\documentclass[11pt]{article}
\usepackage{amsmath}
\usepackage{amsthm}
\usepackage{amssymb}
\usepackage{algorithm}
\usepackage{subfig}
\usepackage{color}
\usepackage[english]{babel}
\usepackage{graphicx}
\usepackage{wrapfig,epsfig}
\usepackage{epstopdf}
\usepackage{url}
\usepackage{graphicx}
\usepackage{color}
\usepackage{epstopdf}
\usepackage{algpseudocode}
\usepackage{scrextend}
\usepackage[T1]{fontenc}
\usepackage{bbm}
\usepackage{comment}
\usepackage{authblk}
\usepackage{multicol}{}
\usepackage{wasysym}
\DeclareMathAlphabet{\mathpzc}{OT1}{pzc}{m}{it}
\usepackage{tikz}
\usepackage{hyperref}
\hypersetup{colorlinks=true,citecolor=red,linkcolor=blue}
\usetikzlibrary{arrows}
\usepackage[margin=1in]{geometry}
\graphicspath{{./figs/}}

\newtheorem{theorem}{Theorem}[section]
\newtheorem{lemma}[theorem]{Lemma}
\newtheorem{definition}[theorem]{Definition}

\newtheorem{corollary}[theorem]{Corollary}

\newtheorem{assumption}[theorem]{Assumption}

\newtheorem{fact}[theorem]{Fact}
\newtheorem{remark}[theorem]{Remark}
\newtheorem{claim}[theorem]{Claim}

\newenvironment{proof-sketch}{\noindent{\textit{Proof Sketch.}}\hspace*{1em}}{\qed\bigskip}

\newcommand{\wh}{\widehat}
\newcommand{\wt}{\widetilde}
\newcommand{\ov}{\overline}

\newcommand{\eps}{\epsilon}

\newcommand{\R}{\mathbb{R}}
\renewcommand{\L}{\mathcal{L}}

\renewcommand{\varepsilon}{\epsilon}

\renewcommand{\hat}{\wh}
\renewcommand{\eps}{\epsilon}

\newcommand{\Tmat}{{\cal T}_{\mathrm{mat}}}

\newcommand{\tmp}{\mathrm{tmp}}
\newcommand{\far}{\mathrm{far}}

\newcommand{\defeq}{:=}
\newcommand{\soft}{\textsc{SoftThreshold}}

\DeclareMathOperator*{\E}{{\mathbb{E}}}
\DeclareMathOperator*{\Var}{{\bf {Var}}}
\DeclareMathOperator*{\Sup}{{\bf {Sup}}}

\DeclareMathOperator{\supp}{supp}

\DeclareMathOperator{\poly}{poly}

\DeclareMathOperator{\nnz}{nnz}

\DeclareMathOperator{\rank}{rank}

\DeclareMathOperator{\diag}{diag}

\DeclareMathOperator{\new}{new}
\DeclareMathOperator{\old}{old}

\DeclareMathOperator{\appr}{appr}
\DeclareMathOperator{\adj}{adj}
\DeclareMathOperator{\sketch}{sketch}

\definecolor{b2}{RGB}{51,153,255}
\definecolor{mygreen}{RGB}{80,180,0}
\definecolor{yl}{RGB}{255,80,0}
\definecolor{zhj}{RGB}{255,50,200}
\definecolor{mycy2}{RGB}{255,51,255}

\newcommand{\Zhao}[1]{}
\newcommand{\Omri}[1]{}
\newcommand{\Hengjie}[1]{}
\newcommand{\Shunhua}[1]{}
\newcommand{\hengjie}[1]{}

\newcommand{\blue}{\textcolor{black}}

\makeatletter
\newcommand*{\RN}[1]{\expandafter\@slowromancap\romannumeral #1@}
\makeatother

\usepackage{lineno}

\allowdisplaybreaks

\begin{document}

\date{}
\title{Faster Dynamic Matrix Inverse for Faster LPs}

\author{
Shunhua Jiang\thanks{\texttt{sj3005@columbia.edu}. Columbia University. Research supported by NSF CAREER award CCF-1844887.}
\quad
Zhao Song\thanks{\texttt{zhaos@ias.edu}. Princeton University and Institute for Advanced Study.}
\quad
Omri Weinstein\thanks{\texttt{omri@cs.columbia.edu}. Columbia University. Research supported by NSF CAREER award CCF-1844887.}
\quad
Hengjie Zhang\thanks{\texttt{hz2613@columbia.edu}. Columbia University. Research supported by NSF CAREER award CCF-1844887.}
}

\begin{titlepage}
  \maketitle
 \begin{abstract}
 
Motivated by recent Linear Programming solvers, we design dynamic data structures for maintaining the inverse of an $n\times n$ real matrix under \emph{low-rank} updates, with polynomially faster amortized running time. Our data structure is based on a recursive application of the Woodbury-Morrison identity for implementing \emph{cascading} low-rank updates, combined with recent sketching technology. Our techniques and amortized analysis of multi-level partial updates, may be of broader interest to dynamic matrix problems.

This data structure leads to the fastest known LP solver for general (dense) linear programs, improving the running time of the recent algorithms of (Cohen et al.'19, Lee et al.'19, Brand'20) from $O^*(n^{2+ \max\{\frac{1}{6}, \omega-2, \frac{1-\alpha}{2}\}})$ to $O^*(n^{2+\max\{\frac{1}{18}, \omega-2, \frac{1-\alpha}{2}\}})$, where $\omega$ and $\alpha$ are the fast matrix multiplication exponent and its dual. Hence, under the common belief that $\omega \approx 2$ and $\alpha \approx 1$, our LP solver runs in $O^*(n^{2.055})$ time instead of $O^*(n^{2.16})$.

 \end{abstract}
  \thispagestyle{empty}
\end{titlepage}

\section{Introduction}

\emph{Dynamic matrix inverse problems} ask to maintain the inverse of an $n\times n$ matrix $M$ over some field (say $\R$), 
when $M$ undergoes a long sequence of row/column updates.  This data structure problem arises in many important 
TCS applications, such as directed reachability in dynamic graphs and maintaining the eigenvalues and rank of a matrix \cite{s04}, 
empirical risk minimization \cite{lsz19}, routing and electrical-flow computation (inverting Laplacians) \cite{st04_lap,m13},   
and in fastest known linear programming solvers \cite{ls15,cls19,blss20}.  
While many variants of this problem have been considered over the years (depending on the specific application), 
the most common one requires the data structure to support \emph{vector-queries}, i.e., computing  
$\textsc{Query}(h) = M^{-1}h$ where $h$ comes from some restricted\footnote{Indeed, supporting \emph{arbitrary}  
online queries $h\in \R^n$ is conjectured to be impossible in truly sub-quadratic time even in the \emph{static} 
case where $M$ remains fixed, see the ``oMV Conjecture'' \cite{hkns15}.} family of $\R^n$, 
under a long sequence of \emph{low-rank} updates (e.g., row/column changes to $M$). 
Recomputing the inverse from scratch upon each update in this setting incurs a daunting computational overhead, 
and therefore the goal is to optimize the tradeoff between 
the (amortized) update time $t_u$ and query time $t_q$ (or the total running time).  
The first dynamic data structure for this problem, tracing back to the 1950's \cite{w49,w50} 
shows how to \emph{explicitly} maintain the inverse of $M$ in $O(n^2)$ time, which is already nontrivial.
Substantial improvements to this data structure were developed more 
recently, showing that \emph{row-queries} (which are equivalent to answering $(M^\top)^{-1}h$ for any \emph{$1$-sparse} vector $h$) 
under row/column updates, can be done in $\max\{t_u,t_q\} = O(n^{1.529})$ 
time \cite{s04,sm10,bns19}. Brand et al. \cite{bns19} also showed a certain conditional $\Omega(n^{1.5})$ worst-case 
lower bound for the \emph{exact} version of this problem, which is important for some of the aforementioned 
applications (e.g., maintaining graph properties such as directed reachability).   

In this work we consider a more challenging dynamic inverse problem: 
The data structure needs to support low-rank updates of the form $\textsc{Update}(u,v) = M + uv^\top$ 
where $u,v$ come from some fixed set of vectors of size $O(n)$, and needs to answer inverse queries $M^{-1}h^{(i)}$ under 
a \emph{slowly changing} vector sequence ($\|h^{(i)} - h^{(i-1)}\|_0 \leq O(1)$).  
In contrast to row-queries (a special case $\|h^{(i)}\|_0=1$), here the online query vectors $h$ may be arbitrarily dense, 
but their \emph{differences} are sparse.\footnote{\label{foot_nonlinear} Note that sparsity of 
$\Delta h$ is \emph{not} equivalent to sparsity of queries  $h$ themselves: If $M$ were \emph{fixed}, 
then by linearity, computing $M^{-1}(\Delta h)$ would indeed suffice, but here $M$ is dynamically changing 
so this standard trick doesn't work.}
An important special case of this dynamic problem 
is ``projection maintenance" \cite{cls19}, where updates to $M$ take the form 
$\textsc{Update}(D) = M + ADA^\top$ where $A$ is an arbitrary fixed matrix and $D$ is a sparse 
diagonal matrix (hence $\rank(ADA^\top)\leq \rank(D)$ is small and can be therefore written as   
$\sum_{i=1}^{rank(D)} A_{i_1}A_{i_2}^\top$).

Nevertheless, an important difference between this work and the aforementioned ones on dynamic inverse 
maintenance (e.g., \cite{s04,sm10,bns19}), is allowing \emph{approximate} answers, i.e., the data 
structure needs to compute $M^{-1}h$ up to some small relative $\ell_\infty$ error. 
This enables the use of randomized tools from \emph{sketching and 
sparse recovery} literature \cite{glps10,cw13,lnnt16,lsz19}.  Another point of departure is our use of  
very heavy \emph{amortization}, augmenting the  approach of \cite{cls19,lsz19,b20} with a novel 
algebraic technique and more sophisticated potential analysis (based on ``high order'' martingales). 
A recurring theme in dynamic data structures (both upper and lower bounds) is that amortized analysis 
is a different ballgame compared to worst-case performance -- Some classic examples %from dynamic graph problems 
are the amortized analysis of  fully-dynamic undirected connectivity \cite{st85} and its matching amortized 
cell-probe lower bound \cite{pd06}, which 
required substantially new techniques (and another decade) compared to  the worst-case lower bound \cite{fs89}. 
In contrast to dynamic graph problems, whose amortized complexity has been studied extensively, its counterpart 
in \emph{dynamic matrix problems} is far less understood (\cite{hkns15, bns19}), 
and we believe our work sheds further light on the power of amortization. 

\paragraph{Dynamic inverse in linear programming} 
The primary application and motivation of this work  is the role of dynamic matrix inverse data structures 
in speeding up  \emph{interior-point methods} (IPMs) for solving  linear programs (LPs) in close to 
matrix-multiplication time. 

Linear programming is one of the cornerstones of algorithm design and convex optimization, 
dating back to as early as Fourier in 1827. LPs are the key toolbox for (literally hundreds of) approximation algorithms, and  
a standard subroutine in convex optimization problems.
Dantzig's 1947 \emph{simplex algorithm} \cite{d47} was the first proposed solution for 
general LPs with $n$ variables and $d$ constraints ($\min_{Ax = b , x\geq 0} c^\top x$). 
Despite its impressive performance in practice, however, the simplex algorithm  
turned out to have exponential worst-case running time (Klee and Minty \cite{km72}). 
The first \emph{polynomial time} algorithm for general LPs was only developed in 1980,  
when Khachiyan \cite{k80} introduced the \emph{Elliposid method}, and showed that 
it runs in $O(n^6)$ time. Unfortunately, this algorithm is very slow in practice compared to 
the simplex algorithm, raising a quest for LP solvers which are efficient in both theory and practice. 

This was the primary motivation behind the development of \emph{interior point methods} 
(IPMs), which uses a %(second order) 
primal-dual gradient descent approach to iteratively converge to a feasible solution (Karmarkar, \cite{k84}). 
An appealing feature of IPM methods for solving LPs is that they are not only guaranteed to run fast in theory, but 
also in practice \cite{s87}. In 1989, Vaidya proposed an $O(n^{2.5})$ LP solver based on a specific implementation of IPMs, 
known as the \emph{Central Path algorithm} \cite{v87, v89_lp}. Vaidya already observed that the main bottleneck of this 
algorithm boils down to a dynamic data structure 
problem of  maintaining the \emph{inverse matrix} $M$ associated with the central 
path equations (see Section~\ref{sec:previous_technique}), under a sequence of updates of the form $(M + A\Delta A^\top )^{-1}$, where $\Delta$ is a sparse diagonal matrix and $A$ is the fixed LP constraint matrix. 
Since $\Delta$ is sparse, each  ``gradient descent'' iteration of this algorithm induces a \emph{low-rank} update to $M$, 
hence it is conceivable to avoid recomputing the inverse matrix from scratch---which would naively cost $n^{\omega}$ 
time per iteration---and gain substantially from amortization. 
This data structure problem was the centerpiece of the recent line of developments on IPM solvers \cite{cls19, lsz19, b20},  
which focused on designing faster  dynamic inverse structures for implementing 
the Central Path algorithm.

The fastest known algorithm for general (dense) LPs,  based on this approach, is due to Cohen, Lee and Song 
\cite{cls19}, 
whose running time is 
$$O^*( n^{\omega} + n^{2.5-\alpha/2} + n^{2+1/6}),$$ 
where $\omega < 2.37$ is the fast matrix-multiplication exponent and $\alpha > 0.31$ is the dual matrix multiplication 
exponent.\footnote{The dual exponent $\alpha$ is defined as the asymptotically maximum number $a\leq 1$ s.t 
multiplying an $n \times n^a$ matrix by an $n^a \times n$ matrix can be done is $n^{2+o(1)}$ time. The current 
best lower bound is $\alpha > 0.31389$ \cite{gu18}.} 
Note that $n^\omega$ is the minimal time for merely inverting a matrix, i.e., finding \emph{any} feasible 
solution ($Ax = b$) to the LP, 
hence it seems quite remarkable that solving the full optimization 
problem ($\min_{Ax = b , x\geq 0} c^\top x$) may be done at virtually no extra cost. 
It is widely believed that $\omega \approx 2$ \cite{cksu05,w12}  and as such, $\alpha \approx 1$ (though the only
formal connection between these constants is $\omega + (\omega /2 ) \alpha \leq 3$ \cite{cglz20}). 
Assuming indeed that $\omega < 2+1/6 \approx 2.166$ and $\alpha > 0.66$, the runtime of the 
aforementioned algorithms is $n^{2.166}$.  Whether the additive $n^{2.166}$ term can be improved or 
completely removed was explicitly posed as an open question in \cite{cls19} and \cite{song19}.  \\

Our main result is an affirmative answer to this open question, asserting that LPs can be solved in 
matrix multiplication time for nearly any value of $\omega$ (i.e., so long as $\omega > 2.055$). 
We design an improved LP solver which runs in time $O^*(n^{\omega} + n^{2.5-\alpha/2} + n^{2+1/18})$.  
In the most notable (ideal) case that $\omega \approx 2$ and $\alpha \approx 1$, our algorithm runs in 
$O^*(n^{2.055})$ time, instead of $O^*(n^{2.166})$ time of previous IPM algorithms \cite{cls19,lsz19,b20}. 
More precisely: 

\begin{theorem}[Main result, Informal statement of Theorem \ref{thm:tech_third_improvement}]\label{thm:main_informal}  
Let $\min_{A x = b, x \geq 0} c^\top x$ be a linear program where $A\in \R^{d\times n}$ and %$d\leq n$ and 
$d=\Omega(n)$. Then for any accuracy parameter $\delta \in (0,1)$, there is a randomized algorithm that 
solves the LP in expected time  
\begin{align*}
& O^*(n^{\omega} + n^{2.5-\alpha/2} + n^{2+1/18})\cdot \log( n / \delta ).
\end{align*}
\end{theorem}

We achieve this result by designing a more efficient projection maintenance data structure, speeding 
up both the update and query times of previous algorithms. This is done via a new algebraic framework 
for bootstrapping lazy updates (described next), combined with randomized compression 
techniques and a sophisticated amortized analysis of the underlying dynamic process.

%%%%%%%%%%%%%%%%%%%%%%%%%%%%%%%%%%%%%%%%%%%%%%%%%%%%

\paragraph{Organization} In Section~\ref{sec:multi} we provide a high-level description of our main 
technique, which is the centerpiece of this paper. Section~\ref{sec:previous_technique} contains some 
necessary background and brief overview of previous related work. In Section~\ref{sec:our_techchniques}, 
we provide a detailed technical 
overview of the proof of Theorem~\ref{thm:main_informal}. This 10-page streamlined %self contained 
overview should be understood as an extended abstract of our entire result, deferring technical proofs and calculations to the Appendix.  

\emph{Appendix Organization.}
Section~\ref{sec:preliminary} contains preliminaries and notation. In Section~\ref{sec:sketching_on_the_left_and_vector_maintenance} we provide an analysis of the Stochastic Central Path algorithm, postponing 
output-feasibility issues to Section~\ref{sec:feasible}. 
We put preliminary part of our data structure in Section~\ref{sec:preliminary_improved}.
We present our full ``cascading data structure'' and prove its correctness in Section~\ref{sec:correctness_improved},
and we analyze its running time in Sections~\ref{sec:time_per_call_improved}, \ref{sec:amortize_time_improved}. Finally, Section~\ref{sec:combine_improved} contains the final runtime analysis of 
our LP solver using the full data structure. In Section~\ref{sec:multi_detail}, we present more details for Section~\ref{sec:multi}.

\section{Bootstrapping low-rank updates via cascading lazy updates}
\label{sec:multi}

One of the main new ideas of this work is ``bootstrapping''\footnote{This term refers to a general approach for speeding 
up dynamic algorithms by repeating a certain technique several times recursively \cite{s04, bns19}. We remark that \cite{bns19} 
uses a different kind of bootstrapping to speed up exact inverse maintenance under simpler row-updates and row-queries 
(via ``one more leve'' of FMM). In particular, \cite{bns19} has nothing to do with \emph{dynamic LU-decompositions} nor 
recursion on Woodburry's identity.  Nevertheless, it is noteworthy that the bottleneck in both works is maintaining certain matrix products for $>2$ ``levels'' of recursion. 
}   
the lazy updates technique of \cite{cls19} by repeated application of the Woodburry-Morrison 
 identity (Lemma~\ref{fact:woodbury}),   
allowing faster low-rank operations  via \emph{cascading lazy updates}.  
For ease of presentation, let us focus on the following simplified dynamic data structure problem, 
which we henceforth call \emph{low-rank inverse maintenance} problem (Definition.~\ref{def:our_omv}).
The data structure is initially given a full rank matrix $M^{(0)} = M \in \R^{n \times n}$, and a fixed vector $h\in \R^n$.
The data structure needs to support a sequence of rank-one updates and vector-queries as follows. 
In the $t$-th \textsc{Update} operation, we are given two vectors $u^{(t)},v^{(t)}\in \R^n$ and a real number 
$c^{(t)}\in \R$, and need to perform a rank-1 update $M^{(t)}=M^{(t-1)} + c^{(t)}\cdot u^{(t)} (v^{(t)})^{\top}$. 
A \textsc{Query} operation asks to calculate $x= (M^{(T)})^{-1}\cdot h$, where $T$ is the number 
of updates in the sequence so far. Note that this setting easily captures rank-$k$ updates invoking $k$ 
consecutive rank-1 updates. 

Achieving sub-quadratic update and query times for this problem requires some restriction on the update vectors $u^{(t)}, u^{(t)}$ 
(we show in Section \ref{sec:multi_detail} that otherwise it would break the oMV Conjecture \cite{hkns15}). Our technique requires one reasonable assumption, 
namely, that all updates $u^{(t)}$ and $v^{(t)}$ come from a fixed set $|S| = O(n)$. As noted in the introduction, LP projection 
maintenance is a special case where $M = (ADA^{\top})$, $A$ is the fixed LP matrix and $D$ is a diagonal matrix which is changing 
slowly under $\ell_0$ norm. Thus, a sparse update $\Delta_i$ to $D$ corresponds to $M \leftarrow M + \Delta_i \cdot A_{i}A_{i}^{\top}$. 

We also remark that the assumption that the query vector $h$ is fixed throughout the sequence---while natural in many streaming applications---is only for simplicity of exposition: Our data structure will actually support an online sequence of slowly-changing queries $h$ (i.e., $\|h^{(t)} - h^{(t-1)}\|_0 = o(n)$). Handling sparse updates to $h$ turns out to be 
much easier than low-rank updates to $M$, hence we focus on the latter task. 

Our technique for solving the \emph{low-rank inverse maintenance} problem is based on an algorithmic generalization of Woodbury's identity to $K>1$ ``levels'' (to be explained below), allowing for \emph{recursive lazy updates} of these $K$ levels using different thresholds. The basic idea is to (dynamically) group updates into $K$ ``epochs'' $0 \leq t_1 \leq \cdots \leq t_{K-1} \leq t_K=T$. The first ``level'' maintains the first epoch $t_1$ and $M^{(t_1)}$. Similarly, in level $k\in \{2,\cdots,K\}$ we group all updates $u^{(t)},v^{(t)},c^{(t)}$ for which $t\in (t_{k-1}, t_k]$, and partition the update sequence in terms of these epochs.  More formally, let $U_k,V_k\in \R^{n \times (t_k-t_{k-1})}$ and the diagonal matrix $C_k\in \R^{(t_k-t_{k-1}) \times (t_k-t_{k-1})}$ be, respectively, the concatenation of all $u_i,v_i$ and $c_i$ in the $k$th epoch, so that $\sum_{i=t_{k-1}+1}^{t_k} c_i \cdot u_i v_i^{\top} = U_kC_kV_k^{\top}$. Let $r_k$ be defined as the rank of $U_kC_kV_k^{\top}$, the epochs are maintained under the invariant that $r_k \leq n_k$, where $n=n_1 \gg  n_2 \gg \cdots \gg n_{K} \geq 1$ are predefined thresholds that decrease exponentially and can later be optimized. In this terminology, for any $h\in \R^n$, the query answer
{\small\begin{align*}
    x := \left(M^{(t_1)} + \sum_{i=t_1+1}^T c_i u_i v_i^{\top}\right)^{-1}h
\end{align*} }
can be equivalently re-written as the following linear system (by adding `dummy' variables $\xi_j$):  {\small
\begin{align} \label{eq_Woodbyrry_K_system}
    \begin{bmatrix}
     x\\ \xi_2 \\ \xi_3 \\ \xi_4 \\ \vdots \\ \xi_K
    \end{bmatrix}=
    \left( \underbrace{\begin{bmatrix}
    M^{(t_1)} & U_2 & U_3 & U_4 & \cdots & U_K \\
    V_2^{\top} & -C_2^{-1} & 0 & 0 & \cdots & 0 \\
    V_3^{\top} & 0 & -C_3^{-1} & 0 & \cdots & 0  \\
    V_4^{\top} & 0 & 0 & -C_4^{-1} & \cdots & 0  \\
    \vdots & \vdots & \vdots & \vdots & \ddots & \vdots  \\
    V_K^{\top} & 0 & 0 & 0 & \cdots & -C_K^{-1}
    \end{bmatrix}}_{D}\right)^{-1}\cdot 
    \begin{bmatrix}
    h\\ 0 \\ 0 \\ 0\\ \vdots \\ 0
    \end{bmatrix}
\end{align}}

Note that this equation is precisely Woodbury's identity, written in a $K$-block matrix form. Indeed, Woodbury's identity
is derived as the solution $x$ to the linear system
\[
\begin{bmatrix}
x\\ \xi
\end{bmatrix}=
\begin{bmatrix}
M^{(t_1)} & U\\
V^{\top} & -C^{-1}
\end{bmatrix}^{-1}\cdot 
\begin{bmatrix}
h\\ 0
\end{bmatrix},
\]
where here $\xi$, $U$, $V$, $C$ denote the concatenation of $\xi_k$, $U_k$, $V_k$, $C_k$ respectively, 
over all $k\in\{2,\cdots,K\}$.
We now explain how the above generalization leads to an efficient data structure for implementing low-rank updates. 

\paragraph{LU-decomposition} 
At update time, the data structure will maintain an LU-decomposition (lower-upper triangular factorization) 
of the matrix $D$ in \eqref{eq_Woodbyrry_K_system}.
\begin{align}\label{eq:LU_decomposition}
    D = L\cdot U,
\end{align}
where $L$ and $U$ are both $K$-block triangular matrices (which are uniquely defined by the Gaussian-elimination 
algorithm and imposing the diagonals of $L$ to be identity matrices, see Eq.\eqref{eq:LU_decomposition_k_2} for an 
example of the case $K=3$).  As we will explain in the next paragraph, such decomposition is useful since the inverse 
of triangular matrices ($L,U$) can be maintained efficiently. This means that the (slowly changing) 
query answers $U^{-1}L^{-1} [h,0]^{\top} = D^{-1}[h,0]^{\top} =[x,\xi]^{\top}$ can be maintained efficiently with respect to the updated $D$.

\paragraph{Cascading lazy updates and query}

The key part of our data structure is performing lazy updates recursively w.r.t different thresholds, to speed up the 
amortized runtime. Recall that in the latest $T$-th update, we are given $u_T$, $v_T$, $c_T$. 
In order to solve the forthcoming queries using this framework, we need to update the epoch in the bottom level $t_K$ 
to reflect the up-to-date $M^{(T)}$, by updating its corresponding tuple $U_K,V_K,C_K$ and 
maintain the LU-decomposition. 
If the update $t_{K}\leftarrow T$ violates the invariant $r_K \leq n_{K}$ (Recall $r_k = t_k-t_{k-1}$ is the rank of 
the update $U_kC_kV_k^{\top}$ in the $k$-th epoch),
we ``cascade'' the update to the next level by updating $t_{K-1}\leftarrow t_{K} \leftarrow T$ and also updating
$U_{K-1},V_{K-1},C_{K-1}$ to include $U_k, V_K, C_K$, and recurse on the next level $t_{K-2}$ and so on, until the threshold invariant 
is restored. 
A \emph{crucial observation} in the implementation of cascading lazy updates and maintaining the LU-decomposition is that 
when level $k$ gets updated, the upper-left $(k-1) \times (k-1)$ blocks of $L$ and $U$, which is the dominating part compared to the entire matrix, \emph{does not change}. Since $L$ and $U$ are triangular matrices, the upper-left $(k-1) \times (k-1)$ blocks of $L^{-1}$ and $U^{-1}$ \emph{also} remain intact. If we write the new $L^{-1},U^{-1}$ as $(L^{-1})^{\new} = L^{-1} + \Delta L$ and $(U^{-1})^{\new} = U^{-1} + \Delta U$, the non-zero part of $\Delta L$ ($\Delta U$) is lower (upper) triangular of $n\times n_{k}$ submatrices that reside on the bottom (right). Thus, 
the query answer can be maintained by computing $(U^{-1})^{\new}(L^{-1})^{\new} [h,0]^{\top}=(U^{-1} + \Delta U)(L^{-1} + \Delta L)[h,0]^{\top}$. 
(The ``heavy'' part $U^{-1}L^{-1}[h,0]^{\top}$ can be reused, and the remaining components are relatively cheap to calculate).

This argument implies that, at level $k$, we can afford time proportional to its size $r_k\leq n_k$ to rebuild the lower-upper triangular matrix $L$ and $U$ on their changed part along with the vector $U^{-1} L^{-1} [h,0]^{\top}$, to perform fast queries. We remark that recomputing $\Delta  L, \Delta U$ requires certain preprocessing which is where we 
exploit the assumption that updates $u,v$ are from a fixed set. 

\paragraph{Application to the LP setting.}
We now explain how the above framework can be successfully applied to the LP setting, i.e., to efficiently 
implement IPM algorithms. As explained in the next section, the goal of each iteration in IPM solvers is to 
(re-)calculate an approximate matrix-vector product $r$ of the following form, 
given new disposition vectors $w^{\appr}$ and $h^{\appr}$ (see Algorithm \ref{alg:alg_CP}):
\[ r := P(w^{\appr}) \cdot h^{\appr} = \sqrt{W^{\appr}} A^{\top} (A W^{\appr} A^{\top})^{-1}A \sqrt{W^{\appr}}\cdot h^{\appr}. \]  
We can use our cascading lazy updates technique to maintain the middle term $(A W^{\appr} A^{\top})^{-1}\cdot h$, 
where for simplicity we assume here that $h$ is some fixed vector.
Indeed, letting $w^{\appr}$ denote the $j$-th iteration update $w^{(j)}$,  
recalculating $(A W^{\appr} A^{\top})^{-1}\cdot h$ corresponds to maintaining a sequence of 
$T_j := \|w^{(j)}-w^{(j-1)}\|_0$ of updates $u_i,v_i,c_i$ followed by one query such that 
$\sum_{i=1}^{T_j} c_i \cdot u_i v_i^{\top} = A (W^{(j)}-W^{(j-1)}) A^{\top}$.

Pictorially, the cascading lazy updates process resembles the following ``chasing game'': 
Child number $t_{k-1}$ is chasing its friend $t_{k}$. Once the distance between them is too large, $t_{k-1}$ 
updates his position to the position of $t_k$ (Figure~\ref{fig:dogwalking}). This process generalizes 
\cite{cls19}, in which there's only one  child $t_1$ chasing its friend $t_2=T$.  
To analyze the amortized cost of this process, we exploit the following special features of the LP problem: 
\begin{figure}[!ht]
    \centering
    \includegraphics[height = 0.18 \textwidth]{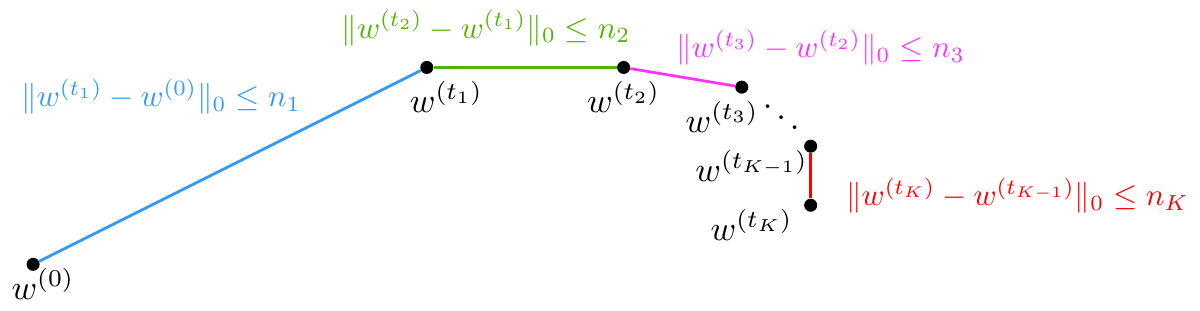}
    \caption{The cascading lazy updates process with per-level  invariants $\|w^{(t_k)}-w^{(t_k-1)}\|_0\leq n_k$. 
    Updates become more expensive but less frequent as we move down the levels.}
    \label{fig:dogwalking}
\end{figure}

\begin{enumerate}
    \item The updates $u^{(t)}, v^{(t)}$ is some row of the original (fixed) LP matrix $A$.
    \item $w^{\appr}$ is slowly changing in each IPM iteration,  
    which imposes nontrivial \emph{sparsity guarantees} that can be used to determine the cascading thresholds 
    $\{n_k\}_{k=1}^K$:  Informally speaking, since $w^{\appr}$ is roughly a martingale, the rank $r_k$ of 
    epoch $k$ will be typically far less than its boundary condition ($r_k\ll t_k-t_{k-1}$) -- It takes  about 
    $ \sqrt{n_k}$ LP-iterations for epoch $k$ to exceed the threshold $n_{k}$, in which case we need to 
    update epoch $k-1$ and compute $\Delta L$, $\Delta U$ which are of size $n\times n_{k-1}$.
    (see Sections~\ref{sec:our_technique_update},~\ref{sec:amortize_time_improved} for more details).
    \item The robustness of the Central Path algorithm implies that queries $M^{-1}\cdot h$ can 
    tolerate small relative $\ell_{\infty}$ error. This allows the use of randomized compression   
    (left-sketching $R^{\top}R\cdot U^{-1}L^{-1}[h,0]^{\top}$) to further reduce the running time (see Section~\ref{sec:our_technique_query} for technique overview).
\end{enumerate}

Therefore, assuming ideal matrix-multiplication constants ($\omega = 2$, $\alpha=1$) and that 
$\Delta L$ and $\Delta U$ can be recomputed in time linear in their sparsity $O(n\cdot n_k)$ for 
\emph{every level} $k \in [K]$, the vector $(L^{-1})^{\new} [h,0]^{\top} = (L^{-1} + \Delta L)[h,0]^{\top}$ can 
be maintained in $O(n\cdot n_k)$ time. Furthermore, the solution
\begin{align*}
(U^{-1})^{\new} (L^{-1})^{\new} [h,0]^{\top} 
= & ~ (U^{-1} + \Delta U)(L^{-1} + \Delta L) [h,0]^{\top}\\
= & ~ U^{-1} L^{-1} [h,0]^{\top} + \Delta U\left((L^{-1} + \Delta L)[h,0]^{\top}\right) + (U^{-1} \Delta L) [h,0]^{\top}
\end{align*}
can be maintained in the same $O(n\cdot n_k)$ time since the non-zero part of $\Delta L$ is an $(n_k\times n)$ block.
The bottom $K$-th level is a special one, as every level $k$ except the last one involves extra pre-computation 
to support the next $(k+1)$-th level. As we show, with some extra (non-trivial) effort, this fact enables maintaining
the $K$-th level in only $O(n_{K-1}\cdot n_K)$ time instead of the brute-force $O(n\cdot n_K)$ time.
(Our algorithm implements $K=3$ levels of this framework. By convention, in later sections we refer to the third level 
as the ``query level''; The first and second levels are referred as ``two-level'' updates). 

In conclusion, the ideal time per LP-iteration according to our framework is proportional to 
\begin{align}\label{eq:recursive_multilevel_runtime_before_setting_parameter}
\sum_{k=1}^{K-1} \underbrace{ ( n \cdot n_k / \sqrt{ n_{k+1} } ) }_{ \text{$K-1$ levels} } + 
\underbrace{ n_{K-1} \cdot n_K }_{ K\text{-th~level}},
\end{align}
where $n\cdot n_k/\sqrt{n_{k+1}}$ is the amortized update cost of level $k\in [K-1]$, and $n_{K-1}\cdot n_K$ is the 
update cost of the last level $K$. 
Carefully balancing the terms by setting geometrically decreasing  thresholds $n_k$'s 
(see Claim~\ref{cla:multi_tex_running_time_prove}), this runtime is optimized to be 
\begin{align}\label{eq:recursive_multilevel_runtime_K}
\wt{O}(K) \cdot n^{2+\frac{1}{6\cdot(2^{K-1}-1)}} = \wt{O}(K)\cdot n^{2+O(2^{-K})}, 
\end{align}
using the fact that the LP algorithm (see Algorithm \ref{alg:alg_CP}) has $\wt{O}(\sqrt{n})$ iterations.  
The formal calculation can be found in Section~\ref{sec:recursive_multilevel_runtime}. 
Note that when $K=2$, this running time is precisely $O^*(n^{2+1/6})$, matching \cite{cls19,lsz19,b20}'s results. 
For $K=3$ levels, this matches our result $O^*(n^{2+1/18})$. 

%%%%%%%%%%%%%%%%%%%%%%%%%%%%%%%%%%%%%%%%%%%%%%%%%
%%%%%%%%%%%%%%%%%%%%%%%%%%%%%%%%%%%%%%%%%%%%%%%%%
Achieving the runtime predicted by Eq.~\eqref{eq:recursive_multilevel_runtime_K} for 
$K = \omega(1)$ levels would show that LPs can be solved in the ideal time 
$O^*(n^{\omega})$. The main challenge in implementing our technique  
beyond $K>3$ levels is maintaining the LU-decomposition of \eqref{eq:LU_decomposition} 
in the desired ($\sim n\cdot n_k$) time, and indeed doing so even for $K=3$ is highly nontrivial, 
as this paper shows. While completing this ambitious program is beyond the reach of this paper, 
we believe the cascading lazy updates framework may have applications in future developments in 
LP and SDP solvers, 
cutting plane methods, and more generally for dynamic inverse problems, hence it is worthwhile presenting it at its full generality.

\section{Background}
\label{sec:previous_technique}

This section provides a brief overview of recent developments in LP solvers, 
the optimization framework of IPMs and its relation 
to dynamic inverse problems. 

\subsection{Recent developments in LP solvers}

The recent work of \cite{cls19} improved the $O(n^{2.5})$ LP algorithm of \cite{v89_lp} 
to\footnote{The running time is actually $O^*(n^{\omega}+n^{2.5-a/2}+n^{1.5+a}+n^{\omega a+0.5}+n^{2a+0.5})$, and the $n^{\omega a+0.5}$ and $n^{2a+0.5}$ terms also come from query time. But they are dominated by other terms. In our paper we improved not only the $n^{1+a}$ term, but also these other two terms.} 
\begin{align}\label{eq_runtime_cls}
O^*(n^{\omega}+n^{2.5-a/2}+n^{1.5+a}),
\end{align}
where $a\leq \alpha$ is a tunable parameter, and $\omega, \alpha$ are the fast matrix multiplication exponent and its dual, respectively. 
Note that for the current values of $\omega\approx 2.38$ and $\alpha\approx 0.31$, this running time is already $O^*(n^{\omega})$. 
However, under common belief that $\omega=2$ and $\alpha=1$ \cite{cksu05,w12}, the running time is $O^*(n^{2+1/6})$, hence there 
is still a polynomial gap to the ideal running time $O^*(n^2)$. 

The three main ingredients in \cite{cls19}'s algorithm are: (i) considering a \emph{stochastic} version of the Central Path algorithm 
(see Algorithm \ref{alg:alg_CP} below), and then leveraging the \emph{robustness} of this algorithm to design a 
more efficient matrix maintenance data structure via \emph{subsampling} (sparsification of the ``gradient'' vector)
yielding $o(n^2)$ query time per iteration. 
(ii) \emph{Lazy updates}: Delaying updates to the projection matrix (associated with the central path equation) via ``soft thresholding'' 
and analyzing their amortized performance via martingale-based potential analysis. 
(iii) Using fast rectangular matrix multiplication to gain extra speedup. 
Our data structure will also take advantage of these building blocks.

\subsection{Optimization: The stochastic central-path algorithm} \label{sec:intro_optimization}

We use a similar optimization framework as that of \cite{cls19} (see Section 2.1 
for a more detailed explanation and context).  Roughly speaking, the Central Path (CP) algorithm maintains a 
primal-dual pair of vectors, $x^{(i)}$ and $s^{(i)}$, and iteratively shrinks the \emph{duality gap}  
$\mu^{(i)} := \sum_{j=1}^n x_j^{(i)}  s_j^{(i)}$ by $\sim (1-1/\sqrt{n})$ in each iteration, until converging to a 
feasible point ($\mu^{(i)} \approx 0$). Hence, the Central Path algorithm has a total of $O(\sqrt{n})$ iterations.
In matrix notation, this algorithm essentially boils down to implementing the following iterative algorithm \cite{cls19}:
  
\begin{algorithm}[!ht]\caption{\small Stochastic Central Path}\label{alg:alg_CP} 
\small
\begin{algorithmic}[1]  
\State $i \leftarrow 1$, initialize $x,s \in \R^n$
\While{$i < \sqrt{n}$} \Comment{In each iteration, we hope $\mu \approx t$}
    \State $t \leftarrow t \cdot (1- 1/ \sqrt{n} )$ \Comment{target decrease of duality gap} 
    \State $\mu \leftarrow x \cdot s$ \Comment{actual decrease in duality gap}
    \State Compute $\delta_{\mu}$ based on $-\frac{\mu}{\sqrt{n}}$ and the gradient $-\nabla\Psi(\mu/t-1)$.
    \State $P \leftarrow \sqrt{\frac{X}{S}} A^\top ( A \frac{X}{S} A^\top )^{-1} A \sqrt{ \frac{X}{S} }$ 
    \Comment{\emph{matrix inverse}, matrix-matrix mult.}
    \State $\delta_x \leftarrow \frac{X}{\sqrt{XS}} (I-P) \frac{1}{\sqrt{XS}} \delta_{\mu} $,  $\delta_s \leftarrow \frac{S}{ \sqrt{XS} } P \frac{1}{ \sqrt{XS} } \delta_{\mu}$ \Comment{matrix-vector mult.}
    \State $x \leftarrow x + \delta_x$, $s \leftarrow s + \delta_s$, $i \leftarrow i+1$
\EndWhile
\end{algorithmic}
\end{algorithm}
Here, $X = \diag(x)$, $S = \diag(s)$ are the primal and dual vectors, 
and $A \in \R^{n\times n}$ is the (fixed) LP constraint matrix. 
$\Psi$ is a potential function measuring how close $\mu$ is from $t$ (the ``target'' duality gap), 
and the vector $\delta_{\mu}$ has two purposes: decreasing $\mu$ by a $(1-1/\sqrt{n})$ factor while keeping 
the potential function bounded. The vectors $\delta_{x}, \delta_{s}$ compute the disposition  
of the primal and dual vectors in each iteration. $P$ is an orthogonal \emph{projection matrix} 
($P^2=P$ and $P=P^{\top}$), and the formulas $\frac{X}{ \sqrt{XS} }$, $\frac{S}{ \sqrt{XS} }$, 
$\frac{1}{\sqrt{XS}}$, and $\frac{X}{S} \in \R^{n \times n}$ are the diagonal matrices of the corresponding vectors.

A key observation in \cite{cls19} is that this algorithm is \emph{robust to small perturbations} along the 
central path: Denoting by $w$ the vector $x/s$, and by $h$ the vector 
$\frac{\delta_{\mu}}{\sqrt{XS}}$, \cite{cls19} shows that in the above algorithm, is enough to 
approximately maintain $ w^{\appr}\approx_{\epsilon_{\mathrm{mp}}} w, h^{\appr}\approx_{\epsilon_{\mathrm{mp}}}h , $  
where $\epsilon_{\mathrm{mp}}<1/4$ and $\approx_{\epsilon_{\mathrm{mp}}}$ denotes coordinate-wise approximation. 

%%%%%%%%%%%%%%%%%%%%%%%%%%%%%%%%%%%%%%%
%%%%%%%%%%%%%%%%%%%%%%%%%%%%%%%%%%%%%%%

\subsection{Data structures: Projection maintenance}

The main bottleneck of Algorithm \ref{alg:alg_CP} is to efficiently maintain the approximate projection matrix
\begin{align}
P(w^{\appr}) = \sqrt{W^{\appr}} A^{\top} (A W^{\appr} A^{\top})^{-1}A \sqrt{W^{\appr}},  
\end{align}
recalculating the queries $r:=P(w^{\appr})h$ on line 7, where $h:=\delta_{\mu}/\sqrt{XS}$. There are $O(\sqrt{n})$ iterations.

\paragraph{Lazy updates.} It was already observed in \cite{v89_lp} that since each iteration only changes the $\ell_2$ mass 
of $w$ by a small amount (which can be turned into an $\ell_0$ sparsity guarantee by ``rounding'' and absorbing a small error), 
most of the time the queries can be answered efficiently by computing the \emph{low-rank} incremental change in $P$, amortizing 
away the rare cases where too many coordinates of $w$ have changed, which are handled using brute force 
fast matrix multiplication.
As noted above, \cite{cls19} further used the power of fast rectangular matrix multiplication: 
By definition of $\alpha$, for any threshold parameter $a\leq \alpha$, the complexity of multiplying an $n\times n^a$ rectangular 
matrix by an $n^a\times n$ rectangular matrix is the same  as multiplying an $n\times 1$ vector with a $1\times n$ vector, so they only update $P$ when 
\emph{at least $n^a$ coordinates} of $w$ have changed. \cite{cls19} further make a clever use of \emph{soft thresholding} on $n^a$, 
which combined with a potential function analysis,  yields an amortized $O(n^{\omega-1/2}+n^{2-a/2})$ update time per iteration. Note that the $n^{2-a/2}$ term needs to be balanced with query time. \cite{lsz19} and \cite{b20} both follow the same update scheme.

\paragraph{Fast queries.}
Computing the queries $r = P(w^{\appr})h$ in each iteration from scratch takes $n^2$ time since the projection matrix 
may be very dense. The three papers \cite{cls19,lsz19,b20} proposed different techniques for speeding up this matrix-vector 
multiplication to $o(n^2)$. In \cite{cls19,lsz19}, the authors use the idea of ``iterating and sketching'' \cite{song19}, 
an adaptive version of the classic ``sketch and solve'' paradigm \cite{cw13}. In \cite{b20}, the author maintains the query answer $r$ 
directly, exploiting the observation that the vector $h$ is also slowly changing (and not just updates $w$). 
Both techniques (\cite{cls19,b20}) are essentially using \emph{sparsification} of the vector $h$, hence involve a 
``right hand side'' linear operation. In contrast, \cite{lsz19} uses a ``left-hand side'' operation by  
\emph{sketching the projection matrix} itself, effectively making it smaller.  

\paragraph{Sampling on the right \cite{cls19}.}  
Here the idea is to apply a $O(\sqrt{n})$-sparse \emph{diagonal sampling matrix} $D \in \R^{n \times n}$ on the 
{\em right} hand side of the maintained matrix:
\begin{align*}
\sqrt{W^{\appr}} A^{\top} ( A W^{\appr} A^{\top} )^{-1} A \sqrt{W^{\appr}} \underbrace{\blue{D}}_{\text{sample~right}} h.
\end{align*}
After this sampling, the vector $Dh$ becomes $O(\sqrt{n})$-sparse, so computing the multiplication of a matrix with 
$Dh$ takes $O(n^{1.5})$ time. Also, since the rank of the change in $W$ is guaranteed to be at most $n^a$, there is 
also a $O(n^{1+a})$ term in the query time.

\paragraph{Sketching on the left \cite{lsz19}.}
The idea here is to apply a (subsampled) Hadamard/Fourier transform matrix \cite{ldfu13,psw17} $R \in \R^{\sqrt{n} \times n}$, by 
{\em sketching on the left} hand side:
\begin{align*}
\underbrace{ \blue{R^\top R} }_{ \text{sketch~left} } \sqrt{W^{\appr}} A^{\top} ( A W^{\appr} A^{\top} )^{-1} A \sqrt{W^{\appr}} h.
\end{align*}
In this way the matrix $RM$ has size $\sqrt{n}\times n$,\footnote{Intuitively, $O(\sqrt{n})$ rows is the minimal 
sketch size one can hope for, as this ensures that with $O(\sqrt{n})$ CP iterations, every coordinate $i\in[n]$ 
has a constant chance to be sampled}
so multiplying this matrix with a vector takes $O(n^{1.5})$ time. So the final query time is also $O(n^{1.5}+n^{1+a})$. 
\cite{lsz19} algorithm is also maintaining some extra vectors, e.g., the explicit/implicit version of $x,s$.

\paragraph{Maintaining query answers \cite{b20}.}
Brand observed that is possible to maintain not only 
the inverse matrix $P(w^{\appr})$ via lazy updates, but also the previously computed 
query answers $r$, by observing that the vector $h$ it also changes slowly. In each iteration, the new $r$ is computed as:
\begin{align*}
r + \sqrt{W^{\appr}} A^{\top} ( A W^{\appr} A^{\top} )^{-1} A \sqrt{W^{\appr}} \underbrace{ \blue{\Delta h}}_{\text{change in }h}.
\end{align*}
\cite{b20} chooses a similar sparsity threshold $n^a$ for the vector and updates the maintained $r$ when $\Delta h$ exceeds $n^a$,
so $\Delta h$ is guaranteed to be $n^a$-sparse at query time. As such, the query time is $O(n^{1+a})$.
It is noteworthy that this algorithm is deterministic, as it avoids sketching/sampling altogether. 
Indeed, the main motivation of \cite{b20} was derandomizing \cite{cls19}.

\paragraph{Our approach.} 
We show how to break the $O(n^{1+a})$ barrier for query time, by \emph{combining} both left-hand and right-hand linear 
transformations on $P$, together with the \emph{cascading lazy updates} technique from Section~\ref{sec:multi}. Such combination 
is needed to further ``compress''  the matrix-vector multiplication when (re)calculating $r$. It turns out that using two sources 
of randomness---sampling on the right \cite{cls19} + sketching on the left \cite{lsz19}---obliterates the error analysis (which needs to 
be controlled to ensure convergence), and this is intuitively where \cite{b20}'s deterministic alternative is useful for us as a 
substitute to right-hand-sampling.

\section{Detailed Technical Overview}\label{sec:our_techchniques}
This section provides a detailed, streamlined technical overview of the proof of Theorem \ref{thm:main_informal}, 
which we restate formally below. This section should be understood as a self-contained extended abstract of our 
entire algorithm. Formal proofs of all technical claims can be found in the Appendix. 

\begin{theorem}[Main result]\label{thm:tech_third_improvement}
Given a linear program $\min_{A x = b, x \geq 0} c^\top x$ with no redundant constraints. Assume that the polytope has 
diameter $R$ in $\ell_1$ norm, namely, for any $x \geq 0$ with $A x = b$, we have $\| x \|_1 \leq R$.

Then, for any $\delta \in (0,1]$, \textsc{Main}$(A,b,c,\delta,a,\wt{a})$ (Algorithm~\ref{alg:main_improved}) outputs $x \geq 0$ such that
\begin{align*}
c^\top x \leq \min_{A x = b, x \geq 0} c^\top x + \delta \| c \|_{\infty} R , \text{~~~and~~~} \| A x - b \|_1 \leq \delta \cdot ( R \| A \|_1 + \| b \|_1 )
\end{align*}
in expected time
\begin{align*}
&~ O( n^{\omega + o(1)} + n^{2.5-a/2+o(1)} + n^{1.5+a-\wt{a}/2+o(1)} + n^{0.5+a+(\omega-1)\wt{a}}) \cdot \log( n / \delta ) 
\end{align*}
for any $0<a\leq \alpha$ and $0<\wt{a}\leq \alpha a$.  
In particular, so long as the constants of fast matrix multiplication satisfy $\omega > 2.055$ and $\alpha > 5-2\omega$,  
general LPs can be solved in $O(n^{\omega+o(1)})$ time. In the ideal case that $\omega=2$ and $\alpha=1$, the running time is $n^{2+1/18} = n^{2.055}$.
\end{theorem}

The first two terms $n^{\omega}$ and $n^{2.5-a/2}$ of our running time are the same as \cite{cls19}, 
stemming from the amortized cost of lazy updates (for $K=1$ levels). The $n^{1.5+a-\wt{a}/2}$ term 
comes from the amortized cost of our cascading lazy updates algorithm for the $K=2$nd level update. 
Finally, the $n^{0.5+a+(\omega-1)\wt{a}}$ term is the worst-case cost of the query algorithm.
We note that a prerequisite for achieving the runtime of Theorem \ref{thm:main_informal} is 
removing the explicit $n^{1+a}$ term as well as the two (implicit) $n^{a\omega}$, $n^{2a}$ terms
in the query time of all previous works. Below we elaborate on how this is achieved step by step, as shown in Table~\ref{tab:running_time}.

\begin{table}[H]
\small
    \centering
    \begin{tabular}{|l|l|l|l|l|} \hline
        {\bf Statement} & {\bf Technique} & {\bf Time} & {\bf Ideal} & {\bf Choice} \\ \hline
        \cite{cls19} & Sec.~\ref{sec:previous_technique} & $n^{2.5-a/2} + n^{0.5+2a} + n^{0.5+\omega a} + n^{1.5+a}$ & $n^{2+1/6}$ & $a=2/3$ \\ \hline
        Thm.~\ref{thm:tech_first_improvement} & Sec.~\ref{sec:tech_first_improvement} & $n^{2.5-a/2} + n^{0.5+2a} + n^{0.5+\omega a}$ & $n^{2+1/10}$ & $a=4/5$\\ \hline
        Thm.~\ref{thm:tech_second_improvement} & Sec.~\ref{sec:tech_second_improvement} & $n^{2.5-a/2} + n^{0.5+2a}$ & $n^{2+1/10}$ & $a=4/5$ \\ \hline
        Thm.~\ref{thm:tech_third_improvement} & Sec.~\ref{sec:tech_third_improvement} & $n^{2.5-a/2} + n^{1.5+a-\wt{a}/2} + n^{0.5+a+(\omega-1)\wt{a}}$ & $n^{2+1/18}$ & $a=8/9, \wt{a} =2/3 $ \\ \hline
    \end{tabular}
    \caption{\small We ignore the $n^{\omega}$ term, and also ignore all the $n^{o(1)}$ terms. {\bf Ideal} denotes the resulting running time when 
    $\omega=2$ and $\alpha=1$. (Note that the current values are $\omega \sim 2.38$ and $\alpha \sim 0.31$).}
    \label{tab:running_time}
\end{table}

\begin{figure}[H]
    \centering
    \includegraphics[height = 0.2\textwidth]{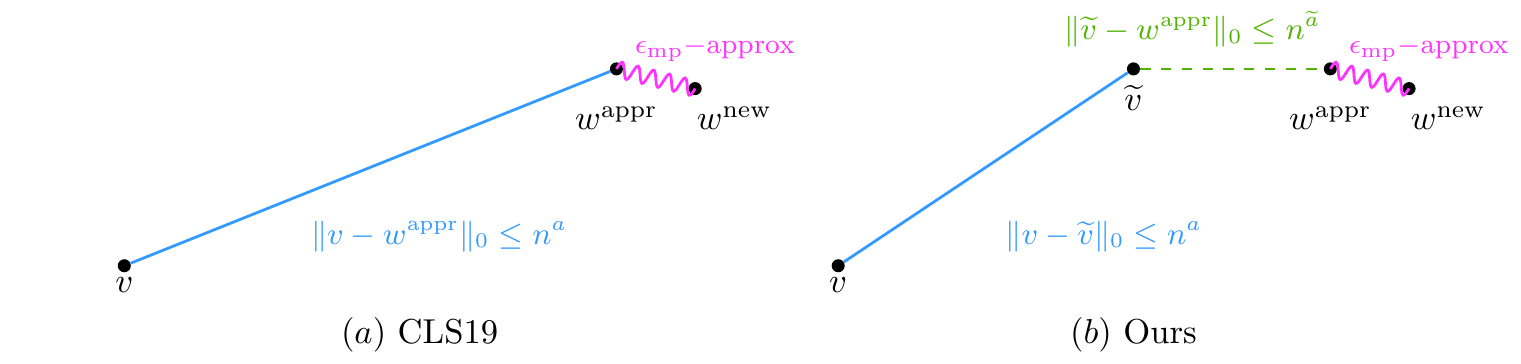}
    \caption{\small $(a)$: In each iteration, we are given $w^{\new}$ which is changing slowly. The algorithm will find $w^{\appr}$ such that $w^{\appr}$ is coordinate-wise $\epsilon_{\mathrm{mp}}$ close to $w^{\new}$({\color{mycy2}pink} wave line), and it is also close to $v$ in $\ell_0$ norm({\color{b2}blue} solid line). $(b)$: Based on $(a)$, we add an intermediate level $\wt{v}$, such that in the query, $\|v-\wt{v}\|_0\leq n^a$({\color{b2}blue} solid line) and $\|\wt{v}-w^{\appr}\|_0\leq n^{\wt{a}}$({\color{mygreen}green} dashed line).}
    \label{fig:intro_chase}
\end{figure}

We note that, although our better running time is achieved by breaking the three bottlenecks of the \emph{query} time of previous works, 
it actually requires non-trivial design of both the \emph{query} part and the \emph{update} part. Indeed, our data structure involves an 
entirely new update subroutine for second-level of updates. 

In the next Section~\ref{sec:our_technique_query}, we walk through the techniques for removing the three Query bottlenecks one by one, 
while introducing the data structure members that must be maintained in order to make queries faster. At a high level, these members 
correspond to maintaining the components appearing in the LU-decomposition
part of our cascading framework for $K=2$ levels (described in Section~\ref{sec:multi}, see Figure~\ref{fig:intro_chase} for illustration). 
In Section~\ref{sec:our_technique_update} we describe and analyze the Update algorithm of our data strucutre. 
This part uses a combination of ``soft thresholding'' for two levels of updates (Figure~\ref{fig:double_threshold}), capturing 
the cascading lazy updates process. We remark that this type of analysis may be of independent interest to ``bootstrapping'' 
techniques in dynamic matrix problems.

%%%%%%%%%%%%%%%%%%%%%%%%%%%%%%%%%%%%%%%%%%%%%%%%%%%%%%%%%%
%%%%%%%%%%%%%%%%%%%%%%%%%%%%%%%%%%%%%%%%%%%%%%%%%%%%%%%%%%

\subsection{Our Query Algorithm}
\label{sec:our_technique_query} 
Our dynamic data structure relies on the following three main ingredients, for 
removing each of the three bottlenecks shown in Table~\ref{tab:running_time} respectively: 
\begin{enumerate}
    \item ``Compressing'' the projection matrix using linear operations on both sides (Section~\ref{sec:tech_first_improvement}): We show that by combining a sketching on the left to reduce the matrix size from $n\times n$ to $\sqrt{n} \times n$, with the direct maintenance of the previous query answer that makes the new query vector sparse, already breaks the $n^{1+a}$ query time barrier of previous works. 
    \item Faster queries through the ``cascading lazy updates'' framework (Section~\ref{sec:tech_second_improvement}): 
    The core technique of previous works is exploiting the fact that the inverse matrix and the query vector are both changing slowly over iterations, 
    hence the data structure can  maintain all intermediate values from previous iterations, and only re-compute the differences. 
    As explained in the last part of Section~\ref{sec:multi}, we show that even the \emph{``derivative''} of updates and queries is  
    slowly changing (these are the ``sparsity thresholds'' we allude to in Section~\ref{sec:multi}). 
    It is therefore natural to maintain a second-level of values (``changes'' of the inverse matrix and the query vector), where each new iteration 
    only computes the changes in these second-level values. We update both levels according to the \emph{cascading lazy updates} 
    framework (recall Figure~\ref{fig:dogwalking}).  We show this technique removes the $n^{a\omega}$ term in the query time of 
    previous data structures (Table~\ref{tab:running_time}).
    \item Maintaining the components of the second-level structure efficiently (Section~\ref{sec:tech_third_improvement}): This is done by essentially recursing the 
    approach of maintaining the first-level to the second-level, by carefully designing the maintained objects (matrix-products, vectors, sets). This removes the $n^{2a}$ term in the query time, which in fact appeared in \emph{multiple} places in previous algorithms. 
\end{enumerate}

We now turn to a detailed description of the query algorithm, using \cite{cls19} as a baseline. 
Recall that each iteration of the CP algorithm generates two vectors $w$ and $h$ 
from the previous outputs of the data structure. The data structure needs to re-calculate 
\begin{align*}
\sqrt{ W } A^\top ( A W A^\top )^{-1} A \sqrt{ W } h.
\end{align*}
By robustness of the stochastic  CP algorithm, it suffices to output an approximated value
\begin{align*}
r = \underbrace{ (\textcolor{blue}{R})^{\top} }_{\rm de-sketch } \cdot \underbrace{ \textcolor{blue}{R}\sqrt{ W^{\appr} } 
A^\top ( A W^{\appr} A^\top )^{-1} A \sqrt{ W^{\appr} } h^{\appr} }_{ \rm sketch },
\end{align*}
where for some fixed parameter $\epsilon_{\mathrm{mp}}<1/4$ the data structure guarantees that $w^{\appr}\approx_{\epsilon_{\mathrm{mp}}}w$ and $h^{\appr}\approx_{\epsilon_{\mathrm{mp}}}h$. 
We use the same sketching matrix $R$ as that of \cite{lsz19} which satisfies $\E[R^{\top}R]=I$ and a guarantee that the variance and $\ell_{\infty}$ error of $(R^\top RPh - Ph)$ are both small. 
$R$ has size $\sqrt{n}\times n$ -- this value is natural, since we have 
$\sim \sqrt{n}$ CP iterations, hence using $o(\sqrt{n})$ linear measurements would fail to even detect changes in all $n$ coordinates. Also, the size of $R$ allows us to pre-batch $O(\sqrt{n})$ copies of sketching matrices in the beginning, which only takes $O(n^{\omega})$ time.

As in \cite{cls19, b20}, our data structure maintains a member $v$ that serves as a proxy for $w$, and a member $g$ that serves as a proxy for $h$.
If the new value $w^{\appr}$ is too far from $v$, then in addition to computing the new displacement $r$, the data structure also 
updates $v$ along with all of its members that depend on $v$. An analogous scheme is used for $g$ w.r.t $h^{\appr}$. Since the new values $w^{\appr}$ and $h^{\mathrm{sample}}$ show up in three places in the output $r$, from now on we will refer to these three places as the left part, the middle part, and the right part, and label them as $a^{\new} \in \R^{\sqrt{n} \times n}$, $b^{\new} \in \R^{n \times n}$, $c^{\new} \in \R^n$ for ease of presentation. Accordingly, we use $a \in \R^{\sqrt{n} \times n}$, $b \in \R^{n \times n}$, $c \in \R^n$ to denote the values that depend on $v$ and $g$:
\begin{align*}
\underbrace{R\sqrt{ W^{\appr} }}_{a^{\new}} \underbrace{A^\top ( A W^{\appr} A^\top )^{-1} A}_{b^{\new}} \underbrace{\sqrt{ W^{\appr} } h^{\appr} }_{c^{\new}} \;\;\; , ~~~
\underbrace{R\sqrt{V}}_{a} \underbrace{A^{\top}(AVA^{\top})^{-1}A}_{b} \underbrace{\sqrt{V}g}_{c}.
\end{align*}
We also denote $\partial a := a^{\new} - a$, $\partial b := b^{\new} - b$, $\partial c := c^{\new} - c$.

For a tunable parameter $a\in (0,\alpha]$, the worst case query time per iteration of \cite{cls19}'s data structure 
for implementing this process is $t_q  =  n^{2a} + n^{\omega a} + n^{1+a}$, the three terms come from the cost of recomputing the query $r$. 
In the remainder of this subsection, we describe how this query time can be improved to  
\begin{align}\label{eq_improved_query_time}
t_q  = n^{a+(\omega-1)\wt{a}} 
\end{align}
for any  $a\in (0,\alpha]$ and $\wt{a} \in (0,\alpha a]$. We show this in three steps, removing the terms  
$n^{1+a}$, $n^{\omega a}$ , and $n^{2a}$ one by one.

\subsubsection{Technique for removing  \texorpdfstring{$n^{1+a}$}{}}\label{sec:tech_first_improvement}
We combine the ``sketching on the left'' technique of \cite{lsz19} and the ``query maintenance'' technique 
of \cite{b20} to remove this term. In \cite{b20} the query is computed as {\small
\begin{align*}
 r =
 \underbrace{\sqrt{W^{\appr}}}_{a^{\new}}\bigg(\underbrace{\beta_2}_{bc} + \underbrace{M}_{b}\underbrace{\big(\sqrt{W^{\appr}}h^{\appr} -\sqrt{V}g\big)}_{\partial c}+ \underbrace{\big(- M_S (\Delta_{S,S}^{-1}+M_{S,S})^{-1}(M_S)^{\top}\big)}_{\partial b} \underbrace{\sqrt{W^{\appr}}h^{\appr}}_{c^{\new}}\bigg),
\end{align*}}
where $M := A^{\top}(AVA^{\top})^{-1}A$, $ \beta_2 :=M\sqrt{V}g \in \R^n$ are members that the data structure maintains,   
$\Delta := W^{\appr} - V$, and $S := \supp(w^{\appr} - v)$. The subscript $M_S$ means taking the sub matrix of columns of $M$ in the 
set $S$. Note that this output is $a^{\new}(bc + b\cdot \partial c + \partial b\cdot c^{\new}) = a^{\new}b^{\new}c^{\new}$.

The $n^{1+a}$ term shows up in three places when computing different terms:
\begin{itemize}
    \item In $a^{\new}\cdot b\cdot \partial c$, multiplying a $n\times n$ matrix $M$ with a $n^a$-sparse vector $\big(\sqrt{W^{\appr}}h^{\appr} -\sqrt{V}g\big)$.
    \item In $a^{\new}\cdot \partial b\cdot c^{\new}$, multiplying a $n\times n^a$ matrix $M_S$ with a $n^a\times 1$ vector $(\Delta^{-1}_{S,S} + M_{S, S})^{-1} (M_{S})^{\top}\sqrt{W^{\appr}}h^{\appr}$.
    \item In $a^{\new}\cdot \partial b\cdot c^{\new}$, multiplying a $n^a\times n$ matrix $(M_S)^{\top}$ with a $n\times 1$ vector $\sqrt{W^{\appr}}h^{\appr}$.
\end{itemize}
The last one is easy, and we deal it by splitting $c^{\new}\in \R^n $ into $c+\partial c $ again:
\begin{align*}
(M_{S})^{\top}\underbrace{\sqrt{W^{\appr}}h^{\appr}}_{c^{\new}} = (\beta_2)_S + (M_{S})^{\top}\underbrace{(\sqrt{W^{\appr}}h^{\appr} - \sqrt{V}g)}_{\partial c}.
\end{align*}
Since the $\partial c$ term is $n^a$-sparse, this computation only takes $n^{2a}$ time now.

We deal with the first two $n^{1+a}$ terms by adding the sketching matrix $R$ on the left, and splitting $a^{\new}$ into $a+\partial a$. 
Aside from maintaining $M=A^{\top}(AVA^{\top})^{-1}A$ and $\beta_2=M\sqrt{V}g$, we further maintain their sketched versions: 
\[
Q=R\sqrt{V}M \in \R^{\sqrt{n} \times n}, ~~\beta_1=R\sqrt{V}\beta_2 \in \R^{\sqrt{n}}.
\]
In addition to the temporary variables $S=\supp(w^{\appr}-v)$ and $\Delta=W^{\appr}-V$, we also define $\Gamma :=\sqrt{W^{\appr}}-\sqrt{V}$. We have the guarantee that $\Delta$, $\Gamma$ and $(\sqrt{W^{\appr}} h^{\appr} - \sqrt{V}g)$ are all $n^a$-sparse and $|S|\leq n^a$, because otherwise the algorithm would first update $V$ and $g$. The query part of our new data structure becomes
{\small
\begin{align*}
    r_1 := &~ \underbrace{\beta_{1}}_{abc}, ~~~~~~
    r_2 := ~ (\underbrace{Q_S}_{ab} + \underbrace{R \Gamma M}_{\partial a\cdot b})\underbrace{ (\sqrt{W^{\appr}}h^{\appr} - \sqrt{V}g)}_{\partial c}, ~~~~~~
    r_3 := ~ \underbrace{R\Gamma}_{\partial a}\cdot \underbrace{\beta_2}_{bc}\\
    r_4 := &~ \lefteqn{\underbrace{\phantom{(Q_{S}+R\Gamma M_S)}}_{a^{\new}\cdot M_S}}(Q_{S}+R\Gamma\cdot \overbrace{M_S)(-(\Delta_{S,S}^{-1}+ M_{S,S} )^{-1}) \lefteqn{\underbrace{\phantom{\Big((M_S)^{\top}\cdot \big(\sqrt{W^{\appr}} h^{\appr}-\sqrt{V}g\big)+ \beta_{2,S}\Big)}}_{(M_S)^{\top}\cdot c^{\new}}}\Big((M_S)^{\top}}^{\partial b} \cdot \big(\sqrt{W^{\appr}} h^{\appr}-\sqrt{V}g\big)+ \beta_{2,S}\Big)
\end{align*}
}
Note that this output is
\begin{align*}
    r := ~ \underbrace{abc}_{r_1} + \underbrace{(a+\partial a)b\cdot \partial c}_{r_2} + \underbrace{\partial a \cdot bc}_{r_3} + \underbrace{a^{\new} \cdot \partial b \cdot c^{\new}}_{r_4}
    = ~ a^{\new} b^{\new} c^{\new}.
\end{align*}

Recall that the $n^{1+a}$ term stemmed from multiplying an $n\times n$ matrix by an $n^a$-sparse vector when computing $a^{\new}\cdot b\cdot \partial c$. 
We split this term as $a^{\new}\cdot b\cdot \partial c := a\cdot b\cdot \partial c+\partial a\cdot b\cdot \partial c$ (see the formula for $r_2$). The data structure will 
now maintain $ab := Q_S$, whose size is only $\sqrt{n}\times n^a$, hence computing $a b \cdot \partial c$ only takes $n^{1/2+a}<n^{1.5}$ time now. Note that 
when computing $\partial a\cdot b\cdot \partial c$, the $n\times n$ matrix $M$ is ``sandwiched'' by a $n^a$-sparse diagonal matrix $\Gamma$ on the left and a $n^a$-sparse vector $(\sqrt{W^{\appr}}h^{\appr}-\sqrt{V}g)$ on the right, thus computing the product $\Gamma M (\sqrt{W^{\appr}}h^{\appr}-\sqrt{V}g)$ takes $n^{2a}$ time.

The last $n^{1+a}$ bottleneck from \cite{cls19} is removed in a similar way, yielding the following intermediate result:  
\begin{theorem}[Informal, first improvement]\label{thm:tech_first_improvement}
For any $a \leq \alpha$, there is a randomized algorithm for solving general LPs in expected time 
$O^*( n^{\omega} + n^{2.5-a/2} + n^{0.5+a\omega} )$.
\end{theorem}
\noindent Note that for $\omega=2$ and $\alpha=1$, 
this approach already yields an improved $n^{2+1/10}$ LP algorithm. 

\subsubsection{Technique for removing \texorpdfstring{$n^{\omega a}$}{}}\label{sec:tech_second_improvement}

The $n^{\omega a}$ term in previous data structures came from computing the inverse of an $n^a\times n^a$ matrix $(\Delta_{S,S}^{-1}+M_{S,S})$, and this inverse is still present in the intermediate algorithm described in Section~\ref{sec:tech_first_improvement} (see the 
formula for $r_4$). This is where the \emph{cascading lazy updates} technique comes useful -- we shall remove the $n^{\omega a}$ by showing 
how to implement it for $K=2$ ``levels''. We now provide the details of this data structure.  

Once again, the main observation is that since we've already computed $(\Delta_{S,S}^{-1}+M_{S,S})^{-1}$ in previous iterations, we do not need to re-compute it from scratch. Instead, we only need to compute the difference between the new inverse matrix and the old one. More concretely, we maintain a second-level data structure member $\wt{v}$. $\wt{v}$ keeps a \emph{closer} distance with $w^{\appr}$ than $v$, and is therefore updated more frequently. We update $v$ whenever $\|\wt{v}-v\|_0>n^a$ (for some $a\leq \alpha$) and update $\wt{v}$ whenever $ \|w^{\appr} - \wt{v}\| > n^{\wt{a}}$ (for some $\wt{a}\leq \alpha a$). 
By abuse of notation, we define 
\begin{align}\label{eq:def_S_intro}
S:=\supp(\wt{v}-v),~~~ S^{\new}:=\supp(w^{\appr}-v),~~~ \partial S:=\supp(w^{\appr}-\wt{v}).
\end{align}
In this overview we can think of $S^{\new}=S\cup \partial S$. Later we also handle the possibly non-empty set $S'=(S\cup \partial S)\backslash S^{\new}$, but the key ideas are the same.
The updates guarantee that in the query we always have that $|S^{\new}|\leq n^a$ and $|\partial S|\leq n^{\wt{a}}$. Our data structure maintains a second-level member 
\begin{align*}
B:=(\Delta_{S,S}^{-1}+M_{S,S})^{-1} \in \R^{n^a \times n^a}.
\end{align*}
Observe that the new matrix $((\Delta^{\new})_{S^{\new},S^{\new}}^{-1}+M_{S^{\new},S^{\new}})$ only differs from $B^{-1}=(\Delta_{S,S}^{-1}+M_{S,S})$ on entries in $S^{\new}\times \partial S$ and $\partial S\times S^{\new}$. (See the left part of Figure~\ref{fig:Ldot_intro}.) 
\begin{figure}[!ht]
    \centering
    \includegraphics[width = \textwidth]{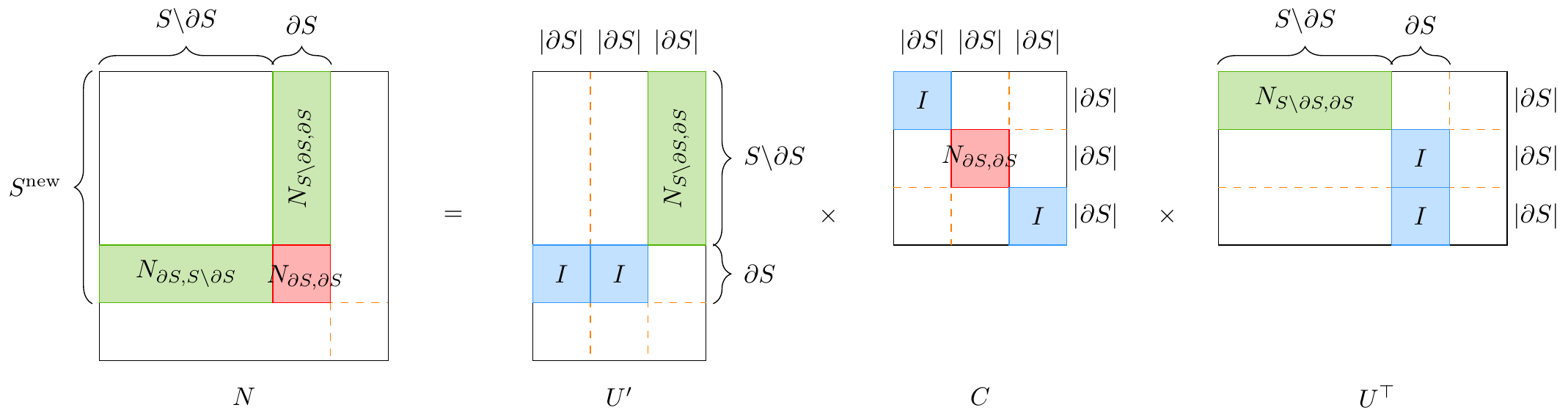}
    \caption{An illustration of the construction of $U', C, U^{\top}$. $N$ is defined as $((\Delta^{\new})_{S^{\new},S^{\new}}^{-1}+M_{S^{\new},S^{\new}}) - (\Delta_{S,S}^{-1}+M_{S,S})$. $I$ denotes the identity matrix.}
    \label{fig:Ldot_intro}
\end{figure}

So we have the following decomposition that can be computed in $O(n^{\wt{a}+a})$ time:
\begin{align*}
U'CU^{\top} = ((\Delta^{\new})_{S^{\new},S^{\new}}^{-1}+M_{S^{\new},S^{\new}}) - (\Delta_{S,S}^{-1}+M_{S,S}),
\end{align*}
where $U',U\in \R^{n^a\times n^{\wt{a}}}$, and $C\in \R^{n^{\wt{a}}\times n^{\wt{a}}}$. In fact $U'$ and $U$ are both constructed by taking a submatrix from $M$ and concatenate it with two identity matrices (see Figure~\ref{fig:Ldot_intro}), this will be useful in the next section where we remove the $n^{2a}$ term. 
Now we can use Woodbury identity to compute 
\begin{align*}
    ((\Delta^{\new})^{-1}_{S^{\new},S^{\new}}+M_{S^{\new}, S^{\new}})^{-1} = ~ (B^{-1}+U'CU^{\top})^{-1}
    = ~ B - BU'(C^{-1} + U^{\top}BU')^{-1}U^{\top}B.
\end{align*}
The most expensive part of this formula is to compute the multiplication $U^{\top}B$ and $BU'$, and the inverse $(C^{-1} + U^{\top}BU')^{-1}$. Since $\wt{a}\leq \alpha a$, using fast rectangular matrix multiplication, multiplying a $n^{\wt{a}}\times n^a$ matrix $U^{\top}$ with a $n^a \times n^a$ matrix $B$ takes $O(n^{2a})$ time. Computing $BU'$ takes the same time. Computing the inverse of a $n^{\wt{a}}\times n^{\wt{a}}$ matrix $(C^{-1} + U^{\top}BU')$ takes $n^{\wt{a}\omega}=\Tmat(n^{\wt{a}}, n^{\wt{a}}, n^{\wt{a}})\leq \Tmat(n^{a}, n^{a}, n^{\wt{a}})=n^{2a}$.
All other parts of the query algorithm remain the same as in Section~\ref{sec:tech_first_improvement}. Hence, so far, the running time is upper bounded by $O(n^{2a})$: 

\begin{theorem}[informal, second improvement]\label{thm:tech_second_improvement}
For any $a \leq \alpha$, there is a randomized algorithm for solving general LPs in expected time 
$O^*( n^{\omega} + n^{2.5-a/2} + n^{0.5+2a} ) $. 
\end{theorem}

We remark that the second-level members $B$ and the local variables $U$, $C$ and $U'$ that the data structure maintains, correspond to 
the variables in the cascading lazy update framework of  Section~\ref{sec:multi}: The matrix $B$ here is exactly the same inverse $B$ as defined in Section~\ref{sec:multi} for $K=2$ (see Eq.\eqref{eq:LU_decomposition_k_2}). In the same vein, the block $C^{-1}$ here corresponds to the term 
$-C_2^{-1}-V_2^{\top}M^{-1}U_2$, $U^{\top}$ corresponds to $V_2^{\top}M^{-1}U_1$, and $U'$ corresponds to $V_1^{\top}M^{-1}U_2$ of Eq.\eqref{eq:LU_decomposition_k_2}. That said, since Section~\ref{sec:multi} deals with a simplified version of our actual inverse problem, we will 
need to maintain several other ad-hoc members to achieve the claimed running time. We turn to describe those next.

\subsubsection{Technique for removing \texorpdfstring{$n^{2 a}$}{} }\label{sec:tech_third_improvement}
Now we show how to remove the $n^{2a}$ term from the current algorithm as described in Section~\ref{sec:tech_first_improvement}, with the inverse matrix of $r_4$ computed as described in Section~\ref{sec:tech_second_improvement}). The $n^{2a}$ term appears in the computation of both $r_2$ and $r_4$. Since removing the $n^{2a}$ terms in $r_4$ is more difficult, we mainly focus on the $r_4$ term in this section. In analogy to the second-level member $\wt{v}$, we also maintain $\wt{g}$, and update it whenever $\|h^{\appr}-\wt{g}\|_0>n^{\wt{a}}$. Similar to the definition of $S^{\new}, S, \partial S$ (Eq.~\ref{eq:def_S_intro}), we shall use the following three variants of notations:
\begin{align*}
    \Delta^{\new}=&~W^{\appr}-V, &~ \Delta = &~ \wt{V}-V, &~ \partial \Delta =&~ W^{\appr}-\wt{V},\\
    \Gamma^{\new}=&~\sqrt{W^{\appr}}-\sqrt{V}, &~ \Gamma = &~\sqrt{\wt{V}}-\sqrt{V}, &~ \partial \Gamma =&~ \sqrt{W^{\appr}}-\sqrt{\wt{V}},\\
    \xi^{\new} =&~ \sqrt{W^{\appr}}h^{\appr}-\sqrt{V}g, &~ \xi =&~ \sqrt{\wt{V}}\wt{g}-\sqrt{V}g, &~ \partial \xi = &~ \sqrt{W^{\appr}}h^{\appr}-\sqrt{\wt{V}}\wt{g}. 
\end{align*}
\begin{figure}[!ht]
    \centering
    \includegraphics[width = 0.98\textwidth]{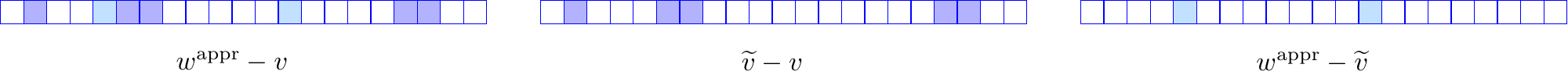}
    \caption{\small The three variants of notations. The left vector $w^{\appr}-v$ corresponds to the notations $X^{\new}$, and $|\supp(w^{\appr}-v)|\leq n^a$. The middle vector $\wt{v}-v$ corresponds to the notations $X$, and $|\supp(\wt{v}-v)|\leq n^a$. The right vector $w^{\appr}-\wt{v}$ corresponds to the notations $\partial X$, and $|\supp(w^{\appr}-\wt{v})|\leq n^{\wt{a}}$. }
    \label{fig:three_version_var}
\end{figure}

Intuitively, for $X\in \{S,\Delta,\Gamma,\xi\}$, $X^{\new}$ represents the difference between $w^{\appr}$ and the first-level proxy $v$, this is the real ``first derivative'', and it is what we need to compute the output; $X$ represents the difference between the first-level proxy $v$ and the second-level proxy $\wt{v}$, this is what we maintain in the data structure, and we can think of it as an out-of-date ``first derivative''; $\partial X$ represents the difference between $w^{\appr}$ and the second-level proxy $\wt{v}$. This term is very sparse, and should be thought of as the ``second derivative''. The actual ``first derivative'' is the sum of the outdated ``first derivative'' plus the ``second derivative''. We can therefore split $X^{\new}$ as  $X + \partial X$ when computing the output. In this way we exploits both the maintained members of the data structure and the \emph{sparsity} of $\partial X$, compares to computing $X^{\new}$ from scratch.

We now turn to formalize the above intuition. Our data structure maintains the second-level members 
$S \subseteq [n], \Delta \in \R^{n \times  n}, \Gamma \in \R^{n \times n},\xi \in \R^n$. By design, our update algorithm ensures that
\begin{align*}
    |S^{\new}|, \|\Delta^{\new}\|_0, \|\Gamma^{\new}\|_0, \|\xi^{\new}\|_0, |S|, \|\Delta\|_0, \|\Gamma\|_0, \|\xi\|_0 \leq n^a;~~~~
    |\partial S|, \|\partial \Delta\|_0, \|\partial \Gamma\|_0, \|\partial \xi\|_0 \leq n^{\wt{a}}.
\end{align*}
Now, recall that from Section~\ref{sec:tech_first_improvement} and Section~\ref{sec:tech_second_improvement}, $r_4$ is computed as follows
\begin{align*}
r_4= -\Big(Q_{S^{\new}}+R\underbrace{\Gamma^{\new}M_{S^{\new}}\Big)\cdot \Big(}_{3}\underbrace{BU'}_{2}(C^{-1}+U^{\top}\underbrace{BU'}_{2})^{-1}U^{\top}-I\Big)\cdot\underbrace{B \Big((\beta_2)_{S^{\new}} + (M_{S^{\new}})^{\top} \xi^{\new}\Big)}_{1},
\end{align*}
where $B=((\Delta_{S,S})^{-1}+M_{S,S})^{-1}$ and $U'CU^{\top}=((\Delta^{\new})^{-1}_{S^{\new}, S^{\new}} + M_{S^{\new}, S^{\new}})^{-1}-B$. Observe that in the above formula the $n^{2a}$ term appears in the following places:
\begin{enumerate}
    \item Multiplying a $n^a\times n^a$ matrix $B$ with a $n^a\times 1$ vector $(\beta_2)_{S^{\new}}+(M_{S^{\new}})^{\top}\xi^{\new}$.
    \item Multiplying a $n^a\times n^a$ matrix $B$ with a $n^a \times n^{\wt{a}}$ matrix $U'$.
    \item Multiplying a $n^a$-sparse diagonal matrix $\Gamma^{\new}$ with a $n\times n^a$ matrix $M_{S^{\new}}$ and then with a $n^a\times 1$ vector that comes from later terms. 
\end{enumerate}
To remove these $n^{2a}$ terms, we maintain additional second-level data structure members (precomputed matrix products): 
\begin{align*}
r_4= &~ -\Big(Q_{S^{\new}}+\underbrace{R\Gamma M_{S}}_{F\text{: member}} + R(\Gamma M_{\partial S\backslash S}+\partial \Gamma M_{S^{\new}})\Big)\cdot \Big((\underbrace{BU'}_{U^{\tmp}}(C^{-1} + U^{\top}\underbrace{BU'}_{U^{\tmp}})^{-1}U^{\top})-I\Big)\cdot \\
&~ \Big(\underbrace{B(\beta_2)_{S} + B(M_{S})^{\top} \cdot\xi}_{\gamma_1\text{: member}} + B(\beta_2)_{\partial S\backslash S} + B(M_{\partial S\backslash S})^{\top}\xi^{\new} + \underbrace{B(M_S)^{\top}}_{E\text{: member}}\partial \xi\Big)
\end{align*}
Each of the previous $n^{2a}$ terms are removed as follows.
\begin{enumerate}
    \item We maintain $\gamma_1:=B(\beta_2)_S+B(M_S)^{\top}\xi\in \R^{n^a}$ so that we only need to compute the difference
    \[
    (B(\beta_2)_{S^{\new}}+B(M_{S^{\new}})^{\top}\xi^{\new}) - \gamma_1 = B(\beta_2)_{\partial S\backslash S}+B(M_{\partial S\backslash S})^{\top}\xi^{\new} + B(M_S)^{\top}\partial \xi.
    \]
    Since $(\beta_2)_{\partial S\backslash S}$ is $n^{\wt{a}}$-sparse, and $(M_{\partial S\backslash S})^{\top}$ only has $n^{\wt{a}}$ non-zero rows, the first two terms $B(\beta_2)_{\partial S\backslash S}$ and $B(M_{\partial S\backslash S})^{\top}\xi^{\new}$ can both be computed in $O(n^{a+\wt{a}})$ time.
    For the third term, we also maintain $E:=B(M_S)^{\top}\in \R^{n^a\times n^a}$. Since $\partial \xi$ is also $n^{\wt{a}}$-sparse, it takes $O(n^{a+\wt{a}})$ time to compute this term as well. 
    \item The construction of $U'$ has the following property:
    $
    BU'=[B_{(\partial S\backslash S)}, B_{\partial S}, BM_{S,{(\partial S\backslash S)}}].
    $
    Since we already maintain $E:=B(M_S)^{\top}\in \R^{n^a\times n}$, we define a local variable $U^{\tmp}:=[B_{(\partial S\backslash S)}, B_{\partial S}, E_{(\partial S\backslash S)}] \in \R^{n^a \times 3n^{\wt{a}}}$. Then $U^{\tmp}=BU'$, and it can be computed in the same time as its size, which is $O(n^{a+\wt{a}})$.
    \item We maintain $F:=R\Gamma M_S\in \R^{\sqrt{n}\times n^a}$, and the difference is 
    $
    R\Gamma^{\new}M_{S^{\new}}-F=R(\Gamma M_{\partial S\backslash S}+\partial \Gamma M_{S^{\new}}).
    $
    Multiplying $F$ with a $n^a\times 1$ vector that comes from later terms only takes $n^{1/2+a}<n^{1.5}$ time. Also since $M_{\partial S\backslash S}$ only has $n^{\wt{a}}$ non-empty columnns, and $\partial \Gamma$ is $n^{\wt{a}}$-sparse, multiplying $(\Gamma M_{\partial S\backslash S}+\partial \Gamma M_{S^{\new}})$ with the $n^a\times 1$ vector that comes from later terms takes $n^{a+\wt{a}}$ time.
\end{enumerate}

A full list of all the second-level members that we maintain is as follows: 
\begin{align*}
   & 1. S=\supp(\wt{v}-v),
 & &2. \Delta = \wt{V}-V,
& & 3. \Gamma = \sqrt{\wt{V}}-\sqrt{V},
& & 4. \xi = \sqrt{\wt{V}}\wt{g}-\sqrt{V}g, \\
& 5. B = (\Delta_{S,S}^{-1}+M_{S,S})^{-1},
 & & 6. E = B(M_S)^{\top},
& & 7. F = R\Gamma M_S,
& & 8. \gamma_1 = B(\beta_2)_S+B(M_S)^{\top}\xi.
\end{align*}

Now whenever our data structure needs to multiply an $n^a\times n^a$ matrix with a $n^a\times 1$ vector, it is always the case that either the vector is $n^{\wt{a}}$-sparse, or the matrix only has $n^{\wt{a}}$ rows, so in both cases this operation takes $O(n^{a+\wt{a}})$ time. We also avoid multiplying the $n^a\times n^a$ matrix $B$ with the $n^{\wt{a}}\times n^a$ matrix $U'$ directly by maintaining $E=B(M_S)^{\top}$, and we can now extract $U^{\tmp}=BU'$ from $E$ efficiently. But still we need to multiply a $n^{\wt{a}}\times n^a$ matrix $U^{\top}$ with a $n^a\times n^{\wt{a}}$ matrix $U^{\tmp}$, which takes time $\Tmat(n^{\wt{a}}, n^a, n^{\wt{a}})\leq n^{a-\wt{a}}\cdot \Tmat(n^{\wt{a}}, n^{\wt{a}}, n^{\wt{a}})= n^{a+(\omega-1)\wt{a}}$. This is our final running time for query, as presented in the last line of Table~\ref{tab:running_time}.

Finally we remark that removing the last $n^{2a}$ term that stems from computing $r_2$, 
can be done in an analogous way to the usage of $\gamma_1$ for $r_4$ (i.e., by maintaining an additional member $\gamma_2 := \Gamma M\xi$). 
We omit the formal details here. 

Thus we finished the proof of how to get the $O(n^{a+(\omega-1)\wt{a}})$ query time of Theorem~\ref{thm:tech_third_improvement}.

\subsection{Our Update Algorithm} \label{sec:our_technique_update}

Bounding the running time of our two-level update scheme requires both algorithmic modifications and a more 
sophisticated amortized analysis than that of \cite{cls19}, in order to capture the cascading lazy updates process 
(as a random process under the sketching of the CP).
This section is organized as follows. In Section~\ref{sec:intro_cascading_updates} we describe the four update subroutines 
for maintaining the two-level members for both $w$ and $h$, and present their worst-case running time per call in Table~\ref{tab:when_to_update_what}. These subroutines correspond to maintaining the LU-decomposition 
of the cascading updates algorithm for a $K=3$ block matrix, described in Section~\ref{sec:multi} (see Figure~\ref{fig:intro_chase}).
Section~\ref{sec:intro_two_level_soft_threshold_and_adjust} describes the main algorithm deciding when to call each of these four 
subroutines, using ``two-level soft thresholding'' (see Figure~\ref{fig:double_threshold}). 
In order to synchronize two levels of soft thresholding, we introduce a new \textsc{Adjust} function.
Finally, Section~\ref{sec:intro_amortized_analysis} describes a potential-based amortized analysis of our update algorithm, and the final amortized running time of the four update subroutines is presented in Table~\ref{tab:four_potential_function}.

\subsubsection{Cascading updates subroutines}
\label{sec:intro_cascading_updates}

We now describe the four subroutines required to efficiently implement the cascading lazy updates process for $K=3$ levels 
as described in the previous section. Recall our data structure maintains two levels of proxies $v$ and $\wt{v}$ for 
the input $w$: 
$v$ is the first-level proxy, and $\wt{v}$ is the second-level proxy. $v$ keeps a larger 
distance of $n^a$ with $w$ and is updated less frequently, while $\wt{v}$ keeps a smaller distance of $n^{\wt{a}}$ with $w$ and is 
updated more frequently. We define two subroutines to update $v$ and $\wt{v}$ (see Figure~\ref{fig:leash} for illustration).

\begin{table}[!ht]
\small
    \centering
    \begin{tabular}{|l|l|l|l|l|} \hline
        {\bf Level } & {\bf Name} & {\bf Time per call} & {\bf Rank/Sparsity} & {\bf Comment} \\ \hline
        1 & Matrix & $\Tmat(n,n,k)$ & $k:=\|v^{\new}-v\|_0$ & Update $v$ and $\wt{v}$ if $\| w^{\appr} - v \|_0 \geq n^a$ \\ \hline
        2 & PartialMatrix & $\Tmat(n,n^a,\wt{k})$ &  $\wt{k}:=\|\wt{v}^{\new}-\wt{v}\|_0$ & Update $\wt{v}$ if $\| w^{\appr} - \wt{v} \|_0 \geq n^{\wt{a}}$ \\ \hline
        1 & Vector & $pn+n^{2a}$ & $p:=\|g^{\new}-g\|_0$ & Update $g$ and $\wt{g}$ if $\| h^{\appr} - g \|_0 \geq n^a$ \\ \hline
        2 & PartialVector & $\wt{p}n+n^{2a}$ & $\wt{p}:=\|\wt{g}^{\new}-\wt{g}\|_0$ & Update $\wt{g}$ if $\| h^{\appr} - \wt{g} \|_0 \geq n^{\wt{a}}$ \\ \hline
    \end{tabular}
    \caption{\small Four update procedures}
    \label{tab:when_to_update_what}
\end{table}

The first subroutine is \textsc{PartialMatrixUpdate}, which corresponds to second-level updates we alluded to in Section~\ref{sec:multi}: so long as the rank of the updates is smaller than the second level threshold $n^{\wt{a}}$, it suffices to only update the second level members as defined in Section~\ref{sec:tech_third_improvement}. These members are relatively cheap to update, and thus \textsc{PartialMatrixUpdate} has a cost of $\Tmat(n, n^a, \wt{k})$ per call (third column of Table~\ref{tab:when_to_update_what}). 
When the algorithm has exceeded the allowable changes in $w$, we cascade to the first level update subroutine 
\textsc{MatrixUpdate}. This subroutine must update all data structure members, and consequently is more expensive: 
it has has a cost of $\Tmat(n, n, k)$ per call (third column of Table~\ref{tab:when_to_update_what}).
The subroutines \textsc{VectorUpdate} and \textsc{PartialVectorUpdate} play an analogous role for updating $h$. 
We proceed to describe when to execute each of these subroutines. 

\begin{figure}[!ht]
    \centering
    \includegraphics[width =0.8 \textwidth]{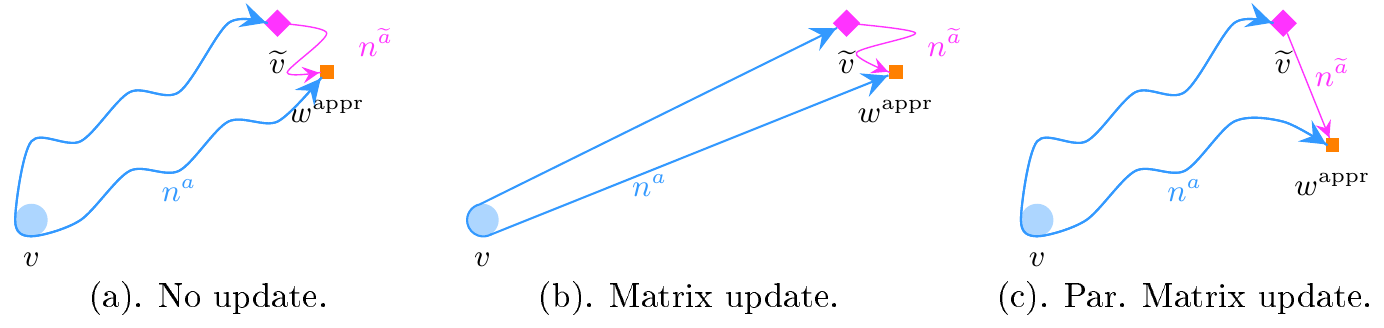}
    \caption{\small An illustration of the two types of updates in our \textsc{Update} algorithm. The blue leash connects $v$ with $\wt{v}$ and $w^{\appr}$ and is of length $n^a$. The pink leash connects $\wt{v}$ with $w^{\appr}$ and  is of length $n^{\wt{a}}$. 
    (a) When all leashes are loose, no update is performed.  (b) The leash on $v$ is tight. In this case we call \textsc{MatrixUpdate}. 
    (c) The leash on $\wt{v}$ is tight. In this case we call \textsc{PartialMatrixUpdate}.} 
    \label{fig:leash}
\end{figure}

\subsubsection{Synchronizing two-level soft thresholding}
\label{sec:intro_two_level_soft_threshold_and_adjust}
For simplicity, we only focus on \textsc{MatrixUpdate} and \textsc{PartialMatrixUpdate} (\textsc{VectorUpdate} and \textsc{PartialVectorUpdate} are analogous). 

In order to bound the amortized cost of the two level updates, we use the ``soft thresholding'' implicit\footnote{The soft thresholding is proposed by \cite{cls19}, and note that they embed it inside their update function since they only have one update function. Since we use it four times, we give it an explicit name $\soft$.} in \cite{cls19} to determine when to invoke each update. The basic idea is to use a smooth threshold (see Algorithm~\ref{alg:binary_search_improved}) to approximate the discrete threshold shown in the last column of Table~\ref{tab:when_to_update_what}. Soft thresholding ensures that there is a gap between errors of coordinates 
that are updated and the errors of other coordinates, and is essential to guarantee a proper decrease in potential. 
Our update algorithm invokes this subroutine {\soft} twice -- once for $\wt{v}$ and once for $v$ (see Figure~\ref{fig:double_threshold}). We now explain the 3 main components of combining the two {\soft} subroutines (see Algorithm~\ref{alg:intro_update_v}): 
\begin{algorithm}
\small
\caption{When to execute \textsc{MatrixUpdate} and \textsc{PartialMatrixUpdate}}
\label{alg:intro_update_v}
\begin{algorithmic}[1]
    \State $\wt{v}^{\new} \leftarrow \soft(y \leftarrow | w^{\new} / \wt{v}  - 1 |, w^{\new},\wt{v} ,\epsilon_{\mathrm{mp}}, n^{\wt{a}}) $
    \State $\textsc{Adjust}(\wt{v}^{\new}, \epsilon_{\mathrm{mp}}/(100\log n))$
    \If{$\|\wt{v}^{\new}-v\|_0\geq n^a$}
        \State $v^{\new} \leftarrow \soft(y \leftarrow | w^{\new} / v  - 1  | + | w^{\new}/ \wt{v}  - 1 |, w^{\new}, v ,  \frac{\epsilon_{\mathrm{mp}}}{100\log^2n}, n^a)$
        \State $w^{\appr}\leftarrow v^{\new}$
        \State \textsc{MatrixUpdate}()
    \Else
        \State $ w^{\appr}\leftarrow \wt{v}^{\new}$
        \If{$\|\wt{v}^{\new} - \wt{v} \|_0 \geq n^{\wt{a}}$}
            \State \textsc{PartialMatrixUpdate}()
        \EndIf
    \EndIf
\end{algorithmic}
\end{algorithm}

\paragraph{Restoring threshold gaps via \textsc{Adjust}.} Our algorithm first computes $\wt{v}^{\new}$ to update all the coordinates $i$ for which $|w^{\new}_i / \wt{v}_i - 1|$ exceeds the threshold $\epsilon_{\mathrm{mp}}$. It then calls the \textsc{Adjust} function to restore all updated coordinates $\wt{v}^{\new}_i$ whose new value $\wt{v}^{\new}_i$ is within a distance of $\epsilon_{\mathrm{mp}}/(100 \log n)$ from $v_i$, back to the original value $v_i$. Since $\wt{v}^{\new}$ is the new value of $\wt{v}$, in this way we ensure that $|\wt{v}_i - v_i| > \epsilon_{\mathrm{mp}} / (100 \log n)$ for all $i \in \supp(\wt{v} - v)$.
Hence when $\|\wt{v} - v\|_0$ exceeds its threshold, there is a large enough decrease in 
our potential function ($\approx \|\wt{v} - v\|_0 \cdot \epsilon_{\mathrm{mp}} / (100\log n)$) for 
``charging'' the update cost (of $v \leftarrow \wt{v}$).  
The purpose of using a smaller threshold-error of $\epsilon_{\mathrm{mp}} / (100 \log n)$ here ensures 
that even when $\wt{v}^{\new}_i$ (which is the new value of $\wt{v}$) is $v_i$ instead of $w^{\new}_i$, the error still decreases by at least a $(1-1/\log n)$ factor after updating $\wt{v} \leftarrow \wt{v}^{\new}$. 
\paragraph{Synchronizing the error function.} When updating $v$, we define the error as $|w^{\new} / v - 1| + |w^{\new} / \wt{v} - 1|$ which is a function of both 
$v$ and $\wt{v}$. This is because as long as \emph{one} of $v_i$ and $\wt{v}_i$ is too far from $w^{\new}_i$, we 
need to update both variables to be the same as $w^{\new}_i$ .

\paragraph{Two error thresholds.} 
When updating $v$, we use a smaller error threshold of $\epsilon_{\mathrm{mp}} / (100 \log^2 n)$ than that of the 
\textsc{Adjust} threshold. This is because \textsc{Adjust} guarantees that $v_i- \wt{v}_i\geq \epsilon_{\mathrm{mp}} / (100 \log n)$ 
on all coordinates $i$ for which $v_i\neq \wt{v}_i$, hence using an even smaller threshold when updating $v$ ensures 
that all such coordinates are counted as ``error larger than threshold''. As such, after \textsc{MatrixUpdate}, $\wt{v}$ is set back 
to be the same as $v$ on all coordinates. 

\begin{figure}[!ht]
    \centering
    \includegraphics[width=1\textwidth]{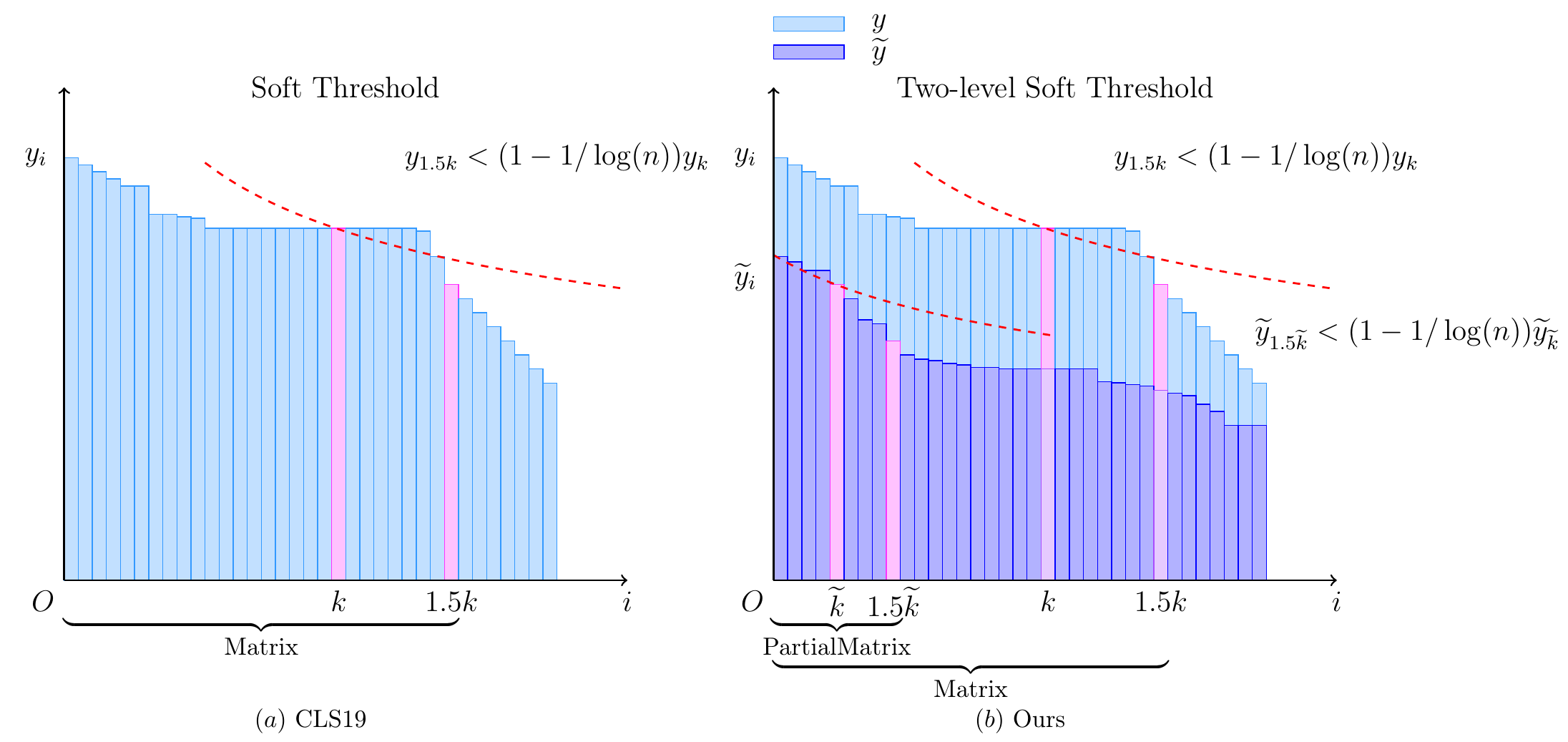}
    \caption{\small  Two-level vs. one-level soft thresholding: (a). \cite{cls19} has a single level update scheme, implemented via one level of soft thresholding. (b). Our Update algorithm has two cascading subroutines (\textsc{MatrixUpdate}, \textsc{PartialMatrixUpdate}), each requires its own potential function and soft threshold.}
    \label{fig:double_threshold}
\end{figure}

\subsubsection{Amortized analysis based on high-order martingales}
\label{sec:intro_amortized_analysis}

As noted in Section~\ref{sec:multi}, the main source of amortization in the update algorithm comes from 
the fact that the maintained variables in all levels are changing slowly. More formally, 
the $j$-th iteration of the CP algorithm calls our data structure with the following inputs: 
A vector $w^{(j+1)}\in \R^n$ and a vector $h^{(j+1)}\in \R^n$ satisfying the following relative error bounds (where the randomness is over the sketching matrix used to generate $w^{(j+1)}$ and $h^{(j+1)}$):
\begin{align} \label{eq:intro:w_constraint}
\|\E[w^{(j+1)}| w^{(j)}] / w^{(j)} - 1\|_2 ~\leq O(1), ~~~ \|\E[h^{(j+1)}| h^{(j)}]/h^{(j)} - 1\|_2 \leq O(1).
\end{align}

The ``evolution'' of $w^{(j)}$ and $h^{(j)}$ is essentially a martingale process with 
the above guarantee. Informally, this is true because we are using \emph{independent} random sketches 
in each iteration. We analyze these random processes by defining four different 
potential functions to capture our four different update subroutines, as shown in Table~\ref{tab:four_potential_function}. 
We next elaborate on the careful design of potentials and what they aim to measure.

\begin{table}[!t]
    \footnotesize
    \centering
    \begin{tabular}{|l|l|l|l|l|}\hline
        {\bf Level} & {\bf Name} & $\Phi \in \R$ & $\phi \in \R^n$,  
        $i> n^{a}$ (resp. $i > n^{\wt{a}}$) & {\bf Amortized} \\ \hline
        1 & Matrix & $\sum_{i=1}^n \phi_i\cdot( |w_{i}/v_{i}-1| + |w_{i}/\wt{v}_{i}-1|)$ & $\phi_i = i^{\frac{\omega-2}{1-a} - 1} \cdot n^{- \frac{a (\omega - 2)}{1 - a}}$
        & $n^{2-a/2}+n^{\omega-1/2}$  \\ \hline
        2 & PartialMatrix & $ \sum_{i=1}^n \phi_i\cdot\left( |w_{i}/\wt{v}_{i}-1|\right)$ &
        $\phi_i = i^{\frac{a (\omega-2)}{a - \wt{a}} - 1} \cdot n^{- \frac{a \wt{a} (\omega - 2)}{a - \wt{a}}}$
        & $ n^{1+a-\wt{a}/2}$\\ \hline
        1 & Vector & $\sum_{i=1}^n \phi_i\cdot\left( |h_{i}/g_{i}-1| + |h_{i}/\wt{g}_{i}-1|\right)$ & $\phi_i=1$ & $n^{1.5}$ \\ \hline
        2 & PartialVector & 
        $\sum_{i=1}^n \phi_i\cdot ( |h_{i}/\wt{g}_{i}-1| )$ &  
        $\phi_i = i^{-1}$
        & $n^{1.5} + n^{2a-\wt{a}/2} $ \\ \hline
    \end{tabular}
    \caption{\small The four different potential functions $\Phi$ are shown in the third column. We always assume that the coordinates are sorted in decreasing order, e.g., for \textsc{PartialMatrix}, we assume that $|w_i / \wt{v}_i - 1| \geq |w_{i+1} / \wt{v}_{i+1} - 1|$. The vector $\phi \in \R^n$ is non-increasing (in $i$): for level-1 subroutines, $\phi_i := n^{-a}$ when 
    $i\leq n^a$, and for level-2 subroutines, $\phi_i := n^{-\wt{a}}$ for $i\leq n^{\wt{a}}$. The definition of $\phi_i$ for 
    $i>n^a$ (level 1) or $n^{\wt{a}}$ (level 2) are shown in the fourth column. The vectors $\phi$ are designed to upper bound the worst-case running time of the update procedures. 
    The last column shows the amortized cost of our four update subroutines.}
    \label{tab:four_potential_function}
\end{table}

\textbf{1. \textsc{MatrixUpdate}.} The potential function for this subroutine is defined in row 1 of Table~\ref{tab:four_potential_function}. Instead of splitting the change in the potential function into ``$w$ move'' and ``$v$ move'' as \cite{cls19}, 
we split the change into a ``$w$ move'' part and a \emph{``$v$ and $\wt{v}$ move''} part:
\begin{align*}
    \Phi^{\mathrm{mat}}_{j+1}-\Phi^{\mathrm{mat}}_{j} = ~ \text{($w$~move)} - \text{($v$~and~$\wt{v}$~move)}
\end{align*}
where the ``$w$ move'' part measures how much the potential can increase due to the input changes from $w^{(j)}$ to $w^{(j+1)}$, and the ``$v$ and $\wt{v}$ move'' part measures how much we can decrease the potential by updating $v^{(j)}, \wt{v}^{(j)}$ to $v^{(j+1)}, \wt{v}^{(j+1)}$. The formal details can be found in Section~\ref{sec:amortize_time_improved}.

Using Eq.~\eqref{eq:intro:w_constraint} which upper bounds the expected relative error of $w^{(j+1)}$ with $w^{(j)}$, it is possible to upper bound the ``$w$ move'' term by $O(\|\phi\|_2)=O(n^{\omega-5/2}+n^{-a/2})$.  

When entering 
\textsc{MatrixUpdate}, for some coordinate $i\in [k]$ (recall that $k :=\|v^{\new}-v\|_0$, see Table~\ref{tab:when_to_update_what}), by design $v^{(j+1)}_i$ and $\wt{v}^{(j+1)}_i$ are both reset 
to $w^{(j+1)}_i$, hence the potential $\phi_i(|{w^{(j+1)}_{i}} / {v^{(j+1)}_{i}}-1 | + |{w^{(j+1)}_{i}}/{\wt{v}^{(j+1)}_{i}}-1|)$ 
decreases to 0. This fact can be used to show that the ``$v$ move'' term in $\Phi$  decreases by at least
\begin{align*}
\Omega(k\cdot \phi_k) \geq  n^{-2}\cdot \Tmat(n,n,k). 
\end{align*}
We also prove that \textsc{PartialMatrixUpdate} can only further decrease the ``$v$ move'' term.
Since the ``$v$ and $\wt{v}$ move'' term is upper bounded by the ``$w$ move'' term, and the cost per call of \textsc{MatrixUpdate} is $\Tmat(n,n,k)$, the amortized cost of \textsc{MatrixUpdate} per iteration is bounded by 
$
O(n^{\omega-1/2+o(1)}+n^{2-a/2+o(1)}).
$

%%%%%%%%%%%%%%%%%%%%%%%%%%%%%%%%%%%%%% 
\

\textbf{2. \textsc{PartialMatrixUpdate}.}  The potential function for this subroutine is defined in row 2 of  Table~\ref{tab:four_potential_function}. Once again, we 
split the change in the potential function, this time into a ``$w$ move'' part and a ``$\wt{v}$ move'' part:
\begin{align*}
\Phi_{j+1}-\Phi_j = \text{($w$~move)} - \text{($\wt{v}$~move)}
\end{align*}
The ``$w$ move'' term is upper bounded by $O(\|\phi\|_2)=O(n^{a\omega-5\wt{a}/2}+n^{-\wt{a}/2})$. To lower bound the ``$\wt{v}$ move'' term, we observe that when entering \textsc{PartialMatrixUpdate}, for any $i\in [\wt{k}]$ (recall that $\wt{k} :=\|\wt{v}^{\new}-\wt{v}\|_0$, see Table~\ref{tab:when_to_update_what}), the term  $|{w^{\new}_i}/{\wt{v}^{\new}_i}-1|$ is decreased by at least a factor of $(1-1/\log n)$. This is where we use the guarantees (and smaller threshold parameter) 
of \textsc{Adjust} and {\soft} for $v$. Using this, we show that the ``$\wt{v}$ move'' term decreases by at least 
\[ \Omega(\wt{k}\cdot \phi_{\wt{k}}) \geq n^{-1-a}\cdot \Tmat(n, n^a, \wt{k}). \]
Therefore, the amortized cost of \textsc{PartialMatrixUpdate} is 
$
O(n^{1+(\omega+1)a-5\wt{a}/2}+n^{1+a-\wt{a}/2})=O(n^{1+a-\wt{a}/2})
$
since the cost per call of \textsc{PartialMatrixUpdate} is $\Tmat(n,n^a,\wt{k})$.

\

\textbf{3. \textsc{VectorUpdate}.} The potential function for this subroutine is defined in row 3 of  Table~\ref{tab:four_potential_function}.  Note that the dominating term in the worst-case cost (per call) of 
this subroutine is the $pn$ term (recall $p:= \|g^{\new}-g\|_0$, see Table~\ref{tab:when_to_update_what}). Again, after amortization in each iteration we have $$\sqrt{n} = \|\phi\|_2 \geq \Omega(p\cdot \phi_{p}) \geq n^{-1}(pn).$$ 
Therefore, the amortized cost of \textsc{VectorUpdate} per iteration is $O(n^{1.5}).$ 

\

\textbf{4. \textsc{PartialVectorUpdate}.}  
The potential function for this subroutine is defined in row 4 of  Table~\ref{tab:four_potential_function}. 
Here, the dominating term in the worst-case cost (per call) is the $n^{2a}$ term (Table~\ref{tab:when_to_update_what}). 
Since \textsc{PartialVectorUpdate} is invoked only when $\wt{p}\geq n^{\wt{a}}$ (recall that $\wt{p} :=\|\wt{g}^{\new}-\wt{g}\|_0$, see Table~\ref{tab:when_to_update_what}), the amortized cost of the $j$-th iteration is $n^{2a}\cdot\mathbf{1}_{\wt{p}_j>n^{\wt{a}}}$. Once again, after amortization in each iteration we have 
$$n^{-\wt{a}} = \|\phi\|_2 \geq \Omega(p\cdot \phi_{p}) \geq \mathbf{1}_{\wt{p}_j>n^{\wt{a}}}.$$ 
Thus the amortized cost of \textsc{VectorUpdate} per iteration is
$
n^{2a}\cdot O(\|\phi\|_2)=O(n^{2a-\wt{a}/2}).
$

\subsection{Putting it all together}

Combining the query time $t_q = n^{a+(\omega - 1) \wt{a}}$ (Eq.~\eqref{eq_improved_query_time} in Section~\ref{sec:our_technique_query}) and the update time $t_u = n^{\omega - 1/2} + n^{2- a/2} + n^{1 + a - \wt{a}/2}$ (see the last column of Table~\ref{tab:four_potential_function}), and since there are $O(\sqrt{n})$ iterations in total, we have that the final running time of our LP algorithm is 
\begin{align*}
   \sqrt{n} \cdot( t_q + t_u ) = n^{0.5 + a+(\omega - 1) \wt{a}} + n^{\omega} + n^{2.5- a/2} + n^{1.5 + a - \wt{a}/2},
\end{align*}
matching the statement of Theorem~\ref{thm:tech_third_improvement}.

Also note that in the ideal case where $\omega=2$ and $\alpha=1$, we can choose $a=\frac{8}{9}$ and $\wt{a}=\frac{2}{3}$, and this leads to an $O(n^{2+1/18})$ algorithm.

\newpage

\addcontentsline{toc}{section}{References}
\bibliographystyle{alpha}
\bibliography{ref}
\newpage

{
\hypersetup{linkcolor=black}

\tableofcontents

}

\newpage

\appendix
\section*{Appendix}

\section{Preliminaries}
\label{sec:preliminary}

Throughout this paper when considering the linear program $\min_{Ax=b,x\geq 0}c^{\top}x$, we always assume that the input matrix $A\in \R^{d \times n}$ is full-rank, and $d=\Omega(n)$, $d\leq n$. For convenience we also assume that $n\geq 10$. In our algorithm we use the standard transformation of LP by \cite{ytm94} such that it is easy to obtain the initial $x$ and $s$ for the transformed LP. We refer the readers to \cite{ytm94} and \cite{cls19} for more details.

\paragraph{Standard notations.}
For a positive integer $n$, we denote $[n]=\{1,2,\cdots,n\}$.

For a positive integer $n$, we use $I_n$ to denote the identity matrix of size $n\times n$. We use standard definitions of hyperbolic functions $\sinh(x)=\frac{e^x-e^{-x}}{2}$, $\cosh(x)=\frac{e^x+e^{-x}}{2}$.

For a vector $v \in \R^n$, we use the standard definition of $\ell_p$ norms: $\forall p\geq 1$, $\|v\|_p = (\sum_{i=1}^n |v_i|^p )^{1/p}$. Specially, $\|v\|_\infty = \max_{ i \in [n] } | v_i |$.

A vector $v \in \R^n$ is called $k-$sparse if at most $k$ entries in $v$ is non-zero.

For any random variable $x$, we use $\E[x]$ to denote its expectation, and we use $\Var[x]$ to denote its variance. We define $\Sup[x]$ to be the deterministic maximum of $x$.

\paragraph{Approximations.}
For any function $f$, we define $\wt{O}(f)=f\cdot \poly \log (f)$, and $O^*(f)=f\cdot f^{o(1)}$.

For vectors $a,b\in \R^n$ and accuracy parameter $\epsilon\in(0,1)$, we use $a\approx_{\epsilon} b$ to denote that $(1-\epsilon)b_i\leq a_i\leq (1+\epsilon)b_i$, for any $i\in [n]$. For constant $t$, we use $a \approx_{\epsilon} t$ to denote that $(1-\epsilon) t \leq a_i \leq (1+\epsilon) t$, for any $i \in [n]$.

\paragraph{Coordinate-wise operations.}
For a vector $x \in \R^n$ and $s \in \R^n$, we denote $x s \in \R^n$ as a length-$n$ vector whose $i$-th coordinate is $(x s)_i=x_i s_i$, $\forall i \in [n]$. 
%We denote $x/s \in \R^n$ as a length-$n$ vector with the $i$-th coordinate $(x/s)_i$ is $x_i/s_i$, $\forall i \in [n]$. 
Similarly, we also define other scalar operations on vectors as coordinate-wise operations.

For a scalar function $f : \R \rightarrow \R$ and a vector $x\in \R^n$, we define $f(x)=[f(x_1),f(x_2),\dots,f(x_n)]^{\top}$.

\paragraph{Upper case as diagonal matrix.}
Given vectors $x,s \in \R^n$, we use $X \in \R^{n \times n}$ and $S \in \R^{n \times n}$ to denote the diagonal matrix of those two vectors. We use $X/S$ to denote the diagonal matrix there the $i$-th entry on the diagonal is $(X/S)_{i,i}=x_i/s_i$, $\forall i \in [n]$. Similarly, we extend other scalar operations to diagonal matrix. 
%Note that matrix $\sqrt{\frac{X}{S}}A^\top (A \frac{X}{S} A^\top)^{-1}A\sqrt{\frac{X}{S}}$ is an orthogonal projection matrix.

\paragraph{Matrix and vectors with subscripts.}
For any matrix $M\in \R^{m\times n}$ where $m>1$ and $n>1$, and any subset $S\subseteq [n]$, we define $M_{S} \in \R^{m \times |S|}$ to be the submatrix of $M$ that only has columns in $S$. 

For any subsets $S_1\subseteq [m], S_2\subseteq [n]$, we also define $M_{S_1,S_2} \in \R^{|S_1|\times |S_2|}$ to be the submatrix of $M$ that only has rows in $S_1$ and columns in $S_2$.

For any vector $v\in \R^{n\times 1}$ where $n>1$, and any subset $S\subseteq [n]$, we define $v_{S} \in \R^{|S| \times 1}$ to be the subvector of $M$ that only has the entries in $S$. 

\paragraph{Fast Matrix Multiplication.}
We use $\omega$ to denote the exponent of matrix multiplication, which is defined as the infimum number such that the multiplication of two $n \times n$ matrices can be done in $O(n^{\omega})$ time. We use $\alpha$ to denote the the dual exponent of matrix multiplication, which is defined as the supremum over all $a\geq 0$ such that the multiplication of a $n\times n^a$ matrix with another $n^a \times n$ matrix can be done in $O(n^{2+o(1)})$ time.

We denote $\Tmat(n,k,r)$ to be the time needed to multiply a $n\times k$ matrix with a $k\times r$ matrix. $\Tmat(n,k,r)$ has the following property.
\begin{lemma}[\cite{s91,gu18,cglz20}]\label{lem:equivent_of_matrix_multiplcation}
\begin{align*}
 \Tmat(n,k,r) = O(\Tmat(n,r,k))
    =  O(\Tmat(k,n,r)) .
\end{align*}
\end{lemma}

\paragraph{Woodbury identity.}
We use the Woodbury matrix identity to calculate the inverse of a matrix $M$ under low-rank updates. %$( M + U C V )^{-1}$, where $C$ is a full-rank square matrix.
\begin{fact}[Woodbury matrix identity, \cite{w49,w50}]\label{fact:woodbury}
For matrices $M\in \R^{n\times n}$, $U\in \R^{n\times k}$, $C\in \R^{k\times k}$, $V\in \R^{k\times n}$,
    \begin{align*}
        ( M + U C V )^{-1} = M^{-1} - M^{-1} U (C^{-1} + V M^{-1} U )^{-1} V M^{-1} .
    \end{align*}
\end{fact}

\begin{comment}
\paragraph{Central path equation.}
Given the input matrix $A$, primal variable $x$, dual variable $y$, slack variable $s$, the central path equation for LP is the following:
\begin{align}\label{eq:central_path_equation}
    X \delta_s + S \delta_x & = \delta_{\mu}, \notag \\
    A \delta_x & = 0,\notag \\
    A^{\top} \delta_y + \delta_s &=0.
\end{align}
The solution for $\delta_x$, $\delta_s$ are explicitly given as
\begin{align*}
    \delta_x & = \frac{X}{\sqrt{XS}}(I-P)\frac{1}{\sqrt{XS}} \delta_{\mu}\\
    \delta_s & = \frac{S}{\sqrt{XS}}P\frac{1}{\sqrt{XS}} \delta_{\mu},
\end{align*}
where $P=\sqrt{\frac{X}{S}}A^{\top}(A\frac{X}{S}A^{\top})^{-1}A\sqrt{\frac{X}{S}}$ is an orthogonal projection matrix. See more details in \cite{cls19}.
\end{comment}

\paragraph{Probability tools.}
We use the following well-known probability tools.
\begin{lemma}[Bernstein Inequality, \cite{b24}]\label{lem:bernstein_inequality}
Let $X_1,\cdots,X_n$ be independent zero-mean random variables. Suppose that $|X_i|\leq M$ almost surely, for all $i$. Then, for all $t>0$,
\begin{align*}
    \Pr\left[\sum_{i=1}^n X_i>t\right]\leq \exp\left(-\frac{t^2/2}{\sum_{j=1}^n\E[X_j^2]+Mt/3}\right).
\end{align*}
\end{lemma}

\paragraph{Central Path Method.}
The stochastic central path method of \cite{cls19} consider the following LP with $A\in \R^{d \times n}$: $\min_{Ax = b, x \geq 0} c^{\top} x$, and its dual: $\max_{A^{\top}y \leq c} b^{\top} y$. By defining $s\in \R^{n}$ as the slack variables, the optimal solution of the above LP must satisfy the following constraints:
\begin{align*}
    x_i s_i = &~ 0, ~ \forall i, \\
    A x = &~ b, \\
    A^{\top} y + s = &~ c, \\
    x_i, s_i \geq &~ 0, ~ \forall i.
\end{align*}
where $\sum_{i=1}^n x_i s_i$ is the duality gap, and the other three equations are the feasibility constraints. In the central path method, the duality gap is parameterized by $t$ that decreases by a factor of $1 - O(\frac{1}{\sqrt{n}})$ in each iteration. Let $\mu := x s$. \cite{cls19} designed a potential function $\Psi(\mu/t - 1)$ such that as long as $\Psi(\mu/t - 1) \leq \poly(n)$, $\mu \approx t$ is satisfied. The change $\delta_{\mu}$ of $\mu$ has two parts: a term $-\frac{\mu}{\sqrt{n}}$ to decrease $\mu$, and a term $- \nabla \Psi(\mu/t - 1)$ to bound the potential function.  

Given $\delta_{\mu}$, the changes $\delta_x$ and $\delta_s$ should satisfy the following constraints (the second-order term $\delta_x \cdot \delta_s$ is ignored):
\begin{align}\label{eq:central_path_equation}
    X \delta_s + S \delta_x = &~ \delta_{\mu}, \\
    A \delta_x = &~ 0,\notag \\
    A^{\top} \delta_y + \delta_s = &~ 0. \notag
\end{align}
The unique solution of the above linear equations is
\begin{align*}
    \delta_x = &~ \frac{X}{\sqrt{XS}} (I - P) \frac{1}{\sqrt{XS}} \delta_{\mu}, \\
    \delta_s = &~ \frac{S}{\sqrt{XS}} P \frac{1}{\sqrt{XS}} \delta_{\mu},
\end{align*}
where $P = \sqrt{\frac{X}{S}} A^{\top} (A \frac{X}{S} A^{\top})^{-1} A \sqrt{\frac{X}{S}}$ is a projection matrix. Thus the stochastic central path method is summarized as the Algorithm~\ref{alg:alg_CP} presented in Section~\ref{sec:intro_optimization}.

\begin{comment}
\begin{algorithm}[ht]\caption{Stochastic Central Path. Restate of algorithm~\ref{alg:alg_CP}.}\label{alg:alg_CP_in_cls_section} 
\begin{algorithmic}[1]  
\State $i \leftarrow 1$, initialize $x,s \in \R^n$
\While{$i < \sqrt{n}$} \Comment{In each iteration, we hope $\mu \approx t$}
    \State $t \leftarrow t \cdot (1- 1/ \sqrt{n} )$ \Comment{target decrease of duality gap} 
    \State $\mu \leftarrow x s$ \Comment{actual decrease in duality gap}
    \State Compute $\delta_{\mu}$ based on $-\frac{\mu}{\sqrt{n}}$ and the gradient $-\nabla\Psi(\mu/t-1)$.
    \State $P \leftarrow \sqrt{\frac{X}{S}} A^\top ( A \frac{X}{S} A^\top )^{-1} A \sqrt{ \frac{X}{S} }$ 
    \Comment{\emph{matrix inverse}, matrix-matrix mult.}
    \State $\delta_x \leftarrow \frac{X}{\sqrt{XS}} (I-P) \frac{1}{\sqrt{XS}} \delta_{\mu} $,  $\delta_s \leftarrow \frac{S}{ \sqrt{XS} } P \frac{1}{ \sqrt{XS} } \delta_{\mu}$ \Comment{matrix-vector mult.}
    \State $x \leftarrow x + \delta_x$, $s \leftarrow s + \delta_s$
    \State $i \leftarrow i+1$
\EndWhile
\end{algorithmic}
\end{algorithm}
\end{comment}

%In order to get better running time, in \cite{cls19} they random sampled $\sqrt{n}$ coordinates of $\delta_{\mu}$ to get a sparse vector $\wt{\delta}_{\mu}$, and they also maintained the projection matrix in a data structure, such that $\delta_x$ and $\delta_s$ are computed using an approximated version of the projection matrix with $W^{\appr} \approx W:= \frac{X}{S}$. It is shown that this central path method is robust to these small perturbations.
\section{Optimization}
\label{sec:sketching_on_the_left_and_vector_maintenance}
In this section we provide the error analysis of central path method. The meaning of $x$, $s$, and the deviation of $\delta_x$, $\delta_s$ are standard (see the Central Path Method paragraph of Section~\ref{sec:preliminary}).

The proof of this whole paper is induction based. The induction hypothesis ensures that at the beginning of each iteration Assumption~\ref{ass:assumption} is satisfied. In this section we use this induction hypothesis to prove the guarantees of the central path method. Later in Section~\ref{sec:combine_improved} we combine the central path guarantees with the data structure guarantees to prove that Assumption~\ref{ass:assumption} is still satisfied at the end of each iteration. More specifically in this section we prove the following:
\begin{enumerate}
    \item Two guarantees Lemma~\ref{lem:w_movement} and \ref{lem:mu_movement} that are needed by the data structure given in Section~\ref{sec:correctness_improved}. Then the properties of the data structure prove that Part 1 of Assumption~\ref{ass:assumption} is still satisfied. 
    \item An upper bound on the potential function (Lemma~\ref{lem:potential_martingale}) which proves that Part 2 of Assumption~\ref{ass:assumption} is still satisfied.
\end{enumerate}

\begin{table}[!t]
\centering
\begin{tabular}{ | l | l | l | }
\hline
{\bf Lemma} & {\bf Section} & {\bf Comment}\\ \hline
Lemma~\ref{lem:stochastic_step} & Section~\ref{sec:stochastic_step} & Bounding $\delta_x, \delta_s, \delta_{\mu}, \delta_t, \delta_{\Phi}$ \\ \hline
Lemma~\ref{lem:bounding_mu_new_minus_mu} & Section~\ref{sec:bounding_mu_new_minus_mu} & Bounding $\mu^{\new} - \mu$ \\ \hline
Lemma~\ref{lem:potential_martingale} & Section~\ref{sec:potential_martingale} & Bounding the expectation of potential function \\ \hline
Lemma~\ref{lem:w_movement} & Section~\ref{sec:w_movement} & Bounding the movement of ${w}$ \\ \hline
Lemma~\ref{lem:mu_movement} & Section~\ref{sec:mu_movement} & Bounding the movement of ${\mu}$ \\ \hline
\end{tabular}	
\caption{Summary of this section}
\end{table}

\subsection{Definitions}

Some of the definitions are used as variables in the algorithm, some of the definitions are used only for analysis, and some of the definitions are used for both.

\begin{assumption}\label{ass:epsilon}
We use the following two error parameters $\epsilon$ and $\epsilon_{\mathrm{mp}}$:
\begin{align*}
    \epsilon \in (0, 10^{-4}), \text{~and~} \epsilon_{\mathrm{mp}} \in (0, 10^{-4}).
\end{align*}
\end{assumption}

We use the same potential function as \cite{cls19,lsz19,b20},
\begin{definition}[Potential function] \label{def:Phi_lambda}
For a parameter $\lambda > \log n$, we define function $\Phi_{\lambda} : \R^n \rightarrow \R$:
\begin{align*}
\Phi_{\lambda}(r) := \sum_{i=1}^n \cosh (\lambda r_i).
\end{align*}
\end{definition}

\begin{definition}[Overline version of parameters]\label{def:overline}
At the beginning of each iteration, we have $t \in \R$, and $\ov{x}, \ov{s} \in \R^n$ from last iteration. $t^{\new} \in \R$ is the new value of $t$. We define $\ov{w} \in \R^n$ and $\ov{\mu} \in \R^n$:
\begin{align*}
\ov{w} := \ov{x} / \ov{s},~ \ov{\mu} := \ov{x}\cdot  \ov{s}.
\end{align*}
We define overline version of projection matrix $\ov{P} \in \R^{n \times n}$:
\begin{align*}
\ov{P} := \sqrt{ \ov{W} } A^\top ( A \ov{W} A^\top )^{-1} A \sqrt{ \ov{W} }.
\end{align*}
The change in $\ov{\mu}$ consists of two parts $\ov{\delta}_t$ and $\ov{\delta}_{\Phi}$:
\begin{align*}
\ov{\delta}_t := & ~ \Big( \frac{ t^{\new} }{ t } - 1 \Big) \ov{\mu}, ~~~
\ov{\delta}_{\Phi} := ~ -\frac{\epsilon}{2} t^{\new} \cdot \frac{ \nabla \Phi_{\lambda} ( \ov{\mu} / t - 1 ) }{ \| \nabla \Phi_{\lambda} ( \ov{\mu} / t - 1 ) \|_2 }, ~~~ \ov{\delta}_{\mu} :=~ \ov{\delta}_t + \ov{\delta}_{\Phi}.
\end{align*}
 The changes to $\ov{x}$ and $\ov{s}$ is computed from $\ov{\delta}_{\mu}$:
\begin{align*}
\ov{\delta}_x := ~ \frac{ \ov{X} }{ \sqrt{ \ov{X} \ov{S} } } ( I - \ov{P} ) \frac{1}{ \sqrt{ \ov{X} \ov{S} } } \ov{\delta}_{\mu}, ~~~
\ov{\delta}_s := ~ \frac{ \ov{S} }{ \sqrt{ \ov{X} \ov{S} } } \ov{P} \frac{1}{ \sqrt{ \ov{X} \ov{S} } } \ov{\delta}_{\mu} .
\end{align*}
\end{definition}

The data structure will take $\ov{w} \in \R^n$ and $\ov{\mu} \in \R^n$ as inputs, and output some $\wt{w} \in \R^n$ and $\wt{\mu} \in \R^n$ such that $\wt{w} \approx_{\epsilon_{\mathrm{mp}}} \ov{w}, ~ \wt{\mu} \approx_{\epsilon_{\mathrm{mp}}} \ov{\mu}$.

\begin{definition}[Tilde version of parameters]\label{def:widetilde}
For a given $\wt{w} \in \R^n$ and $\wt{\mu} \in \R^n$ that is returned by the data structure, we define $\wt{x} \in \R^n$ and $\wt{s} \in \R^n$:
\begin{align*}
\wt{x} := \sqrt{ \wt{w} \wt{\mu} } ,~~~ \wt{s} := \sqrt{ \wt{\mu} / \wt{w} }.
\end{align*}
Note that $\wt{x} / \wt{s} = \wt{w}$, and $\wt{x} \wt{s} = \wt{\mu}$. We define tilde version of projection matrix $\wt{P} \in \R^{n \times n}$:
\begin{align*}
\wt{P} := \sqrt{ \wt{W} } A^\top ( A \wt{W} A^\top )^{-1} A \sqrt{ \wt{W} }.
\end{align*}
Further, we define $\wt{\delta}_t$, $\wt{\delta}_{\Phi}$, $\wt{\delta}_{\mu} \in \R^n$:
\begin{align*}
\wt{\delta}_t  := ~ ( \frac{ t^{\new} }{ t } - 1 ) \wt{\mu}, ~~~
\wt{\delta}_{\Phi} := ~ -\frac{\epsilon}{2} t^{\new} \cdot \frac{ \nabla \Phi_{\lambda} ( \wt{\mu} / t - 1 ) }{ \| \nabla \Phi_{\lambda} ( \wt{\mu} / t - 1 ) \|_2 }, ~~~
\wt{\delta}_{\mu} := ~ \wt{\delta}_t + \wt{\delta}_{\Phi}.
\end{align*}
And we define $\wt{\delta}_x$, $\wt{\delta}_s \in \R^n$:
\begin{align*}
\wt{\delta}_x := ~ \frac{ \wt{X} }{ \sqrt{ \wt{X} \wt{S} } } ( I - \wt{P} ) \frac{1}{ \sqrt{ \wt{X} \wt{S} } } \wt{\delta}_{\mu}, ~~~
\wt{\delta}_s := ~ \frac{ \wt{S} }{ \sqrt{ \wt{X} \wt{S} } } \wt{P} \frac{1}{ \sqrt{ \wt{X} \wt{S} } } \wt{\delta}_{\mu}.
\end{align*}
\end{definition}

Given these definitions, we state the following important assumptions. We assume they are satisfied in the beginning of each iteration, and use them to prove the correctness of the algorithm in this iteration. Later we will verify that these assumptions are always satisfied by induction on iterations. Note that the second assumption is true if the potential function is bounded by $\poly(n)$.
\begin{assumption}\label{ass:assumption}
We make the following assumptions:
\begin{enumerate}
\item $\wt{\mu} \approx_{\epsilon_{\mathrm{mp}}} \ov{\mu}$, $\wt{w} \approx_{\epsilon_{\mathrm{mp}}} \ov{w}$, where $\wt{\mu}$ and $\wt{w}$ are returned by the data structure, and $\ov{\mu}$, $\ov{w}$ are the input to the data structure.
\item $\ov{\mu} \approx_{0.1} t$.
\end{enumerate}
%$\epsilon$ is defined from $t^{\new}=\left(1-\frac{\epsilon}{\sqrt{n}}\right)t$. $\epsilon_{\mathrm{mp}}$ is defined as the approximation ratio from data structure.

%Proved in Lemma~\ref{lem:lemma_5.2_in_b20}: $\ov{x} \approx_{\epsilon_{\mathrm{mp}}} \wt{x}$,
%$\ov{s} \approx_{2\epsilon_{\mathrm{mp}}} \wt{s}$,
%$\wt{\mu} \approx_{\epsilon_{\mathrm{mp}}} \ov{\mu}$.

%Proved in Lemma~\ref{lem:lemma_5.6_in_b20}: $\|\wt{\delta}_{\mu}\|_2\leq \epsilon t$.

%Not proved yet: $\wt{\mu}\approx_{0.1} t$, 
%(Note that this assumption will hold after calling the data-structure.)
\end{assumption}

We further make the following definitions that make use of sketching matrices:
\begin{definition}[Hat version of parameters]\label{def:hat}
Let $b = o(n)$. For any $\wt{x} \in \R^n$, $\wt{s} \in \R^n$, $\wt{P} \in \R^{n \times n}$, $\wt{\delta}_{\mu} \in \R^n$, and a sketching matrix $R\in \R^{b\times n}$, we define $\wh{\delta}_x$, $\wh{\delta}_s \in \R^n$:
\begin{align*}
\wh{\delta}_x := ~ \frac{ \wt{X} }{ \sqrt{ \wt{X} \wt{S} } } ( I - R^\top R \wt{P} ) \frac{1}{ \sqrt{ \wt{X} \wt{S} } } \wt{\delta}_{\mu}, ~~~
\wh{\delta}_s := ~ \frac{ \wt{S} }{ \sqrt{ \wt{X} \wt{S} } } R^\top R \wt{P} \frac{1}{ \sqrt{ \wt{X} \wt{S} } } \wt{\delta}_{\mu}.
\end{align*}
\end{definition}
\begin{remark}
In our case, $R\in \R^{b\times n}$ is a subsampled randomized Hadamard matrix. See more details in Definition~\ref{def:song_coordiante_embedding}.
\end{remark}

\begin{definition}[New versions of definition]\label{def:new}
We use a superscript ``new'' to denote the corresponding variables at the beginning of the next iteration. Specifically, we define $t^{\new} \in \R$ as follows:
\begin{align*}
t^{\new} := (1-\frac{\epsilon}{3\sqrt{n}})t.
\end{align*}
We define $\ov{\mu}^{\new},\ov{w}^{\new} \in \R^n$ as follows:
\begin{align*}
\ov{\mu}^{\new} := (\ov{x} + \wh{\delta}_x) \cdot ( \ov{s} + \wh{\delta}_s ), ~~~
\ov{w}^{\new} := (\ov{x} + \wh{\delta}_x)/( \ov{s} + \wh{\delta}_s ).
\end{align*}
\end{definition}

The following facts directly follow from the definitions and Assumption~\ref{ass:assumption}.
\begin{fact}
\label{fact:delta_mu}
We have the following properties:
\begin{enumerate}
\item $\wt{X}\wh{\delta}_s + \wt{S}\wh{\delta}_x = \wt{\delta}_{\mu} = \wt{\delta}_{t} + \wt{\delta}_{\Phi}$,
\item $\wt{x}\approx_{2\epsilon_{\mathrm{mp}}} \ov{x}$, $\wt{s}\approx_{2\epsilon_{\mathrm{mp}}} \ov{s}$,
\item $\|\wt{\delta}_t\|_2 \leq 0.5\epsilon t$, $\|\wt{\delta}_{\Phi}\|_2\leq 0.5 \epsilon t$, $\|\wt{\delta}_{\mu}\|_2\leq  \epsilon t$.
\end{enumerate}
\end{fact}

\begin{proof}
\textbf{Part 1.}
From the definition of $\wh{\delta}_x$ and $\wh{\delta}_s$ (Definition~\ref{def:hat}) we have that 
\begin{align*}
\wt{X}\wh{\delta}_s + \wt{S}\wh{\delta}_x = &~ \frac{ \wt{X} \wt{S} }{ \sqrt{ \wt{X} \wt{S} } } R^\top R \wt{P} \frac{1}{ \sqrt{ \wt{X} \wt{S} } } \wt{\delta}_{\mu} + \frac{ \wt{S} \wt{X} }{ \sqrt{ \wt{X} \wt{S} } } ( I - R^\top R \wt{P} ) \frac{1}{ \sqrt{ \wt{X} \wt{S} } } \wt{\delta}_{\mu}\\
= &~ \frac{ \wt{X} \wt{S} }{ \sqrt{ \wt{X} \wt{S} } } I \frac{1}{ \sqrt{ \wt{X} \wt{S} } } \wt{\delta}_{\mu}
=  \wt{\delta}_{\mu} =~ \wt{\delta}_{t} + \wt{\delta}_{\Phi}.
\end{align*}

\noindent \textbf{Part 2.} By part 1 and part 2 of Assumption~\ref{ass:assumption}, we have that $\wt{\mu}\approx_{\epsilon_{\mathrm{mp}}}\ov{\mu}$ and $\wt{w}\approx_{\epsilon_{\mathrm{mp}}}\ov{w}$, so we have
\begin{align*}
\wt{x} = \sqrt{\wt{w} \cdot \wt{\mu}}
~\approx_{2\epsilon_{\mathrm{mp}}} ~ \sqrt{\ov{w} \cdot \ov{\mu}}
= \sqrt{ (\ov{x}/\ov{s}) \cdot (\ov{x} \cdot \ov{s})}
= \ov{x},
\end{align*}
where the first step follows from Definition~\ref{def:widetilde}, the second step follows from the fact that if $a \approx_{\epsilon_{\mathrm{mp}}} a'$ and $b \approx_{\epsilon_{\mathrm{mp}}} b'$, then $ab \approx_{2\epsilon_{\mathrm{mp}}} a'b'$, and the third step follows from Definition~\ref{def:overline}.

Using a similar argument, we also have that $\wt{s}\approx_{2\epsilon_{\mathrm{mp}}} \ov{s}$.

\noindent \textbf{Part 3.} We upper bound $\|\wt{\delta}_t\|_2$ as follows: 
\begin{align*}
\|\wt{\delta}_t\|_2 
=  \Big\| \Big(\frac{t^{\new}}{t}-1\Big) \wt{\mu} \Big\|_2
=  \Big\| \frac{\epsilon}{3\sqrt{n}} \wt{\mu} \Big\|_2
\leq  (1+\epsilon_{\mathrm{mp}}) \Big\| \frac{\epsilon }{3\sqrt{n}} \ov{\mu} \Big\|_2
\leq  1.1(1+\epsilon_{\mathrm{mp}}) \Big\|\frac{\epsilon t}{3\sqrt{n}} {\bf 1} \Big\|_2
\leq 0.5 \epsilon t.
\end{align*}
where the first step follows from the definition of $\wt{\delta}_t$ (Definition~\ref{def:widetilde}), the second step follows from the definition of $t^{\new}$ (Definition~\ref{def:new}), the third step follows from $\wt{\mu}\approx_{\epsilon_{\mathrm{mp}}}\ov{\mu}$ (Part 1 of Assumption~\ref{ass:assumption}), the forth step follows from $\ov{\mu}\approx_{0.1} t$ (Part 2 of Assumption~\ref{ass:assumption}), and the last step follows from $\epsilon_{\mathrm{mp}} \leq 10^{-4}$ (Assumption~\ref{ass:epsilon}).

Then we upper bound $\|\wt{\delta}_{\Phi}\|_2$ as follows:
\begin{align*}
\|\wt{\delta}_{\Phi}\|_2 =  \|\frac{\epsilon}{2}t^{\new}\cdot \frac{\nabla \Phi_{\lambda}(\wt{\mu}/t-1)}{\|\nabla \Phi_{\lambda}(\wt{\mu}/t-1)\|_2}\|_2 
= \frac{\epsilon}{2}t^{\new} 
=  \frac{\epsilon}{2} (1-\frac{\epsilon}{3\sqrt{n}})t 
\leq  0.5 \epsilon t,
\end{align*}
where the first step follows from the definition of $\wt{\delta}_{\Phi}$ (Definition~\ref{def:widetilde}), and the third step follows from the definition of $t^{\new}$ (Definition~\ref{def:new}).

Finally, we can use triangle inequality to upper bound \[
\|\wt{\delta}_\mu \|_2\leq \|\wt{\delta}_t\|_2 + \|\wt{\delta}_{\Phi}\|_2\leq 0.5 \epsilon t+ 0.5 \epsilon t\leq \epsilon t.
\]
\end{proof}

\subsection{Facts}

Random sketching matrices are usually used to give subspace embedding and approximate matrix product. Instead of using subspace embedding \cite{s06} and approximate matrix product \cite{kn12}, LP solver requires a different version of embedding, it was from \cite{psw17} implicitly and defined in \cite{lsz19} explicitly.
\begin{definition}[Coordinate-wise embedding]\label{def:song_coordiante_embedding}
Let $\Pi$ denote a distribution on $b\times n$ matrices $R$. We say $\Pi$ is an $(\alpha,\beta,\delta)$-coordinate-wise embedding if for any fixed vector $h\in \R^n$, the following properties hold:
\begin{align*}
1. & \E_{R \sim \Pi}[ R^\top R h ] = h, \\
2. & \E_{R \sim \Pi}[ ( R^\top R h )_i^2 ] \leq h_i^2 + \frac{\alpha}{b} \| h \|_2^2, \\
3. & \Pr_{R \sim \Pi}\Big[ | (R^\top R h)_i - h_i | > \| h \|_2 \frac{ \beta }{ \sqrt{b} } \Big] \leq \delta.
\end{align*}
\end{definition}

%\subsubsection{Tools from \texorpdfstring{\cite{cls19}}{}}

\begin{lemma}[Lemma A.1 in \cite{cls19}]\label{lem:var_xy}
Let $x$ and $y$ be (possibly dependent) random variables such that $|x| \leq c_x$ and $|y| \leq c_y$ almost surely. Then, we have
\begin{align*}
\Var[xy] \leq 2 c_x^2 \cdot \Var[y] + 2 c_y^2 \cdot \Var[x].
\end{align*}
\end{lemma}

\begin{fact}[Gradient and Hessian of potential function]
Let $\Phi_{\lambda}(r) = \sum_{i=1}^n \cosh (\lambda r_i)$ for some $\lambda > 0$. The gradient and the Hessian of $\Phi_{\lambda}(r)$ are
\begin{align*}
\nabla \Phi_{\lambda}(r) = &~ \lambda \cdot \big( \sinh ( \lambda r_1 ) , \sinh ( \lambda r_2 ) , \cdots, \sinh ( \lambda r_n ) \big)^\top \in \R^n, \\
\nabla^2 \Phi_{\lambda}(r) = &~ \diag \big(\lambda^2 \cosh ( \lambda r_1 ) , \lambda^2\cosh ( \lambda r_2 ) , \cdots, \lambda^2\cosh ( \lambda r_n ) \big)\in \R^{n\times n}.
\end{align*}
\end{fact}

\begin{lemma}[Basic properties of potential function, Lemma 4.12 in \cite{cls19}]\label{lem:potential_function_Phi_cosh}
Let $\Phi_{\lambda}(r) = \sum_{i=1}^n \cosh (\lambda r_i)$ for some $\lambda > 0$. For any vector $r \in \R^n$,\\
1. For any vector $\| v \|_{\infty} \leq 1/\lambda$, we have that
\begin{align*}
\Phi_{\lambda}( r + v ) \leq \Phi_{\lambda} (r) + \langle \nabla \Phi_{\lambda} (r), v \rangle + 2 \| v \|_{\nabla^2 \Phi_{\lambda}(r)}^2.
\end{align*}
2.
$
\| \nabla \Phi_{\lambda}(r) \|_2 \geq \frac{\lambda}{ \sqrt{n} } ( \Phi_{\lambda}(r) - n ) .
$\\
3.
$
\left( \sum_{i=1}^n \lambda^2 \cosh^2 (\lambda r_i) \right)^{1/2} \leq \lambda\sqrt{n} + \| \nabla \Phi_{\lambda} (r) \|_2.
$
\end{lemma}

%\subsubsection{Tools from \texorpdfstring{\cite{lsz19}}{}}

\begin{lemma}[Appendix E in \cite{lsz19}]\label{lem:E.5}
Let $R \in \R^{b \times n}$ denote a subsample randomized Hadamard transform, then it gives $(\alpha = 1, \beta = O(\log(n/\delta)), \delta)$-Coordinate-wise Embedding (Definition~\ref{def:song_coordiante_embedding}).
\end{lemma}

%\subsubsection{Tools from \texorpdfstring{\cite{b20}}{}}
We state another property of potential function
\begin{lemma}[Basic properties of potential function, general version of Lemma 5.14 of \cite{b20}]\label{lem:phi_gradient_bound}
Let $\Phi_{\lambda}(r) = \sum_{i=1}^n \cosh (\lambda r_i)$ for some $\lambda > 0$. %If $\wt{\mu} \approx_{\epsilon_{\mathrm{mp}}} \ov{\mu}$, $\ov{\mu} \approx_{0.1} t$ and $\epsilon_{\mathrm{mp}} \leq 1/(20\lambda)$,
If $\| v \|_{\infty} \leq 1/(30\lambda)$, then we have
\begin{align*}
\Big\langle \nabla\Phi_{\lambda}(r ), - \frac{\nabla \Phi_{\lambda}( r + v)}{\| \nabla \Phi_{\lambda}(r + v) \|_2} \Big\rangle \leq - 0.9 \| \nabla \Phi_{\lambda}(r) \|_2 + 0.1 \lambda \sqrt{n} .
\end{align*}
\end{lemma}
The proof is very similar to that of \cite{b20}, we put it here for completeness.
\begin{proof}
We have
\begin{align*}
\langle \nabla \Phi_\lambda( r ), -\nabla \Phi_\lambda( r + v ) \rangle
=&\:
- \langle \nabla \Phi_\lambda( r + v ), \nabla \Phi_\lambda( r + v ) \rangle
+ \langle \nabla \Phi_\lambda( r + v ) - \nabla \Phi_\lambda( r ), \nabla \Phi_\lambda( r + v ) \rangle \\
\le&\:
- \| \nabla \Phi_\lambda(r+v)\|_2^2
+ \| \nabla \Phi_\lambda(r) - \nabla \Phi_\lambda(r+v)\|_2 \cdot \| \nabla \Phi_\lambda(r+v) \|_2
\end{align*}
where the last step follows from $\langle a, b \rangle \leq \| a \|_2 \cdot \| b \|_2$. Then we have that
\begin{align}
&~ \Big\langle
\nabla \Phi_\lambda(r), 
-\frac{\nabla \Phi_\lambda(r+v)}{\|\nabla \Phi_\lambda(r+v)\|_2}
\Big\rangle \notag
\\ \leq &~
- \| \nabla \Phi_\lambda(r+v)\|_2
+ \| \nabla \Phi_\lambda(r) - \nabla \Phi_\lambda(r+v)\|_2 \notag \\
\leq &~ -\Big(\| \nabla \Phi_\lambda(r)\|_2 - \| \nabla \Phi_\lambda(r) - \nabla \Phi_\lambda(r+v)\|_2\Big) + \| \nabla \Phi_\lambda(r) - \nabla \Phi_\lambda(r+v)\|_2 \notag \\
\leq &~ -\| \nabla \Phi_\lambda(r)\|_2 + 2\| \nabla \Phi_\lambda(r) - \nabla \Phi_\lambda(r+v)\|_2, \label{eq:brand:inner_product}
\end{align}
where the second step follows from the fact that $\|a\|_2 \geq \|b\|_2 - \|b-a\|_2$.

Now we need to upper bound the norm 
$\| \nabla \Phi_\lambda(r) - \nabla \Phi_\lambda(r+v)\|_2$.
Note that $\nabla \Phi_\lambda(x)_i = \lambda \sinh(\lambda x_i)$ and $\sinh(x) = (e^x - e^{-x})/2$. We have the following property for  $|\sinh(x+y)-\sinh(x)|$:
\begin{align}\label{eq:brand:sinh}
|\sinh(x+y)-\sinh(x)|
= & ~
|e^x \cdot e^y - e^{-x} \cdot e^{-y} - (e^x - e^{-x})|/2 \notag \\
= & ~
|e^x \cdot (e^y-1) + e^{-x} \cdot (1-e^{-y})|/2 \notag \\
\le & ~
(e^x \cdot |e^y-1| + e^{-x} \cdot |1-e^{-y}|)/2 \notag \\
\le & ~ (e^x + e^{-x})/2 \cdot \max \{ |e^y-1|, |1-e^{-y}| \} \notag \\
\le & ~ (e^x + e^{-x})/2 \cdot (e^{|y|}-1)
= ~ \cosh(x) (e^{|y|}-1),
\end{align}
where the first and the last step follows from the definitions of $\sinh$ and $\cosh$, the third step follows from the triangle inequality of absolute values, and the fifth step follows from $e^y + e^{-y} \geq 2$.

Thus we can upper bound the difference as follows
\begin{align}\label{eq:brand:norm}
\| \nabla \Phi_\lambda(r) - \nabla \Phi_\lambda(r+v)\|_2
= &~ 
\lambda \| \sinh(\lambda(r+v)) - \sinh(\lambda r)\|_2 \notag \\
\leq &~
\lambda \| \cosh(\lambda r) (e^{|\lambda v|}-1)\|_2 \notag \\
\leq &~
\lambda \| \cosh(\lambda r)\|_2 (e^{\lambda\|v\|_\infty}-1) \notag \\
\leq &~
(\lambda \sqrt{n} + \|\nabla \Phi_\lambda(r)\|_2) (e^{\lambda\|v\|_\infty}-1), 
\end{align}
where the second step follows from Eq.~\eqref{eq:brand:sinh}, the third step follows from the fact that $\|a\cdot b\|_2 \leq \|a\|_2 \|b\|_{\infty}$, and the last step follows from Part 3 of Lemma~\ref{lem:potential_function_Phi_cosh}.

%We then bound $\|(\wt{\mu}-\ov{\mu})/t\|_{\infty}$ as follows:
%\begin{align*}
%\|(\wt{\mu}-\ov{\mu})/t\|_\infty 
%\leq & ~ \| \epsilon_{\mathrm{mp}} \ov{\mu} / t \|_\infty \\
%\leq & ~ 1.1 \epsilon_{ \mathrm{mp} },
%\end{align*}
%where in the first step we used that $\wt{\mu} \approx_{\epsilon_{\mathrm{mp}}} \ov{\mu}$, in the second step we used that $\ov{\mu} \approx_{0.1} t$.
We then have
\begin{align}\label{eq:brand:e_bound}
(e^{\lambda\|v\|_\infty}-1) < ~ e^{1 / 30} - 1 \leq ~ 0.05,
\end{align}
where the first step follows from $\| v \|_{\infty} \leq 1/(30\lambda)$.

Finally, this allows us to obtain
\begin{align*}
\Big\langle
\nabla \Phi_\lambda(r), 
-\frac{\nabla \Phi_\lambda(r+v)}{\|\nabla \Phi_\lambda(r+v)\|_2}
\Big\rangle
\leq &~
- \| \nabla \Phi_\lambda(r)\|_2
+ 2\| \nabla \Phi_\lambda(r) - \nabla \Phi_\lambda(r+v)\|_2 \\
< &~ - \| \nabla \Phi_\lambda(r)\|_2
+ 0.1 (\lambda \sqrt{n} + \|\nabla \Phi_\lambda(r)\|_2) \\
\leq &~ - 0.9 \| \nabla \Phi_\lambda(r)\|_2
+ 0.1\lambda \sqrt{n}
\end{align*}
where the first step follows from Eq.~\eqref{eq:brand:inner_product}, and the second step follows from Eq.~\eqref{eq:brand:norm} and Eq.~\eqref{eq:brand:e_bound}.
\end{proof}

\subsection{Bounding \texorpdfstring{${\delta}_s$, ${\delta}_x$, $\delta_t$, $\delta_{\Phi}$}{} and \texorpdfstring{${\delta}_\mu$}{}}\label{sec:stochastic_step}

\begin{table}[!t]
\small
\centering
\begin{tabular}{ | l | l | l | l | l | }
\hline
{\bf Quantity} & {\bf Bound} & {\bf Stated in Lem.~\ref{lem:stochastic_step}} & {\bf Used by Lem.~\ref{lem:bounding_mu_new_minus_mu}}\\ \hline
$\| \ov{s}^{-1} \wt{\delta}_s \|_2$, $\| \ov{x}^{-1} \wt{\delta}_x \|_2$ & $\epsilon$ & Part 1 & Part 1 \\ \hline
$\| \E[ \ov{s}^{-1} \wh{\delta}_s ] \|_2$, $\| \E[ \ov{x}^{-1} \wh{\delta}_x ] \|_2$ & $\epsilon$ & Part 1 & Part 1 \\ \hline
$\| \ov{\mu}^{-1} ( \wt{\delta}_t - \ov{\delta}_t ) \|_2$ & $\epsilon_{\mathrm{mp}} \cdot \epsilon$ & Part 1 & Part 1 \\ \hline
$\| \ov{\mu}^{-1} ( \ov{\delta}_t + \wt{\delta}_{\Phi} ) \|_2$ & $\epsilon$ & Part 1 & Part 4 \\ \hline
$\Var[ \ov{x}_i^{-1} \wh{\delta}_{x,i} ] , \Var[ \ov{s}_i^{-1} \wh{\delta}_{s,i} ]$ & $\epsilon^2 / b$ & Part 2 & Part 1, 2 \\ \hline
$ \| \ov{x}^{-1} ( \ov{x} - \wt{x} )  \|_{\infty} $, $  \| \ov{s}^{-1} ( \ov{s} - \wt{s} ) \|_{\infty} $ & $\epsilon_{\mathrm{mp}}$ & Part 3 & /  \\ \hline
$\| \ov{x}^{-1} \wt{\delta}_x \|_{\infty}, \| \ov{s}^{-1} \wt{\delta}_s \|_{\infty} $ & $\epsilon$ & Part 3 & / \\ \hline 
$ \| \ov{\mu}^{-1} \wt{\delta}_{\mu} \|_{\infty}$ & $\epsilon$ & Part 3 & Part 3 \\ \hline
$ \| \ov{x}^{-1} \wh{\delta}_x \|_{\infty}$, $ \| \ov{s}^{-1} \wh{\delta}_s \|_{\infty} $ & $\epsilon$ & Part 4 & Part 2, 3 \\ \hline
\end{tabular}	\caption{Summary of Lemma~\ref{lem:stochastic_step}. We ignore the constants. Note that the Part 4 (last row of this table) holds with probability $1-1/\poly(n)$ and requires $b \geq 1000 \log^2 n$.}
\end{table}

The goal of this section is to prove Lemma~\ref{lem:stochastic_step}.
\begin{lemma}[A deep version of Lemma 4.3 in \cite{cls19}]\label{lem:stochastic_step}
Under Assumption~\ref{ass:assumption}, and given that $ b\geq 1000\log^2n$, we have the following:\\
1.
$\begin{aligned}[t]
& ~ \| \ov{s}^{-1} \wt{\delta}_s \|_2 \leq 2\epsilon, ~
    \| \ov{x}^{-1} \wt{\delta}_x \|_2 \leq 2\epsilon,\\
& ~ \| \E[ \ov{s}^{-1} \wh{\delta}_s ] \|_2 \leq 2\epsilon, ~
    \| \E[ \ov{x}^{-1} \wh{\delta}_x ] \|_2 \leq 2\epsilon,\\
& ~ \| \ov{\mu}^{-1} ( \wt{\delta}_t - \ov{\delta}_t ) \|_2 \leq \epsilon_{\mathrm{mp}} \cdot \epsilon, \\
& ~ \| \ov{\mu}^{-1} ( \ov{\delta}_t + \wt{\delta}_{\Phi} ) \|_2 \leq 5 \epsilon.
\end{aligned}$

\noindent 2. 
$\begin{aligned}[t]
\Var[ \ov{x}_i^{-1} \wh{\delta}_{x,i} ] \leq 2 \epsilon^2 / b,~ \Var[ \ov{s}_i^{-1} \wh{\delta}_{s,i} ] \leq 2 \epsilon^2 / b.
\end{aligned}$

\noindent 3.
$\begin{aligned}[t]
& ~ \| \ov{x}^{-1} ( \ov{x} - \wt{x} )  \|_{\infty} \leq 2 \epsilon_{\mathrm{mp}},~  \| \ov{s}^{-1} ( \ov{s} - \wt{s} ) \|_{\infty} \leq 2 \epsilon_{\mathrm{mp}}, \\
& ~ \|\ov{x}^{-1} \wt{\delta}_x\|_\infty \leq 2 \epsilon, ~ \|\ov{s}^{-1} \wt{\delta}_s\|_\infty \leq 2\epsilon,\\
& ~ \| \ov{\mu}^{-1} \wt{\delta}_{\mu} \|_{\infty} \leq 5\epsilon.
\end{aligned}$

\noindent 4.
$\begin{aligned}[t]
\| \ov{x}^{-1} \wh{\delta}_x \|_{\infty} \leq 3\epsilon, ~ \| \ov{s}^{-1} \wh{\delta}_s \|_{\infty} \leq 3\epsilon
\end{aligned}$ hold with probability $1-1/n^4$.
\end{lemma}

%\begin{remark}\label{rem:stochastic_step}
%In Part 4 of Lemma~\ref{lem:stochastic_step}, $\| \ov{x}^{-1} \wh{\delta}_x \|_{\infty} \leq 3\epsilon,  \| \ov{s}^{-1} \wh{\delta}_s \|_{\infty} \leq 3\epsilon$ hold with probability $1-1/n^4$.
%\end{remark}

%\begin{remark}
%For notational simplicity, the $\E$ and $\Var$ in the proof are for the case without resampling (Line~\ref{line:resample}). Since the all the additional terms due to resampling are polynomially bounded and since we can set failure probability to an arbitrarily small inverse polynomial (see Claim~\ref{claim:prob}), the proof does not change and the result remains the same.
%\end{remark}

%\begin{proof}

%%%%%%%%%%%%%%%%%%%%%%%%%%%%%%%%%%%%%%%%%%%%%%%%%%%%%%%%%%%%%%%%%%%%%%%%%%%%%%%%%%%%%
\begin{claim}[Part 1 of Lemma~\ref{lem:stochastic_step}, bounding the $\ell_2$ norm]\label{cla:bounding_l2_norm_expectation}
\begin{align*}
(1) & ~ \| \ov{s}^{-1} \wt{\delta}_s \|_2 \leq 2\epsilon, ~
    \| \ov{x}^{-1} \wt{\delta}_x \|_2 \leq 2\epsilon,\\
(2) & ~ \| \E[ \ov{s}^{-1} \wh{\delta}_s ] \|_2 \leq 2\epsilon,~
    \| \E[ \ov{x}^{-1} \wh{\delta}_x ] \|_2 \leq 2 \epsilon,\\
(3) & ~ \| \ov{\mu}^{-1} ( \wt{\delta}_t - \ov{\delta}_t ) \|_2 \leq \epsilon_{\mathrm{mp}} \cdot \epsilon, \\
(4) & ~ \| \ov{\mu}^{-1} ( \ov{\delta}_t + \wt{\delta}_{\Phi} ) \|_2 \leq 5 \epsilon.
\end{align*}
\end{claim}
\begin{proof}

\textbf{Proof of (1).} We first upper bound the $\ell_2$ norm of $\wt{P} \frac{\wt{\delta}_{\mu}}{\sqrt{\wt{X}\wt{S}}}$ in the following way:
\begin{align}\label{eq:upper_bound_P_xs_delta_mu}
\Big\| \wt{P}\frac{1}{\sqrt{\wt{X}\wt{S}}}\wt{\delta}_{\mu} \Big\|_2 
\leq & ~ \Big\| \frac{1}{\sqrt{\wt{X}\wt{S}}}\wt{\delta}_{\mu} \Big\|_2 
\leq  \Sup \Big[ \frac{1}{\sqrt{\wt{\mu}}} \Big] \cdot \|\wt{\delta}_{\mu}\|_2 
\leq \Sup \Big[ \frac{1}{\sqrt{(1-\epsilon_{\mathrm{mp}})\ov{\mu}}} \Big] \cdot \epsilon t \notag\\
\leq & ~ \frac{1}{\sqrt{0.9(1-\epsilon_{\mathrm{mp}})t}} \cdot \epsilon t 
~ \leq ~ 1.1\epsilon \sqrt{t},
\end{align}
where the first step holds since $\wt{P}$ is an orthogonal projection matrix, the second step is because $\wt{x}\wt{s}=\wt{\mu}$ (Definition~\ref{def:widetilde}) and $\|a\cdot b\|_2\leq \Sup[a]\|b\|_2$, the third step follows from $\wt{\mu} \approx_{\epsilon_{\mathrm{mp}}} \ov{\mu}$ (Part 1 of Assumption~\ref{ass:assumption}) and $\| \wt{\delta}_{\mu} \|_2 \leq \epsilon t$ (Part 3 of Fact~\ref{fact:delta_mu}), the fourth step follows from $\ov{\mu} \approx_{0.1} t$ (Part 2 of Assumption~\ref{ass:assumption}), and the last step follows from $\epsilon_{\mathrm{mp}}\leq 10^{-4}$ (Assumption~\ref{ass:epsilon}).

Then we can upper bound $\| \ov{s}^{-1} \wt{\delta}_s \|_2$ as follows:
\begin{align*}
\| \ov{s}^{-1} \wt{\delta}_s \|_2 
= & ~ \left\|\frac{\ov{S}^{-1}\wt{S}}{\sqrt{\wt{X}\wt{S}}} \wt{P} \frac{1}{\sqrt{\wt{X}\wt{S}}} \wt{\delta}_{\mu} \right\|_2 
\leq  \Sup\left[\frac{\ov{S}^{-1}\wt{S}}{\sqrt{\wt{X}\wt{S}}}\right] \left\|\wt{P} \frac{1}{\sqrt{\wt{X}\wt{S}}} \wt{\delta}_{\mu} \right\|_2
\leq  \frac{1+2\epsilon_{\mathrm{mp}}}{\sqrt{(1-\epsilon_{\mathrm{mp}}) 0.9t}}\cdot \left\|\wt{P} \frac{1}{\sqrt{\wt{X}\wt{S}}} \wt{\delta}_{\mu} \right\|_2\\
\leq & ~ \frac{1+2\epsilon_{\mathrm{mp}}}{\sqrt{(1-\epsilon_{\mathrm{mp}}) 0.9t}}\cdot 1.1\epsilon \sqrt{t}
~ \leq ~ 2\epsilon,
\end{align*}
where the first step follows by definition of $\wt{\delta}_s$ (Definition~\ref{def:widetilde}), the second step follows from
$\|a\cdot b\|_2\leq \Sup[a]\cdot \|b\|_2$, the third step follows from  $\wt{s}\approx_{2\epsilon_{\mathrm{mp}}} \ov{s}$ (Part 2 of Fact~\ref{fact:delta_mu}) and $\wt{x}\wt{s}=\wt{\mu}\approx_{\epsilon_{\mathrm{mp}}} \ov{\mu} \approx_{0.1} t$ (Definition~\ref{def:widetilde}, Assumption~\ref{ass:assumption}), the fourth step follows from Eq.~\eqref{eq:upper_bound_P_xs_delta_mu}, the last step follows from $\epsilon_{\mathrm{mp}}\leq 10^{-4}$ (Assumption~\ref{ass:epsilon}).

The proof for $\| \ov{x}^{-1} \wt{\delta}_x \|_2 \leq 2\epsilon$ is similar since $I-\wt{P}$ is also an orthogonal projection matrix.

\noindent \textbf{Proof of (2).} 
Note that from Lemma~\ref{lem:E.5} and definition of $\wh{\delta}_s$ and $\wt{\delta}_s$ we have $\E[\wh{\delta}_s]=\wt{\delta}_s$, therefore,
\begin{align*}
\| \E[ \ov{s}^{-1} \wh{\delta}_s ] \|_2 = \| \ov{s}^{-1} \wt{\delta}_s \|_2 \leq 2\epsilon.
\end{align*}

Similarly, we can prove $\|\ov{x}^{-1}\wt{\delta}_x\|_2\leq 2\epsilon$.

\noindent {\bf Proof of (3).} 
\begin{align*}
\| \ov{\mu}^{-1} ( \wt{\delta}_t - \ov{\delta}_t ) \|_2 
=  \left\| \ov{\mu}^{-1} \left(\left(\frac{t^{\new}}{t}-1\right) \wt{\mu} -\left(\frac{t^{\new}}{t}-1\right) \ov{\mu}\right) \right\|_2 
= \left\| \ov{\mu}^{-1}\frac{\epsilon}{3\sqrt{n}} (\wt{\mu}-\ov{\mu}) \right\|_2
\leq  \epsilon_{\mathrm{mp}} \cdot \epsilon,
\end{align*}
where the first step is by definition of $\wt{\delta}_t$ and $\ov{\delta}_t$ (Definition~\ref{def:widetilde} and \ref{def:overline}), the second step is by $\frac{t^{\new}}{t}-1=\frac{(1-\epsilon/3 \sqrt{n})t}{t}-1=-\frac{\epsilon}{3\sqrt{n}}$, and the last step is by $\wt{\mu} \approx_{\epsilon_{\mathrm{mp}}} \ov{\mu}$ (Part 1 of Assumption~\ref{ass:assumption}).

%\end{proof}

\noindent {\bf Proof of (4).} We use triangle inequality to upper bound 
\begin{align*}
\| \ov{\mu}^{-1} ( \ov{\delta}_t + \wt{\delta}_{\Phi} ) \|_2
=  \| \ov{\mu}^{-1} ( (\ov{\delta}_t - \wt{\delta}_t) + (\wt{\delta}_t + \wt{\delta}_\Phi) ) \|_2
\leq  \| \ov{\mu}^{-1} (\ov{\delta}_t - \wt{\delta}_t)\|_2  + \| \ov{\mu}^{-1} (\wt{\delta}_t + \wt{\delta}_\Phi) \|_2.
\end{align*}

The first term is upper bounded in Part (3): $\| \ov{\mu}^{-1} ( \wt{\delta}_t - \ov{\delta}_t ) \|_2 \leq \epsilon_{\mathrm{mp}} \cdot \epsilon$. 

For the second term, we have
\begin{align*}
\| \ov{\mu}^{-1} ( \wt{\delta}_t + \wt{\delta}_{\Phi} ) \|_2 
= & ~ \| \ov{\mu}^{-1} \wt{\delta}_{\mu} \|_2 
=  \| \ov{\mu}^{-1} ( \wt{x} \wt{\delta}_s + \wt{s} \wt{\delta}_x ) \|_2 
\leq  \| \ov{\mu}^{-1} \wt{x} \wt{\delta}_s \|_2 + \| \ov{\mu}^{-1} \wt{s} \wt{\delta}_x  \|_2 \\
\leq & ~ (1+2\epsilon_{\mathrm{mp}}) \| \ov{s}^{-1} \wt{\delta}_s \|_2 + (1+2\epsilon_{\mathrm{mp}}) \| \ov{x}^{-1} \wt{\delta}_x \|_2
\leq  4\epsilon(1+2\epsilon_{\mathrm{mp}}) 
\leq  5\epsilon,
\end{align*}
where the first step follows from $\wt{\delta}_\mu = \wt{\delta}_t + \wt{\delta}_{\Phi}$ (Definition~\ref{def:widetilde}), the second step follows from 
\begin{align*}
\wt{x}\wt{\delta}_s+\wt{s}\wt{\delta}_x=\frac{\wt{S}\wt{X}}{\sqrt{\wt{X}\wt{S}}} (I-\wt{P}) \frac{1}{\sqrt{\wt{X}\wt{S}}} \wt{\delta}_{\mu}
+ \frac{\wt{X}\wt{S}}{\sqrt{\wt{X}\wt{S}}} \wt{P} \frac{1}{\sqrt{\wt{X}\wt{S}}} \wt{\delta}_{\mu}=\wt{\delta}_\mu,
\end{align*}
the third step follows from triangle inequality, the forth step follows from $\wt{x}\approx_{2\epsilon_{\mathrm{mp}}}\ov{x}$, $\wt{s}\approx_{2\epsilon_{\mathrm{mp}}} \ov{s}$ (Part 2 of Fact~\ref{fact:delta_mu}) and $\ov{x}\cdot \ov{s}=\ov{\mu}$ (Definition~\ref{def:overline}), 
the fifth step follows from Part (1) that $\| \ov{x}^{-1} \wt{\delta}_x \|_2\leq 2\epsilon$ and $\| \ov{s}^{-1} \wt{\delta}_s \|_2\leq 2\epsilon$, and the sixth step follows from $\epsilon_{\mathrm{mp}}\leq 10^{-4}$ (Assumption \ref{ass:epsilon}).
\end{proof}

%%%%%%%%%%%%%%%%%%%%%%%%%%%%%%%%%%%%%%%%%%%%%%%%%%%%%%%%%%%%%%%%%%%%%%%%%%%%%%%%%%%%%
\begin{claim}[Part 2 of Lemma~\ref{lem:stochastic_step}, bounding the variance per coordinate]\label{cla:bounding_coordinate_variance}
\begin{align*}
\Var[ \ov{x}_i^{-1} \wh{\delta}_{x,i} ] \leq 2 \epsilon^2 / b,~ \Var[ \ov{s}_i^{-1} \wh{\delta}_{s,i} ] \leq 2 \epsilon^2 / b.
\end{align*}
\end{claim}

\begin{proof}
For each $i\in [n]$, we can rewrite the expectation of $\ov{s}_i^{-1}\wh{\delta}_{s,i}$ as follows:
\begin{align}\label{eq:var_per_coordinate:expectation}
\E[\ov{s}_i^{-1} \wh{\delta}_{s,i}] = & \E \left[ \frac{\wt{s}_i}{\ov{s}_i \sqrt{\wt{x}_i \wt{s}_i}} \left(R^\top R \wt{P} \frac{1}{\sqrt{\wt{X} \wt{S}}} \wt{\delta}_{\mu}\right)_i \right] 
= \frac{\wt{s}_i}{\ov{s}_i \sqrt{\wt{x}_i \wt{s}_i}} \left(\wt{P} \frac{1}{\sqrt{\wt{X} \wt{S}}} \wt{\delta}_{\mu}\right)_i 
=  \ov{s}_i^{-1} \wt{\delta}_{s,i},
\end{align}
where the first step follows from the definition of $\wh{\delta}_s$ (Definition \ref{def:hat}), the second step follows from the property of matrix $R$ (Part 1 of Definition \ref{def:song_coordiante_embedding} and Lemma~\ref{lem:E.5}), and the third step follows from the definition of $\wt{\delta}_s$ (Definition \ref{def:widetilde}).

We then upper bound the expectation of $(\ov{s}_i^{-1}\wh{\delta}_{s,i})^2$ as follows:
\begin{align} \label{eq:var_per_coordinate:expectation_squared}
\E[(\ov{s}_i^{-1}\wh{\delta}_{s,i})^2] 
= & ~ \frac{\wt{s}_i^2}{\ov{s}_i^2 \cdot \wt{x}_i \wt{s}_i} \E \Big[ (R^\top R \wt{P} \frac{1}{\sqrt{\wt{X} \wt{S}}} \wt{\delta}_{\mu})_i^2 \Big] 
\leq  \frac{\wt{s}_i^2}{\ov{s}_i^2 \cdot \wt{x}_i \wt{s}_i}\Big((\wt{P} \frac{1}{\sqrt{\wt{X} \wt{S}}} \wt{\delta}_{\mu})_i^2 + \frac{1}{b} \Big\|\wt{P} \frac{1}{\sqrt{\wt{X} \wt{S}}} \wt{\delta}_{\mu} \Big\|_2^2 \Big) \notag \\
= & ~ (\ov{s}_i^{-1} \wt{\delta}_{s,i})^2 + \frac{\wt{s}_i^2}{\ov{s}_i^2 \cdot \wt{x}_i \wt{s}_i} \cdot \frac{1}{b} \Big\| \wt{P} \frac{1}{\sqrt{\wt{X} \wt{S}}} \wt{\delta}_{\mu} \Big\|_2^2,
\end{align}
where the first step follows from definition of $\wh{\delta}_{s}$ (Definition~\ref{def:hat}), the second step follows from Part 2 of Definition~\ref{def:song_coordiante_embedding}, and the third step follows from the definition of $\wt{\delta}_s$ (Definition \ref{def:widetilde}).

%Then we upper bound $\|\wt{P} \frac{1}{\sqrt{\wt{X} \wt{S}}} \wt{\delta}_{\mu}\|_2$ in the following way:
%\begin{align}\label{eq:upper_bound_P_xs_delta_mu}
%\Big\| \wt{P}\frac{1}{\sqrt{\wt{X}\wt{S}}}\wt{\delta}_{\mu} \Big\|_2 
%\leq & ~ \Big\| \frac{1}{\sqrt{\wt{X}\wt{S}}}\wt{\delta}_{\mu} \Big\|_2 %\notag \\
%= & ~ \Big\| \frac{1}{\sqrt{\wt{\mu}}}\wt{\delta}_{\mu} \Big\|_2 \notag \\
%\leq & ~ \Sup \Big[ \frac{1}{\sqrt{\wt{\mu}}} \Big] \cdot %\|\wt{\delta}_{\mu}\|_2 \notag \\
%\leq & ~ \Sup \Big[ \frac{1}{\sqrt{(1-\epsilon_{\mathrm{mp}})\ov{\mu}}} %\Big] \cdot \epsilon t \notag\\
%\leq & ~ \frac{1}{\sqrt{0.9(1-\epsilon_{\mathrm{mp}})t}} \cdot \epsilon t %\notag\\
%\leq & ~ 1.1\epsilon \sqrt{t},
%\end{align}
%where the first step holds since $\wt{P}$ is an orthogonal projection matrix, the second step is because $\wt{x}\wt{s}=\wt{\mu}$ (Definition~\ref{def:widetilde}), the third step follows by $\|a\cdot b\|_2\leq \Sup[a]\|b\|_2$, the forth step follows from $\wt{\mu} \approx_{\epsilon_{\mathrm{mp}}} \ov{\mu}$ (Part 1 of Assumption~\ref{ass:assumption}) and $\| \wt{\delta}_{\mu} \|_2 \leq \epsilon t$ (Part 3 of Fact~\ref{fact:delta_mu}), the fifth step follows from $\ov{\mu} \approx_{0.1} t$ (Part 3 of Assumption~\ref{ass:assumption}), and the last step follows from $\epsilon_{\mathrm{mp}}\leq 10^{-4}$ (Assumption~\ref{ass:epsilon}).

Now we have
\begin{align*}
\Var[\ov{s}_i^{-1}\wh{\delta}_{s,i}] 
= &~ \E[(\ov{s}_i^{-1}\wh{\delta}_{s,i})^2] - ( \E[\ov{s}_i^{-1}\wh{\delta}_{s,i}] )^2 
=  \frac{1}{b} \frac{\wt{s}_i^2}{\ov{s}_i^2 \cdot \wt{x}_i \wt{s}_i}  \|\wt{P} \frac{1}{\sqrt{\wt{X} \wt{S}}} \wt{\delta}_{\mu}\|_2^2 
\leq  \frac{1}{b} \frac{(1 + 2\epsilon_{\mathrm{mp}})^2}{\wt{\mu}_i} \|\wt{P} \frac{1}{\sqrt{\wt{X} \wt{S}}} \wt{\delta}_{\mu}\|_2^2 \\
\leq & ~ \frac{1}{b} \frac{(1 + 2\epsilon_{\mathrm{mp}})^2}{\wt{\mu}_i} (1.1 \epsilon)^2 t
\leq  \frac{1}{b} (1 + 2\epsilon_{\mathrm{mp}})^2 (1.1\epsilon)^2(1.1+\epsilon_{\mathrm{mp}}) 
\leq  \frac{2\epsilon^2}{b},
\end{align*}
where the first step follows from definition of variance, the second step follows from Eq.~\eqref{eq:var_per_coordinate:expectation} and Eq.~\eqref{eq:var_per_coordinate:expectation_squared}, the third step follows from $\wt{s}_i\approx_{2\epsilon_{\mathrm{mp}}}\ov{s}_i$ (Part 2 of Fact~\ref{fact:delta_mu}) and $\wt{\mu}=\wt{x}\cdot\wt{s}$ (Definition~\ref{def:widetilde}), the forth step follows from  $\|\wt{P}\frac{1}{\sqrt{\wt{X}\wt{S}}}\wt{\delta}_{\mu}\|_2 \leq 1.1\epsilon \sqrt{t}$ (Eq.~\eqref{eq:upper_bound_P_xs_delta_mu}), and the sixth step follows from $\wt{\mu}\approx_{0.1+\epsilon_{\mathrm{mp}}}t$ (Part 1 and 2 of Assumption~\ref{ass:assumption}), the last step follows from $\epsilon_{\mathrm{mp}}\leq 10^{-4}$ (Assumption~\ref{ass:epsilon}).

The other part that $\Var[\ov{x}_i^{-1}\wh{\delta}_{x,i}]\leq 2\epsilon^2/b$ follows from a similar argument.
\end{proof}

\begin{claim}[Part 3 of Lemma~\ref{lem:stochastic_step}, bounding the infinity norm]\label{cla:infitynorm}
\begin{align*}
(1) & ~ \| \ov{x}^{-1} ( \ov{x} - \wt{x} )  \|_{\infty} \leq  2\epsilon_{\mathrm{mp}}, ~ \| \ov{s}^{-1} ( \ov{s} - \wt{s} ) \|_{\infty} \leq 2 \epsilon_{\mathrm{mp}}, \\
(2) & ~ \|\ov{x}^{-1} \wt{\delta}_x\|_\infty \leq 2\epsilon, ~ \|\ov{s}^{-1} \wt{\delta}_s\|_\infty \leq 2\epsilon,\\
(3) & ~ \| \ov{\mu}^{-1} \wt{\delta}_{\mu} \|_{\infty} \leq 5\epsilon.
\end{align*}
%Without resampling, the following holds with probability $1-2n \exp(-\frac{0.003 k}{\epsilon\sqrt{n}\log{n}})$. 
%$$\|\overline{s}^{-1}\wt{\delta}_{s}\|_{\infty}\leq \frac{0.01}{\log{n}},\|s^{-1}\wt{\delta}_{s}\|_{\infty}\leq \frac{0.02}{\log{n}},\|\overline{x}^{-1}\wt{\delta}_{x}\|_{\infty}\leq \frac{0.01}{\log{n}},\|x^{-1}\wt{\delta}_{x}\|_{\infty}\leq \frac{0.02}{\log{n}},\|\mu^{-1}\wt{\delta}_{\mu}\|_{\infty}\leq \frac{0.02}{\log{n}}.$$
%With resampling, it always holds.
\end{claim}
\begin{proof}
{\bf Proof of (1).} From Part 2 of Fact~\ref{fact:delta_mu}, we have that $\wt{x}\approx_{2\epsilon_{\mathrm{mp}}} \ov{x}$ and $\wt{s}\approx_{2\epsilon_{\mathrm{mp}}} \ov{s}$. Therefore $\| \ov{x}^{-1} ( \ov{x} - \wt{x} )  \|_{\infty} \leq  2\epsilon_{\mathrm{mp}}$,  $\| \ov{s}^{-1} ( \ov{s} - \wt{s} ) \|_{\infty} \leq 2 \epsilon_{\mathrm{mp}},$

\noindent {\bf Proof of (2).} From Part 1 of Claim~\ref{cla:bounding_l2_norm_expectation}, we have $\|\ov{x}^{-1} \wt{\delta}_x\|_2\leq 2\epsilon$. Therefore, $\|\ov{x}^{-1} \wt{\delta}_x\|_\infty \leq \|\ov{x}^{-1} \wt{\delta}_x\|_2\leq 2\epsilon$. Similarly, we have $\|\ov{s}^{-1} \wt{\delta}_s\|_\infty \leq \|\ov{s}^{-1} \wt{\delta}_s\|_2\leq 2\epsilon$.

\noindent {\bf Proof of (3).} Now, the last term follows by
\begin{align*}
| \ov{\mu}_i^{-1} \wt{\delta}_{\mu,i} | 
=  & ~ | \ov{x}_i^{-1} \ov{s}_i^{-1} ( \wt{x}_i \wt{\delta}_{s,i} + \wt{s}_i \wt{\delta}_{x,i} ) | 
\leq  (1+2\epsilon_{\mathrm{mp}})| \ov{s}_i^{-1} \wt{\delta}_{s,i} | + (1+2\epsilon_{\mathrm{mp}})| \ov{x}_i^{-1} \wt{\delta}_{x,i} | \\
\leq & ~ (1+2\epsilon_{\mathrm{mp}}) 2\epsilon + (1+2\epsilon_{\mathrm{mp}}) 2\epsilon 
= ~ 5\epsilon,
\end{align*}
where the first step is by $\wt{x}\wt{\delta}_s+\wt{s}\wt{\delta}_x=\wt{\delta}_{\mu}$ (Definition~\ref{def:widetilde}), the second step is by $\ov{x}\approx_{2\epsilon_{\mathrm{mp}}} \wt{x}$ and $\ov{s}\approx_{2\epsilon_{\mathrm{mp}}} \wt{s}$ (Part 2 of Fact~\ref{fact:delta_mu}), the third step follows from Part (2) that $\| \ov{s}^{-1} \wt{\delta}_s \|_{\infty} \leq 2\epsilon$ and $\| \ov{x}^{-1} \wt{\delta}_x \|_{\infty} \leq 2\epsilon$, and the last step follows from $\epsilon_{\mathrm{mp}}\leq 10^{-4}$ (Assumption~\ref{ass:epsilon}).
\end{proof}

\begin{claim}[Part 4 of Lemma~\ref{lem:stochastic_step}, bounding the infinity norm with high probability]\label{cla:infitynorm_high_prob}
Given $b\geq 1000\log^2n$,\footnote{This assumption is added in Part 3 of Assumption~\ref{ass:assumption2}.} we have
\begin{align*}
\| \ov{x}^{-1} \wh{\delta}_x \|_{\infty} \leq 3\epsilon, ~ \| \ov{s}^{-1} \wh{\delta}_s \|_{\infty} \leq 3\epsilon,
\end{align*}
holds with probability $1-1/n^4$.
\end{claim}

\begin{proof}

By triangle inequality, we have \[\|\ov{s}^{-1}\wh{\delta}_s\|_{\infty}\leq \|\ov{s}^{-1}\wt{\delta}_s\|_{\infty} + \|\ov{s}^{-1}(\wh{\delta}_s-\wt{\delta}_s)\|_{\infty}.\]

The first term is upper bounded by $\|\ov{s}^{-1}\wt{\delta}_s\|_{\infty}\leq 2\epsilon$ (Part 2 of Claim~\ref{cla:infitynorm}). The second part involves randomness, therefore we need to prove that it holds with high probability. Note that $\wh{\delta}_s$ is the unbiased estimation of $\wt{\delta}_s$, i.e. $\E[\wh{\delta}_s]=\wt{\delta}_s$. We have
\begin{align}\label{eq:wh_s_minus_wt_s}
 \wh{\delta}_s-\wt{\delta_s}
= \frac{\wt{S}}{\sqrt{\wt{X}\wt{S}}}\left(R^{\top}R\wt{P}\frac{1}{\sqrt{\wt{X}\wt{S}}}\wt{\delta}_{\mu} - \wt{P}\frac{1}{\sqrt{\wt{X}\wt{S}}}\wt{\delta}_{\mu}\right)
= \frac{\wt{S}}{\sqrt{\wt{X}\wt{S}}}\left(R^{\top}R h - h\right),
\end{align}
where the first steps is by definition of $\wh{\delta}_s$(Definition~\ref{def:hat}) and $\wt{\delta}_s$ (Definition~\ref{def:widetilde}), and in the second step we define $h := \wt{P}\frac{1}{\sqrt{\wt{X}\wt{S}}}\wt{\delta}_{\mu}$. And by Eq.\eqref{eq:upper_bound_P_xs_delta_mu}, we have $\|h\|_2\leq 1.1\epsilon\sqrt{t}$.

%The next task is to upper bound $\|h\|_2$, we have
%\begin{align*}
%\|h\|_2 = &~ \|\wt{P}\frac{1}{\sqrt{\wt{X}\wt{S}}}\wt{\delta}_{\mu}\|_2 \\
%\leq & ~ 1.1\epsilon \sqrt{t},
%\end{align*}
%where the last step is from Equation~\eqref{eq:upper_bound_P_xs_delta_mu}.

Definition~\ref{def:song_coordiante_embedding} and Lemma~\ref{lem:E.5} guarantee that for any vector $h \in \R^n$, a subsample randomized Hadamard transform matrix $R\in \R^{b\times n}$ satisfies 
\begin{align*}
\Pr_{R} \left[ | (R^\top R h)_i - h_i | > \| h \|_2 \cdot \frac{ \log(n/\delta) }{ \sqrt{b} } \right] \leq \delta .
\end{align*}
In every iteration we use a fresh subsample Hadamard matrix $R$ which is independent of $h$, therefore we can apply this bound using the same $h$ and failure probability $\delta=1/n^4$, and we have that with probability at least $1-1/n^4$, $|(R^{\top}R h)_i-h_i| \leq \frac{5.5\epsilon \sqrt{t}\log n}{\sqrt{b}}$. Therefore, 

\begin{align}\label{eq:s_wh_s_wt_s}
\left|\ov{s}_i^{-1}(\wh{\delta}_s-\wt{\delta}_s)_i\right| 
& = \left|\frac{\ov{s}_i^{-1}\wt{s}_i}{\sqrt{\wt{x}_i\wt{s}_i}}\left(R^{\top}R h_i - h_i \right)\right| 
 \leq \left|\frac{1+2\epsilon_{\mathrm{mp}}}{\sqrt{0.9(1-\epsilon_{\mathrm{mp}})t}} \left(R^{\top}R h_i - h_i \right)\right|\notag \\
& \leq \left|\frac{1+2\epsilon_{\mathrm{mp}}}{\sqrt{0.9(1-\epsilon_{\mathrm{mp}})t}} \frac{5.5\epsilon \sqrt{t}\log n}{\sqrt{b}}\right|
~\leq~ \epsilon,
\end{align}
where the first step is by Eq.~\eqref{eq:wh_s_minus_wt_s}, the second step is because $\wt{s} \approx_{2\epsilon_{\mathrm{mp}}} \ov{s}$ (Part 2 of Fact~\ref{fact:delta_mu}) and $\wt{x}\wt{s}=\wt{\mu}\approx_{\epsilon_{\mathrm{mp}}} \ov{\mu} \approx_{0.1} t$ (Part 1 and 2 of Assumption~\ref{ass:assumption}), and the third step is by the upper bound on $|(R^{\top}R h)_i-h_i|$, the last step follows by $b\geq 1000 \log^2n$ and $\epsilon_{\mathrm{mp}}\leq 10^{-4}$ (Assumption~\ref{ass:epsilon}).

Finally, we have
\begin{align*}
\|\ov{s}^{-1}\wh{\delta}_s\|_{\infty}
\leq ~ \|\ov{s}^{-1}\wt{\delta}_s\|_{\infty} + \|\ov{s}^{-1}(\wh{\delta}_s-\wt{\delta}_s)\|_{\infty}
\leq ~ 2\epsilon + \|\ov{s}^{-1}(\wh{\delta}_s-\wt{\delta}_s)\|_{\infty}
%\leq & ~ 2\epsilon+\epsilon 
\leq ~ 3\epsilon,
\end{align*}
where the second step is by $\|\ov{s}^{-1}\wt{\delta}_s\|_\infty \leq 2\epsilon$ (Part 2 of Claim~\ref{cla:infitynorm}), the third step is by Eq.\eqref{eq:s_wh_s_wt_s}.

Similarly, we can show $\|\ov{x}^{-1}\wh{\delta}_x\|_\infty \leq 3\epsilon$ with probability $1 - 1/n^4$.
\end{proof}

\subsection{Bounding \texorpdfstring{$\ov{\mu}^{\new} - \ov{\mu}$}{}}\label{sec:bounding_mu_new_minus_mu}

\begin{table}[!t]
\small
\centering
\begin{tabular}{ | l | l | l | l | l | }
\hline
{\bf Quantity} & {\bf Bound} & {\bf Part} & {\bf Prob.} & {\bf Use Lem.~\ref{lem:stochastic_step}}\\ \hline
$\| \E[ \ov{\mu}^{-1} ( \ov{\mu}^{\new} - \ov{\mu} - \ov{\delta}_{t} - \wt{\delta}_{\Phi} ) ] \|_2 $ & $ \epsilon_{\mathrm{mp}} \epsilon + \epsilon^2 + \epsilon^2 \sqrt{n} / b$ & Part 1 & 1 & Part 1,2 %Part 1,2,3
\\ \hline
$\Var[ \ov{\mu}_i^{-1} \ov{\mu}_i^{\new} ] $ & $ \epsilon_{\mathrm{mp}}^2\epsilon^2/b + \epsilon^4/b$ & Part 2 & $1-1/\poly(n)$ & Part 2,4 \\ \hline
$\| \ov{\mu}^{-1} ( \ov{\mu}^{\new} - \ov{\mu} ) \|_{\infty} $ & $ \epsilon$ & Part 3 & $1-1/\poly(n)$ & Part 3,4\\ \hline
$\| \E[ \ov{\mu}^{-1} ( \ov{\mu}^{\new} - \ov{\mu} ) ] \|_2 $ & $ \epsilon + \epsilon^2 \sqrt{n} / b$ & Part 4 & 1 & Part 1%Part 1,2,3 
\\ \hline
\end{tabular}\caption{Summary of Lemma~\ref{lem:bounding_mu_new_minus_mu}. We ignore the constants.}
\end{table}

The goal of this section is to prove Lemma~\ref{lem:bounding_mu_new_minus_mu}.
\begin{lemma}[A deep version of Lemma 4.8 in \cite{cls19}]\label{lem:bounding_mu_new_minus_mu}
Let $\ov{\mu}$ and $\ov{\mu}^{\new}$ be defined as that of Definition~\ref{def:overline} and Definition~\ref{def:new}: $\ov{\mu} = \ov{x} \cdot \ov{s}$, and $\ov{\mu}^{\new} = ( \ov{x} + \wh{\delta}_x ) ( \ov{s} + \wh{\delta}_s )$.
We have\\
1. $\| \E[ \ov{\mu}^{-1} ( \ov{\mu}^{\new} - \ov{\mu} - \ov{\delta}_{t} - \wt{\delta}_{\Phi} ) ] \|_2 \leq 9 \epsilon_{\mathrm{mp}} \epsilon + 4 \epsilon^2 + 2\epsilon^2 \sqrt{n} / b$,\\
2. $\Var[ \ov{\mu}_i^{-1} \ov{\mu}_i^{\new} ] \leq 16\epsilon_{\mathrm{mp}}^2\epsilon^2/b+320\epsilon^4/b$ holds with probability at least $1-1/\poly(n)$ for all $i\in [n]$, \\
3. $\| \ov{\mu}^{-1} ( \ov{\mu}^{\new} - \ov{\mu} ) \|_{\infty} \leq 6\epsilon$, \\
4. $\| \E[ \ov{\mu}^{-1} ( \ov{\mu}^{\new} - \ov{\mu} ) ] \|_2 \leq 6 \epsilon + 2\epsilon^2 \sqrt{n} / b$.
\end{lemma}

\begin{claim}[Part 1 of Lemma~\ref{lem:bounding_mu_new_minus_mu}]
$
\| \E[ \ov{\mu}^{-1} ( \ov{\mu}^{\new} - \ov{\mu} - \ov{\delta}_{t} - \wt{\delta}_{\Phi} ) ] \|_2 \leq 9 \epsilon_{\mathrm{mp}} \epsilon + 4\epsilon^2 + 2 \epsilon^2 \sqrt{n} / b. 
$
\end{claim}

\begin{proof}
From the definition of $\ov{\mu}^{\new}$, we have
\begin{align} \label{eq:ov_mu_new_expand_terms}
\ov{\mu}^{\new} = & ( \ov{x} + \wh{\delta}_x ) ( \ov{s} + \wh{\delta}_s ) 
=  \ov{\mu} + \ov{x} \wh{\delta}_s + \ov{s} \wh{\delta}_x + \wh{\delta}_x \wh{\delta}_s \notag \\
= & \ov{\mu} +  ( \wt{x} \wh{\delta}_s + \wt{s} \wh{\delta}_x ) + ( \ov{x} - \wt{x} ) \wh{\delta}_s + ( \ov{s} - \wt{s} ) \wh{\delta}_x + \wh{\delta}_x \wh{\delta}_s \notag \\
= & \ov{\mu} + ( \wt{\delta}_{t} + \wt{\delta}_{\Phi} ) + ( \ov{x} - \wt{x} ) \wh{\delta}_s + ( \ov{s} - \wt{s} ) \wh{\delta}_x + \wh{\delta}_x \wh{\delta}_s \notag \\
= & \ov{\mu} + (\ov{\delta}_t + \wt{\delta}_{\Phi}) + (\wt{\delta}_t - \ov{\delta}_t) + ( \ov{x} - \wt{x} ) \wh{\delta}_s + ( \ov{s} - \wt{s} ) \wh{\delta}_x + \wh{\delta}_x \wh{\delta}_s,
\end{align}
where in the forth step we use the fact $\wt{x} \wh{\delta}_s + \wt{s} \wh{\delta}_x = \wt{\delta}_{\mu} = \wt{\delta}_{t} + \wt{\delta}_{\Phi}$ (Part 1 of Fact~\ref{fact:delta_mu}).
Subtracting $\ov{\mu} + (\ov{\delta}_t + \wt{\delta}_{\Phi})$ on both sides and taking the expectation, we have
\[
\E[\ov{\mu}^{\new} - \ov{\mu} - \ov{\delta}_t - \wt{\delta}_{\Phi} ] = (\wt{\delta}_t - \ov{\delta}_t) + ( \ov{x} - \wt{x} ) \E[ \wh{\delta}_s ] + ( \ov{s} - \wt{s} ) \E[ \wh{\delta}_x ] + \E[ \wh{\delta}_x \wh{\delta}_s ] .
\]

Hence, we have that
\begin{align}\label{eq:bounding_mu_new_minus_mu_l2_norm_1}
 & ~ \| \ov{\mu}^{-1}  \E [ \ov{\mu}^{\new}  - \ov{\mu} - \ov{\delta}_t - \wt{\delta}_{\Phi} ]  \|_2 \notag \\
\leq & ~ \|\ov{\mu}^{-1} (\wt{\delta}_t - \ov{\delta}_t)\|_2 + \| \ov{\mu}^{-1} (\ov{x} - \wt{x}) \ov{s} \cdot \ov{s}^{-1} \E[ \wh{\delta}_s ] \|_2 + \| \ov{\mu}^{-1} ( \ov{s} - \wt{s} ) \ov{x} \cdot \ov{x}^{-1} \E[ \wh{\delta}_x ] \|_2 + \| \ov{\mu}^{-1} \E[ \wh{\delta}_x \wh{\delta}_s ] \|_2 \notag \\
\leq & ~ \epsilon_{\mathrm{mp}} \cdot \epsilon + \| \ov{\mu}^{-1} (\ov{x} - \wt{x}) \ov{s} \cdot \ov{s}^{-1} \E[ \wh{\delta}_s ] \|_2 + \| \ov{\mu}^{-1} ( \ov{s} - \wt{s} ) \ov{x} \cdot \ov{x}^{-1} \E[ \wh{\delta}_x ] \|_2 + \| \ov{\mu}^{-1} \E[ \wh{\delta}_x \wh{\delta}_s ] \|_2 \notag \\
\leq & ~ \epsilon_{\mathrm{mp}} \cdot \epsilon + \| \ov{\mu}^{-1} (\ov{x} - \wt{x}) \ov{s} \|_{\infty} \cdot \| \ov{s}^{-1} \E [ \wh{\delta}_s ] \|_2 + \| \ov{\mu}^{-1} ( \ov{s} - \wt{s} ) \ov{x} \|_{\infty} \cdot \| \ov{x}^{-1} \E [ \wh{ \delta}_x ] \|_2 + \| \ov{\mu}^{-1} \E[ \wh{\delta}_x \wh{\delta}_s ] \|_2 \notag \\
\leq & ~ \epsilon_{\mathrm{mp}} \cdot \epsilon + 2\epsilon_{\mathrm{mp}} \cdot \| \ov{s}^{-1} \E[ \wh{\delta}_s ] \|_2 + 2\epsilon_{\mathrm{mp}} \cdot \| \ov{x}^{-1} \E[ \wh{\delta}_x ] \|_2 + \| \ov{\mu}^{-1} \E[ \wh{\delta}_x \wh{\delta}_s ] \|_2 \notag \\
\leq & ~  9\epsilon_{\mathrm{mp}}\cdot \epsilon + \| \ov{\mu}^{-1} \E [ \wh{\delta}_x \wh{\delta}_s ] \|_2,
\end{align}
where the first step follows by triangle inequality, the second step follows by Part 1 of Lemma~\ref{lem:stochastic_step}, the third step follows by $\| a b\|_2 \leq \| a \|_{\infty} \cdot \| b \|_2$, the forth step follows by $\| \ov{\mu}^{-1} ( \ov{x} - \wt{x} ) \ov{s} \|_{\infty} \leq 2\epsilon_{\mathrm{mp}}$ and $\| \ov{\mu}^{-1} ( \ov{s} - \wt{s} ) \ov{x} \|_{\infty} \leq 2\epsilon_{\mathrm{mp}}$ (since $\wt{x} \approx_{2\epsilon_{\mathrm{mp}}} \ov{x}$, $\wt{s} \approx_{2\epsilon_{\mathrm{mp}}} \ov{s}$ by Part 2 of Fact~\ref{fact:delta_mu}, and $\ov{\mu}=\ov{x}\cdot \ov{s}$ by Definition~\ref{def:overline}), the last step follows by $\| \E [ \ov{s}^{-1} \wh{\delta}_s ] \|_2 \leq 2 \epsilon$ and $\| \E [ \ov{x}^{-1} \wh{\delta}_x ] \|_2 \leq 2 \epsilon$ (Part 1 of Lemma~\ref{lem:stochastic_step}). 

To bound the last term of Eq.~\eqref{eq:bounding_mu_new_minus_mu_l2_norm_1}, using $\E[\wh{\delta}_s] = \wt{\delta}_s$ and $\E[\wh{\delta}_x ] = \wt{\delta}_x$, we have that
\begin{align*}
\E[ \wh{\delta}_{x,i} \wh{\delta}_{s,i} ] = \wt{\delta}_{x,i} \wt{\delta}_{s,i} + \E [ ( \wh{\delta}_{x,i} - \wt{\delta}_{x,i} ) ( \wh{\delta}_{s,i} - \wt{\delta}_{s,i} ) ].
\end{align*}
Hence, we have
\begin{align}\label{eq:bounding_mu_new_minus_mu_l2_norm_2}
\| \ov{\mu}^{-1} \E [ \wh{\delta}_x \wh{\delta}_s ] \|_2 
\leq & ~ \| \ov{\mu}^{-1} \wt{\delta}_x \wt{\delta}_s \|_2 + \left( \sum_{i=1}^n \left( \E \left[ \ov{x}_i^{-1} ( \wh{\delta}_{x,i} - \wt{\delta}_{x,i} ) \cdot \ov{s}_i^{-1} ( \wh{\delta}_{s,i} - \wt{\delta}_{s,i} ) \right] \right)^2 \right)^{1/2} \notag \\
\leq & ~ 4\epsilon^2 + \frac{1}{2} \left( \sum_{i=1}^n \left( \Var[  \ov{x}_i^{-1} \wh{\delta}_{x,i} ]  + \Var[ \ov{s}_i^{-1} \wh{\delta}_{s,i}  ] \right)^2 \right)^{1/2} \notag \\
\leq & ~ 4\epsilon^2 + \frac{1}{2} \left( \sum_{i=1}^n 2 ( \Var[ \ov{x}_i^{-1} \wh{\delta}_{x,i} ] )^2 + 2 ( \Var[ \ov{s}_i^{-1} \wh{\delta}_{s,i} ] )^2 \right)^{1/2} \notag \\
\leq & ~ 4\epsilon^2 + 2 \sqrt{ n \cdot \epsilon^4 / b^2  } = 4\epsilon^2 + 2 \epsilon^2 \sqrt{ n } / b,
\end{align}
where the first step follows from triangle inequality and $\ov{\mu} = \ov{x}\ov{s}$, the second step follows by $\| \ov{\mu}^{-1} \wt{\delta}_x \wt{\delta}_s \|_2 \leq \| \ov{x}^{-1} \wt{\delta}_x \|_2 \cdot \| \ov{s}^{-1} \wt{\delta}_s \|_2 \leq 4\epsilon^2$ (Part 1 of Lemma~\ref{lem:stochastic_step}) and $2ab \leq a^2 + b^2$, the third step follows by $(a+b)^2 \leq 2 a^2 + 2 b^2$, the fourth step follows by $\Var[ \ov{x}_i^{-1} \wh{\delta}_{x,i} ] \leq 2 \epsilon^2 / b$ and $\Var[ \ov{s}_i^{-1} \wh{\delta}_{s,i} ] \leq 2 \epsilon^2 / b$ (Part 2 of Lemma~\ref{lem:stochastic_step}).

Finally, we have that
\begin{align*}
\| \ov{\mu}^{-1} ( \E [ \ov{\mu}^{\new} - \ov{\mu} - \ov{\delta}_t - \wt{\delta}_{\Phi} ] ) \|_2 
\leq ~ 9\epsilon_{\mathrm{mp}}  \epsilon + \| \ov{\mu}^{-1} \E[ \wh{\delta}_x \wh{\delta}_s ] \|_2 
\leq ~ 9 \epsilon_{\mathrm{mp}} \epsilon  + 4\epsilon^2 + 2 \epsilon^2 \sqrt{n} / b.
\end{align*}
where the first step follows from Eq.~\eqref{eq:bounding_mu_new_minus_mu_l2_norm_1}, and the last step follows from Eq.~\eqref{eq:bounding_mu_new_minus_mu_l2_norm_2}. 
\end{proof}

\begin{claim}[Part 4 of Lemma~\ref{lem:bounding_mu_new_minus_mu}]
We have
\begin{align*}
\| \E[ \ov{\mu}^{-1} ( \ov{\mu}^{\new} - \ov{\mu} ) ] \|_2 \leq 6 \epsilon + 2 \epsilon^2 \sqrt{n} / b.
\end{align*}
\end{claim}
\begin{proof}
From Part 1 of Lemma~\ref{lem:stochastic_step}, we know that $\| \ov{\mu}^{-1} ( \ov{\delta}_t + \wt{\delta}_{\Phi} ) \|_2 \leq 5 \epsilon$. Thus using triangle inequality and Part 1 of Lemma~\ref{lem:bounding_mu_new_minus_mu}, we know
\begin{align*}
\| \ov{\mu}^{-1} ( \E [ \ov{\mu}^{\new} - \ov{\mu} ] ) \|_2 \leq & ~ \| \ov{\mu}^{-1} ( \E [ \ov{\mu}^{\new} - \ov{\mu} - \ov{\delta}_t - \wt{\delta}_{\Phi} ] ) \|_2 + \| \ov{\mu}^{-1} ( \ov{\delta}_t + \wt{\delta}_{\Phi} ) \|_2 \\
\leq & ~ 9 \epsilon_{\mathrm{mp}} \epsilon   + 4\epsilon^2 +2\epsilon^2\sqrt{n}/b + 5\epsilon 
\leq  6 \epsilon + 2\epsilon^2\sqrt{n}/b,
\end{align*}

where the last step follows by $\epsilon_{\mathrm{mp}}< 10^{-4}$ and $\epsilon < 10^{-4}$ (Assumption~\ref{ass:epsilon}).
\end{proof}

\begin{claim}[Part 2 of Lemma~\ref{lem:bounding_mu_new_minus_mu}]
$\Var[ \ov{\mu}_i^{-1} \ov{\mu}_i^{\new} ] \leq 16\epsilon_{\mathrm{mp}}^2\epsilon^2/b+320\epsilon^4/b$ holds with probability at least $1-1/\poly(n)$ for all $i\in [n]$.
% $\E[ (\mu_i^{-1} \mu_i^{\new}) ] \lesssim \epsilon^2 / b$ for all $i$.
\end{claim}
\begin{proof}
Recall that we showed in Eq.~\eqref{eq:ov_mu_new_expand_terms} that
\begin{align*}
\ov{\mu}^{\new} 
%= &~ (\ov{x}+\wh{\delta}_{x})(\ov{s}+\wh{\delta}_{s})\\
%= &~ \ov{\mu}+\wt{x}\wh{\delta}_s+(\ov{x}-\wt{x})\wh{\delta}_s+\wt{s}\wh{\delta}_x+(\ov{s}-\wt{s})\wh{\delta}_x+\wh{\delta}_x\wh{\delta}_s\\
= &~ \ov{\mu} + \wt{\delta}_{\mu} + ( \ov{x} - \wt{x} ) \wh{\delta}_s + ( \ov{s} - \wt{s} ) \wh{\delta}_x + \wh{\delta}_x \wh{\delta}_s.
\end{align*}

We compute the variance of each of the terms in this formula. For $(\ov{x} - \wt{x})\wh{\delta}_s$ we have
\begin{align}\label{eq:mu_diff_x_delta_s}
\Var[ \ov{\mu}_i^{-1} (\ov{x}_i - \wt{x}_i) \wh{\delta}_{s,i} ]  
= \Var[ \ov{x}_i^{-1} (\ov{x}_i - \wt{x}_i) \ov{s}_i^{-1} \wh{\delta}_{s,i} ] 
\leq  4 \epsilon_{\mathrm{mp}}^2 \Var[ \ov{s}_i^{-1} \wh{\delta}_{s,i} ] 
\leq ~ 8\epsilon_{\mathrm{mp}}^2\epsilon^2/b
\end{align}
where the second step is by $\ov{x}\approx_{2\epsilon_{\mathrm{mp}}}\wt{x}$ (Part 2 of Fact~\ref{fact:delta_mu}), and the third step is by $\Var[ \ov{s}_i^{-1} \wh{\delta}_{s,i} ] \leq 2 \epsilon^2 / b$ (Part 2 of Lemma~\ref{lem:stochastic_step}).

And similarly for $(\ov{s} - \wt{s}) \wh{\delta}_x$ we can show
\begin{align}\label{eq:mu_diff_s_delta_x}
    \Var[ \ov{\mu}_i^{-1} (\ov{s}_i - \wt{s}_i) \wh{\delta}_{x,i} ] \leq 8\epsilon_{\mathrm{mp}}^2\epsilon^2/b.
\end{align}

Now we can upper bound the variance of $\ov{\mu}_i^{-1} \ov{\mu}_i^{\new}$,
\begin{align*}
 \Var[ \ov{\mu}_i^{-1} \ov{\mu}_i^{\new} ] 
\leq & ~ 4 \Var [ \ov{\mu}_i^{-1} \wt{\delta}_{\mu,i} ] + 4 \Var [ \ov{\mu}_i^{-1} (\ov{x}_i -\wt{x}_i) \wh{\delta}_{s,i} ] + 4 \Var [ \ov{\mu}_i^{-1} ( \ov{s}_i - \wt{s}_i ) \wh{\delta}_{x,i} ] + 4 \Var [ \ov{\mu}_i^{-1} \wh{\delta}_{x,i} \wh{\delta}_{s,i} ] \\
\leq & ~ 4\cdot 0+ 8\epsilon_{\mathrm{mp}}^2\epsilon^2/b + 8\epsilon_{\mathrm{mp}}^2\epsilon^2/b + 4 \Var[ \ov{\mu}_i^{-1} \wh{\delta}_{x,i} \wh{\delta}_{s,i} ] \\
= & ~ 16\epsilon_{\mathrm{mp}}^2\epsilon^2/b + 4\Var[ \ov{x}_i^{-1} \wh{\delta}_{x,i} \cdot \ov{s}_i^{-1} \wh{\delta}_{s,i} ] \\
\leq & ~ 16\epsilon_{\mathrm{mp}}^2\epsilon^2/b +8 \Sup[(\ov{x}_i^{-1} \wh{\delta}_{x,i} )^2] \cdot \Var[\ov{s}_i^{-1} \wh{\delta}_{s,i}] + 8 \Sup[(\ov{s}_i^{-1} \wh{\delta}_{s,i} )^2] \cdot \Var[ \ov{x}_i^{-1} \wh{\delta}_{x,i} ]  \\
\leq & ~ 16\epsilon_{\mathrm{mp}}^2\epsilon^2/b  + 8 \cdot (3\epsilon)^2 \cdot \frac{ 2\epsilon^2 }{ b } + 8\cdot (3\epsilon)^2 \cdot \frac{ 2\epsilon^2 }{ b } \\
\leq & ~16 \epsilon_{\mathrm{mp}}^2 \epsilon^2 / b + 320 \epsilon^4 / b,
\end{align*}
where the first step follows from triangle inequality and the fact that $\Var[1]=0$, the second step follows by $\Var[ \ov{\mu}_i^{-1} \wt{\delta}_{\mu,i} ] =0$ (since $\ov{\mu}_i^{-1}$ and $\wt{\delta}_{\mu,i}$ don't involve randomness) and plugging in Eq.~\eqref{eq:mu_diff_x_delta_s} and Eq.~\eqref{eq:mu_diff_s_delta_x}, the third step follows by $\ov{\mu} = \ov{x} \cdot \ov{s}$ (Definition~\ref{def:overline}), the fourth step follows by $\Var[xy]\leq 2 \Sup[x^2] \Var[y] + 2 \Sup[y^2] \Var[x]$ (Lemma~\ref{lem:var_xy}) with $\Sup$ denoting the deterministic maximum of the random variable, the fifth step follows by $\Var[ \ov{s}_i^{-1} \wh{\delta}_{s,i} ] \leq 2 \epsilon^2 / b$ and $\Var[ \ov{x}_i^{-1} \wh{\delta}_{x,i} ] \leq 2 \epsilon^2 / b $ (Part 2 of Lemma~\ref{lem:stochastic_step}) and $\|\ov{x}^{-1}\hat{\delta}_x\|_{\infty}\leq 3\epsilon$ and $ \|\ov{s}^{-1}\hat{\delta}_s\|_{\infty}\leq 3\epsilon$ (Part 4 of Lemma~\ref{lem:stochastic_step}).
\end{proof}

\begin{claim}[Part 3 of Lemma~\ref{lem:bounding_mu_new_minus_mu}]
$\| \ov{\mu}^{-1} ( \ov{\mu}^{\new} - \ov{\mu} ) \|_{\infty} \leq 6\epsilon$ holds with probability at least $1-1/\poly(n)$.
\end{claim}
\begin{proof}
We again note that from Eq.~\eqref{eq:ov_mu_new_expand_terms} we have
\begin{align*}
\ov{\mu}^{\new} = \ov{\mu} + \wt{\delta}_{\mu} + ( \ov{x} - \wt{x} ) \wh{\delta}_s + ( \ov{s} - \wt{s} ) \wh{\delta}_x + \wh{\delta}_x \wh{\delta}_s.
\end{align*}
Hence, we have that with probability at least $1-1/n^4$ the following is true:
\begin{align}\label{eq:mu_mu_new_mu_delta_mu}
| \ov{\mu}_i^{-1} ( \ov{\mu}_i^{\new} - \ov{\mu}_i - \wt{\delta}_{\mu,i} ) |
\leq & ~ | ( \ov{x} - \wt{x} )_i \ov{\mu}_i^{-1} \wh{\delta}_{s,i} | + | (\ov{s} - \wt{s})_i \ov{\mu}_i^{-1} \wh{\delta}_{x,i} | + | \ov{\mu}_i^{-1} \wh{\delta}_{x,i} \wh{\delta}_{s,i} | \notag \\
= & ~ | ( \ov{x} - \wt{x} )_i \ov{x}_i^{-1} | \cdot | \ov{s}_i^{-1} \wh{\delta}_{s,i} | + | (\ov{s} - \wt{s})_i \ov{s}_i^{-1} | \cdot | \ov{x}_i^{-1} \wh{\delta}_{x,i} | + | \ov{x}_i^{-1} \wh{\delta}_{x,i} | \cdot | \ov{s}_i^{-1} \wh{\delta}_{s,i} | \notag\\
\leq & ~ 2\epsilon_{\mathrm{mp}} | \ov{s}_i^{-1} \wh{\delta}_{s,i} | + 2\epsilon_{\mathrm{mp}} | \ov{x}_i^{-1} \wh{\delta}_{x,i} | + | \ov{x}_i^{-1} \wh{\delta}_{x,i} | \cdot | \ov{s}_i^{-1} \wh{\delta}_{s,i} | \notag \\
\leq & ~ 2\epsilon_{\mathrm{mp}} \cdot 3\epsilon + 2\epsilon_{\mathrm{mp}} \cdot 3\epsilon + (3\epsilon)^2 \notag \\
\leq & ~ 20\epsilon_{\mathrm{mp}}\cdot \epsilon + 10\epsilon^2,
\end{align}
where the first step follows by triangle inequality, the second step follows by $\ov{\mu}_i = \ov{x}_i \ov{s}_i$ (Definition~\ref{def:overline}), the third step follows by $\ov{x} \approx_{2\epsilon_{\mathrm{mp}}} \wt{x}$ and $\ov{s} \approx_{2\epsilon_{\mathrm{mp}}} \wt{s}$ (Part 2 of Fact~\ref{fact:delta_mu}), the forth step follows by $ | \ov{s}_i^{-1} \wh{\delta}_{s,i} | \leq 3\epsilon$ and $| \ov{x}_i^{-1} \wh{\delta}_{x,i} | \leq 3\epsilon$ holds with $1 - 1/ n^4$ (Part 4 of Lemma~\ref{lem:stochastic_step}).

Finally, we have
\begin{align*}
| \ov{\mu}_i^{-1} ( \ov{\mu}_i^{\new} - \ov{\mu}_i ) |
\leq & ~ | \ov{\mu}_i^{-1} ( \ov{\mu}_i^{\new} - \ov{\mu}_i -\wt{\delta}_{\mu,i}) | + |\ov{\mu}_i^{-1} \wt{\delta}_{\mu,i}|
\leq  20\epsilon_{\mathrm{mp}}\cdot \epsilon + 10\epsilon^2 + |\ov{\mu}_i^{-1} \wt{\delta}_{\mu,i}|\\
\leq & ~ 20\epsilon_{\mathrm{mp}}\cdot \epsilon + 10\epsilon^2 +5\epsilon 
~\leq ~ 6\epsilon,
\end{align*}
where the first step follows from triangle inequality, the second step follows from Eq.\eqref{eq:mu_mu_new_mu_delta_mu}, and the third step follows from $|\ov{\mu}_i^{-1} \wt{\delta}_{\mu,i}|\leq 5\epsilon$ (Part 3 of Lemma~\ref{lem:stochastic_step}), and fourth step follows from $\epsilon, \epsilon_{\mathrm{mp}}\leq 10^{-4}$ (Assumption~\ref{ass:epsilon}).
\end{proof}

\subsection{Potential martingale}\label{sec:potential_martingale}

%\begin{table}[!t]
%\caption{Summary of choice of $\epsilon$, $\epsilon_{\mathrm{mp}}$, $\lambda$ and $b$. Assigned in \textsc{Main} procedure (Algorithm~\ref{alg:main}). They are used to prove Theorem~\ref{thm:main_0.5_and_a_omega_copied}.}
%\label{table:parameters}
%\centering
%\begin{tabular}{ | l | l | l |}
%\hline
%{\bf Notation} & {\bf Our Choice} & {\bf \cite{cls19} Choice} \\ \hline
%$\epsilon$ & $10^{-6}/\log n$ & $1/40000$ \\ \hline
%$\epsilon_{\mathrm{mp}}$ & $10^{-5}$ & $ 1/(40000\log n)$  \\ \hline 
%$\lambda$ & $40\log n$& $40 \log n$ \\ \hline 
%$b$ & $20000\sqrt{n}\log^2 n$ & $\frac{1000 \epsilon \sqrt{n} \log^2 n }{ %\epsilon_{\mathrm{mp}} }$ (sampling sparsity) \\
%\hline
%\end{tabular}	
%\end{table}

%\begin{table}[!t]
%\centering
%\begin{tabular}{ | l | l |}
%\hline
%{\bf Notation} & {\bf Choice} \\ \hline
%$\epsilon$ & $10^{-6}/\log n$ \\ \hline
%$\epsilon_{\mathrm{mp}}$ & $10^{-5}$   \\ \hline 
%$\lambda$ & $40\log n$\\ \hline 
%$b$ & $20000\sqrt{n}\log^2 n$\\
%\hline
%\end{tabular}
%\caption{Summary of choice of $\epsilon$, $\epsilon_{\mathrm{mp}}$, $\lambda$ and $b$. Assigned in \textsc{Main} procedure (Algorithm~\ref{alg:main_improved}). They are used to prove Theorem~\ref{thm:main_1.5_a_wta_copied}.}
%\label{table:parameters}
%\end{table}

\begin{table}[!t]
\small
\centering
\begin{tabular}{ | l | c | c | c | c |}
\hline
{\bf Notation} & $\epsilon$ & $\epsilon_{\mathrm{mp}}$ & $\lambda$ & $b$ \\ \hline
{\bf Choice} & $10^{-7}/\log n$ & $10^{-5}/\log n$ & $40\log n$ & $10^{22}\sqrt{n}\log^{10} n$ \\ 
\hline
\end{tabular}
\caption{Choice of $\epsilon$, $\epsilon_{\mathrm{mp}}$, $\lambda$ and $b$ that satisfies all constraints in Assumption~\ref{ass:assumption} and Assumption~\ref{ass:assumption2}. These parameters are assigned in \textsc{Main} procedure (Algorithm~\ref{alg:main_improved}). Later they are used to prove Theorem~\ref{thm:main_1.5_a_wta_copied}.
}
\label{table:parameters}
\end{table}

We first state the constraints of the parameters.
\begin{assumption}\label{ass:assumption2}
Let parameters $b, \lambda, \epsilon, \epsilon_{\mathrm{mp}}$ satisfying the following constraints: 
\begin{align*}
& 1. ~ b\geq 20000\cdot ( \lambda \epsilon_{\mathrm{mp}}^2 \epsilon \sqrt{n} + \epsilon^3 \sqrt{n} ),
&& 2. ~ \lambda \geq 30\log n,\\
& 3. ~ b\geq 1000\log^2n,
&& 4. ~ \frac{1}{30\lambda} \geq \frac{\epsilon}{\sqrt{n}}+8\epsilon,\\
& 5. ~ b \geq 20000 \epsilon \sqrt{n},
&& 6. ~ \lambda \epsilon < 10^{-5},\\
& 7. ~ \lambda \leq 60\log n,
&& 8. ~ 1.2 \epsilon_{\mathrm{mp}} < 1/30\lambda.
\end{align*}
\end{assumption}

Now we are ready to prove the main lemma for bounding the potential function. The goal of this section is to prove Lemma~\ref{lem:potential_martingale}.
\begin{lemma}[A deep version of Lemma 4.13 in \cite{cls19}]\label{lem:potential_martingale}
Under the Assumptions~\ref{ass:epsilon}, \ref{ass:assumption}, and~\ref{ass:assumption2}, we have
\begin{align*}
\E \left[ \Phi_{\lambda} \left( \frac{ \ov{\mu}^{\new} }{ t^{\new} } - 1 \right) \right] \leq \Phi_{\lambda} \left( \frac{ \ov{\mu} }{ t } - 1 \right) - \frac{\lambda \epsilon}{15\sqrt{n}} \left( \Phi_{\lambda} \left( \frac{ \ov{\mu} }{ t } - 1 \right) - 10 n \right).
\end{align*}
\end{lemma}
\begin{proof}
Let $\epsilon_{\mu} = \ov{\mu}^{\new} - \ov{\mu} - \ov{\delta}_t -\wt{\delta}_{\Phi}$. From this definition, we have
\begin{align*}
 \ov{\mu}^{\new} - t^{\new} 
= \ov{\mu} + \ov{\delta}_t +\wt{\delta}_{\Phi} + \epsilon_{\mu} - t^{\new},
\end{align*}
which implies
\begin{align}\label{eq:rewrite_mu_t_new}
\frac{ \ov{\mu}^{\new} }{ t^{\new} } - 1 = & ~ \frac{\ov{\mu}}{t^{\new}} + \frac{1}{t^{\new} } (  \ov{\delta}_t + \wt{\delta}_{\Phi} + \epsilon_{\mu} )  -  1  
=  \frac{\ov{\mu}}{t} + \frac{\ov{\mu}}{t} ( \frac{t}{t^{\new}} - 1 ) + \frac{1}{t^{\new} } ( \ov{\delta}_t + \wt{\delta}_{\Phi} + \epsilon_{\mu} )  -  1 \notag \\
= & ~ \frac{\ov{\mu}}{t} - 1 + \underbrace{ \frac{\ov{\mu}}{t} ( \frac{t}{t^{\new}} - 1 ) + \frac{1}{t^{\new} } ( \ov{\delta}_t + \wt{\delta}_{\Phi} + \epsilon_{\mu} ) }_{v}.
\end{align}

To apply Lemma~\ref{lem:potential_function_Phi_cosh} with $r = \ov{\mu} / t - 1$ and $r+v=\ov{\mu}^{\new}/t^{\new}-1$, we first compute $\E[v]$:
\begin{align}\label{eq:expectation_v}
\E[v] = & ~ \frac{\ov{\mu}}{t} ( \frac{t}{t^{\new}} - 1 ) + \frac{1}{t^{\new}}  ( \ov{\delta}_t + \wt{\delta}_{\Phi} + \E[ \epsilon_{\mu} ] ) \notag \\
= & ~ \frac{\ov{\mu}}{t} ( \frac{t}{t^{\new}} - 1 ) + \frac{1}{t^{\new}}  \left( (\frac{t^{\new}}{t} - 1)   \ov{\mu} - \frac{\epsilon}{2} t^{\new} \frac{ \nabla \Phi_{\lambda} (\wt{\mu}/t -1 ) }{ \| \nabla \Phi_{\lambda} (\wt{\mu}/t -1 ) \|_2 } + \E[ \epsilon_{\mu} ] \right) \notag \\
= & ~ - \frac{\epsilon}{2} \frac{ \nabla \Phi_{\lambda} (\wt{\mu}/t-1) }{ \| \nabla \Phi_{\lambda}(\wt{\mu}/t-1) \|_2} + \frac{1}{t^{\new}} \E[ \epsilon_{\mu} ],
\end{align}
where the second step follows by definition of $\ov{\delta}_{t}$ (Definition~\ref{def:overline}) and $\wt{\delta}_{\Phi}$ (Definition~\ref{def:widetilde}).

Next, we bound $\|v\|_\infty$ as follows:
\begin{align*}
\|v\|_{\infty} 
\leq & ~ \left\| \frac{\ov{\mu}}{t}(\frac{t}{t^{\new}}-1) \right\|_{\infty}+ \left\| \frac{1}{t^{\new}}(\ov{\mu}^{\new}-\ov{\mu}) \right\|_{\infty}
\leq \frac{\epsilon}{\sqrt{n}}+\frac{\|\ov{\mu}^{-1}(\ov{\mu}^{\new}-\ov{\mu})\|_{\infty}}{0.9} \\
\leq & ~ \frac{\epsilon}{\sqrt{n}}+8\epsilon
\leq ~\frac{1}{30\lambda},
\end{align*}
where the second step follows from $t^{\new} = (1 - \epsilon / (3 \sqrt{n})) \cdot t$ (Definition~\ref{def:new}) and $\ov{\mu} \approx_{0.1} t$ (Part 2 of Assumption~\ref{ass:assumption}), the third step follows from Part 3 of Lemma~\ref{lem:bounding_mu_new_minus_mu}, and the last step follows from Part 4 of Assumption~\ref{ass:assumption2} that $ \frac{1}{30\lambda}\geq \frac{\epsilon}{\sqrt{n}}+8\epsilon$.

Since $\|v\|_\infty\leq \frac{1}{30\lambda}$, we can apply Part 1 of Lemma~\ref{lem:potential_function_Phi_cosh} and get
\begin{align}\label{eq:split_E_potential}
\E [ \Phi_{\lambda} ( \ov{\mu} / t + v - 1 ) ]
\leq & ~ \Phi_{\lambda} ( \ov{\mu} / t - 1 ) +  \langle \nabla \Phi_{\lambda}( \ov{\mu} / t - 1 ), \E[v] \rangle + 2 \E [ \| v \|_{ \nabla^2 \Phi_{\lambda}(\ov{\mu} / t - 1) }^2 ] \notag \\
= & ~ \Phi_{\lambda} (\ov{\mu} / t - 1 ) + \underbrace{ \Big(- \frac{\epsilon}{2} \Big\langle  \nabla \Phi_{\lambda}( \ov{\mu} / t - 1 ), \frac{\nabla \Phi_{\lambda}( \wt{\mu} / t - 1 )}{\|\nabla \Phi_{\lambda}( \wt{\mu} / t - 1 )\|_2} \Big\rangle \Big)}_{a_1} \notag \\
& ~ + \underbrace{\frac{t}{t^{\new}}  \langle \nabla \Phi_{\lambda}( \ov{\mu} / t - 1 ), \E[ t^{-1} \epsilon_{\mu} ] \rangle}_{a_2} + \underbrace{2 \E [ \| v \|_{ \nabla^2 \Phi_{\lambda}( \ov{\mu}/t - 1)}^2 ]}_{a_3},
\end{align}
where the second step follows by Eq.~\eqref{eq:expectation_v}. 

We have $\|(\wt{\mu} - \ov{\mu})/t\| \leq 1.1\epsilon_{\mathrm{mp}} \leq \frac{1}{30\lambda}$ since $\wt{\mu} \approx_{\epsilon_{\mathrm{mp}}} \ov{\mu}$ and $\ov{\mu} \approx_{0.1} t$ (Assumption~\ref{ass:assumption}) and $1.2 \epsilon_{\mathrm{mp}} < 1/30\lambda$ (Part 8 of Assumption~\ref{ass:assumption2}).
So we can use Lemma~\ref{lem:phi_gradient_bound} and let $r\leftarrow \ov{\mu}/t - 1$ and $v \leftarrow (\wt{\mu} - \ov{\mu})/t - 1$ in the lemma statement to upper bound the $a_1$ term in Eq.~\eqref{eq:split_E_potential}:
\begin{align}\label{eq:split_E_potential_term_1}
a_1\leq - 0.45\epsilon \| \nabla \Phi_{\lambda}(\ov{\mu} / t - 1) \|_2 +0.1 \lambda  \epsilon \sqrt{n}. 
\end{align}
We upper bound $a_2$ term in Eq.~\eqref{eq:split_E_potential} as follows:
\begin{align}\label{eq:split_E_potential_term_2}
a_2 = & ~  \frac{t}{t^{\new}}  \langle \nabla \Phi_{\lambda}( \ov{\mu} / t - 1 ), \E[ t^{-1} \epsilon_{\mu} ] \rangle 
\leq  \frac{t}{t^{\new}} \| \nabla \Phi_{\lambda} (\ov{\mu} / t - 1) \|_2 \cdot \| \E[t^{-1} \epsilon_{\mu} ] \|_2 \notag \\
\leq & ~ 1.1 \| \nabla \Phi_{\lambda} (\ov{\mu} / t - 1) \|_2 \cdot \| \E[t^{-1} \epsilon_{\mu} ] \|_2 
\leq  2( 9 \epsilon_{\mathrm{mp}} \epsilon  + 4 \epsilon^2 + 2 \epsilon^2 \sqrt{n} / b ) \| \nabla \Phi_{\lambda}( \ov{\mu} / t - 1 ) \|_2 
\end{align}
where the second step follows by $\langle a, b \rangle \leq \| a \|_2 \cdot \| b \|_2$, the third step follows from definition of $t^{\new}$ (Definition~\ref{def:new}), the forth step follows by $\| \E[ \ov{\mu}^{-1} \epsilon_{\mu} ] \|_2 \leq 9 \epsilon_{\mathrm{mp}} \epsilon + 4\epsilon^2 + 2 \epsilon^2 \sqrt{n} / b$ (Part 1 of Lemma~\ref{lem:bounding_mu_new_minus_mu}) and $\ov{\mu}\approx_{0.1} t$ (Part 2 of Assumption~\ref{ass:assumption}).

We still need to bound $ a_3 = 2\E[ \| v \|_{ \nabla^2 \Phi_{\lambda}( \ov{\mu} / t - 1 ) }^2 ]$ term in Eq.~\eqref{eq:split_E_potential}. Before bounding it, we first bound $\E[v_i^2]$,
\begin{align}\label{eq:bounding_v_i_square}
\E[v_i^2] 
\leq & ~ 2 \E\left[  \left( \frac{\ov{\mu}_i }{t} ( \frac{t}{ t^{\new} } - 1 ) \right)^2 \right] + 2 \E \left[ \left( \frac{1}{t^{\new}} (\ov{\mu}_i^{\new}-\ov{\mu}_i) \right)^2 \right] 
\leq  \epsilon^2 / n + 3 \E \left[ (  (\ov{\mu}^{\new}_i - \ov{\mu}_i) / \ov{\mu}_i )^2 \right] \notag \\
= & ~ \epsilon^2 / n + 3 \Var[ ( \ov{\mu}_i^{\new} - \ov{\mu}_i ) / \ov{\mu}_i ] + 3 ( \E[  ( \ov{\mu}_i^{\new} - \ov{\mu}_i ) / \ov{\mu}_i ]  )^2 \notag \\
\leq & ~\epsilon^2 / n + 40\epsilon_{\mathrm{mp}}^2\epsilon^2/b + 1000\epsilon^4/b + 3 ( \E[ ( \ov{\mu}_i^{\new} - \ov{\mu}_i ) / \ov{\mu}_i ]  )^2 ,
\end{align}
where the first step follows by definition of $v$ (see Eq.~\eqref{eq:rewrite_mu_t_new}), the second step follows by $\ov{\mu} \approx_{0.1} t$ (Part 2 of Assumption~\ref{ass:assumption}) and $(t/t^{\new} -1)^2 \leq \epsilon^2/(4n)$ (Definition~\ref{def:new}), the third step follows by $\E[x^2] = \Var[x] + (\E[x])^2$, the fourth step follows by Part 2 of Lemma~\ref{lem:bounding_mu_new_minus_mu}.

Now, we are ready to bound $a_3/2=\E[ \| v \|_{ \nabla^2 \Phi_{\lambda}( \ov{\mu} / t - 1 ) }^2 ] $:
\begin{align}\label{eq:split_E_potential_term_3_1}
 & ~ \E[ \| v \|_{ \nabla^2 \Phi_{\lambda}( \ov{\mu} / t - 1 ) }^2 ] \notag \\
= & ~ \lambda^2 \sum_{i=1}^n \E[ \Phi_{\lambda}(\ov{\mu} / t -1 )_i v_i^2 ] \notag \\
\leq & ~ \lambda^2 \sum_{i=1}^n \Phi_{\lambda}(\ov{\mu} / t - 1 )_i \cdot ( \epsilon^2 / n + 40\epsilon_{\mathrm{mp}}^2\epsilon^2/b + 1000\epsilon^4/b + 3 ( \E[ ( \ov{\mu}_i^{\new} - \ov{\mu}_i ) / \ov{\mu}_i ]  )^2  \notag\\
= & ~ (\epsilon^2 / n + 40\epsilon_{\mathrm{mp}}^2\epsilon^2/b + 1000\epsilon^4/b)\lambda^2 \cdot \Phi_{\lambda} ( \ov{\mu} / t - 1 ) \notag \\
& ~ + 3 \lambda^2 \sum_{i=1}^n \Phi_{\lambda} ( \ov{\mu} / t - 1 )_i \cdot ( \E[ ( \ov{\mu}_i^{\new} - \ov{\mu}_i ) /  \ov{\mu}_i  ]  )^2,
\end{align}
where the first step follows by defining $\Phi_{\lambda}(x)_i = \cosh( \lambda x_i )$, 
the second step follows from Eq.~\eqref{eq:bounding_v_i_square}.

For the second term in Eq.~\eqref{eq:split_E_potential_term_3_1}, we can upper bound it in the following way:
\begin{align}\label{eq:split_E_potential_term_3_2}
3 \lambda^2 \sum_{i=1}^n \Phi_{\lambda} ( \ov{\mu} / t - 1 )_i \cdot ( \E[ ( \ov{\mu}_i^{\new} - \ov{\mu}_i ) /  \ov{\mu}_i  ]  )^2
\leq & ~ 
3 \lambda \Big( \sum_{i=1}^n \lambda^2 \Phi_{\lambda} ( \ov{\mu} / t - 1 )_i^2  \Big)^{1/2}
\cdot \| \E[ \ov{\mu}^{-1} ( \ov{\mu}^{\new} - \ov{\mu} ) ] \|_4^2 \notag \\
\leq & ~
3 \lambda \left( \lambda \sqrt{n} + \| \nabla \Phi_{\lambda}(\ov{\mu} /t - 1) \|_2 \right)
\cdot 
 (6\eps + 2\epsilon^2 \sqrt{n}/b)^2 \notag \\
\leq & ~ 
3 \lambda \left( \lambda \sqrt{n} + \| \nabla \Phi_{\lambda}(\ov{\mu} /t - 1) \|_2 \right)
\cdot (100\eps^2 + 10 \epsilon^4 n / b^2 )
\end{align} 
where the first step follows from Cauchy-Schwarz inequality,
the second step follows from Part 3 of Lemma~\ref{lem:potential_function_Phi_cosh} and the fact that
$\| \E[ \ov{\mu}^{-1} ( \ov{\mu}^{\new} - \ov{\mu} ) ] \|_4^2 \leq \| \E[ \ov{\mu}^{-1} ( \ov{\mu}^{\new} - \ov{\mu} ) ] \|_2^2 \leq (6\eps + 2 \epsilon^2 \sqrt{n} / b)^2$ (Part 4 of Lemma~\ref{lem:bounding_mu_new_minus_mu}), the last step follows by $(a+b)^2 \leq 2 a^2 + 2b^2$.

Combining Eq.~\eqref{eq:split_E_potential_term_3_1} and Eq.~\eqref{eq:split_E_potential_term_3_2}, we have
\begin{align}\label{eq:split_E_potential_term_3}
    a_3 = & ~ 2\E[ \| v \|_{ \nabla^2 \Phi_{\lambda}( \ov{\mu} / t - 1 ) }^2 ] \notag \\
    \leq & ~
    2(\epsilon^2 / n + 40\epsilon_{\mathrm{mp}}^2\epsilon^2/b + 1000\epsilon^4/b)\lambda^2 \cdot \Phi_{\lambda} ( \ov{\mu} / t - 1 ) \notag \\
    & ~ + 6 \lambda \left( \lambda \sqrt{n} + \| \nabla \Phi_{\lambda}(\ov{\mu} /t - 1) \|_2 \right)
\cdot (100\eps^2 + 10 \epsilon^4 n / b^2 )
\end{align}

Then, loading Eq.~\eqref{eq:split_E_potential_term_1}, \eqref{eq:split_E_potential_term_2}, \eqref{eq:split_E_potential_term_3} back into Eq.~\eqref{eq:split_E_potential}
\begin{align*}
\E[ \Phi_{\lambda} ( \ov{\mu} / t + v - 1 )  ]
\leq & ~ \Phi_{\lambda} ( \ov{\mu} / t - 1 ) + a_1 + a_2 + a_3 \\
\leq & ~ \Phi_{\lambda} ( \ov{\mu} / t - 1 )  
 + (\text{Eq.}~\eqref{eq:split_E_potential_term_1})  
 + (\text{Eq.}~\eqref{eq:split_E_potential_term_2})
 + (\text{Eq.}~\eqref{eq:split_E_potential_term_3}) \\
= & ~ \Phi_{\lambda} ( \ov{\mu} / t - 1 ) +  \Phi_{\lambda}( \ov{\mu}/t - 1 ) \cdot ( b_{1,\Phi} + b_{2,\Phi} + b_{3,\Phi} ) \\
& ~ + \| \nabla \Phi_{\lambda}( \ov{\mu}/t - 1 ) \|_2 \cdot ( b_{1,\nabla} + b_{2,\nabla} + b_{3,\nabla} ) \\
& ~ + \lambda \epsilon \sqrt{n} \cdot ( b_{1,\sqrt{n}} + b_{2,\sqrt{n}} + b_{3,\sqrt{n}}  )
\end{align*}
where we define terms that come from $a_1$:
\begin{align*}
    b_{1,\Phi} = 0, ~~~ b_{1,\nabla} = -0.45 \epsilon, ~~~ b_{1,\sqrt{n}} = 0.1,
\end{align*}
and terms that come from $a_2$:
\begin{align*}
    b_{2,\Phi} = 0 , ~~~ b_{2,\nabla} = 2 (9\epsilon_{\mathrm{mp}} \epsilon + 4 \epsilon^2 + 2 \epsilon^2 \sqrt{n} / b) = \epsilon \cdot 2( 9 \epsilon_{\mathrm{mp}} + 4 \epsilon + 2\epsilon \sqrt{n} / b ) , ~~~ b_{2,\sqrt{n}} = 0 ,
\end{align*}
and terms that come from $a_3$:
\begin{align*}
    %b_{1,\Phi} = & ~ 0 ,\\
    %b_{1,\nabla} = & ~ -0.45\epsilon ,\\
    %b_{1,\sqrt{n}} = & ~ 0.1, \\
    %b_{2,\Phi} = & ~ 0 , \\
    %b_{2,\nabla} = & ~ 2 (9\epsilon_{\mathrm{mp}} \epsilon + 4 \epsilon^2 + 2 \epsilon^2 \sqrt{n} / b) = \epsilon \cdot 2( 9 \epsilon_{\mathrm{mp}} + 4 \epsilon + \epsilon \sqrt{n} / b ) , \\
    %b_{2,\sqrt{n}} = & ~ 0 ,\\
    b_{3,\Phi} = & ~ 2(\epsilon^2/n + 40\epsilon_{\mathrm{mp}}^2 \epsilon^2 / b + 1000 \epsilon^4 / b ) \lambda^2 = ( \lambda \epsilon / \sqrt{n}) \cdot 2 ( \lambda \epsilon / \sqrt{n} + 40 \epsilon_{\mathrm{mp}}^2 \lambda \epsilon \sqrt{n} /b + 1000 \lambda \epsilon^3 \sqrt{n} /b ), \\
    b_{3,\nabla} = & ~ 6\lambda (100\epsilon^2 + 10 \epsilon^4 n / b^2) = \epsilon \cdot 6( 100 \lambda \epsilon + 10 \lambda \epsilon \cdot \epsilon^2 n / b^2 ) , \\
    b_{3,\sqrt{n}} = & ~ 6(100 \lambda \epsilon + 10 \lambda \epsilon^3 n /b^2  ).
\end{align*}

Note that, if $b \geq 20000 \epsilon \sqrt{n}$ (Part 5 of Assumption~\ref{ass:assumption2}), $\epsilon_{\mathrm{mp}} < 1/1000$ (Assumption~\ref{ass:epsilon}), $\lambda \epsilon  < 10^{-5}$ (Part 6 of Assumption~\ref{ass:assumption2}), we have
\begin{align}\label{eq:bound_coefficient_of_nabla_Phi}
    b_{1,\nabla} + b_{2,\nabla} + b_{3,\nabla} < -0.4 \epsilon
\end{align}

Thus, using Eq.~\eqref{eq:bound_coefficient_of_nabla_Phi} and Part 2 of Lemma~\ref{lem:potential_function_Phi_cosh}, we have
\begin{align*}
\E[ \Phi_{\lambda} ( \ov{\mu} / t + v - 1 )  ]
\leq & ~ \Phi_{\lambda} ( \ov{\mu} / t - 1 ) +  \Phi_{\lambda}( \ov{\mu}/t - 1 ) \cdot ( b_{1,\Phi} + b_{2,\Phi} + b_{3,\Phi} ) \\
& ~ + \| \nabla \Phi_{\lambda}( \ov{\mu}/t - 1 ) \|_2 \cdot (-0.4 \epsilon ) \\
& ~ + \lambda \epsilon \sqrt{n} \cdot ( b_{1,\sqrt{n}} + b_{2,\sqrt{n}} + b_{3,\sqrt{n}}  ) \\
\leq & ~ \Phi_{\lambda} ( \ov{\mu} / t - 1 ) +  \Phi_{\lambda}( \ov{\mu}/t - 1 ) \cdot ( b_{1,\Phi} + b_{2,\Phi} + b_{3,\Phi} ) \\
& ~ + \frac{\lambda}{\sqrt{n}} (\Phi_{\lambda} ( \ov{\mu}/t -1 ) - n) \cdot (-0.4\epsilon) \\
& ~ + \lambda \epsilon \sqrt{n} \cdot ( b_{1,\sqrt{n}} + b_{2,\sqrt{n}} + b_{3,\sqrt{n}}  ) \\
= & ~ \Phi_{\lambda} ( \ov{\mu} / t - 1 ) \cdot \underbrace{ ( 1 + b_{1,\Phi} + b_{2,\Phi} + b_{3,\Phi} - 0.4 \lambda \epsilon /\sqrt{n}  ) }_{ c_{\Phi} } \\
& ~ + \lambda \epsilon \sqrt{n} \cdot \underbrace{ (  b_{1,\sqrt{n}} + b_{2,\sqrt{n}} + b_{3,\sqrt{n}}  + 0.4 ) }_{c_{\sqrt{n}}} ,
\end{align*}
where the first step follows from Eq.~\eqref{eq:bound_coefficient_of_nabla_Phi} and the second step follows from Part 2 of Lemma~\ref{lem:potential_function_Phi_cosh}.

If $b \geq 20000 \cdot ( \lambda \epsilon_{\mathrm{mp}}^2 \epsilon  \sqrt{n} + \epsilon^3 \sqrt{n} )$ (Part 1 of Assumption~\ref{ass:assumption2}) and $\lambda \epsilon < 10^{-5}$ (Part 6 of Assumption~\ref{ass:assumption2}), we have $c_{\Phi} \leq 1 - 0.2 \lambda \epsilon / \sqrt{n}$.

If $\lambda \epsilon< 10^{-5}$ (Part 6 of Assumption~\ref{ass:assumption2}) and $b \geq 20000 \epsilon \sqrt{n} $ (Part 5 of Assumption~\ref{ass:assumption2}), we have $c_{\sqrt{n}} \leq 0.6$.

Thus, we obtain
\begin{align*}
    \E[ \Phi_{\lambda} ( \ov{\mu} / t + v - 1 )  ] 
    \leq & ~ \Phi_{\lambda} ( \ov{\mu} / t - 1 ) \cdot ( 1 - 0.2 \lambda \epsilon  /\sqrt{n} ) + 0.6 \lambda \epsilon \sqrt{n} \\
    \leq & ~ \Phi_{\lambda} ( \ov{\mu} / t - 1 ) -  \frac{\lambda \epsilon}{15 \sqrt{n}} ( \Phi_{\lambda} (\ov{\mu} / t - 1) - 10 n ).
\end{align*}
\end{proof}

\subsection{Bounding the movement of \texorpdfstring{$\ov{w}$}{}}\label{sec:w_movement}
The goal of this section is to prove Lemma~\ref{lem:w_movement}
\begin{lemma}[Bounding the movement of $\ov{w}$]\label{lem:w_movement}
Let $\ov{x}^{\new} = \ov{x} + \wh{\delta}_{x}$, $\ov{s}^{\new} = \ov{s} + \wh{\delta}_{s}$, $\ov{w} = \frac{ \ov{x} }{ \ov{s} }$, and $\ov{w}^{\new} = \frac{ \ov{x}^{\new} }{ \ov{s}^{\new} }$ (same as Definition~\ref{def:new}). Let $b$ denote the size of sketching matrix. Then we have
\begin{align*}
1. & ~ \sum_{i=1}^n \left(  \E [ \ov{w}_i^{\new} ] / \ov{w}_i - 1 \right)^2 \leq 100 \epsilon^2, \\
2. & ~ \sum_{i=1}^n \left(  \E \left[ \left( \ov{w}_i^{\new} / \ov{w}_i - 1 \right)^2 \right]  \right)^2 \leq 10^3 \cdot \epsilon^4 n / b^2 + 4 \cdot 10^4 \cdot \epsilon^4, \\
3. & ~ \left| \ov{w}_i^{\new} / \ov{w}_i - 1 \right| \leq 10\epsilon.
\end{align*} 
\end{lemma}
\begin{proof}

From the definition, we know that
\begin{align*}
\frac{ \ov{w}_i^{\new} }{ \ov{w}_i} = \frac{ 1 }{ \ov{s}_i^{-1} \ov{x}_i } \frac{ \ov{x}_i + \wh{\delta}_{ x,i} }{ \ov{s}_i + \wh{\delta}_{ s,i} } = \frac{ 1 + \ov{x}_i^{-1} \wh{\delta}_{ x,i} }{ 1 + \ov{s}_i^{-1} \wh{\delta}_{ s,i} }.
\end{align*}

\noindent {\bf Part 1.}
For each $i\in [n]$, we have
\begin{align*}
\frac{ \E[ \ov{w}_i^{\new}] }{ \ov{w}_i} - 1 
= & ~ \E\left[ \frac{ 1 +  \ov{x}_i^{-1} \wh{\delta}_{ x,i} }{ 1 + \ov{s}_i^{-1} \wh{\delta}_{ s,i} } \right] - 1 
=  \E \left[ \frac{ \ov{x}_i^{-1} \wh{\delta}_{ x,i} - \ov{s}_i^{-1} \wh{\delta}_{ s,i} }{ 1 + \ov{s}_i^{-1} \wh{\delta}_{ s,i} } \right] 
\leq  2 | \E[ \ov{x}_i^{-1} \wh{\delta}_{x,i} - \ov{s}_i^{-1} \wh{\delta}_{s,i} ] |  \\
\leq & ~ 2 | \E[ \ov{x}_i^{-1} \wh{\delta}_{x,i}] | + 2 | \E[ \ov{s}_i^{-1} \wh{\delta}_{s,i} ] |,
\end{align*}
where the third step follows from $| \ov{s}_i^{-1} \wh{\delta}_{s,i} | \leq 3 \epsilon$ (Part 4 of Lemma~\ref{lem:stochastic_step}), the last step follows from triangle inequality. Then summing over all the coordinates we have
\begin{align*}
\sum_{i=1}^n \left( \E[ \ov{w}_i^{\new}] / \ov{w}_i - 1 \right)^2 \leq \sum_{i=1}^n 8 (\E[ \ov{x}_i^{-1} \wh{\delta}_{x,i}] )^2 + 8 ( \E[ \ov{s}_i^{-1} \wh{\delta}_{s,i}] )^2 \leq 100 \epsilon^2.
\end{align*}
where the first step follows by $(a+b)^2\leq 2a^2 + 2b^2$, the last step follows by $ \|\E[ \ov{s}^{-1} \wh{\delta}_{s}] \|_2^2, \|\E[ \ov{x}^{-1} \wh{\delta}_{x}] \|_2^2 \leq 4 \epsilon^2$ (Part 1 of Lemma~\ref{lem:stochastic_step}).

\noindent {\bf Part 2.} For each $i \in [n]$, we have
\begin{align*}
\E \left[ \left( \frac{ \ov{w}_i^{\new} }{ \ov{w}_i } - 1 \right)^2 \right]
= & ~ \E \left[ \left( \frac{ \ov{x}_i^{-1} \wh{\delta}_{x,i} - \ov{s}_i^{-1} \wh{\delta}_{s,i} }{ 1 + \ov{s}_i^{-1} \wh{\delta}_{s,i} } \right)^2 \right] 
\leq 2 \E[ ( \ov{x}_i^{-1} \wh{\delta}_{x,i} - \ov{s}_i^{-1} \wh{\delta}_{s,i}  )^2 ] \\
\leq & ~ 2 \E[ 2 ( \ov{x}_i^{-1} \wh{\delta}_{x,i})^2 + 2 ( \ov{s}_i^{-1} \wh{\delta}_{s,i})^2 ] 
=  4 \E [( \ov{x}_i^{-1} \wh{\delta}_{x,i})^2] + 4 \E[( \ov{s}_i^{-1} \wh{\delta}_{s,i})^2] \\
= & ~ 4 \Var[ \ov{x}_i^{-1} \wh{\delta}_{x,i} ] + 4 ( \E[ \ov{x}_i^{-1} \wh{\delta}_{x,i}])^2 + 4 \Var[ \ov{s}_i^{-1} \wh{\delta}_{s,i} ] + 4 ( \E[ \ov{s}_i^{-1} \wh{\delta}_{s,i}])^2 \\
\leq & ~ 16 \epsilon^2 / b + 4 ( \E[ \ov{x}_i^{-1} \wh{\delta}_{x,i}])^2 + 4 ( \E[ \ov{s}_i^{-1} \wh{\delta}_{s,i}])^2,
\end{align*}
where the last step follows by $\Var[ \ov{x}_i^{-1} \wh{\delta}_{x,i} ], \Var[ \ov{s}_i^{-1} \wh{\delta}_{s,i} ] \leq 2 \epsilon^2 / b$ (Part 2 of Lemma~\ref{lem:stochastic_step}).

Thus summing over all the coordinates
\begin{align*}
\sum_{i=1}^n \left( \E \left[  \left( \ov{w}_i^{\new} / \ov{w}_i - 1 \right)^2 \right] \right)^2 
\leq & ~  10^3 \epsilon^4 n/b^2 + 64 \sum_{i=1}^n \left( (\E[ \ov{x}_i^{-1} \wh{\delta}_{x,i}])^4 + (\E[ \ov{s}_i^{-1} \wh{\delta}_{s,i}])^4 \right) \\
\leq & ~ 10^3 \epsilon^4 n /b^2 + 4 \cdot 10^4 \cdot \epsilon^4 ,
\end{align*}
where the last step follows by $\| \E [ \ov{s}^{-1} \wh{\delta}_s ] \|_2^2, \| \E[ \ov{x}^{-1} \wh{\delta}_x ] \|_2^2 \leq 4 \epsilon^2 $(Part 1 of Lemma~\ref{lem:stochastic_step}).

\noindent {\bf Part 3.} For each $i\in [n]$
\begin{align*}
\left| \frac{ \ov{w}_i^{\new} }{ \ov{w}_i } - 1 \right|
= \left| \frac{1 + \ov{x}_i^{-1} \wh{\delta}_{x,i} }{ 1 + \ov{s}_i^{-1} \wh{\delta}_{s,i} } - 1 \right|
\leq \left| \frac{1+3\epsilon}{1-3\epsilon} - 1 \right| 
\leq 10\epsilon,
\end{align*}
where the second step follows by $| \ov{x}_i^{-1} \wh{\delta}_{x,i}| \leq 3\epsilon$ and $| \ov{s}_i^{-1} \wh{\delta}_{s,i} | \leq 3\epsilon$ (Part 4 of Lemma~\ref{lem:stochastic_step}).
\end{proof}

\subsection{Bounding the movement of \texorpdfstring{$\ov{\mu}$}{}}\label{sec:mu_movement}
The goal of this section is to prove Lemma~\ref{lem:mu_movement}
\begin{lemma}[Bounding the movement of $\ov{\mu}$]\label{lem:mu_movement}
Let $\ov{x}^{\new} = \ov{x} + \wh{\delta}_{x}$, $\ov{s}^{\new} = \ov{s} + \wh{\delta}_s$, $\ov{\mu} = \ov{x} \cdot \ov{s}$, and $\ov{\mu}^{\new} = \ov{x}^{\new} \cdot \ov{s}^{\new} $. Let $b$ denote the size of sketching matrix. Then we have
\begin{align*}
1. & ~ \sum_{i=1}^n \left( \E [ \ov{\mu}_i^{\new} ] / \ov{\mu}_i  - 1 \right)^2 \leq 100 \epsilon^2, \\
2. & ~ \sum_{i=1}^n \left(  \E \left[ \left( \ov{\mu}_i^{\new} / \ov{\mu}_i - 1 \right)^2 \right]  \right)^2 \leq  4\cdot 10^4 \epsilon^4 n / b^2 + 10^5 \cdot \epsilon^4, \\ 
3. & ~ \left| \ov{\mu}_i^{\new} / \ov{\mu}_i - 1 \right| \leq 10\epsilon.
\end{align*} 
\end{lemma}
%\begin{remark}
%Since $b \geq \sqrt{n}$, the bound of the second term can be simplified to $O(\epsilon^4)$.
%\end{remark}

\begin{proof}

From the definition, we know that 
\begin{align*}
\ov{\mu}^{\new}_i / \ov{\mu}_i = (\ov{x}_i \ov{s}_i)^{-1}\cdot (\ov{x}_i+\wh{\delta}_{x,i})(\ov{s}_i+\wh{\delta}_{s,i}) = (1+\ov{x}_i^{-1}\wh{\delta}_{x,i})(1+\ov{s}_i^{-1}\wh{\delta}_{s,i}).
\end{align*}

\noindent {\bf Part 1.} For each $i\in[n]$, we have 
\begin{align*}
\E[\ov{\mu}^{\new}_i] / \ov{\mu}_i - 1
= & ~ \E[(1+\ov{x}_i^{-1}\wh{\delta}_{x,i})(1+\ov{s}_i^{-1}\wh{\delta}_{s,i})] -1
= \E[\ov{x}_i^{-1}\wh{\delta}_{x,i}+(1+\ov{x}_i^{-1}\wh{\delta}_{x,i})\cdot \ov{s}_i^{-1}\wh{\delta}_{s,i}]\\
\leq & ~ 2\E[\ov{x}_i^{-1}\wh{\delta}_{x,i}+\ov{s}_i^{-1}\wh{\delta}_{s,i}]
=  2\E[\ov{x}_i^{-1}\wh{\delta}_{x,i}]+2\E[\ov{s}_i^{-1}\wh{\delta}_{s,i}]
\end{align*}
where the third step follows by $|\ov{x}_i^{-1}\wh{\delta}_{x,i}|\leq 3\epsilon$ (Part 4 of Lemma~\ref{lem:stochastic_step}).

Thus, summing over all the coordinates gives us
\begin{align*}
\sum_{i=1}^n \left(\E[\ov{\mu}_i^{\new}] / \ov{\mu}_i - 1 \right)^2
\leq  \sum_{i=1}^n 8(\E[\ov{x}_i^{-1}\wh{\delta}_{x,i}])^2+8(\E[\ov{s}_i^{-1}\wh{\delta}_{s,i}])^2
\leq  100\epsilon^2,
\end{align*}
where the first step follows by $(a+b)^2 \leq 2a^2 + 2b^2$, the last step is by $\|\E[\ov{x}^{-1}\wh{\delta}_{x}]\|_2^2,\|\E[\ov{s}^{-1}\wh{\delta}_{s}]\|_2^2\leq 4\epsilon^2$ (Part 1 of Lemma~\ref{lem:stochastic_step}).

\noindent {\bf Part 2.} For each $i\in [n]$, we have
\begin{align*}
\E [ ( \ov{\mu}_i^{\new} / \ov{\mu}_i - 1 )^2 ]
= & ~ \E [(\ov{x}_i^{-1}\wh{\delta}_{x,i}+(1+\ov{x}_i^{-1}\wh{\delta}_{x,i})\cdot \ov{s}_i^{-1}\wh{\delta}_{s,i} )^2 ]\\
\leq & ~ 4\E [ (\ov{x}_i^{-1}\wh{\delta}_{x,i}+\ov{s}_i^{-1}\wh{\delta}_{s,i} )^2 ]\\
\leq & ~ 4\E[2(\ov{x}_i^{-1}\wh{\delta}_{x,i})^2+2(\ov{s}_i^{-1}\wh{\delta}_{s,i})^2]\\
%= & ~ 8 \E[(\ov{x}_i^{-1}\wh{\delta}_{x,i})^2]+4\E[(\ov{s}_i^{-1}\wh{\delta}_{s,i})^2]\\
= & ~ 8 \Var[\ov{x}_i^{-1}\wh{\delta}_{x,i}] + 8(\E[\ov{x}_i^{-1}\wh{\delta}_{x,i}])^2 + 8 \Var[\ov{s}_i^{-1}\wh{\delta}_{s,i}] + 8(\E[\ov{s}_i^{-1}\wh{\delta}_{s,i}])^2 \\
\leq & ~ 32\epsilon^2/b + 8(\E[\ov{x}_i^{-1}\wh{\delta}_{x,i}])^2 + 8(\E[\ov{s}_i^{-1}\wh{\delta}_{s,i}])^2,
\end{align*}
where the second step follows by $|\ov{x}_i^{-1}\wh{\delta}_{x,i}|\leq 3\epsilon$ (Part 4 of Lemma~\ref{lem:stochastic_step}), and the last step follows by $\Var[\ov{x}_i^{-1}\wh{\delta}_{x,i}],~ \Var[\ov{s}_i^{-1}\wh{\delta}_{s,i}]\leq 2\epsilon^2/b$ (Part 2 of Lemma~\ref{lem:stochastic_step}).

Thus summing over all the coordinates
\begin{align*}
\sum_{i=1}^n \left(  \E [ ( \ov{\mu}_i^{\new} / \ov{\mu}_i - 1 )^2 ]  \right)^2
\leq & ~ 4\cdot 10^4 \epsilon^4 n / b^2 + 256\sum_{i=1}^n \left((\E[\ov{x}_i^{-1}\wh{\delta}_{x,i}])^4+(\E[\ov{s}_i^{-1}\wh{\delta}_{s,i})^4\right)\\
\leq & ~ 4\cdot 10^4 \epsilon^4 n /b^2 + 10^5 \epsilon^4
\end{align*}
where the second step follows by $\|\E[\ov{x}^{-1}\wh{\delta}_{x}]\|_2^2,~\|\E[\ov{s}^{-1}\wh{\delta}_s ]\|_2^2\leq 4\epsilon^2$ (Part 1 of Lemma~\ref{lem:stochastic_step}).

\noindent {\bf Part 3.} For each $i\in[n]$,
\begin{align*}
|\ov{\mu}^{\new}_i / \ov{\mu}_i - 1 | = | (1+\ov{x}_i^{-1}\wh{\delta}_{x,i})(1+\ov{s}_i^{-1}\wh{\delta}_{s,i}) -1 | \leq | (1+3\epsilon )^2 -1 | \leq 10\epsilon,
\end{align*}
where the second step follows by $|\ov{x}_i^{-1}\wh{\delta}_{x,i}|\leq 3\epsilon$ and $|\ov{s}_i^{-1}\wh{\delta}_{s,i}|\leq 3\epsilon$ (Part 4 of Lemma~\ref{lem:stochastic_step}).
\end{proof}

\subsection{One step of central path}
\label{sec:central_path}
The central path method is implemented as Algorithm~\ref{alg:one_step_central_path}. In this section we prove that the output of this algorithm indeed matches the definitions of previous sections. First note that the $\ov{x}$, $\ov{s}$, $\ov{w}$, and $\ov{\mu}$ matches Definition~\ref{def:overline}, and $\wt{x}$, $\wt{s}$, $\wt{w}$, $\wt{\mu}$ matches Definition~\ref{def:widetilde}.

\begin{algorithm}[!t]\caption{One step central path }\label{alg:one_step_central_path}
	\begin{algorithmic}[1]
		\Procedure{\textsc{OneStepCentralPath}}{$\text{mp}_t,\text{mp}_{\Phi},\ov{x},\ov{s},t,t^{\new}$} \Comment{Lemma~\ref{lem:one_step_central_path_correctness}, Lemma~\ref{lem:one_step_central_path_time}}
		\State $\ov{w} \leftarrow \ov{x} / \ov{s}$
		\State $\ov{\mu} \leftarrow \ov{x} \ov{s}$ 
		\State $(\wt{w}, \wt{g}_t, p_t) \leftarrow \mathrm{mp}_{t}.\textsc{UpdateQuery}(\ov{w}, \ov{\mu})$ \label{alg:line:p_t} \Comment{Algorithm~\ref{alg:update_query_improved}}
		\State \Comment{mp$_t$ works with function $f_t(x)=\sqrt{x}$}
		\State \Comment{$\wt{w} \approx_{\epsilon_{\mathrm{mp}}} \ov{w}$, $\wt{g}_t\approx_{\epsilon_{\mathrm{mp}}} \ov{\mu}$, $p_t = P(\wt{w})f_t(\wt{g}_t)$}
		%\State $q_{t} \leftarrow f_{t}(\wt{g}_t)$ \label{alg:line:q_t}
		\State $(\wt{w}, \wt{g}_{\Phi}, p_{\Phi}) \leftarrow \mathrm{mp}_{\Phi}.\textsc{UpdateQuery}(\ov{w}, \ov{\mu}/t)$ \label{alg:line:p_phi}\Comment{Algorithm~\ref{alg:update_query_improved}} 
		\State \Comment{mp$_\Phi$ works with function $f_{\Phi}(x)=\nabla\Phi(x-1)/\sqrt{x}$}
		\State \Comment{ $\wt{w} \approx_{\epsilon_{\mathrm{mp}}} \ov{w}$, $\wt{g}_{\Phi} \approx_{\epsilon_{\mathrm{mp}}} \ov{\mu} / t$, $p_{\Phi} = P(\wt{w})f_{\Phi}(\wt{g}_{\Phi})$}
		\State \Comment{Two data structures will return the same $\wt{w}$}
		\State $q_{\Phi} \leftarrow f_{\Phi}( \wt{g}_{\Phi} )$ \label{alg:line:q_phi}
		\State $\wt{\mu} \leftarrow \wt{g}_{\Phi} \cdot t$ \label{alg:line:one_step_wt_mu}
		\State $\wt{x} \leftarrow \sqrt{ \wt{\mu} \wt{w} }$ \label{alg:line:one_step_wt_x}
		\State $\wt{s} \leftarrow \sqrt{ \wt{\mu} / \wt{w} }$ \label{alg:line:one_step_wt_s} \Comment{$\wt{x}$ and $\wt{s}$ satisfies $\wt{\mu} = \wt{x} \wt{s}$ and $\wt{w} = \wt{x} / \wt{s}$}
		\State $\wt{\delta}_t \leftarrow (\frac{t^{\new}}{t}-1)\wt{\mu}$ \label{alg:line:wt_delta_t}
		%,  or $\leftarrow (\frac{t^{\new}}{t}-1) q_t \sqrt{\wt{\mu}} $
		\State $\wt{\delta}_{\Phi} \leftarrow -\frac{\epsilon}{2}\cdot t^{\new}\cdot\frac{\sqrt{\wt{\mu}/t}\cdot q_{\Phi}}{\|\nabla\Phi_{\lambda}(\wt{\mu}/t-1)\|_2} $  \label{alg:line:wt_delta_phi}
		\State $\wt{\delta}_{\mu} \leftarrow \wt{\delta}_t + \wt{\delta}_{\Phi}$ \label{alg:line:wt_delta_mu} %\Comment{Maintain gradient}
		\State $p_{\mu} \leftarrow ( \frac{t^{\new}}{t} - 1 ) p_t - \frac{\epsilon}{2} \cdot t^{\new} \cdot \frac{p_{\Phi}}{\sqrt{t}\| \nabla\Phi_{\lambda}( \wt{\mu} / t - 1 ) \|_2}$ \label{alg:line:p_mu}
		\State $\wh{\delta}_s \leftarrow \frac{\wt{s}}{\sqrt{\wt{\mu}}}p_{\mu}$ \label{alg:line:wh_delta_s}
		\State $\wh{\delta}_x \leftarrow \frac{1}{\wt{s}}\wt{\delta}_{\mu}-\frac{\wt{x}}{\sqrt{\wt{\mu}}}p_{\mu}$ \label{alg:line:wh_delta_x}
		\State \Return $( \wh{\delta}_x, \wh{\delta}_s )$
		\EndProcedure
	\end{algorithmic}
\end{algorithm}

\begin{lemma}[Correctness of one step central\label{lem:central_path_correct} path]\label{lem:one_step_central_path_correctness}
The $\wh{\delta}_s$ and $\wh{\delta}_x$ returned by Algorithm~\ref{alg:one_step_central_path} matches the definition in Definition~\ref{def:hat}, that
\begin{align*}
\wh{\delta}_x = \frac{\wt{X}}{\sqrt{\wt{X}\wt{S}}} (I-(R[l])^\top R[l]\wt{P})\frac{1}{\sqrt{\wt{X}\wt{S}}}\wt{\delta}_{\mu}, ~~~~
\wh{\delta}_s=  \frac{\wt{S}}{\sqrt{\wt{X}\wt{S}}} (R[l])^\top R[l]\wt{P}\frac{1}{\sqrt{\wt{X}\wt{S}}}\wt{\delta}_{\mu},
\end{align*}
where $l$ is the parameter maintained in the data structure, note that the two data structures $mp_t$ and $mp_\Phi$ use the same $l$ and $R[l]$.
\end{lemma}

We have the following claims from Algorithm~\ref{alg:one_step_central_path}.
\begin{claim}
$\wt{\delta}_t$ (Line~\ref{alg:line:wt_delta_t}) matches Definition~\ref{def:widetilde}, that
$\wt{\delta}_t = (\frac{t^{\new}}{t}-1 )\wt{\mu}$.
\end{claim}

\begin{claim}
$\wt{\delta}_\Phi$ (Line~\ref{alg:line:wt_delta_phi}) matches Definition~\ref{def:widetilde}, that $\wt{\delta}_{\Phi} 
= ~-\frac{\epsilon}{2}\cdot t^{\new} \cdot \frac{\nabla \Phi_\lambda (\wt{\mu}/t-1)}{\|\nabla \Phi_{\lambda}(\wt{\mu}/t-1)\|_2}$.
\end{claim}
\begin{proof}
We have
\begin{align*}
\wt{\delta}_{\Phi} 
= -\frac{\epsilon}{2}\cdot t^{\new} \cdot \frac{\sqrt{\wt{\mu}/t}\cdot q_{\Phi}}{\|\nabla \Phi_{\lambda}(\wt{\mu}/t-1)\|_2}
= -\frac{\epsilon}{2}\cdot t^{\new} \cdot \frac{\sqrt{\wt{\mu}/t}\cdot f_{\Phi}(\wt{g}_{\Phi}) }{\|\nabla \Phi_{\lambda}(\wt{\mu}/t-1)\|_2}
= -\frac{\epsilon}{2}\cdot t^{\new} \cdot \frac{\nabla \Phi_{\lambda}(\wt{\mu}/t-1)}{\|\nabla \Phi_{\lambda}(\wt{\mu}/t-1)\|_2},
\end{align*}
where the first step is by definition of $\wt{\delta}_\Phi$ (Line~\ref{alg:line:wt_delta_phi}of Algorithm~\ref{alg:one_step_central_path}), the second step is by definition of $q_\Phi$ (Line~\ref{alg:line:q_phi} of Algorithm~\ref{alg:one_step_central_path}), the third step is by definition of $f_\Phi$ (Line~\ref{alg:line:f_phi} of Algorithm~\ref{alg:main_improved}) and definition of $\wt{\mu}$ (Line~\ref{alg:line:one_step_wt_mu} of Algorithm~\ref{alg:one_step_central_path}) which implies $\wt{g}_{\phi} = \wt{\mu} / t$.
\end{proof}

\begin{claim}
$\wt{\delta}_\mu$ (Line~\ref{alg:line:wt_delta_mu}) matches Definition~\ref{def:widetilde}, that $\wt{\delta}_\mu = \wt{\delta}_t + \wt{\delta}_\Phi$.
\end{claim}

\begin{claim} \label{cla:p_mu}
$p_{\mu}$ (Line~\ref{alg:line:p_mu}) satisfies $p_{\mu} = (R[l])^\top R[l] \wt{P}\frac{1}{\sqrt{\wt{X}\wt{S}}}\wt{\delta}_{\mu}$, where $\wt{P}$ is defined in Definition~\ref{def:widetilde}.
\end{claim}
\begin{proof}
\begin{align*}
p_{\mu} 
= & ~ \left(\frac{t^{\new}}{t}-1\right)p_t - \frac{\epsilon}{2}\cdot t^{\new}\cdot \frac{p_{\Phi}}{\sqrt{t}\|\nabla(\Phi_{\lambda}(\wt{\mu}/t-1)\|_2}\\
= & ~ \left(\frac{t^{\new}}{t}-1\right)(R[l])^\top R[l]\wt{P}\sqrt{\wt{\mu}}+ \frac{\epsilon}{2}\cdot t^{\new}\cdot \frac{(R[l])^\top R[l]\wt{P}\nabla \Phi_{\lambda}( \wt{\mu}/t -1 )/\sqrt{\wt{\mu}/t}}{\sqrt{t}\|\nabla(\Phi_{\lambda}(\wt{\mu}/t-1)\|_2}\\
= & ~ (R[l])^\top R[l]\wt{P}\frac{1}{\sqrt{\wt{\mu}}}\left(\frac{t^{\new}}{t}-1\right)\wt{\mu}+ \frac{\epsilon}{2}\cdot t^{\new}\cdot \frac{(R[l])^\top R[l]\wt{P}\nabla \Phi_{\lambda}( \wt{\mu}/t -1 )/\sqrt{\wt{\mu}}}{\|\nabla(\Phi_{\lambda}(\wt{\mu}/t-1)\|_2}\\
= & ~ (R[l])^\top R[l]\wt{P} \frac{1}{\sqrt{\wt{X}\wt{S}}} \left(\left(\frac{t^{\new}}{t}-1\right)\wt{\mu} + \frac{\epsilon}{2}\cdot t^{\new}\cdot \frac{\nabla \Phi_{\lambda}( \wt{\mu}/t -1 )}{\|\nabla(\Phi_{\lambda}(\wt{\mu}/t-1)\|_2}\right)\\
= & ~ (R[l])^\top R[l]\wt{P}\frac{1}{\sqrt{\wt{X}\wt{S}}}(\wt{\delta_t}+\wt{\delta}_{\Phi})
~ = ~(R[l])^\top R[l]\wt{P}\frac{1}{\sqrt{\wt{X}\wt{S}}}\wt{\delta}_{\mu},
\end{align*}
where the first step is by definition of $p_{\mu}$ (Line~\ref{alg:line:p_mu}), the second step is by definitions of $p_t$ (Line~\ref{alg:line:p_t}) and $p_\Phi$ (Line~\ref{alg:line:p_phi}) and the correctness of \textsc{UpdateQuery}(Part 2 of Theorem~\ref{thm:update_query_correct_improved}), the fourth step is by $\wt{\mu} = \wt{x}\wt{s}$ (Line~\ref{alg:line:one_step_wt_x} and \ref{alg:line:one_step_wt_s}), and the last two steps are by definitions of $\wt{\delta}_t$, $\wt{\delta}_\Phi$ and $\wt{\delta}_{\mu}$ (Definition~\ref{def:widetilde}).
\end{proof}

\begin{claim}\label{cla:wh_delta_s}
$\wh{\delta}_s$ (Line~\ref{alg:line:wh_delta_s}) matches Definition~\ref{def:hat}, that $\wh{\delta}_s=\frac{\wt{S}}{\sqrt{\wt{X}\wt{S}}} (R[l])^\top R[l]\wt{P}\frac{1}{\sqrt{\wt{X}\wt{S}}}\wt{\delta}_{\mu}$.
\end{claim}
\begin{proof}
$\wh{\delta}_s
=  \frac{\wt{S}}{\sqrt{\wt{X}\wt{S}}} p_\mu
=  \frac{\wt{S}}{\sqrt{\wt{X}\wt{S}}} (R[l])^\top R[l]\wt{P}\frac{1}{\sqrt{\wt{X}\wt{S}}}\wt{\delta}_{\mu}$ by Claim~\ref{cla:p_mu}.
\end{proof} 

\begin{claim}\label{cla:wh_delta_x}
$\wh{\delta}_x$ (Line~\ref{alg:line:wh_delta_x}) matches Definition~\ref{def:hat}, that $\wh{\delta}_x = \frac{\wt{X}}{\sqrt{\wt{X}\wt{S}}} (I-(R[l])^\top R[l]\wt{P})\frac{1}{\sqrt{\wt{X}\wt{S}}}\wt{\delta}_{\mu}$.
\end{claim}
\begin{proof}
$\wh{\delta}_x 
=  \frac{1}{\wt{s}}\wt{\delta}_\mu - \frac{\wt{x}}{\sqrt{\mu}} p_\mu
=  \frac{\wt{X}}{\sqrt{\wt{X}\wt{S}}} \frac{1}{\sqrt{\wt{X}\wt{S}}}\wt{\delta}_\mu - \frac{\wt{X}}{\sqrt{\wt{X}\wt{S}}} (R[l])^\top R[l]\wt{P} \frac{1}{\sqrt{\wt{X}\wt{S}}} \wt{\delta}_\mu$ by Claim~\ref{cla:p_mu}.
\end{proof}

\begin{proof}[Proof of Lemma~\ref{lem:central_path_correct}]
Combine Claim~\ref{cla:wh_delta_s} and \ref{cla:wh_delta_x}.
\end{proof}

We also have the following lemma about the running time of Algorithm~\ref{alg:one_step_central_path}:
\begin{lemma}[Running time of one step central path]\label{lem:one_step_central_path_time}
The cost of every operation except data structure calls in Algorithm~\ref{alg:one_step_central_path} is linear in $n$, so the bottleneck is the two calls to the data structure.
\end{lemma}

\section{Data structure : preliminary}
\label{sec:preliminary_improved}

\begin{table}[ht]
\small
\centering
\begin{tabular}{ | l | l | l | l | l | l | }
\hline
{\bf Procedure} & {\bf Algorithm} & {\bf Type} & {\bf Correctness} & {\bf Time/Call} & {\bf Amortized}\\ \hline
\textsc{Initialize} & Algorithm~\ref{alg:initialize_improved} & public & Lemma~\ref{lem:initialize_correct_improved} & Lemma~\ref{lem:initialize_time_improved} & / \\ \hline
\textsc{UpdateQuery} & Algorithm~\ref{alg:update_query_improved} & public & Theorem~\ref{thm:update_query_correct_improved} & /  & Theorem~\ref{thm:main_data_structure_improved} \\ \hline
\textsc{Query} & Algorithm~\ref{alg:query_improved} & private & Lemma~\ref{lem:query_correct_improved} & Lemma~\ref{lem:query_time_improved} & / \\ \hline
\textsc{UpdateV} & Algorithm~\ref{alg:update_v_improved} & private & Lemma~\ref{lem:update_v_correct_improved} & / & / \\ \hline 
\textsc{UpdateG} & Algorithm~\ref{alg:update_g_improved} & private & Lemma~\ref{lem:update_g_correct_improved} & / & / \\ \hline 
\textsc{MatrixUpdate} & Algorithm~\ref{alg:matrix_update_improved} & private & Lemma~\ref{lem:matrix_update_correct_improved} & Lemma~\ref{lem:matrix_update_time_improved} & Lemma~\ref{lem:main_amortize_matrix_update} \\ \hline 
\textsc{PartialMatrixUpdate} & Algorithm~\ref{alg:partial_matrix_update_improved} & private & Lemma~\ref{lem:partial_matrix_update_correct_improved} & Lemma~\ref{lem:partial_matrix_update_time_improved} & Lemma~\ref{lem:main_amortize_partial_matrix_update}\\ \hline
\textsc{VectorUpdate} & Algorithm~\ref{alg:vector_update_improved} & private & Lemma~\ref{lem:vector_update_correct_improved} & Lemma~\ref{lem:vector_update_time_improved} & Lemma~\ref{lem:main_amortize_vector_update}\\ \hline
\textsc{PartialVectorUpdate} & Algorithm~\ref{alg:partial_vector_update_improved} & private & Lemma~\ref{lem:partial_vector_update_correct_improved} & Lemma~\ref{lem:partial_vector_update_time_improved} & Lemma~\ref{lem:main_amortize_partial_vector_update}\\ \hline
\textsc{ComputeLocalVariables} & Algorithm~\ref{alg:compute_local_variables_improved} & private & Definition~\ref{def:local} & Remark~\ref{rem:compute_local_variables} &/ \\ \hline
\end{tabular}	\caption{Summary of the improved data structure. {\bf Amortized} denotes the ``amortized time''.}
\end{table}

\subsection{Preliminary and Definitions}

For ease of presentation, we define the following $\L$ operator that extends the size of a matrix. This $\L$ operator determines the way that our algorithm stores matrices, and executes matrix additions and multiplications.
\begin{definition}[Operator $\L_c$, $\L_r$, $\L$] \label{def:L_cL_rL}
The operator $\L_c$ can only be applied to some sub-columns of a matrix. For a matrix $M_S$ where $M\in \R^{k_1 \times k_2}$, $S\subseteq[k_2]$ with $|S|\leq 6n^a$, $\L_c[M_S]$ means to store the matrix $M_S$ in a $k_1 \times 6n^a$ block by appending extra $0$s.
The algorithm executes $\L_c$ operator in the following way:
\begin{enumerate}
    \item Addition: The $\L_c$ operator supports storing two disjoint submatrices of the same matrix in the same block. %even these two matrices are of different size, as long as their total column size is less than $6n^a$. 
    For a matrix $M\in \R^{k_1 \times k_2}$ and two subsets $S_1,S_2\subseteq[k_2]$ with $S_1\cap S_2 = \emptyset$ and $|S_1\cup S_2| \leq 6n^a$, $\L_c[M_{S_1} ] + \L_c[M_{S_2}] := \L_c[M_{S_1\cup S_2}]$.
    \item Multiplication: When we multiply a matrix $\L_c[M_S]$ with column subscript $S$ with another matrix (or vector) $(B_{S})^{\top}$ with row subscript $S$, if their subscripts are the same, the algorithm will align columns of $\L_c[M_S]$ and rows of $(B_{S})^{\top}$ before doing multiplication.
\end{enumerate}
 
In the same way, we define $\L_r$ as the row operator, and we define $\L = \L_r \circ \L_c$.
\end{definition}

Similarly, we define $\L_*$ that extends a square matrix to $6n^a$ by appending an identity matrix. The motivation of appending identity matrix instead of appending $0$s is to let the matrix inverse being well-defined.

\begin{definition}[Operator $\L_*$]
For any square matrix $M \in \R^{k \times k}$ and $S\subseteq[k]$ with $|S|\leq 6n^a$, 
we define the operator $\L_*$ such that $\L_*[M_{S,S}]$ is stored in a $6n^a\times 6n^a$ block by appending $1$ in the diagonal (so we have $6n^a-|S|$ extra $1$s) and appending $0$ otherwise. We do the same alignment as what we did for $\L$. The extra $1$s are also involved in addition and multiplication. 
\end{definition}

%We define two operations as follows:
%\begin{definition}[Operators $\ins$ and $\del$]
%For any matrix $M \in \R^{2n^a \times 2n^a}$ and any set $S \subseteq [2n^a]$, we define an operator $\ins$ such that $\ins_S[M]\in \R^{2n^a \times 2n^a}$, and for any $i \notin S$, we let the $(i,i)$-th entry of $\ins_S[M]$ to be 1. 

%Similarly, for any matrix $M \in \R^{n^a\times n^a}$ and any set $S \subseteq [n^a]$, we define an operator $\del$ such that $\del_S[M] \in \R^{n^a\times n^a}$, and for any $i \notin S$, we let the $(i,i)$-th entry of $\del_S[M]$ to be 0.
%\end{definition}

Note that in our algorithm whenever we use $\L_r$, $\L_c$, $\L$ or $\L_*$ on a matrix, we always ensure that the size of the extended rows or columns is no larger than $6n^a$. From their definitions, we directly have the following properties. 
\begin{remark}
\label{rem:L_operator}
The operators $\L_r$, $\L_c$, $\L$ and $\L_*$ satisfy the following properties:
\begin{enumerate}
    \item \label{rem:L_operator:enlarge} \textbf{Non-zero entries:} For any $A\in \R^{k_1\times k_2}$ and two subsets $S_1,S_2\subseteq [k_2]$ where $S_1\subseteq S_2$ and $|S_2|\leq 6n^a$, if $A$ only has non-zero entries on columns in $S_1$, then
    \begin{align*}
    \L_c[A_{S_2}] = &~ \L_c[A_{S_1}],\\
    \L_r[(A_{S_2})^{\top}] = &~ \L_r[(A_{S_1})^{\top}].
    \end{align*}
    \item \label{rem:L_operator:addition} \textbf{Addition:} For any $A\in \R^{k_1\times k_2}$ and two subsets $S_1,S_2\subseteq [k_2]$ with $S_1\cap S_2=\emptyset$ and $|S_1\cup S_2|\leq 6n^a$,
    \begin{align*}
    \L_c[A_{S_1}]+\L_c[A_{S_2}] = &~ \L_c[A_{S_1\cup S_2}],\\
    \L_r[(A_{S_1})^{\top}]+\L_r[(A_{S_2})^{\top}] = &~ \L_r[(A_{S_1\cup S_2})^{\top}].
    \end{align*}
    \item \label{rem:L_operator:multiplication_1} \textbf{Multiplication 1:} For any $A\in \R^{k_2\times k_1},B\in \R^{k_2\times k_3}$, and $S
    _1\subseteq [k_1]$, $S_2\subseteq [k_3]$ where $|S_1|,|S_2|\leq 6n^a$, we have 
    \begin{align*}
    \L_r[(A_{S_1})^{\top}\cdot B] = &~ \L_r[(A_{S_1})^{\top}]\cdot B,\\
    \L_c[A^{\top} \cdot B_{S_2}] = &~ A^{\top}\cdot \L_c[B_{S_2}].
    \end{align*}
    \item \label{rem:L_operator:multiplication_2} \textbf{Multiplication 2:} For any $A\in \R^{k_1\times k_2},B\in \R^{k_3\times k_2}$, $C\in \R^{k_2\times k_2}$, and $S_1,S_2\subseteq [k_2]$ where $S_1\subseteq S_2$ and $|S_2|\leq 6n^a$, we have
    \begin{align*}
    \L_c[A_{S_1}]\cdot \L_r[(B_{S_2})^{\top}] = &~ A_{S_1}\cdot (B_{S_1})^{\top}\\
    \L_c[A_{S_2}]\cdot \L_r[(B_{S_1})^{\top}] = &~ A_{S_1} \cdot (B_{S_1})^{\top}\\
    \L_c[A_{S_1}]\cdot \L_*[C_{S_1,S_1}]\cdot \L_r[(B_{S_1})^{\top}] = &~ A_{S_1}\cdot C_{S_1,S_1} \cdot (B_{S_1})^{\top}.
    \end{align*}
    \item \label{rem:L_operator:inverse} \textbf{Inverse:} For any $C\in \R^{k_2\times k_2}$, and $S\subseteq [k_2]$ where $|S|\leq 6n^a$, we have
    \begin{align*}
        \L_*[(C_{S,S})^{-1}] = (\L_*[C_{S,S}])^{-1}.
    \end{align*}
\end{enumerate}
\end{remark}

%We have the following hand-wavy observation:
%\begin{observation}
%\begin{align*}
%\|\wt{v}^{\new} - \wt{v}\|_0 \leq \sqrt{n} \text{ in amortization,}
%\|\wt{g}^{\new} - \wt{g}\|_0 \leq \sqrt{n} \text{ in amortization.}
%\end{align*}
%This is because $\E[\|w^{-1} (w^{\new} - w)\|_2] \leq O(\epsilon)$, and in the worst-case each coordinate of $w$ is changing at the same rate, then in one iteration each coordinate changes by $O(\frac{\epsilon}{\sqrt{n}})$. So after $O(\epsilon_{\mathrm{mp}}\sqrt{n}/\epsilon)$ iterations, all $n$ $w_i$'s have changed by more than $\epsilon_{\mathrm{mp}}$, thus passing the threshold. And in amortization this means at each iteration $O(\sqrt{n})$ coordinates will pass the threshold.
%\end{observation}

\subsection{Facts}
We first prove the following facts that are the cornerstones of the correctness of our data structure. The first lemma shows that we can efficiently decompose low-rank matrices with certain structure.

\begin{figure}[ht]
    \centering
    \includegraphics[width=0.95\textwidth]{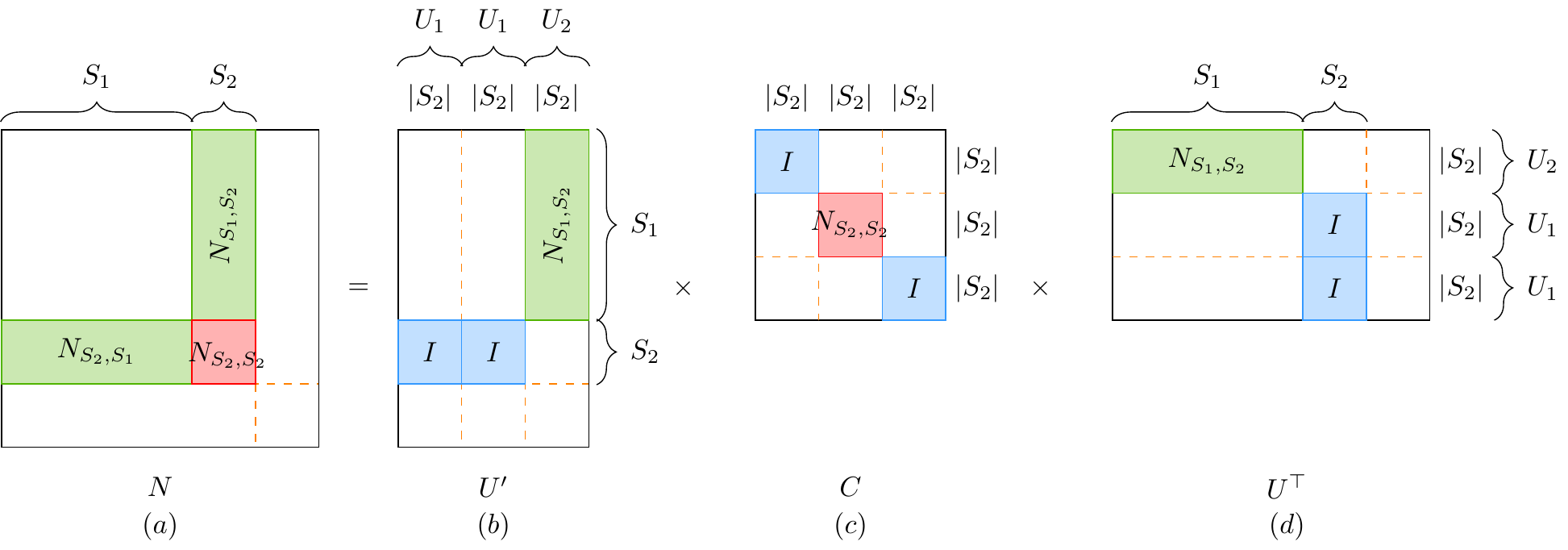}
    \caption{A visualization of the decomposition $N=U'CU^{\top}$ constructed by the \textsc{Decompose} function. (See Lemma~\ref{lem:UCU_decomposition}.) }
    \label{fig:UCU_decomposition_2}
\end{figure}

\begin{lemma}[$U'CU^{\top}$ Decomposition] \label{lem:UCU_decomposition}
If all the non-zero entries of the $6n^a\times 6n^a$ symmetric matrix $N$ can be split into three parts: $N_{S_1, S_2}$, $N_{S_2, S_1}$, and $N_{S_2, S_2}$, where $S_1, S_2\subseteq [n]$, $S_1$ and $S_2$ are disjoint, and $|S_1\cup S_2|\leq 2n^a$. Then there exist matrices $U',U\in \R^{6n^a\times 3|S_2|}, C\in \R^{3|S_2|\times 3|S_2|}$ such that the following decomposition holds: 
\begin{align*}
U'CU^{\top}=N.
\end{align*}
Also, we can compute this decomposition in $O(|S_1|\cdot |S_2|)$ time. We define a function 
$\textsc{Decompose}()$ such that $\textsc{Decompose}(N) = (U',C, U)$.
\end{lemma}
\begin{proof}
We explicitly give a construction of $U'$, $C$ and $U$. See Figure~\ref{fig:UCU_decomposition_2}(a) for an illustration of the structure of $N$, and see Figure~\ref{fig:UCU_decomposition_2}(b),(c),(d) for an illustration of the construction of $U'$, $C$, and $U$. 

First note that from Part~\ref{rem:L_operator:enlarge} and \ref{rem:L_operator:addition} of Remark~\ref{rem:L_operator}, we have
\begin{align*}
N = \L[N_{S_1, S_2}] + \L[N_{S_2, S_1}] + \L[N_{S_2, S_2}].
\end{align*}

We define $U_1\in \R^{6n^a\times |S_2|}$ such that $((U_1^{\top})_{S_2})^{\top}=I_{|S_2|}$ and all other entries of $U_1$ are 0. We let $U_2=\L_r[N_{S_1, S_2}] \in \R^{6n^a\times |S_2|}$. We construct $U'$ as $U'=[U_1, U_1, U_2]$, note that $U'$ has size $6n^a\times 3|S_2|$. And we construct $U$ as $U=[U_2, U_1, U_1]$, note that $U$ also has size $6n^a\times 3|S_2|$.

We construct $C$ as $\begin{bmatrix}
I_{|S_2|}& 0& 0\\
0& N_{S_2,S_2}& 0\\
0& 0 &I_{|S_2|}
\end{bmatrix}$. Note that $C$ has size $3|S_2|\times 3|S_2|$.

It is easy to check that this decomposition is correct:
\begin{align*}
U'CU^{\top} = &~ 
\begin{bmatrix}
U_1& U_1& U_2
\end{bmatrix}
\cdot
\begin{bmatrix}
I_{|S_2|}& 0& 0\\
0& N_{S_2,S_2}& 0\\
0& 0 &I_{|S_2|}
\end{bmatrix}
\cdot
\begin{bmatrix}
U_2^{\top} \\ U_1^{\top} \\ U_1^{\top}
\end{bmatrix}\\
= &~ U_1U_2^{\top} + U_1\cdot N_{S_2,S_2}\cdot U_1^{\top} + U_2U_1^{\top}\\
= &~ \L[((U_1^{\top})_{S_2})^{\top}\cdot U_2^{\top}] +\L[ ((U_1^{\top})_{S_2})^{\top}\cdot N_{S_2,S_2}\cdot (U_1^{\top})_{S_2} ] + \L[U_2\cdot(U_1^{\top})_{S_2}] \\
= &~ \L[U_2^{\top}] + \L[N_{S_2,S_2}] + \L[U_2] \\
= &~ \L[N_{S_2, S_1}] + \L[N_{S_2, S_2}] + \L[N_{S_1, S_2}]
= ~ N,
\end{align*}
where the third step follows from the fact that $U_1$ only has non-zero entries on the rows in $S_2$ and Part~\ref{rem:L_operator:enlarge} of Remark~\ref{rem:L_operator}, the fourth step follows from $((U_1^{\top})_{S_2})^{\top} = (U_1)_{S_2,S_2}=I$, the fifth step follows from $U_2$ only has one block of non-zero entries: $(U_2)_{S_1,S_2}=N_{S_1,S_2}$ and Part~\ref{rem:L_operator:enlarge} of Remark~\ref{rem:L_operator}.

Finally, since $U'$, $C$, $U$ are all constructed by copying certain entries of $N$, the running time of this decomposition is the sum of the sizes of the three matrices. Thus we can compute this decomposition in $O(|S_1|\cdot |S_2|)$ time.
\end{proof}

The next lemma shows that a particular matrix satisfies the constraints of the previous lemma.

\begin{figure}[ht]
    \centering
    \includegraphics[width=0.8\textwidth]{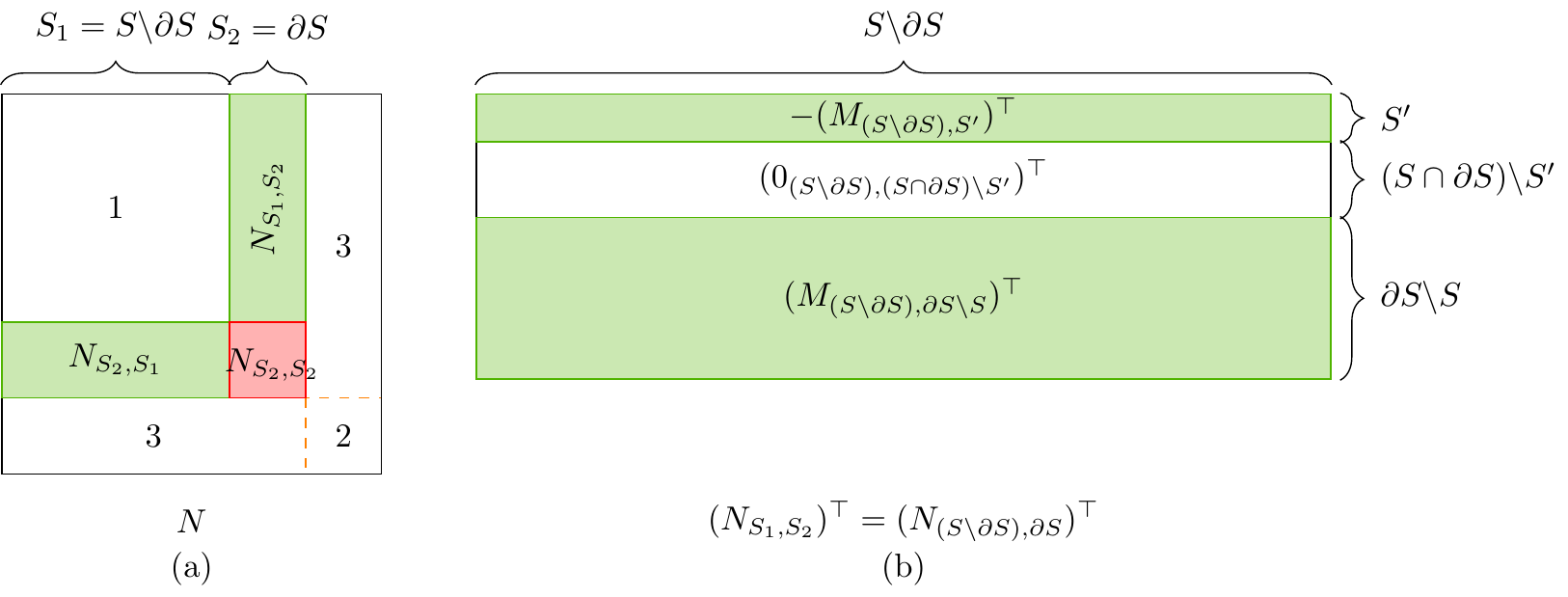}
    \caption{A visualization of the matrices involved in \textsc{Decompose} function. Figure (a) illustrates the structure of the input matrix $N$ (see Lemma~\ref{lem:UCU_decomposition} and Lemma~\ref{lem:structure_inverse_change}). 
    Figure (b) illustrates the structure of $N_{S_1,S_2}=N_{(S\backslash \partial S), \partial S}$. (See Lemma~\ref{lem:structure_inverse_change}.) For clarity we show its transpose here.}
    \label{fig:UCU_decomposition_1}
\end{figure}

\begin{lemma}[Structure of the change in inverse matrix]
\label{lem:structure_inverse_change}
For $v,\wt{v},w^{\appr}\in \R^n$, and a symmetric matrix $M\in \R^{n\times n}$, let $S = \supp(\wt{v}-v)$, $\partial S = \supp(w^{\appr}-\wt{v})$, $S^{\new} = \supp(w^{\appr}-v)$, and $S'=(S\cup \partial S)\backslash S^{\new}$. Let $\Delta=\wt{V}-V$, $\Delta^{\new}=W^{\appr}-V$. Let $N=\L_*[(\Delta^{\new}_{S^{\new},S^{\new}})^{-1}+ M_{S^{\new},S^{\new}}] - \L_*[\Delta_{S,S}^{-1}+ M_{S,S}]$.

Then the non-zero entries of $N$ can be split into three parts: $N_{(S\backslash \partial S), \partial S}$, $N_{\partial S, (S\backslash \partial S)}$, and $N_{\partial S, \partial S}$, as shown in Figure~\ref{fig:UCU_decomposition_1}(a). And $N_{(S\backslash \partial S), \partial S}$ has the following structure (as shown in Figure~\ref{fig:UCU_decomposition_1}(b)):
\begin{itemize}
    \item For columns in $\partial S\backslash S$, $N_{(S\backslash \partial S), (\partial S\backslash S)} = M_{(S\backslash \partial S), (\partial S\backslash S)}$.
    \item For columns in $S'$, $N_{(S\backslash \partial S), S'} = -M_{(S\backslash \partial S), S'}$.
    \item For other columns, $N_{(S\backslash \partial S), (S\cap \partial S)\backslash S'} = 0$.
\end{itemize}

\end{lemma}
\begin{proof}
Note that $N$ is a symmetric matrix. It is easy to see that $S^{\new}\subseteq S\cup \partial S$ and $S'\subseteq S\cap \partial S$. We also have the following observations:
\begin{itemize}
    \item $\forall i\in S\backslash \partial S$, $v_i\neq \wt{v}_i$, and $\wt{v}_i=w^{\appr}_i$.
    \item $\forall i\in \partial S\backslash S$, $v_i=\wt{v}_i$, and $\wt{v}_i\neq w^{\appr}_i$.
    \item $\forall i\in S'$, $v_i\neq \wt{v}_i$, and $v_i=w^{\appr}_i$.
    \item $\forall i\in (S\cap \partial S)\backslash S'$, $v_i\neq \wt{v}_i$, $v_i\neq w^{\appr}_i$, and $\wt{v}_i\neq w^{\appr}_i$.
\end{itemize}

Using the above observations and the definition of $\L_*$, $N$ has the following properties:
\begin{itemize}
\item $\forall i,j\in S\backslash \partial S$, we have $N_{i,j} = (\Delta^{\new}_{i,j} + M_{i,j})- (\Delta_{i,j} + M_{i,j})=0$, where the second step follows from the fact that $\Delta^{\new}_{i,i}=w^{\appr}_i-v_i=\wt{v}_i-v_i=\Delta_{i,i}$. This means the block with label 1 in Figure~\ref{fig:UCU_decomposition_1}(a) only has zeroes.
\item $\forall i,j\notin S\cup \partial S$, if $i=j$, we have $N_{i,j} = 1-1=0$, otherwise we have $N_{i,j} = 0-0=0$. This means the block with label 2 in Figure~\ref{fig:UCU_decomposition_1}(a) only has zeroes.
\item $\forall i\in S\cup \partial S, j\notin S\cup \partial S$, we have $N_{i,j}=0-0=0$, then we also have $N_{j,i}=N_{i,j}=0$. This means the two blocks with label 3 in Figure~\ref{fig:UCU_decomposition_1}(a) only has zeroes.
%\item $\forall i\in S\backslash \partial S, j\in \partial S\cap S$, we have $N_{i,j}=M_{i,j}-M_{i,j}=0$, then we also have $N_{j,i}=N_{i,j}=0$. This means the two blocks with label 4 in sub-figure (b) of Figure~\ref{fig:UCU_decomposition_1} only has zeroes.
\end{itemize}
Thus we prove the first statement of this lemma: the non-zero entries of $N$ can be split into three parts: $N_{(S\backslash\partial S), \partial S}$, $N_{\partial S, (S\backslash\partial S)}$, and $N_{\partial S, \partial S}$. Using the above observations we also have the following properties for entries in $N_{(S\backslash \partial S), \partial S}$. For any $i\in S\backslash \partial S$, note that $i\in S$ and $i\in S^{\new}$.
\begin{itemize}
    \item $\forall j\in (S\cap \partial S)\backslash S'$, note that $i\neq j$, and $j\in S$ and $j\in S^{\new}$. So $N_{i,j}=(0+M_{i,j})-(0+M_{i,j})=0$.
    \item $\forall j\in S'$, note that $i\neq j$, and $j\in S$ and $j\notin S^{\new}$. So $N_{i,j}=(0+0)-(0+M_{i,j})=-M_{i,j}$.
    \item $\forall j\in \partial S\backslash S$, note that $i\neq j$, and $j\notin S$ and $j\in S^{\new}$. So $N_{i,j}=(0+M_{i,j})-(0+0)=M_{i,j}$.
\end{itemize}
Thus $N_{(S\backslash \partial S), \partial S}$ has the structure as described in lemma statement.
%\[
%N_{(S\backslash \partial S), \partial S} = [0_{(S\backslash \partial S), (S\cap \partial S)\backslash S'},~ -M_{(S\backslash \partial S), S'},~ M_{(S\backslash \partial S), (\partial S\backslash S)}].
%\]
\end{proof}

The next lemma shows that we can use Woodbury identity together with the previous $U'CU^{\top}$ decomposition to efficiently maintain the inverse of a matrix.
\begin{lemma}[Correctness of $B$ using Woodbury Identity]
\label{lem:B_tmp_correctness}
For $v,\wt{v},w^{\appr}\in \R^n$, and a symmetric matrix $M\in \R^{n\times n}$, let $S = \supp(\wt{v}-v)$, $\partial S = \supp(w^{\appr}-\wt{v})$, $S^{\new} = \supp(w^{\appr}-v)$, and $S'=(S\cup \partial S)\backslash S^{\new}$. Let $\Delta=\wt{V}-V$, $\Delta^{\new}=W^{\appr}-V$. Let $B = \L_* [(\Delta_{S, S}^{-1}+ M_{S,S})^{-1}]$, and let 
\[
N=\L_*[(\Delta^{\new}_{S^{\new},S^{\new}})^{-1}+ M_{S^{\new},S^{\new}}] - \L_*[\Delta_{S,S}^{-1}+ M_{S,S}].
\]
Suppose $|S\cup \partial S|\leq 2n^a$, and let $(U',C,U):=\textsc{Decompose}(N)$, where $\textsc{Decompose}()$ is the function of Lemma \ref{lem:UCU_decomposition}, note that $U',U\in \R^{6n^a\times (3|\partial S|)},C\in \R^{(3|\partial S|)\times (3|\partial S||)}$, then we have
\[
B - B U' (C^{-1} + U^{\top}BU')^{-1} U^{\top} B = \L_*[\big((\Delta^{\new}_{S^{\new}, S^{\new}})^{-1} + M_{S^{\new}, S^{\new}}\big)^{-1}].
\]
\end{lemma}
\begin{proof}
From Lemma~\ref{lem:structure_inverse_change}, we know that the non-zero entries of $N$ can be split into three parts: $N_{(S\backslash\partial S), \partial S}$, $N_{\partial S, (S\backslash\partial S)}$, and $N_{\partial S, \partial S}$. $S\backslash \partial S$ and $\partial S$ are disjoint, and $|S\cup \partial S|\leq 2n^a$, so $N$ satisfies the requirements of Lemma~\ref{lem:UCU_decomposition}. And in the lemma statement of Lemma~\ref{lem:UCU_decomposition}, $S_1$ corresponds to $S\backslash \partial S$ here, $S_2$ corresponds to $\partial S$ here. Thus $U', C, U$ are well-defined, and we have $U'CU^{\top} = N$.

Also note that from the property of $\L_*$ operator (Part~\ref{rem:L_operator:inverse} of Remark~\ref{rem:L_operator}) we have 
\[
B = \L_*[(\Delta_{S, S}^{-1}+ M_{S,S})^{-1}] = (\L_*[\Delta_{S, S}^{-1}+ M_{S,S}])^{-1}.
\]
Using Woodbury identity (Fact \ref{fact:woodbury}), we have that
\begin{align*}
\L_*[\big((\Delta^{\new}_{S^{\new}, S^{\new}})^{-1} + M_{S^{\new}, S^{\new}}\big)^{-1}] = &~ \big(\L_*[(\Delta^{\new}_{S^{\new}, S^{\new}})^{-1} + M_{S^{\new}, S^{\new}}]\big)^{-1}&~ \\
= &~\Big(\L_* [\Delta_{S, S}^{-1}+ M_{S,S}] + U'CU^{\top}\Big)^{-1}\\
= &~ \Big(B^{-1} + U'CU^{\top}\Big)^{-1}\\
= &~ B - B U' (C^{-1} + U^{\top}BU')^{-1} U^{\top} B
\end{align*}
where the first step follows from Part~\ref{rem:L_operator:inverse} of Remark~\ref{rem:L_operator}, the second step follows from $U'CU^{\top}=N$, the third step follows from the fact that $B = (\L_*[\Delta_{S, S}^{-1}+ M_{S,S}])^{-1}$, and the forth step follows from Woodbury identity.
\end{proof}

We can exploit the fact that the output matrices $U$ and $U'$ of $\textsc{Decompose}$ resembles the input matrix, and we have the following corollary:
\begin{corollary}[Correctness of $U^{\tmp}$] \label{cor:U_tmp_correctness}
Given $v, \wt{v}, w^{\appr}\in \R^n$, and a symmetric matrix $M\in \R^{n\times n}$. Let $S$, $\partial S$, $S^{\new}$, $S'$, $\Delta$, $\Delta^{\new}$, $B$, $N$, $U$, $U'$, and $C$ be defined the same way as in Lemma~\ref{lem:structure_inverse_change} and \ref{lem:B_tmp_correctness}.
Let $E = B\cdot \L_r[(M_S)^{\top}]$. Define $\partial E\in \R^{6n^a\times |\partial S|}$ such that 
\begin{align*}
    (\partial E)_{(\partial S\backslash S)} = &~ E_{(\partial S\backslash S)}-B_{(\partial S\cap S)}M_{(\partial S\cap S), (\partial S\backslash S)}\\
    (\partial E)_{S'} = &~ -E_{S'}+B_{\partial S \cap S}M_{(\partial S\cap S), S'}
\end{align*}
and other entries of $\partial E$ are all zero. Define $U^{\tmp}\in \R^{6n^a\times 3|\partial S|}$ as
\[
U^{\tmp}=\left[B_{\partial S},~ B_{\partial S},~ \partial E\right],
\]
%\[
%\left[B_{\partial S},~ B_{\partial S},~ \L_r[0_{(S\backslash \partial S),~ (S\cap \partial S)\backslash S'}],~ \left(E_{(\partial S\backslash S)}-B_{\partial S\backslash S}M_{(\partial S\backslash S), (\partial S\backslash S)}\right),~ \left(-E_{S'}+B_{\partial S\backslash S}M_{(\partial S\backslash S), S'}\right)\right],
%\]
then we have $U^{\tmp} = BU'$. Note that $\partial E$ is the same one as defined on Line~\ref{alg:line:partial_E_1} and \ref{alg:line:partial_E_1} of Algorithm~\ref{alg:query_improved}, and $U^{\tmp}$ is the same one as defined on Line~\ref{alg:line:U_tmp}.
\end{corollary}
\begin{proof}
From Lemma~\ref{lem:structure_inverse_change} we know that $N$ can be split into three parts: $N_{(S\backslash \partial S), \partial S}$, $N_{\partial S, (S\backslash \partial S)}$, and $N_{\partial S, \partial S}$. From the proof of Lemma~\ref{lem:UCU_decomposition}, we have that $U'=[U_1,U_1,U_2]$, where
\begin{enumerate}
    \item $U_1\in \R^{6n^a\times |\partial S|}$ such that $((U_1^{\top})_{\partial S})^{\top}=I_{|\partial S|}$ and all other entries are 0,
    \item $U_2 \in \R^{6n^a\times |\partial S|}$ and $U_2=\L_r[N_{(S\backslash \partial S), \partial S}]$.
\end{enumerate}
Since $BU'=[BU_1, BU_1, BU_2]$, it suffices to prove that $BU_1=B_{\partial S}$ and $BU_2=\partial E$. We have
\begin{align*}
    B U_1 = B_{\partial S}\cdot ((U_1^{\top})_{\partial S})^{\top}
    = B_{\partial S},
\end{align*}
where the first step follows from $U_1$ only has non-zero rows in $\partial S$, and the second step follows from $((U_1^{\top})_{\partial S})^{\top} = I_{|\partial S|}$. For $BU_2$ we first prove the following:
\begin{align*}
    E_{(\partial S\backslash S)} = &~ B\cdot \L_r[M_{S,(\partial S\backslash S)}]
    = ~ B\cdot \L_r[M_{(\partial S\cap S), (\partial S\backslash S)}] + B\cdot \L_r[M_{(S\backslash \partial S), (\partial S\backslash S)}]\\
    = &~ B_{(\partial S\cap S)}\cdot M_{(\partial S\cap S), (\partial S\backslash S)} + B\cdot \L_r[M_{(S\backslash \partial S), (\partial S\backslash S)}],
\end{align*}
where the first step follows from $E=B\cdot \L_r[(M_S)^{\top}]$, the second step follows from Part~\ref{rem:L_operator:addition} of Remark~\ref{rem:L_operator}, and the third step follows from Part~\ref{rem:L_operator:multiplication_2} of Remark~\ref{rem:L_operator}. So we have
\[
B\cdot \L_r[M_{(S\backslash \partial S), (\partial S\backslash S)}] = E_{(\partial S\backslash S)} - B_{(\partial S\cap S)}\cdot M_{(\partial S\cap S), (\partial S\backslash S)},
\]
and similarly we also have $B\cdot \L_r[M_{(S\backslash \partial S), S'}] = E_{S'} - B_{(\partial S\cap S)}\cdot M_{(\partial S\cap S), S'}$.

Thus combining with the structure of $N_{(S\backslash \partial S), \partial S}$ proved in Lemma~\ref{lem:structure_inverse_change}, we have

1. For columns in $\partial S\backslash S$, 
    \begin{align*}
        (BU_2)_{\partial S\backslash S} = B\cdot \L_r[N_{(S\backslash\partial S), (\partial S\backslash S)}]
        = B\cdot \L_r[M_{(S\backslash\partial S), (\partial S\backslash S)}]
        = E_{(\partial S\backslash S)} - B_{(\partial S\cap S)}\cdot M_{(\partial S\cap S), (\partial S\backslash S)}.
    \end{align*}
    
2. For columns in $S'$,
    \begin{align*}
        (BU_2)_{S'} = B\cdot \L_r[N_{(S\backslash\partial S), S'}]
        = B\cdot \L_r[-M_{(S\backslash\partial S), S'}]
        = -E_{S'} + B_{(\partial S\cap S)}\cdot M_{(\partial S\cap S), S'}.
    \end{align*}

3. For all other columns, $(BU_2)_{(S\cap \partial S)\backslash S'} = 0$.

Thus we have $BU_2=\partial E$, and therefore $BU'=U^{\tmp}$.
\end{proof}

We also have the following lemma that shows how to use Woodbury identity to update the inverse of a matrix directly.
\begin{lemma}[Correctness of $M$ using Woodbury Identity]
\label{lem:M_tmp_correctness}
Let $A\in \R^{n\times n}$, and let $\wt{V}, V\in \R^{n\times n}$ be two diagonal matrices. Let $M=A^{\top}(AVA^{\top})^{-1}A \in \R^{n\times n}$, and let $\Delta = \wt{V} - V \in \R^{n\times n}$, let $S = \supp(\wt{v}-v)\subseteq[n]$, then we have
\begin{align}
\label{eq:M_correctness}
M - M_{S} \cdot \left((\Delta_{S,S})^{-1} + M_{S,S}\right)^{-1}\cdot (M_{S})^{\top}= A^{\top}(A\wt{V}A^{\top})^{-1}A
\end{align}
\end{lemma}

\begin{proof}
We have{\small
\begin{align*}
(A\wt{V}A^{\top})^{-1}= &~ (A (V + \Delta) A^{\top})^{-1} 
= (AVA^{\top} + A\Delta A^{\top})^{-1}
= (AVA^{\top} + A_{S}\Delta_{S, S} (A_{S})^{\top})^{-1} \\
= & ~ (AVA^{\top})^{-1} - (AVA^{\top})^{-1}\cdot A_{S}\cdot  \Big((\Delta_{S, S})^{-1} + (A_{S})^{\top} (AVA^{\top})^{-1} A_{S}\Big)^{-1} \cdot (A_{S})^{\top} \cdot (AVA^{\top})^{-1}\\
= &~ (AVA^{\top})^{-1} - (AVA^{\top})^{-1}\cdot A_{S}\cdot \left((\Delta_{S,S})^{-1} + M_{S,S}\right)^{-1}\cdot(A_{S})^{\top} \cdot (AVA^{\top})^{-1},
\end{align*}}
where the first step follows from the definition of $\Delta$, the third step follows from $\Delta$ only has non-zero entries on $(i,i)$-th entries where $i\in S$, the fourth step follows from Woodbury identity, and the fifth step follows from the definition of $M$.
Then from the definition of $M$ we have Eq.~\eqref{eq:M_correctness}.
%\begin{align*}
%A^{\top}(A\wt{V}A^{\top})^{-1}A = M - M_{S} \cdot \left((\Delta_{S,S})^{-1} + M_{S,S}\right)^{-1}\cdot (M_{S})^{\top}.
%\end{align*}
\end{proof}

\begin{table}[!t]
\small
\centering
\begin{tabular}{ | l | l | l | l | l | }
\hline
{\bf Notation} & \textsc{MatrixUpdate} & \textsc{P.MatrixUpdate} & \textsc{VectorUpdate} & \textsc{P.VectorUpdate} \\ \hline
$v$ & $\surd$ &  &  & \\ \hline
$\wt{v}$ & $\surd$ & $\surd$ &  & \\ \hline
$g$ &  &  & $\surd$ & \\ \hline
$\wt{g}$ &  &  & $\surd$ & $\surd$\\ \hline
$M$ & $\surd$ &  &  & \\ \hline
$Q$ & $\surd$ & & & \\ \hline
$\beta_1$ & $\surd$ & & $\surd$ & \\ \hline
$\beta_2$ & $\surd$ & & $\surd$ & \\ \hline
$S$ & $\surd$ & $\surd$ & & \\ \hline
$T$ &  & & $\surd$ & $\surd$ \\ \hline
$\Delta$ & $\surd$ & $\surd$ & & \\ \hline
$\Gamma$ & $\surd$ & $\surd$ & & \\ \hline
$\xi$ &$\surd$ &$\surd$ &$\surd$ &$\surd$ \\ \hline
$B$ & $\surd$ & $\surd$ & & \\ \hline
$E$ & $\surd$ & $\surd$ & & \\ \hline
$F$ & $\surd$ & $\surd$ & & \\ \hline
$\gamma_1$ & $\surd$ & $\surd$ & $\surd$ & $\surd$ \\ \hline
$\gamma_2$ & $\surd$ & $\surd$ & $\surd$ & $\surd$ \\ \hline
Goal & $v,\wt{v}$ & $\wt{v}$ & $g,\wt{g}$ & $\wt{g}$ \\ \hline
Algorithm & Algorithm~\ref{alg:matrix_update_improved} & Algorithm~\ref{alg:partial_matrix_update_improved} & Algorithm~\ref{alg:vector_update_improved} & Algorithm~\ref{alg:partial_vector_update_improved} \\ \hline
Correctness & Lemma~\ref{lem:matrix_update_correct_improved} & Lemma~\ref{lem:partial_matrix_update_correct_improved} & Lemma~\ref{lem:vector_update_correct_improved} & Lemma~\ref{lem:partial_vector_update_correct_improved} \\ \hline
Time & Lemma~\ref{lem:matrix_update_time_improved} & Lemma~\ref{lem:partial_matrix_update_time_improved} & Lemma~\ref{lem:vector_update_time_improved} & Lemma~\ref{lem:partial_vector_update_time_improved} \\ \hline
\end{tabular}\caption{Summary of things got changed over different updates. List of members in Algorithm~\ref{alg:member_improved}.}\label{tab:member_changes}
\end{table}

\subsection{Main result}
The goal of this section is to present Theorem~\ref{thm:main_data_structure_improved}.
\begin{theorem}[Main data structure theorem]\label{thm:main_data_structure_improved}
Given a full rank matrix $A\in \R^{d\times n}$ with $d\leq n$, two error parameters $0<\epsilon_{\mathrm{mp}}<1/4$ and $\epsilon_{\far}\leq \frac{\epsilon_{\mathrm{mp}}}{100\log n}$, two threshold parameters $a\leq \alpha$ and $\wt{a}\leq \alpha\cdot a$ where $\alpha$ is the dual exponent of matrix multiplication, a parameter of sketching size $b \in (0,1)$. Let $f:\R \rightarrow \R$ be some pre-defined function which can be computed in $O(1)$ time, and extend the definition of $f$ on vector $v$ with $f(v)_i:=f(v_i)$. Let $\omega$ denote the exponent of matrix multiplication. There is a data structure (in Algorithm \ref{alg:member_improved},
%\ref{alg:invariant_improved},
\ref{alg:initialize_improved}, \ref{alg:adjust}, \ref{alg:binary_search_improved}, \ref{alg:update_query_improved}, \ref{alg:update_v_improved}, \ref{alg:update_g_improved}, \ref{alg:compute_local_variables_improved}, \ref{alg:query_improved}, \ref{alg:matrix_update_improved}, \ref{alg:partial_matrix_update_improved}, \ref{alg:vector_update_improved}, \ref{alg:partial_vector_update_improved}) that approximately maintains the vector
\[
\sqrt{W}A^\top (AWA^\top)^{-1}A\sqrt{W}f(h).
\]
The data structure uses $n^{2+o(1)}$ space and supports the following operations:
\begin{enumerate}
\item \textsc{Initialize}$(f, \epsilon_{\mathrm{mp}}, b, L, A, w_0, h_0, R)$: Initialize all data structure members including $L=n^{1-b+o(1)}$ sketching matrices $R_1,R_2,\dots,R_{L}\in \R^{n^{b}\times n}$. This operation takes $O(n^\omega)$ time.
\item \textsc{UpdateQuery}$(w, h)$: On the $j$-th call to this function, it outputs the following three vectors (see Theorem~\ref{thm:update_query_correct_improved}):
    \begin{enumerate}
    \item A vector $w^{\appr} \in \R^n$ such that $w^{\appr} \approx_{\epsilon_{\mathrm{mp}}}w$.
    \item A vector $h^{\appr} \in \R^n$ such that $h^{\appr}\approx_{\epsilon_{\mathrm{mp}}}h$.
    \item A vector $r \in \R^n$ such that
    \[
    r = R_l^\top R_l \sqrt{W^{\appr}}A^\top (AW^{\appr}A^\top)A\sqrt{W^{\appr}}f(h^{\appr}).
    \]
    \end{enumerate}
\end{enumerate}

Furthermore, if the initial vectors $(w^{(0)}, h^{(0)})$ and the update sequence $(w^{(1)}, h^{(1)}),\dots,(w^{(T)}, h^{(T)}) \in (\R^n, \R^n)$ satisfy the following constraints: \\
\begin{align*}
1. ~~~ \sum_{i=1}^n \left( \frac{ \E [ w_i^{(j+1)} ] - w_i^{(j)} }{ w_i^{(j)} } \right)^2 \leq C_1^2,~ \sum_{i=1}^n \left(  \E \left[ \left( \frac{ w_i^{(j+1)} - w_i^{(j)} }{ w_i^{(j)}} \right)^2 \right]  \right)^2 \leq C_2^2,~  \left| \frac{ w_i^{(j+1)} - w_i^{(j)} }{ w_i^{(j)} } \right| \leq C_3,
\end{align*}
\begin{align*}
2. ~~~ \sum_{i=1}^n \left( \frac{ \E [ \mu_i^{(j+1)} ] - \mu_i^{(j)} }{ \mu_i^{(j)} } \right)^2 \leq C_4^2,~ \sum_{i=1}^n \left(  \E \left[ \left( \frac{ \mu_i^{(j+1)} - \mu_i^{(j)} }{\mu_i^{(j)}} \right)^2 \right]  \right)^2 \leq C_5^2, ~ \left| \frac{ \mu_i^{(j+1)} - \mu_i^{(j)} }{ \mu_i^{(j)} } \right| \leq C_6,
\end{align*}
for $j=0,1,\dots,T-1$, where the expectation is conditioned on $w^{(j)}$ for Part 1, and conditioned on $h^{(j)}$ for Part 2, and the parameters $C_1,C_2,C_3,C_4,C_5,C_6$ satisfy that $C_1, C_2, C_4, C_5 > 0$ and $ 0 < C_3, C_6\leq \frac{1}{4}$. Then the worst-case running time of \textsc{Query} per iteration is
\[
O^*(\Tmat(n^{\wt{a}}, n^a, n^{\wt{a}})+n^{1+b}),
\]
and the expected amortized running time per iteration of the other procedures are
\begin{enumerate}
    \item \textsc{MatrixUpdate}: $O^*\Big(\left(C_1\epsilon_{\mathrm{mp}}/\epsilon_{\far}^2 + C_2/\epsilon_{\far}^2\right)\cdot (n^{2-a/2}+n^{\omega-1/2}) \Big)$,
    \item \textsc{PartialMatrixUpdate}: $O^*\Big((C_1/\epsilon_{\mathrm{mp}} + C_2/\epsilon_{\mathrm{mp}}^2)\cdot (n^{1+a-\wt{a}/2} + n^{1+(\omega-3/2)a})\Big)$,
    \item \textsc{VectorUpdate}: $O^*\Big((C_4\epsilon_{\mathrm{mp}}/\epsilon_{\mathrm{\far}}^2 + C_5/\epsilon_{\mathrm{\far}}^2)\cdot n^{1.5}\Big)$,
    \item \textsc{PartialVectorUpdate}: $O^*\Big((C_4\epsilon_{\mathrm{mp}}/\epsilon_{\mathrm{\far}}^2 + C_5/\epsilon_{\mathrm{\far}}^2)\cdot (n^{1.5}+n^{2a-\wt{a}/2})\Big)$,
\end{enumerate}
where $O^*$ notation hides all $n^{o(1)}$ terms.
\end{theorem}
\begin{proof}
The properties of the output of \textsc{UpdateQuery} follow from Theorem~\ref{thm:update_query_correct_improved}. 
The running time of \textsc{Initialize} is by Lemma~\ref{lem:initialize_time_improved}, the running time of \textsc{Query} is by Lemma~\ref{lem:query_time_improved}. The amortized running time of \textsc{MatrixUpdate} is by Lemma~\ref{lem:main_amortize_matrix_update}, the amortized running time of \textsc{PartialMatrixUpdate} is by Lemma~\ref{lem:main_amortize_partial_matrix_update}, the amortized running time of \textsc{VectorUpdate} is by Lemma~\ref{lem:main_amortize_vector_update}, and the amortized running time of \textsc{PartialVectorUpdate} is by Lemma~\ref{lem:main_amortize_partial_vector_update}.
\end{proof}

We prove the correctness of the data structure in Section~\ref{sec:correctness_improved}. We give the worst-case analysis of the running time per call for all procedures in Section~\ref{sec:time_per_call_improved}, and the amortized analysis in Section~\ref{sec:amortize_time_improved}.

\begin{algorithm}[!t]\caption{Data structure : members }\label{alg:member_improved} 
\small
	\begin{algorithmic}[1]
		\State{\bf data structure} \Comment{Theorem~\ref{thm:main_data_structure_improved}}
		\State
		\State{\bf members} \Comment{Table~\ref{tab:member_changes}}
		\State \hspace{4mm} $A \in \R^{d\times n}$
		\State \hspace{4mm} Function $f:\R\rightarrow \R$
		\State \hspace{4mm} $\epsilon_{\mathrm{mp}}, \epsilon_{\far}\in\R$
		\State \hspace{4mm} $v,~\wt{v},~g,~\wt{g} \in \R^n$
		\State \hspace{4mm} $a, \wt{a}, b \in  (0,1]$
		\State \hspace{4mm} $L \in \mathbb{N}$ \Comment{number of sketching matrices}
		\State \hspace{4mm} $l \in \mathbb{N}$ \Comment{count of iterations}
		\State \hspace{4mm} $\forall i\in [L]$, $R_i\in \R^{n^b \times n}$
		\State \hspace{4mm} $R = [ R_1^\top, R_2^\top, \cdots, R_{L}^\top ]^{\top} \in \R^{n^{1+o(1)}\times n}$ \Comment{We have the guarantee that $L n^b = n^{1+o(1)}$}
		\State \hspace{4mm} \Comment{Below are the invariant variables}
		\State \hspace{4mm} $M \in \R^{n\times n}$
		\State \hspace{4mm} $Q \in \R^{n^{1+o(1)}\times n}$
		\State \hspace{4mm} $\beta_1 \in \R^{n+o(1)}$
		\State \hspace{4mm} $\beta_2 \in \R^n$
		\State \hspace{4mm} Set $S \subseteq [n]$ 
		\State \hspace{4mm} Set $T \subseteq [n]$
		\State \hspace{4mm} $\Delta \in \R^{n \times n}$ 
		\State \hspace{4mm} $\Gamma \in \R^{n \times n}$
		\State \hspace{4mm} $\xi \in \R^n$
		\State \hspace{4mm} $B \in \R^{6n^a \times 6n^a}$
		\State \hspace{4mm} $F\in \R^{n^{1+o(1)}\times 6n^a}$
		\State \hspace{4mm} $E\in \R^{6n^a\times n}$
		\State \hspace{4mm} $\gamma_1 \in \R^{6n^a}$
		\State \hspace{4mm} $\gamma_2 \in \R^{n}$
		\State{\bf end members}
		\State {\bf end data structure}
\end{algorithmic}
\end{algorithm}

\begin{algorithm}[!t]\caption{Data structure : \textsc{Initialize}()}\label{alg:initialize_improved}
\small
\begin{algorithmic}[1]
        \State {\bf data structure} \Comment{Theorem~\ref{thm:main_data_structure_improved}}
        \State
		\Procedure{\textsc{initialize}}{$f,\epsilon_{\mathrm{mp}}, \epsilon_{\far}, a, \wt{a}, b, L, A, w_0, h_0, R$}\label{alg:line:initialize_improved} \Comment{Lemma~\ref{lem:initialize_correct_improved}, Lemma~\ref{lem:initialize_time_improved}}
		\State $f \leftarrow f$ \Comment{$f : \R \rightarrow \R$}
		\State $\epsilon_{\mathrm{mp}} \leftarrow \epsilon_{\mathrm{mp}}$ \Comment{$\epsilon_{\mathrm{mp}} \in (0,1)$}
		\State $\epsilon_{\far} \leftarrow \epsilon_{\far}$ \Comment{$\epsilon_{\far} \in (0,1)$}
		\State $a \leftarrow a$ \Comment{$a \in (0,1]$}
		\State $\wt{a} \leftarrow \wt{a}$ \Comment{$\wt{a} \in (0,1]$}
		\State $b \leftarrow b $ \Comment{$b \in (0,1]$}
		\State $L \leftarrow L $ \Comment{$L \in \mathbb{N}$}
		\State $A \leftarrow A$ \Comment{$A\in \R^{d\times n}$}
		\State $R \leftarrow R $ \Comment{$R = [R_1^\top, R_2^\top, \cdots, R_L^\top]^\top \in \R^{n^{1+o(1)} \times n}$}
		\State $l\leftarrow 1$ \Comment{count of iterations}
		\State $v \leftarrow \wt{v} \leftarrow w_0$ \label{alg:line:initialize_v_wt_v_improved} \Comment{$v, \wt{v} \in \R^n$}
		\State {$g \leftarrow \wt{g} \leftarrow h_0$} \Comment{$g, \wt{g} \in \R^n$}
		\State \Comment{Below are the invariant variables}
		\State $M \leftarrow A^\top ( AVA^\top )^{-1}A $ \Comment{$M \in \R^{n\times n}$}
		\State $Q \leftarrow R \sqrt{V}M $ \Comment{$Q \in \R^{n^{1+o(1)}\times n}$}
		\State $\beta_1 \leftarrow Q \sqrt{V} f(g)$ \Comment{$\beta_1 \in \R^{n^{1+o(1)}}$}
		\State $\beta_2 \leftarrow M \sqrt{V} f(g) $ \Comment{$\beta_2 \in \R^{n}$}
		\State $S\leftarrow\emptyset$ \Comment{$S \subseteq [n]$}
		\State $T\leftarrow\emptyset$ \Comment{$T \subseteq [n]$}
		\State $\Delta \leftarrow 0$ \Comment{$\Delta \in \R^{n\times n}$ is a diagonal matrix}
		\State $\Gamma \leftarrow 0$ \Comment{$\Gamma \in \R^{n\times n}$ is a diagonal matrix}
		\State $\xi \leftarrow 0$ \Comment{$\xi \in \R^n$}
		\State $B \leftarrow I$ \Comment{$B\in \R^{6n^a \times 6n^a}$}
		\State $E\leftarrow 0$ \Comment{$E \in \R^{6n^a\times n}$}
		\State $F \leftarrow 0$ \Comment{$F \in \R^{n^{1+o(1)}\times 6n^a}$}
		\State $\gamma_1 \leftarrow 0$ \Comment{$\gamma_1 \in \R^{6n^a}$}
		\State $\gamma_2 \leftarrow 0$ \Comment{$\gamma_2 \in \R^{n}$}
		\EndProcedure
		\State
		\State {\bf end data structure}
	\end{algorithmic}
\end{algorithm}

\begin{algorithm}[!t]\caption{Data structure : \textsc{Adjust}()}\label{alg:adjust}
\small
    \begin{algorithmic}[1]
    \State {\bf data structure}
    \State \Comment{This procedure doesn't use any members in the memory of data structure.}
    \Procedure{\textsc{Adjust}}{$\wt{v}^{\tmp}, \wt{v}, v, \epsilon_{\far}$} 
    \State \Comment{ $\wt{v}^{\tmp}$ is the temporary new update of $\wt{v}$. $\wt{v}^{\tmp}$ is adjusted to $\wt{v}^{\adj}$.}
    \State $\wt{v}^{\adj}\leftarrow \wt{v}^{\tmp}$ \label{alg:line:adjust_start}
    \For{$i=1$ to $n$}
		   \If{$\wt{v}_i^{\tmp} \neq \wt{v}_i$ and $\wt{v}_i^{\tmp}\in[ (1- \epsilon_{\far} )v_i,(1 + \epsilon_{\far} )v_i]$}
		        \State $\wt{v}_i^{\adj} \leftarrow v_i$  \label{alg:line:adjust_chase_parent}
		   \EndIf
	\EndFor
	\State \Return $\wt{v}^{\adj}$
    \EndProcedure
    \State \Comment{If in coordinate $i$, $\wt{v}_i^{\tmp} \neq \wt{v}_i$ and $\wt{v}_i^{\tmp}$ is close to $v_i$, then $\wt{v}_i^{\tmp}$ should move back to $v_i$.}
	\State {\bf end data structure}
    \end{algorithmic}
\end{algorithm}

\begin{algorithm}[!t]\caption{Data structure : \soft()}\label{alg:binary_search_improved}
\small
    \begin{algorithmic}[1]
    \State {\bf data structure}
    \State \Comment{This procedure doesn't use any members in the memory of data structure.}
    \Procedure{\soft}{$y,w^{\new}, v, \epsilon, n^a$} 
		\State Let $\pi : [n] \rightarrow [n]$ be a sorting permutation such that $y_{\pi(i)} \geq y_{\pi(i+1)}$		
		\State $k \leftarrow$ the number of indices $i$ such that $y_i \geq \epsilon$ \label{alg:line:k_initial_binary_search}
		\If{$k \geq n^a$}
		\Repeat \label{alg:line:repeat_start_binary_search_improved}
		\State $k \leftarrow \min \{ \lceil 1.5 k \rceil, n \} $ \label{alg:line:k_times_1.5_binary_search_improved}
		\Until{$k = n$ or $y_{ \pi(k) } < (1-1/\log n) \cdot y_{\pi(k/1.5)}$} \label{alg:line:repeat_end_binary_search_improved}
		\EndIf
		\State $v^{\new}_{\pi(i)} \leftarrow \begin{cases} w_{\pi(i)}^{\new} , & i \in \{1,2,\cdots,k \} ; \\ v_{\pi(i)}, & i \in \{ k+1, \cdots, n \} . \end{cases}$ \label{alg:line:v_new_binary_search_improved}
		\State \Return $v^{\new}$, $k$
	\EndProcedure
	\State
	\State {\bf end data structure}
    \end{algorithmic}
\end{algorithm}

\begin{algorithm}[!t]\caption{Data structure : \textsc{UpdateQuery}()}\label{alg:update_query_improved}
\small
	\begin{algorithmic}[1]	
	\State{\bf data structure} \Comment{Theorem~\ref{thm:main_data_structure_improved}}
	\State	
		\Procedure{\textsc{UpdateQuery}}{$w^{\new},~h^{\new}$} \Comment{Theorem~\ref{thm:update_query_correct_improved}}
		\State $w^{\appr}, k, \wt{k}\leftarrow$ \textsc{UpdateV}$(w^{\new})$ \Comment{Algorithm~\ref{alg:update_v_improved}, $k$ and $\wt{k}$ are only used for analysis.} \label{alg:line:update_v_in_update_query_improved}
		\State $h^{\appr}, p, \wt{p}\leftarrow$ \textsc{UpdateG}$(h^{\new})$ \Comment{Algorithm~\ref{alg:update_g_improved}, $p$ and $\wt{p}$ are only used for analysis.} \label{alg:line:update_g_in_update_query_improved}
		\State $r\leftarrow$ \textsc{Query}$(w^{\appr}, h^{\appr})$ \label{alg:line:query_in_update_query_improved} \Comment{Algorithm~\ref{alg:query_improved}, Lemma~\ref{lem:query_correct_improved}, Lemma~\ref{lem:query_time_improved}}
		\State \Comment{ Compute $r = R[l]^{\top}R[l]\sqrt{ W^{\appr} } A^\top ( A W^{\appr} A^\top )^{-1} A \sqrt{ W^{\appr} } f( h^{\appr} ) $}
		\State \Return $w^{\appr}, h^{\appr}, r$
		\EndProcedure
		\State
		\State {\bf end data structure}
	\end{algorithmic}
\end{algorithm}

\begin{algorithm}[!t]\caption{Data structure : \textsc{UpdateV()}}\label{alg:update_v_improved}
\small
	\begin{algorithmic}[1]
	%\small
	\State{\bf data structure} \Comment{Theorem~\ref{thm:main_data_structure_improved}}
	\State	
		\Procedure{\textsc{UpdateV}}{$w^{\new}$} \Comment{Return $(w^{\appr}, k, \wt{k}) $. Lemma~\ref{lem:update_v_correct_improved}}
		\State $\wt{v}^{\tmp}, \wt{k} \leftarrow \soft(y_i \leftarrow \psi( w^{\new}_i/ \wt{v}_i - 1), w^{\new},\wt{v},\epsilon_{\mathrm{mp}}/2, n^{\wt{a}})$ \label{alg:line:binary_search_for_wt_v_improved} \Comment{Algorithm~\ref{alg:binary_search_improved}}
		\State $\wt{v}^{\new} \leftarrow \textsc{Adjust}(\wt{v}^{\tmp}, \wt{v}, v, \epsilon_{\far})$ \label{alg:line:adjust_for_wt_v_improved} \Comment{Algorithm~\ref{alg:adjust}}
		\If{ $|\supp(\wt{v}^{\new}-v)|\geq n^a$} \label{alg:line:if_for_matrix_update} 
		\State $v^{\new}, k \leftarrow \soft(y_i \leftarrow (\psi ( w^{\new}_i/v_i -1) + \psi ( w^{\new}_i/\wt{v}_i - 1)), w^{\new}, v, \frac{\epsilon_{\far}^2}{32\epsilon_{\mathrm{mp}}}, n^a)$ \label{alg:line:binary_search_for_v_improved} 
		\State \Comment{If $|\supp(\wt{v}^{\new} - v) | \geq  n^a$, then $k \geq n^a$. See Fact~\ref{fac:bound_on_k_j}}
		\State { \textsc{MatrixUpdate}$(v^{\new})$}\label{alg:line:enter_matrix_update_improved} \Comment{Algorithm~\ref{alg:matrix_update_improved}, Lemma~\ref{lem:matrix_update_correct_improved}, Lemma~\ref{lem:matrix_update_time_improved}, Lemma~\ref{lem:main_amortize_matrix_update}}
		\State \Comment{Update $v,\wt{v}$ to be $v^{\new}$.}
		\State \Comment{Update invariants $M$, $Q$, $\beta_1$, $\beta_2$, $S$, $\Delta$, $\Gamma$, $\xi$, $B$, $\gamma_1$, $\gamma_2$, $E$, $F$.}
		\State \Return ($v^{\new}$, $k$, 0) \label{alg:line:return_by_matrix_update}
		\Else \label{alg:line:else_for_matrix_update}
		\If{$|\supp(\wt{v}^{\new} - \wt{v})| \geq n^{\wt{a}}$} \label{alg:line:if_for_partial_matrix_update}
		\State \Comment{If $|\supp(\wt{v}^{\new} - \wt{v}) | \geq n^{\wt{a}}$, then $\wt{k}\geq n^{\wt{a}}$. See Fact~\ref{fac:bound_wt_k_j}.}
		\State { \textsc{PartialMatrixUpdate}$(\wt{v}^{\new})$}\label{alg:line:enter_partial_matrix_update_improved} \Comment{Algorithm~\ref{alg:partial_matrix_update_improved}, Lemma~\ref{lem:partial_matrix_update_correct_improved}, Lemma~\ref{lem:partial_matrix_update_time_improved}, Lemma~\ref{lem:main_amortize_partial_matrix_update}}
		\State \Comment{Update $\wt{v}$ to be $\wt{v}^{\new}$.}
		\State \Comment{Update invariants $S$, $\Delta$, $\Gamma$, $\xi$, $B$, $\gamma_1$, $\gamma_2$, $E$, $F$.}
		\State \Return ($\wt{v}^{\new}$, 0, $\wt{k}$) \label{alg:line:return_by_partial_matrix_update}
		\EndIf
		\EndIf
		\State \Return ($\wt{v}^{\new}$, 0, 0) \label{alg:line:update_v_return_by_default}
		\EndProcedure
		\State
		\State {\bf end data structure}
	\end{algorithmic}
\end{algorithm}

\begin{algorithm}[!t]\caption{Data structure : \textsc{UpdateG()}}\label{alg:update_g_improved}
\small
	\begin{algorithmic}[1]
	%\small
	\State{\bf data structure} \Comment{Theorem~\ref{thm:main_data_structure_improved}}
	\State	
		\Procedure{\textsc{UpdateG}}{$h^{\new}$} \Comment{Return $(h^{\appr}, p, \wt{p}) $. Lemma~\ref{lem:update_g_correct_improved}}
%       \State $w^{\appr} \leftarrow w^{\new}$, $h^{\appr} \leftarrow h^{\new}$  \Comment{In the end, this procedure will return $w^{\appr}$ and $h^{\appr}$}
%       \State \Comment{ as $\epsilon_{\mathrm{mp}}$ approximation of $w^{\new}$ and $h^{\new}$}
		\State $\wt{g}^{\tmp}, \wt{p} \leftarrow \soft(y_i \leftarrow \psi ( h^{\new}_i/ \wt{g}_i - 1), h^{\new},\wt{g},\epsilon_{\mathrm{mp}}/2, n^{\wt{a}})$ \label{alg:line:binary_search_for_wt_g_improved} \Comment{Algorithm~\ref{alg:binary_search_improved}}
		\State $\wt{g}^{\new} \leftarrow \textsc{Adjust}(\wt{g}^{\tmp}, \wt{g}, g, \epsilon_{\far})$ \label{alg:line:adjust_for_wt_g_improved} \Comment{Algorithm~\ref{alg:adjust}}
		\If{ $|\supp(\wt{g}^{\new}-g)|\geq n^a$} \label{alg:line:if_for_vector_update} 
		\State $g^{\new}, p \leftarrow \soft(y_i \leftarrow (\psi ( h^{\new}_i/g_i - 1) + \psi ( h^{\new}_i/\wt{g}_i - 1)), h^{\new}, g, \frac{\epsilon_{\far}^2}{32\epsilon_{\mathrm{mp}}}, n^a)$ \label{alg:line:binary_search_for_g_improved} 
		\State \Comment{Similarly, if $|\supp(\wt{g}^{\new} - g) | > n^a$, then $p>n^a$.}
		\State \textsc{VectorUpdate}$(g^{\new})$ \label{alg:line:enter_vector_update_improved} \Comment{Algorithm~\ref{alg:vector_update_improved}, Lemma~\ref{lem:vector_update_correct_improved}, Lemma~\ref{lem:vector_update_time_improved}, Lemma~\ref{lem:main_amortize_vector_update}}
		\State \Comment{Update $g,\wt{g}$ to be $g^{\new}$. Update invariants $\beta_1$, $\beta_2$, $\xi$, $\gamma_1$, $\gamma_2$, $T$.}
		\State \Return ($g^{\new}$, $p$, 0) \label{alg:line:return_by_vector_update}
		\Else
		\If{$|\supp(\wt{g}^{\new} - \wt{g})| \geq n^{\wt{a}}$} \label{alg:line:if_for_partial_vector_update}
		\State \Comment{Similarly, if $|\supp(\wt{g}^{\new} - \wt{g}) | > n^{\wt{a}}$, then $\wt{p}>n^{\wt{a}}$.}
		\State \textsc{PartialVectorUpdate}$(\wt{g}^{\new})$ \label{alg:line:enter_partial_vector_update_improved} \Comment{Algorithm~\ref{alg:partial_vector_update_improved}, Lemma~\ref{lem:partial_vector_update_correct_improved}, Lemma~\ref{lem:partial_vector_update_time_improved}, Lemma~\ref{lem:main_amortize_partial_vector_update}}
		\State \Comment{Update $\wt{g}$ to be $\wt{g}^{\new}$. Update invariants $\xi$, $\gamma_1$, $\gamma_2$, $T$.}
		\State \Return ($\wt{g}^{\new}$, 0, $\wt{p}$) \label{alg:line:return_by_partial_vector_update}
		\EndIf
		\EndIf
		\State \Return ($\wt{g}^{\new}$, 0, 0) \label{alg:line:update_g_return_by_default}
		\EndProcedure
		\State
		\State {\bf end data structure}
	\end{algorithmic}
\end{algorithm}

\begin{algorithm}[!t]\caption{Data structure : \textsc{ComputeLocalVariables}()}\label{alg:compute_local_variables_improved}
\small
	\begin{algorithmic}[1]	
	\State{\bf data structure} \Comment{Theorem~\ref{thm:main_data_structure_improved}}
	\State	
		\Procedure{\textsc{ComputeLocalVariables}}{$w^{\appr},~h^{\appr}$} 
		\State $\partial \Delta \leftarrow  W^{\appr} - \wt{V}$
		\State $\partial \Gamma \leftarrow \sqrt{W^{\appr}} - \sqrt{\wt{V}}$
		\State $\partial S \leftarrow \supp(w^{\appr} - \wt{v})$
		\State $\Delta^{\new} \leftarrow \Delta + \partial \Delta$
		\State $\Gamma^{\new} \leftarrow \Gamma + \partial \Gamma$
		\State $S^{\new} \leftarrow \supp(w^{\appr}-v)$
		\State $S' \leftarrow (S\cup \partial S)\backslash S^{\new}$
		\State \Comment{If the input $h^{\appr}$ is null, we don't need to compute the following two local variables.}
		\State $\partial \xi \leftarrow \sqrt{W^{\appr}}f(h^{\appr}) - \sqrt{\wt{V}}f(\wt{g})$
		\State $\xi^{\new} \leftarrow \xi + \partial \xi$
		\State \Return $(\partial \Delta, \partial \Gamma, \partial \xi, \partial S, \Delta^{\new}, \Gamma^{\new}, \xi^{\new}, S^{\new}, S')$
		\EndProcedure
		\State
		\State {\bf end data structure}
	\end{algorithmic}
\end{algorithm}

\begin{algorithm}[!t]\caption{Data structure : \textsc{Query}()}\label{alg:query_improved}
\small
	\begin{algorithmic}[1]
	\State{\bf data structure} \Comment{Theorem~\ref{thm:main_data_structure_improved}}
	\State	
	\Procedure{\textsc{Query}}{$w^{\appr}, h^{\appr}$} \Comment{Lemma~\ref{lem:query_correct_improved}, Lemma~\ref{lem:query_time_improved}}
	    	\State $\partial \Delta, \partial \Gamma, \partial \xi,  \partial S, \Delta^{\new}, \_, \_, S^{\new}, S' \leftarrow \textsc{ComputeLocalVariables}(w^{\appr}, \wt{g}^{\new})$ \Comment{Algorithm~\ref{alg:compute_local_variables_improved}}
	    	%\State \Comment{$S^{\new} = \supp(w^{\appr}-v)$, $\partial S=\supp(w^{\appr}-\wt{v})$, and $S'=(S\cup\partial S)\backslash S^{\new}$}
		\State $r_1\leftarrow  \beta_{1}[l]$ \label{alg:line:query:r_1_improved} \Comment{$r_1 \in \R^{n^b}$}
       	\State {$r_2\leftarrow Q[l]\xi + R[l]\gamma_2 + R[l] \partial \Gamma M (\xi + \partial \xi) + \big(Q[l] + R[l] \Gamma M\big) \partial \xi$} \label{alg:line:query:r_2_improved} \Comment{$r_2 \in \R^{n^b}$}
		\State $r_3\leftarrow R[l](\Gamma+\partial \Gamma) \beta_2$ \label{alg:line:query:r_3_improved} \Comment{$r_3 \in \R^{n^b}$}
		%\State {\color{pink} $\partial \gamma \leftarrow B\cdot \L_r [\beta_{2, \partial S\backslash S}] +   B\cdot \L_r[(M_{\partial S} )^{\top}] \cdot \partial \xi$ \Comment{local variable}} \label{alg:line:partial_gamma}
		\State $\partial \gamma \leftarrow B\cdot (\L_r [(\beta_{2})_{\partial S\backslash S}]-\L_r [(\beta_{2})_{S'}]) + B\cdot (\L_r[(M_{\partial S \backslash S})^{\top}]-\L_r[(M_{S'})^{\top}]) \cdot (\xi+\partial \xi) + E \cdot \partial \xi$ \label{alg:line:partial_gamma}
		\State \Comment{local variable $\partial \gamma\in \R^{6n^a}$}
		%\State {\color{pink} $r_{4,1}\leftarrow \L_c[-Q_{l, S} - R_l \Gamma M_{S}] \cdot \gamma_{1}$}
		%%\State $r_{4,1}\leftarrow \big(-\L_c[(Q[l])_{S}] - F[l] \big) \cdot \gamma_{1}$ \label{alg:line:r_4_1}
		%\State {\color{pink} $r_{4,2}\leftarrow \L_c [ -Q_{l,S^{\new}} - R_l(\Gamma+\partial \Gamma) M_{S^{\new}}] \cdot  \partial \gamma$ }
		%\State $r_{4,2}\leftarrow \L_c [ -(Q[l])_{S^{\new}} - F[l] - R[l](\Gamma M_{\partial S\backslash S} + \partial \Gamma M_{S^{\new}})] \cdot  \partial \gamma$
		%%\State $r_{4,2}\leftarrow \big(-\L_c [(Q[l])_{S^{\new}}] - F[l] - R[l]\Gamma\cdot \L_c[M_{\partial S\backslash S}] - R[l] \partial \Gamma\cdot \L_c[M_{S^{\new}}] \big) \cdot \partial \gamma$ \label{alg:line:r_4_2}
		%\State {\color{pink} $r_{4,3}\leftarrow \L_c[ -(Q_{l,\partial S\backslash S}-R_{l}\partial \Gamma M_{\partial S} ] \cdot \gamma_1$}
		%\State $r_{4,3}\leftarrow \L_c[ -((Q[l])_{\partial S\backslash S}-R[l] \Gamma M_{\partial S\backslash S} -R[l] \partial \Gamma M_{ S^{\new}} ] \cdot \gamma_1$
		%%\State $r_{4,3}\leftarrow \big( -\L_c[(Q[l])_{\partial S\backslash S}] - R[l] \Gamma\cdot \L_c[M_{\partial S\backslash S}] - R[l] \partial \Gamma \cdot \L_c[M_{ S^{\new}}] \big)\cdot \gamma_1$ \label{alg:line:r_4_3}
		\State $(U', C, U) \leftarrow \textsc{Decompose}\Big(\L_*[( \Delta^{\new}_{S^{\new}, S^{\new}})^{-1} +  M_{S^{\new}, S^{\new}} ] - \L_* [\Delta_{S, S}^{-1}+ M_{S,S}]\Big)$
		\State \Comment{\textsc{Decompose} is defined in Lemma \ref{lem:UCU_decomposition}. $U',U\in \R^{6n^a\times 3|\partial S|}$, $C\in \R^{3|\partial S|\times 3|\partial S|}$}
		%\State {\color{pink} $r_{4,4}\leftarrow \L_c[ Q_{l,S^{\new}}+R_l(\Gamma+\partial \Gamma) M_{S^{\new}} ] \cdot BU'(C^{-1}+U^{\top}BU')^{-1}U^{\top} \cdot (\gamma_1 + \partial \gamma )$ }
		\State $\partial E\leftarrow E_{\partial S} - B_{(\partial S\cap S)}\cdot M_{(\partial S\cap S), \partial S}$\label{alg:line:partial_E_1}
		\State $(\partial E)_{S'}\leftarrow -(\partial E)_{S'}$, ~ $(\partial E)_{(S\cap \partial S)\backslash S'}\leftarrow 0$ \Comment{local variable $\partial E\in \R^{6n^a\times |\partial S|}$} \label{alg:line:partial_E_2}
		%\State $U^{\tmp}\leftarrow [B_{(\partial S\backslash S)}, B_{\partial S}, E_{(\partial S\backslash S)}]$
		\State $U^{\tmp}\leftarrow [B_{\partial S}, B_{\partial S}, \partial E]$ 
		\State \Comment{local variable $U^{\tmp}\in \R^{6n^a\times 3|\partial S|}$, $U^{\tmp}=BU'$ (Corollary~\ref{cor:U_tmp_correctness})} \label{alg:line:U_tmp}
		\State $\gamma^{\tmp} \leftarrow U^{\tmp}(C^{-1}+U^{\top}U^{\tmp})^{-1}U^{\top} \cdot (\gamma_1 + \partial \gamma )$ \label{alg:line:gamma_tmp} \Comment{local variable, $\gamma^{\tmp}\in \R^{6n^a}$}
		%%\State $r_{4,4}\leftarrow \left(\L_c[(Q[l])_{S^{\new}}] + F[l] + R[l]\Gamma\cdot \L_c[M_{\partial S\backslash S}] + R[l]\partial\Gamma \cdot \L_c[M_{S^{\new}}]\right)\cdot r_{4,4}^{\tmp} $ \label{alg:line:r_4_4}
		%%\State $r_4 \leftarrow r_{4,1} + r_{4,2} + r_{4,3} + r_{4,4}$ \Comment{$r_4 \in n^b$} \label{alg:line:r_4}
		\State $r_4\leftarrow \Big(\L_c[(Q[l])_{S^{\new}}]+F[l]+R[l] \Gamma (\L_c[M_{\partial S\backslash S}]-\L_c[M_{S'}]) + R[l]\partial \Gamma \L_c[M_{S^{\new}}]\Big) (\gamma^{\tmp}-\gamma_1-\partial \gamma)$ \label{alg:line:r_4_improved}
%		\State \Comment {$r_4= (-Q_{l,S}-R_l\Gamma M_S) \cdot (\Delta_{S,S}^{-1}+ M_{S,S} )^{-1} \cdot (\beta_{2,S}+(M_{S})^{\top} (\sqrt{\wt{V}} f(\wt{g})-\sqrt{V}f(g) ))$} \label{alg:line:query:r_4}
		\State $r \leftarrow R[l]^{\top}(r_1+r_2+r_3+r_4)$ \Comment{$r \in \R^n$} \label{alg:line:r_improved}
		\State $l \leftarrow l+1$
		\State \Return $r$
	\EndProcedure
	\State
	\State {\bf end data structure}
	\end{algorithmic}
\end{algorithm}

\begin{algorithm}[!t]\caption{Data structure : \textsc{MatrixUpdate}()}\label{alg:matrix_update_improved}
\small
	\begin{algorithmic}[1]	
	\State{\bf data structure} \Comment{Theorem~\ref{thm:main_data_structure_improved}}
	\State	
		\Procedure{\textsc{MatrixUpdate}}{$w^{\appr} $} \Comment{Lemma~\ref{lem:matrix_update_correct_improved}, Lemma~\ref{lem:matrix_update_time_improved}, Lemma~\ref{lem:main_amortize_matrix_update}}
	    \State $\_, \_, \_,  \_, \Delta^{\new}, \Gamma^{\new}, \_, S^{\new}, \_ \leftarrow \textsc{ComputeLocalVariables}(w^{\appr}, \_)$ \label{alg:line:improved_matrix_update_com_local} \Comment{Algorithm~\ref{alg:compute_local_variables_improved}}
		\State $M^{\tmp} \leftarrow M - M_{S^{\new}} \cdot ( (\Delta^{\new}_{S^{\new},S^{\new}})^{-1} + M_{S^{\new},S^{\new}} )^{-1} \cdot ( M_{S^{\new}} )^\top$ \label{alg:line:matrix_update_improved:M^new}
		\State $Q^{\tmp} \leftarrow Q + R(\Gamma^{\new} M^{\tmp}) + R\sqrt{V}(M^{\tmp} -M) $ \label{alg:line:matrix_update_improved:Q^new}
		\State $\beta_1^{\tmp} \leftarrow Q^{\tmp} \sqrt{ W^{\appr} } f ( g ) $ \label{alg:line:matrix_update_improved:beta_1}
		\State $\beta_2^{\tmp} \leftarrow M^{\tmp} \sqrt{ W^{\appr} } f ( g ) $ \label{alg:line:matrix_update_improved:beta_2}
		\State $\xi^{\tmp} \leftarrow \sqrt{W^{\appr}}(f(\wt{g})-f(g))$ \label{alg:line:matrix_update_improved:xi}
		\State \Comment{We start to refresh variables in the memory of data structure}
		\State $Q\leftarrow Q^{\tmp}$, $M\leftarrow M^{\tmp}$
		\State $\beta_1\leftarrow \beta_1^{\tmp}$, $\beta_2\leftarrow \beta_2^{\tmp}$, $\xi \leftarrow \xi^{\tmp}$
		\State $v \leftarrow \wt{v} \leftarrow w^{\appr}$ \label{alg:line:matrix_update_improved:v_g}
		\State $B\leftarrow I$, $F\leftarrow 0$, $E\leftarrow 0$ \label{alg:line:matrix_update_improved:B_F_E}
		\State $S\leftarrow\emptyset$, $\Delta \leftarrow \Gamma \leftarrow 0$, $\gamma_1 \leftarrow \gamma_2 \leftarrow 0$ \label{alg:line:matrix_update_improved:everything_else}
		\EndProcedure
		\State
		\State {\bf end data structure}
	\end{algorithmic}
\end{algorithm}

\begin{algorithm}[!t]\caption{Data structure : \textsc{PartialMatrixUpdate}(). }\label{alg:partial_matrix_update_improved}
\small
	\begin{algorithmic}[1]	
	\State{\bf data structure} \Comment{Theorem~\ref{thm:main_data_structure_improved}}
	\State	
		\Procedure{\textsc{PartialMatrixUpdate}}{$w^{\appr}$} \Comment{Lemma~\ref{lem:partial_matrix_update_correct_improved}, Lemma~\ref{lem:partial_matrix_update_time_improved}, Lemma~\ref{lem:main_amortize_partial_matrix_update}}
	    \State $\_, \partial \Gamma, \_, \partial S, \Delta^{\new}, \Gamma^{\new}, \_, S^{\new}, \_ \leftarrow \textsc{ComputeLocalVariables}(w^{\appr}, \_)$ \Comment{Algorithm~\ref{alg:compute_local_variables_improved}}
		\State $(U', C, U) \leftarrow \textsc{Decompose}\Big(\L_*[(\Delta^{\new}_{S^{\new}, S^{\new}})^{-1} +  M_{S^{\new}, S^{\new}} ] - \L_* [\Delta_{S, S}^{-1}+ M_{S,S}]\Big)$ \label{alg:line:partial_matrix_update_improved:UCU}
		\State \Comment{\textsc{Decompose} is defined in Lemma \ref{lem:UCU_decomposition}}
		\State $B^{\tmp} \leftarrow B - BU'(C^{-1}+U^{\top}BU')^{-1}U^{\top}B$ \label{alg:line:partial_matrix_update_improved:B}
		%\State $F^{\tmp} \leftarrow F + \L_c[R\Gamma M_{\partial S\backslash S} + R\partial \Gamma M_{S^{\new}}]$
		\State $F^{\tmp} \leftarrow F +R \Gamma \cdot (\L_c[M_{\partial S\backslash S}]-\L_c[M_{S'}]) + R\partial \Gamma\cdot \L_c[M_{S^{\new}}]$
		\label{alg:line:partial_matrix_update_improved:F}
%		\State {\color{pink} $E^{\tmp} \leftarrow E + B\L_r[ (M_{\partial S\backslash S})^{\top}]$}
		%\State $E^{\tmp} \leftarrow E + B^{\tmp}\L_r[ (M_{\partial S\backslash S})^{\top}] - BU'(C^{-1}+U^{\top}BU')^{-1}U^{\top}E$
		\State $E^{\tmp} \leftarrow E + B^{\tmp}(\L_r[ (M_{\partial S\backslash S})^{\top}]-\L_r[ (M_{S'})^{\top}]) - BU'(C^{-1}+U^{\top}BU')^{-1}U^{\top}E$ \label{alg:line:partial_matrix_update_improved:E}
		\State $\xi^{\tmp} \leftarrow \sqrt{W^{\appr}}f(\wt{g}) - \sqrt{V}f(g)$ \label{alg:line:partial_matrix_update_improved:xi}
		\State $\gamma_1^{\tmp} \leftarrow B^{\tmp}\cdot \L_r [\beta_{2,S^{\new}}]+B^{\tmp}\cdot \L_r [(M_{S^{\new}})^{\top}] \xi^{\tmp}$ \label{alg:line:partial_matrix_update_improved:gamma_1}
		\State $\gamma_2^{\tmp} \leftarrow \gamma_2 + (\Gamma+\partial \Gamma) M (\sqrt{W^{\appr}} - \sqrt{\wt{V}})f(\wt{g}) + \partial \Gamma M (\sqrt{\wt{V}}f(\wt{g}) - \sqrt{V}f(g))$ \label{alg:line:partial_matrix_update_improved:gamma_2}
		\State \Comment{We start to refresh variables in the memory of data structure}
		\State $B\leftarrow B^{\tmp}$, $F\leftarrow F^{\tmp}$, $E\leftarrow E^{\tmp}$
		\State $\xi \leftarrow \xi^{\tmp}$ $\gamma_1\leftarrow \gamma_1^{\tmp}$, $\gamma_2\leftarrow \gamma_2^{\tmp}$
		\State $\wt{v} \leftarrow w^{\appr}, S \leftarrow S^{\new}$, $\Delta \leftarrow \Delta^{\new}$, $\Gamma \leftarrow \Gamma^{\new}$ \label{alg:line:partial_matrix_update_improved:everything_else}
		\EndProcedure
		\State
		\State {\bf end data structure}
	\end{algorithmic}
\end{algorithm}

\begin{algorithm}[!t]\caption{Data structure : \textsc{VectorUpdate}(). }\label{alg:vector_update_improved}
\small
	\begin{algorithmic}[1]	
	\State{\bf data structure} \Comment{Theorem~\ref{thm:main_data_structure_improved}}
	\State	
		\Procedure{\textsc{VectorUpdate}}{$h^{\appr} $} \Comment{Lemma~\ref{lem:vector_update_correct_improved}, Lemma~\ref{lem:vector_update_time_improved}, Lemma~\ref{lem:main_amortize_vector_update}}
	    %\State $\_ , \_, \_, \_, \_, \_, \_, \_ \leftarrow \textsc{ComputeLocalVariables}(v, g^{\new})$ \Comment{Algorithm~\ref{alg:compute_local_variables_improved}}
		\State $\beta_1^{\tmp} \leftarrow \beta_1 + Q \sqrt{V} (f(h^{\appr}) - f(g))$ \label{alg:line:vector_update_improved:beta_1}
		\State $\beta_2^{\tmp} \leftarrow \beta_2 + M \sqrt{V} (f(h^{\appr}) - f(g))$ \label{alg:line:vector_update_improved:beta_2}
		\State $\xi^{\tmp} \leftarrow (\sqrt{\wt{V}}-\sqrt{V})f(h^{\appr})$ \label{alg:line:vector_update_improved:xi}
		\State $\gamma_1^{\tmp} \leftarrow B\cdot \L_r [(\beta^{\tmp}_2)_{S}]+B\cdot \L_r [(M_{S})^{\top}]\cdot \xi^{\tmp}$ \label{alg:line:vector_update_improved:gamma_1}
		\State $\gamma_2^{\tmp} \leftarrow \Gamma M\cdot \xi^{\tmp}$ \label{alg:line:vector_update_improved:gamma_2}
		\State \Comment{We start to refresh variables in the memory of data structure}
		\State $\beta_1\leftarrow \beta_1^{\tmp}$, $\beta_2\leftarrow \beta_2^{\tmp}$, $\xi\leftarrow \xi^{\tmp}$, $\gamma_1\leftarrow \gamma_1^{\tmp}$, $\gamma_2\leftarrow \gamma_2^{\tmp}$
		\State $g\leftarrow \wt{g}\leftarrow h^{\appr}$, \label{alg:line:vector_update_improved:g}
		\State $T\leftarrow \emptyset$ \label{alg:line:vector_update_improved:T}
		\EndProcedure
		\State
		\State {\bf end data structure}
	\end{algorithmic}
\end{algorithm}

\begin{algorithm}[!t]\caption{Data structure : \textsc{PartialVectorUpdate}().}\label{alg:partial_vector_update_improved}
\small
	\begin{algorithmic}[1]	
	\State{\bf data structure} \Comment{Theorem~\ref{thm:main_data_structure_improved}}
	\State	 
		\Procedure{\textsc{PartialVectorUpdate}}{$h^{\appr} $} \Comment{Lemma~\ref{lem:partial_vector_update_correct_improved}, Lemma~\ref{lem:partial_vector_update_time_improved}, Lemma~\ref{lem:main_amortize_partial_vector_update}}
	    %\State $\_, \_, \_, \_, \_, \_, \_, \_ \leftarrow \textsc{ComputeLocalVariables}(\wt{v}^{\new}, \wt{g}^{\new})$
		\State $\xi^{\tmp} \leftarrow \sqrt{\wt{V}}f(h^{\appr})-\sqrt{V}f(g)$ \label{alg:line:partial_vector_update_improved:xi}
		\State $\gamma_1^{\tmp} \leftarrow \gamma_1 + B \cdot \L_r [(M_{S})^{\top}] \cdot \sqrt{\wt{V}}\big( f(h^{\appr}) - f(\wt{g})\big)$ \label{alg:line:partial_vector_update_improved:gamma_1}
		\State $\gamma_2^{\tmp} \leftarrow \gamma_2 + \Gamma M\sqrt{\wt{V}}\big( f(h^{\appr}) - f(\wt{g})\big)$ \label{alg:line:partial_vector_update_improved:gamma_2}
		\State \Comment{We start to refresh variables in the memory of data structure}
		\State $\xi \leftarrow \xi^{\tmp}$,  $\gamma_1\leftarrow \gamma_1^{\tmp}$, $\gamma_2 \leftarrow \gamma_2^{\tmp}$
		\State $T \leftarrow \supp(h^{\appr} - g)$ \label{alg:line:partial_vector_update_improved:T}
		\State $\wt{g}\leftarrow \wt{g}^{\new}$ \label{alg:line:partial_vector_update_improved:wt_g}
		\EndProcedure
		\State
		\State {\bf end data structure}
	\end{algorithmic}
\end{algorithm}

\section{Data structure : correctness}
\label{sec:correctness_improved}
The purpose of this section is to show the correctness of our data structure, stated in Theorem~\ref{thm:update_query_correct_improved}. 
We start with the invariants that we maintain for data structure members.
\begin{assumption}[Invariants]\label{ass:invariant_improved}
The following invariants are maintained in the data structure:
\begin{multicols}{2}
\begin{enumerate}
\item \label{ass:invariant_improved:M} $M = A^\top ( A V A^\top )^{-1} A$,
\item \label{ass:invariant_improved:Q} $Q = R \sqrt{V} M$, 
\item \label{ass:invariant_improved:beta_1} $\beta_1 = Q \sqrt{V} f(g)$,
\item \label{ass:invariant_improved:beta_2} $\beta_2 = M \sqrt{V} f(g)$,
\item \label{ass:invariant_improved:S} $S = \supp( \wt{v} - v )$,
\item \label{ass:invariant_improved:T} $T = \supp( \wt{g} - g )$,
\item \label{ass:invariant_improved:Delta} $\Delta = \wt{V} - V$,
\item \label{ass:invariant_improved:Gamma} $\Gamma = \sqrt{\wt{V}} - \sqrt{V}$,
\item \label{ass:invariant_improved:xi} $\xi = \sqrt{\wt{V}} f(\wt{g}) - \sqrt{V} f(g)$,
\item \label{ass:invariant_improved:B} $B = \L_*[ (\Delta_{S,S}^{-1} + M_{S,S})^{-1} ]$,
\item \label{ass:invariant_improved:E} $E = B\cdot \L_r[(M_S)^{\top}]$,
\item \label{ass:invariant_improved:F} $F = R\Gamma \cdot \L_c[M_S]$,
\item \label{ass:invariant_improved:gamma_1} $\gamma_1 = B \cdot \L_r [\beta_{2,S}] + B \cdot \L_r[(M_S)^{\top}] \cdot \xi$,
\item \label{ass:invariant_improved:gamma_2} $\gamma_2 = \Gamma M \cdot \xi$.
\end{enumerate}
\end{multicols}
\end{assumption}
We will prove that if these invariants are true before we enter a procedure, they are still true when the proecedure returns. Thus the correctness of the invariants can be proved by induction.

The following local variables are used %extensively
in procedures \textsc{Query} (Algorithm~\ref{alg:query_improved}), \textsc{MatrixUpdate} (Algorithm~\ref{alg:matrix_update_improved}), and \textsc{PartialMatrixUpdate} (Algorithm~\ref{alg:partial_matrix_update_improved}). For clarity of the presentation, we write their definitions here.
\begin{definition}[Local variables]\label{def:local}
Given inputs $w^{\appr}$ and $h^{\appr}$, we define these local variables:
\begin{multicols}{2}
\begin{enumerate}
\item \label{def:local:partial_Delta} $\partial \Delta \leftarrow  W^{\appr} - \wt{V}$,
\item \label{def:local:partial_Gamma} $\partial \Gamma \leftarrow \sqrt{W^{\appr}} - \sqrt{\wt{V}}$,
\item \label{def:local:partial_xi}  $\partial \xi \leftarrow \sqrt{W^{\appr}}f(h^{\appr}) - \sqrt{\wt{V}}f(\wt{g})$,
\item \label{def:local:partial_S} $\partial S \leftarrow \supp(w^{\appr} - \wt{v})$,
\item \label{def:local:Delta_new} $\Delta^{\new} \leftarrow \Delta + \partial \Delta$,
\item \label{def:local:Gamma_new} $\Gamma^{\new} \leftarrow \Gamma + \partial \Gamma$,
\item \label{def:local:xi_new} $\xi^{\new} \leftarrow \xi + \partial \xi$,
\item \label{def:local:S_new} $S^{\new} \leftarrow \supp(w^{\appr}-v)$,
\item \label{def:local:S'} $S' \leftarrow (S \cup \partial S)\backslash S^{\new}$.
{\color{white} \item}
\end{enumerate}
\end{multicols}
\end{definition}
\begin{remark}[Compute local variables]\label{rem:compute_local_variables}
 The private procedure \textsc{ComputeLocalVariables} (Algorithm~\ref{alg:compute_local_variables_improved}) computes these local variables (defined in Definition~\ref{def:local}) correctly.
\end{remark}
\begin{remark}[Temporary variables]
All variables with super-script $``\tmp"$ are temporary local variables that are only used in update procedures. 
\end{remark}
\begin{remark}[Properties of $S'$ and $S^{\new}$]
Note that $S^{\new}\subseteq S\cup \partial S$, and $S'\subseteq S\cap \partial S$.
\end{remark}

In this section we prove the following main theorem using lemmas proved in later sections.
\begin{theorem}[Correctness of \textsc{UpdateQuery}]
\label{thm:update_query_correct_improved}
%Under Assumption~\ref{ass:epsilon_far}, 
On the $j$-th call to the procedure \textsc{UpdateQuery} (Algorithm~\ref{alg:update_query_improved}), the output satisfies the following:
\begin{enumerate}
\item $w^{\appr} \approx_{\epsilon_{\mathrm{mp}}} w^{\new}$, $h^{\appr} \approx_{\epsilon_{\mathrm{mp}}} h^{\new}$, 
\item $r = R[l]^{\top}R[l] \sqrt{ W^{\appr} } M^{\new} \sqrt{ W^{\appr} } f( h^{\appr} ) $, where $M^{\new} = A^{\top}(AW^{\appr}A^{\top})^{-1}A$.
\end{enumerate}
\end{theorem}
\begin{proof}
\textbf{Part 1.} 
The output $w^{\appr}$ is returned by \textsc{UpdateV}, it can be $v^{\new}$ returned from Line~\ref{alg:line:return_by_matrix_update}, or it can be $\wt{v}^{\new}$ returned from Line~\ref{alg:line:return_by_partial_matrix_update} and \ref{alg:line:update_v_return_by_default} (in Algorithm~\ref{alg:update_v_improved}). 

In the case of $w^{\appr}=v^{\new}$, the properties of $v^{\new}$ are given in Fact~\ref{fac:characterization_k_v_new}.
%the algorithm must enter the \textsc{MatrixUpdate}. 
According to Part 2 and 3 of Fact~\ref{fac:characterization_k_v_new}, there exists a permutation $\pi:[n]\rightarrow [n]$ and a number $k$ such that $\forall i \in \pi([k])$, $v^{\new}_i=w^{(j+1)}_i$ and $\forall i \notin \pi([k])$, $v^{\new}_i\approx_{\epsilon_{\mathrm{mp}}}w^{(j+1)}_i$, where $w^{(j+1)}$ is defined as $w^{\new}$ in the $j$-th iteration. So $v^{\new}\approx_{\epsilon_{\mathrm{mp}}}w^{\new}$ in this case.

If $w^{\appr}=\wt{v}^{\new}$, the properties of $\wt{v}^{\new}$ are given in Fact~\ref{fac:characterization_wt_k_tmp_wt_v_tmp_wt_v_new}. According to Part 3 and 4 of Fact~\ref{fac:characterization_wt_k_tmp_wt_v_tmp_wt_v_new}, there exists a permutation $\pi:[n]\rightarrow [n]$ and a number $\wt{k}$ such that $\forall i \in \pi([\wt{k}])$, $\wt{v}^{\new}_i \approx_{\epsilon_{\far}} w^{(j+1)}_i$ and $\forall i \notin \pi([\wt{k}])$, $\wt{v}^{\new}_i\approx_{\epsilon_{\mathrm{mp}}}w^{(j+1)}_i$, where $w^{(j+1)}$ is defined as $w^{\new}$ in the $j$-th iteration. Using the assumption that $\epsilon_{\far}<\epsilon_{\mathrm{mp}}$, we get $w^{\appr}\approx_{\epsilon_{\mathrm{mp}}} w^{\new}$.

$h^{\appr} \approx_{\epsilon_{\mathrm{mp}}} h^{\new}$ follows by similar reasons.

%Since $w^{\appr}$ is the output of $\soft(w^{\new},\wt{v},\epsilon_{\mathrm{mp}})$ and $h^{\appr}$ is the output of $\soft(h^{\new},\wt{g},\epsilon_{\mathrm{mp}})$, we have 
%\[
%w^{\appr} \approx_{\epsilon_{\mathrm{mp}}} w^{\new}, h^{\appr} \approx_{\epsilon_{\mathrm{mp}}} h^{\new}.
%\]

\noindent \textbf{Part 2.} First we prove by induction that all invariants of Assumption~\ref{ass:invariant_improved} hold all the time. In the beginning, the data structure calls \textsc{Initialize}. By Lemma~\ref{lem:initialize_correct_improved}, all invariants hold. 

In the following iterations, the data structure is only accessed via calls to its procedure \textsc{UpdateQuery} by \textsc{OneStepCentralPath}
(Line~\ref{alg:line:p_t} and \ref{alg:line:p_phi} in Algorithm~\ref{alg:one_step_central_path}). \textsc{UpdateQuery} calls \textsc{UpdateV}, \textsc{UpdateG} and \textsc{Query}(Line~\ref{alg:line:update_v_in_update_query_improved}, \ref{alg:line:update_g_in_update_query_improved}, \ref{alg:line:query_in_update_query_improved} in Algorithm~\ref{alg:update_query_improved}). The procedure \textsc{Query} does not modify any data structure member, so it won't violate any invarint. By Part 2 of Lemma~\ref{lem:update_v_correct_improved} and \ref{lem:update_g_correct_improved}, if all invariants are satisfied before entering the procedure \textsc{UpdateV} (or \textsc{UpdateG}), then after executing the procedure \textsc{UpdateV} (or \textsc{UpdateG}), all invariants are still satisfied. By induction, all invariants of Assumption~\ref{ass:invariant_improved} hold all the time. 

Next we prove that before entering \textsc{Query}, we always have $|S\cup \partial S| \leq 2n^a$, so that all $\L$, $\L_c$, $\L_r$, $\L_*$ operators are well-defined. Note that
\begin{align*}
S =&~ \supp(\wt{v}-v) && \text{(by Part \ref{ass:invariant_improved:S} of Assumption~\ref{ass:invariant_improved}),}\\
\partial S =&~ \supp(w^{\appr} - \wt{v}) && \text{(by Part \ref{def:local:partial_S} of Definition~\ref{def:local}).}
\end{align*}
By Part 1 of Lemma~\ref{lem:update_v_correct_improved}, we have 
$\|w^{\appr} - \wt{v}\|_0\leq n^{\wt{a}}$, so $|\partial S|\leq n^{\wt{a}}$. By Corollary~\ref{cor:improved:sparsity_of_members}, we have
$\|\wt{v}-v\|_0 \leq n^a$, so $|S|\leq n^a$. Therefore
\begin{align*}
    |S\cup \partial S| \leq |S| + |\partial S| \leq n^a + n^{\wt{a}} \leq 2n^a,
\end{align*}
where we use the fact that $\wt{a}\leq a$.

%From the correctness of the procedure \textsc{Initialize} (Lemma~\ref{lem:initialize_correct_improved}), we know that all invariants are satisfied in the beginning. Since we also know that if all invariants are satisfied before entering the procedure \textsc{Update}, they are also satisfied after leaving \textsc{Update} (Part 2 of Lemma~\ref{lem:update_correct_improved}), and also since the procedure \textsc{Query} does not modify any data structure member, by induction we know that the invariants are satisfied throughout the algorithm.

Now the two conditions of Lemma~\ref{lem:query_correct_improved} are both satisfied, so we have \[r = R[l]^{\top}R[l] \sqrt{ W^{\appr} } M^{\new} \sqrt{ W^{\appr} } f( h^{\appr} ),\]
where $M^{\new} = A^{\top}(AW^{\appr}A^{\top})^{-1}A$.
\end{proof}

\begin{table}[ht]
\small
    \centering
    \begin{tabular}{|l|l|l|} 
         \hline
         {\bf Procedure} & {\bf Lemma} & {\bf Section} \\ \hline
         \textsc{UpdateQuery} & Theorem~\ref{thm:update_query_correct_improved} & -- \\ \hline
         \textsc{Query} & Lemma~\ref{lem:query_correct_improved} & Section~\ref{sec:query_correct_improved} \\ \hline
         \textsc{UpdateV} & Lemma~\ref{lem:update_v_correct_improved} & Section~\ref{sec:update_v_g_correct_improved} \\ \hline
         \textsc{UpdateG} & Lemma~\ref{lem:update_g_correct_improved} & Section~\ref{sec:update_v_g_correct_improved} \\ \hline
         \textsc{MatrixUpdate} & Lemma~\ref{lem:matrix_update_correct_improved} & Section~\ref{sec:matrix_update_correct_improved} \\ \hline 
         \textsc{PartialMatrixUpdate} & Lemma~\ref{lem:partial_matrix_update_correct_improved} & Section~\ref{sec:partial_matrix_update_correct_improved} \\ \hline
         \textsc{VectorUpdate} & Lemma~\ref{lem:vector_update_correct_improved} & Section~\ref{sec:vector_update_correct_improved} \\ \hline
         \textsc{PartialVectorUpdate} & Lemma~\ref{lem:partial_vector_update_correct_improved} & Section~\ref{sec:partial_vector_update_correct_improved} \\ \hline
         \textsc{Initialize} & Lemma~\ref{lem:initialize_correct_improved} & Section~\ref{sec:initialize_correct_improved} \\ \hline
    \end{tabular}
     \caption{Summary of the section that proves the correctness of the data structure.}\label{tab:correct_summary_improved}
\end{table}

%%%%%%%%%%%%%%%%%%%%%%%%%%%%%%%%%%%%%%%%%%%%%%%%
\subsection{Correctness of \textsc{Query}}\label{sec:query_correct_improved}
In this section we follow the notation of the procedure \textsc{Query} (Algorithm~\ref{alg:query_improved}). Note that $v$, $\wt{v}$, $g$, $\wt{g}$, $M$, $Q$, $R$, $\beta_1$, $\beta_2$, $\gamma_1$, $\gamma_2$, $B$, $E$, $F$, $\Delta$, $\Gamma,S$ are all members of the data structure (See Algorithm~\ref{alg:member_improved}). \textsc{Query} (Algorithm~\ref{alg:query_improved}) takes $w^{\appr}$ and $h^{\appr}$ as input, and uses the inputs and members of the data structure to compute the following local variables:
$\partial \Delta$, $\partial \Gamma$, $\partial \xi$,  $\partial S$, $\partial \gamma$, $\Delta^{\new}$, $S^{\new}$, $S'$, $U'$, $C$, $U$, $\partial E$, $U^{\tmp}$, $\gamma^{\tmp}$, $r_1$, $r_2$, $r_3$, $r_4$. Finally, \textsc{Query} (Algorithm~\ref{alg:query_improved}) outputs $r$. The goal of this section is to prove Lemma~\ref{lem:query_correct_improved} which gives a close-form formula of the output $r$.

\begin{lemma}[Correctness of \textsc{Query}]\label{lem:query_correct_improved}
Before entering \textsc{Query} (Algorithm~\ref{alg:query_improved}), if we have the following two guarantees: $|S\cup \partial S| \leq 2n^a$, and all the invariants of Assumption~\ref{ass:invariant_improved} are satisfied, then the output $r$ is
\[r = R[l]^{\top}R[l] \sqrt{ W^{\appr} } M^{\new} \sqrt{ W^{\appr} } f( h^{\appr} ),\] where $M^{\new} = A^{\top}(AW^{\appr}A^{\top})^{-1}A$.
\end{lemma}
This lemma is proved in Claim~\ref{cla:query_correct_improved} using the following:
\begin{enumerate}
\item $r_1 = Q[l] \sqrt{V} f(g)$ \hspace{102mm} (Claim~\ref{cla:query_correct_improved_part_1})
\item $r_2 = (Q[l] + R[l](\Gamma + \partial \Gamma) M)\cdot (\sqrt{W^{\appr}}f(h^{\appr}) -\sqrt{V}f(g) )$ \hspace{31mm} (Claim~\ref{cla:query_correct_improved_part_2})
\item $r_3 = R[l](\Gamma+\partial \Gamma) M\sqrt{V} f(g)$ \hspace{85mm} (Claim~\ref{cla:query_correct_improved_part_3})
\item $r_4 = -R[l]\sqrt{W^{\appr}} M_{S^{\new}}  ( (\Delta_{S^{\new},S^{\new}}^{\new})^{-1} + M_{S^{\new},S^{\new}} )^{-1} (M_{S^{\new}})^{\top} \sqrt{W^{\appr}}f(h^{\appr})$ (Claim~\ref{cla:query_correct_improved_part_4})
\end{enumerate}
Now we prove these claims one by one. In the following we assume that $|S\cup \partial S|\leq 2n^a$ and all the invariants of Assumption~\ref{ass:invariant_improved} are satisfied. Note that when $|S\cup \partial S|\leq 2n^a$, all of the $\L$, $\L_c$, $\L_r$, and $\L_*$ are well-defined, and the \textsc{Decompose} function is also well-defined.

\begin{claim}[Close-form formula for $r_1$]%[Part 1 of Lemma~\ref{lem:query_correct_improved}]
\label{cla:query_correct_improved_part_1}
We have $r_1 = Q[l] \sqrt{V} f(g)$.
\end{claim}

\begin{proof}
From the assignment of $r_1$ (Line~\ref{alg:line:query:r_1_improved} of Algorithm~\ref{alg:query_improved}), we have
\begin{align*}
r_1 = & ~ \beta_{1}[l] =  Q[l] \sqrt{V} f(g).
\end{align*}
where the last step follows from $\beta_{1} = Q \sqrt{V} f(g)$ (Part~\ref{ass:invariant_improved:Q} of Assumption~\ref{ass:invariant_improved}).
\end{proof}

%%%%%%%%%%%%%%%%%%%%%%%%%%%%%%%%%%%%%%%%%%%%%%%%
\begin{claim}[Close-form formula for $r_2$]%[Part 2 of Lemma~\ref{lem:query_correct_improved}]
\label{cla:query_correct_improved_part_2}
We have
\begin{align*}
    r_2 = (Q[l] + R[l](\Gamma + \partial \Gamma) M) \cdot (\sqrt{W^{\appr}}f(h^{\appr}) -\sqrt{V}f(g) ).
\end{align*}

\end{claim}
\begin{proof}
From the assignment of $r_2$ (Line~\ref{alg:line:query:r_2_improved} of Algorithm~\ref{alg:query_improved}), we have
\begin{align*}
r_2
= & ~ Q[l] \xi + R[l]\gamma_2 + R[l] \partial\Gamma M  (\xi + \partial \xi) + (Q[l] + R[l]\Gamma M)\partial \xi\\
= & ~ (Q[l] + R[l]\Gamma M) \xi + R[l] \partial\Gamma M  (\xi + \partial \xi) + (Q[l] + R[l]\Gamma M)\partial \xi\\
%= &~ (Q[l] + R[l]\Gamma M)(\xi + \partial \xi) + R[l] \partial\Gamma M  (\xi + \partial \xi)\\
= &~ \big(Q[l] + R[l](\Gamma+\partial \Gamma) M\big)\cdot (\xi + \partial \xi)\\
= & ~\big(Q[l] + R[l](\Gamma + \partial \Gamma) M\big) \cdot (\sqrt{W^{\appr}}f(h^{\appr}) -\sqrt{V}f(g)),
\end{align*}
where the second step follows from $\gamma_2 = \Gamma M \cdot \xi$ (Part~\ref{ass:invariant_improved:gamma_2} of Assumption~\ref{ass:invariant_improved}), the third step follows from merging terms, and the fourth step follows from the invariant $\xi=\sqrt{\wt{V}}f(\wt{g}) - \sqrt{V}f(g)$ (Part~\ref{ass:invariant_improved:xi} of Assumption~\ref{ass:invariant_improved}) and the definition $\partial \xi = \sqrt{W^{\appr}}f(h^{\appr}) - \sqrt{\wt{V}}f(\wt{g})$ (Part~\ref{def:local:partial_xi} of Definition~\ref{def:local}).
\end{proof}

\begin{claim}[Close-form formula for $r_3$]%[Part 3 of Lemma~\ref{lem:query_correct_improved}]
\label{cla:query_correct_improved_part_3}
We have $r_3 = R[l](\Gamma+\partial \Gamma) M\sqrt{V} f(g)$.
\end{claim}
\begin{proof}
From the assignment of $r_3$ (Line~\ref{alg:line:query:r_3_improved} of Algorithm~\ref{alg:query_improved}), we have
\begin{align*}
r_3 
= R[l](\Gamma+\partial \Gamma) \beta_2
= R[l](\Gamma+\partial \Gamma) M\sqrt{V} f(g),
\end{align*}
where the second step follows from the invariant $\beta_2 = M \sqrt{V} f(g)$ (Part~\ref{ass:invariant_improved:beta_2} of Assumption~\ref{ass:invariant_improved}).
\end{proof}

\begin{claim}[Close-form formula for $r_4$]%[Part 4 of Lemma~\ref{lem:query_correct_improved}]
\label{cla:query_correct_improved_part_4}
We have
\begin{align*}
r_4 = -R[l]\sqrt{W^{\appr}} M_{S^{\new}} \cdot ((\Delta_{S^{\new},S^{\new}}^{\new})^{-1} + M_{S^{\new},S^{\new}} )^{-1} \cdot (M_{S^{\new}})^{\top} \sqrt{W^{\appr}}f(h^{\appr}).
\end{align*}
\end{claim}
\begin{proof}
First note that the left part of $r_4$ is
\begin{align}\label{eq:r_4_left}
    &~ \L_c[(Q[l])_{S^{\new}}]+F[l]+R[l] \Gamma \cdot (\L_c[M_{\partial S\backslash S}]-\L_c[M_{S'}]) + R[l]\partial \Gamma\cdot \L_c[M_{S^{\new}}] \notag \\
    = &~ \L_c[(Q[l])_{S^{\new}}]+R[l]\Gamma\cdot \L_c[M_S]+R[l] \Gamma \cdot (\L_c[M_{\partial S\backslash S}]-\L_c[M_{S'}]) + R[l]\partial \Gamma\cdot \L_c[M_{S^{\new}}]\notag\\
    = &~ \L_c[(Q[l])_{S^{\new}}]+R[l]\Gamma\cdot \L_c[M_{S^{\new}}] + R[l]\partial \Gamma\cdot \L_c[M_{S^{\new}}]\notag\\
    = &~ \L_c[(Q[l])_{S^{\new}}]+R[l](\Gamma+\partial \Gamma)\cdot \L_c[M_{S^{\new}}]\notag \\
    = &~ R[l]\sqrt{V}\L_c[M_{S^{\new}}]+R[l](\Gamma+\partial \Gamma)\cdot \L_c[M_{S^{\new}}]\notag\\
    = &~ R[l]\sqrt{W^{\appr}}\L_c[M_{S^{\new}}],
\end{align}
where the first step follows from $F=R\Gamma \cdot \L_c[M_S]$ (Part~\ref{ass:invariant_improved:F} of Assumption~\ref{ass:invariant_improved}), the second step follows from $S'=(S\cup \partial S)\backslash S^{\new}$ (Part~\ref{def:local:S'} of Definition~\ref{def:local}) and thus $\L_c[M_{S'}]+\L_c[M_{S^{\new}}]=\L_c[M_{S\cup \partial S}]=\L_c[M_S]+\L_c[M_{\partial S\backslash S}]$ by Part~\ref{rem:L_operator:addition} of Remark~\ref{rem:L_operator}, the fourth step follows from $Q=R\sqrt{V}M$ (Part~\ref{ass:invariant_improved:Q} of Assumption~\ref{ass:invariant_improved}) and Part~\ref{rem:L_operator:multiplication_1} of Remark~\ref{rem:L_operator}, and the fifth step follows from $\Gamma=\sqrt{\wt{V}}-\sqrt{V}$ (Part~\ref{ass:invariant_improved:Gamma} of Assumption~\ref{ass:invariant_improved}) and $\partial \Gamma=\sqrt{W^{\appr}}-\sqrt{\wt{V}}$ (Part~\ref{def:local:partial_Gamma} of Assumption~\ref{def:local}).

We also have
\begin{align} \label{eq:r_4_gamma}
    \gamma_1 + \partial \gamma = &~B\cdot \L_r[(\beta_2)_S]+B\cdot \L_r[(M_S)^{\top}]\cdot \xi +\notag\\
    &~B\cdot (\L_r [(\beta_{2})_{\partial S\backslash S}]-\L_r [(\beta_{2})_{S'}]) + B\cdot (\L_r[(M_{\partial S \backslash S})^{\top}]-\L_r[(M_{S'})^{\top}]) \cdot (\xi+\partial \xi) + E \cdot \partial \xi \notag\\
    = &~ \underbrace{B\cdot \L_r[(\beta_2)_S]+B\cdot (\L_r [(\beta_{2})_{\partial S\backslash S}]-\L_r [(\beta_{2})_{S'}])}_{a_1}+\notag\\
    &~\underbrace{B\cdot \L_r[(M_S)^{\top}]\cdot \xi + B\cdot (\L_r[(M_{\partial S \backslash S})^{\top}]-\L_r[(M_{S'})^{\top}]) \cdot (\xi+\partial \xi) + E \cdot \partial \xi}_{a_2},
\end{align}
where the first step follows from $\gamma_1=B\cdot \L_r[(\beta_2)_S]+B\cdot \L_r[(M_S)^{\top}]\cdot \xi$ (Part~\ref{ass:invariant_improved:gamma_1} of Assumption~\ref{ass:invariant_improved}) and the assignment of $\partial \gamma$ on Line~\ref{alg:line:partial_gamma} of Algorithm~\ref{alg:query_improved}, the second step follows from changing the order of terms. And
\begin{align}\label{eq:r_4_a_1}
    a_1 %= &~ B\cdot \L_r[(\beta_2)_S]+B\cdot (\L_r [(\beta_{2})_{\partial S\backslash S}]-\L_r [(\beta_{2})_{S'}])\notag\\
    = &~ B\cdot \L_r[(\beta_2)_{S^{\new}}],
\end{align}
which follows from $S'=(S\cup \partial S)\backslash S^{\new}$ (Part~\ref{def:local:S'} of Definition~\ref{def:local}) and thus $\L_c[(\beta_2)_{S'}]+\L_c[(\beta_2)_{S^{\new}}]=\L_c[(\beta_2)_{S\cup \partial S}]=\L_c[(\beta_2)_S]+\L_c[(\beta_2)_{\partial S\backslash S}]$ by Part~\ref{rem:L_operator:addition} of Remark~\ref{rem:L_operator}. And also
\begin{align}\label{eq:r_4_a_2}
    a_2 %=&~\Big(B\cdot \L_r[(M_S)^{\top}]\cdot \xi + B\cdot (\L_r[(M_{\partial S \backslash S})^{\top}]-\L_r[(M_{S'})^{\top}]) \cdot (\xi+\partial \xi) + E \cdot \partial \xi\Big) \notag\\
    =&~B\cdot \L_r[(M_S)^{\top}]\cdot \xi + B\cdot (\L_r[(M_{\partial S \backslash S})^{\top}]-\L_r[(M_{S'})^{\top}]) \cdot (\xi+\partial \xi) + B\cdot \L_r[(M_S)^{\top}] \cdot \partial \xi \notag\\
    = &~B\cdot \Big( \L_r[(M_S)^{\top}]+\L_r[(M_{\partial S \backslash S})^{\top}]-\L_r[(M_{S'})^{\top}]\Big) \cdot (\xi+\partial \xi) \notag\\
    = &~  B\cdot \L_r[(M_{S^{\new}})^{\top}]\cdot (\xi+\partial \xi),
\end{align}
where the first step follows from $E=B\cdot \L_r[(M_S)^{\top}]$, the second step follows from merging terms, and the third step follows from $S'=(S\cup \partial S)\backslash S^{\new}$ (Part~\ref{def:local:S'} of Definition~\ref{def:local}) and thus $\L_r[(M_{S'})^{\top}]+\L_r[(M_{S^{\new}})^{\top}]=\L_r[(M_{S\cup \partial S})^{\top}]=\L_r[(M_S)^{\top}]+\L_r[(M_{\partial S\backslash S})^{\top}]$ by Part~\ref{rem:L_operator:addition} of Remark~\ref{rem:L_operator}.

Combining Eq.~\eqref{eq:r_4_gamma}, \eqref{eq:r_4_a_1}, and \eqref{eq:r_4_a_2} together, we have
\begin{align}\label{eq:r_4_right}
    \gamma_1+\partial \gamma \notag = &~ B\cdot \L_r[(\beta_2)_{S^{\new}}] + B\cdot  \L_r[(M_{S^{\new}})^{\top}]\cdot (\xi+\partial \xi)\notag\\
    = &~ B\cdot \L_r[(M_{S^{\new}})^{\top}]\sqrt{V}f(g) +  B\cdot L_r[(M_{S^{\new}})^{\top}]\cdot (\sqrt{W^{\appr}}f(h^{\appr})-\sqrt{V}f(g))\notag\\
    = &~ B\cdot L_r[(M_{S^{\new}})^{\top}]\cdot \sqrt{W^{\appr}}f(h^{\appr}),
\end{align}
where the second step follows from $\beta_2=M\sqrt{V}f(g)$ (Part~\ref{ass:invariant_improved:beta_2} of Assumption~\ref{ass:invariant_improved}) and using Part~\ref{rem:L_operator:multiplication_1} of Remark~\ref{rem:L_operator}, and $\xi+\partial \xi=\sqrt{W^{\appr}}f(h^{\appr})-\sqrt{V}f(g)$ (Part~\ref{ass:invariant_improved:xi} of Assumption~\ref{ass:invariant_improved} and Part~\ref{def:local:partial_xi} of Definition~\ref{def:local}), and the third step follows from merging terms.

Therefore, 
\begin{align}\label{eq:gamma_tmp_gamma_1_partial_gamma}
    \gamma^{\tmp} - \gamma_1 - \partial \gamma
    = & ~ -(I - U^{\tmp}(C^{-1}+U^{\top}U^{\tmp})^{-1}U^{\top}) \cdot (\gamma_1 + \partial \gamma ) \notag \\
    = & ~ -(I - BU'(C^{-1}+U^{\top}BU')^{-1}U^{\top}) \cdot (\gamma_1 + \partial \gamma )  \notag \\
    = & ~ -(I - BU'(C^{-1}+U^{\top}BU')^{-1}U^{\top}) \cdot B\cdot L_r[(M_{S^{\new}})^{\top}]\cdot \sqrt{W^{\appr}}f(h^{\appr})  \notag \\
    = & ~ -\L_*[((\Delta^{\new}_{S^{\new}, S^{\new}}) + M_{S^{\new}, S^{\new}})^{-1}] \cdot L_r[(M_{S^{\new}})^{\top}]\cdot \sqrt{W^{\appr}}f(h^{\appr}),
\end{align}
where the first step is by assignment of $\gamma^{\tmp}$ on Line~\ref{alg:line:gamma_tmp} of Algorithm~\ref{alg:query_improved}, the second step is by $U^{\tmp}=BU'$(Corollary~\ref{cor:U_tmp_correctness}), the third step follows by Eq.~\eqref{eq:r_4_right}, the fourth step is by 
$B-BU(C^{-1}+U^{\top}BU)^{-1}U^{\top}B=\L_*[((\Delta^{\new}_{S^{\new}, S^{\new}}) + M_{S^{\new}, S^{\new}})^{-1}]$ (Lemma~\ref{lem:B_tmp_correctness}),

Then from the assignment of $r_4$ on Line~\ref{alg:line:r_4_improved} of Algorithm~\ref{alg:query_improved}, we have
\begin{align} \label{eq:r_4_improved}
    r_4 
    = &~ \Big(\L_c[(Q[l])_{S^{\new}}]+F[l]+R[l] \Gamma \cdot (\L_c[M_{\partial S\backslash S}]-\L_c[M_{S'}]) + R[l]\partial \Gamma\cdot \L_c[M_{S^{\new}}]\Big)\cdot (\gamma^{\tmp}-\gamma_1-\partial \gamma) \notag \\
    = &~ R[l]\sqrt{W^{\appr}}\L_c[M_{S^{\new}}]\cdot (\gamma^{\tmp}-\gamma_1-\partial \gamma) \notag \\
    = &~ -R[l]\sqrt{W^{\appr}}\L_c[M_{S^{\new}}]\cdot \L_*[((\Delta^{\new}_{S^{\new}, S^{\new}}) + M_{S^{\new}, S^{\new}})^{-1}] \cdot L_r[(M_{S^{\new}})^{\top}]\cdot \sqrt{W^{\appr}}f(h^{\appr}) \notag\\
    = &~ -R[l]\sqrt{W^{\appr}}M_{S^{\new}}\cdot ((\Delta^{\new}_{S^{\new}, S^{\new}}) + M_{S^{\new}, S^{\new}})^{-1} \cdot (M_{S^{\new}})^{\top}\cdot \sqrt{W^{\appr}}f(h^{\appr}),
\end{align}
where the second step follows from Eq.~\eqref{eq:r_4_left}, the third step follows from Eq.\eqref{eq:gamma_tmp_gamma_1_partial_gamma}, the fourth step follows from the property of $\L$ operators (Part~\ref{rem:L_operator:multiplication_2} of Remark~\ref{rem:L_operator}).
\end{proof}

\begin{claim}[Close-form formula for $r$] %[Part 5 of Lemma~\ref{lem:query_correct_improved}]
\label{cla:query_correct_improved}
\begin{align*}
r = R[l]^{\top}R[l]\sqrt{W^{\appr}} M^{\new}\sqrt{W^{\appr}}f(h^{\appr}),
\end{align*}
where $M^{\new} = A^{\top}(AW^{\appr}A^{\top})^{-1}A$.
\end{claim}

\begin{proof}

From Claim~\ref{cla:query_correct_improved_part_1}, we have $r_1 = Q[l] \sqrt{V} f(g)$.

From Claim~\ref{cla:query_correct_improved_part_2}, we have $r_2 = (Q[l] + R[l](\Gamma + \partial \Gamma) M) \cdot (\sqrt{W^{\appr}}f(h^{\appr}) -\sqrt{V}f(g) )$.

From Claim~\ref{cla:query_correct_improved_part_3}, we have $r_3 = R[l](\Gamma+\partial \Gamma) M\sqrt{V} f(g)$.

From Claim~\ref{cla:query_correct_improved_part_4}, we have
\begin{align*}
r_4 = -R[l]\sqrt{W^{\appr}} M_{S^{\new}} \cdot ( (\Delta_{S^{\new},S^{\new}}^{\new})^{-1} + M_{S^{\new},S^{\new}} )^{-1}\cdot (M_{S^{\new}})^{\top} \sqrt{W^{\appr}}f(h^{\appr}).
\end{align*}

The proof sketch is as follows: we first compute $r_1+r_3$, then compute $(r_1+r_3) + r_2$, and finally we compute $(r_1 + r_2 + r_3) + r_4$. First we compute $r_1 + r_3$ as follows:
\begin{align}\label{eq:add_r_1_r_3}
r_1 + r_3
= & ~ Q[l]\sqrt{V}f(g) + R[l](\Gamma+\partial \Gamma) M\sqrt{V} f(g)
=  R[l]\sqrt{V}M \sqrt{V}f(g) + R[l](\Gamma+\partial \Gamma) M\sqrt{V} f(g) \notag \\
= & ~  R[l] ( \sqrt{V} + (\Gamma+\partial \Gamma) ) M\sqrt{V} f(g) 
=  R[l]\sqrt{W^{\appr}}M\sqrt{V}f(g),
\end{align}
where the first step follows from Claim \ref{cla:query_correct_improved_part_1} and \ref{cla:query_correct_improved_part_3}, the second step follows from $Q = R \sqrt{V} M$ (Part~\ref{ass:invariant_improved:Q} of Assumption~\ref{ass:invariant_improved}), the third step follows from merging terms, and the fourth step follows from $\Gamma + \partial \Gamma =(\sqrt{\wt{V}}-\sqrt{V}) + (\sqrt{W^{\appr}}-\sqrt{\wt{V}})= \sqrt{W^{\appr}}-\sqrt{V}$, ($\Gamma$ from Part~\ref{ass:invariant_improved:Gamma} in Assumption~\ref{ass:invariant_improved}, $\partial \Gamma$ from Part~\ref{def:local:partial_Gamma} in Definition~\ref{def:local}).

Secondly, we can compute $(r_1 + r_3) + r_2$,{\footnotesize
\begin{align}\label{eq:add_r_1_r_3_r_2}
(r_1+r_3)+r_2
%= & ~  R[l]\sqrt{W^{\appr}}M\sqrt{V}f(g) + r_2 \notag \\
= & ~ R[l]\sqrt{W^{\appr}}M\sqrt{V}f(g)
 + (Q[l] + R[l](\Gamma + \partial \Gamma) M)(\sqrt{W^{\appr}}f(h^{\appr}) -\sqrt{V}f(g)) \notag \\
= & ~R[l]\sqrt{W^{\appr}}M\sqrt{V}f(g) + (R[l]\sqrt{V}M + R[l](\Gamma + \partial \Gamma) M)(\sqrt{W^{\appr}}f(h^{\appr}) -\sqrt{V}f(g)) \notag \\
= & ~ R[l]\sqrt{W^{\appr}}M\sqrt{V}f(g) + R[l]\sqrt{W^{\appr}} M(\sqrt{W^{\appr}}f(h^{\appr}) -\sqrt{V}f(g)) \notag \\
= & ~ R[l]\sqrt{W^{\appr}}M\sqrt{W^{\appr}}f(h^{\appr}),
\end{align}}
where the first step follows from Eq.~\eqref{eq:add_r_1_r_3} and Claim \ref{cla:query_correct_improved_part_2}, the second step follows from $Q = R \sqrt{V} M$ (Part~\ref{ass:invariant_improved:Q} of Assumption~\ref{ass:invariant_improved}), the third step follows from $\Gamma + \partial \Gamma =(\sqrt{\wt{V}}-\sqrt{V}) + (\sqrt{W^{\appr}}-\sqrt{\wt{V}})= \sqrt{W^{\appr}}-\sqrt{V}$ ($\Gamma$ from Part~\ref{ass:invariant_improved:Gamma} in Assumption~\ref{ass:invariant_improved}, $\partial \Gamma$ from Part~\ref{def:local:partial_Gamma} in Definition~\ref{def:local}), and the fourth step follows from merging terms.

Finally, we can compute $r_1+r_2+r_3+r_4$.{\footnotesize
\begin{align}\label{eq:add_r_1_r_2_r_3_r_4}
( r_1 + r_2 + r_3 ) + r_4 
%= & R[l]\sqrt{W^{\appr}}M\sqrt{W^{\appr}} f ( h^{\appr} ) + r_4 \\
= & ~ R[l]\sqrt{W^{\appr}}M\sqrt{W^{\appr}} f ( h^{\appr} ) - \notag \\
& ~ R[l]\sqrt{W^{\appr}} M_{S^{\new}} \left((\Delta_{S^{\new},S^{\new}}^{\new})^{-1} + M_{S^{\new},S^{\new}}\right)^{-1} (M_{S^{\new}})^{\top} \sqrt{W^{\appr}}f(h^{\appr}) \notag \\
= & ~ R[l]\sqrt{W^{\appr}}\left( M - M_{S^{\new}} \left((\Delta_{S^{\new},S^{\new}}^{\new})^{-1} + M_{S^{\new},S^{\new}}\right)^{-1} (M_{S^{\new}})^{\top}\right) \sqrt{W^{\appr}}f(h^{\appr}) \notag \\
= & ~ R[l]\sqrt{W^{\appr}} M^{\new}\sqrt{W^{\appr}}f(h^{\appr}) ,
\end{align}}
where the first step follows from Eq.~\eqref{eq:add_r_1_r_3_r_2} and Claim~\ref{cla:query_correct_improved_part_4}, the second step follows from merging terms, and the third step follows from Lemma~\ref{lem:M_tmp_correctness} (by setting the parameters in the lemma statement as $\Delta \leftarrow \Delta^{\new}$, $S \leftarrow S^{\new}$, $\wt{v} \leftarrow w^{\appr}$) and the definition that $M^{\new}=A^{\top}(AW^{\appr}A^{\top})^{-1}A$.

Therefore, from the assignment of $r$ on Line~\ref{alg:line:r_improved} of Algorithm~\ref{alg:query_improved}, we have 
\[
r=R[l]^{\top} (r_1+r_2+r_3+r_4) =R[l]^{\top} R[l]\sqrt{W^{\appr}} M^{\new}\sqrt{W^{\appr}}f(h^{\appr}).
\]
\end{proof}

\subsection{Correctness of \textsc{UpdateV} and \textsc{UpdateG}}\label{sec:update_v_g_correct_improved}
\begin{lemma}[Correctness of \textsc{UpdateV}]\label{lem:update_v_correct_improved}
After executing the procedure \textsc{UpdateV}, the following properties are satisfied:
\begin{enumerate}
\item $\|w^{\appr} - v\|_0\leq n^a$, $\|w^{\appr} - \wt{v}\|_0\leq n^{\wt{a}}$.
\item If all invariants of Assumption~\ref{ass:invariant_improved} are satisfied before entering the procedure \textsc{UpdateV} (Algorithm~\ref{alg:update_v_improved}), then all invariants are still satisfied after \textsc{UpdateV}.
\end{enumerate}
\end{lemma}

\begin{proof}
\noindent\textbf{Part 1.} 
%All discussion are around procedure~\textsc{UpdateV}(see Algorithm~\ref{alg:update_v_improved}).
The procedure \textsc{UpdateV} could exit in three places: Line~\ref{alg:line:update_v_return_by_default}, \ref{alg:line:return_by_partial_matrix_update} and \ref{alg:line:return_by_matrix_update}. We discuss them case by case.

\textbf{(a.} Line~\ref{alg:line:update_v_return_by_default}\textbf{)}. $w^{\appr}$ is assigned to be $\wt{v}^{\new}$. The algorithm avoids the if branches on Line~\ref{alg:line:if_for_matrix_update} and the if branch on Line~\ref{alg:line:if_for_partial_matrix_update}. Therefore, both of the conditions of the if branches are false, so we have $\|\wt{v}^{\new} - v\|_0< n^a$ and $\|\wt{v}^{\new} - \wt{v}\|_0< n^{\wt{a}}$.

\textbf{(b.} Line~\ref{alg:line:return_by_partial_matrix_update}\textbf{)}. $w^{\appr}$ is assigned to be $\wt{v}^{\new}$. The algorithm avoids the if branch on Line~\ref{alg:line:if_for_matrix_update}, so we have $\|\wt{v}^{\new} - v\|_0< n^a$. Then the algorithm enters procedure \textsc{PartialMatrixUpdate} (see Line~\ref{alg:line:enter_partial_matrix_update_improved}) to update $\wt{v} \leftarrow \wt{v}^{\new}$ (see Line~\ref{alg:line:partial_matrix_update_improved:everything_else} of Algorithm~\ref{alg:partial_matrix_update_improved}), so $\|\wt{v}^{\new} - \wt{v}\|_0 = 0 < n^{\wt{a}}$.

\textbf{(c.} Line~\ref{alg:line:return_by_matrix_update}\textbf{)}. $w^{\appr}$ is assigned to be $v^{\new}$. The algorithm enters the procedure \textsc{MatrixUpdate} (see Line~\ref{alg:line:enter_matrix_update_improved}) to update $v \leftarrow \wt{v} \leftarrow w^{\appr}$ (see Line~\ref{alg:line:matrix_update_improved:v_g} of Algorithm~\ref{alg:matrix_update_improved}). Thus $\|v^{\new} - v\|_0 = 0< n^a$ and $\|v^{\new} - \wt{v}\|_0 = 0< n^{\wt{a}}$.

\noindent \textbf{Part 2.} We first prove that $|S\cup \partial S|\leq 2n^a$ is satisfied if we enter \textsc{PartialMatrixUpdate}. Note that the input $w^{\appr}$ of \textsc{PartialMatrixUpdate} is $\wt{v}^{\new}$ (Line~\ref{alg:line:enter_partial_matrix_update_improved}), so $\partial S=\supp(w^{\appr}-\wt{v})=\supp(\wt{v}^{\new}-\wt{v})$ (Part~\ref{def:local:partial_S} of Definition~\ref{def:local}). We have
\begin{align*}
    |S\cup \partial S| = &~ |S| + |\partial S\backslash S|
    \leq  n^a + |\partial S\backslash S|
    =  n^a + |\{i
    \in[n] :v_i=\wt{v}_i, \wt{v}^{\new}_i\neq \wt{v}_i\}|\\
    = &~ n^a + |\{i\in[n]:\wt{v}^{\new}_i\neq v_i\}| ~\leq ~ 2n^a,
\end{align*}
where the second step follows from $|S|\leq n^a$ which is a direct implication of Part 1 of this lemma (see proof of Corollary~\ref{cor:improved:sparsity_of_members}), the third step follows from $S=\supp(v-\wt{v})$ (Part~\ref{ass:invariant_improved:S} of Assumption~\ref{ass:invariant_improved}) and $\partial S = \supp(\wt{v}^{\new}-\wt{v})$, and the fifth step follows from $|\supp(\wt{v}^{\new}-v)|<n^a$ since the if-clause of Line~\ref{alg:line:if_for_matrix_update}  of \textsc{UpdateV} (Algorithm~\ref{alg:update_v_improved}) has to be false to enter \textsc{PartialMatrixUpdate}.

In procedure \textsc{UpdateV}, the data structure members are modified only by procedure \textsc{MatrixUpdate} (on Line~\ref{alg:line:enter_matrix_update_improved}) and procedure \textsc{PartialMatrixUpdate} (on Line~\ref{alg:line:enter_partial_matrix_update_improved}). Since all invariants are satisfied before entering \textsc{MatrixUpdate} or \textsc{PartialMatrixUpdate}, and $|S\cup \partial S|\leq 2n^a$ is also satisfied before entering \textsc{PartialMatrixUpdate}, from Lemma~\ref{lem:matrix_update_correct_improved} and Lemma~\ref{lem:partial_matrix_update_correct_improved} we know that all the invariants are still satisfied after \textsc{MatrixUpdate} and \textsc{PartialMatrixUpdate}.
\end{proof}

\begin{lemma}[Correctness of \textsc{UpdateG}]\label{lem:update_g_correct_improved}
After executing the procedure \textsc{UpdateG}, the following properties are satisfied:
\begin{enumerate}
\item $\|h^{\appr}-g\|_0\leq n^a$, $\|h^{\appr}-\wt{g}\|_0\leq n^{\wt{a}}$.
\item If all invariants of Assumption~\ref{ass:invariant_improved} are satisfied before entering the procedure \textsc{UpdateG} (Algorithm~\ref{alg:update_g_improved}), then all invariants are still satisfied after \textsc{UpdateG}.
\end{enumerate}
\end{lemma}

\begin{proof}

\noindent\textbf{Part 1.} The proof is analogous to that of Part 1 of Lemma~\ref{lem:update_g_correct_improved}.

\noindent \textbf{Part 2.} In procedure \textsc{UpdateG}, the data structure members are modified only by procedure \textsc{VectorUpdate} (see Line~\ref{alg:line:enter_vector_update_improved}) and procedure \textsc{PartialVectorUpdate} (see Line~\ref{alg:line:enter_partial_vector_update_improved}). So this directly follows from the fact that all the invariants are satisfied after \textsc{VectorUpdate} (Lemma~\ref{lem:vector_update_correct_improved}) and \textsc{PartialVectorUpdate} (Lemma~\ref{lem:partial_vector_update_correct_improved}).
\end{proof}

By the same reasoning, we immediately have the following corollary:

\begin{corollary}[Sparsity of $\wt{v} - v$ and $\wt{g} - g$]
\label{cor:improved:sparsity_of_members}
Throughout the algorithm the following is always satisfied: 
\[\|\wt{v} - v\|_0 \leq  n^a,~ \|\wt{g} - g\|_0 \leq n^a.\]
\end{corollary}
%\begin{proof}
%When initializing, $\wt{v}=v$ and $\wt{g}=g$, so these two statements are satisfied.

%It is easy to follow the proof of Part 1 of Lemma~\ref{lem:update_v_correct_improved} and Lemma~\ref{lem:update_g_correct_improved} and check that after each procedure \textsc{MatrixUpdate}, \textsc{PartialMatrixUpdate}, \textsc{VectorUpdate}, and \textsc{PartialVectorUpdate}, whenever the data structure members $\wt{v}$, $v$, $\wt{g}$, and $g$ are modified, they still satisfy the two statements of this corollary.
%\end{proof}

\begin{corollary}[Sparsity guarantees when entering the procedures]
\label{cor:improved:sparsity_guarantee_procedures}
{\color{white}.}
Let $k$ and $\wt{k}$ be the output returned by \textsc{UpdateV} (Line~\ref{alg:line:update_v_in_update_query_improved} in Algorithm~\ref{alg:update_query_improved}).
Let $p$, $\wt{p}$ be the output returned by \textsc{UpdateG} (Line~\ref{alg:line:update_g_in_update_query_improved} in Algorithm~\ref{alg:update_query_improved}).
%Let $k$ be defined on Line~\ref{alg:line:binary_search_for_v_improved} of \textsc{UpdateV} (Algorithm~\ref{alg:update_v_improved}), let $\wt{k}$ be defined on Line~\ref{alg:line:binary_search_for_wt_v_improved} of \textsc{UpdateV} (Algorithm~\ref{alg:update_v_improved}), let $p$ be defined on Line~\ref{alg:line:binary_search_for_g_improved} of \textsc{UpdateG} (Algorithm~\ref{alg:update_g_improved}), and let $\wt{p}$ be defined on Line~\ref{alg:line:binary_search_for_wt_g_improved} of \textsc{UpdateG} (Algorithm~\ref{alg:update_g_improved}).

\begin{enumerate}
\item When entering the procedure \textsc{MatrixUpdate} (Algorithm~\ref{alg:matrix_update_improved}), we have 
\[
\|w^{\appr}-v\|_0 = k.
\]
\item When entering the procedure \textsc{PartialMatrixUpdate} (Algorithm~\ref{alg:partial_matrix_update_improved}), we have \[\|w^{\appr}-v\|_0 \leq  n^a, ~ \|w^{\appr}-\wt{v}\|_0 = \wt{k}\leq 2n^a.
\]
\item When entering the procedure \textsc{VectorUpdate} (Algorithm~\ref{alg:vector_update_improved}), we have 
\[
\|w^{\appr}-v\|_0 \leq n^a, ~ \|w^{\appr}-\wt{v}\|_0 \leq  n^{\wt{a}}, ~ \|h^{\appr}-g\|_0 = p.
\]
\item When entering the procedure \textsc{PartialVectorUpdate} (Algorithm~\ref{alg:partial_vector_update_improved}), we have \[
\|w^{\appr}-v\|_0 \leq n^a, ~ \|w^{\appr}-\wt{v}\|_0 \leq n^{\wt{a}}, ~ \|h^{\appr}-g\|_0 \leq n^a, ~ \|h^{\appr}-\wt{g}\|_0 = \wt{p}\leq 2n^a.
\]
\item When entering the procedure \textsc{Query} (Algorithm~\ref{alg:query_improved}), we have \[\|w^{\appr}-v\|_0\leq n^a,~\|w^{\appr}-\wt{v}\|_0\leq n^{\wt{a}},~\|h^{\appr}-g\|_0\leq n^a,~\|h^{\appr}-\wt{g}\|_0\leq n^{\wt{a}}.\]
\end{enumerate}
\end{corollary}
\begin{proof}
All line number mentioned in the proof is in \textsc{UpdateV} (Algorithm~\ref{alg:update_v_improved}).

\noindent \textbf{Part 1.} 
\textsc{MatrixUpdate} is entered in Line~\ref{alg:line:enter_matrix_update_improved}, and its input $w^{\appr}$ is $v^{\new}$, which is defined on Line~\ref{alg:line:binary_search_for_v_improved}. So $\|w^{\appr}-v\|_0=\|v^{\new}-v\|_0=k$ directly follows from the definition of $k$.

\noindent \textbf{Part 2.} 
\textsc{PartialMatrixUpdate} is entered in Line~\ref{alg:line:enter_partial_matrix_update_improved}, and its input $w^{\appr}$ is $\wt{v}^{\new}$, which is defined on Line~\ref{alg:line:binary_search_for_wt_v_improved}. $\|w^{\appr}-v\|_0\leq n^a$ is because the algorithm bypass the if-branch in Line~\ref{alg:line:if_for_matrix_update}. And $\|w^{\appr}-\wt{v}\|_0=\|\wt{v}^{\new}-v\|_0=\wt{k}$ directly follows from the definition of $\wt{k}$. And $\|w^{\appr}-\wt{v}\|_0\leq \|w^{\appr}-v\|_0 + \|v-\wt{v}\|_0\leq 2n^a$ follows from triangle inequality and Corollary~\ref{cor:improved:sparsity_of_members}.

\noindent \textbf{Part 3,4.}
The guarantee that $\|w^{\appr}-v\|_0 \leq n^a, ~ \|w^{\appr}-\wt{v}\|_0 \leq  n^{\wt{a}}$ is from Part 1 of Lemma~\ref{lem:update_v_correct_improved}.
The remaining proof is the same as Part 1 and Part 2.

\noindent \textbf{Part 5.} This directly follows from Part 1 of Lemma~\ref{lem:update_v_correct_improved} and Part 1 of Lemma~\ref{lem:update_g_correct_improved}.
\end{proof}

\subsection{Correctness of \textsc{MatrixUpdate}}\label{sec:matrix_update_correct_improved}

\begin{lemma}[Correctness of \textsc{MatrixUpdate}]\label{lem:matrix_update_correct_improved}
If all invariants of Assumption~\ref{ass:invariant_improved} are satisfied before entering the procedure \textsc{MatrixUpdate} (Algorithm~\ref{alg:matrix_update_improved}), then after the procedure \textsc{MatrixUpdate} we have the following guarantees:
\begin{enumerate}
\item $v = \wt{v} = w^{\appr}$.
\item $g$ and $\wt{g}$ both remain the same.
\item All invariants of Assumption~\ref{ass:invariant_improved} are still satisfied.
\end{enumerate}
\end{lemma}

Part 1 and 2 are proved in Claim~\ref{cla:matrix_update_correct_improved:part_1_part_2}, and Part 3 is proved in Claim~\ref{cla:matrix_update_correct_improved:part_3} by using Claim~\ref{cla:matrix_update_correct_improved:M_Q_beta_1_beta_2}.

\begin{claim}[Part 1 and 2 of Lemma~\ref{lem:matrix_update_correct_improved}]
\label{cla:matrix_update_correct_improved:part_1_part_2}
After the procedure \textsc{MatrixUpdate} (Algorithm~\ref{alg:matrix_update_improved}), we have $v = \wt{v} = w^{\appr}$, and $g$ and $\wt{g}$ remain the same.

\end{claim}
\begin{proof}
This follows directly from the value assignment of $v$ and $\wt{v}$ on Line~\ref{alg:line:matrix_update_improved:v_g} of Algorithm~\ref{alg:matrix_update_improved}, and the fact that $g$ and $\wt{g}$ are not modified by Algorithm~\ref{alg:matrix_update_improved}.
\end{proof}

\begin{claim}
\label{cla:matrix_update_correct_improved:M_Q_beta_1_beta_2}
In the procedure \textsc{MatrixUpdate} (Algorithm~\ref{alg:matrix_update_improved}), before we refresh variables in the memory of data structure, we have the following:
\begin{enumerate}
\item $M^{\tmp} = A^{\top} (A W^{\appr} A^{\top})^{-1}A$,
\item $Q^{\tmp} = R\sqrt{W^{\appr}}M^{\tmp}$,
\item $\beta_1^{\tmp} = Q^{\tmp}\sqrt{W^{\appr}}f(g)$,
\item $\beta_2^{\tmp} = M^{\new}\sqrt{W^{\appr}}f(g)$,
\item $\xi^{\tmp} = \sqrt{W^{\appr}}f(\wt{g}) - \sqrt{W^{\appr}}f(g)$.
\end{enumerate}
\end{claim}
\begin{proof}
\noindent \textbf{Part 1.}
On Line~\ref{alg:line:matrix_update_improved:M^new} of Algorithm~\ref{alg:matrix_update_improved} we assigned $M^{\tmp}$ as
\[
M^{\tmp} = M - M_{S^{\new}} \cdot ( (\Delta^{\new}_{S^{\new},S^{\new}})^{-1} + M_{S^{\new},S^{\new}} )^{-1} \cdot ( M_{S^{\new}} )^\top.
\]
Since $M=A^{\top}(AVA^{\top})^{-1}A$ (Part~\ref{ass:invariant_improved:M} of Assumption~\ref{ass:invariant_improved}), $\Delta^{\new} = W^{\appr} - V$ (Part~\ref{ass:invariant_improved:Delta} of Assumption~\ref{ass:invariant_improved}), and $S = \supp(w^{\appr}-v)$ (Part~\ref{ass:invariant_improved:S} of Assumption~\ref{ass:invariant_improved}), from Lemma~\ref{lem:M_tmp_correctness} we have
\[
M^{\tmp} = A^{\top} (A W^{\appr} A^{\top})^{-1}A.
\]

\noindent \textbf{Part 2.} 
We have
\begin{align*}
Q^{\tmp} = &~ Q + R(\Gamma^{\new} M^{\tmp}) + R\sqrt{V}(M^{\tmp} -M) 
=  R\sqrt{V}M + R(\Gamma^{\new} M^{\tmp}) + R\sqrt{V}(M^{\tmp} -M) \\
= &~ R(\Gamma^{\new} + \sqrt{V})M^{\tmp} 
=  R(\Gamma + \partial \Gamma + \sqrt{V})M^{\tmp} \\
= &~ R((\sqrt{\wt{V}} - \sqrt{V}) + (\sqrt{W^{\appr}} - \sqrt{\wt{V}}) + \sqrt{V})M^{\tmp} 
=  R \sqrt{W^{\appr}}M^{\tmp},
\end{align*}
where the first step follows from the assigned value for $Q^{\tmp}$ on Line~\ref{alg:line:matrix_update_improved:Q^new} of Algorithm~\ref{alg:matrix_update_improved}, the second step follows from $Q = R\sqrt{V}M$ (Part~\ref{ass:invariant_improved:Q} of Assumption~\ref{ass:invariant_improved}), the third step is by merging terms, the fourth step follows from $\Gamma^{\new}=\Gamma + \partial \Gamma$ (Part~\ref{def:local:Gamma_new} of Definition~\ref{def:local}), the fifth step follows from $\Gamma =\sqrt{\wt{V}} - \sqrt{V}$ (Part~\ref{ass:invariant_improved:Gamma} of Assumption~\ref{ass:invariant_improved}) and $\partial \Gamma =\sqrt{W^{\appr}} - \sqrt{\wt{V}}$ (Part~\ref{def:local:partial_Gamma} of Definition~\ref{def:local}), and the last step is by merging terms.

\noindent \textbf{Part 3, 4 and 5.}
These directly follow from the assignment of $\beta_1^{\tmp}$ on Line~\ref{alg:line:matrix_update_improved:beta_1} of Algorithm~\ref{alg:matrix_update_improved}, the assignment of $\beta_2^{\tmp}$ on Line~\ref{alg:line:matrix_update_improved:beta_2}, and the assignment of $\xi^{\tmp}$ on Line~\ref{alg:line:matrix_update_improved:xi}.
\end{proof}

\begin{claim}[Part 3 of Lemma~\ref{lem:matrix_update_correct_improved}] \label{cla:matrix_update_correct_improved:part_3}
All invariants of Assumption~\ref{ass:invariant_improved} are satisfied after the procedure \textsc{MatrixUpdate} (Algorithm~\ref{alg:matrix_update_improved}).
\end{claim}
\begin{proof}
First note that $g$, $\wt{g}$, and $T$ all remain the same after \textsc{MatrixUpdate}. Note that $v$ and $\wt{v}$ are both assigned the value $w^{\appr}$ (Line~\ref{alg:line:matrix_update_improved:v_g} of Algorithm~\ref{alg:matrix_update_improved}). Then using Claim~\ref{cla:matrix_update_correct_improved:M_Q_beta_1_beta_2}, and since we assigned $M^{\tmp}$ to $M$, $Q^{\tmp}$ to $Q$, $\beta_1^{\tmp}$ to $\beta_1$, $\beta_2^{\tmp}$ to $\beta_2$, and $\xi^{\tmp}$ to $\xi$, we have
\begin{align*}
& M =  A^{\top} (A V A^{\top})^{-1}A,
&& Q =  R\sqrt{V}M,\\
& \beta_1 =  Q\sqrt{V}f(g),
&& \beta_2 =  M\sqrt{V}f(g),\\
&\xi = \sqrt{\wt{V}}f(\wt{g}) - \sqrt{V}f(g).
\end{align*}
We prove the other invariants directly from the assignment on Line~\ref{alg:line:matrix_update_improved:everything_else} of Algorithm~\ref{alg:matrix_update_improved}:
\begin{align*}
& S = \emptyset = \supp(\wt{v}-v),
&& \Delta = 0 = \wt{V} - V,\\
& \Gamma = 0 = \sqrt{\wt{V}} - \sqrt{V},
&& B = I = \L_*[(\Delta_{S,S}^{-1} + M_{S,S})^{-1}],\\
& \gamma_1 = 0 = B\cdot \L_r [\beta_{2,S}]+B \cdot \L_r[(M_S)^{\top}] \cdot \xi,
&& \gamma_2 = 0 = \Gamma M \cdot \xi,\\
& E = 0 = B\cdot \L_r[(M_{\emptyset})^{\top}] = B\cdot \L_r[(M_S)^{\top}],
&& F = 0 = R\Gamma \cdot \L_c[M_{\emptyset}] = R\Gamma \cdot \L_c[M_{S}]. 
\end{align*}
\end{proof}

\subsection{Correctness of \textsc{PartialMatrixUpdate}}\label{sec:partial_matrix_update_correct_improved}

\begin{lemma}[Correctness of \textsc{PartialMatrixUpdate}]\label{lem:partial_matrix_update_correct_improved}
If all invariants of Assumption~\ref{ass:invariant_improved} are satisfied before entering the procedure \textsc{PartialMatrixUpdate} (Algorithm~\ref{alg:partial_matrix_update_improved}), and $|S\cup \partial S|\leq 2n^a$, then after the procedure \textsc{PartialMatrixUpdate} we have the following guarantees:
\begin{enumerate}
\item $\wt{v} = w^{\appr}$, and $v$ remains the same.
\item $g$, $\wt{g}$ both remain the same.
\item All invariants of Assumption~\ref{ass:invariant_improved} are still satisfied.
\end{enumerate}
\end{lemma}

Note that since we have the guarantee that $|S\cup \partial S|\leq 2n^a$, all $\L_r,\L_c,\L_*$ operators that appear in the procedure \textsc{PartialMatrixUpdate} are well-defined. And the \textsc{Decompose} function used in \textsc{PartialMatrixUpdate} is also well-defined.

Part 1 and 2 are proved in Claim~\ref{cla:partial_matrix_update_correct_improved:part_1_part_2}, and Part 3 is proved in Claim~\ref{cla:partial_matrix_update_correct_improved:part_3} by using Claim~\ref{cla:partial_matrix_update_correct_improved:B_xi_gamma_1_gamma_2}.

\begin{claim}[Part 1 and 2 of Lemma~\ref{lem:partial_matrix_update_correct_improved}]
\label{cla:partial_matrix_update_correct_improved:part_1_part_2}
After the procedure \textsc{PartialMatrixUpdate} (Algorithm~\ref{alg:partial_matrix_update_improved}), we have $\wt{v}=w^{\appr}$, and $v$, $g$, $\wt{g}$ all remain the same.
\end{claim}
\begin{proof}
$\wt{v}=w^{\appr}$ follows directly from the value assignment of $\wt{v}$ on Line~\ref{alg:line:partial_matrix_update_improved:everything_else} of Algorithm~\ref{alg:partial_matrix_update_improved}.

The procedure \textsc{PartialMatrixUpdate} (Algorithm~\ref{alg:partial_matrix_update_improved}) does not modify $v$, $g$, $\wt{g}$, so they all remain the same.
\end{proof}

\begin{claim}
\label{cla:partial_matrix_update_correct_improved:B_xi_gamma_1_gamma_2}
In the procedure \textsc{PartialMatrixUpdate} (Algorithm~\ref{alg:partial_matrix_update_improved}) before we refresh variables in the memory of the data structure, we have the following:
%\begin{multicols}{1}
\begin{enumerate}
\item $B^{\tmp} = \L_*[ ((\Delta^{\new}_{S^{\new},S^{\new}})^{-1} + M_{S^{\new},S^{\new}})^{-1} ]$,
\item $\xi^{\tmp} = \sqrt{W^{\appr}}f(\wt{g}) - \sqrt{V}f(g)$,
\item $\gamma_1^{\tmp} = B^{\tmp}\cdot \L_r [\beta_{2,S^{\new}}]+B^{\tmp} \cdot \L_r[(M_{S^{\new}})^{\top}] \cdot \xi^{\tmp}$,
\item $\gamma_2^{\tmp} = \Gamma^{\new} M \cdot \xi^{\tmp}$,
\item $F^{\tmp} = R\Gamma^{\new} \cdot \L_c[M_{S^{\new}}]$,
\item $E^{\tmp} = B^{\tmp} \cdot \L_r[(M_{S^{\new}})^{\top}]$.
\end{enumerate}
%\end{multicols}
\end{claim}
\begin{proof}
\noindent \textbf{Part 1.} On Line~\ref{alg:line:partial_matrix_update_improved:B} of Algorithm~\ref{alg:partial_matrix_update_improved}, we assigned $B^{\tmp}$ as
\[
B^{\tmp} \leftarrow B - BU'(C^{-1}+U^{\top}BU')^{-1}U^{\top}B,
\]
where $(U', C, U) = \textsc{Decompose}( \L_*[(\Delta^{\new}_{S^{\new}, S^{\new}})^{-1} +  M_{S^{\new}, S^{\new}} ] - \L_* [\Delta_{S, S}^{-1}+ M_{S,S}] )$,
then using Lemma~\ref{lem:B_tmp_correctness} we have that $B^{\tmp} = \L_*[ ((\Delta^{\new}_{S^{\new},S^{\new}})^{-1} + M_{S^{\new},S^{\new}})^{-1} ]$.

\noindent \textbf{Part 2 and 3.} 
Part 2 directly follows from the assignment of $\xi^{\tmp}$ on Line~\ref{alg:line:partial_matrix_update_improved:xi} of Algorithm~\ref{alg:partial_matrix_update_improved}.

And Part 3 directly follows from the assignment of $\gamma_1^{\tmp}$ on Line~\ref{alg:line:partial_matrix_update_improved:gamma_1} of Algorithm~\ref{alg:partial_matrix_update_improved}.

\noindent \textbf{Part 4.} We have
\begin{align*}
\gamma_2^{\tmp} = &~ \gamma_2 + (\Gamma+\partial \Gamma) M (\sqrt{W^{\appr}} - \sqrt{\wt{V}})f(\wt{g}) + \partial \Gamma M (\sqrt{\wt{V}}f(\wt{g}) - \sqrt{V}f(g))\\
= &~ \Gamma M \cdot (\sqrt{\wt{V}}f(\wt{g}) - \sqrt{V}f(g)) + (\Gamma+\partial \Gamma) M (\sqrt{W^{\appr}} - \sqrt{\wt{V}})f(\wt{g}) + \partial \Gamma M (\sqrt{\wt{V}}f(\wt{g}) - \sqrt{V}f(g))\\
= &~ \Gamma M \cdot (\sqrt{W^{\appr}}f(\wt{g}) - \sqrt{V}f(g)) + \partial \Gamma M (\sqrt{W^{\appr}}f(\wt{g}) - \sqrt{V}f(g))\\
= &~(\Gamma+\partial \Gamma) M \cdot (\sqrt{W^{\appr}}f(\wt{g}) - \sqrt{V}f(g))\\
= &~\Gamma^{\new} M \cdot (\sqrt{W^{\appr}}f(\wt{g}) - \sqrt{V}f(g))\\
= &~ \Gamma^{\new} M \cdot \xi^{\tmp},
\end{align*}
where the first step follows from the assignment of $\gamma_2^{\new}$ (Line~\ref{alg:line:partial_matrix_update_improved:gamma_2} of Algorithm~\ref{alg:partial_matrix_update_improved}), the second step follows from $\gamma_2=\Gamma M \xi=\Gamma M (\sqrt{\wt{V}}f(\wt{g}) - \sqrt{V}f(g))$ (Part~\ref{ass:invariant_improved:gamma_2} and \ref{ass:invariant_improved:xi} of Assumption~\ref{ass:invariant_improved}), the third and the fourth step both follow from merging terms, the fifth step follows from $\Gamma^{\new} = \Gamma+\partial \Gamma$ (Part~\ref{def:local:partial_Gamma} of Definition~\ref{def:local}), and the sixth step follows from the Part 2 of this claim.

\noindent {\bf Part 5.}
\begin{align*}
    F^{\tmp} 
    = & ~ F +R \Gamma \cdot (\L_c[M_{\partial S\backslash S}]-\L_c[M_{S'}]) + R\partial \Gamma\cdot \L_c[M_{S^{\new}}]\\
    = &~ R\Gamma\cdot \L_c[M_S]+R \Gamma \cdot (\L_c[M_{\partial S\backslash S}]-\L_c[M_{S'}]) + R\partial \Gamma\cdot \L_c[M_{S^{\new}}]\\
    = &~ R\Gamma\cdot \L_c[M_{S^{\new}}] + R\partial \Gamma\cdot \L_c[M_{S^{\new}}]\\
    = &~ R\Gamma^{\new}\cdot \L_c[M_{S^{\new}}],
\end{align*}
where the first step follows from the assignment of $F^{\tmp}$(Line~\ref{alg:line:partial_matrix_update_improved:F} of Algorithm~\ref{alg:partial_matrix_update_improved}), the second step follows from $F=R\Gamma\cdot \L_c[M_S]$ (Part~\ref{ass:invariant_improved:F} of Assumption~\ref{ass:invariant_improved}), the third step follows from $S'=(S\cup \partial S)\backslash S^{\new}$ (Part~\ref{def:local:S'} of Definition~\ref{def:local}) and thus $\L_c[M_{S'}]+\L_c[M_{S^{\new}}]=\L_c[M_{S\cup \partial S}]=\L_c[M_S]+\L_c[M_{\partial S\backslash S}]$ by Part~\ref{rem:L_operator:addition} of Remark~\ref{rem:L_operator}, and the fourth step follows from $\Gamma^{\new}=\Gamma+\partial \Gamma$ (Part~\ref{def:local:Gamma_new} of Definition~\ref{def:local}).

\noindent {\bf Part 6}
\begin{align*}
    E^{\tmp} 
    = &~ E + B^{\tmp}(\L_r[ (M_{\partial S\backslash S})^{\top}]-\L_r[ (M_{S'})^{\top}]) - BU'(C^{-1}+U^{\top}BU')^{-1}U^{\top}E\\
    = &~ B\cdot \L_r[(M_S)^{\top}] + B^{\tmp}(\L_r[ (M_{\partial S\backslash S})^{\top}]-\L_r[ (M_{S'})^{\top}]) - BU'(C^{-1}+U^{\top}BU')^{-1}U^{\top}B\cdot \L_r[(M_S)^{\top}]\\
    = & ~ (B-BU'(C^{-1}+U^{\top}BU')^{-1}U^{\top}B)\cdot  \L_r[(M_S)^{\top}] + B^{\tmp}(\L_r[ (M_{\partial S\backslash S})^{\top}]-\L_r[ (M_{S'})^{\top}])\\
    = & ~ B^{\tmp}\L_r[(M_S)^{\top}]+ B^{\tmp}(\L_r[ (M_{\partial S\backslash S})^{\top}]-\L_r[ (M_{S'})^{\top}])\\
    = & ~ B^{\tmp}\L_r[ (M_{S^{\new}})^{\top}]
\end{align*}
where the first step follows from the assignment of $E^{\tmp}$(Line~\ref{alg:line:partial_matrix_update_improved:E} of Algorithm~\ref{alg:partial_matrix_update_improved}), the second step follows from $E=B\cdot \L_r[(M_S)^{\top}]$ (Part~\ref{ass:invariant_improved:E} of Assumption~\ref{ass:invariant_improved}), the fourth step follows from the definition of $B^{\tmp}$, and the fifth step follows from $S'=(S\cup \partial S)\backslash S^{\new}$ (Part~\ref{def:local:S'} of Definition~\ref{def:local}) and thus $\L_r[(M_{S'})^{\top}]+\L_r[(M_{S^{\new}})^{\top}]=\L_r[(M_{S\cup \partial S})^{\top}]=\L_r[(M_S)^{\top}]+\L_r[(M_{\partial S\backslash S})^{\top}]$ by Part~\ref{rem:L_operator:addition} of Remark~\ref{rem:L_operator}, and the fourth step follows from $\Gamma^{\new}=\Gamma+\partial \Gamma$ (Part~\ref{def:local:Gamma_new} of Definition~\ref{def:local}).
\end{proof}

\begin{claim}[Part 3 of Lemma~\ref{lem:partial_matrix_update_correct_improved}]
\label{cla:partial_matrix_update_correct_improved:part_3}
All invariants of Assumption~\ref{ass:invariant_improved} are satisfied after the procedure \textsc{PartialMatrixUpdate} (Algorithm~\ref{alg:partial_matrix_update_improved}).
\end{claim}

\begin{proof}
First note that the value of $v$, $g$, $\wt{g}$, $T$, $M$, $Q$, $\beta_1$ and $\beta_2$ all remain the same after the procedure \textsc{PartialMatrixUpdate} (Algorithm~\ref{alg:partial_matrix_update_improved}). Also note that $\wt{v}$ is assigned the value $w^{\appr}$ (Line~\ref{alg:line:partial_matrix_update_improved:everything_else} in Algorithm~\ref{alg:partial_matrix_update_improved}).

From the assignment of $S$, $\Delta$, and $\Gamma$ on Line~\ref{alg:line:partial_matrix_update_improved:everything_else} of Algorithm~\ref{alg:partial_matrix_update_improved}, we have
\begin{align*}
S = &~ S^{\new} = \supp(w^{\appr} - v) = \supp(\wt{v} - v) , \\
\Delta = &~ \Delta^{\new} = W^{\appr} - V = \wt{V} - V , \\
\Gamma = &~ \Gamma^{\new} = \sqrt{W^{\appr}} - \sqrt{V} = \sqrt{\wt{V}} - \sqrt{V}.
\end{align*}

Then using Claim~\ref{cla:partial_matrix_update_correct_improved:B_xi_gamma_1_gamma_2}, and since we assigned $B^{\tmp}$ to $B$, $\xi^{\tmp}$ to $\xi$, $\gamma_1^{\new}$ to $\gamma_1$, and $\gamma_2^{\new}$ to $\gamma_2$, $E^{\tmp}$ to $E$, $F^{\tmp}$ to $F$, we have
\begin{align*}
B = &~ \L_*[ ((\Delta_{S,S})^{-1} + M_{S,S})^{-1} ], 
&\xi = &~ \sqrt{\wt{V}} f(\wt{g}) - \sqrt{V} f(g) , \\
\gamma_1 = &~ B \cdot \L_r [\beta_{2,S}]+B \cdot \L_r[(M_{S})^{\top}] \cdot \xi, 
&\gamma_2 = &~ \Gamma M \cdot \xi,\\
E = &~ B\cdot \L_r[(M_S)^{\top}],
&F = &~ R\Gamma \cdot \L_c[M_S].
\end{align*}
\end{proof}

\subsection{Correctness of \textsc{VectorUpdate}}\label{sec:vector_update_correct_improved}

\begin{lemma}[Correctness of \textsc{VectorUpdate}]\label{lem:vector_update_correct_improved}
If all invariants of Assumption~\ref{ass:invariant_improved} are satisfied before entering the procedure \textsc{VectorUpdate} (Algorithm~\ref{alg:vector_update_improved}), then after the procedure \textsc{VectorUpdate} we have the following guarantees:
\begin{enumerate}
\item $v$, $\wt{v}$ both remain the same.
\item $g = \wt{g} = h^{\appr}$.
\item All invariants of Assumption~\ref{ass:invariant_improved} are still satisfied.
\end{enumerate}
\end{lemma}
First note that from Corollary~\ref{cor:improved:sparsity_of_members} we have that $|S|\leq n^a$, so all $\L_r$ operators that appear in the procedure \textsc{VectorUpdate} are well-defined.

Part 1 and 2 are proved in Claim~\ref{cla:vector_update_correct_improved:part_1_part_2}, and Part 3 is proved in Claim~\ref{cla:vector_update_correct_improved:part_3} by using Claim~\ref{cla:vector_update_correct_improved:beta_1_beta_2_xi_gamma_1_gamma_2}.

\begin{claim} [Part 1 and 2 of Lemma~\ref{lem:vector_update_correct_improved}]
\label{cla:vector_update_correct_improved:part_1_part_2}
After the procedure \textsc{VectorUpdate} (Algorithm~\ref{alg:vector_update_improved}), $v$ and $\wt{v}$ both remain the same, and $g = \wt{g} = h^{\appr}$.
\end{claim}
\begin{proof}
The procedure \textsc{VectorUpdate} does not modify $v$ or $\wt{v}$, so they both remain the same. And $g = \wt{g} = h^{\appr}$ follows directly from the value assignment of $g$ and $\wt{g}$ on Line~\ref{alg:line:vector_update_improved:g} of Algorithm~\ref{alg:vector_update_improved}.
\end{proof}

\begin{claim}
\label{cla:vector_update_correct_improved:beta_1_beta_2_xi_gamma_1_gamma_2}
In the procedure \textsc{VectorUpdate} (Algorithm~\ref{alg:vector_update_improved}) before we refresh variables in the memory of the data structure, we have the following:
\begin{multicols}{2}
\begin{enumerate}
\item $\beta_1^{\tmp} = Q\sqrt{V}f(h^{\appr})$,
\item $\beta_2^{\tmp} = M\sqrt{V}f(h^{\appr})$,
\item $\xi^{\tmp} = (\sqrt{\wt{V}}-\sqrt{V})f(h^{\appr})$,
\item $\gamma_1^{\tmp} = B \cdot \L_r [\beta_{2,S}^{\tmp}] + B \cdot \L_r[(M_S)^{\top}] \cdot\xi^{\tmp}$,
\item $\gamma_2^{\tmp} = \Gamma M \cdot \xi^{\tmp}$.
\end{enumerate}
\end{multicols}
\end{claim}
\begin{proof}
\noindent \textbf{Part 1.} From the assignment of $\beta_1^{\tmp}$ on Line~\ref{alg:line:vector_update_improved:beta_1} of Algorithm~\ref{alg:vector_update_improved}, we have
\begin{align*}
\beta_1^{\tmp} 
=  \beta_1 + Q\sqrt{V}(f(h^{\appr}) - f(g)) 
=  Q \sqrt{V} f(g) + Q\sqrt{V}(f(h^{\appr}) - f(g)) 
=  Q\sqrt{V} f(h^{\appr}),
\end{align*}
where the second step follows from $\beta_1 = Q\sqrt{V}f(g)$ (Part~\ref{ass:invariant_improved:beta_1} of Assumption~\ref{ass:invariant_improved}).

\noindent \textbf{Part 2.} From the assignment of $\beta_2^{\tmp}$ on Line~\ref{alg:line:vector_update_improved:beta_2} of Algorithm~\ref{alg:vector_update_improved}, we have
\begin{align*}
\beta_2^{\tmp} 
=  \beta_2 + M\sqrt{V}(f(h^{\appr}) - f(g)) 
=  M \sqrt{V} f(g) + M\sqrt{V}(f(h^{\appr}) - f(g)) 
=  M\sqrt{V} f(h^{\appr}),
\end{align*}
where the second step follows from $\beta_2 = M\sqrt{V}f(g)$ (Part~\ref{ass:invariant_improved:beta_2} of Assumption~\ref{ass:invariant_improved}).

\noindent \textbf{Part 3, 4 and 5}. These directly follow from the assignment of $\xi^{\tmp}$ (Line~\ref{alg:line:vector_update_improved:xi}), $\gamma_1^{\tmp}$ (Line~\ref{alg:line:vector_update_improved:gamma_1}), and $\gamma_2^{\tmp}$ (Line~\ref{alg:line:vector_update_improved:gamma_2}) in Algorithm~\ref{alg:vector_update_improved}.
\end{proof}

\begin{claim}[Part 3 of Lemma~\ref{lem:vector_update_correct_improved}]
\label{cla:vector_update_correct_improved:part_3}
All invariants of Assumption~\ref{ass:invariant_improved} are satisfied after the procedure \textsc{VectorUpdate} (Algorithm~\ref{alg:vector_update_improved}).
\end{claim}
\begin{proof}
First note that $v$, $\wt{v}$, $M$, $Q$, $B$, $\Delta$, $\Gamma$, and $S$ all remain the same after the procedure \textsc{VectorUpdate}. Also note that $g$ and $\wt{g}$ are both assigned the value $h^{\appr}$ (Line~\ref{alg:line:vector_update_improved:g} of Algorithm~\ref{alg:vector_update_improved}).

Then using Claim~\ref{cla:vector_update_correct_improved:beta_1_beta_2_xi_gamma_1_gamma_2}, and since we assigned $\beta_1^{\tmp}$ to $\beta_1$, $\beta_2^{\tmp}$ to $\beta_2$, $\xi^{\tmp}$ to $\xi$, $\gamma_1^{\tmp}$ to $\gamma_1$, and $\gamma_2^{\tmp}$ to $\gamma_2$, we have
\begin{align*}
\beta_1 = &~ Q \sqrt{V} f(g),
& \beta_2 = &~ M \sqrt{V} f(g), \\
\gamma_1 = &~ B \cdot \L_r [\beta_{2,S}] + B \cdot \L_r[(M_S)^{\top}] \cdot \xi,
& \gamma_2 = &~ \Gamma M \cdot \xi, \\
\xi = &~ (\sqrt{\wt{V}}-\sqrt{V})f(h^{\appr}) = \sqrt{\wt{V}}f(\wt{g})-\sqrt{V}f(g).
\end{align*}

Also, from the assignment of $T$ on Line~\ref{alg:line:vector_update_improved:T} of Algorithm~\ref{alg:vector_update_improved}, we have $T = \emptyset = \supp(\wt{g} - g)$.
\end{proof}

\subsection{Correctness of \textsc{PartialVectorUpdate}}\label{sec:partial_vector_update_correct_improved}

\begin{lemma}[Correctness of \textsc{PartialVectorUpdate}]\label{lem:partial_vector_update_correct_improved}
If all invariants of Assumption~\ref{ass:invariant_improved} are satisfied before entering the procedure \textsc{PartialVectorUpdate} (Algoritm \ref{alg:partial_vector_update_improved}), then after the procedure \textsc{PartialVectorUpdate} we have the following guarantees:
\begin{enumerate}
\item $v$, $\wt{v}$ both remain the same.
\item $\wt{g} = h^{\appr}$, and $g$ remains the same.
\item All invariants of Assumption~\ref{ass:invariant_improved} are still satisfied.
\end{enumerate}
\end{lemma}
First note that from Corollary~\ref{cor:improved:sparsity_of_members} we have that $|S|\leq n^a$, so all $\L_r$ operators that appear in the procedure \textsc{PartialVectorUpdate} are well-defined.

Part 1 and 2 are proved in Claim~\ref{cla:partial_vector_update_correct_improved:part_1_part_2}, and Part 3 is proved in Claim~\ref{cla:partial_vector_update_correct_improved:part_3} by using Claim~\ref{cla:partial_vector_update_correct_improved:xi_gamma_1_gamma_2}.

\begin{claim}[Part 1 and 2 of Lemma~\ref{lem:partial_vector_update_correct_improved}]
\label{cla:partial_vector_update_correct_improved:part_1_part_2}
After the procedure \textsc{PartialVectorUpdate} (Algorithm~\ref{alg:partial_vector_update_improved}), we have $\wt{g} = h^{\appr}$, and
$v$, $\wt{v}$, $g$ all remain the same.
\end{claim}
\begin{proof} 
$\wt{g} = h^{\appr}$ follows directly from the value assignment of $\wt{g}$ on Line~\ref{alg:line:partial_vector_update_improved:wt_g} of Algorithm~\ref{alg:partial_vector_update_improved}.

The procedure \textsc{PartialVectorUpdate} (Algorithm~\ref{alg:partial_vector_update_improved}) does not modify $v$, $\wt{v}$, $g$, so they all remain the same.
\end{proof}

\begin{claim}
\label{cla:partial_vector_update_correct_improved:xi_gamma_1_gamma_2} In the procedure \textsc{PartialVectorUpdate} (Algorithm~\ref{alg:partial_vector_update_improved}) before we refresh variables in the memory of the data structure, we have the following:
\begin{enumerate}
\item $\xi^{\tmp} = \sqrt{\wt{V}}f(h^{\appr})-\sqrt{V}f(g)$,
\item $\gamma_1^{\tmp} = B\cdot \L_r [\beta_{2,S}]+B \cdot \L_r[(M_S)^{\top}] \cdot \xi^{\tmp}$,
\item $\gamma_2^{\tmp} =  \Gamma M \cdot \xi^{\tmp}$.
\end{enumerate}
\end{claim}
\begin{proof}
\noindent \textbf{Part 1.} This directly follows from the assignment of $\xi^{\tmp}$ on Line~\ref{alg:line:partial_vector_update_improved:xi} of Algorithm~\ref{alg:partial_vector_update_improved}.

\noindent \textbf{Part 2.} From the assignment of $\gamma_1^{\tmp}$ on Line~\ref{alg:line:partial_vector_update_improved:gamma_1} of Algorithm~\ref{alg:partial_vector_update_improved}, we have
\begin{align*}
\gamma_1^{\tmp}
= & ~ \gamma_1 + B \cdot \L_r [(M_{S})^{\top}] \cdot \sqrt{\wt{V}}\big( f(h^{\appr}) - f(\wt{g})\big)\\
= & ~ B\cdot \L_r [\beta_{2,S}] + B \cdot \L_r[(M_S)^{\top}] (\sqrt{\wt{V}} f(\wt{g})-\sqrt{V}f(g)) + B \cdot \L_r [(M_{S})^{\top}] \sqrt{\wt{V}}\big( f(h^{\appr}) - f(\wt{g})\big)\\
= & ~ B\cdot \L_r [\beta_{2,S}] + B \cdot \L_r[(M_S)^{\top}] \cdot \Big( (\sqrt{\wt{V}} f(\wt{g})-\sqrt{V}f(g)) + \sqrt{\wt{V}} ( f(h^{\appr}) - f(\wt{g}) ) \Big)\\
= & ~ B\cdot \L_r [\beta_{2,S}]+B \cdot \L_r[(M_S)^{\top}] \cdot (\sqrt{\wt{V}} f(h^{\appr})-\sqrt{V}f(g))\\
= & ~ B\cdot \L_r [\beta_{2,S}]+B \cdot \L_r[(M_S)^{\top}] \cdot \xi^{\tmp},
\end{align*}
where the second step follows from the invariant of $\gamma_1$ (Part~\ref{ass:invariant_improved:gamma_1} in Assumption~\ref{ass:invariant_improved}), the third and the fourth steps follow from merging terms, and the last step follows from Part 1 of this lemma.

\noindent \textbf{Part 3.} From the assignment of $\gamma_2^{\tmp}$ on Line~\ref{alg:line:partial_vector_update_improved:gamma_2} of Algorithm~\ref{alg:partial_vector_update_improved}, we have
\begin{align*} 
\gamma_2^{\tmp} 
= & ~ \gamma_2 + \Gamma M\sqrt{\wt{V}}\big( f(h^{\appr}) - f(\wt{g})\big)
=  \Gamma M (\sqrt{\wt{V}} f(\wt{g}) - \sqrt{V}f(g)) + \Gamma M\sqrt{\wt{V}}\big( f(h^{\appr}) - f(\wt{g})\big)\\
= & ~ \Gamma M(\sqrt{\wt{V}} f(h^{\appr}) - \sqrt{V}f(g))
=  \Gamma M \cdot \xi^{\tmp},
\end{align*}
where the second step follows from the invariant of $\gamma_2$ (Part~\ref{ass:invariant_improved:gamma_2} in Assumption~\ref{ass:invariant_improved}), the third step follows from merging terms, and the last step follows from Part 1 of this lemma.
\end{proof}

\begin{claim}[Part 3 of Lemma~\ref{lem:partial_vector_update_correct_improved}]
\label{cla:partial_vector_update_correct_improved:part_3}
All invariants of Assumption~\ref{ass:invariant_improved} are satisfied after the procedure \textsc{PartialVectorUpdate} (Algorithm~\ref{alg:partial_vector_update_improved}).
\end{claim}
\begin{proof}
First note that $v$, $\wt{v}$, $g$, $M$, $Q$, $B$, $\Delta$, $\Gamma$, $S$, $\beta_1$, and $\beta_2$ all remain the same after the procedure \textsc{VectorUpdate}. Also note that $\wt{g}$ is assigned the value $h^{\appr}$ (Line~\ref{alg:line:partial_vector_update_improved:wt_g} of Algorithm~\ref{alg:partial_vector_update_improved}).

Then using Claim~\ref{cla:vector_update_correct_improved:beta_1_beta_2_xi_gamma_1_gamma_2}, and since we assigned $\xi^{\tmp}$ to $\xi$, $\gamma_1^{\tmp}$ to $\gamma_1$, and $\gamma_2^{\tmp}$ to $\gamma_2$, we have
\begin{align*}
\gamma_1 = &~ B \cdot \L_r [\beta_{2,S}] + B \cdot \L_r[(M_S)^{\top}] \cdot \xi, 
& \gamma_2 = &~ \Gamma M \cdot \xi, \\
\xi = &~ \sqrt{\wt{V}}f(h^{\appr})-\sqrt{V}f(g) = \sqrt{\wt{V}}f(\wt{g})-\sqrt{V}f(g).
\end{align*}

Finally, from the assignment of $T$ on Line~\ref{alg:line:partial_vector_update_improved:T}, we have $T = \supp(h^{\appr} - g) = \supp(\wt{g} - g)$.
\end{proof}

\subsection{Correctness of \textsc{Initialize}}\label{sec:initialize_correct_improved}

\begin{lemma}[Correctness of \textsc{Initialize}]\label{lem:initialize_correct_improved}
When initialized, all the invariants of the data structure members stated in Assumption~\ref{ass:invariant_improved} are satisfied.
\end{lemma}

\begin{proof}
Since in the beginning $v$ and $\wt{v}$ are both assigned the value $w_0$, and $g$ and $\wt{g}$ are both assigned the value $h_0$, it is obvious that $S=\emptyset$, $T=\emptyset$, $\Delta=0$, $\Gamma=0$, $\xi=0$, $\gamma_1=0$, $\gamma_2=0$ all satisfy their invariant requirement of Assumption~\ref{ass:invariant_improved}. Also, $B=I=\L_*[0]$ also satisfies the invariant requirement. Finally note that the initial assignment of $M$, $Q$, $\beta_1$, $\beta_2$ directly satisfy their invariant requirement of Assumption~\ref{ass:invariant_improved}.
\end{proof}
\section{Data structure : time per call}
\label{sec:time_per_call_improved}

In Section~\ref{sec:time_per_call_improved} we provide a worst-case analysis of the running time per call for the five major procedures \textsc{MatrixUpdate}, \textsc{PartialMatrixUpdate}, \textsc{VectorUpdate}, \textsc{PartialVectorUpdate}, \textsc{Query}. We prove the amortized running time of these procedures later in Section~\ref{sec:amortize_time_improved}. %Our amortized time analysis builds upon the amortized analysis from \cite{cls19}.
In Section~\ref{sec:time_per_call_improved}, we ignore the running time of adding two vectors since it is only $O(n)$.

\begin{table}[h]
\small
    \centering
    \begin{tabular}{|l|l|l|l|} 
    \hline
    {\bf Procedure} & {\bf Time per Call} & {\bf Amortized Time} & {\bf Lemma}  \\ \hline
    \textsc{Query}   & $\Tmat(n^{\wt{a}}, n^a, n^{\wt{a}})+n^{1+b} $ & $\Tmat(n^{\wt{a}}, n^a, n^{\wt{a}})+n^{1+b} $ & Lemma~\ref{lem:query_time_improved} \\ \hline
    \textsc{MatrixUpdate}   & $\Tmat(k,n,n)$ & $\wt{O}(n^{\omega-1/2}+n^{2-a/2})$ & Lemma~\ref{lem:matrix_update_time_improved}, \ref{lem:main_amortize_matrix_update} \\ \hline
    \textsc{PartialMatrixUpdate} & $\Tmat(\wt{k}, n^a,n )$ & $\wt{O}(n^{1+(\omega-3/2)a}+n^{1+a-\wt{a}/2})$ & Lemma~\ref{lem:partial_matrix_update_time_improved}, \ref{lem:main_amortize_partial_matrix_update} \\ \hline
    \textsc{VectorUpdate} & $p n + n^{2a}$ & $\wt{O}(n^{1.5})$& Lemma~\ref{lem:vector_update_time_improved}, \ref{lem:main_amortize_vector_update} \\ \hline
    \textsc{PartialVectorUpdate} & $\wt{p} n^a + n^{2a}$ & $ \wt{O}(n^{2a-\wt{a}/2})$ & Lemma~\ref{lem:partial_vector_update_time_improved}, \ref{lem:main_amortize_partial_vector_update} \\ \hline
    \end{tabular}
    \caption{Time for different procedures. Summary of Section~\ref{sec:time_per_call_improved} and Section~\ref{sec:amortize_time_improved}.}
\end{table}

\subsection{Sparsity guarantees}
We first present some sparsity bounds that will be useful in the time analysis.

\begin{table}[h]
\small
    \centering
    \begin{tabular}{ | l | l | l | l | l | l | l | }
        \hline
        {\bf Procedure} & $\|w^{\appr} - v\|_0$ & $\|w^{\appr} - \wt{v}\|_0$ & $\|h^{\appr} - g\|_0$ & $\|h^{\appr} - \wt{g}\|_0$ & $\|\wt{v}-v\|_0$ & $\|\wt{g}-g\|_0$ \\ \hline
        \textsc{MatrixUpdate} & $ =k$ & / & / & / & $\leq n^a$ & $\leq n^a$ \\ \hline
        \textsc{P.MatrixUpdate} & $\leq n^a$ & $= \wt{k} \leq 2n^a$ & / & / & $\leq n^a$ & $\leq n^a$ \\ \hline
        \textsc{VectorUpdate} & $\leq n^a$ & $\leq n^{\wt{a}}$ &  $= p$ & / & $\leq n^a$ & $\leq n^a$ \\ \hline
        \textsc{P.VectorUpdate} & $\leq n^a$ & $\leq n^{\wt{a}}$ & $\leq n^a$ & $= \wt{p} \leq 2n^a$ & $\leq n^a$ & $\leq n^a$ \\ \hline
        \textsc{Query} & $\leq n^a$ & $\leq n^{\wt{a}}$ & $\leq n^a$ & $\leq n^{\wt{a}}$ & $\leq n^a$ & $\leq n^a$ \\ \hline
    \end{tabular}
\caption{Sparsity guarantees of $w^{\appr}$, $\wt{v}$, $v$, $h^{\appr}$, $\wt{g}$, $g$ when entering the procedures (Part 1 of Lemma~\ref{lem:improved:sparsity_guarantee}). We say some vector $x\in \R^n$ is $k$-sparse, it means that $\supp(x)=k$.}
\label{tab:improved:sparsity_guarantee_part_1}
\end{table}

\begin{table}[h]
\small
    \centering
    \begin{tabular}{ | l | p{2cm} | p{1.8cm} | l | l | p{1.8cm} | p{1.5cm} | }
        \hline
        {\bf Procedure} & $\|\Delta + \partial \Delta\|_0$, $\|\Gamma + \partial \Gamma\|_0$ & $\|\partial \Delta\|_0$, $\|\partial \Gamma\|_0$, $|\partial S|$ & $\|\xi + \partial \xi\|_0$ & $\|\partial \xi\|_0$ & $\|\Delta\|_0, \|\xi\|_0$, $\|\Gamma\|_0$, $|S|$ & $|S\cup \partial S|$  \\ \hline
        \textsc{MatrixUpdate} & $=k$ & / & / & / & $\leq n^a$ & $\leq 3k$ \\ \hline
        \textsc{P.MatrixUpdate} & $\leq n^a$ & $= \wt{k} \leq 2n^a$ & / & / & $\leq n^a$ & $\leq 3n^a$ \\ \hline
        \textsc{VectorUpdate} & $\leq n^a$ & $\leq n^{\wt{a}}$ & $= p$ & / & $\leq n^a$ & / \\ \hline
        \textsc{P.VectorUpdate} & $\leq n^a$ & $\leq n^{\wt{a}}$ & $\leq n^a$ & $= \wt{p} \leq 2n^a$ & $\leq n^a$ & / \\ \hline
        \textsc{Query} & $\leq n^a$ & $\leq n^{\wt{a}}$ & $\leq n^a$ & $\leq n^{\wt{a}}$ & $\leq n^a$ & $\leq 2n^a$ \\ \hline
    \end{tabular}
\caption{Sparsity guarantees of other local variables (Definition~\ref{def:local}) when entering the procedures (Part 2 of Lemma~\ref{lem:improved:sparsity_guarantee}). We say some vector $x\in \R^n$ is $k$-sparse, it means that $\supp(x)=k$. }\label{tab:improved:sparsity_guarantee_part_2}
\end{table}

\begin{lemma}[Sparsity guarantees]
\label{lem:improved:sparsity_guarantee}
The members of data structure presented in Table~\ref{tab:improved:sparsity_guarantee_part_1} and Table~\ref{tab:improved:sparsity_guarantee_part_2} all follow the invariants of Assumption~\ref{ass:invariant_improved}, and the local variables presented in Table~\ref{tab:improved:sparsity_guarantee_part_2} all follow the definition of Definition~\ref{def:local}.
We have the following sparsity guarantees.
\begin{enumerate}
\item When entering the procedures \textsc{MatrixUpdate}, \textsc{PartialMatrixUpdate}, \textsc{VectorUpdate}, \textsc{PartialVectorUpdate}, and \textsc{Query}, we have the sparsity bounds on $\|w^{\appr} - v\|_0$, $\|w^{\appr} - \wt{v}\|_0$, $\|h^{\appr} - g\|_0$, $\|h^{\appr} - \wt{g}\|_0$, $\|\wt{v}-v\|_0$, and $\|\wt{g}-g\|_0$ as presented in Table~\ref{tab:improved:sparsity_guarantee_part_1}.
\item When entering the procedures \textsc{MatrixUpdate}, \textsc{PartialMatrixUpdate}, \textsc{VectorUpdate}, \textsc{PartialVectorUpdate}, and \textsc{Query}, we have the sparsity bounds on $\Delta + \partial \Delta$, $\Gamma + \partial \Gamma$, $S\cup \partial S$, $\partial \Delta$, $\partial \Gamma$, $\partial S$, $\xi + \partial \xi$, $\partial \xi$, $\Delta$, $\Gamma$, $S$, $\xi$ as presented in Table~\ref{tab:improved:sparsity_guarantee_part_2}.
\end{enumerate}
\end{lemma}
\begin{proof}
\noindent \textbf{Part 1.}
The first four columns follow from Corollary~\ref{cor:improved:sparsity_guarantee_procedures}, and the last two columns follow from Corollary~\ref{cor:improved:sparsity_of_members}.

\noindent \textbf{Part 2.} For Col.~1, we have
\begin{align*}
\|\Delta + \partial \Delta\|_0 = &~ \|W^{\appr} - V\|_0 = \|w^{\appr} - v\|_0, \\
\|\Gamma + \partial \Gamma\|_0 = &~ \|\sqrt{W^{\appr}} - \sqrt{V}\|_0 = \|w^{\appr} - v\|_0,
%\\
%|S\cup \partial S| = &~ |\supp(w^{\appr} - v)| = \|w^{\appr} - v\|_0,
\end{align*}
the bounds of this column then follow from Col.~1 of Table~\ref{tab:improved:sparsity_guarantee_part_1}.

For Col.~2, we have
\begin{align*}
\|\partial \Delta\|_0 = &~ \|W^{\appr} - \wt{V}\|_0 = \|w^{\appr} - \wt{v}\|_0, \\
\|\partial \Gamma\|_0 = &~ \|\sqrt{W^{\appr}} - \sqrt{\wt{V}}\|_0 = \|w^{\appr} - \wt{v}\|_0, \\
|\partial S| = &~ \|W^{\appr} - \wt{V}\|_0 = \|w^{\appr} - \wt{v}\|_0, 
\end{align*}
the bounds of this column then follow from Col.~2 of Table~\ref{tab:improved:sparsity_guarantee_part_1}.

For Col.~3 we have
\begin{align*}
\|\xi + \partial \xi\|_0 = &~ \|\sqrt{W^{\appr}}f(h^{\appr}) - \sqrt{V}f(g)\|_0 
\leq  \max \{ \|\sqrt{W^{\appr}} - \sqrt{V}\|_0, \|f(h^{\appr}) - f(g)\|_0 \} \\
= &~ \max \{ \| w^{\appr} - v \|_0, \|h^{\appr} - g\|_0 \},
\end{align*}
the bounds of this column then follow from Col.~1 and Col.~3 of Table~\ref{tab:improved:sparsity_guarantee_part_1} and the fact $p\geq n^a$ (since we have the equivalent fact for $p$ as of Fact~\ref{fac:bound_on_k_j} for $k$).

For Col.~4 we have
\begin{align*}
\|\partial \xi\|_0 = &~ \|\sqrt{W^{\appr}}f(h^{\appr}) - \sqrt{\wt{V}}f(\wt{g})\|_0 
\leq  \max \{ \|\sqrt{W^{\appr}} - \sqrt{\wt{V}}\|_0, \|f(h^{\appr}) - f(\wt{g})\|_0 \} \\
= &~ \max \{ \| w^{\appr} - \wt{v} \|_0, \|h^{\appr} - \wt{g}\|_0 \},
\end{align*}
the bounds of this column then follow from Col.~2 and Col.~4 of Table~\ref{tab:improved:sparsity_guarantee_part_1} and the fact $\wt{p}\geq n^{\wt{a}}$ (since we have the equivalent fact for $\wt{p}$ as of Fact~\ref{fac:bound_wt_k_j} for $\wt{k}$).

For Col.~5 we have 
\begin{align*}
\|\Delta\|_0 = &~ \|\wt{V} - V\|_0 = \|\wt{v} - v\|_0, \\
\|\Gamma\|_0 = &~ \|\sqrt{\wt{V}} - \sqrt{V}\|_0 = \|\wt{v} - v\|_0, \\
|S| = &~ |\supp(\wt{v} - v)| = \|\wt{v} - v\|_0,\\
\|\xi\|_0 = &~ \|\sqrt{\wt{V}}f(\wt{g}) - \sqrt{V}f(g)\|_0 \leq \max \{ \| \wt{v} - v \|_0, \| \wt{g} - g \|_0 \},
\end{align*}
the bounds of this column then follow from Col.~5 and Col.~6 of Table~\ref{tab:improved:sparsity_guarantee_part_1}.

For the first part (\textsc{MatrixUpdate}) of Col.~6, we have
\begin{align*}
|S\cup \partial S| 
\leq & ~ |S| + |\partial S|
\leq  |\supp(v-\wt{v})| + |\supp(w^{\appr} - \wt{v})| \\
\leq & ~ |\supp(v-\wt{v})| + |\supp(w^{\appr} - v)| + |\supp(v-\wt{v})|
\leq  n^a+k+n^a
\leq  3k,
\end{align*}
where the first and the third steps both follow from triangle inequality, the second step follows from $S=\supp(v-\wt{v})$ (part~\ref{ass:invariant_improved:S} Assumption~\ref{ass:invariant_improved}) and $\partial S=\supp(w^{\appr}-\wt{v})$ (Part~\ref{def:local:partial_S} of Definition~\ref{def:local}), the fourth step follows from Col.~1 and Col.~5 of Table~\ref{tab:improved:sparsity_guarantee_part_1}, the last step follows from the fact that $k\geq n^a$ (Fact~\ref{fac:bound_on_k_j}). 

For the second part (\textsc{PartialMatrixUpdate}) of Col.~6, we have
\begin{align*}
|S\cup \partial S| 
\leq  |S| + |\partial S| 
\leq |\supp(v-\wt{v})| + |\supp(w^{\appr} - \wt{v})| 
\leq  n^a + 2n^a = 3n^a,
\end{align*}
where the first two steps are the same as previous inequality, and the third step follows from Col.~2 and Col.~5 of Table~\ref{tab:improved:sparsity_guarantee_part_1}.

For the fifth part (\textsc{Query}) of Col.~6, we have
\begin{align*}
|S\cup \partial S| 
\leq  |S| + |\partial S| 
\leq |\supp(v-\wt{v})| + |\supp(w^{\appr} - \wt{v})| 
\leq n^a + n^{\wt{a}}
\leq 2n^a,
\end{align*}
where the first two steps are the same as previous inequality, the third step follows from Col.~2 and Col.~5 of Table~\ref{tab:improved:sparsity_guarantee_part_1}.
\end{proof}

\begin{table}[h]
\small
    \centering
    \begin{tabular}{|l|l|l|} 
         \hline
         {\bf Procedure} & {\bf Lemma} & {\bf Section} \\ \hline
         \textsc{Query} & Lemma~\ref{lem:query_time_improved} & Section~\ref{sec:query_time_improved} \\ \hline
         \textsc{MatrixUpdate} & Lemma~\ref{lem:matrix_update_time_improved} & Section~\ref{sec:matrix_update_time_improved} \\ \hline 
         \textsc{PartialMatrixUpdate} & Lemma~\ref{lem:partial_matrix_update_time_improved} & Section~\ref{sec:partial_matrix_update_time_improved} \\ \hline
         \textsc{VectorUpdate} & Lemma~\ref{lem:vector_update_time_improved} & Section~\ref{sec:vector_update_time_improved} \\ \hline
         \textsc{PartialVectorUpdate} & Lemma~\ref{lem:partial_vector_update_time_improved} & Section~\ref{sec:partial_vector_update_time_improved} \\ \hline
         %\textsc{UpdateQuery} & Lemma~\ref{lem:update_query_time_improved} & Section~\ref{sec:update_query_time_improved} \\ \hline
         %\textsc{Update} & Lemma~\ref{lem:update_time_improved} & Section~\ref{sec:update_time_improved} \\ \hline
         \textsc{Initialize} & Lemma~\ref{lem:initialize_time_improved} & Section~\ref{sec:initialize_time_improved} \\ \hline
    \end{tabular}
     \caption{Summary of the section that proves running time per call.}\label{tab:time_summary_improved}
    
\end{table}

\subsection{Running time of \textsc{Query}}\label{sec:query_time_improved}
The goal of this section is to prove Lemma~\ref{lem:query_time_improved}. We will use the following sparsity guarantees that are proved in Lemma~\ref{lem:improved:sparsity_guarantee}.
\begin{fact}[Sparsity guarantees for \textsc{Query}]
\label{fac:query_time_improved_sparsity}
When entering \textsc{Query} we have the following sparsity guarantee (from Table~\ref{tab:improved:sparsity_guarantee_part_2}):
\begin{multicols}{3}
\begin{enumerate}
\item \label{fac:query_time_improved_sparsity_Gamma_partial_Gamma} $\|\Gamma + \partial \Gamma\|_0 \leq n^a$,
\item \label{fac:query_time_improved_sparsity_partial_Gamma} $\|\partial \Gamma\|_0 \leq n^{\wt{a}}$,
\item \label{fac:query_time_improved_sparsity_Gamma} $\|\Gamma\|_0 \leq n^a$,
\item \label{fac:query_time_improved_sparsity_xi_partial_xi} $\|\xi + \partial \xi\|_0 \leq n^a$,
\item \label{fac:query_time_improved_sparsity_partial_xi} $\|\partial \xi\|_0 \leq n^{\wt{a}}$,
\item \label{fac:query_time_improved_sparsity_partial_S} $|\partial S|\leq n^{\wt{a}}$,
\item \label{fac:query_time_improved_sparsity_S} $|S| \leq n^a$.
%\item \label{fac:query_time_improved_sparsity_part_1} $\|\Delta + \partial \Delta\|_0\leq n^a $,~ $\|\Gamma + \partial \Gamma\|_0\leq n^a$,~ $\|\xi + \partial \xi\|_0 \leq n^a$,
%\item \label{fac:query_time_improved_sparsity_part_2} $\|\partial \Delta\|_0\leq n^{\wt{a}}$, $\|\partial \Gamma\|_0\leq n^{\wt{a}}$, $\|\partial \xi\|_0 \leq n^{\wt{a}}$,
%\item \label{fac:query_time_improved_sparsity_part_3} $\|\Delta\|_0\leq n^a$, $\|\Gamma\|_0\leq n^a$, $\|\xi\|_0 \leq n^a$,
%\item \label{fac:query_time_improved_sparsity_part_4} $|S\cup \partial S|\leq n^a$, $|\partial S|\leq n^{\wt{a}}$, $|S| \leq n^a$.
{\color{white} \item}
{\color{white} \item}
\end{enumerate}
\end{multicols}
\end{fact}

\begin{lemma}[Running time of \textsc{Query}]\label{lem:query_time_improved}
In procedure \textsc{Query} (Algorithm~\ref{alg:query_improved}), it takes
\begin{enumerate}
\item $O(n^{1+b} + n^{a+\wt{a}})$ time to compute
\begin{align*}
r_2\leftarrow Q[j]\xi + R[j]\gamma_2 + R[j] \partial \Gamma M (\xi + \partial \xi) + \big(Q[j] + R[j] \Gamma M\big) \partial \xi,
\end{align*}
\item $O(n^{1+b})$ time to compute 
\[
r_3 \leftarrow R[j](\Gamma+\partial \Gamma) \beta_2,
\]
\item $O(n^{a+\wt{a}})$ time to compute 
\[
\partial \gamma \leftarrow B\cdot \L_r [(\beta_{2})_{\partial S\backslash S}] + B\cdot \L_r[(M_{\partial S \backslash S})^{\top}] \cdot (\xi+\partial \xi) + E \cdot \partial \xi,
\]
\item $O(n^{a+\wt{a}})$ time to compute
\[
(U', C, U) \leftarrow \textsc{Decompose}\Big(\L_*[( \Delta^{\new}_{S^{\new}, S^{\new}})^{-1} +  M_{S^{\new}, S^{\new}} ] - \L_* [\Delta_{S, S}^{-1}+ M_{S,S}]\Big),
\]
%\item $O(n^{a+b})$ time to compute \[
%r_{4,1}\leftarrow -\L_c[(Q[l])_{S}]\cdot \gamma_{1} - F[l] \cdot %\gamma_{1},
%\]
%\item $O(n^{a+\wt{a}} + n^{a+b})$ time to compute 
%\[
%r_{4,2}\leftarrow \big(-\L_c [(Q[l])_{S^{\new}}] - F[l] - R[l]\Gamma\cdot %\L_c[M_{\partial S\backslash S}] - R[l] \partial \Gamma\cdot %\L_c[M_{S^{\new}}] \big) \cdot \partial \gamma,
%\]
%\item $O(n^{a+b}+n^{a+\wt{a}})$ time to compute 
%\[
%r_{4,3}\leftarrow \big( -\L_c[((Q[l])_{\partial S\backslash S}] - R[l] %\Gamma\cdot \L_c[M_{\partial S\backslash S}] - R[l] \partial \Gamma \cdot %\L_c[M_{ S^{\new}}] \big)\cdot \gamma_1,
%\]
\item $O(\Tmat(n^a, n^{\wt{a}}, n^{\wt{a}}))$ time to compute
\begin{align*}
E^{\tmp}\leftarrow & E_{\partial S} - B_{(\partial S\backslash S)}M_{(\partial S\backslash S), \partial S}, ~E^{\tmp}_{S'}\leftarrow -E^{\tmp}_{S'}, ~ E^{\tmp}_{(S\cap \partial S)\backslash S'}\leftarrow 0, \text{~and} \\
U^{\tmp}\leftarrow & [B_{\partial S}, B_{\partial S}, E^{\tmp}],
\end{align*}
\item $O(\Tmat(n^{\wt{a}},n^a,n^{\wt{a}}))$ time to compute
\[
\gamma^{\tmp} \leftarrow U^{\tmp}(C^{-1}+U^{\top}U^{\tmp})^{-1}U^{\top} \cdot (\gamma_1 + \partial \gamma ),
\]
%\item $O(n^{a+\wt{a}}+n^{a+b})$ time to compute 
%\[
%r_{4,4}\leftarrow \left(\L_c[(Q[l])_{S^{\new}}] + F[l] + R[l]\Gamma\cdot %\L_c[M_{\partial S\backslash S}] + R[l]\partial\Gamma \cdot %\L_c[M_{S^{\new}}]\right)\cdot r_{4,4}^{\tmp}.
%\]
\item $O(n^{\wt{a}+a} + n^{1+b})$ time to compute 
\[
r_4\leftarrow \Big(\L_c[(Q[j])_{S^{\new}}]+F[j]+R[j] \Gamma \cdot (\L_c[M_{\partial S\backslash S}]-\L_c[M_{S'}]) + R[j]\partial \Gamma\cdot \L_c[M_{S^{\new}}]\Big)\cdot (\gamma^{\tmp}-\gamma_1-\partial \gamma).
\]
\end{enumerate}
Overall, the time to compute $r$ (which is $r_1 + r_2 + r_3 + r_4$) is
\begin{align*}
O(n^{1+b} + \Tmat(n^{\wt{a}},n^a,n^{\wt{a}})).
\end{align*}
\end{lemma}

\begin{claim}[Part 1 of Lemma~\ref{lem:query_time_improved}]\label{cla:query_time_improved_r_2}
In procedure \textsc{Query} (Algorithm~\ref{alg:query_improved}), it takes $O(n^{1+b} + n^{a+\wt{a}})$
time to compute 
\begin{align*}
r_2 \leftarrow Q[j] \xi + R[j]\gamma_2 + \underbrace{R[j] \partial \Gamma M(\xi + \partial \xi)}_{a_1} + \underbrace{\big(Q[j] + R[j] \Gamma M\big)\partial \xi}_{a_2}.
\end{align*}
\end{claim}
\begin{proof}
The running time of this step can be split into the following parts:

The first part is to compute $Q[j]\xi$ by multiplying a $n^b \times n$ matrix $Q[j]$ with a $n \times 1 $ vector $\xi$. It takes $O(n^{1+b})$ time. The second part is to compute $R[j]\gamma_2$ by multiplying a $n^b \times n$ matrix $R[j]$ with a $n \times 1$ vector $\gamma_2$. It takes $O(n^{1+b})$ time.

The third part is to compute $a_1$ as follows:
\begin{enumerate}
\item We multiply a $n^{\wt{a}}$-sparse $n \times n$ diagonal matrix $\partial \Gamma$ (Part~\ref{fac:query_time_improved_sparsity_partial_Gamma} of Fact~\ref{fac:query_time_improved_sparsity}) with a $n \times n$ matrix $M$ and then with a $n^a$-sparse $n \times 1$ vector $(\xi + \partial \xi) $ (Part~\ref{fac:query_time_improved_sparsity_xi_partial_xi} of Fact~\ref{fac:query_time_improved_sparsity}). It takes $O(n^{a+\wt{a}})$ time.
\item We multiply a $ n^b \times n$ matrix $ R[j]$ and a $n \times 1 $ vector $
(\partial \Gamma M(\xi+\partial \xi)).$
It takes $O( n^{1+b})$ time.
\end{enumerate}
So computing $a_1$ takes $O(n^{a+\wt{a}} + n^{1+b})$ time in total.

The fourth part is to compute $a_2$ as follows:
\begin{enumerate}
\item We multiply a $n^a$-sparse $n \times n$ diagonal matrix $\Gamma$ (Part~\ref{fac:query_time_improved_sparsity_Gamma} of Fact~\ref{fac:query_time_improved_sparsity}) with a $n\times n$ matrix $M$ and then with a $n^{\wt{a}}$-sparse $n \times 1$ vector $\partial \xi$ (Part~\ref{fac:query_time_improved_sparsity_partial_xi} of Fact~\ref{fac:query_time_improved_sparsity}). It takes $O(n^{a+\wt{a}})$ time.
\item We multiply a $n^b \times n$ matrix $R[j]$ with a $n \times 1$ vector $(\Gamma M \partial \xi)$. It takes $O(n^{1+b})$ time.
\item We multiply a $n^b \times n$ matrix $Q[j]$ and a $n^{\wt{a}}$-sparse $n \times 1$ vector $\partial \xi$ (Part~\ref{fac:query_time_improved_sparsity_partial_xi} of Fact~\ref{fac:query_time_improved_sparsity}). It takes $O(n^{\wt{a}+b})$ time.
\end{enumerate}
So computing $a_2$ takes $O(n^{a+\wt{a}} + n^{1+b})$ time in total since $\wt{a}\leq 1$.

%Finally, we add the four $n^b \times 1$ vectors $Q[j]\xi$, $R[j]\gamma_2$, $a_1$, $a_2$ together, and this step takes $O(n)$ time.

Thus the total running time is $O(n^{1+b} + n^{a+\wt{a}})$.
\end{proof}

\begin{claim}[Part 2 of Lemma~\ref{lem:query_time_improved}]\label{cla:query_time_improved_r_3}
In procedure \textsc{Query} (Algorithm~\ref{alg:query_improved}), it takes $O(n^{1+b})$ time to compute $r_3 \leftarrow R[j](\Gamma+\partial \Gamma) \beta_2$.
\end{claim}

\begin{proof}
The running time of this step can be split into the following parts.
\begin{enumerate}
\item We multiply a $n^a$-sparse $ n \times n $ diagonal matrix $(\Gamma + \partial \Gamma) $ (Part~\ref{fac:query_time_improved_sparsity_Gamma_partial_Gamma} of Fact~\ref{fac:query_time_improved_sparsity}) and a $n \times 1$ vector $\beta_2$. It takes $O(n^a)$ time.
\item We multiply a $ n^b\times n $ matrix $ R[j]$ and a $n \times 1$ vector $( (\Gamma+\partial \Gamma) \beta_2 )$. It takes $O(n^{1+b})$ time.
\end{enumerate}
The total running time is
$O(n^a + n^{1+b}) = O(n^{1+b})$ since $a \leq 1$.
\end{proof}

\begin{claim}[Part 3 of Lemma~\ref{lem:query_time_improved}]\label{cla:query_time_improved_partial_gamma}
In procedure \textsc{Query} (Algorithm~\ref{alg:query_improved}), it takes $O(n^{a+\wt{a}})$ time to compute $\partial \gamma \leftarrow B\cdot \L_r [(\beta_{2})_{\partial S\backslash S}] + B\cdot \L_r[(M_{\partial S \backslash S})^{\top}] \cdot (\xi+\partial \xi) + E \cdot \partial \xi$.
\end{claim}

\begin{proof}
The running time of this step can be split into the following parts.
\begin{enumerate}
\item We multiply a $6n^a\times 6n^a$ matrix $B$ with a $n^{\wt{a}}$-sparse $6n^a\times 1$ vector $\L_r[(\beta_2)_{\partial S\backslash S}]$ (since $|\partial S\backslash S|\leq |\partial S|\leq n^{\wt{a}}$ by Part~\ref{fac:query_time_improved_sparsity_partial_S} of Fact~\ref{fac:query_time_improved_sparsity}). It takes $O(n^{a+\wt{a}})$ time.
\item We multiply a $6n^a\times n$ matrix $\L_r[(M_{\partial S\backslash S})^{\top}]$ that only has $n^{\wt{a}}$ non-zero rows (since $|\partial S\backslash S|\leq |\partial S|\leq n^{\wt{a}}$ by Part~\ref{fac:query_time_improved_sparsity_partial_S} of Fact~\ref{fac:query_time_improved_sparsity}) with a $n^a$-sparse $n \times 1$ vector $(\xi + \partial \xi)$ (Part~\ref{fac:query_time_improved_sparsity_xi_partial_xi} of Fact~\ref{fac:query_time_improved_sparsity}). It takes $O(n^{a+\wt{a}})$ time.
\item We multiply a $6n^a\times 6n^a$ matrix $B$ with a $n^{\wt{a}}$-sparse $6n^a\times 1$ vector $(\L_r[(M_{\partial S\backslash S})^{\top}](\xi+\partial \xi))$ (since $|\partial S\backslash S|\leq |\partial S|\leq n^{\wt{a}}$ by Part~\ref{fac:query_time_improved_sparsity_partial_S} of Fact~\ref{fac:query_time_improved_sparsity}). It takes $O(n^{a+\wt{a}})$ time.
\item We multiply a $6n^a \times n$ matrix $E$ with a $n^{\wt{a}}$-sparse vector $\partial \xi$ (Part~\ref{fac:query_time_improved_sparsity_partial_xi} of Fact~\ref{fac:query_time_improved_sparsity}). It takes $O(n^{a+\wt{a}})$ time.
\end{enumerate}
Thus the total running time is $O(n^{a+\wt{a}})$.
\end{proof}

\begin{claim}[Part 4 of Lemma~\ref{lem:query_time_improved}]\label{cla:query_time_improved_UCU}
In procedure \textsc{Query} (Algorithm~\ref{alg:query_improved}), it takes $O(n^{a+\wt{a}})$ time to compute 
\[
(U', C, U) \leftarrow \textsc{Decompose}\Big(\L_*[( \Delta^{\new}_{S^{\new}, S^{\new}})^{-1} +  M_{S^{\new}, S^{\new}} ] - \L_* [\Delta_{S, S}^{-1}+ M_{S,S}]\Big).
\]
And the size of the computed matrices are $U',U \in \R^{n^a\times c}$, $C\in \R^{c\times c}$, where $c\leq O(n^{\wt{a}})$.
\end{claim}
\begin{proof}
For ease of notation, we denote $N:=\L_*[( \Delta^{\new}_{S^{\new}, S^{\new}})^{-1} +  M_{S^{\new}, S^{\new}} ] - \L_* [\Delta_{S, S}^{-1}+ M_{S,S}]$.

From Lemma~\ref{lem:structure_inverse_change} we know that the non-zero entries of $N$ can be split into three parts: $N_{\partial S, (S\backslash\partial S)}$, $N_{(S\backslash\partial S), \partial S}$, and $N_{\partial S, \partial S}$. Thus we don't need to compute $N$ explicitly, but only compute the non-zero entries of $N$, which takes $O(|\partial S|\cdot |S\backslash\partial S| + |\partial S|\cdot|\partial S|)= O(n^{a+\wt{a}})$ time (from Part~\ref{fac:query_time_improved_sparsity_S} and Part~\ref{fac:query_time_improved_sparsity_partial_S} of Fact~\ref{fac:query_time_improved_sparsity} and since $\wt{a}\leq a$).

Then $N$ satisfies the requirement of Lemma~\ref{lem:UCU_decomposition} with $S_1=S\backslash\partial S$, $S_2=\partial S$, so the function \textsc{Decompose} can be computed in $O(n^a|S_2|)=O(n^a|\partial S|)=O(n^{a+\wt{a}})$ time (Part~\ref{fac:query_time_improved_sparsity_partial_S} of Fact~\ref{fac:query_time_improved_sparsity}). Also using Lemma~\ref{lem:UCU_decomposition}, we know that the computed matrix $C$ has size $c\times c = (3|S_2|)\times (3|S_2|) = (3|\partial S|)\times (3|\partial S|)$, thus $c\leq O(n^{\wt{a}})$ (Part~\ref{fac:query_time_improved_sparsity_partial_S} of Fact~\ref{fac:query_time_improved_sparsity}).

Thus the total running time is $O(n^{a+\wt{a}})$.
\end{proof}

\begin{claim}[Part 5 of Lemma~\ref{lem:query_time_improved}]\label{cla:query_time_improved_U_tmp}
In procedure \textsc{Query} (Algorithm~\ref{alg:query_improved}), it takes $O(\Tmat(n^a, n^{\wt{a}}, n^{\wt{a}}))$ time to compute
\begin{align*}
E^{\tmp}\leftarrow E_{\partial S} - B_{(\partial S\cap S)}M_{(\partial S\cap S), \partial S};~~
E^{\tmp}_{S'}\leftarrow -E^{\tmp}_{S'}; ~~
E^{\tmp}_{(S\cap \partial S)\backslash S'}\leftarrow 0;~~
U^{\tmp}\leftarrow [B_{\partial S}, B_{\partial S}, E^{\tmp}].
\end{align*}
\end{claim}
\begin{proof}
To compute the initial $E^{\tmp}\leftarrow E_{\partial S} - B_{(\partial S\cap S)}M_{(\partial S\cap S), \partial S}$, we need to multiply a $6n^a\times |\partial S\cap S|$ matrix $B_{(\partial S\cap S)}$ with a $|\partial S\cap S|\times |\partial S|$ matrix $M_{(\partial S\cap S), \partial S}$, and this takes $\Tmat(6n^a, |\partial S\cap S|, |\partial S|)=O(\Tmat(n^a, n^{\wt{a}}, n^{\wt{a}}))$ time (Part~\ref{fac:query_time_improved_sparsity_partial_S} of Fact~\ref{fac:query_time_improved_sparsity}).

Finally note that the other two steps $E^{\tmp}_{S'}\leftarrow -E^{\tmp}_{S'}$ and $E^{\tmp}_{(S\cap \partial S)\backslash S'}\leftarrow 0$ to adjust $E^{\tmp}$ takes the same time as the size of $E^{\tmp}$, which is $n^a\times |\partial S|=O(n^{a+\wt{a}})$ (Part~\ref{fac:query_time_improved_sparsity_partial_S} of Fact~\ref{fac:query_time_improved_sparsity}). Computing $U^{\tmp}$ only needs to copy entries from the already computed $B$ and $E^{\tmp}$, so it also takes the same time as the size of $U^{\tmp}$, which is $n^a\times 3|\partial S|=O(n^{a+\wt{a}})$. Thus the total running time is
\begin{align*}
O(\Tmat(n^a, n^{\wt{a}}, n^{\wt{a}})+n^{a+\wt{a}}) = O(\Tmat(n^a, n^{\wt{a}}, n^{\wt{a}})).
\end{align*}
\end{proof}

\begin{claim}[Part 6 of Lemma~\ref{lem:query_time_improved}]\label{cla:query_time_improved_r_4_4_tmp}
In procedure \textsc{Query} (Algorithm~\ref{alg:query_improved}), it takes $O(\Tmat(n^{\wt{a}},n^a,n^{\wt{a}}) )$ time to compute 
\begin{align*}
    \gamma^{\tmp} \leftarrow U^{\tmp}\underbrace{(C^{-1}+U^{\top}U^{\tmp})^{-1}}_{N}U^{\top} \cdot (\gamma_1 + \partial \gamma ).
\end{align*}
\end{claim}

\begin{proof}
First note that from Lemma~\ref{cla:query_time_improved_UCU} we have that the sizes of $U'$ and $U$ are $6n^a\times 3|\partial S|$, and the size of $C$ is $3|\partial S|\times 3|\partial S|$. From the assignment of $U^{\tmp}$ on Line~\ref{alg:line:U_tmp}, we know that the size of $U^{\tmp}$ is also $6n^a\times 3|\partial S|$. The running time of this step can be split into the following parts:

The first part is to compute the $3|\partial S| \times 3|\partial S|$ matrix $N$.
\begin{enumerate}
\item Computing the inverse of a $3|\partial S|\times 3|\partial S|$ matrix $C$ takes $O(n^{\wt{a}\omega})$ time (Part~\ref{fac:query_time_improved_sparsity_partial_S} of Fact~\ref{fac:query_time_improved_sparsity}).
\item We multiply a $3|\partial S| \times 6n^a$ rectangular matrix $U^{\top}$ with a $6n^a \times 3|\partial S|$ matrix $U^{\tmp}$, which takes $\Tmat (3|\partial S|, 6n^a, 3|\partial S| )= O(\Tmat (n^{\wt{a}}, n^a , n^{\wt{a}} ))$ time ($|\partial S|\leq n^{\wt{a}}$ from Part~\ref{fac:query_time_improved_sparsity_partial_S} of Fact~\ref{fac:query_time_improved_sparsity}). 
\item We add a $3|\partial S|\times 3|\partial S|$ matrix $C^{-1}$ with a $3|\partial S|\times 3|\partial S|$ matrix $U^{\top}U^{\tmp}$ and calculate its inverse. It takes $O(|\partial S|^{\omega})=O(n^{\wt{a}\omega})$ time (Part~\ref{fac:query_time_improved_sparsity_partial_S} of Fact~\ref{fac:query_time_improved_sparsity}).
\end{enumerate}
Thus in total computing $N$ takes time $O( \Tmat ( n^{\wt{a}}, n^a, n^{\wt{a}} ) + n^{\wt{a}\omega}) = O(\Tmat (n^{\wt{a}}, n^a , n^{\wt{a}} ))$, sinnce $\wt{a}\leq a$ and thus $n^{\wt{a}\omega}=\Tmat(n^{\wt{a}}, n^{\wt{a}}, n^{\wt{a}})\leq \Tmat(n^{\wt{a}}, n^a, n^{\wt{a}})$.

The second part is to compute the $6n^a \times 1$ vector $\gamma^{\tmp}$.
\begin{enumerate}
\item We multiply the $3|\partial S|\times 6n^a$ matrix $U^{\top}$ with a $6n^a\times 1$ vector $(\gamma_1+\partial \gamma)$. This takes $O(n^a\cdot |\partial S|)=O(n^{a+\wt{a}})$ time (Part~\ref{fac:query_time_improved_sparsity_partial_S} of Fact~\ref{fac:query_time_improved_sparsity}).
\item We multiply the $3|\partial S|\times 3|\partial S|$ matrix $N$ with a $3|\partial S|\times 1$ vector $U^{\top}(\gamma_1+\partial \gamma)$. This takes $O(|\partial S|^2)=O(n^{2\wt{a}})$ time (Part~\ref{fac:query_time_improved_sparsity_partial_S} of Fact~\ref{fac:query_time_improved_sparsity}).
\item We multiply the $6n^a\times 3|\partial S|$ matrix $U^{\tmp}$ with a $3|\partial S|\times 1$ vector $NU^{\top}(\gamma_1+\partial \gamma)$. This takes $O(n^a\cdot |\partial S|)=O(n^{a+\wt{a}})$ time (Part~\ref{fac:query_time_improved_sparsity_partial_S} of Fact~\ref{fac:query_time_improved_sparsity}).
\end{enumerate}
Thus this part takes $O(n^{a+\wt{a}})$ time since $\wt{a}\leq a$.

Thus the total time to compute $\gamma^{\tmp}$ is 
\[
O(\Tmat(n^{\wt{a}}, n^a, n^{\wt{a}}) + n^{a+\wt{a}}) = O(\Tmat(n^{\wt{a}}, n^a, n^{\wt{a}})).
\]
\end{proof}

\begin{claim}[Part 7 of Lemma~\ref{lem:query_time_improved}]\label{cla:query_time_improved_r_4_4}
In procedure \textsc{Query} (Algorithm~\ref{alg:query_improved}), it takes $O(n^{\wt{a}+a} + n^{1+b})$ time to compute 
\begin{align*}
r_4\leftarrow \Big(\L_c[(Q[j])_{S^{\new}}]+F[j]+R[j] \Gamma \cdot (\L_c[M_{\partial S\backslash S}]-\L_c[M_{S'}]) + R[j]\partial \Gamma\cdot \L_c[M_{S^{\new}}]\Big)\cdot (\gamma^{\tmp}-\gamma_1-\partial \gamma).
\end{align*}
\end{claim}

\begin{proof}
The running time can be split into the following parts:
\begin{enumerate}
    \item We multiply a $n^{\wt{a}}$-sparse $n\times n$ diagonal matrix $\partial \Gamma$ (Part~\ref{fac:query_time_improved_sparsity_partial_Gamma} of Fact~\ref{fac:query_time_improved_sparsity}) with a $n\times 6n^a$ matrix $\L_c[M_{S^{\new}}]$ with a $6n^a\times 1$ vector $(\gamma^{\tmp}-\gamma_1-\partial \gamma)$. It takes $O(n^{\wt{a}+a})$ time.
    \item We multiply a $n^b \times n$ matrix $R[j]$ with a $n\times 1$ vector $(\partial \Gamma \L_c[M_{S^{\new}}](\gamma^{\tmp}-\gamma_1-\partial \gamma))$. It takes $O(n^{1+b})$ time.
    \item We multiply a $n^a$-sparse $n\times n$ diagonal matrix $\Gamma$ (Part~\ref{fac:query_time_improved_sparsity_Gamma} of Fact~\ref{fac:query_time_improved_sparsity}) with a $n\times 6n^a$ matrix $\L_c[M_{\partial S\backslash S}]$ that only has $n^{\wt{a}}$ non-zero columns (since $|\partial S\backslash S|\leq |\partial S|\leq n^{\wt{a}}$ by Part~\ref{fac:query_time_improved_sparsity_partial_S} of Fact~\ref{fac:query_time_improved_sparsity}) and then with a $6n^a\times 1$ vector $(\gamma^{\tmp}-\gamma_1-\partial \gamma)$. It takes $O(n^{a+\wt{a}})$ time.
    \item We multiply a $n^b\times n$ matrix $R[j]$ with a $n\times 1$ vector $(\Gamma (\L_c[M_{\partial S\backslash S}]-\L_c[M_{S'}])(\gamma^{\tmp}-\gamma_1-\partial \gamma))$. It takes $O(n^{1+b})$ time.
    \item We multiply a $n^b\times 6n^a$ matrix $F[j]$ with a $6n^a \times 1$ vector $(\gamma^{\tmp}-\gamma_1-\partial \gamma)$. It takes $O(n^{1+b})$ time.
    \item We multiply a $n^b\times 6n^a$ matrix $\L_c[(Q[j])_{S^{\new}}]$ with a $6n^a \times 1$ vector $(\gamma^{\tmp}-\gamma_1-\partial \gamma)$. It takes $O(n^{b+a})$ time.
\end{enumerate}
Thus this part takes $O(n^{\wt{a}+a} + n^{1+b} + n^{b+a}) = O(n^{\wt{a}+a} + n^{1+b})$ time, since $a\leq 1$.

Thus, the total running time to compute $r_4$ is $O(n^{\wt{a}+a} + n^{1+b})$.
\end{proof}

\begin{proof}[Proof of Lemma~\ref{lem:query_time_improved}]
Summing over the time to compute $r_2$, $r_3$, $(U', C, U)$, $\partial \gamma$, $r_{4,1}$, $r_{4,2}$, $r_{4,3}$, $U^{\tmp}$, and $r_{4,4}$, the total running time of computing $r$ is 
\begin{align*}
 O(n^{1+b} + n^{a+\wt{a}} + n^{a+b} + \Tmat(n^{\wt{a}},n^a,n^{\wt{a}})) 
=  O(n^{1+b} + \Tmat(n^{\wt{a}},n^a,n^{\wt{a}})),
\end{align*}
which follows by $a\leq 1$ and $n^{a+\wt{a}} \leq \Tmat(n^{\wt{a}},n^a,n^{\wt{a}}))$.
%\begin{align*}
%& ~ O(\sqrt{n} + \wt{p}n^a + n^{1+b} + \wt{p} n^b + \wt{k} n^a + n^{2a} + n^{a+b} + \Tmat (n^a, \wt{k}, n^a) + \wt{k}^{\omega})\\
%= & ~ O( \wt{p}n^a + n^{1+b} + \wt{p} n^b + \wt{k} n^a + n^{2a} + n^{a+b} + \Tmat (n^a, \wt{k}, n^a) + \wt{k}^{\omega})\\
%= & ~ O( n^{1+b} + \wt{k} n^a + n^{2a} + n^{a+b} + \Tmat (n^a, \wt{k}, n^a) + \wt{k}^{\omega})\\
%= & ~ O( n^{1+b} + \wt{k} n^a + n^{2a}  + \Tmat (n^a, \wt{k}, n^a) + \wt{k}^{\omega})\\
%= & ~ O(n^{2a} + n^{1+b} +  \Tmat (n^a, \wt{k}, n^a) ),
%\end{align*}
%where the first step follows from $\sqrt{n}$ is dominated by $n^{1+b}$, the second step follows from $\wt{p}\leq n^a$, the third step follows from $a \leq 1$, and last step follows from the fact $\wt{k}\leq n^a$ and $\Tmat (n^a, \wt{k}, n^a) \geq \wt{k}^{\omega}$ when $\wt{k} \leq n^a$.
\end{proof}

\subsection{Running time of \textsc{MatrixUpdate}}\label{sec:matrix_update_time_improved}

The goal of this section is to prove Lemma~\ref{lem:matrix_update_time_improved}.
We will use the following sparsity guarantees that are proved in Lemma~\ref{lem:improved:sparsity_guarantee}.
\begin{fact}[Sparsity guarantees for \textsc{MatrixUpdate}]
\label{fac:matrix_update_time_improved_sparsity}
When entering \textsc{MatrixUpdate} we have the following sparsity guarantee (from Table~\ref{tab:improved:sparsity_guarantee_part_2}):
\begin{multicols}{2}
\begin{enumerate}
\item \label{fac:matrix_update_time_improved_sparsity_S_partial_S} $|S^{\new}|\leq |S\cup \partial S|\leq 3k$,
\item \label{fac:matrix_update_time_improved_sparsity_Gamma_partial_Gamma} $\|\Gamma^{\new}\|_0=\|\Gamma+\partial \Gamma\|_0=k$.
\end{enumerate}
\end{multicols}
\end{fact}

\begin{lemma}[Running time of \textsc{MatrixUpdate}]
\label{lem:matrix_update_time_improved}
In the procedure \textsc{MatrixUpdate} (in Algorithm~\ref{alg:matrix_update_improved}) it takes 
\begin{enumerate}
\item $O(k^{\omega+o(1)}+ \Tmat(n,k,n))$ time to compute 
\begin{align*}
M^{\tmp} \leftarrow M - M_{S^{\new}} \cdot ( (\Delta_{S^{\new},S^{\new}}^{\new} )^{-1} + M_{S^{\new},S^{\new}} )^{-1} ( M_{S^{\new}} )^\top, 
\end{align*} 
\item $O(k^{\omega+o(1)} + \Tmat(n^{1+o(1)}, k , n) )$ time to compute 
\begin{align*}
Q^{\tmp} \leftarrow Q + R (\Gamma^{\new} M^{\tmp}) + R \sqrt{V} (M^{\tmp} -M),
\end{align*}
\item $O(n^{2+o(1)})$ time to compute \begin{align*}
\beta_1 \leftarrow Q^{\tmp} \sqrt{ W^{\appr} } f ( h^{\appr} ),
\end{align*}
\item $O(n^{2})$ time to compute \begin{align*}
\beta_2 \leftarrow M^{\tmp} \sqrt{ W^{\appr} } f ( h^{\appr} ).
\end{align*}
\end{enumerate}
Further, it takes $O(n)$ time to do all other computation. So the overall running time of the procedure \textsc{MatrixUpdate} is $O( \Tmat( n^{1+o(1)} , k , n ) )$.
\end{lemma}

\begin{claim}[Part 1 of Lemma~\ref{lem:matrix_update_time_improved}]
\label{cla:matrix_update_time_improved_M_tmp}
In procedure \textsc{MatrixUpdate}, it takes $O(k^{\omega+o(1)}+ \Tmat(n,k,n))$ time to compute
\[
M^{\tmp} \leftarrow M - M_{S^{\new}} \cdot ( (\Delta_{S^{\new},S^{\new}}^{\new} )^{-1} + M_{S^{\new},S^{\new}} )^{-1} ( M_{S^{\new}} )^\top. 
\]
\end{claim}
\begin{proof} The running time of computing $M^{\new}$ can be split into the following parts:
\begin{enumerate}
\item Computing the inverse of a $O(k)
\times O(k)$ matrix $( (\Delta_{S^{\new},S^{\new}}^{\new} )^{-1} + M_{S^{\new},S^{\new}})$ (the size follows from $|S^{\new}|=O(k)$, see Part~\ref{fac:matrix_update_time_improved_sparsity_S_partial_S} of Fact~\ref{fac:matrix_update_time_improved_sparsity}), this part takes $O(k^{\omega+o(1)})$ time.
\item Computing the multiplication of a $n \times O(k)$ matrix $M_{S^{\new}}$ with a $O(k) \times O(k)$ matrix $( (\Delta_{S^{\new},S^{\new}}^{\new} )^{-1} + M_{S^{\new},S^{\new}})^{-1}$, this part takes $O(\Tmat(n,k,k))$ time. 
\item Computing the multiplication of a $n \times O(k)$ matrix $M_{S^{\new}}((\Delta_{S^{\new},S^{\new}}^{\new} )^{-1} + M_{S^{\new},S^{\new}})^{-1}$ with a $O(k) \times n$ matrix $(M_{S^{\new}})^{\top}$, this part takes $O(\Tmat(n,k,n))$ time. 
\item Subtracting a $n\times n$ matrix $M_{S^{\new}} \cdot ( (\Delta_{S^{\new},S^{\new}}^{\new} )^{-1} + M_{S^{\new},S^{\new}} )^{-1} ( M_{S^{\new}} )^\top$ from the $n\times n$ matrix $M$, this part takes $O(n^2)$ time.
\end{enumerate}
So in total computing $M^{\tmp}$ takes time
\[
O(k^{\omega+o(1)}+ \Tmat(n,k,k) + \Tmat(n,k,n) + n^2)=O(k^{\omega+o(1)}+ \Tmat(n,k,n)).
\]
\end{proof}

\begin{claim}[Part 2 of Lemma~\ref{lem:matrix_update_time_improved}]
\label{cla:matrix_update_time_improved_Q_tmp}
In \textsc{MatrixUpdate}, it takes $O(k^{\omega+o(1)}+\Tmat( n^{1+o(1)} , k , n ))$ time to compute
\[
Q^{\tmp} \leftarrow Q + \underbrace{R (\Gamma^{\new} M^{\tmp})}_{N_1} + \underbrace{R \sqrt{V} (M^{\tmp} -M)}_{N_2}.
\]
\end{claim}

\begin{proof}
The running time of computing $Q^{\tmp}$ can be split into the following parts:

The first part is to compute the $n\times n$ matrix $N_1$:
\begin{enumerate}
\item Multiplying a $k$-sparse $n\times n$ diagonal matrix $\Gamma^{\new}$ (Part~\ref{fac:matrix_update_time_improved_sparsity_Gamma_partial_Gamma} of Fact~\ref{fac:matrix_update_time_improved_sparsity}) with a $n\times n$ matrix $M^{\tmp}$ takes $O(kn)$ time.
\item Multiplying a $n^{1+o(1)}\times n$ matrix $Q^{\tmp}$ with a $k$-row-sparse $n\times n$ matrix $(\Gamma^{\new} M^{\tmp})$ (Part~\ref{fac:matrix_update_time_improved_sparsity_S_partial_S} of Fact) takes $O(\Tmat(n^{1+o(1)} , k, n) )$ time.
\end{enumerate}
So in total computing $N_1$ takes $O(\Tmat(n^{1+o(1)} , k, n) )$ time.

The second part is to compute the $n\times n$ matrix $N_2$. By the definition of $M^{\tmp}$ we have
\begin{align*}
N_2 =  R \sqrt{V} (M^{\tmp} -M)
=  R \sqrt{V} M_{S^{\new}} \cdot \underbrace{( (\Delta_{S^{\new},S^{\new}}^{\new} )^{-1} + M_{S^{\new},S^{\new}} )^{-1} ( M_{S^{\new}} )^\top}_{N_3}.
\end{align*}
And we computes $N_2$ as follows:
\begin{enumerate}
\item Multiplying a $n^{1+o(1)}\times n$ matrix $R$ with a $n\times n$ diagonal matrix $\sqrt{V}$ takes $O(n^{2+o(1)})$ time.
\item Multiplying a $n^{1+o(1)}\times n$ matrix $R \sqrt{V}$ with a $n \times O(k)$ matrix $M_{S^{\new}}$ (Part~\ref{fac:matrix_update_time_improved_sparsity_S_partial_S} of Fact~\ref{fac:matrix_update_time_improved_sparsity}) takes $O(\Tmat(n^{1+o(1)}, n, k))$ time. 
\item The $O(k)\times n$ matrix $N_3$ is already computed when we were computing $M^{\tmp}$ (Claim~\ref{cla:matrix_update_time_improved_M_tmp}), and this takes $O(k^{\omega+o(1)} + \Tmat(k,k,n))$ time. 
\item Multiplying a $n^{1+o(1)}\times O(k)$ matrix $(R\sqrt{V} M_{S^{\new}} )$ with a $O(k) \times n$ matrix $N_3$ takes $O(\Tmat(n^{1+o(1)},k,n))$ time.
\end{enumerate} 
So in total computing $Q^{\tmp}$ takes time
\begin{align*}
O( n^{2+o(1)} + k^{\omega+o(1)} + \Tmat(k,k,n)+ \Tmat( n^{1+o(1)} , k , n ) ) = O(k^{\omega+o(1)} + \Tmat( n^{1+o(1)} , k , n )).
\end{align*}
\end{proof}

\begin{claim}[Part 3 of Lemma~\ref{lem:matrix_update_time_improved}]
\label{cla:matrix_update_time_improved_beta_1}
In \textsc{MatrixUpdate}, it takes $O(n^{2+o(1)})$ time to compute
\[
\beta_1 \leftarrow Q^{\tmp} \sqrt{ W^{\appr} } f ( h^{\appr} ).
\]
\end{claim}
\begin{proof}
To compute $\beta_1$, we first multiply a $n^{1+o(1)} \times n$ matrix $Q^{\tmp}$ with a $n\times n$ diagonal matrix $\sqrt{W^{\appr}}$, which takes $O(n^{2+o(1)})$ time. Then we multiply a $n^{1+o(1)} \times n$ matrix $Q^{\tmp}\sqrt{W^{\appr}}$ with a $n \times 1$ vector, which takes $O(n^{2+o(1)})$ time. So in total computing $\beta_1$ takes time $O(n^{2+o(1)})$.
\end{proof}

\begin{claim}[Part 4 of Lemma~\ref{lem:matrix_update_time_improved}]
\label{cla:matrix_update_time_improved_beta_2}
In procedure \textsc{MatrixUpdate}, it takes $O(n^{2})$ time to compute
\[
\beta_2 \leftarrow M^{\tmp} \sqrt{ W^{\appr} } f ( h^{\appr} ).
\]
\end{claim}
\begin{proof}
To compute $\beta_2$, we first multiply a $n \times n$ matrix $M^{\tmp}$ with a $n\times n$ diagonal matrix $\sqrt{W^{\appr}}$, which takes $O(n^{2})$ time. Then we multiply a $n^{1} \times n$ matrix $M^{\tmp}\sqrt{W^{\appr}}$ with a $n \times 1$ vector, which takes $O(n^{2})$ time. So in total computing $\beta_1$ takes time $O(n^{2})$.
\end{proof}

\begin{proof}[Proof of Lemma~\ref{lem:matrix_update_time_improved}]
\noindent \textbf{Overall time.}
The running time of procedure \textsc{MatrixUpdate} is 

$O(k^{\omega+o(1)} + \Tmat( n^{1+o(1)} , k , n )  + n^{2+o(1)}) = O( \Tmat( n^{1+o(1)} , k , n ) )$.
\end{proof}

\subsection{Running time of \textsc{PartialMatrixUpdate}}\label{sec:partial_matrix_update_time_improved}
The goal of this section is to prove Lemma~\ref{lem:partial_matrix_update_time_improved}.
We will use the following sparsity guarantees that are proved in Lemma~\ref{lem:improved:sparsity_guarantee}.

\begin{fact}[Sparsity guarantees for \textsc{PartialMatrixUpdate}]
\label{fac:partial_matrix_update_time_improved_sparsity}
When entering \textsc{PartialMatrixUpdate} (Algorithm~\ref{alg:partial_matrix_update_improved}) we have the following sparsity guarantee (from Table~\ref{tab:improved:sparsity_guarantee_part_1} and Table~\ref{tab:improved:sparsity_guarantee_part_2}):

\begin{multicols}{3}
\begin{enumerate}
\item \label{fac:partial_matrix_update_time_improved_sparsity_sqrt_V}
$\|\sqrt{W^{\appr}}-\sqrt{V}\|_0 \leq n^a$,
\item \label{fac:partial_matrix_update_time_improved_sparsity_partial_Gamma}
$\|\partial \Gamma\|_0 \leq \wt{k} \leq 2n^a$,
\item \label{fac:partial_matrix_update_time_improved_sparsity_Gamma_partial_Gamma}
$\|\Gamma+\partial \Gamma\|_0 \leq n^a$,
\item \label{fac:partial_matrix_update_time_improved_sparsity_f_g}
$\|f(\wt{g})-f(g)\|_0 \leq n^a$,
\item \label{fac:partial_matrix_update_time_improved_sparsity_xi}
$\|\xi\|_0 \leq n^a$,
\item \label{fac:partial_matrix_update_time_improved_sparsity_S}
$|S|\leq n^a$,
\item \label{fac:partial_matrix_update_time_improved_sparsity_partial_S}
$|\partial S|\leq \wt{k} \leq 2n^a$.
{\color{white} \item }
{\color{white} \item }
\end{enumerate}
\end{multicols}
\end{fact}

%Note that since $|S|\leq n^a$ and $|\partial S|\leq \wt{k} \leq 2n^a$, all the $\L$ operators are well-defined, and the input for procedure \textsc{Decompose} also satisfies the size constraint.

\begin{lemma}[Running time of \textsc{PartialMatrixUpdate}]
\label{lem:partial_matrix_update_time_improved}
In the procedure \textsc{PartialMatrixUpdate} (Algorithm~\ref{alg:partial_matrix_update_improved}) it takes 
\begin{enumerate}
\item $O(\wt{k}n^{a})$ time to compute
\begin{align*}
(U', C, U) \leftarrow \textsc{Decompose}\Big( \L_*[(\Delta^{\new}_{S^{\new}, S^{\new}})^{-1} +  M_{S^{\new}, S^{\new}} ] - \L_* [\Delta_{S, S}^{-1}+ M_{S,S}]\Big),
\end{align*}
\item $O(\Tmat(\wt{k}, n^a, n^a))$ time to compute 
\begin{align*}
B^{\tmp} \leftarrow B - BU'(C^{-1}+U^{\top}BU')^{-1}U^{\top}B,
\end{align*}
\item $O(n)$ time to compute
\begin{align*}
\xi^{\tmp} \leftarrow \sqrt{W^{\appr}}f(\wt{g}) - \sqrt{V}f(g),
\end{align*}
\item $O(n^{2a})$ time to compute 
\begin{align*}
\gamma_1^{\tmp} \leftarrow B^{\tmp} \cdot \L_r [\beta_{2,S^{\new}}] + B^{\tmp} \cdot \L_r [(M_{S^{\new}})^{\top}] \xi^{\tmp} ,
\end{align*}
\item $O(\wt{k}n^a)$ time to compute 
\begin{align*}
\gamma_2^{\tmp} \leftarrow \gamma_2 + (\Gamma + \partial \Gamma) M ( \sqrt{ W^{\appr} } - \sqrt{\wt{V}})f(\wt{g}) + \partial \Gamma M (\sqrt{\wt{V}}f(\wt{g}) - \sqrt{V}f(g)) .
\end{align*}
\item $O(\Tmat(n^{1+o(1)},n^a,\wt{k}))$ time to compute 
\begin{align*}
    F^{\tmp} \leftarrow F +R \Gamma \cdot (\L_c[M_{\partial S\backslash S}]-\L_c[M_{S'}]) + R\partial \Gamma\cdot \L_c[M_{S^{\new}}].
\end{align*}
\item $O(\Tmat(n,n^a,\wt{k}))$ time to compute 
\begin{align*}
    E^{\tmp} = E + B^{\tmp}(\L_r[ (M_{\partial S\backslash S})^{\top}]-\L_r[ (M_{S'})^{\top}]) - BU'(C^{-1}+U^{\top}BU')^{-1}U^{\top}E.
\end{align*}
\end{enumerate}
Further, it takes $O(n)$ time to do all other computation. So the overall running time of the procedure \textsc{PartialMatrixUpdate} is $O(\Tmat(n^{1+o(1)},n^a,\wt{k}))$.
\end{lemma}

\begin{claim}[Part 1 of Lemma~\ref{lem:partial_matrix_update_time_improved}]
\label{cla:partial_matrix_update_time_improved_UCU}
In the procedure \textsc{PartialMatrixUpdate} (Algorithm~\ref{alg:partial_matrix_update_improved}), it takes $O(\wt{k} n^a)$ time to compute 
\begin{align*}
(U', C, U) \leftarrow \textsc{Decompose}\Big(\L_*[(\Delta^{\new}_{S^{\new}, S^{\new}})^{-1} +  M_{S^{\new}, S^{\new}} ] - \L_* [\Delta_{S, S}^{-1}+ M_{S,S}]\Big).
\end{align*}
And the size of the computed matrices are $U',U \in \R^{6n^a\times c}$, $C\in \R^{c\times c}$, where $c= O(\wt{k})$.
\end{claim}
\begin{proof}
The analysis for the \textsc{Decompose} function in \textsc{PartialMatrixUpdate} is the same as the analysis in \textsc{Query} (Claim~\ref{cla:query_time_improved_UCU}). We still have $|S|\leq n^a$ (Part~\ref{fac:partial_matrix_update_time_improved_sparsity_S} of Fact~\ref{fac:partial_matrix_update_time_improved_sparsity}), but the bound on $\partial S$ is different and now we have $|\partial S| = \wt{k}\leq 2n^a$ (Part~\ref{fac:partial_matrix_update_time_improved_sparsity_partial_S} of Fact~\ref{fac:partial_matrix_update_time_improved_sparsity}). So $c=O(|\partial S|)=O(\wt{k})$. 

And to compute $(U', C, U)$ it takes $O(\wt{k} n^a)$ time.
\end{proof}

\begin{claim}[Part 2 of Lemma~\ref{lem:partial_matrix_update_time_improved}]
\label{cla:partial_matrix_update_time_improved_B_tmp}
In the procedure \textsc{PartialMatrixUpdate} (Algorithm~\ref{alg:partial_matrix_update_improved}), it takes $O(\Tmat(\wt{k}, n^a, n^a))$ time to compute
\begin{align*}
B^{\tmp} \leftarrow B - BU'\underbrace{(C^{-1}+U^{\top}BU')^{-1}}_{N}U^{\top}B.
\end{align*}
\end{claim}
\begin{proof}
The running time of computing $B^{\tmp} \in \R^{6n^a \times 6n^a}$ can be split into the following parts.

The first part is to compute $N\in \R^{c\times c}$ as follows:
\begin{enumerate}
\item Multiplying a $c\times 6n^a$ matrix $U^{\top}$ with a $6n^a\times 6n^a$ matrix $B$ takes $O(\Tmat(c, n^a, n^a))$ time.
\item Multiplying a $c\times 6n^a$ matrix $(U^{\top}B)$ with a $6n^a\times c$ matrix $U'$ takes $O(\Tmat(c, n^a, c))$ time.
\item Computing the inverse of a $c \times c$ matrix $(C^{-1}+U^{\top}BU')$ takes $O(c^{\omega+o(1)})$ time.
\end{enumerate}
So the total running time to compute $N$ is $O(\Tmat(c, n^a, n^a) + \Tmat(c, n^a, c) + c^{\omega+o(1)}) = O(\Tmat(\wt{k}, n^a, n^a))$ (since $c=O(\wt{k})\leq O(n^a)$ by Claim~\ref{cla:partial_matrix_update_time_improved_UCU}), and $\wt{k}\leq 2n^a$.

The second part is to compute $BU'NU^{\top}B\in \R^{6n^a\times 6n^a}$ as follows:
\begin{enumerate}
\item Multiplying a $6n^a\times 6n^a$ matrix $B$ with a $6n^a\times c$ matrix $U'$ takes $O(\Tmat(n^a, n^a, c))$ time.
\item Multiplying a $c \times 6n^a$ matrix $U^\top$ with a $6n^a\times 6n^a$ matrix $B$ takes $O(\Tmat(c, n^a, n^a))$ time.
\item Multiplying a $6n^a \times c$ matrix $BU'$ with a $c \times c$ matrix $N$ takes $O(\Tmat(n^a, c, c))$ time.
\item Multiplying a $6n^a\times c$ matrix $BU'N$ with a $c\times 6n^a$ matrix $U^{\top}B$ takes $O(\Tmat(n^a, c, n^a))$ time.
\end{enumerate}
So the total running time to compute $BU'NU^{\top}B$ is $O(\Tmat(n^a, n^a, c) + \Tmat(n^a, c, c)) = O(\Tmat(n^a, n^a, \wt{k}))$ (since $c=O(\wt{k})\leq O(n^a)$ by Claim~\ref{cla:partial_matrix_update_time_improved_UCU}), and $\wt{k}\leq 2n^a$.

Finally, subtracting $BU'NU^{\top}B$ from $B$ takes $O(n^{2a})$ time since both matrices has size $6n^a\times 6n^a$.

So in total computing $B^{\tmp} \in \R^{6n^a \times 6n^a} $ takes time $O ( \Tmat(\wt{k}, n^a, n^a) + n^{2a} ) = O ( \Tmat(\wt{k}, n^a, n^a) )$.
\end{proof}

\begin{claim}[Part 3 of Lemma~\ref{lem:partial_matrix_update_time_improved}]
\label{cla:partial_matrix_update_time_improved_xi_tmp}
In the procedure \textsc{PartialMatrixUpdate} (Algorithm~\ref{alg:partial_matrix_update_improved}), it takes $O(n)$ time to compute $\xi^{\tmp} \leftarrow \sqrt{W^{\appr}}f(\wt{g}) - \sqrt{V}f(g)$. And the computed vector $\xi^{\tmp}\in \R^{n\times 1}$ is $n^a$-sparse.
\end{claim}
\begin{proof}
The $O(n)$ running time follows trivially from the fact that $\sqrt{W^{\appr}}$ and $\sqrt{V}$ are $n\times n$ diagonal matrices, and $f(\wt{g})$ and $f(g)$ are $n\times 1$ vectors. We also have
\begin{align*}
\|\xi^{\tmp}\|_0 = \|\sqrt{W^{\appr}}f(\wt{g}) - \sqrt{V}f(g)\|_0 
=  \max \{\|\sqrt{W^{\appr}}-\sqrt{V}\|_0, \|f(\wt{g})-f(g)\|_0\}
\leq  n^a,
\end{align*}
which follows from Part~\ref{fac:partial_matrix_update_time_improved_sparsity_sqrt_V} and \ref{fac:partial_matrix_update_time_improved_sparsity_f_g} of Fact~\ref{fac:partial_matrix_update_time_improved_sparsity}.
\end{proof}

\begin{claim}[Part 4 of Lemma~\ref{lem:partial_matrix_update_time_improved}]
\label{cla:partial_matrix_update_time_improved_gamma_1_tmp}
In the procedure \textsc{PartialMatrixUpdate} (Algorithm~\ref{alg:partial_matrix_update_improved}), it takes $O(n^{2a})$ time to compute $\gamma_1^{\tmp} \leftarrow B^{\tmp} \cdot \L_r [\beta_{2,S^{\new}}] + B^{\tmp} \cdot \L_r [(M_{S^{\new}})^{\top}] \cdot \xi^{\tmp}$.
\end{claim}
\begin{proof}
The running time of computing $\gamma_1^{\tmp}$ can be split into the following parts:
\begin{enumerate}
\item Multiplying a $6n^a\times 6n^a$ matrix $B^{\tmp}$ with a $6n^a \times 1$ vector $\L_r[\beta_{2,S^{\new}}]$ takes $O(n^{2a})$ time.
\item Multiplying a $6n^a \times n$ matrix $\L_r[(M_{S^{\new}})^{\top}]$ with a $n^a$-sparse $n \times 1$ vector $\xi^{\tmp}$ (Claim~\ref{cla:partial_matrix_update_time_improved_xi_tmp}) takes $O(n^{2a})$ time.
\item Multiplying a $6n^a\times 6n^a$ matrix $B^{\tmp}$ with a $6n^a \times 1$ vector $\L_r[(M_{S^{\new}})^{\top}]\xi^{\tmp}$ takes $O(n^{2a})$ time.
\end{enumerate}
So in total computing $\gamma_1^{\tmp}$ takes $O(n^{2a})$ time.
\end{proof}

\begin{claim}[Part 5 of Lemma~\ref{lem:partial_matrix_update_time_improved}]
\label{cla:partial_matrix_update_time_improved_gamma_2_tmp}
In the procedure \textsc{PartialMatrixUpdate} (Algorithm~\ref{alg:partial_matrix_update_improved}), it takes $O(\wt{k}n^a)$ time to compute
\begin{align*}
\gamma_2^{\tmp} \leftarrow \gamma_2 + (\Gamma + \partial \Gamma) M \underbrace{( \sqrt{ W^{\appr} } - \sqrt{\wt{V}})f(\wt{g})}_{a_1} + \partial \Gamma M \underbrace{(\sqrt{\wt{V}}f(\wt{g}) - \sqrt{V}f(g))}_{a_2}.
\end{align*}
\end{claim}
\begin{proof}
First note that the $n\times 1$ vectors $a_1$ and $a_2$ can both be computed in $O(n)$ time. Also,
\begin{align*}
\|a_1\|_0 =  \|( \sqrt{ W^{\appr} } - \sqrt{\wt{V}})f(\wt{g})\|_0
\leq  \min\{ \| \sqrt{ W^{\appr} } - \sqrt{\wt{V}} \|_0 , \| f(\wt{g}) \|_0 \} \leq \wt{k},
\end{align*}
where the last step follows from Part~\ref{fac:partial_matrix_update_time_improved_sparsity_sqrt_V} of Fact~\ref{fac:partial_matrix_update_time_improved_sparsity}. And
\begin{align*}
\|a_2\|_0 =  \|\sqrt{\wt{V}}f(\wt{g}) - \sqrt{V}f(g)\|_0
=  \|\xi\|_0 \leq n^a,
\end{align*}
where the last step follows from Part~\ref{fac:partial_matrix_update_time_improved_sparsity_xi} of Fact~\ref{fac:partial_matrix_update_time_improved_sparsity}.

The running time of computing $\gamma_2^{\tmp}$ can be split into the following parts:
\begin{enumerate}
\item Multiplying a $n^a$-sparse $n \times n$ diagonal matrix $(\Gamma + \partial \Gamma)$ (Part~\ref{fac:partial_matrix_update_time_improved_sparsity_Gamma_partial_Gamma} of Fact~\ref{fac:partial_matrix_update_time_improved_sparsity}) with a $n\times n$ matrix $M$ and then with a $\wt{k}$-sparse $n \times 1$ vector $a_1$ takes $O(\wt{k}n^a)$ time.
\item Multiplying a $\wt{k}$-sparse $n \times n$ diagonal matrix $\partial \Gamma$ (Part~\ref{fac:partial_matrix_update_time_improved_sparsity_partial_Gamma} of Fact~\ref{fac:partial_matrix_update_time_improved_sparsity}) with a $n\times n$ matrix $M$ and then with a $n^a$-sparse $n \times 1$ vector $a_2$ takes $O(\wt{k}n^a)$ time.
\end{enumerate}
So in total computing $\gamma_2^{\tmp}$ takes time $O(\wt{k}n^a)$.
\end{proof}

\begin{claim}[Part 6 of Lemma~\ref{lem:partial_matrix_update_time_improved}]
\label{cla:partial_matrix_update_time_improved_F_tmp}
In the procedure \textsc{PartialMatrixUpdate} (Algorithm~\ref{alg:partial_matrix_update_improved}), it takes $O(\Tmat(n^{1+o(1)},n^a,\wt{k}))$ time to compute
\begin{align*}
F^{\tmp} \leftarrow F +R \Gamma \cdot (\L_c[M_{\partial S\backslash S}]-\L_c[M_{S'}]) + R\partial \Gamma\cdot \L_c[M_{S^{\new}}].
\end{align*}
\end{claim}
\begin{proof}
By sparsity guarantee(Part~\ref{fac:partial_matrix_update_time_improved_sparsity_partial_S}, \ref{fac:partial_matrix_update_time_improved_sparsity_S} of Fact~\ref{fac:partial_matrix_update_time_improved_sparsity}), we have $|\partial S|+|S'|\leq |\partial S| \leq \wt{k} \leq 2n^a$.
%$|S|\leq n^a$. So $|S^{\new}|\leq |S\cup \partial S|\leq |S|+|\partial S| \leq 3n^a$.
The running time of computing $F^{\tmp}$ can be split into the following parts:
\begin{enumerate}
\item Multiplying a $n^{1+o(1)}\times n$ matrix $R$ with a $n^a$-sparse $n\times n$ diagonal matrix $\Gamma$ and then with a $\wt{k}$-column-sparse $n\times 6n^{a}$ matrix $\L_c[M_{\partial S \backslash S}]-\L_c[M_{S'}]$ takes $O(\Tmat(n^{1+o(1)},n^a,\wt{k}))$ time.
\item Multiplying a $n^{1+o(1)}\times n$ matrix $R$ with a $\wt{k}$-sparse $n\times n$ diagonal matrix $\partial \Gamma$ and then with a $n \times 6n^a$ matrix $\L_c[M_{S^{\new}}]$ takes $O(\Tmat(n^{1+o(1)},\wt{k},n^a)$ time.
\item Adding two matrices $R\Gamma (\L_c[M_{\partial S \backslash S}]-\L_c[M_{S'}])$,  $R\partial \Gamma \L_c[M_{S^{\new}}]$ on the current stored matrix $F$ takes $O(n^{1+o(1)+a})$ time.
\end{enumerate}
So in total computing $F^{\tmp}$ takes $O(\Tmat(n^{1+o(1)},n^a,\wt{k}) + n^{1+o(1)+a})
= O(\Tmat(n^{1+o(1)},n^a,\wt{k}))$ time since $O(n^{1+o(1)+a}) \leq O(\Tmat(n^{1+o(1)},n^a,\wt{k}))$.
%\begin{align*}
%O(\Tmat(n^{1+o(1)},n^a,\wt{k}) + \Tmat(n^{1+o(1)},\wt{k},n^a) + n^{1+o(1)+a})
%= & ~ O(\Tmat(n^{1+o(1)},n^a,\wt{k}) + n^{1+o(1)+a})\\
%= & ~ O(\Tmat(n^{1+o(1)},n^a,\wt{k})),
%\end{align*}
%where the first step follows from Lemma~\ref{lem:equivent_of_matrix_multiplcation} that $O(\Tmat(n^{1+o(1)},n^a,\wt{k}))=O(\Tmat(n^{1+o(1)},\wt{k},n^a))$, the second step is from $O(n^{1+o(1)+a}) \leq O(\Tmat(n^{1+o(1)},n^a,\wt{k}))$.
\end{proof}

\begin{claim}[Part 7 of Lemma~\ref{lem:partial_matrix_update_time_improved}]
\label{cla:partial_matrix_update_time_improved_E_tmp}
In the procedure \textsc{PartialMatrixUpdate} (Algorithm~\ref{alg:partial_matrix_update_improved}), it takes $O(\Tmat(n,n^a,\wt{k}))$ time to compute
\begin{align*}
E^{\tmp} \leftarrow E + \underbrace{B^{\tmp}(\L_r[ (M_{\partial S\backslash S})^{\top}]-\L_r[ (M_{S'})^{\top}])}_{N_2} - \underbrace{BU'\underbrace{(C^{-1}+U^{\top}BU')^{-1}}_{N}U^{\top}E}_{N_1}
\end{align*}
\end{claim}
\begin{proof}
By sparsity guarantee(Part~\ref{fac:partial_matrix_update_time_improved_sparsity_partial_S}, \ref{fac:partial_matrix_update_time_improved_sparsity_S} of Fact~\ref{fac:partial_matrix_update_time_improved_sparsity}), we have $|\partial S| \leq \wt{k} \leq 2n^a$, $|S|\leq n^a$. So $|S^{\new}|\leq |S\cup \partial S|\leq |S|+|\partial S| \leq 3n^a$.
The running time of computing $F^{\tmp}$ can be split into the following parts:

First we compute $N_1\in \R^{6n^a\times n}$.

\begin{enumerate}
\item We already compute the time of computing $N$ in Claim~\ref{cla:partial_matrix_update_time_improved_B_tmp}. It takes $O(\Tmat(c, n^a, n^a))$ time.
\item Multiplying a $6n^a\times 6n^a$ matrix $B$ with a $6n^a\times c$ matrix $U'$ takes $O(\Tmat(n^a, n^a, c))$ time.
\item Multiplying a $6n^a \times c$ matrix $BU'$ with a $c \times c$ matrix $N$ takes $O(\Tmat(n^a, c, c))$ time.
\item Multiplying a $c \times 6n^a$ matrix $U^{\top}$ with a $6n^a \times n$ matrix $E$ takes $O(\Tmat(c,n^a,n))$ time.
\item Multiplying a $6n^a \times c$ matrix $BU'N$ with a $c\times n$ matrix $U^{\top}E$ takes time $O(\Tmat(n^a,c,n))$.
\end{enumerate}

Next, we compute $N_2\in \R^{6n^a\times n}$. Multiplying a $6n^a \times 6n^a$ matrix $B^{\tmp}$ with a $\wt{k}$-row-sparse $6n^a\times n$ matrix $(\L_r[ (M_{\partial S\backslash S})^{\top}]-\L_r[ (M_{S'})^{\top}])$ takes $O(\Tmat(n^a,\wt{k},n)$ time.

So in total computing $E^{\tmp}$ takes time
\begin{align*}
& ~ O(\Tmat(c, n^a, n^a) + \Tmat(n^a, n^a, c) + \Tmat(n^a, c, c) + \Tmat(c,n^a,n) + \Tmat(n^a,c,n))\\
= & ~ O(\Tmat(n,n^a,c))
= ~ O(\Tmat(n,n^a,\wt{k})),
\end{align*}
where the first step follows since $c=O(\wt{k})\leq O(n^a)$ by Claim~\ref{cla:partial_matrix_update_time_improved_UCU}), and $\wt{k}\leq 2n^a$.
\end{proof}

\begin{proof}[Proof of Lemma~\ref{lem:partial_matrix_update_time_improved}]
\textbf{Overall time.} So in total the running time of procedure \textsc{PartialMatrixUpdate} (Algorithm~\ref{alg:partial_matrix_update_improved}) is 
\begin{align*}
O(\wt{k}n^a + \Tmat(\wt{k},n^a,n^a) + n + n^{2a} + \Tmat(n,n^a,\wt{k})) = O(\Tmat(n,n^a,\wt{k})).
\end{align*}
\end{proof}

\subsection{Running time of \textsc{VectorUpdate}}\label{sec:vector_update_time_improved}

The goal of this section is to prove Lemma~\ref{lem:vector_update_time_improved}. We will use the following sparsity guarantees that are proved in Lemma~\ref{lem:improved:sparsity_guarantee}.
\begin{fact}[Sparsity guarantees for \textsc{VectorUpdate}]
\label{fac:vector_update_time_improved_sparsity}
When entering \textsc{VectorUpdate} (Algorithm~\ref{alg:vector_update_improved}) we have the following sparsity guarantee (from Table~\ref{tab:improved:sparsity_guarantee_part_1} and Table~\ref{tab:improved:sparsity_guarantee_part_2}):
\begin{multicols}{2}
\begin{enumerate}
\item \label{fac:vector_update_time_improved_sparsity_Gamma}
$\|\Gamma\|_0 = \|\sqrt{\wt{V}}-\sqrt{V}\|_0\leq n^a$,
\item \label{fac:vector_update_time_improved_sparsity_f_g}
$\|f(h^{\appr}) - f(g)\|_0 = p$.
\end{enumerate}
\end{multicols}
\end{fact}

\begin{lemma}[Running time of \textsc{VectorUpdate}]\label{lem:vector_update_time_improved}
The procedure \textsc{VectorUpdate} (Algorithm~\ref{alg:vector_update_improved}) takes
\begin{enumerate}
    \item $O(pn^{1+o(1)})$ time to compute 
        \begin{align*}
        \beta_1^{\tmp} \leftarrow \beta_1 + Q\sqrt{V}(f(h^{\appr})-f(g)),
        ~~~ \beta_2^{\tmp} \leftarrow \beta_2 + M\sqrt{V}(f(h^{\appr})-f(g)),
        \end{align*}
    \item $O(n)$ time to compute
        \begin{align*}
        \xi^{\tmp} \leftarrow (\sqrt{\wt{V}}-\sqrt{V})f(h^{\appr}),
        \end{align*}
    \item $O(n^{2a})$ time to compute
        \begin{align*}
        \gamma_1^{\tmp} \leftarrow B \cdot \L_r [\beta_{2,S}^{\tmp}] + B \cdot \L_r [(M_{S})^{\top}] \cdot \xi^{\tmp},
        \end{align*}
    \item $O(n^{2a})$ time to compute 
        \begin{align*}
        \gamma_2^{\tmp} \leftarrow \Gamma M \cdot \xi^{\tmp}.
        \end{align*}
\end{enumerate}
Overall, the running time of procedure \textsc{VectorUpdate} is
\begin{align*}
    O( n^{2a} + pn^{1+o(1)}).
\end{align*}
\end{lemma}

\begin{claim}[Part 1 of Lemma~\ref{lem:vector_update_time_improved}]
\label{cla:vector_update_time_improved_beta_1_beta_2}
In the procedure \textsc{VectorUpdate} (Algorithm~\ref{alg:vector_update_improved}), it takes $O(pn^{1+o(1)})$ time to compute
\begin{align*}
    \beta_1^{\tmp} \leftarrow \beta_1 + Q\sqrt{V}(f(h^{\appr})-f(g)),
    ~~~ \beta_2^{\tmp} \leftarrow \beta_2 + M\sqrt{V}(f(h^{\appr})-f(g)).
\end{align*}
\end{claim}
\begin{proof}
We compute $\beta_1^{\tmp}$ as follows.
\begin{enumerate}
\item Multiplying a $n\times n$ diagonal matrix $\sqrt{V}$ with a $p$-sparse vector $f(h^{\appr})-f(g)$ (Part~\ref{fac:vector_update_time_improved_sparsity_f_g} of Fact~\ref{fac:vector_update_time_improved_sparsity}) takes $O(p)$ time, and the resulting vector $\sqrt{V}(f(h^{\appr})-f(g))$ is also $p$-sparse.
\item Multiplying a $n^{1+o(1)} \times n$ matrix $Q$ with a $p$-sparse vector $\sqrt{V}(f(h^{\appr})-f(g))$ takes $O(pn^{1+o(1)})$ time.
\end{enumerate}
So the total running time to compute $\beta_1^{\tmp}$ is $O(pn^{1+o(1)})$.

Following the same reason and note that $M$ has size $n\times n$, we can compute $\beta_2^{\tmp}$ in $O(pn)\leq O(pn^{1+o(1)})$ time.
\end{proof}

\begin{claim}[Part 2 of Lemma~\ref{lem:vector_update_time_improved}]
\label{cla:vector_update_time_improved_xi_tmp}
In the procedure \textsc{VectorUpdate} (Algorithm~\ref{alg:vector_update_improved}), it takes $O(n)$ time to compute $\xi^{\tmp} \leftarrow (\sqrt{\wt{V}}-\sqrt{V})f(h^{\appr})$.
And the computed $n\times 1$ vector $\xi^{\tmp}$ is $n^a$-sparse.
\end{claim}
\begin{proof}
$\xi^{\tmp}$ is computed by multiplying a $n\times n$ diagonal matrix $(\sqrt{\wt{V}} - \sqrt{V})$ with a $n\times 1$ vector $f(h^{\appr})$, and this takes $O(n)$ time. We also have 
\begin{align*}
    \|\xi^{\tmp}\|_0 =  \|(\sqrt{\wt{V}}-\sqrt{V})f(h^{\appr})\|_0
    \leq \min\{\|\sqrt{\wt{V}}-\sqrt{V}\|_0, \|f(h^{\appr})\|_0\} \leq n^a,
\end{align*}
where the last step follows from Part~\ref{fac:vector_update_time_improved_sparsity_Gamma} of Fact~\ref{fac:vector_update_time_improved_sparsity}.
\end{proof}

\begin{claim}[Part 3 of Lemma~\ref{lem:vector_update_time_improved}]
\label{cla:vector_update_time_improved_gamma_1_tmp}
In the procedure \textsc{VectorUpdate} (Algorithm~\ref{alg:vector_update_improved}), it takes $O(n^{2a})$ time to compute $\gamma_1^{\tmp} \leftarrow B \cdot \L_r [\beta_{2,S}^{\tmp}] + B \cdot \L_r [(M_{S})^{\top}] \cdot \xi^{\tmp}$.
\end{claim}
\begin{proof}
The running time of computing $\gamma_1^{\tmp}$ can be split into the following parts:
\begin{enumerate}
\item Multiplying a $6n^a \times 6n^a$ matrix $B$ with a $6n^a \times 1$ vector $\L_r[\beta_{2,S}^{\tmp}]$ takes $O(n^{2a})$ time.
\item Multiplying a $6n^a \times n$ matrix $\L_r [(M_{S})^{\top}]$ with a $n^a$-sparse $n \times 1$ vector $\xi^{\tmp}$ (Claim~\ref{cla:vector_update_time_improved_xi_tmp}) takes $O(n^{2a})$ time.
\item Multiplying a $6n^a \times 6n^a $ matrix $B$ with a $6n^a \times 1$ vector $\L_r [(M_{S})^{\top}] \xi^{\tmp}$ takes $O(n^{2a})$ time.
\end{enumerate}
The total running time to compute $\gamma_1^{\tmp}$ is $O(n^{2a})$.
\end{proof}

\begin{claim}[Part 4 of Lemma~\ref{lem:vector_update_time_improved}]
\label{cla:vector_update_time_improved_gamma_2_tmp}
In the procedure \textsc{VectorUpdate} (Algorithm~\ref{alg:vector_update_improved}), it takes $O(n^{2a})$ time to compute $\gamma_2^{\tmp} \leftarrow \Gamma M \cdot \xi^{\tmp}$.
\end{claim}
\begin{proof}
We compute $\gamma_2^{\tmp} \in \R^n$ by multiplying a $n^a$-sparse $n \times n$ diagonal matrix $\Gamma$ (Part~\ref{fac:vector_update_time_improved_sparsity_Gamma} of Fact~\ref{fac:vector_update_time_improved_sparsity}) with a $n \times n$ matrix $M$ and then with a $n^a$-sparse $n \times 1$ vector $\xi^{\tmp}$ ((Claim~\ref{cla:vector_update_time_improved_xi_tmp})), which takes $O(n^{2a})$ time.
\end{proof}

\begin{proof}[Proof of Lemma~\ref{lem:vector_update_time_improved}]
Overall, the total running time to compute $\gamma_1^{\tmp}$ is $O ( pn^{1+o(1)} + n^{2a} ) = O( pn^{1+o(1)})$, since we always have $p\geq n^a$.
\end{proof}

\subsection{Running time of \textsc{PartialVectorUpdate}}\label{sec:partial_vector_update_time_improved}
The goal of this section is to prove Lemma~\ref{lem:partial_vector_update_time_improved}. We will use the following sparsity guarantees that are proved in Lemma~\ref{lem:improved:sparsity_guarantee}.
\begin{fact}[Sparsity guarantees for \textsc{PartialVectorUpdate}]
\label{fac:partial_vector_update_time_improved_sparsity}
When entering \textsc{PartialVectorUpdate} (Algorithm~\ref{alg:partial_vector_update_improved}) we have the following sparsity guarantee (from Table~\ref{tab:improved:sparsity_guarantee_part_1} and Table~\ref{tab:improved:sparsity_guarantee_part_2}):
\begin{multicols}{2}
\begin{enumerate}
\item \label{fac:partial_vector_update_time_improved_sparsity_Gamma}
$\|\Gamma\|_0\leq n^a$,
\item \label{fac:partial_vector_update_time_improved_sparsity_f_g}
$\|f(h^{\appr}) - f(\wt{g})\|_0=\wt{p}\leq 2n^a$.
\end{enumerate}
\end{multicols}
\end{fact}

\begin{lemma}[Running time of \textsc{PartialVectorUpdate}]\label{lem:partial_vector_update_time_improved}
In the procedure \textsc{PartialVectorUpdate} (in Algorithm~\ref{alg:partial_vector_update_improved}) it takes
\begin{enumerate}
\item $O(n)$ time to compute
\begin{align*}
    \xi^{\tmp} \leftarrow \sqrt{\wt{V}}f(h^{\appr}) - \sqrt{V}f(g),
\end{align*}
\item $O(n^{2a}+\wt{p}n^a)$ time to compute 
\begin{align*}
    \gamma_1^{\tmp} \leftarrow \gamma_1 + B \cdot \L_r [(M_{S})^{\top}] \cdot \sqrt{\wt{V}}\big( f(h^{\appr})-f(\wt{g})\big),
\end{align*}
\item $O(\wt{p}n^a)$ time to compute 
\begin{align*}
    \gamma_2^{\tmp} \leftarrow \gamma_2 + \Gamma M \sqrt{\wt{V}}\big( f(h^{\appr}) - f(\wt{g})\big).
\end{align*}

\end{enumerate}
Overall, the procedure \textsc{PartialVectorUpdate} takes $O( n^{2a} + \wt{p} n^a )$ time.
\end{lemma}

\begin{claim}[Part 1 of Lemma~\ref{lem:partial_vector_update_time_improved}]
\label{cla:partial_vector_update_time_improved_xi_tmp}
In the procedure \textsc{PartialVectorUpdate} (Algorithm~\ref{alg:partial_vector_update_improved}), it takes $O(n)$ time to compute $\xi^{\tmp} \leftarrow \sqrt{\wt{V}}f(h^{\appr})-\sqrt{V}f(g)$.
\end{claim}
\begin{proof}
$\xi^{\tmp}$ is computed by multiplying the $n\times n$ diagonal matrices $\sqrt{\wt{V}}$ and $\sqrt{V}$ with the $n\times 1$ vectors $f(h^{\appr})$ and $f(g)$ respectively, and this takes $O(n)$ time.
\end{proof}

\begin{claim}[Part 2 of Lemma~\ref{lem:partial_vector_update_time_improved}]
\label{cla:partial_vector_update_time_improved_gamma_1_tmp}
In the procedure \textsc{PartialVectorUpdate} (Algorithm~\ref{alg:partial_vector_update_improved}), it takes $O(n^{2a}+\wt{p}n^a)$ time to compute $\gamma_1^{\tmp} \leftarrow \gamma_1 + B \cdot \L_r [(M_{S})^{\top}] \cdot \sqrt{\wt{V}}\big( f(h^{\appr})-f(\wt{g})\big)$.
\end{claim}
\begin{proof}
The running time of computing $\gamma_1^{\tmp}$ can be split into the following parts:
\begin{enumerate}
\item Multiplying a $n\times n$ diagonal matrix $\sqrt{\wt{V}}$ with a $\wt{p}$-sparse $n \times 1$ vector $(f(h^{\appr})-f(\wt{g}))$ (Part~\ref{fac:partial_vector_update_time_improved_sparsity_f_g} of Fact~\ref{fac:partial_vector_update_time_improved_sparsity}) takes $O(\wt{p})$ time, and the resulting vector $\sqrt{\wt{V}}(f(h^{\appr})-f(\wt{g}))$ is also $\wt{p}$-sparse.
\item Multiplying a $6n^a \times n$ matrix $\L_r[(M_S)^{\top}]$ with a $\wt{p}$-sparse $n \times 1$ vector $ \sqrt{\wt{V}}(f(h^{\appr}) - f(\wt{g}))$ takes $O(\wt{p}n^a)$ time.
\item Multiplying a $6n^a\times 6n^a$ matrix $B$ with a $6n^a \times 1$ vector $\L_r[(M_S)^{\top}]\sqrt{\wt{V}}(f(h^{\appr}) - f(\wt{g}))$ takes $O(n^{2a})$ time.
\item Adding two $n\times 1$ vectors together takes $O(n)$ time.
\end{enumerate}
So in total the running time to compute $\gamma_1^{\tmp}$ is $O(n^{2a}+\wt{p}n^a)$.
\end{proof}

\begin{claim}[Part 3 of Lemma~\ref{lem:partial_vector_update_time_improved}]
\label{cla:partial_vector_update_time_improved_gamma_2_tmp}
In the procedure \textsc{PartialVectorUpdate} (Algorithm~\ref{alg:partial_vector_update_improved}), it takes $O(\wt{p}n^a)$ time to compute $\gamma_2^{\tmp} \leftarrow \gamma_2 + \Gamma M \sqrt{\wt{V}}\big( f(h^{\appr}) - f(\wt{g})\big)$.
\end{claim}
\begin{proof}
The running time of computing $\gamma_2^{\tmp}$ can be split into the following parts:
\begin{enumerate}
\item Multiplying a $n\times n$ diagonal matrix $\sqrt{\wt{V}}$ with a $\wt{p}$-sparse $n \times 1$ vector $(f(h^{\appr})-f(\wt{g}))$ (Part~\ref{fac:partial_vector_update_time_improved_sparsity_f_g} of Fact~\ref{fac:partial_vector_update_time_improved_sparsity}) takes $O(\wt{p})$ time, and the resulting vector $\sqrt{\wt{V}}(f(h^{\appr})-f(\wt{g}))$ is also $\wt{p}$-sparse.
\item Multiplying a $n^a$-sparse $n \times n$ diagonal matrix $\Gamma$ (Part~\ref{fac:partial_vector_update_time_improved_sparsity_Gamma} of Fact~\ref{fac:partial_vector_update_time_improved_sparsity}) with a $n \times n$ matrix $M$ and then with a $\wt{p}$-sparse $n \times 1$ vector $\sqrt{\wt{V}}(f(h^{\appr})-f(\wt{g}))$ takes $O( \wt{p} n^a )$ time.
\item Adding two $n\times 1$ vectors together takes $O(n)$ time.
\end{enumerate}
So in total the running time to compute $\gamma_2^{\tmp}$ is $O(\wt{p}n^a)$.
\end{proof}

\begin{proof}[Proof of Lemma~\ref{lem:partial_vector_update_time_improved}]
Overall, the running time of \textsc{PartialVectorUpdate} (Algorithm~\ref{alg:partial_vector_update_improved}) is $O( n^{2a} + \wt{p} n^a )$.
\end{proof}

%
%\subsection{Running time of \textsc{Update}}\label{sec:update_time_improved}

%The goal of this section is to prove Lemma~\ref{lem:update_time_improved}.
%\begin{lemma}[Running time of \textsc{Update}]\label{lem:update_time_improved}
%The procedure \textsc{UpdateQuery} (in Algorithm~\ref{alg:update_improved}) takes ??? %time.
%\end{lemma}
%\begin{proof}
%\end{proof}

\subsection{Running time of \textsc{Initialize}}\label{sec:initialize_time_improved}

The goal of this section is to prove Lemma~\ref{lem:initialize_time_improved}.
\begin{lemma}[Running time of \textsc{Initialize}]\label{lem:initialize_time_improved}

The procedure \textsc{Initialize} (in Algorithm~\ref{alg:initialize_improved}) takes $O(n^{\omega+o(1)})$ time.
\end{lemma}

\begin{proof}
The bottleneck in \textsc{Initialize} is to compute $M=A^{\top}(AVA^{\top})^{-1}A$ and $Q=R\sqrt{V}M$. Other operations take at most $O(n^{2+o(1)})$ time.

Computing $M$ involves matrix multiplication of two $n\times n$ matrix, and inversion of a $n\times n$ matrix. Both of them can be done in $O(n^{\omega+o(1)})$ time.

Computing $Q$ involves matrix multiplication of $n^{1+o(1)}\times n$ matrix with a $n\times n$ matrix. This can be done in $O(n^{\omega+o(1)})$ time.

Therefore, the total running time is $O(n^{\omega+o(1)})$.
\end{proof}

\section{Data structure : amortized time}
\label{sec:amortize_time_improved}
In this section we provide an amortized analysis of the four procedures \textsc{MatrixUpdate}, \textsc{PartialMatrixUpdate}, \textsc{VectorUpdate}, and \textsc{PartialVectorUpdate}. Our amortized analysis builds upon the amortized analysis of \cite{cls19}.
\subsection{Definitions and  Preliminaries}\label{sec:facts_amortized_improved}

Our algorithm makes $T$ calls to the procedure \textsc{UpdateQuery}. For clarity, we define some notations with superscripts that represent the number of iterations. For the input diagonal matrix $W$ and its approximations $V$ and $\wt{V}$, we define the following notations:

\begin{definition}[Definitions of sequences $\{w^{(j)}\}_{j=0}^T$, $\{v^{(j)}\}_{j=0}^T$ and $\{\wt{v}^{(j)}\}_{j=0}^T$]
\label{def:w_j_v_j_wt_v_j}
When initializing, we use $v^{(0)}$ and $\wt{v}^{(0)}$ to denote the initial values of the data structure members $v$ and $\wt{v}$. Note that $v^{(0)}=\wt{v}^{(0)}=w^{(0)}$. 

In the $j$-th iteration, we use $w^{(j+1)}$ to denote the input $w^{\new}$ of \textsc{UpdateQuery}. Since the procedure \textsc{Update} might update both $v$ and $\wt{v}$, we use $v^{(j+1)}$ and $\wt{v}^{(j+1)}$ to denote the new values of $v$ and $\wt{v}$ respectively.
\end{definition}

We define similar notations for the input query vector $h$ and its approximations $g$ and $\wt{g}$:
\begin{definition}[Definitions of sequences $\{h^{(j)}\}_{j=0}^T$, $\{g^{(j)}\}_{j=0}^T$ and $\{\wt{g}^{(j)}\}_{j=0}^T$]
\label{def:h_j_g_j_wt_g_j}
When initializing, we use $g^{(0)}$ and $\wt{g}^{(0)}$ to denote the initial values of the data structure members $g$ and $\wt{g}$. Note that $g^{(0)}=\wt{g}^{(0)}=h^{(0)}$. 

In the $j$-th iteration, we use $h^{(j+1)}$ to denote the input $h^{\new}$ of \textsc{UpdateQuery}. Since the procedure \textsc{Update} might update both $g$ and $\wt{g}$, we use $g^{(j+1)}$ and $\wt{g}^{(j+1)}$ to denote the new values of $g$ and $\wt{g}$ respectively.
\end{definition}

The following four notations will be helpful to prove the amortized cost of the procedures:
\begin{definition}[Definition of $k$, $\wt{k}$, $p$, and $\wt{p}$]\label{def:k_wt_k_p_wt_p}
In the $j$-th iteration, we define the following notations:
\begin{enumerate}
    \item $k_j$ and $\wt{k}_j$ denote outputs $k$ and $\wt{k}$ of \textsc{UpdateV} on Line~\ref{alg:line:update_v_in_update_query_improved} of \textsc{UpdateQuery} (Algorithm~\ref{alg:update_query_improved}).
    %\item $\wt{k}_j$ denotes the output $\wt{k}$ of \textsc{UpdateV} on Line~\ref{alg:line:update_v_in_update_query_improved} of \textsc{UpdateQuery} (Algorithm~\ref{alg:update_query_improved}),
    \item $p_j$ and $\wt{p}_j$ denote outputs $p$ and $\wt{p}$ of \textsc{UpdateG} on Line~\ref{alg:line:update_g_in_update_query_improved} of \textsc{UpdateQuery} (Algorithm~\ref{alg:update_query_improved}),
    %\item $\wt{p}_j$ denotes the output $\wt{p}$ of \textsc{UpdateG} on Line~\ref{alg:line:update_g_in_update_query_improved} of \textsc{UpdateQuery} (Algorithm~\ref{alg:update_query_improved}).
\end{enumerate}
\end{definition}

%On the $j$-th round, we define $k_j:=0$ if we never execute the procedure \textsc{MatrixUpdate} (which is equivalent to that the if-clause on Line~\ref{alg:line:if_for_matrix_improved} of procedure \textsc{Update}(Algorithm~\ref{alg:update_improved}) is False). And otherwise we define 
%\begin{align*}
%k_j:=|\supp(v^{\new}-v)|,
%\end{align*}
%where $v^{\new}$ is defined on line~\ref{alg:line:binary_search_for_v_improved} of procedure \textsc{Update}(Algorithm~\ref{alg:update_improved}). (Here WLOG we assume for any $i$, if $v^{\new}_i$ is set to be $w_i$ on Line~\ref{alg:line:tilde_approximate_improved} of \soft, then $w_i\neq v_i$. Note that otherwise we could simply define $k_j$ be the set of $i$ such that $v^{\new}_i$ is set to be $w_i$ on Line~\ref{alg:line:tilde_approximate_improved}.)

%On the $j$-th round, we define $\wt{k}_j:=0$ if we never execute the procedure \textsc{PartialMatrixUpdate} (which is equivalent to that the if-clause on Line~\ref{alg:line:if_for_partial_matrix_improved} of procedure \textsc{Update}(Algorithm~\ref{alg:update_improved}) is False). And otherwise we define 
%\begin{align*}
%\wt{k}_j:=|\supp(\wt{v}^{\new}-\wt{v})|,
%\end{align*}
%where $\wt{v}^{\new}$ is defined on Line~\ref{alg:line:binary_search_for_wt_v_improved} and \ref{alg:line:adjust_for_wt_v_improved} of procedure \textsc{UpdateQuery}(Algorithm~\ref{alg:update_query_improved}).

We also define a function $\psi$ that approximates absolute function $|x|$ around $0$. It will be used to define the potential functions for amortized analysis.

\begin{definition}[Definition of $\psi$ function] \label{def:psi}
Let $\epsilon_{\mathrm{mp}} \in (0,1/10)$. 
Let $\psi:\R \to \R$ be defined by
\begin{align*}%\label{eq:def_psi}
\psi(x) = 
\begin{cases}
\frac{|x|^2}{2\epsilon_{\mathrm{mp}}}, & |x| \in [0,\epsilon_{\mathrm{mp}}] ; \\
\epsilon_{\mathrm{mp}} - \frac{ (2\epsilon_{\mathrm{mp}} - |x|)^2 }{2\epsilon_{\mathrm{mp}}},  & |x| \in (\epsilon_{\mathrm{mp}},2\epsilon_{\mathrm{mp}}] ; \\
\epsilon_{\mathrm{mp}}, & |x| \in (2\epsilon_{\mathrm{mp}}, +\infty) .
\end{cases}
\end{align*}
\end{definition}

We also make the following assumption about the values of $\epsilon_{\mathrm{mp}}$ and $\epsilon_{\far}$:
\begin{assumption}\label{ass:epsilon_far}
The error parameters satisfy $0 < \epsilon_{\mathrm{mp}} < 1/10$ and $0 < \epsilon_{\far} < \frac{\epsilon_{\mathrm{mp}}}{100\log n}$.
\end{assumption}

\subsection{Facts based on \textsc{Adjust} and two level of \soft}
The following facts are direct results from the algorithms, and they are proved solely based on the algorithms. They will be useful when proving the lemmas of amortized running time.

\begin{fact}[Characterization of $\wt{k}$ (Line~\ref{alg:line:binary_search_for_wt_v_improved}), $\wt{v}^{\tmp}$ (Line~\ref{alg:line:binary_search_for_wt_v_improved}), $\wt{v}^{\new}$ (Line~\ref{alg:line:adjust_for_wt_v_improved}) of \textsc{UpdateV}]\label{fac:characterization_wt_k_tmp_wt_v_tmp_wt_v_new} 
In the $j$-th iteration, $\forall i\in [n]$, let $y_i=|w_i^{(j+1)}/\wt{v}_i^{(j)}-1|$. Let $\pi:[n]\to [n]$ be the sorting such that $y_{\pi(i)}\geq y_{\pi(i+1)}$. We have the following guarantees from the description of procedures \soft (Algorithm~\ref{alg:binary_search_improved}) and \textsc{Adjust} (Algorithm~\ref{alg:adjust}).
    \begin{align*}
        1. & ~ \{i \in [n] : \psi( w^{(j+1)}_i / \wt{v}^{(j)}_i - 1 ) \geq \epsilon_{\mathrm{mp}}/2 \} \subseteq \pi([\wt{k}]) \\
        2. & ~ \wt{v}^{\tmp}_i =
        \begin{cases}
        w^{(j+1)}_i, & ~ \mathrm{~if~}i \in \pi( [\wt{k}] ) ; \\
        \wt{v}^{(j)}_i, & ~ \mathrm{~otherwise}.
        \end{cases} \\
        3. & ~ \wt{v}^{\new}_i =
        \begin{cases}
        v^{(j)}_i, & ~ \mathrm{~if~} i \in \pi( [ \wt{k} ] ) \mathrm{~and~} w^{(j+1)}_i \in [(1-\epsilon_{\far})v^{(j)}_i,(1+\epsilon_{\far})v^{(j)}_i] ;\\
        w^{(j+1)}_i, & ~ \mathrm{~if~}i \in \pi( [ \wt{k} ] ) \mathrm{~but~} w^{(j+1)}_i \notin [(1-\epsilon_{\far})v^{(j)}_i,(1+\epsilon_{\far})v^{(j)}_i] ; \\
        \wt{v}^{(j)}_i, & ~ \mathrm{~otherwise}.
        \end{cases} \\
        4. & ~ \forall i\notin \pi([\wt{k}]), ~ | w_i^{(j+1)}/\wt{v}^{\new}_i - 1|<\epsilon_{\mathrm{mp}}.
    \end{align*}

\end{fact}
\begin{proof}
\noindent {\bf Part 1 and 2.}
In Line~\ref{alg:line:binary_search_for_wt_v_improved} in Procedure~\textsc{UpdateV}(Algorithm~\ref{alg:update_v_improved}), we create $\wt{v}^{\tmp},\wt{k}$ by calling procedure $\soft (y_i \leftarrow \psi( w^{(j+1)}_i / \wt{v}^{(j)}_i - 1 ), w^{(j+1)},\wt{v}^{(j)},\epsilon_{\mathrm{mp}}/2, n^{\wt{a}})$. 

From the initial assignment of $\wt{k}$ on Line~\ref{alg:line:k_initial_binary_search} of \soft (Algorithm~\ref{alg:binary_search_improved}), and the fact that the repeat-loop (Line~\ref{alg:line:repeat_start_binary_search_improved} to \ref{alg:line:repeat_end_binary_search_improved}) only increases $k$, we know that for any $i$ such that $y_i = \psi ( ( w^{\new}_i - \wt{v}_i ) / \wt{v}_i ) \geq \epsilon_{\mathrm{mp}} / 2$, $i \in \pi([\wt{k}])$.

By how we calculate $\wt{v}^{\tmp}$ in $\soft$ (Line~\ref{alg:line:v_new_binary_search_improved} of Algorithm~\ref{alg:binary_search_improved}), Part 2 follows directly.

\noindent {\bf Part 3.}
We get $\wt{v}^{\new}$ by calling procedure $\textsc{Adjust}\left(\wt{v}^{\tmp}, \wt{v}^{(j)}, v^{(j)}, \epsilon_{\far}\right)$ (Line~\ref{alg:line:adjust_for_wt_v_improved} in Procedure~\textsc{UpdateV}, Algorithm~\ref{alg:update_v_improved}). By Part 2
%that $\wt{v}^{\tmp}_i = w^{(j+1)}_i$ for $i \in \pi( [\wt{k}] )$ and $\wt{v}^{\tmp}_i = \wt{v}^{(j)}_i$ for other $i$,
%\begin{align*}
%\wt{v}^{\tmp}_i =
%\begin{cases}
%w^{(j+1)}_i, & ~ \mathrm{~if~}i \in \pi( [\wt{k}] ) ; \\
%\wt{v}^{(j)}_i, & ~ \mathrm{~otherwise~} .
%\end{cases}
%\end{align*}
and the rule of how we calculate $\wt{v}^{\new}_i$(see Line~\ref{alg:line:adjust_start} to \ref{alg:line:adjust_chase_parent} in Procedure~\textsc{Adjust}(Algorithm~\ref{alg:adjust})):
\begin{align*}
    & \wt{v}^{\new}_i\leftarrow \wt{v}^{\tmp}_i\\
   &\textbf{if}~~~\wt{v}_i^{\tmp} \neq \wt{v}^{(j)}_i~~\text{and}~~\wt{v}_i^{\tmp}\in[ (1- \epsilon_{\far} )v^{(j)}_i,(1 + \epsilon_{\far} )v^{(j)}_i]~~~~\textbf{then}\\
    &~~~~~\wt{v}_i^{\new} \leftarrow v^{(j)}_i
\end{align*}
It is easy to see that for $i \notin \pi([\wt{k}])$, $\wt{v}^{\tmp}_i = \wt{v}^{(j)}_i$, so $\wt{v}^{\new}_i = \wt{v}^{\tmp}_i = \wt{v}^{(j)}_i$. And for $i \in \pi([\wt{k}])$, if $w^{(j+1)}_i \in [(1-\epsilon_{\far})v^{(j)}_i,(1+\epsilon_{\far})v^{(j)}_i]$, then $\wt{v}^{\new}_i$ is adjusted to be $v^{(j)}_i$, otherwise $\wt{v}^{\new}_i = \wt{v}^{\tmp}_i = w^{(j+1)}_i$.
%\begin{align*}
    %\wt{v}^{\new}_i = 
    %\begin{cases}
    %v^{(j)}_i, & ~ \mathrm{~if~} i \in \pi( [ \wt{k} ] ) \mathrm{~and~} w^{(j+1)}_i \in ((1-\epsilon_{\far})v^{(j)}_i,(1+\epsilon_{\far})v^{(j)}_i) ;\\
    %w^{(j+1)}_i, & ~ \mathrm{~if~}i \in \pi( [ \wt{k} ] ) \mathrm{~but~} w^{(j+1)}_i \notin ((1-\epsilon_{\far})v^{(j)}_i,(1+\epsilon_{\far})v^{(j)}_i) ; \\
    %\wt{v}^{(j)}_i, & ~ \mathrm{~otherwise}.
    %\end{cases}
%\end{align*}

\noindent {\bf Part 4.} From Part 3, $\wt{v}^{\new}_i$ has three possible values. If $\wt{v}^{\new}_i = w^{(j+1)}_i$, we have $|w_{i}^{(j+1)} / \wt{v}^{\new}_{i} - 1 | = 0$.

If $\wt{v}^{\new}_i = v^{(j)}_i$, we have $ w_i^{(j+1)} \in [(1 - \epsilon_{\far})v^{(j)}_i, (1 + \epsilon_{\far})v^{(j)}_i ]$. So $|w_{i}^{(j+1)} / \wt{v}^{\new}_{i} - 1 | < \epsilon_{\far} < \epsilon_{\mathrm{mp}}$ (Assumption~\ref{ass:epsilon_far}).

If $\wt{v}^{\new}_i = \wt{v}^{(j)}_i$, we have $i\notin \pi([\wt{k}])$, then by Part 1 we know that %$i\notin \{i'\in [n] : \psi ( w_{i'}^{(j+1)} / \wt{v}^{(j)}_{i'} -  1 ) \geq \epsilon_{\mathrm{mp}}/2\}$, which means 
$\psi(  w_i^{(j+1)} / \wt{v}^{(j)}_i - 1 ) < \epsilon_{\mathrm{mp}}/2$. By definition of $\psi$ function (Definition~\ref{def:psi}), if $\psi(x) < \epsilon_{\mathrm{mp}}/2$, we have $|x| < \epsilon_{\mathrm{mp}}$.
\end{proof}

%Combining Part 1 and Part 3 of Fact~\ref{fac:characterization_wt_k_tmp_wt_v_tmp_wt_v_new}, we directly have the following corollary:
%\begin{corollary} \label{cor:when_wt_v_new_equal_wt_v}
%On the $j$-th round, for any $i\in [n]$, if $\wt{v}^{\new}_i=\wt{v}_i^{(j)}$, then $\psi \left(\frac{w_i^{(j+1)}-\wt{v}^{(j)}_i}{\wt{v}^{(j)}_i}\right)<\epsilon_{\mathrm{mp}}$.
%\end{corollary}
%\begin{proof}
%Since $\wt{v}^{\new}_i=\wt{v}_i^{(j)}$, using Part 3 we know that $i\notin \pi([\wt{k}])$, so using Part 1 we know that $i\notin \left\{i'\in [n]:\psi \left(\frac{w_{i'}^{(j+1)}-\wt{v}^{(j)}_{i'}}{\wt{v}^{(j)}_{i'}}\right)\geq \epsilon_{\mathrm{mp}}\right\}$, which means $\psi\left(\frac{w_i^{(j+1)}-\wt{v}^{(j)}_i}{\wt{v}^{(j)}_i}\right)<\epsilon_{\mathrm{mp}}$.
%\end{proof}

\begin{fact}[Characterization of $k$ (Line~\ref{alg:line:binary_search_for_v_improved}), $v^{\new}$ (Line~\ref{alg:line:binary_search_for_v_improved}) of \textsc{UpdateV}]\label{fac:characterization_k_v_new}
In the $j$-th iteration, $\forall i\in [n]$, let $y_i=\psi (w_i^{(j+1)}/v_i^{(j)}-1)+\psi(w_i^{(j+1)}/\wt{v}_i^{(j)}-1)$. Let $\pi:[n]\to [n]$ be the sorting such that $y_{\pi(i)}\geq y_{\pi(i+1)}$. We have the following guarantees from the description of procedure $\soft$ (Algorithm~\ref{alg:binary_search_improved}). 
    \begin{align*}
        1. & ~\left \{i \in [n] : \psi ( w^{(j+1)}_i / v^{(j)}_i - 1 )+\psi( w^{(j+1)}_i / \wt{v}^{(j)}_i - 1 ) \geq  \epsilon_{\far}^2/(32\epsilon_{\mathrm{mp}}) \right\} \subseteq \pi([k]) \\
        2. & ~v^{\new}_i = 
        \begin{cases}
        w^{(j+1)}_i, & ~ \mathrm{~if~}i \in \pi( [k] ) ; \\
        v^{(j)}_i, & ~ \mathrm{~otherwise}.
        \end{cases}\\
        3. & ~  \forall i\notin \pi([k]),~ | w^{(j+1)}_i / v^{\new}_i -1 | < \epsilon_{\mathrm{mp}}.
    \end{align*}

\end{fact}
\begin{proof}
\noindent {\bf Part 1 and 2.}
From Line~\ref{alg:line:binary_search_for_v_improved} of \textsc{UpdateV} (Algorithm~\ref{alg:update_v_improved}) $v^{\new} \in \R^n$ is the output of
\begin{align*}
\soft (y_i \leftarrow \psi(  w^{(j+1)}_i / v^{(j)}_i - 1 ) + \psi(  w^{(j+1)}_i / \wt{v}^{(j)}_i - 1 ), w^{(j+1)}, v^{(j)}, \frac{\epsilon_{\far}^2}{32\epsilon_{\mathrm{mp}}}, n^a ).
\end{align*}
The statements then follow from a similar proof as that of Part 1 and 2 of Fact~\ref{fac:characterization_wt_k_tmp_wt_v_tmp_wt_v_new}.
%From the initial assignment of $k$ on Line~\ref{alg:line:k_initial_binary_search} of \soft (Algorithm~\ref{alg:binary_search_improved}), and the fact that the repeat-loop (Line~\ref{alg:line:repeat_start_binary_search_improved} to \ref{alg:line:repeat_end_binary_search_improved}) only increases $k$, we know that for any $i$ such that
%\begin{align*}
%y_i =  ~ \psi( ( w^{\new}_i - v_i ) / v_i ) + \psi ( ( w^{\new}_i - \wt{v}_i ) / \wt{v}_i ) \geq  ~ \frac{\epsilon_{\far}^2}{32\epsilon_{\mathrm{mp}}},
%\end{align*}
%we have $i\in \pi([k])$.
%So we have 
%\begin{align*}
        %\left \{i \in [n] : \psi \left( w^{(j+1)}_i / v^{(j)}_i - 1 \right)+\psi\left( w^{(j+1)}_i / \wt{v}^{(j)}_i \right) \geq \frac{\epsilon_{\far}^2}{32\epsilon_{\mathrm{mp}}} \right\} \subseteq \pi([k]).
%\end{align*}
%This finishes the proof of Part 1. From the assignment of $v^{\new}$ on Line~\ref{alg:line:v_new_binary_search_improved} of \soft (Algorithm~\ref{alg:binary_search_improved}), we directly have that 
%\begin{align*}
    %v^{\new}_i = 
    %\begin{cases}
    %w^{(j+1)}_i, & ~ \mathrm{~if~}i \in \pi( [k] ) ; \\
    %v^{(j)}_i, & ~ \mathrm{~otherwise}.
    %\end{cases}
%\end{align*}
%This finishes the proof of Part 2.

\noindent {\bf Part 3.} From Part 2, $v^{\new}_i$ has two possible values. If $v^{\new}_i = w^{(j+1)}_i$, we have $|w^{(j+1)}_i / v^{\new}_i - 1| = 0$.

If $v^{\new}_i = v^{(j)}_i$, we have $i \notin \pi([k])$, then by Part 1 we know that $\psi(w^{(j+1)}_i / v^{(j)}_i - 1) < \epsilon_{\far}^2 / (32 \epsilon_{\mathrm{mp}}) < \epsilon_{\mathrm{mp}} / 2$ (Assumption~\ref{ass:epsilon_far}). Then $|w^{(j+1)}_i / v^{(j)}_i - 1| < \epsilon_{\mathrm{mp}}$ by Definition~\ref{def:psi}.
\end{proof}

\begin{fact}[$\wt{v}$ is far away from $v$]\label{fac:v_far_from_wt_v}
For any $j\in \{0,...,T\}$, %at the beginning of the $j$-th iteration, 
for any $i \in [n]$, either $v^{(j)}_i= \wt{v}^{(j)}_i$, or
$| \wt{v}^{(j)}_i / v^{(j)}_i - 1 | > \epsilon_{\far}$.
\end{fact}
\begin{proof}
We prove it by induction. When initializing, we set $v^{(0)}=\wt{v}^{(0)}$ (Line~\ref{alg:line:initialize_v_wt_v_improved} of \textsc{Initialize}, Algorithm~\ref{alg:initialize_improved}), so the statement is true when $j=0$. 

In later iterations, $v$ and $\wt{v}$ can only be modified by procedure \textsc{UpdateV}. And in \textsc{UpdateV} (Algorithm~\ref{alg:update_v_improved}), the if branch on Line~\ref{alg:line:if_for_matrix_update} ensures that we can enter at most one of \textsc{MatrixUpdate} or \textsc{PartialMatrexUpdate}. Now we analyze iteration $j$ by looking into these two cases.

\noindent \textbf{Case 1.} When we enter Procedure~\textsc{MatrixUpdate}, we always set $v\leftarrow v^{\new}$ and $\wt{v}\leftarrow v^{\new}$ (Line~\ref{alg:line:matrix_update_improved:v_g} of Algorithm~\ref{alg:matrix_update_improved}). So $v^{(j+1)}=\wt{v}^{(j+1)}$, and the statement holds in this case.

\noindent \textbf{Case 2.} When we enter Procedure~\textsc{PartialMatrexUpdate}, we know we did not enter Procedure~\textsc{MatrixUpdate} due to the else-if branch, and we do not modify $v$ in \textsc{PartialMatrexUpdate}, so $v$ does not change, i.e., $v^{(j+1)}=v^{(j)}$.

And for $\wt{v}$, we set $\wt{v}^{(j+1)} \leftarrow \wt{v}^{\new}$ (Line~\ref{alg:line:partial_matrix_update_improved:everything_else} in Algorithm~\ref{alg:partial_matrix_update_improved}), where $\wt{v}^{\new}$ is defined on Line~\ref{alg:line:adjust_for_wt_v_improved} of Procedure~\textsc{UpdateV}(Algorithm~\ref{alg:update_v_improved}). By Part 3 of Fact~\ref{fac:characterization_wt_k_tmp_wt_v_tmp_wt_v_new} we have
\begin{align*}
    \wt{v}_i^{(j+1)} \leftarrow \wt{v}^{\new}_i = 
    \begin{cases}
        v^{(j)}_i, & ~ \mathrm{~if~} i \in \pi( [ \wt{k}] ) \mathrm{~and~} w^{(j+1)}_i \in [(1-\epsilon_{\far})v^{(j)}_i,(1+\epsilon_{\far})v^{(j)}_i] ;\\
        w^{(j+1)}_i, & ~ \mathrm{~if~}i \in \pi( [ \wt{k}] ) \mathrm{~but~} w^{(j+1)}_i \notin [(1-\epsilon_{\far})v^{(j)}_i,(1+\epsilon_{\far})v^{(j)}_i] ; \\
        \wt{v}^{(j)}_i, & ~ \mathrm{~otherwise}.
    \end{cases}
\end{align*}
In the first case, we have $\wt{v}^{(j+1)}_i = v^{(j)}_i=v^{(j+1)}_i$, and the lemma statement holds. 

In the second case, we have $|\wt{v}_i^{(j+1)} / v^{(j+1)}_i - 1| = |w_i^{(j+1)} / v^{(j)}_i - 1| > \epsilon_{\far}$, and the lemma statement is also true. 

In the third case we have $\wt{v}^{(j+1)}=\wt{v}^{(j)}$. The lemma statement is true by induction hypothesis.
\end{proof}

The following corollary follows from the above fact and triangle inequality:
\begin{corollary}
\label{cor:v_far_from_wt_v}
For any $j\in \{0,\dots, T\}$ and any $i\in [n]$, if $v^{(j)}_i\neq \wt{v}^{(j)}_i$, then $\forall x\in \R$, we have
\[
|\frac{x-v^{(j)}_i}{v^{(j)}_i}| + |\frac{x-\wt{v}^{(j)}_i}{\wt{v}^{(j)}_i}| \geq \epsilon_{\far} /2.
\]
\end{corollary}
\begin{proof}
Since $v^{(j)}_i\neq \wt{v}^{(j)}_i$, from Fact~\ref{fac:v_far_from_wt_v}, we have
\begin{align} \label{eq:cor:v_far_from_wt_v}
| \wt{v}^{(j)}_i / v^{(j)}_i - 1| > \epsilon_{\far}.
\end{align}
We consider the following two cases depend on the value of $|\wt{v}^{(j)}_i / v^{(j)}_i|$:

\textbf{Case 1, $|\wt{v}^{(j)}_i / v^{(j)}_i|\leq 1$:}
    We have{\small
    \begin{align*}
        |\frac{x - \wt{v}^{(j)}_i}{\wt{v}^{(j)}_i}| + |\frac{x - v^{(j)}_i}{v^{(j)}_i}| =  |\frac{x - \wt{v}^{(j)}_i}{v^{(j)}_i}|\cdot |\frac{v^{(j)}_i}{\wt{v}^{(j)}_i}| + |\frac{x - v^{(j)}_i}{v^{(j)}_i}|
        \geq 
        |\frac{x - \wt{v}^{(j)}_i}{v^{(j)}_i}| + |\frac{x - v^{(j)}_i}{v^{(j)}_i}|
        \geq  |\frac{\wt{v}^{(j)}_i-v^{(j)}_i}{v^{(j)}_i}| 
        >  \epsilon_{\far}.
    \end{align*}}
    where the second step follows from the assumption of Case 1 that $|\wt{v}^{(j)}_i / v^{(j)}_i|\leq 1$, the
    third step follows from triangle inequality, and the fourth step follows from Eq.~\eqref{eq:cor:v_far_from_wt_v}.
    
\textbf{Case 2, $|\wt{v}^{(j)}_i / v^{(j)}_i|> 1$:} We have
    \begin{align*} %\label{eq:cor:v_far_from_wt_v_divide_by_wt_v}
        |\frac{\wt{v}^{(j)}_i-v^{(j)}_i}{\wt{v}^{(j)}_i}|= &~ \frac{|{(\wt{v}^{(j)}_i-v^{(j)}_i)}/{v^{(j)}_i}|}{|{\wt{v}^{(j)}_i}/{v^{(j)}_i}|}
        \geq ~ \frac{|{(\wt{v}^{(j)}_i-v^{(j)}_i)}/{v^{(j)}_i}|}{|{(\wt{v}^{(j)}_i-v^{(j)}_i)}/{v^{(j)}_i}|+1}
        \geq ~ \frac{\epsilon_{\far}}{\epsilon_{\far}+1} 
        \geq ~ \epsilon_{\far}/2.
    \end{align*}
    where the second step follows from triangle inequality, the third step follows from Eq.~\eqref{eq:cor:v_far_from_wt_v} and the fact that function $\frac{x}{x+1}$ is monotonically increasing, and the last step follows from $\epsilon_{\far}\leq 1$.
    
    Then similar as Case 1 we have{\small
    \begin{align*}
        |\frac{x - \wt{v}^{(j)}_i}{\wt{v}^{(j)}_i}| + |\frac{x - v^{(j)}_i}{v^{(j)}_i}| = &~ |\frac{x - \wt{v}^{(j)}_i}{\wt{v}^{(j)}_i}| + |\frac{x - v^{(j)}_i}{\wt{v}^{(j)}_i}|\cdot |\frac{\wt{v}^{(j)}_i}{v^{(j)}_i}|
        \geq 
        |\frac{x - \wt{v}^{(j)}_i}{\wt{v}^{(j)}_i}| + |\frac{x - v^{(j)}_i}{\wt{v}^{(j)}_i}|
        \geq  |\frac{\wt{v}^{(j)}_i-v^{(j)}_i}{\wt{v}^{(j)}_i}| 
        >  \epsilon_{\far}/2.
    \end{align*}}
    %where the second step follows from the assumption of Case 2 that $|\wt{v}^{(j)}_i / v^{(j)}_i|> 1$, the third step follows from triangle inequality, and the fourth step follows from Eq.~\eqref{eq:cor:v_far_from_wt_v_divide_by_wt_v}. Thus the statement of this corollary is true.
\end{proof}

\begin{fact}[Lower bound on $k_j$]\label{fac:bound_on_k_j}
%Under Assumption~\ref{ass:epsilon_far}, 
In the $j$-th iteration, we have that either $k_j = 0$, or $k_j \geq n^a$.
\end{fact}
\begin{proof}
Recall that $k_j$ is the second returned value by \textsc{UpdateV} on Line~\ref{alg:line:update_v_in_update_query_improved} of \textsc{UpdateQuery} (Algorithm~\ref{alg:update_query_improved}). If in \textsc{UpdateV} (Algorithm~\ref{alg:update_v_improved}) we do not enter the if-branch on Line~\ref{alg:line:if_for_matrix_update}, then by the return clauses on Line~\ref{alg:line:return_by_partial_matrix_update} and Line~\ref{alg:line:update_v_return_by_default}, we have that $k_j=0$.

Now we only need to prove that when we enter the if-branch on Line~\ref{alg:line:if_for_matrix_update} of \textsc{UpdateV} (Algorithm~\ref{alg:update_v_improved}), the returned value $k_j=k\geq n^a$. When the if-clause on Line~\ref{alg:line:if_for_matrix_update} is true, we have
\begin{align}\label{eq:bound_on_kj_1}
    | \supp ( \wt{v}^{\new} - v^{(j)} ) | \geq n^a,
\end{align}
where $\wt{v}^{\new}$ is defined on Line~\ref{alg:line:adjust_for_wt_v_improved}. $\forall i\in [n]$, let $y_i=\psi (w_i^{(j+1)}/v_i^{(j)}-1)+\psi(w_i^{(j+1)}/\wt{v}_i^{(j)}-1)$. Let $\pi:[n]\to [n]$ be the sorting such that $y_{\pi(i)}\geq y_{\pi(i+1)}$. %Note that $v^{\new}$ is defined on Line~\ref{alg:line:binary_search_for_v_improved}. 
From Part 1 of Fact~\ref{fac:characterization_k_v_new}, we have
\begin{align}\label{eq:bound_on_kj_2}
    \left \{i \in [n] : \psi( w^{(j+1)}_i /  v^{(j)}_i - 1 )+\psi( w^{(j+1)}_i / \wt{v}^{(j)}_i - 1 ) \geq \epsilon_{\far}^2/ (32\epsilon_{\mathrm{mp}} ) \right\} \subseteq \pi([k]).
\end{align}
Now it suffices to prove
\begin{align}\label{eq:bound_on_kj_3}
    \supp(\wt{v}^{\new}-v^{(j)}) \subseteq \left \{i : \psi( w^{(j+1)}_i / v^{(j)}_i - 1 )+\psi(  w^{(j+1)}_i / \wt{v}^{(j)}_i - 1 ) \geq \epsilon_{\far}^2/ (32\epsilon_{\mathrm{mp}}) \right\},
\end{align}
because then we would have
\begin{align*}
    k_j = &~ k
    =  |\pi([k])|
    \geq  \left|\left \{i : \psi( w^{(j+1)}_i / v^{(j)}_i - 1 )+\psi( w^{(j+1)}_i / \wt{v}^{(j)}_i - 1 ) \geq \epsilon_{\far}^2/(32\epsilon_{\mathrm{mp}}) \right\} \right| \\
    \geq & ~ |\supp(\wt{v}^{\new}-v^{(j)})| 
    \geq  n^a,
\end{align*}
where the first step follows from the definition of $k_j$, the third step follows from Eq.~\eqref{eq:bound_on_kj_2}, the fourth step follows from Eq.~\eqref{eq:bound_on_kj_3}, and the fifth step follows from Eq.~\eqref{eq:bound_on_kj_1}.

Now it remains to prove Eq.~\eqref{eq:bound_on_kj_3}. We first prove that 
\begin{align*}
\supp(\wt{v}^{\new}-v^{(j)}) \subseteq \left \{i : | w^{(j+1)}_i / v^{(j)}_i - 1 | + | w^{(j+1)}_i / \wt{v}^{(j)}_i - 1 | \geq \epsilon_{\far}/2 \right\}.
\end{align*}
Using Part 3 of Fact~\ref{fac:characterization_wt_k_tmp_wt_v_tmp_wt_v_new} we know that $\wt{v}_i^{\new}$ can be $v^{(j)}_i$, $w^{(j+1)}_i$, or $\wt{v}^{(j)}_i$. So for any $i \in \supp( \wt{v}^{\new} - v^{(j)} )$, since $\wt{v}^{\new}_i\neq v^{(j)}_i$, we know that $\wt{v}_i^{\new}$ is either $\wt{v}^{(j)}_i$ or $w_i^{(j+1)}$.
We consider these two cases:
\begin{enumerate}
    \item $\wt{v}^{\new}_i = \wt{v}^{(j)}_i \neq v^{(j)}_i$. 
    Using Corollary~\ref{cor:v_far_from_wt_v}, and plugging in $x\leftarrow w^{(j+1)}_i$ we directly have that $| w_i^{(j+1)} / v^{(j)}_i - 1 |+| w_i^{(j+1)}/\wt{v}^{(j)}_i - 1| \geq \epsilon_{\far} / 2$.
    \item $\wt{v}_i^{\new}=w_i^{(j+1)}\neq v^{(j)}_i$. By Part 3 of Fact~\ref{fac:characterization_wt_k_tmp_wt_v_tmp_wt_v_new} we have $w_i^{(j+1)} \notin [ (1-\epsilon_{\far})v^{(j)}_i, (1+\epsilon_{\far})v^{(j)}_i ]$,
    which then gives us
    $| w_i^{(j+1)} / v^{(j)}_i - 1 | \geq \epsilon_{\far} \geq \epsilon_{\far}/2$.
\end{enumerate}
Thus we always have that $\forall i \in \supp( \wt{v}^{\new} - v^{(j)} )$, $|  w_i^{(j+1)} / v^{(j)}_i - 1 | + |  w_i^{(j+1)} / \wt{v}^{(j)}_i - 1 | \geq \epsilon_{\far} / 2$.
So at least one of the following is true:
\begin{align*}
|  w_i^{(j+1)} / v^{(j)}_i - 1 | \geq \epsilon_{\far}/4 \mathrm{~~~or~~~} |  w_i^{(j+1)} / \wt{v}^{(j)}_i - 1 | \geq \epsilon_{\far}/4.
\end{align*}

Without loss of generality, assume $| w_i^{(j+1)} / v^{(j)}_i - 1 |\geq \frac{\epsilon_{\far}}{4}$. We have
\begin{align*}
    \psi(   w_i^{(j+1)} / v^{(j)}_i - 1 ) + \psi(  w_i^{(j+1)} / \wt{v}^{(j)}_i - 1 )
    \geq ~ \psi(  w_i^{(j+1)} / v^{(j)}_i - 1 )
    \geq ~ \psi( \epsilon_{\far} / 4 )
    =  ~ \epsilon_{\far}^2/(32\epsilon_{\mathrm{mp}}).
\end{align*}
where the first two steps follow from $\psi$ is non-negative and non-decreasing, and the third step follows from $\psi( \epsilon_{\far}/4) = \epsilon_{\far}^2 / (32\epsilon_{\mathrm{mp}})$ (see Definition~\ref{def:psi} of $\psi$). Thus this proves Eq.~\eqref{eq:bound_on_kj_3}.
\end{proof}

\begin{fact}[Lower bound on $\wt{k}_j$]\label{fac:bound_wt_k_j}
In the $j$-th iteration, we have that either $\wt{k}_j=0$, or $\wt{k}_j\geq n^{\wt{a}}$.
\end{fact}
\begin{proof}
Recall that $\wt{k}_j$ is the third returned value by \textsc{UpdateV} on Line~\ref{alg:line:update_v_in_update_query_improved} of \textsc{UpdateQuery} (Algorithm~\ref{alg:update_query_improved}). If in \textsc{UpdateV} (Algorithm~\ref{alg:update_v_improved}) we do not enter the else-if branch on Line~\ref{alg:line:if_for_partial_matrix_update}, then by the return clauses on Line~\ref{alg:line:return_by_matrix_update} and Line~\ref{alg:line:update_v_return_by_default}, we have that $\wt{k}_j=0$.

Now we only need to prove that when we enter the else-if branch on Line~\ref{alg:line:if_for_partial_matrix_update} of \textsc{UpdateV} (Algorithm~\ref{alg:update_v_improved}), the returned value $\wt{k}_j=\wt{k}\geq n^{\wt{a}}$. When the if-clause on Line~\ref{alg:line:if_for_partial_matrix_update} is true, we have 
\[
|\supp(\wt{v}^{\new} - \wt{v}^{(j)})| \geq n^{\wt{a}}.
\]
From Part 3 of Fact~\ref{fac:characterization_wt_k_tmp_wt_v_tmp_wt_v_new}, we have $|\supp(\wt{v}^{\new} - \wt{v}^{(j)})| \subseteq \pi([\wt{k}])$. Therefore, 
\begin{align*}
    \wt{k}_j = \wt{k} \geq |\supp(\wt{v}^{\new} - \wt{v}^{(j)})| \geq n^{\wt{a}}.
\end{align*}
\end{proof}

\begin{fact}[Characterization of \textsc{MatrixUpdate}]\label{fac:characterization_MatrixUpdate}
Assume Assumption~\ref{ass:epsilon_far} is true.
In the $j$-th iteration, $\forall i \in [n]$, let $y_i = \psi( w_i^{(j+1)} / v_i^{(j)} - 1) + \psi( w_i^{(j+1)} / \wt{v}_i^{(j)} - 1)$.
Let $\pi:[n]\to [n]$ be the sorting such that $y_{\pi(i)}\geq y_{\pi(i+1)}$. If $k_j>0$, we have the following:
\begin{enumerate}
    \item $k_j$ satisfies that either $k_j=n$ or $y_{\pi(k_j)} < (1-1/\log n) \cdot y_{\pi(k_j/1.5)}$. 
    \item $y_{\pi(k_j)} \geq \epsilon_{\far}^2 / (3200\epsilon_{\mathrm{mp}})$. 
    \item After the procedure \textsc{MatrixUpdate},
    \begin{align*}
        \begin{cases}
        \wt{v}^{(j+1)}_{\pi(i)} = v^{(j+1)}_{\pi(i)} =  w^{(j+1)}_{\pi(i)}, & \forall i \leq k_j; \\
       \wt{v}_{\pi(i)}^{(j+1)} = v_{\pi(i)}^{(j+1)} = v_{\pi(i)}^{(j)} = \wt{v}_{\pi(i)}^{(j)}, & \forall i > k_j.
        \end{cases}
    \end{align*}
    \item The running time of \textsc{MatrixUpdate} in the $j$th iteration is $\Tmat(k_j, n^{1+o(1)}, n)$.
\end{enumerate}
\end{fact}
\begin{proof}

\noindent {\bf Part 1 and Part 2.} 
Note that as long as $k_j>0$, $k_j$ is the $k$ computed in Line~\ref{alg:line:binary_search_for_v_improved} in Procedure~\textsc{UpdateV}(Algorithm~\ref{alg:update_v_improved}). By Fact~\ref{fac:bound_on_k_j} and $k_j\neq 0$, we have $k_j\geq n^a$. Therefore, when calculating $k$ using one call of $\soft$ (Line~\ref{alg:line:binary_search_for_v_improved} in Algorithm~\ref{alg:update_v_improved}), we must have entered the repeat-until branch (Line~\ref{alg:line:repeat_start_binary_search_improved} to \ref{alg:line:repeat_end_binary_search_improved} in Algorithm~\ref{alg:binary_search_improved}). So $k_j$ must satisfy the end condition of the repeat-loop that either $k_j=n$ or $y_{\pi(k_j)} < (1-1/\log n) \cdot y_{\pi(k_j/1.5)}$. This finishes the proof of Part 1.

%\noindent \textbf{Part 1.} If we execute the procedure \textsc{MatrixUpdate} on the $j$-th round, $k_j$ is defined as $k_j=|\supp(v^{\new}-v)|$. Note that since $v^{\new}$ is assigned the output of \soft on Line~\ref{alg:line:binary_search_for_v_improved} of \textsc{Update} (Algorithm~\ref{alg:update_improved}), $k_j$ is exactly the $k$ on Line~\ref{alg:line:tilde_approximate_improved} of \soft (Algorithm~\ref{alg:binary_search_improved}). Thus by the until condition on Line~\ref{alg:line:repeat_end_binary_search_improved} of \soft (Algorithm~\ref{alg:binary_search_improved}), we have that either 
%\begin{align*}
%k_j = n \mathrm{~~~or~~~} y_{ \pi(k_j) } < (1-1/\log n) \cdot y_{\pi(k_j/1.5)}.
%\end{align*}

%\noindent \textbf{Part 2.}
%Using the same argument as of Part 1, we know that $k_j$ is exactly the $k$ on Line~\ref{alg:line:tilde_approximate_improved} of \soft (Algorithm~\ref{alg:binary_search_improved}).
Further, let $k^*$ denote the largest index such that $y_{\pi(k^*)} \geq \epsilon_{\far}^2 / (32\epsilon_{\mathrm{mp}})$. We have
%The ending condition of the repeat-loop on Line~\ref{alg:line:repeat_end_binary_search_improved} of \soft (Algorithm~\ref{alg:binary_search_improved}) shows that
\begin{align*}
 y_{\pi(k_j)} 
 \geq & ~ ( 1 - 1 / \log n)^{\log_{1.5} k_j - \log_{1.5} k^*} \cdot y_{\pi(k^*)}
 \geq  ( 1 - 1 / \log n )^{\log_{1.5} n} \cdot \frac{\epsilon_{\far}^2}{32\epsilon_{\mathrm{mp}}}
 \geq  \frac{\epsilon_{\far}^2}{3200\epsilon_{\mathrm{mp}}},
\end{align*}
where the first step follows from the repeat-loop (Line~\ref{alg:line:k_times_1.5_binary_search_improved} in Algorithm~\ref{alg:binary_search_improved}), the second step follows from $\log_{1.5} k_j - \log_{1.5} k^* \leq \log_{1.5}n$ and $ y_{\pi(k^*)} \geq \epsilon_{\far}^2 / (32\epsilon_{\mathrm{mp}})$, and the last step follows from the fact that for $n \geq 4$, $(1-1/\log n)^{\log_{1.5}n}\geq 1/100$. This finishes the proof of Part 2.

\noindent \textbf{Part 3.} From Line~\ref{alg:line:matrix_update_improved:v_g} of \textsc{MatrixUpdate} (Algorithm~\ref{alg:matrix_update_improved}) we have that $v^{(j+1)}=\wt{v}^{(j+1)}=v^{\new}$.
Using Part 2 of Fact~\ref{fac:characterization_k_v_new}, we have that 
\begin{align}
    \begin{cases}
    v^{\new}_{\pi(i)} = w^{(j+1)}_{\pi(i)}, & \forall i \leq k_j; \\
    v^{\new}_{\pi(i)} = v_{\pi(i)}^{(j)}, & \forall i > k_j,
    \end{cases}
\end{align}
so we have
    \begin{align*}
        \begin{cases}
        \wt{v}^{(j+1)}_{\pi(i)} = v^{(j+1)}_{\pi(i)} =  w^{(j+1)}_{\pi(i)}, & \forall i \leq k_j; \\
       \wt{v}_{\pi(i)}^{(j+1)} = v_{\pi(i)}^{(j+1)} = v_{\pi(i)}^{(j)}, & \forall i > k_j.
        \end{cases}
    \end{align*}
Note that by Part 1, either $k_j=n$ or $y_{\pi(k_j)} < (1-1/\log n) \cdot y_{\pi(k_j/1.5)}$. If $k_j=n$, there is no $i>k_j$, so the proof is already finished.

Otherwise, for any $i>k_j$, we prove that $\wt{v}^{(j)}_{\pi(i)}= v^{(j)}_{\pi(i)}$ by contradiction. If $\wt{v}^{(j)}_{\pi(i)}\neq v^{(j)}_{\pi(i)}$, using Corollary~\ref{cor:v_far_from_wt_v} and plugging in $x\leftarrow w^{(j+1)}_{\pi(i)}$, we have $| w^{(j+1)}_{\pi(i)} / v^{(j)}_{\pi(i)} - 1 | + | w^{(j+1)}_{\pi(i)} / \wt{v}^{(j)}_{\pi(i)} - 1 | \geq \epsilon_{\far} / 2$,
so at least one of the following is true:
\begin{align*}
    | w^{(j+1)}_{\pi(i)} / v^{(j)}_{\pi(i)} - 1 | \geq \epsilon_{\far}/4 \mathrm{~~~or~~~} | w^{(j+1)}_{\pi(i)} / \wt{v}^{(j)}_{\pi(i)} - 1| \geq \epsilon_{\far} / 4.
\end{align*}
Thus we have
\begin{align*}
y_{\pi(i)} 
=  \psi ( w^{(j+1)}_{\pi(i)} / v^{(j)}_{\pi(i)} - 1 ) + \psi ( w^{(j+1)}_{\pi(i)} / \wt{v}^{(j)}_{\pi(i)} - 1 ) 
\geq  \psi(\epsilon_{\far}/4)
= \epsilon_{\far}^2/(32\epsilon_{\mathrm{mp}}),
\end{align*}
where the first two steps follow from $\psi$ is non-negative and non-decreasing, and the third step follows from $\psi( \epsilon_{\far}/4)=\epsilon_{\far}^2/(32\epsilon_{\mathrm{mp}})$ (see Definition~\ref{def:psi} of $\psi$). 

Then we have
$y_{\pi(k_j)}\geq y_{\pi(i)}\geq \epsilon_{\far}^2 / (32\epsilon_{\mathrm{mp}})$ since $y$ is decreasingly sorted according to $\pi$ and $i>k_j$. But from the initial assignment of $k_j$ on Line~\ref{alg:line:k_initial_binary_search} of \soft (Algorithm~\ref{alg:binary_search_improved}), and that $k_j$ can only strictly increase afterwards, we have that $y_{\pi(k_j)}< \epsilon_{\far}^2 / (32\epsilon_{\mathrm{mp}})$. This leads to contradiction. Thus we have $\wt{v}_{\pi(i)}^{(j+1)} = v_{\pi(i)}^{(j+1)} = v_{\pi(i)}^{(j)} = \wt{v}_{\pi(i)}^{(j)},~ \forall i > k_j$.

\noindent \textbf{Part 4.}
This directly follows from Lemma~\ref{lem:matrix_update_time_improved} which proves the running time of \textsc{MatrixUpdate} per call.
\end{proof}

\begin{fact}[Characterization of \textsc{PartialMatrixUpdate}]\label{fac:characterization_PartialMatrixUpdate}
Assume Assumption~\ref{ass:epsilon_far} is true. In the $j$-th iteration, $\forall i\in [n]$, let $y_i = \psi( w_i^{(j+1)} / \wt{v}_i^{(j)} - 1 )$. Let $\pi:[n]\to [n]$ be the sorting such that $y_{\pi(i)}\geq y_{\pi(i+1)}$. If $\wt{k}_j>0$,
% , we execute the procedure \textsc{PartialMatrixUpdate} (Algorithm~\ref{alg:partial_matrix_update_improved}) on the $j$-th round. And 
we have the following:
\begin{enumerate}
    \item $\wt{k}_j$ satisfies that either $\wt{k}_j=n$ or $y_{\pi(\wt{k}_j)} < (1-1/\log n) \cdot y_{\pi(\wt{k}_j/1.5)}$.
    \item $y_{\pi(\wt{k}_j)} \geq \epsilon_{\mathrm{mp}}/100$.
    \item After the procedure \textsc{PartialMatrixUpdate}, $\forall i\in \pi( [\wt{k}_j] )$, $\wt{v}^{(j+1)}_{i}$ satisfies
    \begin{align}\label{eq:become_small_PartialMatrixUpdate}
        \psi ( w_{i}^{(j+1)}/\wt{v}^{(j+1)}_{i} - 1 ) \leq \epsilon_{\mathrm{mp}}/(200\log n),
    \end{align}
    and $\forall i \notin \pi( [\wt{k}_j] )$, $\wt{v}_{i}^{(j+1)}=\wt{v}_{i}^{(j)}$. Also, $\forall i\in [n]$, $v^{(j+1)}_i=v^{(j)}_i$.
    \item The running time of \textsc{PartialMatrixUpdate} in the $j$-th iteration is $\Tmat(\wt{k}_j, n^a, n^a)$.
\end{enumerate}
\end{fact}
\begin{proof}
\noindent \textbf{Part 1 and 2.} Note that as long as $\wt{k}_j>0$, $\wt{k}_j$ is the $\wt{k}$ computed in Line~\ref{alg:line:binary_search_for_wt_v_improved} in Procedure~\textsc{UpdateV} (Algorithm~\ref{alg:update_v_improved}). By Fact~\ref{fac:bound_wt_k_j}, we have $\wt{k}\geq n^{\wt{a}}$. So by a similar argument as that of the proof Part 1 and 2 of Fact~\ref{fac:characterization_MatrixUpdate}, we can prove Part 1 and 2 of this fact.

\noindent {\bf Part 3.} 
Because we do not modify $v$ in \textsc{PartialMatrixUpdate}, so $\forall i\in [n]$, $v^{(j+1)}_i=v^{(j)}_i$.

In procedure~\textsc{PartialMatrixUpdate}, we set $\wt{v}^{(j+1)}\leftarrow \wt{v}^{\new}$ (Line~\ref{alg:line:partial_matrix_update_improved:everything_else} of Algorithm~\ref{alg:partial_matrix_update_improved}), where $\wt{v}^{\new}$ is created by one call to $\soft$ (Line~\ref{alg:line:binary_search_for_wt_v_improved} in \textsc{UpdateV}, Algorithm~\ref{alg:update_query_improved}). By Part 3 of Fact~\ref{fac:characterization_wt_k_tmp_wt_v_tmp_wt_v_new}, we have
\begin{align*}
    \wt{v}^{\new}_i \leftarrow 
    \begin{cases}
    v^{(j)}_i, & ~ \mathrm{~if~} i \in \pi( [ \wt{k} ] ) \mathrm{~and~} w^{(j+1)}_i \in [(1-\epsilon_{\far})v^{(j)}_i,(1+\epsilon_{\far})v^{(j)}_i] ;\\
    w^{(j+1)}_i, & ~ \mathrm{~if~}i \in \pi( [ \wt{k} ] ) \mathrm{~but~} w^{(j+1)}_i \notin [(1-\epsilon_{\far})v^{(j)}_i,(1+\epsilon_{\far})v^{(j)}_i] ; \\
    \wt{v}^{(j)}_i, & ~ \mathrm{~otherwise}.
    \end{cases}
\end{align*}

For $ i \notin \pi( [ \wt{k}_j ] )$, $\wt{v}_{i}^{(j+1)}=\wt{v}_{i}^{(j)}$.

For $i\in \pi( [ \wt{k}_j ] )$, if $\wt{v}^{(j+1)}_{i} = \wt{v}^{\new}_i = w^{(j+1)}_i$, Eq.~\eqref{eq:become_small_PartialMatrixUpdate} is trivially true since $\psi ( w_{i}^{(j+1)}/\wt{v}^{(j+1)}_{i} - 1)=0$.

Otherwise $\wt{v}^{(j+1)}_{i} = \wt{v}^{\new}_i = v^{(j)}_i$ and $w^{(j+1)}_i \in [(1-\epsilon_{\far})v^{(j)}_i,(1+\epsilon_{\far})v^{(j)}_i]$, we have:
\begin{align*}
    \psi ( w_{i}^{(j+1)}/\wt{v}^{(j+1)}_{i} - 1 ) 
    =  \psi ( w_{i}^{(j+1)}/v^{(j)}_{i} - 1 ) 
    \leq  \psi(\epsilon_{\far}) 
    =  \epsilon_{\far}^2/(2\epsilon_{\mathrm{mp}})
    \leq  (\epsilon_{\mathrm{mp}}/(200\log n)),
\end{align*}
where the last step is by $\epsilon_{\far} \leq \epsilon_{\mathrm{mp}}/(100\log n)$ (Assumption~\ref{ass:epsilon_far}).

%And we have lower bound on the right hand side of inequality~\eqref{eq:become_small_PartialMatrixUpdate}.

%\begin{align*}
%    \psi\Big(\frac{w_{i}^{(j+1)}-\wt{v}^{(j)}_{i}}{\%wt{v}^{(j)}_{i}}\Big)
%    = y(i) 
%    \geq y(\wt{k}_j)
%    \geq \frac{\epsilon_{\mathrm{mp}}}{100},
%\end{align*}
%where the first step is by definition of $y(i)$, the second step is by $i\in \pi(1,...,\wt{k}_j)$, the third step is by $y_{\pi(\wt{k}_j)} \geq \frac{\epsilon_{\mathrm{mp}}}{100}$ proved in part 2.
%The prove is finished by combining upper bound and lower bound.

\noindent {\bf Part 4.}
This directly follows from Lemma~\ref{lem:partial_matrix_update_time_improved} which proves the running time of \textsc{PartialMatrixUpdate} per call.
\end{proof}

\begin{corollary}\label{cor:potential_not_increase_after_partialmatrixUpdate}
Assume Assumption~\ref{ass:epsilon_far} is true. In the $j$-th iteration, if we enter \textsc{PartialMatrixUpdate}, we must have $\forall i\in[n]$, $\psi( w_i^{(j+1)} / \wt{v}_i^{(j+1)} - 1 ) \leq \psi( w_i^{(j+1)} / \wt{v}_i^{(j)} - 1 )$.
\end{corollary}
\begin{proof}
If we enter \textsc{PartialMatrixUpdate}, we must have $\wt{k}_j>0$. By Part 2,3 of Fact~\ref{fac:characterization_PartialMatrixUpdate}, there is some permutation $\pi$ such that $\forall i\in \pi([\wt{k}_j])$, 
\begin{align*}
    \psi( w_i^{(j+1)} / \wt{v}_i^{(j)} - 1 ) \geq \epsilon_{\mathrm{mp}}/100 \text{~~~and~~~}\psi ( w_{i}^{(j+1)}/\wt{v}^{(j+1)}_{i} - 1 ) \leq \epsilon_{\mathrm{mp}}/(200\log n) < \epsilon_{\mathrm{mp}}/100.
\end{align*}
And also by Part 3 of Fact~\ref{fac:characterization_PartialMatrixUpdate}, $\forall i\notin\pi([\wt{k}_j])$, $\wt{v}_{i}^{(j+1)}=\wt{v}_{i}^{(j)}$. Therefore, $\forall i\notin\pi([\wt{k}_j])$, we have that $\psi ( w_{i}^{(j+1)}/\wt{v}^{(j+1)}_{i} - 1 ) = \psi( w_i^{(j+1)} / \wt{v}_i^{(j)} - 1 )$.
\end{proof}

%\begin{lemma}
%When we enter the procedure~\ref{alg:query_improved}, we can ensure that:
%\begin{align*}
%k:=supp(\wt{v}^{\new} - v) \leq n^a,
%& ~~\wt{k}:=supp(\wt{v}^{\new} - \wt{v} ) \leq n^{\wt{a}},\\
%g:=supp(\wt{h}^{\new} - h) \leq n^a,
%& ~~\wt{g}:=supp(\wt{h}^{\new} - \wt{h} ) \leq n^{\wt{a}}.\\
%\end{align*}
%\end{lemma}
%\begin{proof}
%The idea is to prove that each time we leave update procedure, $k$ and $\wt{k}$ indeed decrease(to 0).
%\end{proof}

%
\subsection{Amortized analysis for \textsc{MatrixUpdate}}

\subsubsection{Definitions}

%%%%%%%%%%% def of x,y
\begin{definition}[$x$ and $y$ for \textsc{MatrixUpdate}]
\label{def:x_y_matrix_update}
For any $j\in \{0, 1, \dots, T-1\}$, and for any $i \in [n]$, we define $x^{(j)}_{i}$ and $y^{(j)}_{i}$ as follows:
\begin{align*}
x^{(j)}_{i} := \psi ( w^{(j)}_{i} / v^{(j)}_{i}  - 1 ) + \psi (  w^{(j)}_{i} / \wt{v}^{(j)}_{i} - 1 ), ~~~~ %\label{eq:def_x_matrix_update}\\
y^{(j)}_{i} := \psi ( w^{(j+1)}_{i} / v^{(j)}_{i} - 1 ) + \psi( w^{(j+1)}_{i} / \wt{v}^{(j)}_{i} - 1 ), %\label{eq:def_y_matrix_update}\\
%x^{(j+1)}_{i} & ~ = \psi ( w^{(j+1)}_{i} / v^{(j+1)}_{i} - 1 ) + \psi ( w^{(j+1)}_{i} / \wt{v}^{(j+1)}_{i} - 1 ),\notag
\end{align*}
where $v^{(j)}$, $\wt{v}^{(j)}$ and $w^{(j)}$ are defined as of Definition~\ref{def:w_j_v_j_wt_v_j}.
\end{definition}
Note that the difference between $x^{(j)}_i$ and $y^{(j)}_i$ is that $w$ changes from $w^{(j)}$ to $w^{(j+1)}$. The difference between $y^{(j)}_i$ and $x^{(j+1)}_i$ is that $v$ and $\wt{v}$ changes from $v^{(j)},\wt{v}^{(j)}$ to $v^{(j+1)},\wt{v}^{(j+1)}$.

For convenience, we define permutations of the coordinates that are sorted according to $x$ or $y$.
\begin{definition}[Sorting permutations for \textsc{MatrixUpdate}]
For any $j\in \{0,1,\dots,T\}$, let $\tau_j$ be the permutation such that $x^{(j)}_{\tau_j(i)} \geq x^{(j)}_{\tau_j(i+1)}$, and let $\pi_j$ be the permutation such that $y^{(j)}_{\pi_j(i)} \geq y^{(j)}_{\pi_j(i+1)}$.

When it is clear from the context that we are only arguing about the $j$-th iteration, for simplicity we assume the coordinates of vector $x^{(j)} \in \R^n$ are sorted such that $x^{(j)}_{i} \geq x^{(j)}_{i+1}$. And we use $\tau$ and $\pi$ to denote the permutations such that $x^{(j+1)}_{\tau(i)} \geq x^{(j+1)}_{\tau(i+1)}$ and $y^{(j)}_{\pi(i)} \geq y^{(j)}_{\pi(i+1)}$.
\end{definition}

%%%%%%%%% def of g
\begin{definition}[$g$ for \textsc{MatrixUpdate}]
\label{def:g_matrix_update}
For some $a \leq \alpha$, where $\alpha$ is the dual exponent of matrix multiplication, we define $g\in \R^n$ as follows:
\begin{align*}%\label{eq:def_g_matrix_update}
g_i=
\begin{cases}
n^{-a}, &\mathrm{~if~} i \leq n^a ;\\
i^{\frac{\omega-2}{1-a}-1} \cdot n^{-\frac{a(\omega-2)}{1-a}}, &\mathrm{~if~} i \in (n^a,n].
\end{cases}
\end{align*}
\end{definition}
Note that $g$ is non-increasing. For all $k_j \in (n^a,n]$, $(k_j \cdot g_{k_j}n^2)=k_j^{\frac{\omega-2}{1-a}}\cdot n^{2-\frac{a(\omega-2)}{1-a}}$ is an upper bound of the running time $\Tmat(n, n, k_j)$ of multiplying a $n \times n$ matrix with a $n \times k_j$ matrix. For more details please refer to \cite{gu18}.

%%%%%%%%% def of Phi
\begin{definition}[Potential function $\Phi$ for \textsc{MatrixUpdate}]
\label{def:phi_matrix_update}
We define the potential function in the $j$-th iteration as
\begin{align*}%\label{eq:def_phi_matrix_update}
\Phi_j = \sum_{i = 1}^n g_i \cdot x^{(j)}_{\tau_j(i)}.
\end{align*}
%where $\tau_j(i)$ is the permutation such that $x^{(j)}_{\tau_j(i)} \geq x^{(j)}_{\tau_j(i+1)}$.
\end{definition}
Note that we always have $\Phi_j\geq 0$ since $\forall i$, $g_i$ and $x_i^{(j)}$ are both non-negative.

\subsubsection{Main result}

\begin{lemma}[Amortized time for \textsc{MatrixUpdate}] \label{lem:main_amortize_matrix_update}
Let sequences $\{w^{(j)}\}_{j=0}^T$, $\{v^{(j)}\}_{j=0}^T$, $\{\wt{v}^{(j)}\}_{j=0}^T$ be defined as of Definition~\ref{def:w_j_v_j_wt_v_j}, let $k_j$ be defined as of Definition~\ref{def:k_wt_k_p_wt_p},
and let $\{x^{(j)}\}_{j=0}^T$, $\{y^{(j)}\}_{j=0}^T$, $g$, $\Phi$ be defined as of Definition~\ref{def:x_y_matrix_update}, \ref{def:g_matrix_update} and \ref{def:phi_matrix_update}.
%in Equation~\eqref{eq:def_x_matrix_update}, \eqref{eq:def_y_matrix_update},\eqref{eq:def_g_matrix_update},\eqref{eq:def_phi_matrix_update}
If we further have the condition that the input sequence satisfies the following: $\forall j\in\{0,...,T-1\}${\small
\begin{align*}
    \sum_{i=1}^n ( \E[w_i^{(j+1)}|w^{(j)}]/ w_i^{(j)} - 1 )^2 \leq  C_1^2, ~~~
\sum_{i=1}^n (\E[( w_i^{(j+1)}/ w_i^{(j)} - 1)^2 ~|~ w^{(j)}])^2 \leq  C_2^2, ~~~
|w_i^{(j+1)} / w_i^{(j)} -  1 | \leq  1/4.
\end{align*}}

Then, we have that in expectation 
\begin{align*}
\frac{1}{T} \sum_{j=1}^T k_j g_{k_j} = &~ O\Big((\frac{C_1\epsilon_{\mathrm{mp}}}{\epsilon_{\far}^2} + \frac{C_2}{\epsilon_{\far}^2})\cdot \log n \cdot \|g\|_2 \Big).
\end{align*}

Further, combining with Lemma~\ref{lem:matrix_update_time_improved}, the expected amortized running time per iteration of \textsc{MatrixUpdate} is
\[
O\Big( (\frac{C_1\epsilon_{\mathrm{mp}}}{\epsilon_{\far}^2} + \frac{C_2}{\epsilon_{\far}^2} )\cdot (n^{2-a/2}+n^{\omega-1/2})\log n \Big).
\]
\end{lemma}

\begin{proof}
First note that in the $j$-th iteration, the value of the potential $\Phi_j$ depends on $w^{(j)}$, $v^{(j)}$ and $\wt{v}^{(j)}$. And the value of $v^{(j)}$ and $\wt{v}^{(j)}$ are affected by both \textsc{MatrixUpdate} and \textsc{PartialMatrixUpdate}. We upper bound how much the potential function can increase due to changing $w^{(j)}$ to $w^{(j+1)}$ in Section~\ref{sec:w_move_matrix_update}, and we also lower bound how much the potential function can decrease because of changing $v^{(j)}$ to $v^{(j+1)}$ and $\wt{v}^{(j)}$ to $\wt{v}^{(j+1)}$ in Section~\ref{sec:v_wt_v_move_matrix_update}.

In the beginning $v^{(0)}=\wt{v}^{(0)}=w^{(0)}$, so $\Phi_0=0$. Also note that $\Phi_j\geq 0, \forall j\in [T]$.
Thus we have
\begin{align*}
    0 \leq \E[\Phi_T]-\Phi_0 = &~ \sum_{j=0}^{T-1}\E[\Phi_{j+1}-\Phi_j]
    = ~ \sum_{j=0}^{T-1}\sum_{i=1}^n g_i \cdot \E\left[ x^{(j+1)}_{\tau(i)} -  x^{(j)}_{i} \right] \notag \\
    = & ~ \sum_{j=0}^{T-1} \Bigg(\sum_{i = 1}^n g_i \cdot \underbrace{ \E\left[ y^{(j)}_{\pi(i)} - x^{(j)}_{i} \right] }_{w\text{~move}} - \sum_{i = 1}^n g_i \cdot \underbrace{ \E\left[ y^{(j)}_{\pi(i)} - x^{(j+1)}_{\tau(i)} \right]  }_{v,\wt{v}\text{~move}}\Bigg)\\
    %\leq &~ \sum_{j = 0}^{T-1}  \Big( O(C_1 + C_2 / \epsilon_{\mathrm{mp}}) \cdot \|g\|_2 - \sum_{i = 1}^n g_i \cdot \underbrace{ \E\left[ y^{(j)}_{\pi(i)} - x^{(j+1)}_{\tau(i)} \right] }_{v,\wt{v}\text{~move}}\Big) \\
    \leq &~ \sum_{j = 0}^{T-1}  \Big( O(C_1 + C_2 / \epsilon_{\mathrm{mp}}) \cdot \|g\|_2 - \Omega \left( \epsilon_{\far}^2 \cdot k_j \cdot g_{k_j} / ( \epsilon_{\mathrm{mp}} \log n ) \right) \Big) \\
    = & ~ T \cdot O(C_1 + C_2 / \epsilon_{\mathrm{mp}}) \|g\|_2 - \sum_{j=1}^T \Omega \left( {\epsilon_{\far}^2 \cdot k_j \cdot g_{k_j} } / ( \epsilon_{\mathrm{mp}} \log n ) \right),
\end{align*}
where the second step follows from splitting terms and the fact that $\Phi_0$ is deterministic, the third step follows from the definition of $\Phi$ (Definition~\ref{def:phi_matrix_update}), the fourth step follows from splitting terms, the fifth step follows from Lemma~\ref{lem:w_move_matrix_update} which states that $\forall w^{(j)}, v^{(j)}, \wt{v}^{(j)}$, we have 
\begin{align*}
\sum_{i=1}^n g_i \cdot \E \left[ y^{(j)}_{\pi(i)} - x^{(j)}_{i} ~\Big|~ w^{(j)}, v^{(j)}, \wt{v}^{(j)}\right] \leq O(C_1 + C_2 / \epsilon_{\mathrm{mp}}) \cdot \|g\|_2,
\end{align*}
then this upper bound also holds for unconditional expectation, the fifth step also follows from Lemma~\ref{lem:v_wt_v_move_matrix_update} which states that $\sum_{i = 1}^n g_i \cdot ( y^{(j)}_{\pi(i)} - x^{(j+1)}_{\tau(i)} ) \geq \Omega( {\epsilon_{\far}^2 \cdot k_j \cdot g_{k_j} } / ( \epsilon_{\mathrm{mp} } \log n) )$.

Therefore, we have
\begin{align*}
\frac{1}{T} \sum_{j=1}^T k_jg_{k_j} \leq O\Big((\frac{C_1\epsilon_{\mathrm{mp}}}{\epsilon_{\far}^2} + \frac{C_2}{\epsilon_{\far}^2})\cdot \log n \cdot \|g\|_2 \Big).
\end{align*}

Using Lemma~\ref{lem:matrix_update_time_improved}, we have that the expected amortized running time per iteration of \textsc{MatrixUpdate} is
\begin{align*}
\frac{1}{T} \sum_{j=1}^T \Tmat(n,n,k_j)\leq &~ \frac{n^2}{T} \sum_{j=1}^T k_j g_{k_j} 
=  O\Big((\frac{C_1\epsilon_{\mathrm{mp}}}{\epsilon_{\far}^2} + \frac{C_2}{\epsilon_{\far}^2})\cdot n^2\log n \cdot \|g\|_2 \Big)\\
= &~ O\Big((\frac{C_1\epsilon_{\mathrm{mp}}}{\epsilon_{\far}^2} + \frac{C_2}{\epsilon_{\far}^2})\cdot (n^{2-a/2}+n^{\omega-1/2})\log n \Big),
\end{align*}
where the first step follows from the definition of $g$ which gives that $\Tmat(n,n,k_j)\leq n^2 k_j g_{k_j}$, and the third step follows from Lemma~\ref{lem:l2_norm_g_matrix_update_improved} that $\|g\|_2 = O(n^{-a/2} + n^{\omega-5/2})$.
\end{proof}

\subsubsection{\texorpdfstring{$w$}{} move}
\label{sec:w_move_matrix_update}
The goal of this section is to prove Lemma~\ref{lem:w_move_matrix_update}.
\begin{lemma}[$w$ move]\label{lem:w_move_matrix_update}
In the $j$-th iteration, for any possible values $w^{(j)}$, $v^{(j)}$, and $\wt{v}^{(j)}$, we have 
\begin{align}
\label{eq:w_move_matrix_update}
\sum_{i=1}^n g_i \cdot \E \left[ y^{(j)}_{\pi(i)} - x^{(j)}_{i} ~\Big|~ w^{(j)}, v^{(j)}, \wt{v}^{(j)}\right] \leq O(C_1 + C_2 / \epsilon_{\mathrm{mp}}) \|g\|_2.
\end{align}
\end{lemma}
\begin{proof}
For simplicity, in this proof we write $\E[\cdot]$ as a shorthand of $\E[\cdot | w^{(j)}, v^{(j)}, \wt{v}^{(j)}]$.

Observe that since the non-negative values $x^{(j)}_{i}$ are sorted in descending order, and $g$ is also non-increasing, we have
\begin{align}\label{eq:w_move_matrix_update:sum_g_i_x_pi_i_leq_sum_g_i_x_i}
\sum_{i = 1}^n g_i x^{(j)}_{\pi(i)} \leq \sum_{i = 1}^n g_i x^{(j)}_{i} .
\end{align}

We then have
{\small
\begin{align*}
    &~ \sum_{i=1}^n g_i \cdot \E [ y^{(j)}_{\pi(i)} - x^{(j)}_{i} ] \leq \sum_{i=1}^n g_i \cdot \E [ y^{(j)}_{\pi(i)} - x^{(j)}_{\pi(i)} ]\\
    = &~ \sum_{i=1}^n g_i \cdot \E [
    \psi (  w^{(j+1)}_{\pi(i)} / v^{(j)}_{\pi(i)} - 1 ) + \psi(  w^{(j+1)}_{\pi(i)} / \wt{v}^{(j)}_{\pi(i)} - 1 ) ] - 
    \sum_{i=1}^n g_i \cdot \E [ \psi ( w^{(j)}_{\pi(i)} / v^{(j)}_{\pi(i)} - 1 ) + \psi( w^{(j)}_{\pi(i)} / \wt{v}^{(j)}_{\pi(i)} - 1 ) ]\\
    = & ~ \sum_{i=1}^n g_i \cdot \E [
    \psi ( w^{(j+1)}_{\pi(i)} / v^{(j)}_{\pi(i)} - 1 ) - \psi ( w^{(j)}_{\pi(i)} / v^{(j)}_{\pi(i)} - 1 ) ] +
    \sum_{i=1}^n g_i \cdot \E [ \psi(  w^{(j+1)}_{\pi(i)} / \wt{v}^{(j)}_{\pi(i)} - 1  ) - \psi(  w^{(j)}_{\pi(i)} / \wt{v}^{(j)}_{\pi(i)} - 1  )]
\end{align*}
}
where the first step follows from Eq.\eqref{eq:w_move_matrix_update:sum_g_i_x_pi_i_leq_sum_g_i_x_i}, the second step follows from the definitions of $x^{(j)}$ and $y^{(j)}$ (Definition~\ref{def:x_y_matrix_update}).

Now $\sum_{i=1}^n g_i\cdot \E[y_{\pi(i)}^{(j)}-x_i^{(j)}]\leq O(C_1+C_2/\epsilon_{\mathrm{mp}})\|g\|_2$ directly follows from Lemma~\ref{lem:bounding_E_psi_w_minus_v_matrix_update} and Lemma~\ref{lem:bounding_E_psi_w_minus_wt_v_matrix_update}.
\end{proof}

It remains to prove the following two lemmas.
\begin{lemma}
\label{lem:bounding_E_psi_w_minus_v_matrix_update}
Under Assumption~\ref{ass:epsilon_far}, in the $j$-th iteration, for any $w^{(j)}$, $v^{(j)}$, and $\wt{v}^{(j)}$ we have
\begin{align*}
    \sum_{i=1}^n g_i \cdot \E [
    \psi (  w^{(j+1)}_{\pi(i)} / v^{(j)}_{\pi(i)} - 1 ) - \psi (   w^{(j)}_{\pi(i)} / v^{(j)}_{\pi(i)} - 1 ) ~|~ w^{(j)}, v^{(j)}, \wt{v}^{(j)} ] = O(C_1 + C_2 / \epsilon_{\mathrm{mp}}) \cdot \|g\|_2.
\end{align*}
\end{lemma}

\begin{proof}
For simplicity, in this proof we write $\E[\cdot]$ as a shorthand of $\E[\cdot | w^{(j)}, v^{(j)}, \wt{v}^{(j)}]$. And we also define $x$ and $y$ as
\begin{align}\label{eq:bounding_E_psi_w_minus_v_matrix_update:x_y_def}
    y_{i}=  w^{(j+1)}_{i} / v^{(j)}_{i}- 1, ~~~ x_{i} = w^{(j)}_{i} / v^{(j)}_{i} - 1,
\end{align}
and they are only used in this proof. Then the lemma statement becomes 
\begin{align*}
\sum_{i=1}^n g_i \cdot \E [
    \psi( y_{\pi(i)} ) - \psi ( x_{\pi(i)} ) ] = O(C_1 + C_2 / \epsilon_{\mathrm{mp}}) \cdot \|g\|_2.
\end{align*}

Let $I$ be the set of indices such that $|x_i| \leq 1$. We separate the term into two:
\begin{align*}
\sum_{i=1}^n g_i \cdot \E[ \psi ( y_{\pi(i)} ) - \psi ( x_{\pi(i)} )] = \sum_{i \in I} g_{\pi^{-1}(i)} \cdot \E[ \psi ( y_i ) - \psi (  x_i ) ] + \sum_{i \in I^c} g_{\pi^{-1}(i)} \cdot \E[ \psi ( y_i ) - \psi (  x_i ) ].
\end{align*}

\noindent {\bf Case 1: Terms from $I$.}
Mean value theorem shows that for any $y_i$, there exist $\zeta$ such that
\begin{align*}
\psi( y_{i} ) - \psi ( x_{i} )
= & ~ \psi'( x_{i} ) ( y_{i} - x_{i} ) + \frac{1}{2} \psi''(\zeta) (y_{i} - x_{i})^2 \\
\leq & ~ \psi' ( x_{i} )\cdot \frac{w^{(j+1)}_{i} - w^{(j)}_{i} }{v^{(j)}_{i}} + \frac{L_2}{2} \cdot ( \frac{w^{(j+1)}_{i} - w^{(j)}_{i} }{v^{(j)}_{i}} )^2,
\end{align*}
where the second step follows from plugging in the definition of $x_i$ and $y_i$ in Eq.~\eqref{eq:bounding_E_psi_w_minus_v_matrix_update:x_y_def}, and letting $L_2 = \max_x \psi''(x)$.
Taking conditional expectation (over $w^{(j)}$, $v^{(j)}$, and $\wt{v}^{(j)}$) on both sides, we get
\begin{align*}
\E[ \psi(y_{i}) - \psi(x_{i}) ] \leq &~ \psi'( x_{i} ) \cdot \frac{ \E[ w^{(j+1)}_{i}] - w^{(j)}_{i} }{ v^{(j)}_{i} } + \frac{L_2}{2} \frac{1}{( v^{(j)}_{i} )^2} \E [ (w^{(j+1)}_{i} - w^{(j)}_{i})^2] \\
= & ~ \psi'( x_{i} ) \cdot \frac{w^{(j)}_{i}}{v^{(j)}_{i}} \cdot \beta_i + \frac{L_2}{2} \frac{ ( w^{(j)}_{i} )^2}{ ( v^{(j)}_{i} )^2} \gamma_i,
\end{align*}
where $\beta_i$ and $\gamma_i$ are defined as $\beta_i = \E[w^{(j+1)}_{i}] / w^{(j)}_{i} - 1$, $\gamma_i = \E [ (w^{(j+1)}_{i} / w^{(j)}_{i} -1)^2 ]$.

This then gives us
\begin{align}\label{eq:bounding_E_psi_w_minus_v_matrix_update:1}
    \sum_{i\in I} g_{\pi^{-1}(i)} \cdot \E[ \psi(y_{i}) - \psi(x_{i}) ] \leq \sum_{i\in I} g_{\pi^{-1}(i)} \cdot \psi'( x_{i} ) \cdot \frac{w^{(j)}_{i}}{v^{(j)}_{i}} \cdot \beta_i + \sum_{i\in I} g_{\pi^{-1}(i)} \cdot \frac{L_2}{2} \cdot \frac{ ( w^{(j)}_{i} )^2}{ ( v^{(j)}_{i} )^2} \cdot \gamma_i.
\end{align}

For the term $w^{(j)}_{i} / v^{(j)}_{i}$, we note that for $i \in I$, we have 
\begin{align}
\label{eq:bounding_E_psi_w_minus_v_matrix_update:wi/vi}
| w^{(j)}_i / v^{(j)}_i  | 
=  ~ |x_i  + 1 |  
\leq  ~ | x_i | + 1  
\leq  ~ 2,
\end{align}
where the first step follows the definition of $x_i$ in Eq.~\eqref{eq:bounding_E_psi_w_minus_v_matrix_update:x_y_def}, the second step follows from triangle inequality, and the third step follows from $|x_i|\leq 1$ for $i\in I$.

Using this, we can bound the first term of Eq.~\eqref{eq:bounding_E_psi_w_minus_v_matrix_update:1} by
\begin{align}\label{eq:bounding_E_psi_w_minus_v_matrix_update:bound_sum_first_derivative}
\sum_{i\in I} g_{\pi^{-1}(i)}\cdot \psi'( x_{i} ) \cdot \frac{ w^{(j)}_{i} }{ v^{(j)}_{i} } \cdot \beta_i 
\leq & ~ \Big( \sum_{i \in I} \big(g_{\pi^{-1}(i)} \cdot \psi'( x_{i} ) \cdot \frac{ w^{(j)}_{i} }{ v^{(j)}_{i} }\big)^2 \cdot \sum_{i\in I} \beta_i^2\Big)^{1/2} \notag \\
\leq & ~ \Big( \sum_{i \in I} (2L_1\cdot g_{\pi^{-1}(i)})^2 \cdot \sum_{i\in I} \beta_i^2 \Big)^{1/2} \notag \\
\leq & ~ O(L_1) \cdot \Big( \sum_{i=1}^n g_i^2  \cdot C_1^2 \Big)^{1/2} 
= ~ O(C_1 L_1 \| g \|_2),
\end{align}
where $L_1 = \max_x |\psi'(x)|$, the first step follows by Cauchy-Schwarz inequality, and the second step follows by $| \psi'( x_i ) | \leq L_1$ and $| w^{(j)}_i / v^{(j)}_i |\leq 2$ by Eq.\eqref{eq:bounding_E_psi_w_minus_v_matrix_update:wi/vi}, and the third step follows from $\sum_{i=1}^n \beta^2_i \leq C_1^2$ from the lemma statement of Lemma~\ref{lem:main_amortize_matrix_update}.

For the second term of Eq.~\eqref{eq:bounding_E_psi_w_minus_v_matrix_update:1}, we have
\begin{align}\label{eq:bounding_E_psi_w_minus_v_matrix_update:bound_sum_second_derivative} 
\sum_{i\in I} g_{\pi^{-1}(i)}\cdot  \frac{L_2}{2}\cdot \frac{ ( w^{(j)}_{i} )^2 }{ ( v^{(j)}_{i} )^2}\cdot \gamma_i 
\leq  O(L_2) \cdot \sum_{i=1}^n g_{\pi^{-1}(i)} \cdot \gamma_i 
\leq O(L_2) \cdot \| g \|_2 \cdot \| \gamma \|_2 
=  O( C_2 L_2  \| g\|_2),
\end{align}
where the first step follows from $|w^{(j)}_i / v^{(j)}_i |\leq 2$ by Eq.\eqref{eq:bounding_E_psi_w_minus_v_matrix_update:wi/vi}, the second step follows from Cauchy-Schwarz inequality, and the third step follows from $\sum_{i=1}^n\gamma_i^2\leq C_2^2$ from the lemma statement of Lemma~\ref{lem:main_amortize_matrix_update}.

Now, plugging the bound of Eq.~\eqref{eq:bounding_E_psi_w_minus_v_matrix_update:bound_sum_first_derivative} and Eq.~\eqref{eq:bounding_E_psi_w_minus_v_matrix_update:bound_sum_second_derivative} into Eq.~\eqref{eq:bounding_E_psi_w_minus_v_matrix_update:1} and using that $L_1 = O(1)$, $L_2 = O(1/\epsilon_{\mathrm{mp}})$ (from Part 4 of Lemma~\ref{lem:def_psi}), we have that
\begin{align*}
  \sum_{i \in I} g_{\pi^{-1}(i)} \cdot \E[ \psi ( y_i ) - \psi (  x_i ) ]  \leq O(C_1 + C_2 / \epsilon_{\mathrm{mp}}) \cdot \|g\|_2.
\end{align*}

\noindent {\bf Case 2: Terms from $I^c$.}
For all $i\in I^c$, we have $|x_i| > 1$.
Note that $\psi(x)$ is a constant for $x \geq 2 \epsilon_{\mathrm{mp}}$, and from Assumption~\ref{ass:epsilon_far} we assume that $\epsilon_{\mathrm{mp}}\leq 1/4$. 
Therefore, if $|y_i| \geq 1/2$, we have that $\psi ( y_i ) - \psi ( x_i )  = 0$.
Hence, we only need to consider the $i \in I^c$ such that $|y_i| < 1/2$. For these $i$, we have that
\begin{align} \label{eq:bounding_E_psi_w_minus_v_matrix_update:y_i_minus_x_i_lower}
    |y_i - x_i| \geq | x_i | - | y_i | > 1/2,
\end{align}
where the first step follows by triangle inequality, and the second step follows from the assumptions $|x_i|> 1$ and $|y_i|<1/2$. We also have{ \small
\begin{align}
\label{eq:bounding_E_psi_w_minus_v_matrix_update:y_i_minus_x_i_upper}
|y_i - x_i| =  \left| (w^{(j+1)}_i - w^{(j)}_i) / v^{(j)}_i \right| 
=  \left|  w^{(j+1)}_i / v^{(j)}_i \right| \cdot \left| ( w^{(j+1)}_i - w^{(j)}_i ) / w^{(j+1)}_i \right| 
\leq  \frac{3}{2} \left| 1 - w^{(j)}_i / w^{(j+1)}_i \right|,
\end{align}}
where the first step follows from the definition of $y_i$ and $x_i$ in Eq.~\eqref{eq:bounding_E_psi_w_minus_v_matrix_update:x_y_def}, the third step follows from $| y_i | = |  w^{(j+1)}_i / v^{(j)}_i - 1 | < 1/2$ and thus $| w^{(j+1)}_i / v^{(j)}_i | < 3/2$. 

Combining Eq.~\eqref{eq:bounding_E_psi_w_minus_v_matrix_update:y_i_minus_x_i_lower} and Eq.~\eqref{eq:bounding_E_psi_w_minus_v_matrix_update:y_i_minus_x_i_upper}, we have that $| 1 - w^{(j)}_i / w^{(j+1)}_i | > 1/3$ and hence $| w^{(j)}_i / w^{(j+1)}_i | < 2/3$ or $| w^{(j)}_i  / w^{(j+1)}_i | > 4/3$, which then gives us $| w^{(j+1)}_i / w^{(j)}_i -1 | > 1/4$, but this contradicts with the lemma statement of Lemma~\ref{lem:main_amortize_matrix_update}, so $|y_i|<1/2$ is impossible.

Hence, we have
\begin{align*}
 \sum_{i \in I^c} g_{\pi^{-1}(i)} \cdot \E[ \psi ( y_i ) - \psi (  x_i ) ]  = 0 .
\end{align*}

Combining both cases, we have the result.
\end{proof}

\begin{lemma}
\label{lem:bounding_E_psi_w_minus_wt_v_matrix_update}
In the $j$-th iteration, for any $w^{(j)}$, $v^{(j)}$, and $\wt{v}^{(j)}$ we have
\begin{align*}
    \sum_{i=1}^n g_i \cdot \E \left[
    \psi (  w^{(j+1)}_{\pi(i)} / \wt{v}^{(j)}_{\pi(i)} - 1 ) - \psi (  w^{(j)}_{\pi(i)} / \wt{v}^{(j)}_{\pi(i)}  - 1 ) ~\Big|~ w^{(j)}, v^{(j)}, \wt{v}^{(j)}\right] = O(C_1 + C_2 / \epsilon_{\mathrm{mp}}) \cdot \|g\|_2.
\end{align*}
\end{lemma}
\begin{proof}
The proof of this lemma is exactly the same as that of Lemma~\ref{lem:bounding_E_psi_w_minus_v_matrix_update}, just replace all $v$ with $\wt{v}$ in the proof of Lemma~\ref{lem:bounding_E_psi_w_minus_v_matrix_update}.
\end{proof}

Note that in the proof of Lemma~\ref{lem:bounding_E_psi_w_minus_v_matrix_update} we do not have any requirement on $v$, so in fact we have the following more generalized lemma.
\begin{lemma}[Generalized ``$w$ move'' lemma]
\label{lem:general_w_move}
Let $\{w^{(j)}\}_{j=0}^{T}$ be a random sequence that satisfies
$\forall j\in\{0,...,T-1\}$, {\small
\begin{align*}
\sum_{i=1}^n \Big( \E[w_i^{(j+1)}|w^{(j)}]/ w_i^{(j)} - 1 \Big)^2 \leq  C_1^2, ~~
\sum_{i=1}^n \Big(\E[( w_i^{(j+1)}/ w_i^{(j)} - 1)^2 ~|~ w^{(j)}]\Big)^2 \leq  C_2^2, ~~
|w_i^{(j+1)} / w_i^{(j)} -  1 | \leq  1/4.
\end{align*}}
Let $\{v^{(j)}\}_{j=0}^T$ be another random sequence such that $\forall j\in [T]$, $v^{(j)}$ only depends on $w^{(j)}$ and $v^{(j-1)}$. Let $\psi$ be defined as of Definition~\ref{def:psi}, where the parameter $\epsilon_{\mathrm{mp}} < 1/4$. Let $g\in \R^n$ be a sequence that is non-increasing, i.e., $g_1\geq g_2 \geq \cdots \geq g_n$. Let $\pi_j: [n]\to [n]$ be the sorting $\psi (  w^{(j+1)}_{\pi(i)} / v^{(j)}_{\pi(i)} - 1 ) \geq \psi (  w^{(j+1)}_{\pi(i+1)} / v^{(j)}_{\pi(i+1)} - 1 )$.

Then $\forall j\in\{0,...,T-1\}$, in the $j$-th iteration, for any $w^{(j)}$ and $v^{(j)}$ we have
\begin{align*}
    \sum_{i=1}^n g_i \cdot \E \left[
\psi (  w^{(j+1)}_{\pi_j(i)} / v^{(j)}_{\pi_j(i)} - 1 ) - \psi (   w^{(j)}_{\pi_j(i)} / v^{(j)}_{\pi_j(i)} - 1 ) ~\Big|~ w^{(j)}, v^{(j)}\right] = O(C_1 + C_2 / \epsilon_{\mathrm{mp}}) \cdot \|g\|_2.
\end{align*}
\end{lemma}

\subsubsection{\texorpdfstring{$v,\wt{v}$}{} move}
\label{sec:v_wt_v_move_matrix_update}
The goal of this section is to prove Lemma~\ref{lem:v_wt_v_move_matrix_update}.

\begin{lemma}[$v,\wt{v}$ move]\label{lem:v_wt_v_move_matrix_update}
In the $j$-th iteration, we have,
 \begin{align*}
 \sum_{i = 1}^n g_i \cdot ( y^{(j)}_{\pi(i)} - x^{(j+1)}_{\tau(i)} ) \geq \Omega\Big( \frac{\epsilon_{\far}^2 \cdot k_j \cdot g_{k_j} }{ \epsilon_{\mathrm{mp}} \log n}\Big),
\end{align*}
in which $\pi,\tau:[n]\to [n]$ are permutations such that $y^{(j)}_{\pi(i)} \geq y^{(j)}_{\pi(i+1)}$ and $x^{(j+1)}_{\tau(i)} \geq x^{(j+1)}_{\tau(i+1)}$.
\end{lemma}
\begin{proof}
\noindent \textbf{Case 1, $k_j=0$.} When $k_j=0$, we didn't enter the \textsc{MatrixUpdate} procedure, so $v^{(j+1)}=v^{(j)}$. If we further enter the \textsc{PartialMatrixUpdate} procedure and change $\wt{v}$, we have that $\forall i \in [n]$,
\begin{align*}
   y^{(j)}_i - x^{(j+1)}_i
=  ~ \psi( w^{(j+1)}_{i} / \wt{v}^{(j)}_{i} - 1 )- \psi( w^{(j+1)}_{i} / \wt{v}^{(j+1)}_{i} - 1 )
\geq ~ 0,
\end{align*}
where the first step follows by the definition of $x^{(j+1)}_i$ and $y^{(j)}_i$ (Definition~\ref{def:x_y_matrix_update}) and the fact that $v$ do not change, and the last step follows by Corollary~\ref{cor:potential_not_increase_after_partialmatrixUpdate}.

%It basically says that $y^{(j)}$ coordinate-wise dominates $x^{(j+1)}$. Therefore, after sorting, 
Thus $y^{(j)}_i \geq x^{(j+1)}_i$, $\forall i \in [n]$. Since $g_i$ and $y^{(j)}_{\pi(i)}$ are both non-increasing, we have
\begin{align*}
\sum_{i = 1}^n g_i \cdot ( y^{(j)}_{\pi(i)} - x^{(j+1)}_{\tau(i)} ) \geq &~ \sum_{i = 1}^n g_i \cdot ( y^{(j)}_{\tau(i)} - x^{(j+1)}_{\tau(i)} )
\geq  0 
=  \Omega( \epsilon_{\far}^2 \cdot k_j \cdot g_{k_j} / ( \epsilon_{\mathrm{mp}} \log n) ),
\end{align*}
where the last step follows by $k_j=0$.

\noindent \textbf{Case 2, $k_j\neq 0$.}
When $k_j\neq 0$, we must have entered \textsc{MatrixUpdate}, and by Fact~\ref{fac:bound_on_k_j}, we must have $k_j\geq n^a$. 
By Part 3 of Fact~\ref{fac:characterization_MatrixUpdate}, the difference between $x^{(j+1)}$ and $y^{(j)}$ is that in coordinates $i\in \pi( [ k_j ] )$, $x^{(j+1)}_i$ is cleared to $0$, and in other coordinates $x^{(j+1)}_i$ is the same with $y^{(j)}_i$. So we have 
\begin{align}\label{eq:matrix_split}
 \sum_{i = 1}^n g_i \cdot ( y^{(j)}_{\pi(i)} - x^{(j+1)}_{\tau(i)} )
= & ~ \sum_{i = 1}^n g_i \cdot ( y^{(j)}_{\pi(i)} - y^{(j)}_{\pi(i + k_j)} ).
\end{align}
Note that when the subscripts are out of range, we define $y^{(j)}_{\pi(n+1)}=\dots=y^{(j)}_{\pi(n+k_j)} = 0$. 

Part 2 of Fact~\ref{fac:characterization_MatrixUpdate} shows that 
\begin{align} \label{eq:lower_bound_on_y_pi_kj_matrix_update}
y^{(j)}_{\pi(k_j)}  \geq \epsilon_{\far}^2 / (3200\epsilon_{\mathrm{mp}}).
\end{align}
Part 1 of Fact~\ref{fac:characterization_MatrixUpdate} shows that either $k_j=n$ or $y_{\pi(k_j)}^{(j)} < (1-1/\log n)y_{\pi(k_j/1.5)}^{(j)}$. If $k_j=n$, we let $L=k_j=n$, otherwise we let $L=k_j/1.5$.
The $L$ we choose always satisfied that
for all $i\in [L]$,
\begin{align}\label{eq:matrix_r1.5_at_most_r}
 y^{(j)}_{\pi (i) } - y^{(j)}_{\pi (i+k_j) } 
\geq  ~ y^{(j)}_{\pi(L)}  - y^{(j)}_{\pi(1+k_j)} 
\geq  ~ \epsilon_{\far}^2 / (3200\epsilon_{\mathrm{mp}}\log n) ,
\end{align}
where the first step follows by $y^{(j)}_{\pi(i)}$ is non-increasing, the second step is true because:
\begin{enumerate}
    \item In the case of $k_j=n$, we have $y^{(j)}_{\pi(L)}=y^{(j)}_{\pi(k_j)}\geq \frac{\epsilon_{\far}^2}{3200\epsilon_{\mathrm{mp}}}$ by Eq.~\eqref{eq:lower_bound_on_y_pi_kj_matrix_update} and $y^{(j)}_{\pi(k_j+1)}=y^{(j)}_{\pi(n+1)}=0$.
    \item In the case of $y_{\pi(k_j)}^{(j)} < (1-1/\log n)y_{\pi(k_j/1.5)}^{(j)}$, we have %from this inequality we can directly infer that $(y_{\pi(k_j/1.5)}^{(j)}-y_{\pi(k_j)}^{(j)})>y_{\pi(k_j/1.5)}^{(j)}/\log n$, then we have
    \begin{align*}
    y^{(j)}_{\pi(L)}  - y^{(j)}_{\pi(1+k_j)} \geq  y^{(j)}_{\pi(k_j/1.5)}  - y^{(j)}_{\pi(k_j)} 
    \geq  y^{(j)}_{\pi(k_j/1.5)}/\log n 
    \geq \epsilon_{\far}^2 / (3200\epsilon_{\mathrm{mp}}\log n),
    \end{align*}
    where the second step follows from the inequality of $y^{(j)}_{\pi(k_j)}$, and the third step follows from Eq.~\eqref{eq:lower_bound_on_y_pi_kj_matrix_update} and the fact that $y^{(j)}_{\pi(i)}$ is non-increasing.
\end{enumerate}

Putting it all together, we have
 \begin{align*}
 \sum_{i = 1}^n g_i \cdot ( y^{(j)}_{\pi(i)} - x^{(j+1)}_{\tau(i)} )
\geq & ~  \sum_{i = 1}^n g_i \cdot ( y^{(j)}_{\pi(i)} - y^{(j)}_{\pi(i+k_j)} )  
\geq  \sum_{i = 1}^{L} g_i \cdot ( y^{(j)}_{\pi(i)} - y^{(j)}_{\pi(i+k_j)} ) \\
\geq & ~ \sum_{i = 1}^{L} g_i \cdot \Big( \frac{\epsilon_{\far}^2}{3200\epsilon_{\mathrm{mp}}\log n} \Big) 
=  \Omega\Big( \frac{\epsilon_{\far}^2 \cdot k_j \cdot g_{k_j} }{ \epsilon_{\mathrm{mp}} \log n}\Big) ,
\end{align*}
where the first step is by Eq.~\eqref{eq:matrix_split}, the second step follows from $y^{(j)}_{\pi(i)}$ is non-increasing and thus all terms $\geq 0$, the third step is by Eq.~\eqref{eq:matrix_r1.5_at_most_r}, and the last step follows from $g_{L}\geq g_{k_j}$ and $L=\Omega(k_j)$.
\end{proof}

\subsubsection{\texorpdfstring{$\ell_2$}{}-norm of \texorpdfstring{$g$}{}}
\begin{lemma}[$\ell_2$-norm of $g$]\label{lem:l2_norm_g_matrix_update_improved}
$g \in \R^n$ (Definition~\ref{def:g_matrix_update}) satisfies $\|g\|_2 = O(n^{-a/2} + n^{\omega-5/2})$.
\end{lemma}
\begin{proof}
%Since function $g$ behaves differently when $i\leq n^a$ and $i > n^a$. We will sum from two parts.
For $i\leq n^a$, we have $\sum_{i=1}^{n^a} g_i^2  = \sum_{i=1}^{n^a} n^{-2a} = n^{-a}$.

For $i > n^a$, note that there exists $i \in [n]$ such that $i > n^a$ implies $a<1$, so we have
\begin{align*}
     \sum_{i=n^{a}+1}^{n} g_i^2
    = & ~ \sum_{i=n^{a}+1}^{n} i^{\frac{2(\omega-2)}{1-a}-2}\cdot n^{-\frac{2a(\omega-2)}{1-a}}
    =  O(1) \cdot \int_{n^{a+1}}^{n} x^{\frac{2(\omega-2)}{1-a}-2}\cdot n^{-\frac{2a(\omega-2)}{1-a}} \mathrm{d}x \\
    = & ~ O(1) \cdot \max\{ n^{\frac{2(\omega-2)}{1-a}-1}, n^{\frac{2a(\omega-2)}{1-a}-a} \}
    \cdot n^{-\frac{2a(\omega-2)}{1-a}} 
    =  O( n^{2\omega-5} + n^{-a}) .
\end{align*}

Therefore, we have $\|g\|_2 = O(n^{-a/2} + n^{\omega-5/2})$.
\end{proof}

\subsection{Amortized analysis for \textsc{PartialMatrixUpdate}}
%%%%%%%%%%% def of x,y

\subsubsection{Definitions}

%%%%%%%%%%% def of x,y
\begin{definition}[$x$ and $y$ for \textsc{PartialMatrixUpdate}]
\label{def:x_y_partial_matrix_update}
For any $j\in \{0, 1, \dots, T-1\}$, and for any $i \in [n]$, we define $x^{(j)}_{i}$ and $y^{(j)}_{i}$ as follows:
\begin{align*}
x^{(j)}_{i} := \psi ( w^{(j)}_{i} / \wt{v}^{(j)}_{i} - 1 ), ~~~~ %\label{eq:def_x_partial_matrix_update} \\
y^{(j)}_{i} := \psi ( w^{(j+1)}_{i} / \wt{v}^{(j)}_{i} - 1 ), %\label{eq:def_y_partial_matrix_update}
%x^{(j+1)}_{i} & ~ = \psi ( w^{(j+1)}_{i} / \wt{v}^{(j+1)}_{i} - 1 ),\notag
\end{align*}
where $v^{(j)}$, $\wt{v}^{(j)}$ and $w^{(j)}$ are defined as of Definition~\ref{def:w_j_v_j_wt_v_j}.
\end{definition}
Note that the difference between $x^{(j)}_i$ and $y^{(j)}_i$ is that $w$ is changing. The difference between $y^{(j)}_i$ and $x^{(j+1)}_i$ is that $\wt{v}$ is changing.

\begin{definition}[Sorting permutations for \textsc{PartialMatrixUpdate}]
For any $j\in \{0,1,\dots,T\}$, let $\tau_j$ be the permutation that $x^{(j)}_{\tau_j(i)} \geq x^{(j)}_{\tau_j(i+1)}$, and let $\pi_j$ be the permutation that $y^{(j)}_{\pi(i)} \geq y^{(j)}_{\pi(i+1)}$.

When it is clear from the context that we are only arguing about the $j$-th iteration, for simplicity we assume the coordinates of vector $x^{(j)} \in \R^n$ are sorted such that $x^{(j)}_{i} \geq x^{(j)}_{i+1}$. And we use $\tau$ and $\pi$ to denote the permutations such that $x^{(j+1)}_{\tau(i)} \geq x^{(j+1)}_{\tau(i+1)}$ and $y^{(j)}_{\pi(i)} \geq y^{(j)}_{\pi(i+1)}$.
\end{definition}

%%%%%%%%% def of g
\begin{definition}[$g$ for \textsc{PartialMatrixUpdate}]
\label{def:g_partial_matrix_update}
For some $\wt{a} \leq \alpha\cdot a$, where $\alpha$ is the dual exponent of matrix multiplication, and $a$ is the parameter for \textsc{MatrixUpdate}, we define $g\in \R^{n}$ as follows:
\begin{align*}%\label{eq:def_g_partial_matrix_update}
g_i=
\begin{cases}
n^{-\wt{a}}, &\mathrm{~if~}i \leq n^{\wt{a}};\\
i^{\frac{a(\omega-2)}{a-\wt{a}}-1}\cdot n^{-\frac{a\wt{a}(\omega-2)}{a-\wt{a}}}, & \mathrm{~if~}i\in (n^{\wt{a}},n^a];\\
0 & \mathrm{~if~} i > n^a.
\end{cases}
\end{align*}
\end{definition}
Note that $g$ is non-increasing. For all $\wt{k}_j\in (n^{\wt{a}},n^a]$, $(\wt{k}_j\cdot g_{\wt{k}_j} n^{2a}) =\wt{k}_j^{\frac{a(\omega-2)}{a-\wt{a}}}\cdot n^{2a-\frac{a\wt{a}(\omega-2)}{a-\wt{a}}}$ is an upper bound of $\Tmat(n^a, n^a, \wt{k}_j)$ of multiplying a $n^a \times n^a$ matrix with a $n^a \times \wt{k}_j$ matrix. For more details please refer to \cite{gu18}. %And thus $\Tmat(n,n^a,\wt{k}_j)\leq n^{1-a}\cdot \Tmat(n^a,n^a,\wt{k}_j)\leq (\wt{k}_j\cdot g_{\wt{k}_j} n^{1+a})$. Note that by Lemma~\ref{lem:improved:sparsity_guarantee} we indeed have $\wt{k}_j \leq 2n^a$ when entering \textsc{PartialMatrixUpdate}, and we ignore the factor $2$ since it can only increase the final amortized time by a constant factor. 

%%%%%%%%% def of Phi
\begin{definition}[Potential function $\Phi$ for \textsc{PartialMatrixUpdate}]
\label{def:phi_partial_matrix_update}
We define the potential function in the $j$-th iteration as
\begin{align*}%\label{eq:def_phi_partial_matrix_update}
\Phi_j = \sum_{i = 1}^n g_i \cdot x^{(j)}_{\tau_j(i)}.
\end{align*}
%where $\tau_j(i)$ is the permutation such that $x^{(j)}_{\tau_j(i)} \geq x^{(j)}_{\tau_j(i+1)}$.
\end{definition}
Note that we always have $\Phi_j\geq 0$ since $\forall i$, $g_i$ and $x_i^{(j)}$ are both non-negative.

\subsubsection{Main result}

\begin{lemma}[Amortized time for \textsc{PartialMatrixUpdate}] \label{lem:main_amortize_partial_matrix_update}
Let sequences $\{w^{(j)}\}_{j=0}^T$, $\{v^{(j)}\}_{j=0}^T$, $\{\wt{v}^{(j)}\}_{j=0}^T$ be defined as of Definition~\ref{def:w_j_v_j_wt_v_j}, let $\wt{k}_j$ be defined as of Definition~\ref{def:k_wt_k_p_wt_p},
and let $\{x^{(j)}\}_{j=0}^T$, $\{y^{(j)}\}_{j=0}^T$, $g$, $\Phi$ be defined as of Definition~\ref{def:x_y_matrix_update}, \ref{def:g_partial_matrix_update} and \ref{def:phi_partial_matrix_update}.
%in Equation~\eqref{eq:def_x_partial_matrix_update},~\eqref{eq:def_y_partial_matrix_update},~\eqref{eq:def_g_partial_matrix_update},~\eqref{eq:def_phi_partial_matrix_update}
If we further have the condition that the input sequence satisfies the following: $\forall j\in\{0,...,T-1\}$ {\small
\begin{align*}
\sum_{i=1}^n ( \E[w_i^{(j+1)}|w^{(j)}]/ w_i^{(j)} - 1 )^2 \leq  C_1^2, ~~~
\sum_{i=1}^n (\E[( w_i^{(j+1)}/ w_i^{(j)} - 1)^2 ~|~ w^{(j)}])^2 \leq  C_2^2, ~~~
|w_i^{(j+1)} / w_i^{(j)} -  1 | \leq  1/4.
\end{align*}}

Then, we have that in expectation 
\begin{align*}
\frac{1}{T} \sum_{j=1}^T \wt{k}_jg_{\wt{k}_j} =  O\Big((C_1/\epsilon_{\mathrm{mp}} + C_2/\epsilon_{\mathrm{mp}}^2)\cdot \log n \cdot \|g\|_2\Big).
\end{align*}

Further, combining with Lemma~\ref{lem:partial_matrix_update_time_improved}, the expected amortized running time per iteration of \textsc{PartialMatrixUpdate} is
\[O\Big((C_1/\epsilon_{\mathrm{mp}} + C_2/\epsilon_{\mathrm{mp}}^2)\cdot (n^{1+a-\wt{a}/2} + n^{1+(\omega-3/2)a})\log n\Big).
\]
\end{lemma}
\begin{proof}
%First note that in the $j$-th iteration, the value of the potential $\Phi_j$ depends on $w^{(j)}$ and $\wt{v}^{(j)}$. And the value of $\wt{v}^{(j)}$ is affected by both the procedure \textsc{MatrixUpdate} and the procedure \textsc{PartialMatrixUpdate}. 
Similar to the proof of Lemma~\ref{lem:main_amortize_matrix_update}, we upper bound how much the potential function can increase due to changing $w^{(j)}$ to $w^{(j+1)}$ (in Section~\ref{sec:w_move_partial_matrix_update}), and also lower bound how much the potential function can decrease because of changing $\wt{v}^{(j)}$ to $\wt{v}^{(j+1)}$ (in Section~\ref{sec:v_wt_v_move_partial_matrix_update}).

%In the beginning $\wt{v}^{(0)}=w^{(0)}$, so $\Phi_0=0$. Also note that $\Phi_j\geq 0$, $\forall j\in [T]$.
Similar to Lemma~\ref{lem:main_amortize_matrix_update}, we have
\begin{align*}
    0 \leq \E[\Phi_T]-\Phi_0 %= &~ \sum_{j=0}^{T-1}\E[\Phi_{j+1}-\Phi_j]
    %=  \sum_{j=0}^{T-1}\sum_{i=1}^n g_i \cdot \E\left[ x^{(j+1)}_{\tau(i)} -  x^{(j)}_{i} \right] \\
    = & ~ \sum_{j=0}^{T-1} \Bigg(\sum_{i = 1}^n g_i \cdot \underbrace{ \E\left[ y^{(j)}_{\pi(i)} - x^{(j)}_{i} \right] }_{w\text{~move}} - \sum_{i = 1}^n g_i \cdot \underbrace{ \E\left[ y^{(j)}_{\pi(i)} - x^{(j+1)}_{\tau(i)} \right]  }_{\wt{v}\text{~move}}\Bigg)\\
    %\leq &~ \sum_{j = 0}^{T-1}  \Bigg( O(C_1 + C_2 / \epsilon_{\mathrm{mp}}) \cdot \|g\|_2 - \sum_{i = 1}^n g_i \cdot \underbrace{ \E\left[ y^{(j)}_{\pi(i)} - x^{(j+1)}_{\tau(i)} \right] }_{\wt{v}\text{~move}} \Bigg) \\
    \leq &~ \sum_{j = 0}^{T-1}  \Big( O(C_1 + C_2 / \epsilon_{\mathrm{mp}}) \cdot \|g\|_2 - \Omega (\epsilon_{\mathrm{mp}} \wt{k}_j g_{\wt{k}_j} / \log n ) \Big) \\
    = & ~ T \cdot O(C_1 + C_2 / \epsilon_{\mathrm{mp}}) \|g\|_2 - \sum_{j=1}^T \Omega (\epsilon_{\mathrm{mp}} \wt{k}_j g_{\wt{k}_j} / \log n),
\end{align*}
%where the first step follows from splitting terms and the fact that $\Phi_0$ is deterministic, the second step follows from the definition of $\Phi$, the third step follows from splitting terms, 
where the third step follows from Lemma~\ref{lem:w_move_partial_matrix_update} which states that $\forall w^{(j)}, v^{(j)}, \wt{v}^{(j)}$, we have 
\begin{align*}
\sum_{i=1}^n g_i \cdot \E \left[ y^{(j)}_{\pi(i)} - x^{(j)}_{i} ~\Big|~ w^{(j)}, v^{(j)}, \wt{v}^{(j)}\right] \leq O(C_1 + C_2 / \epsilon_{\mathrm{mp}}) \cdot \|g\|_2,
\end{align*}
then this upper bound also holds for unconditional expectation, the third step also follows from Lemma~\ref{lem:v_wt_v_move_partial_matrix_update} which states that $\sum_{i = 1}^n g_i \cdot \left( y^{(j)}_{\pi(i)} - x^{(j+1)}_{\tau(i)} \right) \geq \Omega( \epsilon_{\mathrm{mp}} \wt{k}_j g_{\wt{k}_j} / \log n)$.

Therefore, we have
\begin{align*}
\frac{1}{T} \sum_{j=1}^T \wt{k}_jg_{\wt{k}_j} = O\Big((C_1/\epsilon_{\mathrm{mp}} + C_2/\epsilon_{\mathrm{mp}}^2)\cdot \log n \cdot \|g\|_2 \Big).
\end{align*}

Using Lemma~\ref{lem:partial_matrix_update_time_improved}, we have that the expected amortized running time per iteration of \textsc{PartialMatrixUpdate} is
\begin{align*}
\frac{1}{T} \sum_{j=1}^T \Tmat(n,n^a,\wt{k}_j) \leq &~ \frac{1}{T} \sum_{j=1}^T n^{1-a}\cdot \Tmat(n^a,n^a,\wt{k}_j) 
\leq \frac{n^{1+a}}{T} \sum_{j=1}^T \wt{k}_j g_{\wt{k}_j}  \\
= &~ O\Big((C_1/\epsilon_{\mathrm{mp}} + C_2/\epsilon_{\mathrm{mp}}^2)\cdot n^{1+a}\log n \cdot \|g\|_2 \Big)\\
= &~ O\Big((C_1/\epsilon_{\mathrm{mp}} + C_2/\epsilon_{\mathrm{mp}}^2)\cdot  (n^{1+a-\wt{a}/2} + n^{1+(\omega-3/2)a})\log n  \Big),
\end{align*}
where the first step follows from the fact that we can always divide a $n\times n^a$ matrix into $n^{1-a}$ copies of $n^a \times n^a$ matrices, the second step follows from Definition~\ref{def:g_partial_matrix_update} of $g$ which gives that $\Tmat(n^a,n^a,\wt{k}_j)\leq  n^{2a} \cdot \wt{k}_j g_{\wt{k}_j}$ for $n^{\wt{a}} \leq \wt{k}_j \leq n^a$, and we indeed have $\wt{k}_j \geq n^{\wt{a}}$ (Fact~\ref{fac:bound_wt_k_j}) and $\wt{k}_j \leq 2n^a$ (Lemma~\ref{lem:improved:sparsity_guarantee}) when entering \textsc{PartialMatrixUpdate}. (We ignore the $2$ factor since it can only increase the final amortized time by a constant factor.)
The fourth step follows from Lemma~\ref{lem:l2_norm_g_partial_matrix_update_improved} that $\|g\|_2 =  O(n^{-\wt{a}/2} + n^{a\omega-5a/2})$.

\end{proof}

\subsubsection{\texorpdfstring{$w$}{} move}
\label{sec:w_move_partial_matrix_update}
The goal of this section is to prove Lemma~\ref{lem:w_move_partial_matrix_update}.
\begin{lemma}[$w$ move]\label{lem:w_move_partial_matrix_update}
In the $j$-th iteration, for any possible values $w^{(j)}$, $v^{(j)}$, and $\wt{v}^{(j)}$, we have 
\begin{align}
\label{eq:w_move_partial_matrix_update}
\sum_{i=1}^n g_i \cdot \E \left[ y^{(j)}_{\pi(i)} - x^{(j)}_{i} ~\Big|~ w^{(j)}, v^{(j)}, \wt{v}^{(j)}\right] \leq O(C_1 + C_2 / \epsilon_{\mathrm{mp}}) \cdot \|g\|_2.
\end{align}
\end{lemma}
\begin{proof}
For simplicity, in this proof we write $\E[ ~ \cdot ~ ]$ as a shorthand of $\E[ ~ \cdot ~|~ w^{(j)}, v^{(j)}, \wt{v}^{(j)} ~ ]$.

%Observe that since the non-negative values $x^{(j)}_{i}$ are sorted in descending order, and $g$ is also non-increasing, we have
%\begin{align}\label{eq:w_move_partial_matrix_update:sum_g_i_x_pi_i_leq_sum_g_i_x_i}
%\sum_{i = 1}^n g_i x^{(j)}_{\pi(i)} \leq \sum_{i = 1}^n g_i x^{(j)}_{i} .
%\end{align}

Similar to the proof of Lemma~\ref{lem:w_move_matrix_update}, we have
\begin{align*}
    \sum_{i=1}^n g_i \cdot \E [ y^{(j)}_{\pi(i)} - x^{(j)}_{i} ] \leq  \sum_{i=1}^n g_i \cdot \E [ y^{(j)}_{\pi(i)} - x^{(j)}_{\pi(i)} ]
    =  \sum_{i=1}^n g_i \cdot \E [ \psi (   w^{(j+1)}_{\pi(i)} / \wt{v}^{(j)}_{\pi(i)} - 1 ) - \psi (  w^{(j)}_{\pi(i)} / \wt{v}^{(j)}_{\pi(i)} - 1 ) ]
\end{align*}
where the first step follows from that the non-negative values $x^{(j)}_{i}$ are sorted in descending order, and $g$ is also non-increasing, the second step follows from the definitions of $x^{(j)}$ and $y^{(j)}$ (Definition~\ref{def:x_y_partial_matrix_update}.

Now $\sum_{i=1}^n g_i\cdot \E[y_{\pi(i)}^{(j)}-x_i^{(j)}]\leq O(C_1+C_2/\epsilon_{\mathrm{mp}})\|g\|_2$ directly follows from Lemma~\ref{lem:bounding_E_psi_w_minus_wt_v_partial_matrix_update}.
\end{proof}

It remains to prove the following Lemma.
\begin{lemma}
\label{lem:bounding_E_psi_w_minus_wt_v_partial_matrix_update}
In the $j$-th iteration, for any $w^{(j)}$, $v^{(j)}$, and $\wt{v}^{(j)}$ we have
\begin{align*}
    \sum_{i=1}^n g_i \cdot \E \left[
    \psi ( w^{(j+1)}_{\pi(i)} / \wt{v}^{(j)}_{\pi(i)} - 1 ) - \psi (  w^{(j)}_{\pi(i)} / \wt{v}^{(j)}_{\pi(i)} - 1 ) ~\Big|~ w^{(j)}, v^{(j)}, \wt{v}^{(j)}\right] = O(C_1 + C_2 / \epsilon_{\mathrm{mp}}) \cdot \|g\|_2.
\end{align*}
\end{lemma}

\begin{proof}
The proof of this lemma is exactly the same as that of Lemma~\ref{lem:bounding_E_psi_w_minus_v_matrix_update}, just replace all $v$ with $\wt{v}$ in the proof of Lemma~\ref{lem:bounding_E_psi_w_minus_v_matrix_update}.
\end{proof}

\subsubsection{\texorpdfstring{$\wt{v}$}{} move}
\label{sec:v_wt_v_move_partial_matrix_update}
The goal of this section is to prove Lemma~\ref{lem:v_wt_v_move_partial_matrix_update}.
\begin{lemma}[$\wt{v}$ move]\label{lem:v_wt_v_move_partial_matrix_update}
In the $j$-th iteration, we have,
 \begin{align*}
 \sum_{i = 1}^n g_i \cdot ( y^{(j)}_{\pi(i)} - x^{(j+1)}_{\tau(i)} ) \geq \Omega( \epsilon_{\mathrm{mp}} \wt{k_j} g_{\wt{k_j}} / \log n).
\end{align*}
\end{lemma}
\begin{proof}
\noindent \textbf{Case 1.} If we enter the \textsc{MatrixUpdate} procedure, we have $\wt{k}_j = 0$ since we won't enter the else branch in Line~\ref{alg:line:else_for_matrix_update} of \textsc{UpdateV} (Algorithm~\ref{alg:update_v_improved}). From Part 3 of Fact~\ref{fac:characterization_MatrixUpdate}, we know that $\forall i\leq k_j$, $\wt{v}^{(j+1)}_{\pi(i)}=w^{(j+1)}_{\pi(i)}$, and $\forall i > k_j$, $\wt{v}_{\pi(i)}^{(j+1)}=\wt{v}_{\pi(i)}^{(j)}$. Therefore, $\forall i \in [n]$, 
\begin{align*}
    x^{(j+1)}_i =  \psi ( w^{(j+1)}_{i} / \wt{v}^{(j+1)}_{i} - 1 )
    \leq \psi ( w^{(j+1)}_{i} / \wt{v}^{(j)}_{i} - 1 )
    =  y^{(j)}_i,
\end{align*}
where the first and the third steps follow by the definition of $x^{(j+1)}$ and $y^{(j)}$ (Definition~\ref{def:x_y_partial_matrix_update}). This means $y^{(j)}_i \geq x^{(j+1)}_i$, $\forall i \in [n]$. Since $g_i$ and $y^{(j)}_{\pi(i)}$ are both non-increasing, we have
\begin{align*}
    \sum_{i = 1}^n g_i \cdot ( y^{(j)}_{\pi(i)} -  x^{(j+1)}_{\tau(i)} ) \geq  \sum_{i = 1}^n g_i \cdot ( y^{(j)}_{\tau(i)} -  x^{(j+1)}_{\tau(i)} )
    \geq  0 
    =  \Omega( \epsilon_{\mathrm{mp}} \wt{k_j} g_{\wt{k_j}} / \log n),
\end{align*}
where the last step follows by $\wt{k}_j = 0$.

\noindent \textbf{Case 2.} If we do not enter both the \textsc{MatrixUpdate} and the \textsc{PartialMatrixUpdate} procedure, nothing happens and $x^{(j+1)}$ is the same as $y^{(j)}$, this statement also holds.

\noindent \textbf{Case 3.} Now we only need to consider the case where we enter the \textsc{PartialMatrixUpdate} procedure. %In this case $\wt{k}_j\neq 0$, so by Fact~\ref{fac:bound_wt_k_j}, we must have $\wt{k}_j\geq n^{\wt{a}}$.
By Part 3 of Fact~\ref{fac:characterization_PartialMatrixUpdate}, $x^{(j+1)}$ satisfies that in coordinates $i\in \pi( [ \wt{k}_j ] )$, $x^{(j+1)}_i\leq \frac{\epsilon_{\mathrm{mp}}}{200\log n}$, and in coordinates $i\notin \pi( [ \wt{k}_j ] )$, $x^{(j+1)}$ is the same with $y^{(j)}$. So we decompose $x^{(j+1)}$ into two pieces $x^{(j+1)}=x_1+x_2$, where $x_1$ copies the values on coordinates $i \in \pi( [ \wt{k}_j ] )$ and has 0 on other coordinates, and $x_2$ copies the values on coordinates $i \notin \pi( [ \wt{k}_j ] )$ and has 0 on other coordinates. And when the subscripts are out of range, we define $y^{(j)}_{\pi(n+1)}=\dots=y^{(j)}_{\pi(n+k_j)} = 0$. We have 
{\small
\begin{align}\label{eq:partial_split_x_1_x_2}
  \sum_{i = 1}^n g_i \cdot ( y^{(j)}_{\pi(i)} - x^{(j+1)}_{\tau(i)} ) 
= & ~ \sum_{i = 1}^n g_i \cdot ( y^{(j)}_{\pi(i)} - x_{1,\tau(i)} - x_{2,\tau(i)})
\geq \sum_{i = 1}^n g_i \cdot ( y^{(j)}_{\pi(i)} - x_{2,\tau(i)}) - \sum_{i = 1}^{\wt{k}_j} g_i\cdot \frac{\epsilon_{\mathrm{mp}}}{200\log n}\notag \\
= & ~ \sum_{i = 1}^n g_i \cdot ( y^{(j)}_{\pi(i)} - y^{(j)}_{\pi(i+\wt{k}_j)} ) - \sum_{i = 1}^{\wt{k}_j} g_i\cdot \frac{\epsilon_{\mathrm{mp}}}{200\log n}\notag \\
\geq & ~ \sum_{i = 1}^n g_i \cdot ( y^{(j)}_{\pi(i)} - y^{(j)}_{\pi(i+\wt{k}_j)} ) - 1.5\sum_{i = 1}^{\wt{k}_j/1.5} g_i\cdot \frac{\epsilon_{\mathrm{mp}}}{200\log n},
\end{align}}
where the second step follows by $x_{1\tau(i)}\leq \frac{\epsilon_{\mathrm{mp}}}{200\log n}$ for $\tau(i)\in \pi([\wt{k}_j])$ and $x_{1\tau(i)}=0$ for $\tau(i)\notin \pi([\wt{k}_j])$, the third step follows by $x_{2, \tau(i)}=0$ for $\tau(i)\in \pi([\wt{k}_j])$ and $x_{2, \tau(i)}=y^{(j)}_{\tau(i)}$ for $i \notin \pi([\wt{k}_j])$, and the last step follows by $g$ is non-increasing.

Part 2 of Fact~\ref{fac:characterization_PartialMatrixUpdate} shows that 
\begin{align}
\label{eq:lower_bound_on_y_pi_kj_partial_matrix_update}
y^{(j)}_{\pi(\wt{k}_j)}  \geq \epsilon_{\mathrm{mp}} / 100
\end{align}
Part 1 of Fact~\ref{fac:characterization_PartialMatrixUpdate} shows that either $\wt{k}_j=n$ or $y_{\pi(\wt{k}_j)}^{(j)} < (1-1/\log n) \cdot y_{\pi(\wt{k}_j/1.5)}^{(j)}$. If $\wt{k}_j=n$, we let $L=\wt{k}_j=n$, otherwise we let $L=\wt{k}_j/1.5$. The $L$ we choose always satisfies that
for all $i\in [L]$,
\begin{align}\label{eq:partial_r1.5_at_most_r}
 y^{(j)}_{\pi (i) } - y^{(j)}_{\pi (i+\wt{k}_j) } 
\geq  ~ y^{(j)}_{\pi(L)}  - y^{(j)}_{\pi(1+\wt{k}_j)} 
\geq ~ \epsilon_{\mathrm{mp}} / (100 \log n),
\end{align}
where the first step follows by $y^{(j)}_{\pi(i)}$ is non-increasing, the second step is true because:
\begin{enumerate}
    \item In the case of $\wt{k}_j=n$, we have $y^{(j)}_{\pi(L)}=y^{(j)}_{\pi(\wt{k}_j)}\geq \epsilon_{\mathrm{mp}} / 100$ by Eq.~\eqref{eq:lower_bound_on_y_pi_kj_partial_matrix_update} and $y^{(j)}_{\pi(\wt{k}_j+1)}=y^{(j)}_{\pi(n+1)}=0$.
    \item In the case of $y_{\pi(\wt{k}_j)}^{(j)} < (1-1/\log n) \cdot y_{\pi(\wt{k}_j/1.5)}^{(j)}$, we have 
    \begin{align*}
    y^{(j)}_{\pi(L)}  - y^{(j)}_{\pi(1+\wt{k}_j)} \geq y^{(j)}_{\pi(\wt{k}_j/1.5)}  - y^{(j)}_{\pi(\wt{k}_j)} \geq y^{(j)}_{\pi(\wt{k}_j)}/\log n \geq \epsilon_{\mathrm{mp}} / (100 \log n),
    \end{align*}
    where the second step follows from the inequality of $y^{(j)}_{\pi(\wt{k}_j)}$, and the third step follows from Eq.~\eqref{eq:lower_bound_on_y_pi_kj_partial_matrix_update}.
\end{enumerate}

Putting it all together, we have
 \begin{align*}
 \sum_{i = 1}^n g_i \cdot ( y^{(j)}_{\pi(i)} - x^{(j+1)}_{\tau(i)} )
\geq & ~ \sum_{i = 1}^n g_i \cdot ( y^{(j)}_{\pi(i)} - y^{(j)}_{\pi(i+\wt{k}_j)} ) - 1.5\sum_{i = 1}^{\wt{k}_j/1.5} g_i\cdot \frac{\epsilon_{\mathrm{mp}}}{200\log n}  \\
\geq & ~ \sum_{i = 1}^{L} g_i \cdot ( y^{(j)}_{\pi(i)} - y^{(j)}_{\pi(i+\wt{k}_j)} ) - 1.5\sum_{i = 1}^{\wt{k}_j/1.5} g_i\cdot \frac{\epsilon_{\mathrm{mp}}}{200\log n}    \\
\geq & ~ \sum_{i = 1}^{L} g_i \cdot ( y^{(j)}_{\pi(i)} - y^{(j)}_{\pi(i+\wt{k}_j)} - 1.5 \frac{\epsilon_{\mathrm{mp}}}{200\log n} )  \\
\geq & ~ \sum_{i = 1}^{L} g_i \cdot ( \frac{\epsilon_{\mathrm{mp}}}{100\log n} - 1.5 \frac{\epsilon_{\mathrm{mp}}}{200\log n} ) 
=  \Omega( \epsilon_{\mathrm{mp}} \cdot \wt{k}_j \cdot g_{\wt{k}_j} / \log n ) .
\end{align*}
where first step is by  Eq.~\eqref{eq:partial_split_x_1_x_2}, the second step follows from $y^{(j)}_{\pi(i)}$ is non-increasing and thus all terms $\geq 0$, the third step follows from $L\geq \wt{k}_j/1.5$, the forth step is by Eq.~\eqref{eq:partial_r1.5_at_most_r}, and the last step follows from $g_{L}\geq g_{\wt{k}_j}$ and $L=\Omega(\wt{k}_j)$.
\end{proof}

\subsubsection{\texorpdfstring{$\ell_2$}{}-norm of \texorpdfstring{$g$}{}}

\begin{lemma}
\label{lem:l2_norm_g_partial_matrix_update_improved}
$g \in \R^n$ (Definition~\ref{def:g_partial_matrix_update}) satisfies $\|g\|_2 = O(n^{-\wt{a}/2} + n^{a\omega-5a/2})$.
\end{lemma}
\begin{proof}
The proof is same as that of Lemma~\ref{lem:l2_norm_g_matrix_update_improved}. We use $n^a$ to replace $n$, and $n^{\wt{a}}$ to replace $n^a$.
\end{proof}
\begin{comment}
\begin{proof}
Since function $g$ behaves differently when $i\leq n^{\wt{a}}$ and $i > n^{\wt{a}}$. We will sum from two parts.

For the first part, we have
\begin{align*}
\sum_{i=1}^{n^{\wt{a}}} g_i^2  = \sum_{i=1}^{n^{\wt{a}}} n^{-2\wt{a}} = n^{-{\wt{a}}}.
\end{align*}

For the second part, we have
\begin{align*}
     \sum_{i=n^{\wt{a}}+1}^{n^a} g_i^2
    = & ~ \sum_{i=n^{\wt{a}}+1}^{n^a} i^{\frac{2a(\omega-2)}{a-\wt{a}}-2}\cdot n^{-\frac{2a\wt{a}(\omega-2)}{a-\wt{a}}}\\
    = & ~ O(1) \int_{n^{\wt{a}+1}}^{n^a} x^{\frac{2a(\omega-2)}{a-\wt{a}}-2}\cdot n^{-\frac{2a\wt{a}(\omega-2)}{a-\wt{a}}} \mathrm{d}x\\
    = & ~ O(1) \cdot \max\{ n^{\frac{2a^2(\omega-2)}{a-\wt{a}}-a}, n^{\frac{2a\wt{a}(\omega-2)}{a-\wt{a}}-\wt{a}} \} \cdot n^{-\frac{2a\wt{a}(\omega-2)}{a-\wt{a}}}\\
    = & ~ O( n^{a(2\omega-5)} + n^{-\wt{a}}) .
\end{align*}

Therefore, we have $\|g\|_2 = O(n^{-\wt{a}/2} + n^{a\omega-5a/2})$.
\end{proof}
\end{comment}

%
\subsection{Amortized analysis for \textsc{VectorUpdate}}
\begin{lemma}[Amortized time for \textsc{VectorUpdate}]
\label{lem:main_amortize_vector_update}
Let sequences $\{h^{(j)}\}_{j=0}^T$, $\{g^{(j)}\}_{j=0}^T$, $\{\wt{g}^{(j)}\}_{j=0}^T$ be defined as of Definition~\ref{def:h_j_g_j_wt_g_j}, and let $p_j$ be defined as of Definition~\ref{def:k_wt_k_p_wt_p}.
If we further have the condition that the input sequence satisfies the following: $\forall j\in\{0,...,T-1\}$
{\small
\begin{align*}
\sum_{i=1}^n ( \E[h_i^{(j+1)}|h^{(j)}]/h_i^{(j)}- 1 )^2 \leq  C_4^2,  ~~~
\sum_{i=1}^n (\E[( h_i^{(j+1)}/h_i^{(j)} - 1 )^2 ~|~ h^{(j)}])^2 \leq  C_5^2, ~~~
| h_i^{(j+1)}/h_i^{(j)} -1| \leq 1/4.
\end{align*}}

Then, we have that in expectation 
\begin{enumerate}
    \item $\frac{1}{T} \sum_{j=1}^T p_jn^{1+o(1)} = O\big((C_4\epsilon_{\mathrm{mp}}/\epsilon_{\mathrm{\far}}^2 + C_5/\epsilon_{\mathrm{\far}}^2)\cdot \log n \cdot n^{1.5+o(1)}\big)$,
    \item $\frac{1}{T} \sum_{j=1}^T n^{2a} \cdot \mathbf{1}_{p_j>0} =  O\big((C_4\epsilon_{\mathrm{mp}}/\epsilon_{\mathrm{\far}}^2 + C_5/\epsilon_{\mathrm{\far}}^2)\cdot \log n \cdot n^{1.5a}\big)$.
\end{enumerate}

Further, combining with Lemma~\ref{lem:vector_update_time_improved}, the expected amortized running time per iteration of \textsc{VectorUpdate} is
\[
O\big((C_4\epsilon_{\mathrm{mp}}/\epsilon_{\mathrm{\far}}^2 + C_5/\epsilon_{\mathrm{\far}}^2)\cdot \log n \cdot n^{1.5+o(1)}\big).
\]
\end{lemma}

\begin{proof}
\noindent \textbf{Part 1.} For the first equation, we define $g\in \R^n$ to be $g_i=1$, $\forall i \in [n]$. Note that $g$ is non-increasing, $n^{1+o(1)}\cdot (p_jg_{p_j})=p_jn^{1+o(1)}$, and $\|g\|_2=\sqrt{n}$. Then we can use the same argument as Lemma~\ref{lem:main_amortize_matrix_update} for \textsc{MatrixUpdate} to prove that
\begin{align*}
\frac{1}{T} \sum_{j=1}^T p_jn^{1+o(1)} = &~ O\big((C_4\epsilon_{\mathrm{mp}}/\epsilon_{\mathrm{\far}}^2 + C_5/\epsilon_{\mathrm{\far}}^2)\cdot \log n \cdot n^{1+o(1)}\cdot \|g\|_2\big)\\
= &~ O\big((C_4\epsilon_{\mathrm{mp}}/\epsilon_{\mathrm{\far}}^2 + C_5/\epsilon_{\mathrm{\far}}^2)\cdot \log n \cdot n^{1.5+o(1)}\big).
\end{align*}

\noindent \textbf{Part 2.} For the second equation, we define $g\in \R^n$ to be
\begin{align*}
    g_i=
    \begin{cases}
    n^{-a}, & ~ i \leq n^a,\\
    i^{-1}, &~ n^a < i < n.
    \end{cases}
\end{align*}
Note that $g$ is non-increasing, and $n^{2a}\cdot \mathbf{1}_{p_j>0} = n^{2a}\cdot \mathbf{1}_{p_j\geq n^a} \leq n^{2a}\cdot (p_jg_{p_j})$ since analogous to Fact~\ref{fac:bound_on_k_j} we have that either $p_j=0$ or $p_j\geq n^a$. We also have
\begin{align*}
\|g\|_2^2 
\leq \frac{n^{a}}{n^{2a}} + \int_{n^{a}}^{n}x^{-2} \mathrm{d} x 
=  n^{-a} + n^{-a} - n^{-1}
=  O(n^{-a}).
\end{align*} Then we can use the same argument as Lemma~\ref{lem:main_amortize_matrix_update} for \textsc{MatrixUpdate} to prove that
\begin{align*}
\frac{1}{T} \sum_{j=1}^T n^{2a}\cdot \mathbf{1}_{p_j>0} = O\big((C_4/\epsilon_{\mathrm{mp}} + C_5/\epsilon_{\mathrm{mp}}^2) \log n \cdot n^{2a} \|g\|_2\big)
= O\big((C_4/\epsilon_{\mathrm{mp}} + C_5/\epsilon_{\mathrm{mp}}^2)\cdot \log n \cdot n^{1.5a}\big).
\end{align*}

\noindent \textbf{Combine Part 1 and Part 2.} Using Lemma~\ref{lem:vector_update_time_improved} and note that $n^{1.5a}\leq n^{1.5+o(1)}$, we have that the expected amortized running time per iteration of \textsc{VectorUpdate} is
\begin{align*}
\frac{1}{T} \sum_{j=1}^T (p_jn^{1+o(1)} + n^{2a}\mathbf{1}_{p_j>0})
= &~ O\big((C_4/\epsilon_{\mathrm{mp}} + C_5/\epsilon_{\mathrm{mp}}^2)\cdot \log n \cdot n^{1.5+o(1)}\big).
\end{align*}
\end{proof}

\subsection{Amortized analysis for \textsc{PartialVectorUpdate}}
\begin{lemma}[Amortized time for \textsc{PartialVectorUpdate}]
\label{lem:main_amortize_partial_vector_update}
Let sequences $\{h^{(j)}\}_{j=0}^T$, $\{g^{(j)}\}_{j=0}^T$, $\{\wt{g}^{(j)}\}_{j=0}^T$ be defined as of Definition~\ref{def:h_j_g_j_wt_g_j}, and let $\wt{p}_j$ be defined as of Definition~\ref{def:k_wt_k_p_wt_p}.
If we further have the condition that the input sequence satisfies the following: $\forall j\in\{0,...,T-1\}$
{\small
\begin{align*}
\sum_{i=1}^n ( \E[h_i^{(j+1)}|h^{(j)}]/h_i^{(j)}- 1 )^2 \leq  C_4^2,  ~~~
\sum_{i=1}^n (\E[( h_i^{(j+1)}/h_i^{(j)} - 1 )^2 ~|~ h^{(j)}])^2 \leq  C_5^2, ~~~
| h_i^{(j+1)}/h_i^{(j)} -1| \leq 1/4.
\end{align*}}

Then, we have that in expectation 
\begin{enumerate}
    \item $\frac{1}{T} \sum_{j=1}^T \wt{p}_jn^{1+o(1)} = O\big((C_4/\epsilon_{\mathrm{mp}} + C_5/\epsilon_{\mathrm{mp}}^2)\cdot \log n \cdot n^{1.5+o(1)}\big)$,
    \item $\frac{1}{T} \sum_{j=1}^T n^{2a} \cdot \mathbf{1}_{\wt{p}_j>0} =  O\big((C_4/\epsilon_{\mathrm{mp}} + C_5/\epsilon_{\mathrm{mp}}^2)\cdot \log n \cdot n^{2a-\wt{a}/2}\big)$.
\end{enumerate}

Further, combining with Lemma~\ref{lem:partial_vector_update_time_improved}, the expected amortized running time per iteration of \textsc{PartialVectorUpdate} is
\[
O\big((C_4\epsilon_{\mathrm{mp}}/\epsilon_{\mathrm{\far}}^2 + C_5/\epsilon_{\mathrm{\far}}^2)\cdot \log n \cdot (n^{1.5+o(1)}+n^{2a-\wt{a}/2})\big).
\]
%\begin{align*}
%\frac{1}{T} \sum_{\substack{1\leq j \leq T,\\p_j>0}} n^{2a} = O\Big((C_4/\epsilon_{\mathrm{mp}} + C_5/\epsilon_{\mathrm{mp}}^2)\cdot \log n \cdot n^{2a-\wt{a}/2}\Big).
%\end{align*}
\end{lemma}

\begin{proof}
First note that we always have $\wt{p}_j\leq 2n^a$ by Lemma~\ref{lem:improved:sparsity_guarantee}.

\noindent \textbf{Part 1.} For the first equation, we define $g\in \R^{n}$ to be $g_i=1$, $\forall i \in [2n^a]$, and $g_i=0$, $\forall i \notin [2n^a]$.
Note that $g$ is non-increasing, $n^{a+o(1)}\cdot (\wt{p}_jg_{\wt{p}_j})=\wt{p}_jn^{a+o(1)}$, and $\|g\|_2=\sqrt{2n^a}$. Then we can use the same argument as Lemma~\ref{lem:main_amortize_partial_matrix_update} for \textsc{PartialMatrixUpdate} to prove that
{\small
\begin{align*}
\frac{1}{T} \sum_{j=1}^T \wt{p}_jn^{a+o(1)} = O\big((C_4/\epsilon_{\mathrm{mp}} + C_5/\epsilon_{\mathrm{mp}}^2) \log n \cdot n^{a+o(1)} \cdot \|g\|_2\big)
= O\big((C_4/\epsilon_{\mathrm{mp}} + C_5/\epsilon_{\mathrm{mp}}^2) \cdot \log n \cdot n^{1.5a+o(1)} \big).
\end{align*}}

\noindent \textbf{Part 2.} For the second equation, we let $g\in \R^n$ be
\begin{align*}
    g_i=
    \begin{cases}
    n^{-\wt{a}}, & ~ i \leq n^{\wt{a}},\\
    i^{-1}, &~ n^{\wt{a}} < i \leq 2n^a, \\
    0, &~ i > 2n^a.
    \end{cases}
\end{align*}
Note that $g$ is non-increasing, and $n^{2a}\cdot \mathbf{1}_{p_j>0} = n^{2a}\cdot \mathbf{1}_{p_j\geq n^{\wt{a}}} \leq n^{2a}\cdot (p_jg_{p_j})$ since analogous to Fact~\ref{fac:bound_wt_k_j} we have that either $\wt{p}_j=0$ or $\wt{p}_j\geq n^{\wt{a}}$. We also have
\begin{align*}
\|g\|_2^2 
=  \frac{n^{\wt{a}}}{n^{2\wt{a}}} + \int_{n^{\wt{a}}}^{2n^a}x^{-2} \mathrm{d}x 
=  n^{-\wt{a}} + n^{-\wt{a}} - n^{-a}/2
=  O(n^{-\wt{a}}).
\end{align*}
Then we can use the same argument as the Lemma~\ref{lem:main_amortize_partial_matrix_update} for \textsc{PartialMatrixUpdate} to prove that
\begin{align*}
\frac{1}{T} \sum_{j=1}^T n^{2a}\cdot \mathbf{1}_{\wt{p}_j>0} = O\big((C_4/\epsilon_{\mathrm{mp}} + C_5/\epsilon_{\mathrm{mp}}^2) \log n \cdot n^{2a} \|g\|_2\big)
= O\big((C_4/\epsilon_{\mathrm{mp}} + C_5/\epsilon_{\mathrm{mp}}^2) \log n \cdot n^{2a-\wt{a}/2}\big).
\end{align*}

\noindent \textbf{Combine Part 1 and Part 2.} Using Lemma~\ref{lem:partial_vector_update_time_improved} and note that $n^{1.5a}\leq n^{1.5}$, we have that the expected amortized running time per iteration of \textsc{PartialVectorUpdate} is
\begin{align*}
\frac{1}{T} \sum_{j=1}^T (p_jn^{1+o(1)} + n^{2a}\cdot \mathbf{1}_{p_j>0})
= &~ O\big((C_4/\epsilon_{\mathrm{mp}} + C_5/\epsilon_{\mathrm{mp}}^2)\cdot \log n \cdot (n^{1.5+o(1)} + n^{2a-\wt{a}/2})\big).
\end{align*}
\end{proof}

\subsection{Potential function \texorpdfstring{$\psi$}{}}\label{sec:potential_function_psi}
%Here we collect the facts we used about $\psi$. 

\begin{figure}[!t]
  \centering
    \includegraphics[width=0.99\textwidth]{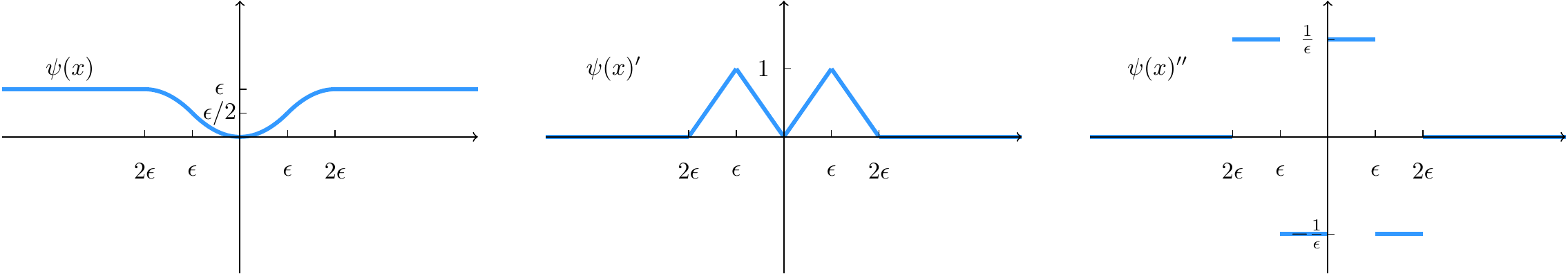}
    \caption{$\psi(x)$, $\psi(x)'$ and $\psi(x)''$. For $\epsilon_{\mathrm{mp}} \in (0,1)$, and for simplicity we use $\epsilon$ in the figures.}
\end{figure}

\begin{lemma}[Properties of function $\psi$, Lemma~5.10 of \cite{cls19}]\label{lem:def_psi}
Let function $\psi$ be defined as of Definition~\ref{def:psi}. Then function $\psi$ satisfies the following properties: \\
1. Symmetric: $\psi(-x)=\psi(x)$ and $\psi(0)=0$; \\
2. $\psi(|x|)$ is non-decreasing; \\
3. $|\psi'(x)| = \Omega(1), \forall |x| \leq 1.5\epsilon_{\mathrm{mp}}$; \\
4. $L_1 \defeq \max_x \psi'(x) = 1$ and $L_2 \defeq \max_x \psi''(x) = 1 / \epsilon_{\mathrm{mp}}$; \\
5. $\psi(x)$ is a constant for $|x| \geq 2 \epsilon_{\mathrm{mp}}$.
%6. $\psi(x) = O(\epsilon_{\mathrm{mp}})$, $\forall x$.
\end{lemma}
\begin{proof}
We can see that
\begin{align*}
\psi(x)' = \begin{cases}
\frac{|x|}{\epsilon_{\mathrm{mp}}} & |x| \in [0,\epsilon_{\mathrm{mp}}] \\
 \frac{ 2\epsilon_{\mathrm{mp}} - |x| }{\epsilon_{\mathrm{mp}}} & |x| \in (\epsilon_{\mathrm{mp}},2\epsilon_{\mathrm{mp}}] \\
0 & x \in (2\epsilon_{\mathrm{mp}}, +\infty)
\end{cases}
\quad\text{and}\quad
\psi(x)'' = \begin{cases}
\frac{1}{\epsilon_{\mathrm{mp}}} & x \in [0,\epsilon_{\mathrm{mp}}] \cup [-2\epsilon_{\mathrm{mp}}, -\epsilon_{\mathrm{mp}}]\\
-\frac{1}{\epsilon_{\mathrm{mp}}} & x \in  (\epsilon_{\mathrm{mp}},2\epsilon_{\mathrm{mp}}] \cup [-\epsilon_{\mathrm{mp}},0] \\
0 & x \in (2\epsilon_{\mathrm{mp}}, +\infty)
\end{cases}
\end{align*}
From the $\psi(x)'$ and $\psi(x)''$, it is not hard to see that $\psi$ satisfies the properties needed.
\end{proof}

\section{Combining data structure with optimization}
\label{sec:combine_improved}
\begin{algorithm}[!t]\caption{Main algorithm}\label{alg:main_improved}
	\begin{algorithmic}[1]
	\Procedure{\textsc{Main}}{$A,b,c,\delta, a,\wt{a}$} \Comment{Theorem~\ref{thm:main_1.5_a_wta_copied}}
	    \State \Comment{$A,b,c$ are inputs of LP}
	    \State \Comment{$\delta$ is the accuracy parameter}
	    \State $\epsilon_{\mathrm{mp}}\leftarrow 10^{-5}/\log n$,~ $\epsilon_{\far} \leftarrow \epsilon_{\mathrm{mp}}/100\log n$,~ $\epsilon\leftarrow 10^{-7}/\log n$
	    \State $\lambda\leftarrow 40\log n$,~ $b_{\text{sketch}} \leftarrow 10^{12}\sqrt{n}\log^8 n / \epsilon_{\mathrm{mp}}^2$,~ $L_{\text{sketch}} \leftarrow n^{1/2+o(1)} $
	    \State $\delta\leftarrow \min\{\frac{\delta}{2},\frac{1}{\lambda}\}$
	    \State Create $L_{\text{sketch}}$ sketching matrices $R_1, R_2, \cdots, R_{L_{\text{sketch}}} \in \R^{b_{\sketch} \times n}$ \Comment{Lemma~\ref{lem:E.5}}
	    \State Let $R = [ R_1^\top, R_2^\top, \cdots, R_{L_{\sketch}}^\top ]^\top$ \label{alg:line:main_initial_x_s}
	    \State Modify the linear program and obtain an initial $x$ and $s$.
	    \State Let $\text{mp}_t$, $\text{mp}_{\Phi}$ be projection maintenance data structures.
	    \State Let $f_t : x \mapsto \sqrt{x}$, and $f_{\Phi} : x \mapsto \lambda \sinh(\lambda(x-1)) / \sqrt{x}$ \label{alg:line:f_phi}
	    \State $t\leftarrow 1$
	    \State $\text{mp}_t$.\textsc{Initialize}$(f_t,\epsilon_{\mathrm{mp}}, \epsilon_{\far} , a, \wt{a}, b_{\text{sketch}},L_{\sketch},A,\frac{x}{s},xs, R)$ \Comment{Algorithm~\ref{alg:initialize_improved}}
	    \State $\text{mp}_{\Phi}$.\textsc{Initialize}$(f_{\Phi},\epsilon_{\mathrm{mp}}, \epsilon_{\far} , a ,\wt{a} ,b_{\text{sketch}},L_{\sketch},A,\frac{x}{s},\frac{xs}{t}, R)$ \Comment{Algorithm~\ref{alg:initialize_improved}}
	    \While{$t > \delta^2 /(32n^3)$}
	        \State $t^{\new}\leftarrow (1-\frac{\epsilon}{3\sqrt{n}})t$
	        \Repeat \label{alg:line:repeat_improved}
	        \State $\delta_x$, $\delta_s\leftarrow$ \textsc{OneStepCentralPath}$(\text{mp}_t,\text{mp}_{\Phi},x,s,t,t^{\new})$ \Comment{Algorithm~\ref{alg:one_step_central_path}}
	            \If{the $L_{\text{sketch}}$ sketching matrices are used up}
	                \State re-initialize $\text{mp}_t$ and $\text{mp}_{\Phi}$ with new skeching matrices.
	           \EndIf
	        \Until{$\|x^{-1}\delta_x\|_{\infty}\leq 3\epsilon$, $\|s^{-1}\delta_s\|_{\infty}\leq 3\epsilon$}
	        \State $x^{\new} \leftarrow x+\delta_x$, $s^{\new} \leftarrow s+\delta_s$
	        \If{$\Phi_{\lambda}(x^{\new} s^{\new}/t-1)>n^3$}
	            \State $(x^{\new},s^{\new}) \leftarrow \textsc{ClassicalStep}(x,s,t^{\new})$ \Comment{Use the central path step of \cite{v89_lp}.} \label{alg:line:classical_step_improved}
	            \State Construct sketching matrices $R$ similar as before.
	            \State $\text{mp}_t$.\textsc{Initialize}$(f_1,\epsilon_{\mathrm{mp}}, \epsilon_{\far}, a, \wt{a}, b_{\text{sketch}}, L_{\sketch}, A,\frac{x^{\new}}{s^{\new}},x^{\new}s^{\new}, R)$ \Comment{Algorithm~\ref{alg:initialize_improved}}
	            \State $\text{mp}_{\Phi}$.\textsc{Initialize}$(f_2,\epsilon_{\mathrm{mp}}, \epsilon_{\far}, a, \wt{a}, b_{\text{sketch}}, L_{\sketch}, A, \frac{x^{\new}}{s^{\new}}, \frac{x^{\new}s^{\new}}{t}, R)$ \Comment{Algorithm~\ref{alg:initialize_improved}}
	        \EndIf
	        \State $x\leftarrow x^{\new}$, $s\leftarrow s^{\new}$, $t\leftarrow t^{\new}$
	    \EndWhile
	    \State Return an approximate solution of the original linear program
	\EndProcedure
	\end{algorithmic}
\end{algorithm}

In this section we combine the results of optimization and data structure to prove Theorem~\ref{thm:tech_third_improvement}.

\begin{lemma} \label{lem:ass_proof_improved}
During the Main algorithm (Algorithm \ref{alg:main_improved}), we have the following guarantees:
\begin{enumerate}
\item Assumption~\ref{ass:epsilon},~\ref{ass:assumption2}, and~\ref{ass:epsilon_far} about the error parameters are always satisfied.
\item Assumption~\ref{ass:assumption} that $\wt{\mu} \approx_{\epsilon_{\mathrm{mp}}} \ov{\mu}$, $\wt{w} \approx_{\epsilon_{\mathrm{mp}}} \ov{w}$, and $\ov{\mu} \approx_{0.1} t$ is always satisfied.
\item The \textsc{ClassicalStep} (line~\ref{alg:line:classical_step_improved} of Algorithm~\ref{alg:main_improved}) is executed with probability at most $\frac{10}{n^2}$ in each iteration.
\item In expectation the repeat-loop on line~\ref{alg:line:repeat_improved} of Algorithm~\ref{alg:main_improved} is executed at most $2$ times.
\end{enumerate}
\end{lemma}

\begin{table}[!t]
\centering
\begin{tabular}{ | l | l | l | l | l | l |}
\hline
{\bf Notation} & $\epsilon$ & $\epsilon_{\mathrm{mp}}$ &  $\epsilon_{\far}$ & $\lambda$ & $b$\\
\hline
{\bf Choice} & $10^{-7}/\log n$ & $10^{-5}/\log n$ & $10^{-7}/\log^2 n$ & $40\log n$ & $10^{22}\sqrt{n}\log^{10} n$\\
\hline
\end{tabular}
\caption{Extension of Table~\ref{table:parameters}. Summary of choice of $\epsilon$, $\epsilon_{\mathrm{mp}}$, $\epsilon_{\far}$, $\lambda$ and $b$. Assigned in \textsc{Main} procedure (Algorithm~\ref{alg:main_improved}). They are used to prove Theorem~\ref{thm:main_1.5_a_wta_copied}.}
\label{table:parameters_improved}
\end{table}

\begin{proof}
\noindent \textbf{Part 1.} After plugging in the parameters in Table~\ref{table:parameters_improved}, it is straightforward to see that the constraints stated in Assumption~\ref{ass:epsilon},~\ref{ass:assumption2}, and~\ref{ass:epsilon_far} are all satisfied.

\noindent \textbf{Part 2.} The assumption that $\wt{\mu} \approx_{\epsilon_{\mathrm{mp}}} \ov{\mu}$ follows from Part (b) of the correctness of \textsc{UpdateQuery} in Theorem~\ref{thm:main_data_structure_improved} that $h^{\appr} \approx_{\epsilon_{\mathrm{mp}}} h^{\new}$, and the assumption that $\wt{w} \approx_{\epsilon_{\mathrm{mp}}} \ov{w}$ follows from Part (a) of the correctness of \textsc{UpdateQuery} in Theorem~\ref{thm:main_data_structure_improved} that $w^{\appr} \approx_{\epsilon_{\mathrm{mp}}} w^{\new}$.

Finally, whenever $\Phi_{\lambda}(xs/t-1)>n^3$, the main algorithm runs the procedure \textsc{ClassicalStep}, and the $(x,s)$ returned by \textsc{ClassicalStep} is guaranteed to satisfy $xs \approx_{0.01} t$ (see \cite{v89_lp}). Also, if $\Phi_{\lambda}(xs/t-1)\leq n^3$, we have that $e^{\lambda |x_is_i/t-1|} \leq n^3$, and since $\lambda \geq 30\log n$ (Part 2 of Assumption~\ref{ass:assumption2}), we have $|x_is_i/t-1| \leq 0.1$. Thus $\ov{\mu} \approx_{0.1} t$ is always satisfied as well.

\noindent \textbf{Part 3.} Let $\Phi^{(i)}=\Phi_{\lambda}(\frac{x^{(i)}s^{(i)}}{t^{(i)}}-1)$ denote the value of the potential function in the $i$-th iteration. We use induction to prove $\E[\Phi^{(i)}]\leq 10n$, for all $i$. In the beginning of the main algorithm, $x_is_i=1=t$, $\forall i\leq n$. Therefore in the base case, $\Phi^{(0)}=n<10n$. If the algorithm executes \textsc{ClassicalStep} in the $i$-th iteration, \textsc{ClassicalStep} outputs $x$ and $s$ that $xs\approx_{0.01}t$, and since $\lambda\leq 60\log n$ (Part 7 of Assumption~\ref{ass:assumption2}), $\Phi^{(i)} \leq 10n$. If the algorithm doesn't execute \textsc{ClassicalStep}, Lemma~\ref{lem:potential_martingale} gives us $\E [ \Phi^{(i)} ] \leq (1 - \frac{\lambda\eps}{15\sqrt{n}})\E [ \Phi^{(i-1)} ] + \frac{\lambda\eps}{15\sqrt{n}} 10n$. Therefore $\E[\Phi^{(i)} ] \leq 10 n$ since we have $\E[\Phi^{(i-1)} ]\leq 10n$ from induction hypothesis.
Then using Markov's inequality we have $\Pr[\Phi^{(k)}> n^3] \leq 10 / n^2$. Thus the \textsc{ClassicalStep} on line~\ref{alg:line:classical_step_improved} of Algorithm~\ref{alg:main_improved} is executed with probability at most $10 / n^2$ in each iteration.

\noindent \textbf{Part 4.} From Part 4 of Lemma~\ref{lem:stochastic_step} we have that $\|x^{-1} \delta_x\|_{\infty} > 3\eps$ and $\|s^{-1} \delta_s\|_{\infty} > 3\eps$ each happens with probability at most $1 / n^4$. Thus in expectation the repeat-loop on line~\ref{alg:line:repeat_improved} of Algorithm~\ref{alg:main_improved} is executed at most $2$ times.
\end{proof}

\begin{lemma}
For $\epsilon \in (0,1/10000)$, $\epsilon_{\mathrm{mp}}\in (0,1/10000)$, and $\epsilon_{\far}= \epsilon_{\mathrm{mp}}/100\log n$, each iteration of \textsc{Main} (Algorithm~\ref{alg:main_improved}) takes 
\[
O^*\Big( (\epsilon/\epsilon_{\mathrm{mp}}) \cdot (n^{\omega-1/2} + n^{2-a/2} + n^{1+a-\wt{a}/2}) + n^{(\omega-1)\wt{a}+a} + n^{1+b}\Big)
\]
expected amortized time per iteration, where $\omega$ is the exponent of matrix multiplication, $\alpha$ is the dual exponent of matrix multiplication, $0\leq a\leq \alpha$ and $0\leq \wt{a}\leq a\alpha$ are the thresholds used by the data structure, and $n^b$ is the sketching size. $O^*$ notation hides all $n^{o(1)}$ terms.
\end{lemma}
\begin{proof}
Part 3 of Lemma~\ref{lem:ass_proof_improved} shows that in each iteration \textsc{ClassicalStep} is executed with probability at most $O(1/n^2)$. Since the cost of \textsc{ClassicalStep} is $O(n^{2.5})$, the amortized cost of executing \textsc{ClassicalStep} is $O(n^{0.5})$ for one iteration.

Part 4 of Lemma~\ref{lem:ass_proof_improved} shows that in expectation \textsc{OneStepCentralPath} is executed at most 2 times in each iteration. So now we only need to bound the running time of the procedure \textsc{OneStepCentralPath} (Algorithm~\ref{alg:one_step_central_path}). In the procedure \textsc{OneStepCentralPath}, we call the procedure \textsc{UpdateQuery} of the data structures (Algorithm~\ref{alg:update_query_improved}) two times. Since the time analysis of these two data structure is the same, we are going to focus on one of them. Also note that the running time of \textsc{UpdateQuery} is the sum of that of \textsc{MatrixUpdate}, \textsc{PartialMatrixUpdate}, \textsc{VectorUpdate}, \textsc{PartialVectorUpdate}, and \textsc{Query}, so we analyze them one by one.

From what we proved in Lemma~\ref{lem:w_movement} and Lemma~\ref{lem:mu_movement}, we have $C_1=C_4=\Theta(\epsilon)$ and $C_2=C_5=\Theta(\epsilon^2)$ and $C_3=C_6=\Theta(\epsilon)$ < 1/4.

By Theorem~\ref{thm:main_data_structure_improved} and plugging in $\epsilon_{\far} = \epsilon_{\mathrm{mp}}/100\log n$, the expected amortized cost per iteration of the following procedures are as follows:
\begin{align*}
\textsc{MatrixUpdate}
= & ~ O^*((C_1\epsilon_{\mathrm{mp}}/\epsilon_{\far}^2 + C_2/\epsilon_{\far}^2)\cdot (n^{2-a/2}+n^{\omega-1/2}) )\\
%= & ~ O^*\Big(\left(\epsilon\cdot \epsilon_{\mathrm{mp}}/\epsilon_{\far}^2 + \epsilon^2/\epsilon_{\far}^2\right)\cdot (n^{2-a/2}+n^{\omega-1/2}) \Big)
= & ~ O^*( (\epsilon/\epsilon_{\mathrm{mp}} )  \cdot (n^{2-a/2}+n^{\omega-1/2}))\\
\textsc{PartialMatrixUpdate}
= & ~ O^*((C_1/\epsilon_{\mathrm{mp}} + C_2/\epsilon_{\mathrm{mp}}^2)\cdot (n^{1+a-\wt{a}/2} + n^{1+(\omega-3/2)a}))\\
%= & ~ O^*\Big((\epsilon/\epsilon_{\mathrm{mp}} + \epsilon^2/\epsilon_{\mathrm{mp}}^2)\cdot (n^{1+a-\wt{a}/2} + n^{1+(\omega-3/2)a})\Big)\\
= & ~ O^*((\epsilon/\epsilon_{\mathrm{mp}})\cdot (n^{1+a-\wt{a}/2} + n^{1+(\omega-3/2)a}))\\
\textsc{VectorUpdate}
= & ~ O^*((C_4\epsilon_{\mathrm{mp}}/\epsilon_{\mathrm{\far}}^2 + C_5/\epsilon_{\mathrm{\far}}^2)\cdot n^{1.5})\\
%= & ~ O^*\Big(\left(\epsilon\cdot \epsilon_{\mathrm{mp}}/\epsilon_{\far}^2 + \epsilon^2/\epsilon_{\far}^2\right)\cdot n^{1.5} \Big)
= & ~ O^*( (\epsilon/\epsilon_{\mathrm{mp}} )  \cdot n^{1.5})\\
\textsc{PartialVectorUpdate}
= & ~ O^*((C_4\epsilon_{\mathrm{mp}}/\epsilon_{\mathrm{\far}}^2 + C_5/\epsilon_{\mathrm{\far}}^2)\cdot (n^{1.5}+n^{2a-\wt{a}/2}))\\
%= & ~ O^*\Big(\left(\epsilon\cdot \epsilon_{\mathrm{mp}}/\epsilon_{\far}^2 + \epsilon^2/\epsilon_{\far}^2\right)\cdot (n^{1.5}+n^{2a-\wt{a}/2}) \Big)\\
= & ~ O^*( (\epsilon/\epsilon_{\mathrm{mp}} )  \cdot (n^{1.5}+n^{2a-\wt{a}/2}))\\
\textsc{Query}
= & ~ O^*(\Tmat(n^{\wt{a}}, n^a, n^{\wt{a}})+n^{1+b})\\
= & ~ O^*(n^{(\omega-1)\wt{a}+a}+n^{1+b}).
\end{align*}
So the overall expected amortized cost of one iteration is
\begin{align*}
 & ~  \textsc{MatrixUpdate} + \textsc{PartialMatrixUpdate}\\
 & ~ + \textsc{VectorUpdate} + \textsc{PartialVectorUpdate} + \textsc{Query} \\
= & ~  O^*\Big(
\underbrace{ (\epsilon/\epsilon_{\mathrm{mp}}) \cdot (n^{2-a/2}+n^{\omega-1/2}) }_{ \textsc{MatrixUpdate} } + 
\underbrace{(\epsilon/\epsilon_{\mathrm{mp}}) \cdot (n^{1+a-\wt{a}/2}+n^{1+(\omega-3/2)a})}_{\textsc{PartialMatrixUpdate}} \\
& ~ + 
\underbrace{ (\epsilon/\epsilon_{\mathrm{mp}}) \cdot n^{1.5} }_{ \textsc{VectorUpdate} }  +  \underbrace{ (\epsilon/\epsilon_{\mathrm{mp}}) \cdot (n^{1.5} + n^{2a-\wt{a}/2}) }_{ \textsc{PartialVectorUpdate} }  +
\underbrace{ n^{(\omega-1)\wt{a}+a}+n^{1+b} }_{ \textsc{Query} }\Big)  \\
= & ~ O^*\Big( (\epsilon/\epsilon_{\mathrm{mp}}) \cdot (n^{\omega-1/2} + n^{2-a/2} + n^{1+a-\wt{a}/2}) + n^{(\omega-1)\wt{a}+a} + n^{1+b}\Big),
\end{align*}
where in the last step we use $\omega\geq 2$, $\wt{a}\leq a\leq 1$ and $\omega-1/2 = 1+(\omega-3/2)\geq 1+(\omega-3/2)a$.
\end{proof}

Now we are ready to prove the main theorem of this paper.

\begin{theorem}[Restate Theorem~\ref{thm:tech_third_improvement}, Main result, third improvement]\label{thm:main_1.5_a_wta_copied}
Given a linear program $\min_{A x = b, x \geq 0} c^\top x$ with no redundant constraints. Assume that the polytope has diameter $R$ in $\ell_1$ norm, namely, for any $x \geq 0$ with $A x = b$, we have $\| x \|_1 \leq R$.

Then, for any $\delta \in (0,1]$, \textsc{Main}$(A,b,c,\delta)$ (Algorithm~\ref{alg:main_improved}) outputs $x \geq 0$ such that
\begin{align*}
c^\top x \leq \min_{A x = b, x \geq 0} c^\top x + \delta \| c \|_{\infty} R , \text{~~~and~~~} \| A x - b \|_1 \leq \delta \cdot ( R \| A \|_1 + \| b \|_1 )
\end{align*}
in expected time
\begin{align*}
&~ \wt{O}( n^{\omega + o(1)} + n^{2.5-a/2+o(1)} + n^{1.5+a-\wt{a}/2+o(1)} + n^{0.5+a+(\omega-1)\wt{a}}) \cdot \log( n / \delta ) 
\end{align*}
where $\omega$ is the exponent of matrix multiplication, $\alpha$ is the dual exponent of matrix multiplication, and $0 < a \leq \alpha$. 

In the ideal case when $\omega=2$ and $\alpha=1$. The running time is $\wt{O}(n^{2+1/18})$. For general $2\leq \omega\leq 3$ and $0\leq \alpha\leq 1$, the running time 
is $O^*(n^{\omega} + n^{2.5-a/2} + n^{(8 + \sqrt{19})/6 })=O^*(n^{\omega} + n^{2.5-a/2} + n^{2.06})$.
\end{theorem}

\begin{proof}
We use the parameters of Table \ref{table:parameters_improved} to prove the theorem. Since $t$ is decreasing by a $(1-\frac{\epsilon}{3\sqrt{n}})$ factor, the \textsc{Main} algorithm will take $O( \epsilon^{-1} {\sqrt{n}} \log(n/\delta))$ iterations in total.

Thus the total running time is
\begin{align*}
 & ~ \# \text{iterations} \cdot \text{cost~per~iteration} \\
= & ~ O( \epsilon^{-1} {\sqrt{n}} \log(n/\delta)) \cdot O^*\Big( (\epsilon/\epsilon_{\mathrm{mp}}) \cdot (n^{\omega-1/2} + n^{2-a/2} + n^{1+a-\wt{a}/2}) + n^{(\omega-1)\wt{a}+a} + n^{1+b}\Big) \\
= & ~ O^*\Big( \epsilon_{\mathrm{mp}}^{-1} ( n^{\omega} + n^{2.5-a/2} + n^{1.5+a-\wt{a}/2} ) + \epsilon^{-1} (n^{0.5+(\omega-1)\wt{a}+a} + n^{1.5+b}) \Big) \cdot \log (n/\delta).
\end{align*}
By plugging in the parameters $\epsilon = O(1/\log n)$, $\epsilon_{\mathrm{mp}} = O(1/\log n)$, and $b=\sqrt{n}\log^{10}n$ (see Table~\ref{table:parameters_improved}), the above running time becomes
\begin{align}\label{eq:main_improved_time}
O^*\Big( n^{\omega} + n^{2.5-a/2} + n^{1.5+a-\wt{a}/2} + n^{0.5+(\omega-1)\wt{a}+a} \Big) \cdot \log (n/\delta),
\end{align}
and recall that parameters $a$ and $\wt{a}$ need to satisfy that $a\leq \alpha$ and $\wt{a}\leq \alpha a$.

Therefore, in the ideal case where $\omega=2$ and $\alpha=1$, we can choose $a=\frac{8}{9}$ and $\wt{a}=\frac{2}{3}$, and we have $2.5-a/2=1.5+a-\wt{a}/2=0.5+(\omega-1)\wt{a}+a=2+1/18$, so the above running time simplifies to
\[
O^*(n^{2+1/18+o(1)})\cdot \log(n/\delta).
\]

For general $\omega$ and $\alpha$, the parameters are optimized as follows:
\begin{align*}
a = 
\begin{cases}
\alpha, &~ \text{ if }\alpha \leq \frac{4w}{3(2w-1)},\\
\frac{4w}{3(2w-1)}, &~ \text{ o.w. }
\end{cases}
&&
\wt{a}=
\begin{cases}
\min\{\alpha^2, \frac{2}{2\omega -1}\}, &~ \text{ if }\alpha \leq \frac{4\omega}{3(2\omega-1)},\\
\frac{2}{2\omega-1}, &~ \text{ o.w. }
\end{cases}
\end{align*}
Here we prove the final running time by discussing two cases.
\begin{enumerate}
    \item In the first case where $\alpha \leq \frac{4\omega}{3(2\omega-1)}$, we have $a=\alpha$ and $\wt{a}\leq \alpha^2=\alpha a$. 
    
    If $\wt{a}=\alpha^2$, then $1.5 + \alpha - \wt{a}/2 = 1.5 + \alpha - \alpha^2/2 \leq 2.5-\alpha/2$ since $\alpha \leq 1$. 
    
    If $\wt{a}=\frac{2}{2\omega-1}$, then $1.5 + \alpha - \wt{a}/2 = 1.5 + \alpha - 1/(2\omega-1)\leq 2.5-\alpha/2$, since $\alpha \leq \frac{4\omega}{3(2\omega-1)}$ and $\frac{4\omega}{3(2\omega-1)}$ is the value that balances the two terms.
    Thus the following inequality holds:
    \[
    n^{1.5+a-\wt{a}/2}\leq n^{2.5-a/2}.
    \]
    
    We also have $0.5 + (\omega -1)\wt{a} + a \leq 1.5+a-\wt{a}/2,$ since $\wt{a}\leq \frac{2}{2\omega-1}$ and $\frac{2}{2\omega-1}$ is the value that balances these two terms.
    Thus the following inequality also holds:
    \[
    n^{0.5 + (\omega -1)\wt{a} + a} \leq n^{1.5+a-\wt{a}/2}.
    \]
    
    Therefore the running time of Eq.~\eqref{eq:main_improved_time} is dominated by $O(n^{\omega}+n^{2.5-\alpha/2})$ in the first case.
    \item In the second case where $\alpha > \frac{4\omega}{3(2\omega-1)}$, we have $a=\frac{4\omega}{3(2\omega-1)}\leq \alpha$, and
    $
    \wt{a}=\frac{2}{2\omega-1}\leq (\frac{4\omega}{3(2\omega-1)})^2 \leq \alpha a,
    $
    where the second step follows since $\omega \geq 2$.
    With these parameters, we have that
    \[
    2.5-a/2=1.5+a-\wt{a}/2=0.5+(\omega-1)\wt{a}+a=\frac{13}{6}-\frac{1}{3(2\omega-1)}.
    \]
    Therefore in the second case the running time of Eq.~\eqref{eq:main_improved_time} is dominated by 
    \[
    O(n^{\omega}+n^{\frac{13}{6}-\frac{1}{3(2\omega-1)}}) \leq O(n^{\omega} + n^{(8 + \sqrt{19})/6}),
    \]
    where $(8 + \sqrt{19})/6\approx 2.0598\leq 2.06$ is the solution of equation $\omega = \frac{13}{6}-\frac{1}{3(2\omega-1)}$.
\end{enumerate}

Thus the running time in Eq.~\eqref{eq:main_improved_time} is always upper bounded by
\begin{align*}
O^*\Big(n^{\omega} + n^{2.5-\alpha/2} + n^{(8 + \sqrt{19})/6}\Big)\cdot \log(n/\delta).
\end{align*}

\end{proof}

\section{Multi-level with more details}
\label{sec:multi_detail}
In this section we provide more details for Section~\ref{sec:multi}.
\subsection{LU-decomposition of Woodbury identity when \texorpdfstring{$K=3$}{}}
The LU-decomposition of matrix $D=\begin{bmatrix}
M & U_2 & U_3 \nonumber\\
V_2^{\top} & -C_2^{-1} & 0 \nonumber \\
V_3^{\top} & 0 & -C_3^{-1}
\end{bmatrix} \nonumber$ where the diagonal blocks of $L$ are identity matrices is
\begin{align}\label{eq:LU_decomposition_k_2}
D = &
\begin{bmatrix}
I & 0 & 0 \nonumber\\
V_2^{\top}M^{-1} & I & 0 \nonumber\\
V_3^{\top}M^{-1} & 0 & I
\end{bmatrix} \cdot 
\begin{bmatrix}
M & U_2 & U_3 \nonumber\\
0 & -C_2^{-1} - V_2^{\top}M^{-1}U_2 & - V_2^{\top}M^{-1}U_3 \nonumber\\
0 & -V_3^{\top}M^{-1}U_2 & -C_3^{-1}-V_3^{\top}M^{-1}U_3
\end{bmatrix} \nonumber \\
= &
\underbrace{
\begin{bmatrix}
I & 0 & 0\\
V_2^{\top}M^{-1} & I & 0\\
V_3^{\top}M^{-1} & 0 & I
\end{bmatrix} \cdot 
\begin{bmatrix}
I & 0 & 0\\
0 & I & 0\\
0 & -V_3^{\top}M^{-1}U_2B^{-1} & I
\end{bmatrix}}_{L} \nonumber \\
\cdot & 
\underbrace{
\begin{bmatrix}
M & U_2 & U_3\\
0 & B & - V_2^{\top}M^{-1}U_3\\
0 & 0 & -C_3^{-1}-V_3^{\top}M^{-1}U_3 - V_3^{\top}M^{-1}U_2B^{-1}V_2^{\top}M^{-1}U_3
\end{bmatrix}}_{U},
\end{align}
where $B:=-C_2^{-1} - V_2^{\top}M^{-1}U_2$.

\subsection{Online low-rank inverse and the oMV conjecture.}
The following \emph{static} data structure problem has received significant attention recently (see \cite{hkns15, bns19} 
and references therein). 
\begin{definition}[Online matrix-vector multiplication (oMV)]\label{def:OMV1}
Preprocess a (fixed) matrix $M \in \R^{n \times n}$ so that, given an online sequence of $T$ vectors $h_t\in \R^n$
(arriving one by one), the data structure can efficiently output $M h_t$ (exactly) before the next iteration $t+1$.
\end{definition}

The oMV Conjecture~\cite{hkns15} states that with $\poly(n)$ preprocessing time, the (amortized) query time 
of any word-RAM data structure for oMV is at least $t_q > n^{2- o(1)}$. Note that this is in sharp contrast to 
the \emph{offline} setting where the vectors $\{h_t\}_{t\in T}$ are given as a batch, in which case fast-matrix 
multiplication can achieve $n^{\omega-1} <  n^{1.37}$ query time on average (assuming $T=n$, say). 
Now consider the following static problem: 

\begin{definition}[Online low-rank inverse multiplication]\label{def:OMV2}
Preprocess a fixed matrix $M \in \R^{n \times n}$  and a fixed vector $h$, 
so that given an online sequence of $T$ pairs of vectors $u_t,v_t \in \R^n$, 
the data structure can efficiently output 
$(M + u_t v_t^\top)^{-1} h$ (exactly)  before the next iteration $t+1$.
\end{definition}

Perhaps surprisingly, we prove that these problems are essentially equivalent in the word-RAM model 
with polynomial preprocessing time. We first show that oMV is at least as hard as the problem in Definition \ref{def:OMV2}:

\begin{lemma}\label{lem:omv_direction_1_to_2}
If there is a data structure $A$ with polynomial preprocessing time 
for oMV (Definition~\ref{def:OMV1}) with worst case query time ${\cal T}_{A} (n, T)$, 
then there is a data structure $A'$ for the online low-rank inverse problem in Definition~\ref{def:OMV2} 
with polynomial preprocessing time and $O( {\cal T}_A(n,T))$ query time.
\end{lemma}

\begin{proof}
We design $A'$ as follows. In the preprocessing time, we use $O(n^\omega)$ time to pre-compute 
the vector $x := M^{-1}\cdot h \in \R^{n}$ and run the oMV data structure $A$ on the input matrix $M^{-1}$. 
Recall in the query stage, we are given $u_t,v_t\in \R^{n}$. By Woodbury's identity, the solution 
$(M+u_tv_t^{\top})^{-1}\cdot h$ can be written as
\begin{align*}
  (M + u_t v_t^{\top})^{-1} h
=  M^{-1}h - M^{-1} u_t ( 1 + v_t^{\top}M^{-1} u_t)^{-1} v_t^{\top} M^{-1}\cdot h.
\end{align*} 
Thus, using a single invocation of the query algorithm of $A$, we can compute the product 
$y := M^{-1}\cdot u_t$, and the remaining calculation is
\begin{align*}
     M^{-1}h - M^{-1} u_t ( 1 + v_t^{\top}M^{-1} u_t)^{-1} v_t^{\top} M^{-1}\cdot h
    =  x - y( 1 + v_t^{\top}\cdot y)^{-1}v_t^{\top}\cdot x,    
\end{align*}
which only involves vector inner product calculation and can be done in $O(n)$ time.

Therefore, $A'$ takes $O( {\cal T}_A(n,T)) + O(Tn)$ worst case query time. The lemma is proved by 
observing that ${\cal T}_A(n,T)$ is at least $\Omega(Tn)$.
\end{proof}

%%%%%%%%%%%%%%%%%%%%%%

We proceed to the other direction of the proof. 

\begin{lemma}\label{lem:omv_direction_2_to_1}
Given a word-RAM data structure $B$ with polynomial preprocessing time for 
the online low-rank inverse problem, with worst-case query time ${\cal T}_{B} (n, T)$, 
there is a data structure $B'$ for the oMV problem with polynomial preprocessing time and $O( {\cal T}_B(n,T) )$ worst case query time.
\end{lemma}

\begin{proof}
We construct $B'$ as follows: Given the input $M$
in Definition~\ref{def:OMV1}, the data structure first compute $M^{-1}$, and finds an arbitrary vector $h$ for which $h^{\top} Mh\neq 0$, 
then pre-computes $x := M h$ and $y := h^{\top} M$. This takes $O(n^{\omega})$ preprocessing time. We can now invoke 
the preprocessing function of the data structure $B$ (for online low-rank inverse) 
with inputs $M \leftarrow M^{-1}$, $h \leftarrow h$.
In each iteration $t\in [T]$, given the vector $h_t$ of Definition~\ref{def:OMV1}, invoke $B$ with query vector $u_t \leftarrow h_t$, $v_t \leftarrow h$ to get an answer $g$. Note that by Woodbury's identity
\begin{align*}
g =  (M^{-1} + u_t v_t^{\top})^{-1} h
=  M h - M u_t ( 1 + v_t^{\top}M u_t)^{-1} v_t^{\top} M \cdot h 
=  M h - M h_t ( 1 + h^{\top}M u_t)^{-1} h^{\top} M \cdot h.
\end{align*} 
Then $B'$ outputs the query answer
\begin{align*}
    (x - g) \cdot \frac{1 + y\cdot h_t}{ y \cdot h} = (M h - g) \cdot \frac{1 + (h^{\top}M) h_t}{ h^{\top} M \cdot h} = M h_t,
\end{align*}
which only involves vector inner product calculation and can be done in $O(n)$ time.

Therefore, $B'$ takes polynomial preprocessing time and $O( {\cal T}_B(n,T)) + O(Tn)$ worst case query time. The lemma is proved by 
observing that ${\cal T}_B(n,T)$ is at least $\Omega(Tn)$.
\end{proof}

As a corollary, we get that the following problem is at least as hard  as the oMV problem: 
\begin{definition}[Online cumulative low-rank inverse]\label{def:our_omv}
Preprocess a fixed matrix $M \in \R^{n \times n}$  and a fixed vector $h$, %and its inverse $M^{-1}\in \R^{n\times n}$
so that given an online sequence of $T$ pairs of vectors $u_t,v_t \in \R^n$, 
%(arriving one by one),  
the data structure can efficiently output 
$(M + \sum_{i=1}^t u_i v_i^\top)^{-1} h$ (exactly)  before the next iteration $t+1$.
\end{definition}

\subsection{Optimizing the parameters of Eq.~\texorpdfstring{\eqref{eq:recursive_multilevel_runtime_before_setting_parameter}}{}}
\label{sec:recursive_multilevel_runtime}
\begin{claim}\label{cla:multi_tex_running_time_prove}
For any positive integer $n,K$, the following equation holds.
\begin{align*}
\min_{n=n_1\geq n_2\geq \cdots n_{K-1}\geq n_K\geq 1}\left( \sum_{k=1}^{K-1} n \cdot n_k / \sqrt{ n_{k+1} } +n_{K-1} \cdot n_K \right) = O(K)\cdot n^{1.5 + \frac{1}{6\cdot (2^{K-1}-1)}}.
\end{align*}
\end{claim}
\begin{proof}
Denote $p=\frac{1}{6\cdot (2^{K-1}-1)}$. Define $a_k$ as the exponent of $n_k$ such that $n_k=n^{a_k}$. Then it is equivalent to finding $a_1,\cdots,a_K\in \R$ such that the following three conditions hold:
\begin{enumerate}
    \item $1=a_1\geq a_2\geq \cdots \geq a_{K-1}\geq a_K\geq 0$.
    \item $1 + a_i - a_{i+1}/2 \leq 1.5 + p$, $\forall i\in[K-1]$.
    \item $a_{K-1} + a_K \leq 1.5 + p$.
\end{enumerate}

The optimal solution is given by $a_{i} = 1-1/(3\cdot 2^{K-i}) + (2-2^{-K+i+1})\cdot p$, $\forall i \in[K]$. Now we prove that the three conditions are satisfied with this solution.

\noindent {\bf Condition 1.}
    \begin{align*}
        a_1 
        =  1-\frac{1}{3\cdot 2^{K-1}} + \frac{2-2^{-K+2}}{6\cdot (2^{K-1}-1)}
        =  1-\frac{1}{3\cdot 2^{K-1}} + \frac{(2^{K-1}-1)/2^{K-1}}{3\cdot (2^{K-1}-1)}
        =  1.
    \end{align*}
    Since $a_{i}$ is decreasing with $i$, $a_i\geq a_{i+1}$ is also satisfied.
    Finally, $a_K = \frac{2}{3} \geq 0$.
    
\noindent {\bf Condition 2.} 
    For all $i\in [K-1]$, 
    \begin{align*}
        1 + a_i - a_{i+1}/2 
        = & ~ 1 + (1-\frac{1}{3\cdot 2^{K-i}} + (2-2^{-K+i+1})\cdot p ) - ( 1-\frac{1}{3\cdot 2^{K-i-1}} + (2-2^{-K+i+2})\cdot p)/2\\
        = & ~ 1.5 -\frac{1}{3\cdot 2^{K-i}} + (2-2^{-K+i+1})\cdot p + \frac{1}{3\cdot 2^{K-i}} - (1-2^{-K+i+1})\cdot p\\
        = & ~ 1.5 + p.
    \end{align*}
    
\noindent {\bf Condition 3.} 
    $a_{K-1} + a_K = (1 - \frac{1}{6} + p) + \frac{2}{3} = 1.5 + p$.
\end{proof}
%\subsection{Algorithm flow of feasible data structure}
%\label{sec:data_structure_fixed}

\begin{algorithm}[!t]
\caption{Feasible version of \textsc{Main} (Algorithm~\ref{alg:main_improved})}\label{alg:main_fixed}
\small
	\begin{algorithmic}[1]
	\Procedure{\textsc{Main}}{$A,b,c,\delta, a,\wt{a}$} \Comment{Theorem~\ref{thm:main_1.5_a_wta_copied}+\ref{thm:correctness_fixed}, $A,b,c$ are inputs of LP, $\delta$ is the accuracy parameter}
	    \State $\epsilon_{\mathrm{mp}}\leftarrow 10^{-5}/\log n$
	    \State $\epsilon_{\far} \leftarrow \epsilon_{\mathrm{mp}}/100\log n$
	    \State $\epsilon\leftarrow 10^{-7}/\log n$
	    \State $\lambda\leftarrow 40\log n$
	    \State $\delta\leftarrow \min\{\frac{\delta}{2},\frac{1}{\lambda}\}$
	    \State $b_{\text{sketch}} \leftarrow 10^{12}\sqrt{n}\log^8 n / \epsilon_{\mathrm{mp}}^2$ \label{alg:line:main_fixed_b}
	    \State $L_{\text{sketch}} \leftarrow n^{1/2+o(1)} $
	    \State Create $L_{\text{sketch}}$ sketching matrices $R_1, R_2, \cdots, R_{L_{\text{sketch}}} \in \R^{b_{\sketch} \times n}$ \Comment{Lemma~\ref{lem:E.5}}
	    \State Let $R = [ R_1^\top, R_2^\top, \cdots, R_{L_{\sketch}}^\top ]^\top$
	    \State Modify the linear program and obtain an initial $x^{(0)}$ and $s^{(0)}$.
	    \State Let $\text{mp}_t$ and $\text{mp}_{\Phi}$ be projection maintenance data structures.
	    \State Let $f_t : x \mapsto \sqrt{x}$, ~$f_{\Phi} : x \mapsto \lambda \sinh(\lambda(x-1)) / \sqrt{x}$	    \label{alg:line:main_fixed_f_t_f_phi}
	    \State \blue{$\text{mp}_t$.\textsc{Initialize}$(f_t,\epsilon_{\mathrm{mp}}, \epsilon_{\far} , a, \wt{a}, b_{\text{sketch}},L_{\sketch},A,\frac{x^{(0)}}{s^{(0)}},x^{(0)}s^{(0)}, R)$} \Comment{Algorithm~\ref{alg:member_invariant_initialize_fixed}}
	    \State \blue{$\text{mp}_{\Phi}$.\textsc{Initialize}$(f_{\Phi},\epsilon_{\mathrm{mp}}, \epsilon_{\far} , a ,\wt{a} ,b_{\text{sketch}},L_{\sketch},A,\frac{x^{(0)}}{s^{(0)}},\frac{x^{(0)}s^{(0)}}{t}, R)$} \Comment{Algorithm~\ref{alg:member_invariant_initialize_fixed}}
	    
	    \State \blue{ \textbf{Global } $\ov{x},\ov{s}, x, s, t, t^{\new}, w^{\old}$}
	    \State \blue{ $\ov{x} \leftarrow x \leftarrow x^{(0)}$, ~ $\ov{s} \leftarrow s \leftarrow s^{(0)}$, ~ $w^{\old} \leftarrow x^{(0)}s^{(0)}$}
	    \State \blue{ $t^{\old} \leftarrow t \leftarrow 1$, ~ $j \leftarrow 0$ }
	    \While{$t > \delta^2 /(32n^3)$} \label{alg:line:t_while_loop_fixed}
	        \State $t^{\new}\leftarrow (1-\frac{\epsilon}{3\sqrt{n}})t$,~ \blue{$j \leftarrow j + 1$}
	        \Repeat \label{alg:line:repeat_fixed}
	        \State \blue{ $\wh{\delta}_x, \wh{\delta}_s, w^{\appr} \leftarrow$ \textsc{OneStepCentralPath}$(\text{mp}_t,\text{mp}_{\Phi},t,t^{\new})$}
	        \Comment{Algorithm~\ref{alg:one_step_central_path_fixed}}
	            \If{the $L_{\text{sketch}}$ sketching matrices are used up}
	                \State re-initialize $\text{mp}_t$ and $\text{mp}_{\Phi}$ with new ones.
	           \EndIf
	        \Until{$\|\ov{x}^{-1}\wh{\delta}_x\|_{\infty}\leq 5\epsilon$, $\|\ov{s}^{-1}\wh{\delta}_s\|_{\infty}\leq 5\epsilon$,~ if this condition is false, revoke the updates of $u_1,u_2,u_3,u_4$ in $\mathrm{mp}_{t}$ and $\mathrm{mp}_{\Phi}$}
	        \State $\ov{x} \leftarrow \ov{x}+\wh{\delta}_x$, $\ov{s} \leftarrow \ov{s}+\wh{\delta}_s$
	       \State \blue{ $\textsc{MakeFeasible}(w^{\appr})$} \label{alg:line:make_feasible_fixed}
	       \If{\blue{ $j>\sqrt{n}$ or $t < t^{\old}/2$}} \label{alg:line:main_fixed_if_j_sqrt_n}
		\State \blue{ $x \leftarrow x - (\mathrm{mp}_{t}.u_1 + \mathrm{mp}_{t}.G\cdot \mathrm{mp}_{t}.u_2) - (\mathrm{mp}_{\Phi}.u_1 + \mathrm{mp}_{\Phi}.G\cdot \mathrm{mp}_{\Phi}.u_2)$} \label{alg:line:x_update_fixed}
		\State \blue{ $s \leftarrow s + (\mathrm{mp}_{t}.u_3 + \mathrm{mp}_{t}.M \cdot \mathrm{mp}_{t}.u_4) + (\mathrm{mp}_{\Phi}.u_3 + \mathrm{mp}_{\Phi}.M \cdot \mathrm{mp}_{\Phi}.u_4)$} \label{alg:line:s_update_fixed}
		\State \blue{ $\mathrm{mp}_t.\textsc{Initialize}(\epsilon_{mp}, \epsilon_{\far}, a, \wt{a}, b, L, A, w^{\appr}, h^{\appr}, R)$}\label{alg:line:main_fixed_initialize_t_when_j_geq_sqrt_n}
		\State \blue{ $\mathrm{mp}_{\Phi}.\textsc{Initialize}(\epsilon_{mp}, \epsilon_{\far}, a, \wt{a}, b, L, A, w^{\appr}, h^{\appr}/t, R)$}\label{alg:line:main_fixed_initialize_phi_when_j_geq_sqrt_n}
		\State \blue{ $j \leftarrow 1$, $t^{\old} \leftarrow t$}
		\EndIf
	        \If{$\Phi_{\lambda}(\ov{x} \ov{s}/t-1)>n^3$}
	            \State \blue{ $(\ov{x}, 
	            \ov{s}) \leftarrow \textsc{ClassicalStep}(\ov{x},\ov{s},t^{\new})$ }\Comment{Use the central path step of \cite{v89_lp}. \label{alg:line:classical_step_fixed}}
	            \State \blue{ $x \leftarrow \ov{x}$,~ $s \leftarrow \ov{s}$}
	            \State Construct sketching matrices $R$ similar as before.
	            \State \blue{$\text{mp}_t$.\textsc{Initialize}$(f_t,\epsilon_{\mathrm{mp}}, \epsilon_{\far}, a, \wt{a}, b_{\text{sketch}}, L_{\sketch}, A,\frac{\ov{x}}{\ov{s}},\ov{x}\ov{s}, R)$} \Comment{Algorithm~\ref{alg:member_invariant_initialize_fixed}}
	            \State\blue{ $\text{mp}_{\Phi}$.\textsc{Initialize}$(f_{\Phi},\epsilon_{\mathrm{mp}}, \epsilon_{\far}, a, \wt{a}, b_{\text{sketch}}, L_{\sketch}, A, \frac{\ov{x}}{\ov{s}}, \frac{\ov{x}\ov{s}}{t}, R)$} \Comment{Algorithm~\ref{alg:member_invariant_initialize_fixed}}
	        \EndIf
	        \State $t\leftarrow t^{\new}$
	    \EndWhile \label{alg:line:main_end_while_fixed}
	    \State \Return $x$%Return an approximate solution of the original linear program
	\EndProcedure
	\end{algorithmic}
\end{algorithm}

\begin{algorithm}[!t]\caption{Feasible version of \textsc{OneStepCentralPath} (Algorithm~\ref{alg:one_step_central_path})}\label{alg:one_step_central_path_fixed}
\small
	\begin{algorithmic}[1]
		\Procedure{\textsc{OneStepCentralPath}}{$\text{mp}_t,\text{mp}_{\Phi},t,t^{\new}$} \Comment{Part 1 of Theorem~\ref{thm:correctness_fixed}} 
		\State $\ov{w} \leftarrow \ov{x} / \ov{s}$  \Comment{$\ov{x},\ov{s}$ is global variable}
		\State $\ov{\mu} \leftarrow \ov{x} \ov{s}$ 
		\State \blue{$(q_{t, x}, p_{t, x}, p_{t, s}, w^{\appr}) \leftarrow \mathrm{mp}_{t}.\textsc{UpdateQuery}(\ov{w}, \ov{\mu})$} \Comment{Algorithm~\ref{alg:update_query_fixed}}
		\State \Comment{ this data-structure works with function $f_t(x)=\sqrt{x}$}
		\State \blue{$(q_{\Phi, x}, p_{\Phi, x}, p_{\Phi, s}, w^{\appr}) \leftarrow \mathrm{mp}_{\Phi}.\textsc{UpdateQuery}(\ov{w}, \ov{\mu}/t)$} \Comment{Algorithm~\ref{alg:update_query_fixed}} 
		\State \Comment{this data-structure works with function $f_{\Phi}(x)=\nabla\Phi(x-1)/\sqrt{x}$}
		\State \blue{ $\wh{\delta}_x \leftarrow q_{t,x} + q_{\Phi,x} - (p_{t,x} + p_{\Phi,x})$} \label{alg:line:one_step_central_path_fixed_wh_delta_x}
		\State \blue{ $\wh{\delta}_s \leftarrow p_{t,s} + p_{\Phi,s}$}
		\State \blue{$x \leftarrow x + q_{t,x} + q_{\Phi,x} $} \label{alg:line:one_step_central_path_fixed_x_add_q_t_x_add_q_phi_x}
		\State \Return $( \wh{\delta}_x, \wh{\delta}_s , w^{\appr})$
		\EndProcedure
	\end{algorithmic}
\end{algorithm}

\begin{algorithm}[!ht]
\caption{Data structure : \blue{\textsc{MakeFeasible}}.}\label{alg:makefeasible_fixed}
\small
	\begin{algorithmic}[1]
	\State{\bf data structure}
    \Procedure{\textsc{MakeFeasible}}{$w^{\appr}$} \Comment{Part 2 of Theorem~\ref{thm:correctness_fixed}}
        \State $\wh{S} \leftarrow \{i : |w^{\old}_i-w^{\appr}_i| > w^{\old}_i/2\}$
        \State $\ov{x}_{\wh{S}} \leftarrow x_{\wh{S}} - ((\mathrm{mp}_{t}.u_1)_{\wh{S}} + (\mathrm{mp}_{t}.G\cdot \mathrm{mp}_{t}.u_2)_{\wh{S}}) - ((\mathrm{mp}_{\Phi}.u_1)_{\wh{S}} + (\mathrm{mp}_{\Phi}.G\cdot \mathrm{mp}_{\Phi}.u_2)_{\wh{S}})$ \label{alg:line:makefeasible_fixed_ov_x_wh_s}
        \State $\ov{s}_{\wh{S}} \leftarrow s_{\wh{S}} + ((\mathrm{mp}_{t}.u_3)_{\wh{S}} + (\mathrm{mp}_{t}.M \cdot \mathrm{mp}_{t}.u_4)_{\wh{S}}) + ((\mathrm{mp}_{\Phi}.u_3)_{\wh{S}} + (\mathrm{mp}_{\Phi}.M \cdot \mathrm{mp}_{\Phi}.u_4)_{\wh{S}})$
		\State $w^{\old}_{\wh{S}} \leftarrow w^{\appr}_{\wh{S}}$
		%\State $\ov{x}_{\wh{S}} \leftarrow x_{\wh{S}}$
		%\State $\ov{s}_{\wh{S}} \leftarrow s_{\wh{S}}$
		%\State $(\mathrm{mp}_{t}.u_1)_{\wh{S}} \leftarrow - (\mathrm{mp}_{t}.G \cdot \mathrm{mp}_{t}.u_2)_{\wh{S}}$
		%\State $(\mathrm{mp}_{\Phi}.u_1)_{\wh{S}} \leftarrow - (\mathrm{mp}_{\Phi}.G \cdot \mathrm{mp}_{\Phi}.u_2)_{\wh{S}}$
		%\State $(\mathrm{mp}_{t}.u_3)_{\wh{S}} \leftarrow - (\mathrm{mp}_{t}.M \cdot \mathrm{mp}_{t}.u_4)_{\wh{S}}$
		%\State $(\mathrm{mp}_{\Phi}.u_3)_{\wh{S}} \leftarrow - (\mathrm{mp}_{\Phi}.M \cdot \mathrm{mp}_{\Phi}.u_4)_{\wh{S}}$
    \EndProcedure
	\State {\bf end data structure}
	\end{algorithmic}
\end{algorithm}

\begin{algorithm}[!ht]
\caption{Data structure : feasible version of members (Algorithm~\ref{alg:member_improved}), invariants (Assumption~\ref{ass:invariant_improved}), and \textsc{Initialize} (Algorithm~\ref{alg:initialize_improved}) }\label{alg:member_invariant_initialize_fixed} 
\small
	\begin{algorithmic}[1]
		\State{\bf data structure}
		\State
		\State{\bf members}
		\State \hspace{4mm} $\cdots$ \Comment{We continue to have all previous members of Algorithm~\ref{alg:member_improved}.}
		\State \hspace{4mm} \blue{ $u_1,u_2,u_3,u_4 \in \R^n$}
		\State \hspace{4mm} \blue{ $G \in \R^{n\times n}$}
%		\State \hspace{4mm} \blue{ $\ov{x} \in \R^n$}
%		\State \hspace{4mm} \blue{ $\ov{s} \in \R^n$}
		\State{\bf end members}
		\State
		\State{\bf invariant}
		\State \hspace{4mm} $\cdots$ \Comment{We continue to maintain all previous invariants of Assumption~\ref{ass:invariant_improved}.}
		\State \hspace{4mm} \blue{ $G=\wt{V}A^{\top}(AVA^{\top})^{-1}A$} \Comment{\blue{$G\in \R^{n\times n}$}}
        %\Comment{$s=s+A^{\top}(AW^{\appr}A^{\top})^{-1}A\sqrt{W^{\appr}}f(h^{\appr})$}
        \State \hspace{4mm} \blue{ $x=u_1+Gu_2$ }\Comment{\blue{$x\in \R^{n}$}}
        %\Comment{$x=x+W^{\appr}A^{\top}(AW^{\appr}A^{\top})^{-1}A\sqrt{W^{\appr}}f(h^{\appr})$}
        %\State \hspace{4mm} \blue{ $G=W^{\appr}A^{\top}(AVA^{\top})^{-1}A$}
        \State \hspace{4mm} \blue{ $s=u_3+Mu_4$ }\Comment{\blue{$x\in \R^{n}$}}
		\State{\bf end invariant}
		\State
		\Procedure{\textsc{initialize}}{$f,\epsilon_{\mathrm{mp}}, \epsilon_{\far}, a, \wt{a}, b, L, A, w_0, h_0, R$}\label{alg:line:initialize_fixed}
		\State $\cdots$ \Comment{All previous members are initialized as of Algorithm~\ref{alg:initialize_improved}.}
		\State \blue{ $G\leftarrow W_0 A^{\top}(AVA^{\top})^{-1}A$  \Comment{$G\in \R^{n \times n}$}}
		\State \blue{$u_1, u_2, u_3, u_4 \leftarrow 0$\Comment{$u_1, u_2, u_3, u_4\in \R^{n}$}}
		\EndProcedure
		\State
		\State {\bf end data structure}
\end{algorithmic}
\end{algorithm}

\begin{algorithm}[!t]
\caption{Data structure : \blue{\textsc{ScalarC}}.}\label{alg:ScalarC_fixed}
\small
	\begin{algorithmic}[1]
		\Procedure{\textsc{ScalarC}}{$h^{\appr}$} \Comment{Output a number $c\in \R$}
		%\State $\wt{w} \leftarrow w^{\appr}$
	    %\If{self.name $= \mathrm{mp}_t$}
	        %\State $\wt{\mu} \leftarrow h^{\appr}$
	    %\EndIf 
	    %\If{self.name $= \mathrm{mp}_{\Phi}$}
	        %\State $\wt{\mu} \leftarrow h^{\appr} \cdot t$
	    %\EndIf
	    %\State $\wt{x} \leftarrow \sqrt{\wt{\mu} \wt{w}}$
	    %\State $\wt{s} \leftarrow \sqrt{\wt{\mu} / \wt{w}}$
		%\State \Comment{$\wt{\mu} = \wt{x} \wt{s}$ and $\wt{w} = \wt{x} / \wt{s}$}
%%%%%%%%%%%%%%%%%%%%%%%%%%%%%%%%
%       
%	    Return $p_t$, $I_t$, 
%	    $\wh{\delta}_x = I_t + I_{\Phi} - (c_x^t p_t + c_x^{\Phi} p_{\Phi}))$,
%	    $\wh{\delta}_s = c_s^t p_t + c_s^{\Phi} p_{\Phi})$
%%%%%%%%%%%%%%%%%%%%%%%%%%%%%%%%
	    \If{self.name $= \mathrm{mp}_t$}
	        \State $c \leftarrow \frac{t^{\new}}{t} - 1$
	        %\State $c_s \leftarrow \sqrt{\frac{\wt{s}}{\wt{x}}} (\frac{t^{\new}}{t} - 1)$
	    \EndIf 
	    \If{self.name $= \mathrm{mp}_{\Phi}$}
	        \State $\wt{\mu} \leftarrow h^{\appr} \cdot t$
	        \State $c \leftarrow - \frac{\epsilon}{2} \cdot t^{\new} \cdot \frac{1}{\sqrt{t} \| \nabla \Phi_{\lambda}(\wt{\mu} / t - 1) \|_2}$
	        %\State $c_s \leftarrow - \sqrt{\frac{\wt{s}}{\wt{x}}} \cdot \frac{\epsilon}{2} \cdot t^{\new} \cdot \frac{1}{\sqrt{t} \| \nabla \Phi_{\lambda}(\wt{\mu} / t - 1) \|_2}$
	        %\State $I_{\Phi} \leftarrow - \frac{1}{\wt{s}} \cdot \frac{\epsilon}{2} \cdot t^{\new} \cdot \frac{ \sqrt{\wt{\mu}} }{\sqrt{t} \| \nabla \Phi_{\lambda}(\wt{\mu} / t - 1) \|_2 $
	    \EndIf
	    \State \Return $c$
		\EndProcedure
	\end{algorithmic}
\end{algorithm}

\begin{algorithm}[!ht]
\caption{Data structure : feasible version of \textsc{UpdateQuery} (Algorithm~\ref{alg:update_query_improved})}\label{alg:update_query_fixed}
\small
	\begin{algorithmic}[1]	
	\State{\bf data structure} \Comment{Theorem~\ref{thm:main_data_structure_improved}}
	\State	
		\Procedure{\textsc{UpdateQuery}}{$w^{\new},~h^{\new}$} \Comment{Theorem~\ref{thm:correctness_fixed}}
		\State $h^{\appr}, p, \wt{p}\leftarrow$ \textsc{UpdateG}$(h^{\new})$ \Comment{Algorithm~\ref{alg:update_g_improved}, $p$ and $\wt{p}$ are only used for analysis.} \label{alg:line:update_g_in_update_query_fixed}
		\State $w^{\appr}, k, \wt{k}\leftarrow$ \textsc{UpdateV}$(w^{\new}, \blue{h^{\appr}})$ \Comment{Algorithm~\ref{alg:update_v_improved}, $k$ and $\wt{k}$ are only used for analysis.} \label{alg:line:update_v_in_update_query_fixed}
		\State $r\leftarrow$ \textsc{Query}$(w^{\appr}, h^{\appr})$ \label{alg:line:query_in_update_query_fixed} \Comment{Algorithm~\ref{alg:query_fixed}, Lemma~\ref{lem:query_correct_improved},~\ref{lem:query_time_improved},~\ref{lem:query_time_fixed}}
		\State \Comment{ Compute $r = R[l]^{\top}R[l]\sqrt{ W^{\appr} } A^\top ( A W^{\appr} A^\top )^{-1} A \sqrt{ W^{\appr} } f( h^{\appr} ) $}
		%\State \blue{\textsc{MakeFeasible}($w^{\appr}$)}
		\State \blue{$c \leftarrow \textsc{ScalarC}(h^{\appr})$}
		\If{\blue{self.name $= \mathrm{mp}_{t}$}}
		    \State \blue{$\wt{\mu} \leftarrow h^{\appr}$} \label{alg:line:update_query_fixed_mu_t}
		\EndIf
		\If{\blue{self.name $= \mathrm{mp}_{\Phi}$}}
		    \State \blue{$\wt{\mu} \leftarrow h^{\appr} \cdot t$} \label{alg:line:update_query_fixed_mu_phi}
		\EndIf
		\State \blue{$\wt{w}\leftarrow w^{\appr}$} \label{alg:line:update_query_fixed_wt_w}
		\State \blue{$\wt{x} \leftarrow \sqrt{\wt{\mu} \cdot \wt{w}}$, $\wt{s} \leftarrow \sqrt{\wt{\mu} / \wt{w}}$} \label{alg:line:update_query_fixed_wt_x_wt_s} \Comment{\blue{$\wt{\mu} = \wt{x} \wt{s}$ and $\wt{w} = \wt{x} / \wt{s}$}}
		\State \blue{$q_x \leftarrow \sqrt{\frac{\wt{x}}{\wt{s}}}\cdot c \cdot f(h^{\appr})$} \label{alg:line:update_query_fixed_q_x}
		\State \blue{$p_x \leftarrow \sqrt{\frac{\wt{x}}{\wt{s}}} \cdot c \cdot r$} \label{alg:line:update_query_fixed_p_x}
		\State \blue{$p_s \leftarrow \sqrt{\frac{\wt{s}}{\wt{x}}} \cdot c \cdot r$} \label{alg:line:update_query_fixed_p_s}
		\State \blue{\Return $q_x, p_x, p_s, w^{\appr}$}
		\EndProcedure
		\State
		\State {\bf end data structure}
	\end{algorithmic}
\end{algorithm}

\begin{algorithm}[!ht]
\caption{Data structure : feasible version of \textsc{Query} (Algorithm~\ref{alg:query_improved})}\label{alg:query_fixed}
\small
	\begin{algorithmic}[1]
	\State{\bf data structure} \Comment{Theorem~\ref{thm:main_data_structure_improved}}
	\State	
	\Procedure{\textsc{Query}}{$w^{\appr}, h^{\appr}$} \Comment{Lemma~\ref{lem:query_correct_improved},~\ref{lem:query_time_improved},~\ref{lem:query_time_fixed}}
	    	\State $\partial \Delta, \partial \Gamma, \partial \xi,  \partial S, \Delta^{\new}, \_, \_, S^{\new}, S' \leftarrow \textsc{ComputeLocalVariables}(w^{\appr}, \wt{g}^{\new})$
	    	\State \Comment{Algorithm~\ref{alg:compute_local_variables_improved}}
	    	%\State \Comment{$S^{\new} = \supp(w^{\appr}-v)$, $\partial S=\supp(w^{\appr}-\wt{v})$, and $S'=(S\cup\partial S)\backslash S^{\new}$}
		\State $r_1\leftarrow  \beta_{1}[l]$ \label{alg:line:query:r_1_fixed} \Comment{$r_1 \in \R^{n^b}$}
       	\State {$r_2\leftarrow Q[l]\xi + R[l]\gamma_2 + R[l] \partial \Gamma M (\xi + \partial \xi) + \big(Q[l] + R[l] \Gamma M\big) \partial \xi$} \label{alg:line:query:r_2_fixed} \Comment{$r_2 \in \R^{n^b}$}
		\State $r_3\leftarrow R[l](\Gamma+\partial \Gamma) \beta_2$ \label{alg:line:query:r_3_fixed} \Comment{$r_3 \in \R^{n^b}$}
		\State $\partial \gamma \leftarrow B\cdot (\L_r [(\beta_{2})_{\partial S\backslash S}]-\L_r [(\beta_{2})_{S'}]) + B\cdot (\L_r[(M_{\partial S \backslash S})^{\top}]-\L_r[(M_{S'})^{\top}]) \cdot (\xi+\partial \xi) + E \cdot \partial \xi$ \label{alg:line:partial_gamma_fixed}
		\State \Comment{local variable $\partial \gamma\in \R^{6n^a}$}
		\State $(U', C, U) \leftarrow \textsc{Decompose}\Big(\L_*[( \Delta^{\new}_{S^{\new}, S^{\new}})^{-1} +  M_{S^{\new}, S^{\new}} ] - \L_* [\Delta_{S, S}^{-1}+ M_{S,S}]\Big)$
		\State \Comment{\textsc{Decompose} is defined in Lemma \ref{lem:UCU_decomposition}. $U',U\in \R^{6n^a\times 3|\partial S|}$, $C\in \R^{3|\partial S|\times 3|\partial S|}$}
		%\State {\color{pink} $r_{4,4}\leftarrow \L_c[ Q_{l,S^{\new}}+R_l(\Gamma+\partial \Gamma) M_{S^{\new}} ] \cdot BU'(C^{-1}+U^{\top}BU')^{-1}U^{\top} \cdot (\gamma_1 + \partial \gamma )$ }
		\State $\partial E\leftarrow E_{\partial S} - B_{(\partial S\cap S)}\cdot M_{(\partial S\cap S), \partial S}$\label{alg:line:partial_E_1_fixed}
		\State $(\partial E)_{S'}\leftarrow -(\partial E)_{S'}$, ~ $(\partial E)_{(S\cap \partial S)\backslash S'}\leftarrow 0$ \Comment{local variable $\partial E\in \R^{6n^a\times |\partial S|}$} \label{alg:line:partial_E_2_fixed}
		%\State $U^{\tmp}\leftarrow [B_{(\partial S\backslash S)}, B_{\partial S}, E_{(\partial S\backslash S)}]$
		\State $U^{\tmp}\leftarrow [B_{\partial S}, B_{\partial S}, \partial E]$ 
		\State \Comment{local variable $U^{\tmp}\in \R^{6n^a\times 3|\partial S|}$, $U^{\tmp}=BU'$ (Corollary~\ref{cor:U_tmp_correctness})} \label{alg:line:U_tmp_fixed}
		\State $\gamma^{\tmp} \leftarrow U^{\tmp}(C^{-1}+U^{\top}U^{\tmp})^{-1}U^{\top} \cdot (\gamma_1 + \partial \gamma )$ \label{alg:line:r_4_4_tmp_fixed} \Comment{local variable, $\gamma^{\tmp}\in \R^{6n^a}$}
		\State $r_4\leftarrow \Big(\L_c[(Q[l])_{S^{\new}}]+F[l]+R[l] \Gamma (\L_c[M_{\partial S\backslash S}]-\L_c[M_{S'}]) + R[l]\partial \Gamma \L_c[M_{S^{\new}}]\Big) (\gamma^{\tmp}-\gamma_1-\partial \gamma)$ \label{alg:line:r_4_fixed}
		\State $r \leftarrow R[l]^{\top}(r_1+r_2+r_3+r_4)$ \Comment{$r \in \R^n$} \label{alg:line:r_fixed}
		\State $l \leftarrow l+1$
		\State \blue{$c \leftarrow \textsc{ScalarC}(h^{\appr})$}
		\State \blue{ $u_1 \leftarrow u_1 + c \cdot (W^{\appr}-\wt{V})\Big(\beta_2 + M \cdot \big(\sqrt{W^{\appr}}f(h^{\appr}) - \sqrt{V} f(g) + \mathbf{1}_{S^{\new}}(\gamma^{\tmp}-\gamma_1-\partial \gamma)\big)\Big)$}  \label{alg:line:query_u_1_update_fixed}
		\State \blue{\Comment{$\mathbf{1}_{S^{\new}} \in \R^{n\times 6n^a}$ only has ones in positions $(i,i)$ for $i\in S^{\new}$}}
		\State \blue{ $u_2 \leftarrow u_2 + c \cdot \left(\sqrt{W^{\appr}}f(h^{\appr}) + \mathbf{1}_{S^{\new}}(\gamma^{\tmp}-\gamma_1-\partial \gamma)\right)$}\label{alg:line:query_u_2_update_fixed}
		\State \blue{ $u_4 \leftarrow u_4 + c \cdot \left( \sqrt{W^{\appr}}f(h^{\appr}) + \mathbf{1}_{S^{\new}}(\gamma^{\tmp}-\gamma_1-\partial \gamma) \right)$}\label{alg:line:query_u_4_update_fixed}
		\State \Return $r$
	\EndProcedure
	\State
	\State {\bf end data structure}
	\end{algorithmic}
\end{algorithm}

\begin{algorithm}[!ht]
\caption{Data structure : feasible version of \textsc{MatrixUpdate} (Algorithm~\ref{alg:matrix_update_improved})}\label{alg:matrix_update_fixed}
\small
	\begin{algorithmic}[1]	
	\State{\bf data structure} \Comment{Theorem~\ref{thm:main_data_structure_improved}}
	\State	
		\Procedure{\textsc{MatrixUpdate}}{$w^{\appr},\blue{h^{\appr}} $} \Comment{Lemma~\ref{lem:matrix_update_correct_improved},~\ref{lem:matrix_update_time_improved},~\ref{lem:matrix_update_time_fixed}}
	    \State $\_, \_, \_,  \_, \Delta^{\new}, \Gamma^{\new}, \_, S^{\new}, \_ \leftarrow \textsc{ComputeLocalVariables}(w^{\appr}, \_)$ \label{alg:line:fixed_matrix_update_com_local} \Comment{Algorithm~\ref{alg:compute_local_variables_improved}}
		\State $M^{\tmp} \leftarrow M - M_{S^{\new}} \cdot ( (\Delta^{\new}_{S^{\new},S^{\new}})^{-1} + M_{S^{\new},S^{\new}} )^{-1} \cdot ( M_{S^{\new}} )^\top$ \label{alg:line:matrix_update_fixed:M^new}
		\State $Q^{\tmp} \leftarrow Q + R(\Gamma^{\new} M^{\tmp}) + R\sqrt{V}(M^{\tmp} -M) $ \label{alg:line:matrix_update_fixed:Q^new}
		\State $\beta_1^{\tmp} \leftarrow Q^{\tmp} \sqrt{ W^{\appr} } f ( g ) $ \label{alg:line:matrix_update_fixed:beta_1}
		\State $\beta_2^{\tmp} \leftarrow M^{\tmp} \sqrt{ W^{\appr} } f ( g ) $ \label{alg:line:matrix_update_fixed:beta_2}
		\State $\xi^{\tmp} \leftarrow \sqrt{W^{\appr}}(f(\wt{g})-f(g))$ \label{alg:line:matrix_update_fixed:xi}
		\State \Comment{We start to refresh variables in the memory of data structure}
		\State \blue{$c \leftarrow \textsc{ScalarC}(h^{\appr})$}
		\State \blue{ $G \leftarrow W^{\appr}M^{\tmp}$} \label{alg:line:matrix_update_fixed:G}
		\State \blue{ $u_1 \leftarrow u_1 + G u_2 + c \cdot W^{\appr}M^{\tmp}\sqrt{W^{\appr}}f(h^{\appr}) $}  \label{alg:line:matrix_update_fixed:u_1}
		\State \blue{ $u_3 \leftarrow u_3 + M u_4 + c \cdot M^{\tmp}\sqrt{W^{\appr}}f(h^{\appr})$}
		\label{alg:line:matrix_update_fixed:u_3}
		\State \blue{ $u_2 \leftarrow 0$, $u_4 \leftarrow 0$}  \label{alg:line:matrix_update_fixed:u_2_u_4}
		\State $Q\leftarrow Q^{\tmp}$, $M\leftarrow M^{\tmp}$ \label{alg:line:matrix_update_fixed:Q_M}
		\State $\beta_1\leftarrow \beta_1^{\tmp}$, $\beta_2\leftarrow \beta_2^{\tmp}$, $\xi \leftarrow \xi^{\tmp}$
		\State $v \leftarrow \wt{v} \leftarrow w^{\appr}$ \label{alg:line:matrix_update_fixed:v_g}
		\State $B\leftarrow I$, $F\leftarrow 0$, $E\leftarrow 0$ \label{alg:line:matrix_update_fixed:B_F_E}
		\State $S\leftarrow\emptyset$, $\Delta \leftarrow \Gamma \leftarrow 0$, $\gamma_1 \leftarrow \gamma_2 \leftarrow 0$ \label{alg:line:matrix_update_fixed:everything_else}
		\EndProcedure
		\State
		\State {\bf end data structure}
	\end{algorithmic}
\end{algorithm}

\begin{algorithm}[!ht]
\caption{Data structure : feasible version of \textsc{PartialMatrixUpdate} (Algorithm~\ref{alg:partial_matrix_update_improved}). }\label{alg:partial_matrix_update_fixed}
\small
	\begin{algorithmic}[1]	
	\State{\bf data structure} \Comment{Theorem~\ref{thm:main_data_structure_improved}}
	\State	
		\Procedure{\textsc{PartialMatrixUpdate}}{$w^{\appr},\blue{h^{\appr}}$} \Comment{Lemma~\ref{lem:partial_matrix_update_correct_improved},~\ref{lem:partial_matrix_update_time_improved},~\ref{lem:partial_matrix_update_time_fixed}}
	    \State $\_, \partial \Gamma, \_, \partial S, \Delta^{\new}, \Gamma^{\new}, \_, S^{\new}, \_ \leftarrow \textsc{ComputeLocalVariables}(w^{\appr}, \_)$ \Comment{Algorithm~\ref{alg:compute_local_variables_improved}}
		\State $(U', C, U) \leftarrow \textsc{Decompose}\Big(\L_*[(\Delta^{\new}_{S^{\new}, S^{\new}})^{-1} +  M_{S^{\new}, S^{\new}} ] - \L_* [\Delta_{S, S}^{-1}+ M_{S,S}]\Big)$ \label{alg:line:partial_matrix_update_fixed:UCU}
		\State \Comment{\textsc{Decompose} is defined in Lemma \ref{lem:UCU_decomposition}}
		\State $B^{\tmp} \leftarrow B - BU'(C^{-1}+U^{\top}BU')^{-1}U^{\top}B$ \label{alg:line:partial_matrix_update_fixed:B}
		%\State $F^{\tmp} \leftarrow F + \L_c[R\Gamma M_{\partial S\backslash S} + R\partial \Gamma M_{S^{\new}}]$
		\State $F^{\tmp} \leftarrow F +R \Gamma \cdot (\L_c[M_{\partial S\backslash S}]-\L_c[M_{S'}]) + R\partial \Gamma\cdot \L_c[M_{S^{\new}}]$
		\label{alg:line:partial_matrix_update_fixed:F}
%		\State {\color{pink} $E^{\tmp} \leftarrow E + B\L_r[ (M_{\partial S\backslash S})^{\top}]$}
		%\State $E^{\tmp} \leftarrow E + B^{\tmp}\L_r[ (M_{\partial S\backslash S})^{\top}] - BU'(C^{-1}+U^{\top}BU')^{-1}U^{\top}E$
		\State $E^{\tmp} \leftarrow E + B^{\tmp}(\L_r[ (M_{\partial S\backslash S})^{\top}]-\L_r[ (M_{S'})^{\top}]) - BU'(C^{-1}+U^{\top}BU')^{-1}U^{\top}E$ \label{alg:line:partial_matrix_update_fixed:E}
		\State $\xi^{\tmp} \leftarrow \sqrt{W^{\appr}}f(\wt{g}) - \sqrt{V}f(g)$ \label{alg:line:partial_matrix_update_fixed:xi}
		\State $\gamma_1^{\tmp} \leftarrow B^{\tmp}\cdot \L_r [\beta_{2,S^{\new}}]+B^{\tmp}\cdot \L_r [(M_{S^{\new}})^{\top}] \xi^{\tmp}$ \label{alg:line:partial_matrix_update_fixed:gamma_1}
		\State $\gamma_2^{\tmp} \leftarrow \gamma_2 + (\Gamma+\partial \Gamma) M (\sqrt{W^{\appr}} - \sqrt{\wt{V}})f(\wt{g}) + \partial \Gamma M (\sqrt{\wt{V}}f(\wt{g}) - \sqrt{V}f(g))$ \label{alg:line:partial_matrix_update_fixed:gamma_2}
		\State \blue{$c \leftarrow \textsc{ScalarC}(h^{\appr})$}
		\State \blue{ $G \leftarrow G + (W^{\appr}-\wt{V}) \cdot M$} \label{alg:line:partial_matrix_update_fixed:G}
		\State \blue{ $u_1 \leftarrow u_1 + (W^{\appr}-\wt{V}) \cdot M u_2$} \label{alg:line:partial_matrix_update_fixed:u_1}
		\State \blue{ $u_2 \leftarrow u_2 + c \cdot \left( \sqrt{W^{\appr}}f(h^{\appr}) - \mathbf{1}_{S^{\new}}B^{\tmp}\L_r[(M_{S^{\new}})^{\top}]\sqrt{W^{\appr}}f(h^{\appr})\right)$} \label{alg:line:partial_matrix_update_fixed:u_2}
		\State \blue{ $u_4 \leftarrow u_4 + c \cdot \left( \sqrt{W^{\appr}}f(h^{\appr}) - \mathbf{1}_{S^{\new}}B^{\tmp}\L_r[(M_{S^{\new}})^{\top}]\sqrt{W^{\appr}}f(h^{\appr}) \right)$} \label{alg:line:partial_matrix_update_fixed:u_4}
		\State \Comment{We start to refresh variables in the memory of data structure}
		\State $B\leftarrow B^{\tmp}$, $F\leftarrow F^{\tmp}$, $E\leftarrow E^{\tmp}$
		\State $\xi \leftarrow \xi^{\tmp}$ $\gamma_1\leftarrow \gamma_1^{\tmp}$, $\gamma_2\leftarrow \gamma_2^{\tmp}$
		\State $\wt{v} \leftarrow w^{\appr}, S \leftarrow S^{\new}$, $\Delta \leftarrow \Delta^{\new}$, $\Gamma \leftarrow \Gamma^{\new}$ \label{alg:line:partial_matrix_update_fixed:everything_else}
		\EndProcedure
		\State
		\State {\bf end data structure}
	\end{algorithmic}
\end{algorithm}

\section{A feasible algorithm}
\label{sec:feasible}
In previous sections, we show a fast algorithm calculating an LP solution $\ov{x}$. However, $\ov{x}$ is not always a feasible solution since we used sketching to calculate $\wh{\delta}_x$ and hence $A\wh{\delta}_x$ is not $0$. In this section we present how to turn the output $x$ of Theorem~\ref{thm:tech_third_improvement} to a feasible LP solution, i.e. $\|Ax-b\|_1$ is bounded. Our technique is based on the robust central path of \cite{lsz19}, and we extend it to two level update setting. Our algorithm use $G,M,u_1,u_2,u_3,u_4$ to implicitly maintain the solutions $x = G u_1 + u_2$ and $s = M u_3 + u_4$. Each iteration, the algorithm updates $G,M,u_1,u_2,u_3,u_4$ so that $x$ and $s$ change by $\wt{\delta}_x$ and $\wt{\delta}_s$ in each iteration (see Definition~\ref{def:widetilde}). In this way, we can postpone the expensive matrix vector multiplication to every $\sqrt{n}$ iterations. The running time of maintaining $x$ and $s$ is dominated by the pre-existing computations, so our algorithm still achieves the same overall running time. Also since we do not multiply the sketching matrix on the left when computing $\wt{\delta}_x$ and $\wt{\delta}_s$, $A \wt{\delta}_x = 0$ is always satisfied and in each iteration we always have $A x = A x^{(0)} = b$.

We introduce a smaller error constant $\epsilon_{\mathrm{tiny}}$ in this section.
\begin{definition} \label{def:epsilon_tiny}
We define $\epsilon_{\mathrm{tiny}} = \frac{ \epsilon_{\mathrm{mp}} \cdot \epsilon }{ 3200 \log^3 n }$.
\end{definition}

To make the output feasible, we present in this section the modified algorithms. In Section~\ref{sec:error_analysis_fixed} we show that the error guarantees of central path method still has the same bound with the modifications, this section should be seen as a complement of Section~\ref{sec:sketching_on_the_left_and_vector_maintenance}. In Section~\ref{sec:correctness_fixed} we prove the correctness of this feasible algorithm when implicitly maintaining $x$ and $s$, this section should be seen as a complement of Section~\ref{sec:correctness_improved}. In Section~\ref{sec:time_per_call_feasible} we bound that the running time of maintaining $x$ and $s$, this section should be seen as a complement of Section~\ref{sec:time_per_call_improved}.

\subsection{Analysis}
\label{sec:error_analysis_fixed}
Consider the $j$-th iteration. Assume at the beginning of the $j$-th iteration, we have $x^{(j)}$, $s^{(j)}$, $\ov{x}^{(j)}$, $\ov{s}^{(j)}$. Define $w := w^{(j)} = x^{(j)} s^{(j)}$, $\mu := \mu^{(j)} =  \frac{x^{(j)}}{s^{(j)}}$, $\ov{w} := \ov{w}^{(j)} =  \ov{x}^{(j)} \ov{s}^{(j)}$, $\ov{\mu} := \ov{\mu}^{(j)} =  \frac{\ov{x}^{(j)}}{\ov{s}^{(j)}}$, $w^{\new} := w^{(j+1)} = x^{(j+1)} s^{(j+1)}$, $\mu^{\new} := \mu^{(j+1)} = \frac{x^{(j+1)}}{s^{(j+1)}}$. We will prove \emph{inductively} that the guarantees of Section~\ref{sec:sketching_on_the_left_and_vector_maintenance} are still satisfied for the modified algorithm.

\paragraph{Robustness of central path.}
We first prove the following two statements about the robustness of central path method:
\begin{enumerate}
    \item The $w^{\appr}$ and $h^{\appr}$ used in the data structures satisfy $w^{\appr} \approx_{2\epsilon_{\mathrm{mp}}} w$, and $h^{\appr} \approx_{2\epsilon_{\mathrm{mp}}} \mu$. (This corresponds to Part 1 and 2 of Assumption~\ref{ass:assumption}. We only lose a constant factor here.)
    \item $\mu^{\new} \approx_{0.2} t$. (This corresponds to Part 3 of Assumption~\ref{ass:assumption}. We only loose a constant factor here.)
\end{enumerate}
In Part 3 of Theorem~\ref{thm:correctness_fixed}, we prove that in any iteration, $x,\ov{x}$, $s,\ov{s}$ are entry-wise close with high probability, i.e. $\ov{x} \approx_{\epsilon_{\mathrm{tiny}}} x$ and $\ov{s} \approx_{\epsilon_{\mathrm{tiny}}} s$ holds with probability $1-1/\poly(n)$. We will use this to prove that the above two statements are still satisfied in our modified algorithm.

\begin{enumerate}
    \item The data structure directly ensures $w^{\appr} \approx_{\epsilon_{\mathrm{mp}}} \ov{w}$, and $h^{\appr} \approx_{\epsilon_{\mathrm{mp}}} \ov{\mu}$. And since $\ov{x}^{(j)} \approx_{\epsilon_{\mathrm{tiny}}} x^{(j)}$ and $\ov{s}^{(j)} \approx_{\epsilon_{\mathrm{tiny}}} s^{(j)}$ (Part 3 of Theorem~\ref{thm:correctness_fixed}) we directly get the desired result that $w^{\appr} \approx_{2\epsilon_{\mathrm{mp}}} w$ and $h^{\appr} \approx_{2\epsilon_{\mathrm{mp}}} \mu$.
    \item Note that previously $\mu \approx_{0.1} t$ was proved in Lemma~\ref{lem:potential_martingale} by bounding the potential function. And the proof of Lemma~\ref{lem:potential_martingale} uses the bounds given by Lemma~\ref{lem:bounding_mu_new_minus_mu}. We define the potential function to be $\Phi(\frac{xs}{t} - 1)$ here instead of $\Phi(\frac{\ov{x}\ov{s}}{t} - 1)$. Note that conditioned on $\ov{x}^{(j)}$ and $\ov{s}^{(j)}$, $x^{(j+1)}$ and $s^{(j+1)}$ are deterministic. Thus Part 2 and Part 4 of Lemma~\ref{lem:bounding_mu_new_minus_mu} are trivial. We still have an analog of Part 1 and Part 3 of Lemma~\ref{lem:bounding_mu_new_minus_mu}: $\| \ov{\mu}^{-1} ( \mu^{\new} - \ov{\mu} - \ov{\delta}_{t} - \wt{\delta}_{\Phi} ) \|_2 \leq O( \epsilon_{\mathrm{mp}}$ and $\| \ov{\mu}^{-1} ( \mu^{\new} - \ov{\mu} ) \|_{\infty} \leq O(\epsilon)$ using the fact that $\ov{x}^{(j)}$ and $\ov{s}^{(j)}$ are close to $x^{(j)}$ and $s^{(j)}$ (Part 3 of Theorem~\ref{thm:correctness_fixed}). Thus we still have an analog of Lemma~\ref{lem:potential_martingale} that 
    \begin{align*}
    \Phi_{\lambda} ( \mu^{\new} / t^{\new}  - 1 ) \leq \Phi_{\lambda} (   \ov{\mu} / t  - 1 ) - \frac{\lambda \epsilon}{15\sqrt{n}} \cdot ( \Phi_{\lambda} (  \ov{\mu} / t  - 1 ) - 10 n ).
    \end{align*}
    Using Part 3 of Theorem~\ref{thm:correctness_fixed} again and the fact that the derivative of $\cosh(x)$ function is constant when $x<2$, we have
    \begin{align*}
    \Phi_{\lambda} ( \mu^{\new} / t^{\new}  - 1 ) \leq \Phi_{\lambda} ( \mu / t  - 1 ) - \Omega(\frac{\lambda \epsilon}{\sqrt{n}}) \cdot ( \Phi_{\lambda} (  \mu / t  - 1 ) - O( n^2 ) ).
    \end{align*}
    Thus we can still inductively prove a polynomial upper bound on the potential function $\Phi_{\lambda} ( \frac{ xs }{ t } - 1 )$. Note that we only loose a constant factor in the approximation ratio of $\mu$ with $t$ when the last term is $O(n^2)$ instead of $O(n)$.
\end{enumerate}

\paragraph{$w$ and $h$ move slowly.} The last part of the inductive analysis is to show that $w$ and $\mu$ are both moving slowly as that of Section~\ref{sec:w_movement} and \ref{sec:mu_movement}. Now we only have the guarantee for $w$ and $\mu$, instead of $\ov{w}$ and $\ov{\mu}$, so we use $\psi ( w_{i} / v_{i}  - 1 )$ and $\psi (  w_{i} / \wt{v}_{i} - 1 )$ in the analysis and use $\psi ( \ov{w}_{i} / v_{i}  - 1 )$ and $\psi (  \ov{w}_{i} / \wt{v}_{i} - 1 )$ for the algorithm (because the algorithm doesn't know $w$ and $\mu$). The amortized analysis of Section~\ref{sec:amortize_time_improved} need to be modified accordingly. We have the following two statements:

\begin{enumerate}
    \item The three bounds of Lemma~\ref{lem:w_movement} and Lemma~\ref{lem:mu_movement} hold for $w^{\new} / w - 1$ and $\mu^{\new} / \mu - 1$. (Originally used in subsection~\ref{sec:w_move_partial_matrix_update}, \ref{sec:w_move_matrix_update}.)
    \item Lemma~\ref{lem:v_wt_v_move_matrix_update} and Lemma~\ref{lem:v_wt_v_move_partial_matrix_update} still hold for new potential function with use $w$ instead of $\ov{w}$.
\end{enumerate}
\begin{proof-sketch}
\begin{enumerate}
    \item Note that Lemma~\ref{lem:w_movement} and Lemma~\ref{lem:mu_movement} only use the relative error bounds stated in Lemma~\ref{lem:stochastic_step}. We can prove that all $\ov{x}^{-1}$, $\ov{s}^{-1}$ and $\ov{\mu}^{-1}$ terms of Lemma~\ref{lem:stochastic_step} can be replaced by $x$, $s$, and $\mu$ since we proved that in iteration $j$, $x^{(j)} \approx_{\epsilon_{\mathrm{tiny}}} \ov{x}^{(j)}$ and $s^{(j)} \approx_{\epsilon_{\mathrm{tiny}}} \ov{s}^{(j)}$. Following the same proof of Lemma~\ref{lem:w_movement} and Lemma~\ref{lem:mu_movement}, and using the error bound of the $x$ and $s$ version of Lemma~\ref{lem:stochastic_step}, we can prove the desired statement.
    \item Since $x \approx_{\epsilon_{\mathrm{tiny}}} \ov{x}$ in any iteration (Part 3 of Theorem~\ref{thm:correctness_fixed}), we have $w^{(j+1)}\approx_{\epsilon_{\mathrm{tiny}}} \ov{w}^{(j+1)}$, hence $\psi ( \ov{w}^{(j+1)}_{i} / v^{(j+1)}_{i}  - 1 )$ is within the range of $[\psi ( w^{(j+1)}_{i} / v^{(j+1)}_{i}  - 1 ) - \epsilon_{\mathrm{tiny}}, \psi ( w^{(j+1)}_{i} / v^{(j+1)}_{i}  - 1 ) + \epsilon_{\mathrm{tiny}}]$. So it is fine to use $\psi ( w^{(j+1)}_{i} / v^{(j+1)}_{i}  - 1 )$ in the analysis while using $\psi ( \ov{w}^{(j+1)}_{i} / v^{(j+1)}_{i}  - 1 )$ in the algorithm. The only problem we need to resolve is that when $v$ and $\wt{v}$ are updated to $w^{\appr}=\ov{w}$, the potential function is not cleared to be $0$, but instead remain a small number $\epsilon_{\mathrm{tiny}} \ll \epsilon_{\mathrm{mp}}$. But this problem is already solved in Lemma~\ref{lem:v_wt_v_move_partial_matrix_update}.
\end{enumerate}
\end{proof-sketch}

\subsection{Correctness of feasible algorithm}
\label{sec:correctness_fixed}
We first present the following main theorem, and prove it using lemmas proved subsequently.
\begin{theorem} \label{thm:correctness_fixed}
In the $j$-th iteration of the while-loop from Line~\ref{alg:line:t_while_loop_fixed} to Line~\ref{alg:line:main_end_while_fixed} of \textsc{Main}(Algorithm~\ref{alg:main_fixed}), let $\ov{x}^{(j)}$ and $\ov{s}^{(j)}$ be the values of global variables $\ov{x}$ and $\ov{s}$ at the beginning of $j$-th iteration, and let $\ov{x}^{(j+1)}$, $\ov{s}^{(j+1)}$, $x^{(j+1)}$, $s^{(j+1)}$ be the values of global variables $\ov{x}$, $\ov{s}$, $x$, $s$ at the end of $j$-th iteration. Let $\wt{x}$, $\wt{s}$, $\wt{P}$, $\wt{\delta}_t$, $\wt{\delta}_{\Phi}$, $\wt{\delta}_x$, $\wt{\delta}_s$ be defined as in Definition~\ref{def:widetilde}. Let $R\in\R^{b\times n}$ be the subsampled randomized Hadamard matrix we used in our algorithm, defined in Definition~\ref{def:song_coordiante_embedding}.

Then our algorithm guarantees that\\
1. The output of \textsc{OneStepCentralPath}(Algorithm~\ref{alg:one_step_central_path_fixed}) satisfies
    \begin{align*}
        \wh{\delta}_x =  \sqrt{\frac{\wt{X}}{\wt{S}}}(I - (R^{\top}R)\wt{P})\frac{1}{\sqrt{\wt{X}\wt{S}}}(\wt{\delta}_t + \wt{\delta}_{\Phi}) , ~~~
        \wh{\delta}_s =  \sqrt{\frac{\wt{S}}{\wt{X}}} (R^{\top}R)\wt{P}\frac{1}{\sqrt{\wt{X}\wt{S}}}(\wt{\delta}_t + \wt{\delta}_{\Phi}),
    \end{align*}$\wh{\delta}_x$ and $\wh{\delta}_s$ match the definition in Definition~\ref{def:hat}. \\
2. 
    In each iteration, when the algorithm reaches \textsc{MakeFeasible}(Algorithm~\ref{alg:makefeasible_fixed}) on Line~\ref{alg:line:make_feasible_fixed} of \textsc{Main} (Algorithm~\ref{alg:main_fixed}), the following holds:
    \begin{align*}
        x - \big((\mathrm{mp}_{t}.u_1) + (\mathrm{mp}_{t}.G)\cdot (\mathrm{mp}_{t}.u_2)\big) - \big((\mathrm{mp}_{\Phi}.u_1) + (\mathrm{mp}_{\Phi}.G)\cdot (\mathrm{mp}_{\Phi}.u_2)\big) = x^{(0)} + \sum_{i=1}^{j} \wt{\delta}_x^{(i)}\\
        s + \big((\mathrm{mp}_{t}.u_3) + (\mathrm{mp}_{t}.M) \cdot (\mathrm{mp}_{t}.u_4)\big) + \big((\mathrm{mp}_{\Phi}.u_3) + (\mathrm{mp}_{\Phi}.M) \cdot (\mathrm{mp}_{\Phi}.u_4)\big) = s^{(0)} + \sum_{i=1}^{j} \wt{\delta}_s^{(i)}.
    \end{align*}
3. $\ov{x} \approx_{\epsilon_{\mathrm{tiny}}} x$, $\ov{s} \approx_{\epsilon_{\mathrm{tiny}}} s$ with high probability.
\end{theorem}
\begin{proof}
For simplicity, we only prove these three statements for $x$. The case for $s$ follows from similar reasons. Since we have two data structures sharing the same code, we denote $q_{x,t}$, $q_{x,\Phi}$, $p_{x,t}$, $p_{x,\Phi}$ as the $q_x$ and $p_x$ defined on Line~\ref{alg:line:update_query_fixed_q_x} and \ref{alg:line:update_query_fixed_p_x} of \textsc{UpdateQuery} (Algorithm~\ref{alg:update_query_fixed}) in data structure $\mathrm{mp}_t$ and $\mathrm{mp}_{\Phi}$ respectively. Also, 
We denote $c_t$ as the output of \textsc{ScalarC} (Algorithm~\ref{alg:ScalarC_fixed}) of data structure $\mathrm{mp}_t$, and $c_{\Phi}$ corresponds to $\mathrm{mp}_{\Phi}$.
From the description of \textsc{ScalarC}, we have 
\begin{align}\label{eq:description_scalarc}
c_t :=  (\frac{t^{\new}}{t} - 1)  ~~~~~~~~~~~~~
c_{\Phi} :=  - \frac{\epsilon}{2} \cdot t^{\new} \cdot \frac{1}{\sqrt{t} \| \nabla \Phi_{\lambda}(\wt{\mu} / t - 1) \|_2}.
\end{align}
And recall $f_t$, $f_{\Phi}$ defined in Line~\ref{alg:line:main_fixed_f_t_f_phi}, Algorithm~\ref{alg:main_fixed}) are
\begin{align}\label{eq:description_f}
f_t(x) := \sqrt{x} ~~~~~~~~~~~~~ f_{\Phi}(x):= \nabla \Phi_{\lambda}(x)/\sqrt{x}.
\end{align}

In \textsc{UpdateQuery} (Algorithm~\ref{alg:update_query_fixed}), the two data structure approximate $w^{\appr}$ and $h^{\appr}$ in the same way, so their $\wt{w}$ (Line~\ref{alg:line:update_query_fixed_mu_t}) and $\wt{\mu}$ (Line~\ref{alg:line:update_query_fixed_mu_t} and \ref{alg:line:update_query_fixed_mu_phi}) are the same. So they also have the same $\wt{x}$ and $\wt{s}$ (Line~\ref{alg:line:update_query_fixed_wt_x_wt_s}). Note that $\wt{w}$, $\wt{\mu}$, $\wt{x}$, $\wt{s}$ calculated in the data structure all matches Definition~\ref{def:widetilde}.

We first show the following that will be used in both Part 1 and 2:
\begin{align}\label{eq:c_t_f_t_plus_c_phi_f_phi}
    c_t\cdot f_t(\wt{\mu}) + c_{\Phi} \cdot f_{\Phi}(\wt{\mu}/t) = & ~ (\frac{t^{\new}}{t} - 1) \cdot \sqrt{\wt{\mu}} - \frac{\epsilon}{2} \cdot t^{\new} \cdot \frac{1}{\sqrt{t} \| \nabla \Phi_{\lambda}(\wt{\mu} / t - 1) \|_2} \cdot \nabla \Phi_{\lambda}(\wt{\mu}/t)/\sqrt{\wt{\mu}/t} \notag \\
    = & ~ \frac{1}{\sqrt{\wt{\mu}}}\left( (\frac{t^{\new}}{t} - 1) \cdot \wt{\mu} - \frac{\epsilon}{2} \cdot t^{\new} \cdot \frac{\nabla \Phi_{\lambda}(\wt{\mu}/t)}{ \| \nabla \Phi_{\lambda}(\wt{\mu} / t - 1) \|_2} \right) \notag \\
    = &~ \frac{1}{\sqrt{\wt{\mu}}}(\wt{\delta}_t + \wt{\delta}_{\Phi}) 
    =  \frac{1}{\sqrt{\wt{X}\wt{S}}}(\wt{\delta}_t + \wt{\delta}_{\Phi}),
\end{align}
where the first step is by the definition of $c$ and $f$ (Eq.\eqref{eq:description_scalarc},\eqref{eq:description_f}), the third step is by the definition of $\wt{\delta}_t$ and $\wt{\delta}_{\Phi}$ (Definition~\ref{def:widetilde}), the last step is by $\wt{\mu}=\wt{x}\wt{s}$.

\noindent {\bf Part 1. }
First we calculate $q_{x,t} + q_{x,\Phi}$:
\begin{align}\label{eq:q_x_t_q_x_phi}
   q_{x,t} + q_{x,\Phi}
=  \sqrt{\frac{\wt{X}}{\wt{S}}}\left(c_t\cdot f_t(\wt{\mu}) + c_{\Phi} \cdot f_{\Phi}(\wt{\mu}/t)\right)
=  \sqrt{\frac{\wt{X}}{\wt{S}}}\frac{1}{\sqrt{\wt{X}\wt{S}}}(\wt{\delta}_t + \wt{\delta}_{\Phi})
\end{align}
where the first step is by assignment of $q_x$ (Line~\ref{alg:line:update_query_fixed_q_x} of \textsc{UpdateQuery}, Algorithm~\ref{alg:update_query_fixed}), the second step is by Eq.\eqref{eq:c_t_f_t_plus_c_phi_f_phi}.

Next, we calculate $p_{x,t} + p_{x,\Phi}$. Note that the output $r$ of \textsc{Query} is calculated in the same way as before, so by Lemma~\ref{lem:query_correct_improved} we have
\begin{align}\label{eq:output_r_guarantee}
    r = R[l]^{\top} R[l] \sqrt{W^{\appr}}A^{\top}(AW^{\appr}A^{\top})^{-1}A\sqrt{W^{\appr}}f(h^{\appr})
    =R[l]^{\top} R[l] \wt{P} f(h^{\appr}).
\end{align}

Therefore,
\begin{align*}
       p_{x,t} + p_{x,\Phi}
    = & ~ \sqrt{\frac{\wt{X}}{\wt{S}}} \left(c_t \cdot r_t + c_{\Phi} \cdot r_{\Phi} \right) 
    =  \sqrt{\frac{\wt{X}}{\wt{S}}} \left(c_t \cdot R[l]^{\top} R[l] \wt{P} f_t(\wt{\mu}) + c_{\Phi} \cdot R[l]^{\top} R[l] \wt{P} f_{\Phi}(\wt{\mu}/t) \right)\\
    = & ~ \sqrt{\frac{\wt{X}}{\wt{S}}}R[l]^{\top} R[l] \wt{P}(c_t \cdot f_t(\wt{\mu}) + c_{\Phi} \cdot f_{\Phi}(\wt{\mu}/t))
    =  \sqrt{\frac{\wt{X}}{\wt{S}}}R[l]^{\top} R[l] \wt{P}\frac{1}{\sqrt{\wt{X}\wt{S}}}\cdot(\wt{\delta_t} + \wt{\delta}_{\Phi}),
\end{align*}
where the first step is by assignment of $p_x$ (Line~\ref{alg:line:update_query_fixed_p_x} of \textsc{UpdateQuery}, Algorithm~\ref{alg:update_query_fixed}), the second step is by Eq.\eqref{eq:output_r_guarantee}, the last step is by Eq.\eqref{eq:c_t_f_t_plus_c_phi_f_phi}.

Then as we calculate $\wh{\delta}_x$ in line~\ref{alg:line:one_step_central_path_fixed_wh_delta_x} of \textsc{OneStepCentralPath} (Algorithm~\ref{alg:one_step_central_path_fixed}), 
\begin{align*}
    \wh{\delta}_x 
    = & ~ q_{t,x} + q_{\Phi,x} - (p_{t,x} + p_{\Phi,x})\\
    = & ~ \sqrt{\frac{\wt{X}}{\wt{S}}}\frac{1}{\sqrt{\wt{X}\wt{S}}}(\wt{\delta}_t + \wt{\delta}_{\Phi}) - \sqrt{\frac{\wt{X}}{\wt{S}}}R[l]^{\top} R[l] \wt{P}\frac{1}{\sqrt{\wt{X}\wt{S}}}\cdot(\wt{\delta_t} + \wt{\delta}_{\Phi})\\
    = & ~ \sqrt{\frac{\wt{X}}{\wt{S}}}\left(I -R[l]^{\top} R[l] \wt{P}\right)\frac{1}{\sqrt{\wt{X}\wt{S}}}\cdot(\wt{\delta_t} + \wt{\delta}_{\Phi}).
\end{align*}

\noindent {\bf Part 2.}
We prove Part 2 by induction. In the basic case when $j=0$, we initialize $x\leftarrow x^{(0)}$, so it is true.

When $j>0$, let $j_1<j$ be the last iteration that we call \textsc{Initialize} on Line~\ref{alg:line:main_fixed_initialize_t_when_j_geq_sqrt_n} and \ref{alg:line:main_fixed_initialize_phi_when_j_geq_sqrt_n} of \textsc{Main} (Algorithm~\ref{alg:main_fixed}). In the $j_1$-th iteration, the \textsc{Main} algorithm executes Line~\ref{alg:line:x_update_fixed}: $x \leftarrow x - (\mathrm{mp}_{t}.u_1 + \mathrm{mp}_{t}.G\cdot \mathrm{mp}_{t}.u_2) - (\mathrm{mp}_{\Phi}.u_1 + \mathrm{mp}_{\Phi}.G\cdot \mathrm{mp}_{\Phi}.u_2)$
, and then re-initialize $\mathrm{mp}_{t}.u_1 \leftarrow \mathrm{mp}_{t}.u_2 \leftarrow \mathrm{mp}_{\Phi}.u_1 \leftarrow \mathrm{mp}_{\Phi}.u_2\leftarrow 0$.
Therefore by induction on $j$, we have $x^{(j_1)}=x^{(0)} + \sum_{i=1}^{j_1} \wt{\delta}_x^{(i)}$.

Now apply Part 3 of Lemma~\ref{lem:invariant_fixed}, we have that
\begin{align*}
    & \big( (\mathrm{mp}_t.u_1) + (\mathrm{mp}_t.G)\cdot (\mathrm{mp}_t.u_2) \big) + \big( (\mathrm{mp}_{\Phi}.u_1) + (\mathrm{mp}_{\Phi}.G)\cdot (\mathrm{mp}_{\Phi}.u_2) \big)\\
    = &~ \sum_{i=j_1+1}^j c_t^{(i)} \cdot \sqrt{W^{\appr, (i)}}\wt{P}^{(i)}f_t(\wt{\mu}^{(i)})
      + \sum_{i=j_1+1}^j c_{\Phi}^{(i)} \cdot \sqrt{W^{\appr, (i)}}\wt{P}^{(i)} f_{\Phi}(\wt{\mu}^{(i)}/t^{(i)})\\
    = &~ \sum_{i=j_1+1}^j \sqrt{\frac{\wt{X}^{(i)}}{\wt{S}^{(i)}}}\wt{P}^{(i)}\frac{1}{\sqrt{\wt{X}^{(i)}\wt{S}^{(i)}}}(\wt{\delta}_{t}^{(i)}+\wt{\delta}_{\Phi}^{(i)}).
\end{align*}
where the first step is by Part 3 of Lemma~\ref{lem:invariant_fixed}, the second step is by Eq.\eqref{eq:c_t_f_t_plus_c_phi_f_phi}.

Also notice that in every iteration, we add $q_{t,x}+q_{\Phi,x}$ to $x$ in \textsc{OneStepCentralPath} (Line~\ref{alg:line:one_step_central_path_fixed_x_add_q_t_x_add_q_phi_x} of Algorithm~\ref{alg:one_step_central_path_fixed}), so \begin{align*}
x^{(j)}-x^{(j_1)}
= \sum_{i=j_1+1}^j ( q_{t,x}^{(i)}+q_{\Phi,x}^{(i)} )
= \sum_{i=j_1+1}^j \sqrt{\frac{\wt{X}^{(i)}}{\wt{S}^{(i)}}}\frac{1}{\sqrt{\wt{X}^{(i)}\wt{S}^{(i)}}}(\wt{\delta}_t^{(i)} + \wt{\delta}_{\Phi}^{(i)}).
\end{align*}
where the last step is by Eq.\eqref{eq:q_x_t_q_x_phi}.

Therefore, 
\begin{align*}
    &x^{(j)} - \big((\mathrm{mp}_{t}.u_1) + (\mathrm{mp}_{t}.G)\cdot (\mathrm{mp}_{t}.u_2)\big) - \big((\mathrm{mp}_{\Phi}.u_1) + (\mathrm{mp}_{\Phi}.G)\cdot (\mathrm{mp}_{\Phi}.u_2)\big)\\
    = & ~ x^{(j_1)} + \sum_{i=j_1+1}^j\left(\sqrt{\frac{\wt{X}^{(i)}}{\wt{S}^{(i)}}}\frac{1}{\sqrt{\wt{X}^{(i)}\wt{S}^{(i)}}}(\wt{\delta}_t^{(i)} + \wt{\delta}_{\Phi}^{(i)}) - \sqrt{\frac{\wt{X}^{(i)}}{\wt{S}^{(i)}}} \cdot \wt{P}^{(i)} \cdot \frac{1}{\sqrt{\wt{X}^{(i)}\wt{S}^{(i)}}}(\wt{\delta}_{t}^{(i)}+\wt{\delta}_{\Phi}^{(i)})\right)\\
    = & ~ x^{(j_1)} + \sum_{i=j_1+1}^j\wt{\delta}_x^{(i)}\\
    = & ~ x^{(0)} + \sum_{i=1}^j \wt{\delta}_x^{(i)}.
\end{align*}

\noindent {\bf Part 3.}
Consider a fixed coordinate $i$. We use $j_i$ to denote the last iteration that the algorithm includes $i$ into $\wh{S}$ when enter the \textsc{MakeFeasible} procedure on Line~\ref{alg:line:make_feasible_fixed} of \textsc{Main} (Algorithm~\ref{alg:main_fixed}) or when re-initialize. We prove the following properties in order to apply Lemma~\ref{lem:accuracy_of_ov_x_respect_to_x}.
\begin{enumerate}
    \item $\ov{x}_i^{(j_i)} = x_i^{(j_i)}$. If the algorithm enter the \textsc{Initialize} procedure, we have $\ov{x}_i^{(j_i)} = x_i^{(j_i)}$. Otherwise we enter the \textsc{MakeFeasible} procedure, and according to the updating rule $\ov{x}_{\wh{S}}\leftarrow x_{\wh{S}}$ (Line~\ref{alg:line:makefeasible_fixed_ov_x_wh_s} of \textsc{MakeFeasible}, Algorithm~\ref{alg:makefeasible_fixed}), we also have $\ov{x}_i^{(j_i)} = x_i^{(j_i)}$. 
    \item $t^{(j)} > t^{(j_i)} /2 $ and $j-j_i\leq \sqrt{n}$, since the algorithm re-\textsc{Initialize} whenever it passes $\sqrt{n}$ iterations or $t$ changes too much (Line~\ref{alg:line:main_fixed_if_j_sqrt_n} in \textsc{Main}, Algorithm~\ref{alg:main_fixed}).
    \item For all $l \in \{j_i+1,\cdots,j\}$, $w^{\appr,(l)}\in [w^{\old}/2,2w^{\old}]$, since the algorithm doesn't include coordinate $i$ in $\wt{S}$ during iteration $l$.
\end{enumerate}
Then by Lemma~\ref{lem:accuracy_of_ov_x_respect_to_x}, $x_i^{(j)}\approx_{\epsilon_x} \ov{x}_i^{(j)}$ holds with probability at least $1-\delta$, where $\epsilon_x = \frac{200n^{1/4}}{\sqrt{b}}\cdot \log(n/\delta)\epsilon \leq \frac{ \epsilon_{\mathrm{mp}} \cdot \epsilon }{ 3200 \log^3 n }$ by our assignment of $b = 10^{12}\sqrt{n}\log^8 n / \epsilon_{\mathrm{mp}}^2$ (Line~\ref{alg:line:main_fixed_b} of \textsc{Main}, Algorithm~\ref{alg:main_fixed}) and $\delta = 1 / \poly(n)$. This satisfies the requirement of Def.~\ref{def:epsilon_tiny}.

By choosing $\delta$ to be $1/n^{10}$ and union bound on all coordinates $i$, $x^{(j)}\approx_{\epsilon_{\mathrm{tiny}}} \ov{x}^{(j)}$ holds w.h.p.
\end{proof}

The following lemma proves the invariants that are true throughout the algorithm. It is used to prove Theorem~\ref{thm:correctness_fixed}. This lemma should be seen as a complement of Section~\ref{sec:correctness_improved}, and we directly use results proved in that section.

\begin{lemma}[Invariants]\label{lem:invariant_fixed}
In the end of the $j$-th iteration, the following invariants hold:
\begin{enumerate}
\item \label{ass:invariant_fixed:G} $G=\wt{V} A^{\top}(AVA^{\top})^{-1}A$.
\item \label{ass:invariant_fixed:s} $u_3 + Mu_4 = \sum_{i=j_1+1}^j c^{(i)} \cdot A^{\top}(AW^{\appr, (i)}A^{\top})^{-1}A\sqrt{W^{\appr, (i)}}f(h^{\appr, (i)})$,
\item \label{ass:invariant_fixed:x} $u_1 + Gu_2 = \sum_{i=j_1+1}^j c^{(i)} \cdot W^{\appr, (i)}A^{\top}(AW^{\appr, (i)}A^{\top})^{-1}A\sqrt{W^{\appr, (i)}}f(h^{\appr, (i)})$,
\end{enumerate}
where $j_1$ is the last iteration we call \textsc{Initialize}, and $c^{(j)}$ is defined as in \textsc{ScalarC}:
\\in data structure $\mathrm{mp}_t$,
\begin{align*}
    c := &~ (\frac{t^{\new}}{t} - 1),
\end{align*}
and in data structure $\mathrm{mp}_{\Phi}$,
\begin{align*}
    c := &~ - \frac{\epsilon}{2} \cdot t^{\new} \cdot \frac{1}{\sqrt{t} \| \nabla \Phi_{\lambda}(\wt{\mu} / t - 1) \|_2}.
\end{align*}
\end{lemma}
\begin{proof}
We consider the $j$-th iteration in this proof. 
Throughout the proof, we call the updated version of all date structure members at the end of the $j$-th iteration as the ``new'' version, with a ``new'' superscript in the notation, e.g. $u_1^{\new}$.

\noindent {\bf Part 1.} Note that $V$, $\wt{V}$ and $G$ are only updated in \textsc{MatrixUpdate} and \textsc{PartialMatrixUpdate}.

{\bf Case 1. \textsc{MatrixUpdate} (Algorithm~\ref{alg:matrix_update_fixed}).} On Line~\ref{alg:line:matrix_update_fixed:G}, $G$ is updated to be $W^{\appr}M^{\tmp}$, where $W^{\appr}$ is the new value of $\wt{V}$ (Line~\ref{alg:line:matrix_update_fixed:v_g}), and $M^{\tmp}$ is the new value of $M$ (Line~\ref{alg:line:matrix_update_fixed:Q_M}). We have
\begin{align*}
    G^{\new} = \wt{V}M = \wt{V}A^{\top}(AVA^{\top})^{-1}A,
\end{align*}
since in Lemma~\ref{lem:matrix_update_correct_improved} we proved that the invariant of $M$ still holds.

{\bf Case 2. \textsc{PartialMatrixUpdate} (Algorithm~\ref{alg:partial_matrix_update_fixed}).} First note that $V$ and $M$ are not updated in \textsc{PartialMatrixUpdate}. On Line~\ref{alg:line:partial_matrix_update_fixed:G}, $G$ is updated to be $W^{\appr}M$, where $W^{\appr}$ is the new value of $\wt{V}$, so the invariant of $G$ still holds.

\noindent {\bf Part 2.}
It suffices to prove that the additive term in the $j$-th iteration is 
\begin{align*}
    u_3^{\new} + M^{\new}u_4^{\new} - (u_3 + Mu_4) = c \cdot A^{\top}(AW^{\appr}A^{\top})^{-1}A\sqrt{W^{\appr}}f(h^{\appr}).
\end{align*}
Then starting from the $j_1$-th iteration where we re-\textsc{initialize} and set $u_3^{(j_1)}=u_4^{(j_1)}=0$, we can inductively prove the statement for iterations $i \in \{j_1+1, \cdots, j\}$.

There are three cases that we update $u_3$ and $u_4$ in the $j$-th iteration: in \textsc{MatrixUpdate}, \textsc{PartialMatrixUpdate}, or \textsc{Query}. Note that even though we might enter \textsc{Query} after executing \textsc{MatrixUpdate} or \textsc{PartialMatrixUpdate}, $u_3$ and $u_4$ won't change inside \textsc{Query} since they were already updated in \textsc{MatrixUpdate} or \textsc{PartialMatrixUpdate}. So we can consider the updates to $u_3$ and $u_4$ in these three procedures separately.

{\bf Case 1, \textsc{MatrixUpdate} (Algorithm~\ref{alg:matrix_update_fixed}).}
\begin{align*}
    u_3^{\new} + M^{\new}u_4^{\new} - (u_3 + Mu_4) = & ~ u_3 + Mu_4 + c \cdot M^{\tmp}\sqrt{W^{\appr}}f(h^{\appr}) + M^{\new}\cdot 0 - (u_3 + Mu_4)\\
    = & ~ c \cdot M^{\tmp}\sqrt{W^{\appr}}f(h^{\appr})\\
    = & ~ c \cdot A^{\top}(AW^{\appr}A^{\top})^{-1}A\sqrt{W^{\appr}}f(h^{\appr}),
\end{align*}
where the first step is by the assignment of $u_3^{\new}$ (Line~\ref{alg:line:matrix_update_fixed:u_3}), $u_4^{\new}$ (Line~\ref{alg:line:matrix_update_fixed:u_2_u_4}),
the last step is by $M^{\tmp} = A^{\top}(AW^{\appr}A^{\top})^{-1}A$ that we already proved in Lemma~\ref{lem:matrix_update_correct_improved}.

{\bf Case 2, \textsc{PartialMatrixUpdate} (Algorithm~\ref{alg:partial_matrix_update_fixed}).}
%We denote $M^{\tmp} = A^{\top} ( A W^{\appr} A^{\top} )^{-1} A$.
\begin{align*}
  & u_3^{\new} + M^{\new}u_4^{\new} - (u_3 + M u_4)\\
= & ~ M(u_4^{\new}-u_4)\\
= & ~ M\cdot c \cdot \left( \sqrt{W^{\appr}} f(h^{\appr}) - \mathbf{1}_{S^{\new}} B^{\tmp} \L_r[(M_{S^{\new}})^{\top}] \sqrt{W^{\appr}}f(h^{\appr}) \right)\\
= & ~ c \cdot \left( M - \L_c[M_{S^{\new}}] B^{\tmp}\L_r[(M_{S^{\new}})^{\top}]\right) \sqrt{W^{\appr}}f(h^{\appr})\\
= & ~ c \cdot \left( M - \L_c[M_{S^{\new}}] \cdot \L_*[\left((\Delta_{S^{\new},S^{\new}}^{\new})^{-1} + M_{S^{\new},S^{\new}}\right)^{-1}] \cdot \L_r[(M_{S^{\new}})^{\top}]\right) \sqrt{W^{\appr}}f(h^{\appr})\\
%= & ~ c \cdot M^{\tmp}\sqrt{W^{\appr}}f(h^{\appr})\\
= & ~ c \cdot A^{\top} ( A W^{\appr} A^{\top} )^{-1} A \sqrt{W^{\appr}}f(h^{\appr}),
\end{align*}
where the first step is by $u_3^{\new}=u_3$ and $M^{\new}=M$ since they are not modified in \textsc{PartialMatrixUpdate}, 
the second step is by the assignment of $u_4^{\new}$ (Line~\ref{alg:line:partial_matrix_update_fixed:u_4}),
the fourth step is by the invariant on $B^{\tmp}$ (Part~\ref{ass:invariant_improved:B} of Assumption~\ref{ass:invariant_improved}),
the fifth step is by Woodbury identity (Lemma~\ref{lem:M_tmp_correctness}) and $M = A^{\top} ( A V A^{\top} )^{-1} A$.

{\bf Case 3, \textsc{Query} (Algorithm~\ref{alg:query_fixed}).} 
\begin{align*}
    & u_3^{\new} + M^{\new}u_4^{\new} - (u_3 + Mu_4)\\
    = & ~ M(u_4^{\new} - u_4)\\
    = & ~ M\cdot c \cdot \left( \sqrt{W^{\appr}}f(h^{\appr}) + \mathbf{1}_{S^{\new}}(\gamma^{\tmp}-\gamma_1-\partial \gamma) \right)\\
    = & ~ c \cdot \left(M - \L_c[M_{S^{\new}}] \cdot \L_*[((\Delta^{\new}_{S^{\new}, S^{\new}}) + M_{S^{\new}, S^{\new}})^{-1}] \cdot L_r[(M_{S^{\new}})^{\top}]\right)\cdot \sqrt{W^{\appr}}f(h^{\appr})\\
    %= & ~ c \cdot M^{\tmp} \sqrt{W^{\appr}}f(h^{\appr})\\
    = & ~ c \cdot A^{\top} ( A W^{\appr} A^{\top} )^{-1} A \sqrt{W^{\appr}}f(h^{\appr}),
\end{align*}
where the first step is by $u_3^{\new}=u_3$ and $M^{\new}=M$ since they are not modified in \textsc{MatrixUpdate},
the second step is by the assignment of $u_4^{\new}$(Line~\ref{alg:line:query_u_4_update_fixed}),
the third step is by the close-form formula of $\gamma^{\tmp}-\gamma_1-\partial \gamma$ (Eq.~\eqref{eq:gamma_tmp_gamma_1_partial_gamma}),
the fourth step is by Woodbury identity (Lemma~\ref{lem:M_tmp_correctness}) and $M = A^{\top} (A V A^{\top})^{-1} A$.

\noindent {\bf Part 3.}
The proof for $u_1 + Gu_2 = \sum_{i=j_1+1}^j c^{(i)} \cdot W^{\appr, (i)}A^{\top}(AW^{\appr, (i)}A^{\top})^{-1}A\sqrt{W^{\appr, (i)}}f(h^{\appr, (i)})$ follows from similar reasons as that of Part 2. We omit the details here.

\end{proof}

\subsection{Bounding \texorpdfstring{$x$}{} and \texorpdfstring{$\ov{x}$}{}}
In this section we prove that the explicit $\ov{x}$ is always within an error of $\epsilon_{\mathrm{tiny}}$ with the implicitly maintained $x$. This fact is used to prove Part 3 of Theorem~\ref{thm:correctness_fixed}. Similar as in previous sections, we use a superscript $(j)$ to denote the variable at the beginning of the $j$-th iteration.

\begin{remark}
Our entire analysis is an induction-based argument. In the $j$-th iteration, the induction hypothesis allows us to assume that Assumption~\ref{ass:assumption} is true for the $(j-1)$-th iteration, it also allows us to assume the following Lemma~\ref{lem:accuracy_of_ov_x_respect_to_x} is true for $x^{(j)}$ and $\ov{x}^{(j)}$. 

Using these two induction hypothesis we proved that Assumption~\ref{ass:assumption} is still true for the $j$-th iteration in Section~\ref{sec:error_analysis_fixed}.

Then we further use these two induction hypothesis together with \emph{Assumption~\ref{ass:assumption} for the $j$-th iteration} to prove the following Lemma~\ref{lem:accuracy_of_ov_x_respect_to_x}.
\end{remark}

\begin{lemma}[$x$ and $\ov{x}$ are close]\label{lem:accuracy_of_ov_x_respect_to_x}
Consider a fixed coordinate $i\in [n]$ and a fixed iteration number $k \leq \sqrt{n}$. Let $b$ denote the size of sketching matrix. If the following are true: (1) $x^{(0)}_i = \ov{x}^{(0)}_i$. (2) $t^{(k)}>t^{(0)}/2$. (3) There is a constant $w>0$ such that for all $j \in [k]$, $w^{\appr,(j)}_i\in [w/2,2w]$. (4) Inductively Assumption~\ref{ass:assumption} is true for iterations $1, 2, \dots, k$.
Then we have:
\begin{align*}
|(x_i^{(k)})^{-1} (\ov{x}_i^{(k)} - x_i^{(k)}) | \leq  \epsilon_x
\end{align*}
holds with probability $1-\delta$ over the randomness of sketching matrices $R^{(1)},\cdots,R^{(k)} \in \R^{b \times n}$, where $\epsilon_x = \frac{200n^{1/4}}{\sqrt{b}}\cdot \log(n/\delta)\epsilon$. 
\end{lemma}

\begin{proof}
We denote $t = t^{(k)}$. Since $t^{(0)} \leq 2t^{(k)}$ and the central path iteration will only decrease $t^{(j)}$, we have $t\leq t^{(j)}\leq 2t$ for all $j\in [k]$.

From the definition of $\wt{\delta}_x$ (Definition~\ref{def:widetilde}) and $\wh{\delta}_x$ (Definition~\ref{def:hat}), we have
\begin{align}\label{eq:difference_wt_delta_x_wh_delta_x}
\wt{\delta}_x - \wh{\delta}_x
= & ~ \frac{\wt{X}}{\sqrt{\wt{X}\wt{S}}}(I-\wt{P})\frac{1}{\sqrt{\wt{X}\wt{S}}}\wt{\delta}_{\mu} - \frac{\wt{X}}{\sqrt{\wt{X}\wt{S}}}(I-R^{\top}R\wt{P})\frac{1}{\sqrt{\wt{X}\wt{S}}}\wt{\delta}_{\mu} \notag\\
= & ~ \frac{\wt{X}}{\sqrt{\wt{X}\wt{S}}}(\wt{P}-R^{\top}R\wt{P})\frac{1}{\sqrt{\wt{X}\wt{S}}}\wt{\delta}_{\mu} 
=  \sqrt{W^{\appr}}(\wt{P}-R^{\top}R\wt{P})\frac{1}{\sqrt{\wt{X}\wt{S}}}\wt{\delta}_{\mu}.
\end{align}

Then the difference between $x^{(k)}_i$ and $\ov{x}^{(k)}_i$ can be written as
\begin{align}\label{eq:difference_x_i_ov_x_i}
|x_i^{(k)} - \ov{x}_i^{(k)}| 
= & ~\Big| \big(x^{(0)}_i + \sum_{j=1}^{k} \wt{\delta}_{x,i}^{(j)}\big) - \big(\ov{x}^{(0)}_i + \sum_{j=1}^{k} \wh{\delta}_{x,i}^{(j)}\big) \Big| 
= \Big| \sum_{j=1}^{k}( \wt{\delta}_{x,i}^{(j)} - \wh{\delta}_{x,i}^{(j)}) \Big| \notag \\
= & ~ \Big|\sum_{j=1}^{k}\sqrt{w^{\appr,(j)}}_i\big( (\wt{P}^{(j)}-R^{\top (j)}R^{(j)}\wt{P}^{(j)})\frac{1}{\sqrt{\wt{X}^{(j)}\wt{S}^{(j)}}}\wt{\delta}_{\mu}^{(j)}\big)_i\Big|,
\end{align}
where the second step is by $x^{(0)}_i = \ov{x}^{(0)}_i$, the third step is by Eq.\eqref{eq:difference_wt_delta_x_wh_delta_x}.

We define a random vector $Y_j$ for the $j$-th iteration: $Y_j = \sqrt{w^{\appr,(j)}}_i\Big( ( I - R^{(j) \top} R^{(j)} ) \wt{P}^{(j)} \frac{\wt{\delta}_{\mu}^{(j)}}{\sqrt{\wt{X}^{(j)}\wt{S}^{(j)}}} \Big)_i$.

We first bound the $\ell_2$ norm of the right part $\frac{\wt{\delta}_{\mu}^{(j)}}{\sqrt{\wt{X}^{(j)}\wt{S}^{(j)}}}$:
\begin{align*}
     \Big\|\frac{\wt{\delta}_{\mu}^{(j)}}{\sqrt{\wt{X}^{(j)}\wt{S}^{(j)}}}\Big\|_2 \leq  \Sup\Big[\frac{1}{\sqrt{\wt{X}^{(j)}\wt{S}^{(j)}}}\Big] \cdot \|\wt{\delta}_{\mu}^{(j)}\|_2
     \leq  \frac{1.1}{\sqrt{t^{(j)}}}\cdot \|\wt{\delta}_{\mu}^{(j)}\|_2
     \leq  2.2\epsilon \sqrt{t^{(j)}}
     \leq  5\epsilon \sqrt{t},
\end{align*}
where the first step is by $\|a\cdot b\|_2\leq \Sup[a]\cdot \|b\|_2$, the second step is by the inductive Assumption~\ref{ass:assumption} that $\wt{X}^{(j)}\wt{S}^{(j)}\approx_{0.1} t^{(j)}$, the third step is by $\|\wt{\delta}_{\mu}^{(j)}\|_2\leq 2\epsilon t^{(j)}$ (Fact~\ref{fact:delta_mu}), the last step is by $t^{(j)}\leq 2t$.

By Lemma~\ref{lem:E.5}, for each $j$ and with randomness over $R^{(j)}$, and use the fact that $w^{\appr,(j)}_i\in[w/2,2w]$ for all $j\in[k]$, we have
\begin{align*}
\E[ Y_j ]  = 0 ~~~\text{and}~~~
\E[ ( Y_j )^2 ]  \leq \frac{w^{\appr,(j)}_i}{b} \Big\| \frac{\wt{\delta}_{\mu}^{(j)}}{\sqrt{\wt{X}^{(j)}\wt{S}^{(j)}}} \Big\|_2^2 \leq  (2w)\cdot 25 \epsilon^2 t/b,
\end{align*}
and with probability $1-\delta/n$,
\begin{align*}
| Y_j | \leq \sqrt{w^{\appr,(j)}_i}\cdot \Big\|\frac{\wt{\delta}_{\mu}^{(j)}}{\sqrt{\wt{X}^{(j)}\wt{S}^{(j)}}}\Big\|_2 \cdot \frac{\log(n/\delta)}{\sqrt{b}} \leq (\sqrt{2w})\cdot 5\epsilon \sqrt{t} \frac{ \log ( n / \delta ) }{ \sqrt{b} } := M.
\end{align*}

Now, we apply Bernstein inequality on these zero-mean independent random variable $\{Y_j\}_{j=1}^k$ (Lemma~\ref{lem:bernstein_inequality}), $\forall \tau>0$,
\begin{align*}
\Pr \Big[ |\sum_{j=1}^k (Y_j) | > \tau \Big] \leq 2 \exp \Big( - \frac{ \tau^2 / 2 }{ \sum_{j=1}^k \E[ (Y_j)^2 ] + M \tau / 3} \Big)
\end{align*}

Choosing $\tau = 64\frac{\sqrt{wkt}}{\sqrt{b}}\log(n/\delta)^2\epsilon$, we have 
\begin{align*}
\Pr \Big[ |\sum_{j=1}^k (Y_j) | > \tau \Big] \leq 
2 \exp \Big( - \frac{ \tau^2 / 2 }{ 50 wk\epsilon^2 t/b + \tau \cdot 5\epsilon\sqrt{2wt}\frac{\log(n/\delta )}{3\sqrt{b}} } \Big) 
\leq 2\exp(-10\log(n/\delta)).
\end{align*}

Then take a union bound on all events that $|Y_j|\leq M$, we have $|\sum_{j=1}^k (Y_j)|\leq \tau$ with probability at least $1-\delta$. Therefore, 
\begin{align*}
|x^{(k)}_i - \ov{x}^{(k)}_i| 
\leq & ~ |\sum_{j=1}^k (Y_j)|
\leq  \tau
\leq  \frac{64\sqrt{wkt}}{\sqrt{b}}\cdot \log(n/\delta)^2\epsilon\\
\leq & ~ x^{(k)}_i\cdot \frac{200\sqrt{k}}{\sqrt{b}}\cdot \log(n/\delta)^2\epsilon
\leq  x^{(k)}_i\cdot \frac{200n^{1/4}}{\sqrt{b}}\cdot \log(n/\delta)^2\epsilon
\leq  x^{(k)}_i \epsilon_x,
\end{align*}
where the first step is by Eq.\eqref{eq:difference_x_i_ov_x_i}, the fourth step is by $wt\leq 2w^{(\appr),(k)}_i\cdot 2t^{(k)}\leq (2x^{(k)}_i/s^{(k)}_i)\cdot (3x^{(k)}_is^{(k)}_i)\leq 6(x^{(k)}_i)^2$ by Assumption~\ref{ass:assumption}, the fifth step is by $k\leq \sqrt{n}$.
\end{proof}

\subsection{Running time of feasible data structure}
\label{sec:time_per_call_feasible}
\begin{lemma}
It takes $O(n)$ time to execute \textsc{ScalarC}.
\end{lemma}
\begin{proof}
The proof is straightforward.
\end{proof}

\begin{lemma}\label{lem:matrix_update_time_fixed}
In the procedure \textsc{MatrixUpdate}, it takes
\begin{enumerate}
    \item $O(n^2)$ time to compute $G \leftarrow W^{\appr}M^{\tmp}$
    \item $O(n^2)$ time to compute $u_1 \leftarrow u_1 + G u_2 + c \cdot W^{\appr}M^{\tmp}\sqrt{W^{\appr}}f(h^{\appr})$
    \item $O(n^2)$ time to compute $u_3 \leftarrow u_3 + M u_4 + c \cdot M^{\tmp}\sqrt{W^{\appr}}f(h^{\appr})$
\end{enumerate}
Overall, refreshing $G$, $u_1$, $u_2$ takes $O(n^2)$ time.
\end{lemma}
\begin{proof}
This lemma directly follows from the algorithm of \textsc{MatrixUpdate} (ALgorithm~\ref{alg:matrix_update_fixed}).
\end{proof}

\begin{lemma}\label{lem:partial_matrix_update_time_fixed}
In the procedure \textsc{PartialMatrixUpdate}, it takes
\begin{enumerate}
    \item $O(n^{1+a})$ time to compute $G \leftarrow G + (W^{\appr}-\wt{V}) \cdot M$.
    \item $O(n^{1+a})$ time to compute $u_1 \leftarrow u_1 + (W^{\appr}-\wt{V}) \cdot M u_2$.
    \item $O(n^{1+a})$ time to compute $u_2 \leftarrow u_2 + c \left( \sqrt{W^{\appr}}f(h^{\appr}) - \mathbf{1}_{S^{\new}}B^{\tmp}\L_r[(M_{S^{\new}})^{\top}]\sqrt{W^{\appr}}f(h^{\appr})\right)$. And computing $u_4$ takes the same time.
\end{enumerate}
Overall, refreshing $G$, $u_1$, $u_2$, $u_4$ takes $O(n^{1+a})$ time.
\end{lemma}
\begin{proof}
From Lemma~\ref{lem:improved:sparsity_guarantee}, we have that when entering \textsc{PartialUpdate}, $\|w^{\appr}-\wt{v}\|_0 \leq O(n^a)$, and $|S^{\new}| \leq O(n^a)$.

\noindent {\bf Part 1.} Since $(W^{\appr}-\wt{V})$ is a $O(n^a)$-sparse diagonal matrix, multiplying it with a $n\times n$ matrix $M$ takes $O(n^{1+a})$ time. By memory operation, adding $(W^{\appr}-\wt{V}) M$ on $G$ also takes $O(n^{1+a})$ time.

\noindent {\bf Part 2.} Multiplying a $O(n^a)$-sparse diagonal matrix $(W^{\appr}-\wt{V})$ with a $n\times n$ matrix $M$ and then with a $n\times 1$ vector $u_2$ takes $O(n^{1+a})$ time.

\noindent {\bf Part 3.} Computing $\sqrt{W^{\appr}} f(h^{\appr})$ takes $O(n)$ time. Multiplying a $O(n^a) \times n$ matrix $\L_r[(M_{S^{\new}})^{\top}]$ with a $n\times 1$ vector $\sqrt{W^{\appr}} f(h^{\appr})$ takes $O(n^{1+a})$ time. Multiplying a $O(n^a) \times O(n^a)$ matrix $B^{\tmp}$ with a $O(n^a)\times 1$ vector $\L_r[(M_{S^{\new}})^{\top}]\sqrt{W^{\appr}} f(h^{\appr})$ takes $O(n^{2a})$ time. Finally multiplying a $n\times O(n^a)$ matrix $\mathbf{1}_{S^{\new}}$ that only has $O(n^a)$ non-zero entries with a $O(n^a)\times 1$ vector $B^{\tmp}(M_{S^{\new}})^{\top}\sqrt{W^{\appr}} f(h^{\appr})$ takes $O(n)$ time. Thus in total this step takes $O(n^{1+a})$ time.
\end{proof}

\begin{remark}
Note that this running time of $O(n^{1+a})$ is always dominated by the $O(\Tmat(n, n^a, \wt{k}))$ time of other computations of \textsc{PartialMatrixUpdate} (see Lemma~\ref{lem:partial_matrix_update_time_improved}).
\end{remark}

\begin{lemma}\label{lem:query_time_fixed}
In the procedure \textsc{Query}, it takes
\begin{enumerate}
    \item $O(n^{a+\wt{a}})$ time to compute $u_1 \leftarrow u_1 + c \cdot (W^{\appr}-\wt{V})\Big(\beta_2 + M \cdot \big(\sqrt{W^{\appr}}f(h^{\appr}) - \sqrt{V} f(g) + \mathbf{1}_{S^{\new}}(\gamma^{\tmp}-\gamma_1-\partial \gamma)\big)\Big)$.
    \item $O(n)$ time to compute $u_2 \leftarrow u_2 + c \cdot \left(\sqrt{W^{\appr}}f(h^{\appr}) + \mathbf{1}_{S^{\new}}(\gamma^{\tmp}-\gamma_1-\partial \gamma)\right)$. And computing $u_4$ takes the same time.
\end{enumerate}
Overall, calculating $u_1$, $u_2$, $u_4$ takes $O(n^{a+\wt{a}})$ time.
\end{lemma}
\begin{proof}
When entering \textsc{Query}, from Lemma~\ref{lem:improved:sparsity_guarantee}, we have that  $\|w^{\appr} - v\|_0 \leq n^{a}$ and $\|h^{\appr} - g\|_0 \leq n^{a}$, thus $\|\sqrt{W}^{\appr}f(h^{\appr}) - \sqrt{V}f(g)\|_0 \leq O(n^{a})$. From Lemma~\ref{lem:improved:sparsity_guarantee}, we also have that $\|w^{\appr} - \wt{v}\|_0 \leq n^{\wt{a}}$,  and $|S^{\new}| \leq O(n^a)$.

\noindent {\bf Part 1.} We need to compute the following four parts:
\begin{enumerate}
    \item Multiplying a $n\times n$ diagonal matrix $W^{\appr} - \wt{V}$ with a $n\times 1$ vector $\beta_2$ takes $O(n)$ time.
    \item Multiplying a $n^{\wt{a}}$-sparse $n\times n$ diagonal matrix $W^{\appr} - \wt{V}$ with a $n\times n$ matrix $M$ and then with a $O(n^a)$-sparse $n\times 1$ vector $(\sqrt{W^{\appr}} f(h^{\appr}) - \sqrt{V} f(g))$ takes $O(n^{a+\wt{a}})$ time.
    \item Computing $\mathbf{1}_{S^{\new}}(\gamma^{\tmp} - \gamma_1 - \partial \gamma)$ takes $O(n)$ time. Multiplying a $n^{\wt{a}}$-sparse $n\times n$ diagonal matrix $W^{\appr} - \wt{V}$ with a $n\times n$ matrix $M$ and then with a $O(n^a)$-sparse $n\times 1$ vector $\mathbf{1}_{S^{\new}}(\gamma^{\tmp} - \gamma_1 - \partial \gamma)$ takes $O(n^{a+\wt{a}})$ time.
\end{enumerate}
So overall this step takes $O(n^{a+\wt{a}})$ time.

\noindent {\bf Part 2.} Computing $\sqrt{W^{\appr}}\cdot f(h^{\appr})$ takes $O(n)$ time. And multiplying a $n\times O(n^a)$ matrix $\mathbf{1}_{S^{\new}}$ that only has $O(n^a)$ non-zero entries with a $O(n^a)\times 1$ vector $(\gamma^{\tmp} - \gamma_1 - \partial \gamma)$ takes $O(n)$ time. Thus in total this step takes $O(n)$ time.
\end{proof}

\begin{remark}
Note that this running time of $O(n^{a+\wt{a}})$ is the same as other computations of \textsc{Query} (see Lemma~\ref{lem:query_time_improved}).
\end{remark}

\begin{lemma}
The amortized running time of procedure \textsc{MakeFeasible} is
\begin{align*}
    (C_1 / \epsilon_{\mathrm{mp}} + C_2 / \epsilon^2_{\mathrm{mp}}) \cdot n^{1.5}.
\end{align*}
\end{lemma}
\begin{proof}
In this proof, we will use superscript notations that are consistent with section~\ref{sec:amortize_time_improved}. 
Let $\ov{w}^{(j+1)}$ denote the input $w^{\new}$ of \textsc{UpdateQuery} in the $j$-th iteration. 
Let $w^{\appr,(j+1)}$ be the output of \textsc{UpdateQuery} in the $j$-th iteration.
Let $w^{\old, (j)}$ be the values of $w^{\old}$ at the beginning of the $j$-th iteration.
Let $\wh{S}^{(j)} := \{i : |w^{\old, (j)}_i - w^{\appr, (j+1)}_i| > w^{\old, (j)}_i / 2\}$, and we define $\wh{k}_j := |\wh{S}^{(j)}|$. Note that in the $j$-th iteration, procedure \textsc{MakeFeasible} takes worst-case $O(n\cdot \wh{k}_j)$ time. We use a similar amortized argument as that of Lemma~\ref{lem:main_amortize_matrix_update}.

We define the following potential function
\begin{align*}
    \Phi_j = \sum_{i=1}^n \psi(w_i^{(j)} / w_i^{\old, (j)} - 1),
\end{align*}
where the function $\psi$ is defined as Definition~\ref{def:psi}. We let $g\in \R^n$ be the all one vector. Applying Lemma~\ref{lem:general_w_move} (set $v^{(j)}$ of the lemma statement to be $w^{\old, (j)}$), we have
\begin{align} \label{eq:make_feasible_amortized_w_move}
    (w\text{~move})^{(j)} := &~ \sum_{i=1}^n g_i \cdot \E \left[ \psi (  w^{(j+1)}_{i} / w^{\old, (j)}_{i} - 1 ) - \psi (   w^{(j)}_{i} / w^{\old, (j)}_{i} - 1 ) ~\Big|~ w^{(j)}, w^{\old, (j)} \right] \notag \\
    = &~ O(C_1 + C_2 / \epsilon_{\mathrm{mp}}) \cdot \|g\|_2
    = O((C_1 + C_2 / \epsilon_{\mathrm{mp}}) \cdot \sqrt{n}).
\end{align}

From Section~\ref{sec:error_analysis_fixed}, we have $w^{\appr, (j+1)} \approx_{\epsilon_{\mathrm{mp}}} \ov{w}^{(j+1)}$. And from Lemma~\ref{lem:accuracy_of_ov_x_respect_to_x} we have $\ov{w}^{(j+1)} \approx_{\epsilon_{\mathrm{tiny}}} w^{(j+1)}$. So $w^{\appr, (j+1)} \approx_{\epsilon_{\mathrm{mp}} + \epsilon_{\mathrm{tiny}}} w^{(j+1)}$.
For every $i\in \wh{S}^{(j)}$, we have $|w^{\appr, (j+1)}_i / w^{\old, (j)}_i - 1| > 1 / 2$ by definition of $\wh{S}^{(j)}$, thus 
\begin{align*}
    |w^{(j+1)}_i / w^{\old, (j)}_i - 1| = &~ |(w^{\appr, (j+1)}_i / w^{\old, (j)}_i) \cdot (w^{(j+1)}_i / w^{\appr, (j+1)}_i) - 1|\\
    > &~ 1 / 2 - 2 \epsilon_{\mathrm{mp}} - 2\epsilon_{\mathrm{tiny}}
    > 2 \epsilon_{\mathrm{mp}},
\end{align*}
where the last step follows from $\epsilon_{\mathrm{mp}} < \frac{1}{10}$ (Assumption~\ref{ass:epsilon_far}) and $\epsilon_{\mathrm{tiny}} < \frac{\epsilon_{\mathrm{mp}}}{1000}$ (Def.~\ref{def:epsilon_tiny}). Thus $\psi(w^{(j+1)}_i / w^{\old, (j)}_i - 1) > \epsilon_{\mathrm{mp}}$ from the definition of $\psi$. Since for $i\in \wh{S}^{(j)}$, $w^{\old, (j+1)}_i$ is updated to be $w^{\appr, (j+1)}_i$, we have 
\begin{align*}
    |w^{(j+1)}_i / w^{\old, (j+1)}_i - 1| =  |w^{(j+1)}_i / w^{\appr, (j+1)}_i - 1| 
    \leq  \epsilon_{\mathrm{mp}} + \epsilon_{\mathrm{tiny}} < 1.1 \epsilon_{\mathrm{mp}}.
\end{align*}
Thus $\psi(w^{(j+1)}_i / w^{\old, (j+1)}_i - 1) < 0.6 \epsilon_{\mathrm{mp}}$ from the definition of $\psi$ (Definition~\ref{def:psi}).

So we have
\begin{align} \label{eq:make_feasible_amortized_v_move}
    (v\text{~move})^{(j)} := &~ \sum_{i=1}^n \E \left[ \psi (  w^{(j+1)}_{i} / w^{\old, (j)}_{i} - 1 ) - \psi (   w^{(j+1)}_{i} / w^{\old, (j+1)}_{i} - 1 ) ~\Big|~ w^{(j)}, w^{\old, (j)} \right] \notag \\
    \geq &~ \sum_{i \in \wh{S}^{(j)}} \E \left[ \psi (  w^{(j+1)}_{i} / w^{\old, (j)}_{i} - 1 ) - \psi (   w^{(j+1)}_{i} / w^{\old, (j+1)}_{i} - 1 ) ~\Big|~ w^{(j)}, w^{\old, (j)} \right] \notag \\
    \geq &~ \sum_{i \in \wh{S}^{(j)}} (\epsilon_{\mathrm{mp}} - 0.6 \epsilon_{\mathrm{mp}})
    = \Omega(\epsilon_{\mathrm{mp}} \wh{k}_{j+1}).
\end{align}

Combining Eq.~\eqref{eq:make_feasible_amortized_w_move} and Eq.~\eqref{eq:make_feasible_amortized_v_move}, we have
\begin{align*}
    \E[\Phi_T] - \Phi_0 = \sum_{j=0}^{T-1} ((w\text{~move})^{(j)} - (v\text{~move})^{(j)}) \leq T \cdot (C_1 + C_2 / \epsilon_{\mathrm{mp}}) \cdot \sqrt{n} - \sum_{j=1}^T \Omega(\epsilon_{\mathrm{mp}} \wh{k}_{j}).
\end{align*}
Since when initialize we set $w^{\old}$ to be $w_0$, we have $\Phi_0 = 0$. And since $\E[\Phi_T] \geq 0$, we have 
\begin{align*}
    \sum_{j=1}^T \wh{k}_{j} \leq T \cdot (C_1 / \epsilon_{\mathrm{mp}} + C_2 / \epsilon^2_{\mathrm{mp}}) \cdot \sqrt{n}.
\end{align*}
Thus the amortized running time of \textsc{MakeFeasible} is
$
    (C_1 / \epsilon_{\mathrm{mp}} + C_2 / \epsilon^2_{\mathrm{mp}}) \cdot n^{1.5}.
$
\end{proof}

\begin{remark}
Note that the amortized time of \textsc{MakeFeasible} is dominated by the amortized time of \textsc{MatrixUpdate} (Lemma~\ref{lem:main_amortize_matrix_update}).
\end{remark}

\newpage
\section{History of Matrix Multiplication and LP}
\label{sec:matrix_multiplication}

\begin{table}[!h]
\small
    \centering
    \begin{tabular}{|l|l|l|l|l|} \hline
        {\bf Year} & {\bf Reference} & $\omega$ & {\bf Reference} & $\alpha$  \\ \hline
        1969 & \cite{s69} & $2.808$ &  & \\ \hline
        1978 & \cite{p78} & $2.796$ & & \\ \hline
        1979 & \cite{bcrl79} & $2.78$ & & \\ \hline
        1981 & \cite{s81} & $2.548$ & & \\ \hline
        1982 & \cite{r82} & $2.517$ & \cite{c82} & $0.172$ \\ \hline
        1982 & \cite{cw82} & $2.496$ & & \\ \hline
        1986 & \cite{s86} & $2.479$ & & \\ \hline
        1987 & \cite{cw87} & $2.376$ & & \\ \hline
        1997 & & & \cite{c97} & $0.29462$ \\ \hline
        2012 & \cite{w12} & $2.3729$ & & \\ \hline
        2014 & \cite{l14} & $2.37286$ & \cite{l14} & $0.30298$ \\ \hline
        2018 & & & \cite{gu18} & $0.31389$ \\ \hline 
    \end{tabular}
    \caption{The history of exponent of matrix multiplication $\omega$ and the dual exponent of matrix multiplication $\alpha$.}
    \label{tab:omega_vs_alpha}
\end{table}

\begin{table}[!ht]
\begin{center}
    \begin{tabular}{ | l | l | l | l | l | l | }
    \hline
    {\bf Year} & {\bf Author} & {\bf Reference} & {\bf Complexity}  \\ \hline
    1947 & Dantzig & \cite{d47} & $2^{O(n)}$ \\ \hline
    1979 & Khachiyan & \cite{k80} & $n^{6} $ \\ \hline %%% 
    1984 & Karmarkar & \cite{k84}  & $n^{3.5}$ \\ \hline %%% 3.2
    1986 & Renegar & \cite{r88}  & $n^{3} $ \\ \hline %%% 2.7
    1987 & Vaidya & \cite{v87} & $n^{3}$ \\ \hline
    1989 & Vaidya & \cite{v89_lp} & $n^{2.5} $ \\ \hline
    1994 & Nesterov, Nemirovskii & \cite{nn94}  & $n^{2.5} $ \\ \hline
    2014 & Lee, Sidford & \cite{ls14} & $n^{2.5} $ \\ \hline
    2015 & Lee, Sidford & \cite{ls15}  & $n^{2.5} $ \\ \hline
    2019 & Cohen, Lee, Song & \cite{cls19} & $ n^{\omega} + n^{2.5-\alpha/2} + n^{2+1/6} $ \\ \hline
    2019 & Lee, Song, Zhang & \cite{lsz19} & $ n^{\omega} + n^{2.5-\alpha/2} + n^{2+1/6} $ \\\hline
    2020 & Brand & \cite{b20}  & $ n^{\omega} + n^{2.5-\alpha/2} + n^{2+1/6} $ \\ \hline
    2020 & Brand, Lee, Sidford, Song & \cite{blss20} & $n^3$ \\ \hline
    2020 & & This paper & $n^{\omega} + n^{2.5-\alpha/2} + n^{2+1/18} $ \\ \hline 
    \end{tabular}
\end{center}\caption{Let $\omega$ denote the exponent of the current matrix multiplication. LP has $n$ variables, $d=\Theta(n)$ constraints, and all number can be encoded in $L$ bits. The running time of all these algorithms has a nearly linear dependence on $L$. We consider the case where $A$ is a dense full rank matrix. $\omega$ denotes the exponent of matrix multiplication, and $\alpha$ denotes the dual exponent of matrix multiplication. We remark that in some previous papers the running time is presented with explicit $d$ and $\nnz(A)$.
%all these results are different when $\rank(A)$ is low rank, $\nnz(A)$ is sparse and $n \not\approx d$. 
Here we present the running time assuming $d=\Theta(n)$, $\rank(A) = n$ and $\nnz(A) = n^2$. 
%Also note that until now the running time remains larger than $n^2$ when $\alpha=1$ and $\omega=2$. Currently $\omega \leq 2.373$ and $\alpha \geq 0.31$.
}
\end{table}

\end{document}